\documentclass[a4paper,12pt,numbered,custommargin,print,index]{Classes/PhDThesisPSnPDF}
\usepackage{amsfonts,amsmath,amssymb,braket}
\usepackage{hyperref,color}
\usepackage{epsf}
\usepackage[italicdiff]{physics}
\usepackage{bm}
\usepackage{mathtools}
\usepackage{tikz}
\usepackage{enumitem}
\usepackage{tensor}
\usepackage{footnote}
\usepackage{here,comment}
\usepackage{bbm}
\usepackage{shorttoc}
\usepackage{bbold}
\usepackage[export]{adjustbox}
\usepackage{scalerel}
\usepackage{enumitem}

\usepackage{titlesec}
\usepackage{slashed}
\usepackage{arydshln}
\usepackage{nccmath}
\hypersetup{% hyperrefオプションリスト
	setpagesize=false,%
	bookmarksnumbered=true,%
	bookmarksopen=true,%
	colorlinks=true,%
	linkcolor={blue!50!black},%when print, comment out the following three
	citecolor={blue!50!black},
	urlcolor={blue!50!black},
	pdftoolbar=true,
	pdfmenubar=true,
	pdfauthor={Takato Mori},
	pdfkeywords={AdS/CFT, holography, tensor network, entanglement, quantum information, quantum field theory, conformal field theory}
}

\usepackage{amsthm}
\theoremstyle{definition}
\newtheorem{definition}{Definition}[section]
\newtheorem{theorem}{Theorem}[section]

\usepackage[T1]{fontenc}

\newcommand{\eqn}[1]{\begin{equation} #1 \end{equation}}%eq without labels

\newcommand{\al}[1]{\begin{align} #1\end{align}}

\allowdisplaybreaks[1]

\newcommand{\nn}{\nonumber\\}

\newcommand{\be}{\begin{equation}}
\newcommand{\e}{\end{equation}}
\newcommand{\aln}[1]{\begin{align}#1\end{align}}
\newcommand{\ga}[1]{\begin{gather}#1\end{gather}}

\newcommand{\de}{\partial}
\newcommand{\ba}{\begin{eqnarray}}
\newcommand{\ea}{\end{eqnarray}}
\newcommand{\ee}{\end{equation}}

\newcommand{\f}{\frac}
\newcommand{\s}{\sqrt}
\newcommand{\vp}{\varphi}

\newcommand{\ti}{\tilde}
\newcommand{\ap}{\alpha}

\newcommand{\ddd}{\cdot\cdot\cdot}
\newcommand{\no}{\nonumber \\}
\newcommand{\la}{\langle}
\newcommand{\lb}{\rangle}
\newcommand{\bea}{\begin{eqnarray}}
\newcommand{\eea}{\end{eqnarray}}
\newcommand{\bes}{\begin{equation*}}
\newcommand{\beas}{\begin{eqnarray*}}
\newcommand{\eeas}{\end{eqnarray*}}
\newcommand{\bas}{\begin{array*}}
\newcommand{\eas}{\end{array*}}
\newcommand{\ees}{\end{equation*}}

\newcommand{\ep}{\epsilon}

\def\mod{\ {\rm mod}\ }

\renewcommand{\*}{ &=& }

\numberwithin{equation}{section}

\input{Preamble/preamble}

\title{Entanglement structure in quantum many-body systems, field theories, and holography}
\subtitle{\leavevmode\\量子多体系、場の理論、ホログラフィー原理における\\エンタングルメント構造}

\author{MORI, Takato}

\dept{Department of Particle and Nuclear Physics,\leavevmode\\ School of High Energy Accelerator Science,}

\university{The Graduate University for Advanced Studies, SOKENDAI}
\crest{\includegraphics[width=0.35\textwidth]{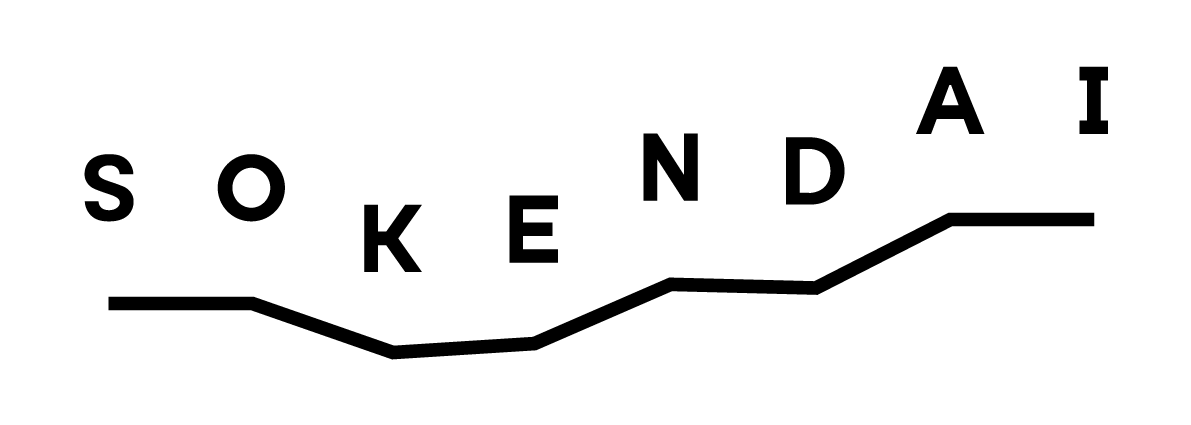}\hspace{0.1\textwidth}\includegraphics[width=0.25\textwidth]{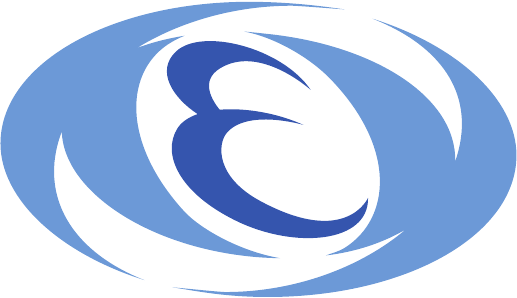}}
\supervisor{Prof. ISO, Satoshi}
\degreetitle{Doctor of Philosophy}

\degreedate{\leavevmode\\[\baselineskip]November 2022} 

\subject{Entanglement} \keywords{{Entanglement} {PhD Dissertation} {Physics} {The Graduate University for Advanced Studies (SOKENDAI)}}

\begin{document}
\begin{CJK}{UTF8}{ipxm}

\renewcommand{\thechapter}{\Roman{chapter}}

\frontmatter

\maketitle

%\include{Dedication/dedication}
%%\include{Declaration/declaration}
%\include{Acknowledgement/acknowledgement}
%\include{Abstract/abstract}
%\include{Publication/publication}

% ******************************* Thesis Dedidcation ********************************

\begin{dedication} 

I would like to dedicate this thesis to my loving fianc\'{e}e Shiori Homma and my parents and brother who supported me for a long time \dots

\end{dedication}
%\include{Declaration/declaration}
% ************************** Thesis Acknowledgements **************************

\begin{acknowledgements}      
I am deeply grateful to my supervisor Satoshi Iso for their invaluable guidance and support throughout my PhD journey. %He imparted upon me a wealth of knowledge regarding the thought processes and methods of a physicist. 
His passion and integrity in the pursuit of research, without resorting to shortcuts, served as a valuable example to me. 
Additionally, he taught me the significance of setting goals that are oriented towards understanding physical phenomena, without getting sidetracked by fleeting trends, but rather performing research like a deep diver, constantly striving to grasp the main objective.

I am extremely indebted to my host Tadashi Takayanagi and Tomonori Ugajin during my long-term stay at YITP for their support and guidance in expanding my expertise in holography.
%
%I am extremely indebted to my host Tadashi Takayanagi during my atom-type fellow at YITP for his support and guidance in expanding my expertise in holography. %His collaborative approach to research and project management is highly commendable and I was honored to have the opportunity to learn and discuss with such a leading pioneer for such an extended period.
%Furthermore, I am also greatly thankful to Tomonori Ugajin for his kindness, mentorship, and invaluable instruction in the large-$c$ calculations and CFT techniques for holography. He was always willing to share his expertise with me and was a constant source of support during my time at YITP. 
%
I extend my sincere gratitude to Hiroaki Matsueda for his steadfast support and collaboration throughout the past three years. %as we delved into various research avenues in the fields of tensor networks and information geometry. %Despite not yet achieving our original research objective, his patience and instruction in teaching me (not trained in condensed matter physics nor his student!) about condensed matter physics and tensor networks through our discussions has been invaluable. I am deeply thankful for his continued support and guidance.
I am also deeply grateful to Beni Yoshida for his exceptional hospitality and invaluable contributions during my long-term stay at PITP. %as part of the SOKENDAI Student Dispatch Program. 
His mastery of quantum information and holography, his skill in devising thought-provoking research problems, and his utilization of quantum information theoretical arguments are highly commendable. %Additionally, I appreciate his informal discussion style which fostered a stimulating and intellectually enriching environment. I am truly grateful for his guidance and support.
I would like to acknowledge the generous support of the Atsumi Foundation that financed my final PhD year, as well as SOKENDAI for my five years there. Furthermore, the atom-type fellow program provided by YITP and the SOKENDAI Student Dispatch Program by SOKENDAI were essential for the completion of this dissertation and provided invaluable experiences for me. I would also like to express my gratitude to the welcoming and stimulating environment at YITP and PITP. %, where many people from different disciplines actively engage in seminars, informal discussions, and workshops. 
This research is supported in part by Perimeter Institute for Theoretical Physics. Research at Perimeter Institute is supported by the Government of Canada through the Department of Innovation, Science and Economic Development and by the Province of Ontario through the Ministry of Colleges and Universities.

%I would like to extend my sincere thanks to my ex-colleague Katsuta Sakai, for his unwavering support, patience, and guidance over the years. His expertise in quantum field theories and willingness to share their knowledge with me was invaluable. %I am deeply grateful for his frank and honest teaching which helped me to gain a deeper understanding of this complex subject.
%I am also beholden to the support of my colleagues and friends, Takao Suyama, Noburo Shiba, Rinto Kuramochi, and Hiroyuki Adachi, during my PhD journey who provided me with valuable feedback, encouragement and support. %Their contributions are greatly appreciated and have played a vital role in the completition of this dissertation.
%
I would also like to acknowledge people I met and discussed at KEK, YITP, PITP, and workshops as well as on Twitter. Many thanks to 
%Katsuta Sakai, Takao Suyama, Noburo Shiba, Rinto Kuramochi, Hiroyuki Adachi,
%Masahiro Hotta, Guifre Vidal, Alex May, Yuya Kusuki, Yoshifumi Nakata, Zixia Wei, Kotaro Tamaoka, Ali Mollabashi, Ryotaro Suzuki, Keiichiro Furuya, Katsumasa Nakayama, Etsuko Ito, Shi Chen, Akira Matsumoto, Yoshiyasu Ito, Sotaro Sugishita, Hiromasa Watanabe, Kotaro Murakami, Jun Nishimura, Yu-ki Suzuki, Yusuke Taki, Taishi Kawamoto, Naritaka Oshita, Issei Koga, Shin'ya Mizoguchi, Kazunori Kohri, Takahiko Matsubara, Furugori Hideo, Makoto Natsuume, Yasusada Nambu, Hidetaka Manabe, Robert Mann, Eduardo Martin-Martinez, Erickson Tjoa, Bruno de Souza Leao Torres, Masazumi Honda, Masataka Watanabe, Tokiro Numasawa and so many more for all the great times and useful discussions.
Alex May, Yu-ki Suzuki, Taishi Kawamoto, Naritaka Oshita, Hidetaka Manabe, Yoshifumi Nakata, Katsuta Sakai, Takao Suyama, Rinto Kuramochi, Hiroyuki Adachi, Kotaro Tamaoka, Yuya Kusuki, Zixia Wei, Yusuke Taki, Etsuko Ito, Jun Nishimura, Issei Koga, Makoto Natsuume, Erickson Tjoa, Kohdai Kuroiwa and so many more for all the great times and useful discussions.

%I would also like to express my appreciation for the support and camaraderie provided by the members of the choral club at Tsukuba University, Mukudori. During my PhD, it was a great source of relaxation and respite from my research to be able to bond and share a common hobby with friends. The club has played a significant role in my well-being and I am thankful for the positive impact it has had on my PhD journey.
%
Finally, I would like to extend my heartfelt thanks to my family for all their unwavering support and encouragement throughout my entire academic journey. %Without their love and support, completing my PhD would not have been possible.
Last but not the least, most importantly, I am very thankful to my fianc\'ee Shiori. Thank you for being there for me through the challenging and the triumphant moments during my PhD. I am fortunate to have met her here in Tsukuba and to have shared so many wonderful moments over the past four years. Thank you for making my life doubly exponentially fun.

\end{acknowledgements}

% ************************** Thesis Abstract *****************************
% Use `abstract' as an option in the document class to print only the titlepage and the abstract.
\begin{abstract}
The aim of this dissertation is to clarify the structure of entanglement, a type of quantum correlations, in various quantum systems with a large number of degrees of freedom for holography between generic quantum systems and spacetimes toward a quantum description of our universe. Previous examinations of entanglement and holography have focused on specific classes of quantum systems due to the lack of computational techniques in field theory and the inherent limitation of holography. This dissertation informs various methods and formalisms to overcome these difficulties by extending the target quantum systems with mass and interactions, boundaries, and geometric variational ansatze. These approaches provide insights into the generalization of holography from the bottom up. This dissertation initiates a comprehensive study beyond conventional holography by establishing new techniques in quantum field theory, holography, and tensor networks. Focusing on entanglement entropy, we found it is generally expressed in terms of renormalized two-point correlators of both fundamental and composite operators. Beyond entanglement entropy, we found the operational meaning of the entanglement structure in generic tensor networks. Furthermore, we established a correct prescription for the AdS/BCFT correspondence with a local operator quench.

%700-2000 words
\end{abstract}

\chapter*{Publication List}
This Ph.D. dissertation is based on
\begin{enumerate}[label={[\arabic*]}]
    %\cite{Mori:2022xec}
    \item%{Mori:2022xec}
    \textbf{Takato~Mori}, Hidetaka~Manabe and Hiroaki~Matsueda,
    ``Entanglement distillation toward minimal bond cut surface in tensor networks,''
    \href{https://journals.aps.org/prd/abstract/10.1103/PhysRevD.106.086008}{Phys. Rev. D \textbf{106}, no.8, 086008 (2022)
    doi:10.1103/PhysRevD.106.086008}
    \href{https://arxiv.org/abs/2205.06633}{[arXiv:2205.06633 [hep-th]]}.
    %0 citations counted in INSPIRE as of 18 May 2022

    %\cite{Kawamoto:2022etl}
    \item%{Kawamoto:2022etl}
    Taishi~Kawamoto, \textbf{Takato~Mori}, Yu-ki~Suzuki, Tadashi~Takayanagi and Tomonori~Ugajin,
    ``Holographic local operator quenches in BCFTs,''
    \href{https://link.springer.com/article/10.1007/JHEP05(2022)060}{JHEP \textbf{05}, 060 (2022)}\\
    \href{https://link.springer.com/article/10.1007/JHEP05(2022)060}{doi:10.1007/JHEP05(2022)060}
    \href{https://arxiv.org/abs/2203.03851}{[arXiv:2203.03851 [hep-th]]}.
    %1 citations counted in INSPIRE as of 13 May 2022

	%\cite{Iso:2021dlj}
    \item%{Iso:2021dlj}
    Satoshi~Iso, \textbf{Takato~Mori} and Katsuta~Sakai,
    ``Wilsonian Effective Action and Entanglement Entropy,''
    \href{https://www.mdpi.com/2073-8994/13/7/1221}{Symmetry \textbf{13}, no.7, 1221 (2021)
    doi:10.3390/sym13071221}
    \href{https://arxiv.org/abs/2105.14834}{[arXiv:2105.14834 [hep-th]]}.
    %1 citations counted in INSPIRE as of 06 Dec 2021
	
	%\cite{Iso:2021rop}
    \item%{Iso:2021rop}
    Satoshi~Iso, \textbf{Takato~Mori} and Katsuta~Sakai,
    ``Non-Gaussianity of entanglement entropy and correlations of composite operators,''
    \href{https://journals.aps.org/prd/abstract/10.1103/PhysRevD.103.125019}{Phys. Rev. D \textbf{103}, no.12, 125019 (2021)
    doi:10.1103/PhysRevD.103.125019}
    \href{https://arxiv.org/abs/2105.02598}{[arXiv:2105.02598 [hep-th]]}.
    %2 citations counted in INSPIRE as of 06 Dec 2021
	
	%\cite{Iso:2021vrk}
	\item%{Iso:2021vrk}
	Satoshi~Iso, \textbf{Takato~Mori} and Katsuta~Sakai,
	``Entanglement entropy in scalar field theory and $\mathbb{Z}_M$ gauge theory on Feynman diagrams,''
	\href{https://journals.aps.org/prd/abstract/10.1103/PhysRevD.103.105010}{Phys. Rev. D \textbf{103}, no.10, 105010 (2021)
	doi:10.1103/PhysRevD.103.105010}
	\href{https://arxiv.org/abs/2103.05303}{[arXiv:2103.05303 [hep-th]]}.
	%2 citations counted in INSPIRE as of 16 Jun 2021
\end{enumerate}

%The other papers I published during the Ph.D. course are

% *********************** Adding TOC and List of Figures ***********************

\tableofcontents

%\listoffigures

%\listoftables

% \printnomenclature[space] space can be set as 2em between symbol and description
%\printnomenclature[3em]

%\printnomenclature

% ******************************** Main Matter *********************************
\mainmatter

%\renewcommand{\cleardoublepage}{}
%\renewcommand{\clearpage}{}
%\renewcommand{\newpage}{}
%\makeatletter
%\renewcommand\chapter{%
%\if@openright\cleardoublepage\else\fi
%\thispagestyle{plain}%
%\global\@topnum\z@
%\@afterindenttrue
%\secdef\@chapter\@schapter}
%\makeatother

\markboth{Introduction}{Introduction}
\chapter*{Introduction}
\addcontentsline{toc}{chapter}{Introduction}
%\fancyhead{}
%\fancyhead[RO,LE]{Introduction}

%\section*{Introduction}
The theory of quantum gravity has been one of the most important goals in high-energy physics for decades. It is thought to be essential for a consistent description of the origin of our universe and the nature of black holes. A promising approach for quantum gravity is a so-called holographic duality, or more concretely, the AdS/CFT correspondence~\cite{Maldacena:1997re}. It relates quantum gravity around the anti-de Sitter (AdS) spacetime with interacting quantum field theory (QFT) on a fixed background. From this correspondence, it is known that a minimal surface in a bulk spacetime called the anti-de Sitter spacetime corresponds to entanglement entropy, a measure of quantum correlations, of a conformal field theory (CFT) with large degrees of freedom~\cite{Ryu:2006bv,Ryu:2006ef,Hubeny:2007xt}. Motivated by this duality, we aim at extending it to more general quantum systems other than CFTs to understand how various spacetimes including ours, which are beyond AdS, arise from their entanglement (Fig.\ref{fig:dissertation}). We could say this is an entanglement-based approach to resolve the lack of background independence, an obvious obstacle in the AdS/CFT correspondence. It is also essential to study black hole physics under various backgrounds and matter fields as the microscopic description of black hole entropy is intimately related to entanglement entropy~\cite{Solodukhin:2011gn,Almheiri:2020cfm,Almheiri:2019hni,Penington:2019kki}. Besides holographic duality, the study of entanglement in various quantum systems is also motivated from condensed matter physics and lattice QFTs. Entanglement entropy is a candidate for order parameters to describe quantum phase transition and topological orders for its nonlocality \cite{PhysRevA.66.032110,Osterloh_2002,PhysRevLett.90.227902,Jin_2004,Holzhey:1994we,Calabrese:2004eu,PhysRevLett.96.110404,PhysRevLett.96.110405,Ibieta-Jimenez:2019wwo}. It is very important to establish a systematic way to study entanglement in various quantum systems beyond the existing cases.

%Towards the emergence of our universe from quantum correlations, we have attempted various approaches to reveal the nature of entanglement in quantum field theories, holographic theories, and tensor networks (Fig.\ref{fig:dissertation}). 

To achieve this goal, we need to understand the anatomy of holography and entanglement. We attempt to answer the following questions through this dissertation: 
\begin{itemize}
\setlength{\parskip}{0pt} % 段落間
\setlength{\itemsep}{3pt} % 項目間
    \item How can entanglement entropy be evaluated in generic QFTs? (Chapter \ref{ch:1})
    \item Can we add any additional ingredients in the %standard
    AdS/CFT correspondence to realize more nontrivial spacetime and dynamics? How can we formulate it? (Chapter \ref{ch:2}, \ref{ch:2-2})
    \item Can we generalize the formalism of holography itself in light of the entanglement structure? (Chapter \ref{ch:3})
\end{itemize}
From top to bottom, each question will be answered by using a diagrammatic analysis in QFTs, the AdS/BCFT correspondence~\cite{Takayanagi2011} from holography, and tensor networks to describe quantum many-body states, respectively.

%%%%%%%%%%%%%%%%%%%%%%%%%%%%%%%%%%%%

Useful sources to complete this dissertation include
\cite{BB09394904} for fundamental knowledge of quantum information related to entanglement (in Japanese);
\cite{nielsen_chuang_2010} for quantum information and quantum computations;
\cite{1130282270516945024} for entanglement in quantum many-body systems and tensor networks (in Japanese); 
\cite{Bridgeman:2016dhh,Biamonte:2017dgr} for tensor networks; 
\cite{R300000001-I031534480-00,Tomitsuka:2019} for relativistic quantum information theory (in Japanese); 
\cite{Casini:2009sr} for calculation techiniques for entanglement entropy in free field theories;
\cite{Rangamani:2016dms,Rangamani_2017,Nishioka_2009,Wu_2019} and \cite{1020000782220175744} (in Japanese) for holographic entanglement entropy; \cite{RevModPhys.90.035007} for entanglement entropy and renormalization group flow; 
\cite{Banerjee:2018,Natsuume:2014sfa,Natsuume_2015,Aharony:1999ti} for reviews in holography;
\cite{Kaplan:2016} for the bottom-up approaches for the AdS/CFT correspondence; 
\cite{Harlow:2014yka,Hartman:2015} for detailed and pedagogical lectures on black holes, quantum information, and holography including some recent progress; 
\cite{Asplund:2013zba} for the large-$c$ calculation involving heavy states;
\cite{Di_Francesco_1997,polchinski_1998} and \cite{BB19618269,BB3163561X,BB29063466} (in Japanese) for two-dimensional (boundary) CFTs; 
\cite{Takayanagi2011,Fujita:2011fp} for the AdS/BCFT correspondence.

%%%%%%%%%%%%%%%%%%%%%%%%%%%%%%%%%%%%%%%%%%%%%%%%%%%%%

%
\begin{figure}[t]
    \centering
\tikzset{every picture/.style={line width=0.75pt}} %set default line width to 0.75pt     
\begin{adjustbox}{width=1.1\textwidth}
\hspace*{-30pt}
\begin{tikzpicture}[x=0.75pt,y=0.75pt,yscale=-1,xscale=1]
%uncomment if require: \path (0,276); %set diagram left start at 0, and has height of 276
%Straight Lines [id:da12003671728809595] 
\draw  [dash pattern={on 0.84pt off 2.51pt}]  (338,50) -- (423,50) ;
\draw [shift={(426,50)}, rotate = 180] [fill={rgb, 255:red, 0; green, 0; blue, 0 }  ][line width=0.08]  [draw opacity=0] (8.93,-4.29) -- (0,0) -- (8.93,4.29) -- cycle    ;
\draw [shift={(335,50)}, rotate = 0] [fill={rgb, 255:red, 0; green, 0; blue, 0 }  ][line width=0.08]  [draw opacity=0] (8.93,-4.29) -- (0,0) -- (8.93,4.29) -- cycle    ;
%Straight Lines [id:da8087346487910968] 
\draw    (293,65) -- (243.05,146.3) ;
\draw [shift={(242,148)}, rotate = 301.57] [color={rgb, 255:red, 0; green, 0; blue, 0 }  ][line width=0.75]    (10.93,-3.29) .. controls (6.95,-1.4) and (3.31,-0.3) .. (0,0) .. controls (3.31,0.3) and (6.95,1.4) .. (10.93,3.29)   ;
%Flowchart: Connector [id:dp32839695733528596] 
\draw  [color={rgb, 255:red, 208; green, 2; blue, 27 }  ,draw opacity=1 ][line width=1.5]  (283,49.5) .. controls (283,42.04) and (293.52,36) .. (306.5,36) .. controls (319.48,36) and (330,42.04) .. (330,49.5) .. controls (330,56.96) and (319.48,63) .. (306.5,63) .. controls (293.52,63) and (283,56.96) .. (283,49.5) -- cycle ;
%Straight Lines [id:da12614396756383572] 
\draw    (376.5,59) -- (376.01,146) ;
\draw [shift={(376,148)}, rotate = 270.32] [color={rgb, 255:red, 0; green, 0; blue, 0 }  ][line width=0.75]    (10.93,-3.29) .. controls (6.95,-1.4) and (3.31,-0.3) .. (0,0) .. controls (3.31,0.3) and (6.95,1.4) .. (10.93,3.29)   ;
%Straight Lines [id:da8322784246338851] 
\draw    (386.5,60) -- (539.27,148.99) ;
\draw [shift={(541,150)}, rotate = 210.22] [color={rgb, 255:red, 0; green, 0; blue, 0 }  ][line width=0.75]    (10.93,-3.29) .. controls (6.95,-1.4) and (3.31,-0.3) .. (0,0) .. controls (3.31,0.3) and (6.95,1.4) .. (10.93,3.29)   ;
%Shape: Brace [id:dp4420384593224935] 
\draw  [line width=1.5]  (197.63,222.21) .. controls (197.63,226.88) and (199.96,229.21) .. (204.63,229.21) -- (464.82,229.12) .. controls (471.49,229.12) and (474.82,231.45) .. (474.82,236.12) .. controls (474.82,231.45) and (478.15,229.12) .. (484.82,229.11)(481.82,229.11) -- (745,229.02) .. controls (749.67,229.02) and (752,226.69) .. (752,222.02) ;
%Straight Lines [id:da31106950315803306] 
\draw [line width=1.5]    (162,36) -- (162,176) ;
\draw [shift={(162,180)}, rotate = 270] [fill={rgb, 255:red, 0; green, 0; blue, 0 }  ][line width=0.08]  [draw opacity=0] (11.61,-5.58) -- (0,0) -- (11.61,5.58) -- cycle    ;
%Shape: Rectangle [id:dp2366591762148209] 
\draw   (1,1) -- (765,1) -- (765,273) -- (1,273) -- cycle ;
% Text Node
\draw (348,23) node [anchor=north west][inner sep=0.75pt]   [align=left] {AdS/CFT};
% Text Node
\draw (290,41) node [anchor=north west][inner sep=0.75pt]   [align=left] {CFT};
% Text Node
\draw (432,40) node [anchor=north west][inner sep=0.75pt]   [align=left] {AdS};
% Text Node
\draw (8,90) node [anchor=north west][inner sep=0.75pt]   [align=left] {\begin{minipage}[lt]{96.19pt}\setlength\topsep{0pt}
{\fontfamily{ptm}\selectfont extending/generalizing}
\begin{center}
{\fontfamily{ptm}\selectfont holography}
\end{center}
\end{minipage}};
% Text Node
\draw (221,156) node [anchor=north west][inner sep=0.75pt]   [align=left] {QFT {\fontfamily{pcr}\selectfont (Ch.II)}};
% Text Node
\draw (230,13) node [anchor=north west][inner sep=0.75pt]  [color={rgb, 255:red, 208; green, 2; blue, 27 }  ,opacity=1 ] [align=left] {\textbf{Entanglement}};
% Text Node
\draw (440,13) node [anchor=north west][inner sep=0.75pt]  [color={rgb, 255:red, 208; green, 2; blue, 27 }  ,opacity=1 ] [align=left] {\textbf{Geometry}};
% Text Node
\draw (337,156) node [anchor=north west][inner sep=0.75pt]   [align=left] {AdS/BCFT {\fontfamily{pcr}\selectfont (Ch.III\&IV)}};
% Text Node
\draw (205,182) node [anchor=north west][inner sep=0.75pt]  [color={rgb, 255:red, 208; green, 2; blue, 27 }  ,opacity=1 ] [align=left] {\textbf{Entanglement entropy}};
% Text Node
\draw (189,76) node [anchor=north west][inner sep=0.75pt]  [font=\small,color={rgb, 255:red, 0; green, 0; blue, 226 }  ,opacity=1, style={fill=white}] [align=left] {Adding mass\\(and interactions)};
% Text Node
\draw (316,91) node [anchor=north west][inner sep=0.75pt]  [font=\small,color={rgb, 255:red, 0; green, 0; blue, 226 }  ,opacity=1, style={fill=white}] [align=left] {Adding boundary\\and excitation};
% Text Node
\draw (321,246) node [anchor=north west][inner sep=0.75pt]   [align=left] {Holographic realization of our universe?};
% Text Node
\draw (548,135) node [anchor=north west][inner sep=0.75pt]   [align=left] {Quantum many-body systems\\Tensor network {\fontfamily{pcr}\selectfont (Ch.V)}};
% Text Node
\draw (390,183) node [anchor=north west][inner sep=0.75pt]  [color={rgb, 255:red, 208; green, 2; blue, 27 }  ,opacity=1 ] [align=left] {Brane dynamics};
% Text Node
\draw (549,182) node [anchor=north west][inner sep=0.75pt]  [color={rgb, 255:red, 208; green, 2; blue, 27 }  ,opacity=1 ] [align=left] {\textbf{Entanglement structure}};
% Text Node
\draw (440,82) node [anchor=north west][inner sep=0.75pt]  [font=\small,color={rgb, 255:red, 0; green, 0; blue, 226 }  ,opacity=1, style={fill=white}] [align=left] {Various geometries\\as a network of tensors};
% Text Node
\draw (596,201) node [anchor=north west][inner sep=0.75pt]  [color={rgb, 255:red, 208; green, 2; blue, 27 }  ,opacity=1 ] [align=left] {Quantum operations};
\end{tikzpicture}
\end{adjustbox}
    \caption{A schematic figure explaining the motivation and goal of this dissertation and their relation to each study}
    \label{fig:dissertation}
\end{figure}
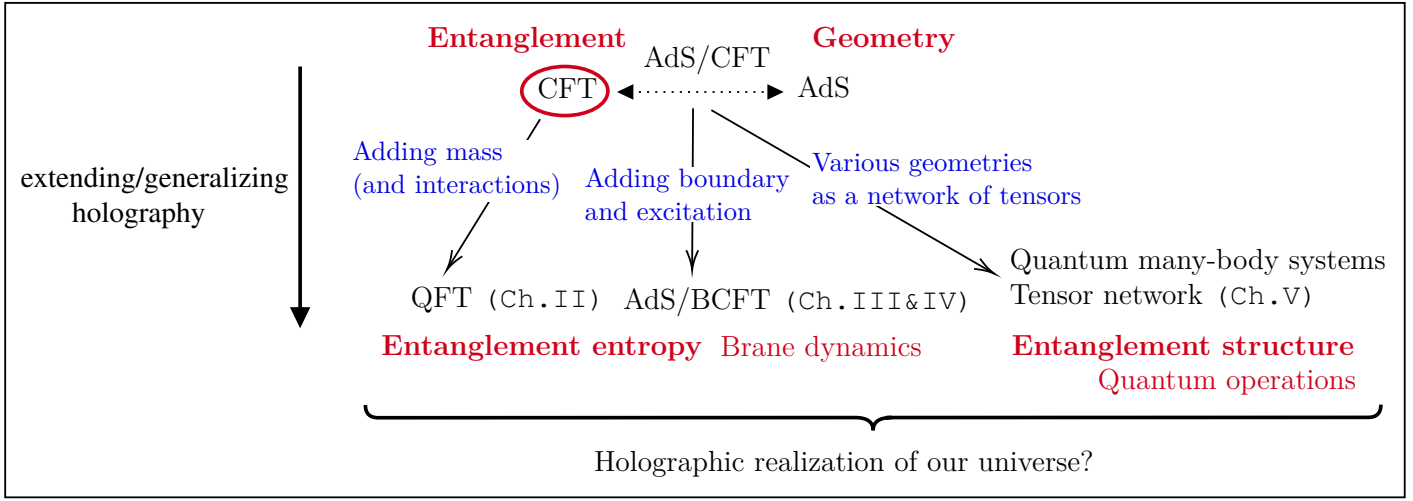
%\vspace{10pt}

The outline of this dissertation is as follows. We begin in Chapter \ref{ch:0} with an introduction to entanglement, in particular, entanglement entropy as its measure. Other than the standard definition from von Neumann entropy, we provide an operational definition as well to motivate a tensor network analysis performed in Chapter \ref{ch:3}. Some useful knowledge of quantum information, such as quantum states, measurements, and channels, is provided in the Appendix \ref{app:quantum}, \ref{app:meas}, \ref{app:quant-channel}, \ref{app:CP-inst}.

Chapter \ref{ch:1} presents a computation of vacuum entanglement entropy in massive, interacting QFTs without relying on conformal symmetry or the AdS/CFT correspondence. 
Section \ref{sec:review-EE-QFT} first reviews the existing methods to calculate entanglement entropy in free field theories and CFTs and point out disadvantages to extend them to interacting or non-conformal cases. Then, Section \ref{sec:orb-int} and subsequent sections present our contributions toward the calculation of entanglement entropy in generic QFTs~\cite{Iso:2021dlj,Iso:2021rop,Iso:2021vrk}. The computation heavily relies on our newly-invented method, the $\mathbb{Z}_M$ lattice-like gauge theory on Feynman diagrams. We show how the perturbative contributions to entanglement entropy are resummed to all orders and how it is expressed in terms of renormalized correlation functions. We also show that there exists a new contribution peculiar to interacting cases due to the splitting of interacting vertices.

Chapter \ref{ch:2} and \ref{ch:2-2} focuses on holography and its extension.
Chapter \ref{ch:2} presents an overview of the AdS/CFT correspondence and its extension to the AdS/BCFT correspondence, in which boundaries are introduced in the CFT. The first part focuses on the top-down model by Maldacena. The middle part describes (asymptotically) AdS spacetimes, CFTs, and the AdS/CFT correspondence in general dimensions. In the latter part, we focus on the AdS$_3$/CFT$_2$ correspondence, in particular, the results from the large-$c$, sparse CFTs. Finally, the last two sections provide some explanation about holographic local quench and the AdS/BCFT correspondence, which become the cornerstone for the next chapter.

Then, in Chapter \ref{ch:2-2}, we investigate a time-dependent setup in the AdS/BCFT correspondence by exciting the BCFT vacuum with a local operator (local operator quench)~\cite{Kawamoto:2022etl}. A naive expectation from a combination of the local operator quench in the AdS/CFT correspondence and the framework of the AdS/BCFT correspondence leads to several puzzles. We present a correct prescription and check it by matching the energy-momentum tensor and entanglement entropy from BCFTs and their gravity duals. We also show that this prescription elegantly resolves the puzzles.

%section 1?
Chapter \ref{ch:3} discusses the structure of entanglement in tensor networks from an operational perspective. Specific classes of tensor networks are widely investigated as finite-dimensional, information-theoretic toy models of holography. Their guiding principle is the Ryu-Takayanagi formula, relating a minimal surface in the bulk to entanglement entropy of the dual quantum system. %In these toy models, entanglement entropy is well-studied. 
To generalize holography, we are interested in the geometry of generic tensor networks and its relation to the structure of entanglement beyond entanglement entropy, which is just one aspect of entanglement. 
Section \ref{sec:tensor-network} together with Appendix \ref{app:TN-rep} provides a concise review of tensor networks for spatially one-dimensional quantum systems, in particular, matrix product states (MPS) and multi-scale entanglement renormalization ansatz (MERA).
In Section \ref{sec:HED} and subsequent sections, we show that in various tensor networks, a geometric operation toward the minimal surface is interpreted as the concentration of entangled states known as entanglement distillation~\cite{Mori:2022xec}. I analytically showed the total amount of entanglement is conserved during the process of distillation while it is being concentrated toward the minimal surface. 

Finally, in Chapter \ref{ch:discussion}, we discuss some possible resolutions for the remaining problems from quantum information and applications of our methods to braneworld, quantum black holes, and quantum computations.
%\markboth{Overview}{Overview}
%\chapter*{Overview}
%\addcontentsline{toc}{chapter}{Overview}
%\include{Introduction/overview}

%\makeatletter
%\renewcommand\chapter{%
%\if@openright\cleardoublepage\else\clearpage\fi
%\thispagestyle{plain}%
%\global\@topnum\z@
%\@afterindenttrue
%\secdef\@chapter\@schapter}
%\makeatother

%\chapter{Entanglement -- Definition in quantum systems}\label{ch:0}
%\include{Chapter0/chapter0}
%\include{Chapter0/appendix0}
%\chapter{Entanglement in quantum field theories}\label{ch:1}
%\include{Chapter1/chapter1}
%\include{Chapter1/appendix1}
%\chapter{Entanglement in holography}\label{ch:2}
%\include{Chapter2/chapter2}
%\include{Chapter2/appendix2}
%\chapter{Entanglement in quantum many-body systems}\label{ch:3}
%\include{Chapter3/chapter3}
%\include{Chapter3/appendix3}
%\chapter{Conclusion and discussion}\label{ch:discussion}
%\include{Conclusion/conclusion}

\chapter{Entanglement -- Definition in quantum systems}\label{ch:0}
%!TEX root = ../thesis.tex
%*******************************************************************************
%*********************************** First Chapter *****************************
%*******************************************************************************

%\chapter{Getting started}  %Title of the First Chapter

\ifpdf
    \graphicspath{{Chapter1/Figs/Raster/}{Chapter1/Figs/PDF/}{Chapter1/Figs/}}
\else
    \graphicspath{{Chapter1/Figs/Vector/}{Chapter1/Figs/}}
\fi
\graphicspath{{./Chapter1/Figs/}}

%Entanglement in QFT
%********************************** %First Section  **************************************
\renewcommand{\thesection}{\thechapter.\arabic{section}}
\setcounter{section}{0}

\textit{In this chapter, we review bipartite entanglement of quantum states and the definition of entanglement entropy. The key sources used here are \cite{BB09394904,nielsen_chuang_2010}.}\\

%\section{Entanglement entropy} %Section - 1.1
One of the main focuses of this dissertation is the structure of entanglement, a type of quantum correlations, in quantum systems. To quantify entanglement, a particularly useful measure is known as \textbf{entanglement entropy}\index{entanglement entropy}. In this chapter, we define it in two ways, either state-based (Section \ref{sec:state-based}) and operation-based definitions (Section \ref{sec:op-aspect-ent} and \ref{sec:EE-op-based}). In Section \ref{sec:examples-ent}, we give some basic examples of entangled states. 
In Appendix \ref{app:quantum}, the notation and the definition for quantum states are summarized. In Appendix \ref{app:meas}, a general formalism for quantum measurements is explained. In Appendix \ref{app:quant-channel}, quantum channels are defined and their representations are introduced. Appendix \ref{app:CP-inst} discusses a comprehensive formalism describing evolutions and measurements called CP-instruments.

\section{Definition -- state-based}\label{sec:state-based}
Given a density matrix $\rho\in \mathcal{S}(\mathcal{H}_A\otimes\mathcal{H}_B)$, where $\mathcal{H}_A$ and $\mathcal{H}_B$ are Hilbert spaces associated to respective subsystems $A$ and $B$, its \textbf{entanglement entropy}\index{entanglement entropy} (EE) of a subregion $A$ is defined as the von Neumann entropy of the reduced density matrix $\rho_A=\Tr_{B}\rho$:
\begin{equation}
    S_A(\rho):=-\Tr\rho_A \log\rho_A.
\end{equation}
For the later analysis, it is useful to introduce the \textbf{singular value decomposition (SVD)}\index{singular value decomposition}\index{SVD|see singular value decomposition }, also known as the \textbf{Schmidt decomposition}\index{Schmidt decomposition}. A bipartite quantum state $\ket{\Psi}_{AB}\in \mathcal{H}_A\otimes\mathcal{H}_B$ is in a one-to-one correspondence with a linear operator from $\mathcal{H}_B$ to $\mathcal{H}_A$. This is called the \textbf{channel-state duality}\index{channel-state duality} or \textbf{Choi-Jamio\l kowski isomorphism}\index{Choi-Jamio\l kowski isomorphism}. Explicitly,
\begin{align}
    \ket{\Psi}_{AB}&=\sum_{I=1}^{d_A} \sum_{\tilde{I}=1}^{d_B} \Psi_{I \tilde{I}} \ket{I}_A\otimes \ket{\Tilde{I}}_B\\
    &\downarrow\nonumber\\
    \hat{\Psi}&=\sum_{I=1}^{d_A} \sum_{\tilde{I}=1}^{d_B} \Psi_{I \tilde{I}} \ket{I}_A \bra{\Tilde{I}}_B,
\end{align}
where $d_A\equiv \dim \mathcal{H}_A, d_B\equiv \dim\mathcal{H}_B$. In the following, we omit the tensor product symbol $\otimes$.
Given the set of basis vectors, the state $\ket{\Psi}$ is completely characterized by a rectangular matrix $\Psi$. Then, via SVD, any rectangular matrix whose rank is $r$ can be written as
\begin{equation}
\Psi_{I\tilde{J}}=\sum_{K,\tilde{M}} U_{IK}\Sigma_{K\tilde{M}}V^\dagger_{\tilde{M}\tilde{J}} \quad (I=1,\cdots , d_A,\ \tilde{J}=1,\cdots , d_B),
\label{eq:svd}
\end{equation}
where $U$ and $V$ are unitary matrices on $\mathcal{H}_A$ and $\mathcal{H}_B$, respectively. $\Sigma$ is known as a \textbf{singular value matrix}. The first $r$ entries of its diagonal components are $\{\sigma_i\}_{i=1,\cdots, r}$ and other components are zero.
$r$ is called \textbf{Schmidt rank}\index{Schmidt rank} or \textbf{Schmidt number}\index{Schmidt number} of the matrix. $\{\sigma_\alpha\}_{\alpha=1,\cdots , r}$ are called the \textbf{singular values}\index{singular values} or the \textbf{Schmidt coefficients}\index{Schmidt coefficient}. They are strictly positive and usually taken to be in the descending order: $\sigma_1\ge \sigma_2 \ge \cdots \ge \sigma_r > 0$.\footnote{We can sort the singular values in this way without loss of generality under suitable permutations by a basis transformation.} For a normalized quantum state, the square of the singular values constitutes a probability distribution as $\sum_{\alpha=1}^r \sigma^2_\alpha=1$. 
When the rank $r$ is less than $d_A$ or $d_B$, the unitaries $U_{IK}$ and $V^\dagger_{\tilde{M}\Tilde{J}}$ can also be regarded as isometries $U_{I\alpha}$ and $V^\dagger_{\alpha\Tilde{J}}$:
\begin{equation}
\Psi_{I\tilde{J}}=\sum_{\alpha=1}^r U_{I\alpha}\sigma_\alpha V^\dagger_{\alpha\tilde{J}}.
\end{equation}
Finally, by absorbing these unitaries by the redefinition of basis, any bipartite states can be written as
\begin{equation}
    \ket{\Psi}_{AB}=\sum_{\alpha=1}^r \sigma_\alpha \ket{\phi_\alpha}_A \ket{\psi_\alpha}_B.
\end{equation}
In terms of the Schmidt coefficients, the eigenvalues of $\rho_A$ are $\{\lambda_\alpha\equiv \sigma_\alpha^2\}_\alpha$. This is nothing but a probability distribution. Then, EE reads
\begin{equation}
    S_A (\dyad{\Psi})= -\sum_{\alpha=1}^r \lambda_\alpha \log \lambda_\alpha.
    \label{eq:EE-SVD}
\end{equation}

In the following, we mostly focus on the case of pure states $\rho=\dyad{\Psi}$. Then, one can show immediately EE is zero if and only if the state is \textbf{separable}\index{separable} $\ket{\Psi}_{AB}=\ket{\phi}_A\otimes\ket{\psi}_B$. Since $-\lambda_\alpha\log\lambda_\alpha \ge 0$ for any $\alpha$, $S_A=0$ is equivalent to $-\lambda_\alpha\log\lambda_\alpha = 0$. This is nothing but $\lambda_1=1$ and otherwise zero. In other words, the Schmidt rank $r$ is 1. This is correct since to compose a pure state, we only need a one-dimensional Hilbert space, whose basis is proportional to the state.

When a bipartite pure state is not separable, it is called \textbf{entangled}\index{entangled}. When the whole system is pure, there is no classical correlation so that one can discriminate whether the state is separable or entangled completely by EE.\footnote{However, we should note that the structure of entanglement like the entanglement spectrum is very rich so that a single entanglement measure like EE is not sufficient for its entire characterization.} Such an appearance of entangled states is unique to quantum systems.

Lastly, let us briefly comment on entanglement of mixed states. A mixed state $\rho_{AB}\in\mathcal{S}(\mathcal{H}_A\otimes \mathcal{H}_B)$ is \textbf{separable}\index{separable} if and only if it can be written as a probabilistic mixture of separable pure states. Equivalently, a separable mixed state is always written as
\begin{equation}
    \rho_{AB}=\sum_i p_i \rho^{(i)}_A \otimes \rho^{(i)}_B,\quad 
    \left(
    \rho^{(i)}_{A}\in \mathcal{S}(\mathcal{H}_A),\ \rho^{(i)}_{B}\in \mathcal{S}(\mathcal{H}_B)
    \right),
\end{equation}
where $\{p_i\}_i$ is a probability distribution.
On the other hand, an \textbf{entangled}\index{entangled} mixed state is a mixed state which is not separable.

\section{Examples in quantum mechanics}\label{sec:examples-ent}
Among various entangled states, pure states whose EE is maximal are called \textbf{maximally entangled states}\index{maximally entangled state}. For a qudit system\footnote{When $d=2$ ($d=3)$, each state is called qubit (qutrit).}, a rank-$d$ ($\le \max(d_A,d_B)$) maximally entangled state in $\mathcal{H}_A\otimes\mathcal{H}_B$ is defined as
\begin{equation}
(U\otimes V)\frac{1}{\sqrt{d}}\sum_{i=0}^{d-1} \ket{i}_A \ket{i}_B
=(UV^T\otimes \mathbf{1}) \sum_{i=0}^{d-1} \ket{i}_A \ket{i}_B
=(\mathbf{1}\otimes VU^T) \sum_{i=0}^{d-1} \ket{i}_A \ket{i}_B,
\end{equation}
where $U$ and $V$ are arbitrary unitaries acting on $\mathcal{H}_A$ and $\mathcal{H}_B$, respectively. From \eqref{eq:EE-SVD}, EE equals $\log d$. This is the maximal possible value of EE for any qudit systems. This can be shown from the non-negativity of a so-called relative entropy for the reduced density matrix with respect to a maximally mixed state. See \cite{nielsen_chuang_2010} for the explanation. %maybe write it in the appendix
One of the maximally entangled states of the form
\begin{equation}
\ket{\mathrm{EPR}}_{AB}\equiv \frac{1}{\sqrt{d}}\sum_{i=0}^{d-1} \ket{i}_A \ket{i}_B
\label{eq:epr-state}
\end{equation}
is known as the rank-$d$ \textbf{Einstein–Podolsky–Rosen (EPR) state (pair)}\index{Einstein–Podolsky–Rosen state}\index{EPR|see Einstein–Podolsky–Rosen state }. From this, we can construct %Particularly simple examples of 
a basis set of a two-qubit system, known as \textbf{Bell states}\index{Bell states} or \textbf{Bell basis}\index{Bell basis}. They are maximally entangled states 
obtained by taking $UV^T=\mathbf{1},X,Y,Z$, where Pauli operators are defined as
\begin{equation*}
	\mathbf{1}=
	\begin{pmatrix}
	\imat{2}
	\end{pmatrix},\quad
	X=\sigma_x=
	\begin{pmatrix}
	0&1\\1&0
	\end{pmatrix},\quad
	Y=\sigma_y=
	\begin{pmatrix}
	0&-i\\i&0
	\end{pmatrix},\quad
	Z=\sigma_z=
	\begin{pmatrix}
	1&0\\0&-1
	\end{pmatrix}.
\end{equation*}
%The Bell states form a basis for the two-qubit system.
%When $d=2$, a single EPR pair has $\log 2$ EE. With the logarithm with base two, EE is $1$. Thus, we call EE ebit.

Another important example of entangled states is the \textbf{thermofield double (TFD) state}\index{thermofield double state}
\begin{equation}
    \ket{\mathrm{TFD}_\beta}_{AB}=\frac{1}{\sqrt{Z(\beta)}}\sum_{n=1}^d e^{-\beta E_n /2} \ket{n}_A \ket{n}_B,\quad Z(\beta)=\sum_n e^{-\beta E_n},
    \label{eq:TFD}
\end{equation}
where $\beta$ is an inverse temperature and $E_n$ is an $n$-th eigenvalue of the system Hamiltonian $H$. The reduced density matrix of the TFD state is nothing but the Gibb state $e^{-\beta H}/Z(\beta)$, whose observable gives the thermal expectation value. In the context of gravitational physics, $\mathcal{H}_A$ is the time reversal (i.e. complex conjugate) of $\mathcal{H}_B$. To manifest this, we sometimes write $\ket{\bar{n}}_B$ instead of $\ket{n}_B$.

%Violation of Bell's inequality
%\mynote{Violation of Bell's inequality}

\section{Operational aspects of entanglement}\label{sec:op-aspect-ent}
In the previous section, we provide a definition of EE based on the density matrix. Then, it has been shown to quantify entanglement of the state. However, in what sense does EE quantify entanglement? Is there any meaning for the value of EE? The answer to these questions is indeed yes! EE has a clear interpretation as the number of ebits (unit information measured by the natural logarithm) from the state in a certain limit. To see this, we first introduce a notion of quantum evolution and measurements and discuss a specific class of operations related to entanglement in this section. The details are described in Appendix \ref{app:meas}, \ref{app:quant-channel}, and \ref{app:CP-inst}.
%
%A quantum evolution, in other words, a most generic time evolution of a quantum state, is a map from a density matrix to another density matrix. In quantum information, it is also known as a \textbf{quantum channel}\index{quantum channel}.\footnote{Note that quantum channels described here do \textit{not} involve measurements. The effects of measurements are considered in Appendix \ref{app:meas}. A most generic measurement process including quantum evolutions is described by a CP-instrument \cite{BB09394904}.} The simplest example is a unitary evolution of a closed quantum system governed by the Schr\"{o}dinger equation. However, when environments are involved, more general evolutions are allowed. %We will describe the most general time evolution in quantum systems on the basis of measurements in quantum theories described in Appendix \ref{app:meas}. 

For now, let us borrow some known results about them to characterize entanglement in an operational manner. For simplicity, we focus on the case when a state evolution is an endomorphism, $\mathcal{H}_A\otimes \mathcal{H}_B\rightarrow \mathcal{H}_A\otimes \mathcal{H}_B$.\footnote{Note that in Appendix \ref{app:meas}, \ref{app:quant-channel}, \ref{app:CP-inst}, $\mathcal{H}_A$ and $\mathcal{H}_B$ denote Hilbert spaces of the input and output, respectively. On the other hand, here and in the rest of this thesis, $\mathcal{H}_A$ and $\mathcal{H}_B$ usually denote subsystems. In particular, when we discuss a bipartite system, $B$ denotes the complement of $A$.} 
A state evolution is described by a Kraus representation
\begin{equation}
    \Lambda(\rho)=\sum_k V_k \rho V_k^\dagger
\end{equation}
in general. Since entanglement is a quantum correlation between two subsystems, let us consider a more restricted process that should not increase entanglement.
It should consist of quantum operations and measurements can be performed \textit{locally} for each subsystem but not globally. We also allow observers in each subsystem to send these measurement results to each other via \textit{classical} communications. These processes are called \textbf{local operations and classical communications (LOCC)}\index{local operations and classical communications}\index{LOCC|see local operations and classical communications }. In this operational context, entanglement is defined as a \textit{quantum} correlation that cannot be generated by LOCC alone as they involve only \textit{classical} communications.

Successive LOCC are described as follows. 
The initial density matrix is given by $\rho\in\mathcal{S}(\mathcal{H}_{AB})$. Let Alice and Bob be an observer who can perform a local \textbf{completely positive trace-preserving (CPTP)}\index{completely positive trace-preserving}\index{CPTP|see completely positive trace-preserving } operation $\Lambda$ on $A$ and $\Gamma$ on $B$, respectively.\footnote{See Appendix \ref{sec:CPTP} for the definition of a CPTP map.}
Alice makes a selective measurement $\Lambda_m$ with outcome $m$. The state transforms under this measurement process:
\begin{equation}
    \rho\mapsto (\Lambda_m \otimes \mathbf{1})[\rho]
\end{equation}
up to a normalization factor. Next, Alice sends the outcome $m$ to Bob via a classical communication like telephone or email (thus the propagation speed of information is bounded by causality). Bob performs an arbitrary CPTP operation given by $\Gamma^{(m)}$. The state is now
\begin{equation}
    (\Lambda_m \otimes \Gamma^{(m)})[\rho]
\end{equation}
up to a normalization factor. 
If Bob performs a selective measurement with outcome $k$, the CPTP map also depends on $k$, i.e. $\Gamma^{(m)}=\Gamma^{(m)}_k$. Then, Bob can send the outcome $k$ to Alice and she can perform another measurement depending on $m$ and $k$. By repeating this procedure, we obtain a general LOCC.

By combining the indices of outcomes for Alice and Bob into a single index $l$, the state after LOCC is in general
\begin{equation}
    \sum_l p_l (\Lambda_l \otimes \Gamma_l) [\rho],
    \label{eq:locc}
\end{equation}
where we recover the normalization $p_m$ associated with the measurement of outcome $m$.
The summation accounts for arbitrary non-selective measurements and local CPTP evolutions. It is evident from \eqref{eq:locc} if we take the initial state to be separable $\rho=\sigma_A\otimes\xi_B$, the state after any LOCC remains separable:
\begin{equation}
    \sum_l p_l \, \Lambda_l[\sigma_A] \otimes \Gamma_l[\xi_B].
    \label{eq:SEP}
\end{equation}
Thus, we conclude LOCC do not generate entanglement. Note that although LOCC can be always written in this form \eqref{eq:SEP}, the converse is not always true. Operations that can be written as \eqref{eq:SEP} are called \textbf{separable operations}\index{separable operations} (SEP). There are some SEP that do not belong to LOCC. The qubit and qutrit examples are discussed in \cite{PhysRevLett.89.147901} and \cite{Bennett:1998ev}, respectively. The essential point is that either Alice or Bob must perform a measurement first in LOCC while both can perform independent local measurements in SEP. Roughly speaking, a locally degenerate state is indistinguishable in LOCC, but not in SEP. Thus, they are different classes of operations. Nevertheless, one can regard LOCC and SEP are the same class of operations when acting them on bipartite pure states \cite{gheorghiu2010separable}.

\section{Definition -- operation-based}\label{sec:EE-op-based}
Understanding fundamental classes of operations related to entanglement, we are now able to define EE in an operational manner.
It is known that one can concentrate and dilute entanglement via LOCC. The rate of concentration or dilution is nothing but EE. In the following, we focus on entanglement distillation and how it is related to EE. 

In \textbf{entanglement distillation}\index{entanglement distillation} (\textbf{entanglement concentration}\index{entanglement concentration}), we consider extracting EPR pairs from copies of a given state $\ket{\psi}^{\otimes n}\rightarrow \ket{\mathrm{EPR}}^{\otimes m}$.

In the following, we demonstrate the case of a two-qubit system with an arbitrary $n$ and $n\rightarrow\infty$. We also briefly comment on a case of a general qudit system.

Let
\begin{equation}
    \ket{\psi}=\sqrt{p}\ket{00}_{AB}+\sqrt{1-p}\ket{11}_{AB}
\end{equation}
with $p\in [0,1]$. Note that one can always write any pure states in a two-qubit system in this way via the SVD. EE for Alice or Bob is given by the \textbf{binary entropy}\index{binary entropy}
\begin{equation}
    h(p)\equiv -p\log p - (1-p)\log (1-p).
\end{equation}
To distill EPR pairs, consider its $n$-copy state\footnote{This is known as an \textbf{independent and identically distributed (i.i.d.) state}\index{independent and identically distributed state}\index{i.i.d. state|see independent and identically distributed state }.}
\begin{align}
    &\phantom{=}
    \ket{\psi}^{\otimes n} \\
    &=\left[\sqrt{p}\ket{00}+\sqrt{1-p} \ket{11} \right]^{\otimes n} _{AB} \\
    &=  \sqrt{\binom{n}{0}} p^{n/2} \frac{1}{\sqrt{\binom{n}{0}}}\ket{0\cdots 0}_A \ket{0\cdots 0}_B \nonumber \\
    &\phantom{=} + \sqrt{\binom{n}{1}} p^{\frac{n-1}{2}} (1-p)^{1/2} 
    \frac{1}{\sqrt{\binom{n}{1}}}
    \left(\ket{10\cdots 0}+ \ket{010\cdots 0}+\cdots +\ket{0\cdots 01} \right)_A \nonumber \\
    &\phantom{= + \sqrt{\binom{n}{1}} p^{\frac{n-1}{2}} (1-p)^{1/2} 
    \frac{1}{\sqrt{\binom{n}{1}}}}
    \otimes \left(\ket{10\cdots 0}+ \ket{010\cdots 0}+\cdots +\ket{0\cdots 01} \right)_B \nonumber \\ 
    &\phantom{=} +\cdots + \sqrt{\binom{n}{n}} (1-p)^{n/2} \frac{1}{\sqrt{\binom{n}{n}}}\ket{1\cdots 1}_A \ket{1\cdots 1}_B \\
    &= \sum_{k=0}^n \sqrt{\binom{n}{k}} p^{\frac{n-k}{2}} (1-p)^{k/2} \ket{\# =k}_{AB},
\end{align}
where we combine terms with the same coefficient into a single term. Note that EE of $\ket{\psi}^{\otimes n}$ is $n\, h(p)$, just $n$ times that of $\ket{\psi}$. We defined a normalized basis
\begin{equation}
    \ket{\#=k}_{AB} \equiv \frac{1}{\sqrt{\binom{n}{k}}} \sum_{\substack{\text{all possible}\\ \text{permutation}\\ \text{of qubits}}} \ket{\text{$k$ ones and $(n-k)$ zeros}}_A \ket{\text{$k$ ones and $(n-k)$ zeros}}_B.
    \label{eq:proj-epr}
\end{equation}
By labeling each element of the permutation, this is a rank-$\binom{n}{k}$ EPR state!

Now, let Alice make a selective measurement whose outcome is the number of $1$'s in the Alice qubits.\footnote{General measurement process described by POVM is explained in Appendix \ref{app:meas}.} The POVM elements for Alice are given by
\begin{equation}
    E_m = \sum_{\substack{\text{all possible}\\ \text{permutation}\\ \text{of qubits}}} \dyad{\text{$m$ ones and $(n-m)$ zeros}}_A.
\end{equation}
Explicitly, they are
\begin{equation}
\begin{cases}
    E_0 &= \dyad{0\cdots 0}_A\\
    E_1 &= \dyad{10\cdots 0}+ \dyad{010\cdots 0}+\cdots +\dyad{0\cdots 01}_A \\
    &\vdots\\
    E_n &= \dyad{1\cdots 1}_A
\end{cases}
.
\end{equation}
From
\begin{equation}
    (E_m\otimes \mathbf{1}) \ket{\# = k}_{AB} = \delta_{mk} \ket{\# = k}_{AB},
\end{equation}
the probability of outcome $m$ measured by Alice is
\begin{equation}
    \mathrm{Prob}\left(m \left| (\dyad{\psi})^{\otimes n}\right.\right) = \Tr \left[(E_m\otimes \mathbf{1}) (\dyad{\psi})^{\otimes n}\right] = \binom{n}{m} p^{(n-m)} (1-p)^m.
    \label{eq:prob-epr}
\end{equation}
Recall that if Bob performs a suitable basis transformation to \eqref{eq:proj-epr} according to the outcome, the state after each measurement becomes a tensor product of the rank-$\binom{n}{m}$ EPR state and remaining separable states. Thus, we can distill a rank-$\binom{n}{m}$ EPR pair with probability \eqref{eq:prob-epr} from $\ket{\psi}^{\otimes n}$ via LOCC.\footnote{Note that one will obtain a rank-$1$ EPR state with probability $p^n$ or $(1-p)^n$. In this specific case, the state is actually separable and the distillation fails. This reflects EE is zero in this case.} The averaged amount of entanglement distilled as EPR pairs is 
\begin{align}
    \overline{S_{\mathrm{EPR}}}&= \sum_{m=0}^{n} \mathrm{Prob}\left(m \left| (\dyad{\psi})^{\otimes n}\right.\right) \log\binom{n}{m} \nonumber\\
    &= \sum_{m=0}^{n}\exp\left[\log\binom{n}{m}+(n-m)\log p +m\log (1-p)\right]\log\binom{n}{m}
    \label{eq:ave-EE}
\end{align}
as EE of each rank-$\binom{n}{m}$ EPR pair equals $\log \binom{n}{m}$.

Let us evaluate the averaged EE in the large $n$ limit (\textbf{i.i.d. limit}\index{independent and identically distributed limit}\index{i.i.d limit|see independent and identically distributed limit }). 
From Stirling's formula $\log n! \sim n \log n -n$,
\begin{align}
     \log\binom{n}{m} 
     & \sim n\log n -m \log m -(n-m) \log (n-m) \\
     & \sim m\log \frac{n}{m} +(n-m)\log\frac{n}{n-m}.
\end{align}
The `free energy' is
\begin{equation}
    \log \mathrm{Prob}\left(m \left| (\dyad{\psi})^{\otimes n}\right.\right) \sim m \log \frac{(1-p)\, n}{m}+(n-m)\log\frac{p\, n}{n-m}.
    \label{eq:free-ene-distil}
\end{equation}
From the law of large numbers, the saddle point approximation should be valid in $n\rightarrow \infty$. Plugging the stationary point of \eqref{eq:free-ene-distil} $m=(1-p)n$ in $\overline{S_{\mathrm{EPR}}}$ \eqref{eq:ave-EE}, we obtain
\begin{equation}
    \overline{S_{\mathrm{EPR}}}\sim n\left[-p \log p - (1-p) \log (1-p)\right]
\end{equation}
in the $n\rightarrow \infty$ limit. This is nothing but the EE of $\ket{\psi}^{\otimes n}$. This indicates that on average the number of EPR pairs per copy which can be distilled from $\ket{\psi}^{\otimes n}$ is equivalent to EE of the original state $\ket{\psi}$. In other words, we can perform entanglement distillation via LOCC in $n\rightarrow \infty$,
\begin{equation}
    \ket{\psi}^{\otimes n} \mapsto \ket{\mathrm{EPR}}^{\otimes n S_A},
    \label{eq:distillation}
\end{equation}
where $S_A=h(p)$ is entanglement entropy of $\ket{\psi}$.
Although we do not discuss this in detail here, it is known that this rate of distillation is optimal through the technique of quantum data compression \cite{BB09394904}.

So far our discussion is limited to the case of a bipartite entanglement between qubits. A generalization to qu\textit{d}its can be done in the same manner. Let us consider a two-qudit state
\begin{equation}
    \ket{\phi}=\sum_{i=0}^{d-1} \sqrt{p_i}\ket{ii}_{AB},
\end{equation}
where $p_i\in [0,1]$ and $\sum_i p_i =1$. Note that any pure states in a finite-dimensional quantum system can be reduced to this expression via the SVD.\index{singular value decomposition} %For the moment, we take the base of logarithm to be 2. 
A term with the same multinomial coefficient in the $n$-copy state $\ket{\phi}^{\otimes n}$ is regarded as a basis. Then, we can perform LOCC to obtain \eqref{eq:distillation} for a general two-qudit system.

Note that in this LOCC protocol, we essentially deal with two qudits. In the case of bipartite quantum many-body states, which will be described in Chapter \ref{ch:3}, the local operations are in fact not so local, i.e. they act on a number of qubits in the subregion. There is still room to investigate a better entanglement distillation protocol that specializes in quantum many-body systems. We will work on this issue in Chapter \ref{ch:3} from the tensor network point of view.

%A {\em \LaTeX{} class file}\index{\LaTeX{} class file@LaTeX class file} is a file, which holds style information for a particular \LaTeX{}.

\nomenclature[a-dim]{$d$}{a spatial dimension of a configuration space of a quantum field theory}
\nomenclature[a-dim]{$X,Y,Z$}{a Pauli operator}
\nomenclature[z-quant]{QFT}{quantum field theory}
\nomenclature[z-quant]{EE}{entanglement entropy}
\nomenclature[z-quant]{SVD}{singular value decomposition}
\nomenclature[z-quant]{EPR}{Einstein–Podolsky–Rosen}
\nomenclature[z-quant]{TFD}{thermofield double}
%\nomenclature[g-p]{$\pi$}{ $\simeq 3.14\ldots$}                                             % first letter G is for Greek Symbols
\nomenclature[x-i]{$\mathbf{1}$}{an identity operator} % first letter X is for Other Symbols

\chapter{Entanglement in quantum field theories}\label{ch:1}
%!TEX root = ../thesis.tex
%*******************************************************************************
%*********************************** First Chapter *****************************
%*******************************************************************************

%\chapter{Getting started}  %Title of the First Chapter

\ifpdf
    \graphicspath{{Chapter1/Figs/Raster/}{Chapter1/Figs/PDF/}{Chapter1/Figs/}}
\else
    \graphicspath{{Chapter1/Figs/Vector/}{Chapter1/Figs/}}
\fi
\graphicspath{{./Chapter1/Figs/}}

\renewcommand{\thesection}{\thechapter.\arabic{section}}
\setcounter{section}{0}

\textit{This chapter contains both reviews and original materials. Useful reviews on various computational techniques for EE include~\cite{Casini:2009sr,RevModPhys.90.035007,Solodukhin:2011gn}. Section \ref{sec:orb-int} and subsequent sections follow my own work with my supervisor Satoshi Iso and ex-colleague Katsuta Sakai~\cite{Iso:2021dlj,Iso:2021rop,Iso:2021vrk}. The author of this dissertation has contributed to proposing the calculation method, perturbative calculations, and writing papers.}\\

Entanglement entropy (EE)\index{entanglement entropy} provides important information about the bipartite correlations of a given 
state. In particular, EE of a ground state or vacuum between two spatially separated regions, which is sometimes called geometric EE~\cite{PhysRevD.34.373,Srednicki:1993im,Callan:1994py,Holzhey:1994we}, 
has been widely discussed in quantum information, condensed matter physics and, even in quantum gravity, {cosmology,} and high energy physics \cite{PhysRevA.66.032110,Osterloh_2002,PhysRevLett.90.227902,Jin_2004,Calabrese:2004eu,Ryu:2006bv,Ryu:2006ef,Hubeny:2007xt,Almheiri:2020cfm,Almheiri:2019hni,Penington:2019kki,Nambu:2008my}.
Despite its  significance, practical computations of EE in field theories are not easy.
%In this chapter, we provide a perturbative method to compute vacuum EE of a half-space subregion for massive, interacting QFTs, its resummation to all orders, and the consequence from Wilsonian renormalization.

In this chapter, we give a systematic study of vacuum EE of half space in interacting field theories~\cite{Iso:2021vrk,Iso:2021rop,Iso:2021dlj}. Section \ref{sec:review-EE-QFT} is devoted to the review of previous calculation techniques. We introduce the replica trick and its generalization, the orbifold method for later discussion. Besides, we mention several other methods and explain why they are not suited for our purpose. Next, in Section \ref{sec:orb-int}, we introduce the $\mathbb{Z}_M$ gauge theory on Feynman diagrams, a novel formalism to perform the orbifold method for interacting field theories, and discuss orbifolding in terms of twists. Then, in Section \ref{sec:pert-EE} we demonstrate perturbative calculations of EE from diagrams with a single twist and show that there are two types of contributions, namely, propagator contributions and vertex contributions. In Section \ref{sec:nonpert-EE}, we perform the resummation of these perturbative contributions to all orders and find EE in terms of renormalized correlators. In Section \ref{sec:unified}, we unify the propagator and vertex contributions and present a universal formula applicable to various interactions including higher spins. Notably, we show that the vertex contributions can be written in terms of renormalized two-point functions of composite operators. In the last two sections Section \ref{sec:num-mult} and Section \ref{sec:wilson-rg}, we discuss the remaining multiple twist contributions from numerics and Wilsonian renormalization group (RG) flow. In particular, by considering the infrared (IR) limit, we expect to find a universal form of EE independent of the ultraviolet (UV) theories. 
%Finally, we give a conclusion of this chapter in Section \ref{s:discussion}. 
In Appendix \ref{app:thermal-EE}, we show the equivalence between the half-space EE and the thermodynamic entropy.
In Appendix \ref{app:spin}, we present partition functions for scalar, vector, and fermionic fields. In Appendix \ref{app:spinor}, we discuss how the Dirac fermion transforms under $SO(2)$ rotation and the counting of the number of components for each spin.
In Appendix \ref{appenarea}, we prove the area law for R\'{e}nyi entropy and
the capacity of entanglement~\cite{PhysRevLett.105.080501,deBoer:2018mzv}.
In Appendix \ref{app:twist}, we summarize the position space representation of a twisted propagator, which appears in EE.
In Appendix \ref{appencomp}, we give a proof that all the single twist contributions from  vertices are written 
in the one-loop type expression of composite operators. This is a generalization of the proof for the 
propagator contributions based on the 2PI formalism.

%********************************** %Second Section  *************************************
\section{Replica trick and its generalization}
\label{sec:review-EE-QFT}%Section - 1.2
In this section, we review the replica trick and several other methods and generalize an analytic continuation of the replica trick toward the computation of EE in interacting QFTs.

\subsection{Replica trick}\label{sec:replica}
The \textbf{replica trick}\index{replica trick} is a way to compute EE.\footnote{This replica trick originated from spin glass theory~\cite{mezard1987spin}.} This method is mathematically equivalent to the original definition (despite a subtlety due to analytical continuation), nevertheless, it is very useful for the computation in field theories.

Consider a Hilbert space that consists of two subspaces 
corresponding to the physical subsystems of interest $A$ and $\bar{A}$: 
$\mathcal{H}_{\mathrm{tot}}=\mathcal{H}_A\otimes\mathcal{H}_{\bar{A}}$. 
The EE of $A$ is defined as 
\[ S_{A}=-\Tr_{A}\rho_{A}\log\rho_{A},\]
where $\rho_A=\Tr_{\bar{A}}\rho_\mathrm{tot}$ is a reduced density matrix of 
the total one, $\rho_\mathrm{tot}$. 
The replica trick \cite{Holzhey:1994we,Calabrese:2004eu}:
\aln{
	S_A%=-\Tr_A \left(\rho_A\log \rho_A \right)
	:=\lim_{n\rightarrow 1} S_n
	= -\lim_{n\rightarrow 1} \pdv{n} \left[\Tr_A\rho_A^n\right],
	\label{eq:replica}
}
where $S_n$ is the \textbf{R\'{e}nyi entropy}\index{R\'{e}nyi entropy}, defined by 
\begin{equation}
    S_n := \frac{1}{1-n}\log \Tr\rho_A^n.
\end{equation}
For EE to be uniquely determined by the replica trick, we assume the analytical continuation of $n\in\mathbb{Z}_{>1}$ to $\mathbb{R}$. This formula holds for general $A$ and $\rho_\mathrm{tot}$ as long as such an analytical continuation exists. The uniqueness of the analytical continuation is supported by the \textbf{Carlson's theorem}\index{Carlson's theorem}~\cite{Casini:2009sr,Witten:2018xfj}.\footnote{The theorem requires $\abs{S_n}$ as a function of the complex $n$ to be bounded by an exponential function, $C \exp (C^\prime \abs{n})$ for some constants $C$ and $C^\prime$ throughout $\Re{n}\ge 1$ (and one more condition for $\Re{n}= 1$). Then, this excludes ambiguities of adding a function that is zero in the $n\rightarrow 1$ limit; $\sin \pi n$, for example.}

Note that all we need to compute here is just a trace of the $n$-fold density matrix and there is no logarithm. This is extremely useful in field theories as it has a clear interpretation in terms of the Euclidean path integral.

In this chapter, we are interested in the vacuum EE when $A$ is a half space on a time slice in a $(d+1)$-dimensional  Minkowski spacetime,
$A=\{x^0=0,x_\perp\ge 0, \forall x_\parallel\}$, where $x_0$ is a Lorentzian temporal coordinate while 
$x_\perp$  is a one-dimensional normal direction and $x_\parallel$ are the rest $(d-1)$-dimensional
parallel directions %directions 
to $\partial A$ %respectively 
(Fig.\ref{fig:subregion}).
%%%%%%%%%%%%%%%%%%%%%%
\begin{figure}[t]
	\centering
	\includegraphics[width=10cm,clip]{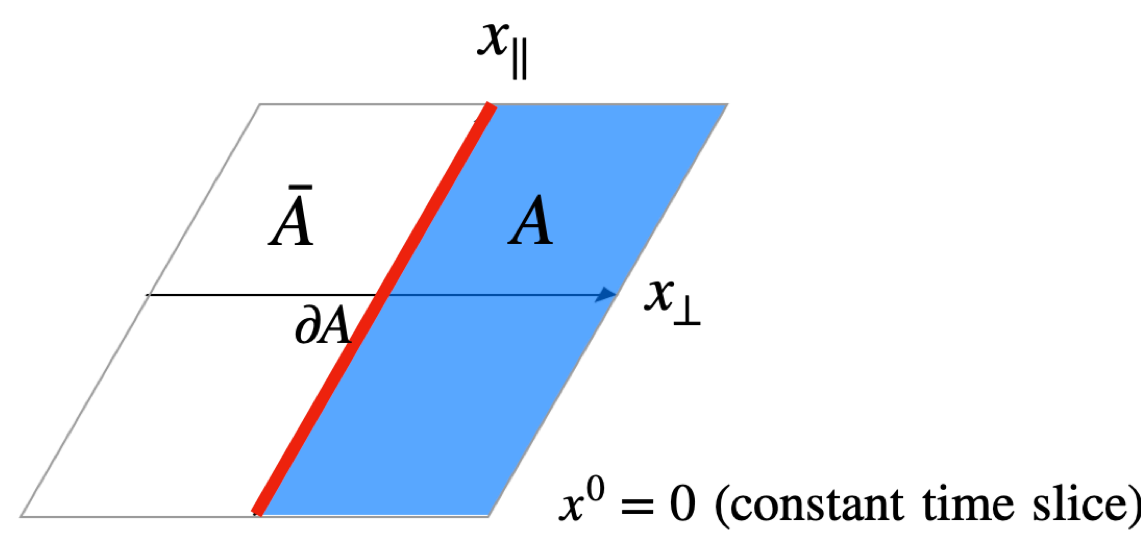}
	\caption{Our choice of the subregion $A$ and its complement $\bar{A}$. It is a half space given by $A=\{x^0=0,x_\perp\ge 0, \forall x_\parallel\}$. The boundary of the subregion is parametrized as $\partial A=\{x^0=0,x_\perp=0,\forall x_\parallel\}$.}
	\label{fig:subregion}
\end{figure}
%%%%%%%%%%%%%%%%%%%%%%

%%%%%%%%%%%%%%%%%%%%%
\begin{figure}[t]
    \centering
    \includegraphics[width=7cm,clip]{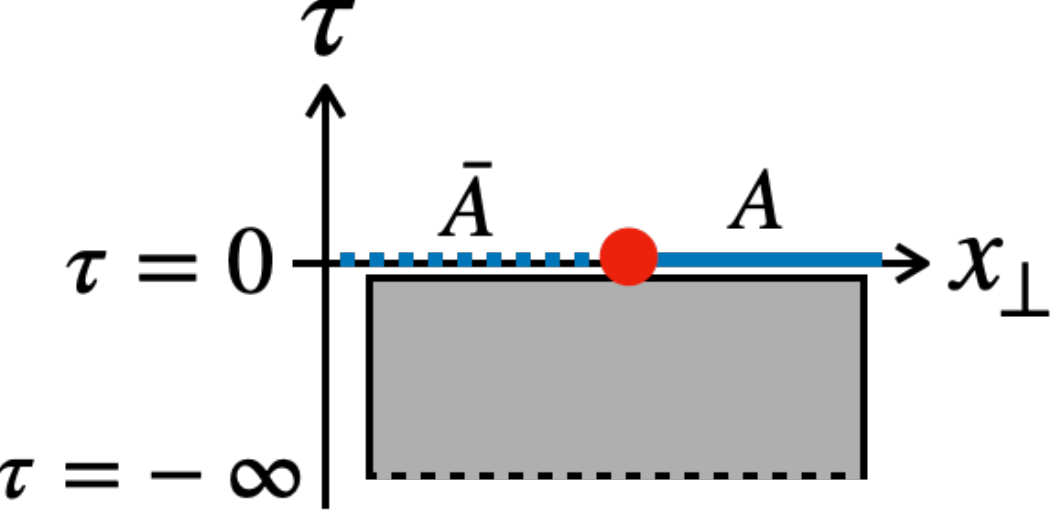}
    \caption{The (nondegenerate) vacuum in the Euclidean path integral. The red dot denotes the subregion boundary $\partial A$. The $x_\parallel$ direction is omitted.}
    \label{fig:ket-vac}
\end{figure}
%%%%%%%%%%%%%%%%%%%%

%%%%%%%%%%%%%%%%%%%%%%
\begin{figure}[t]
	\begin{tabular}{c}%prevent line break
		\hspace*{-0.05\linewidth}
		\begin{minipage}{0.5\hsize}%align figs horizontally
			\centering
			\includegraphics[width=\linewidth]{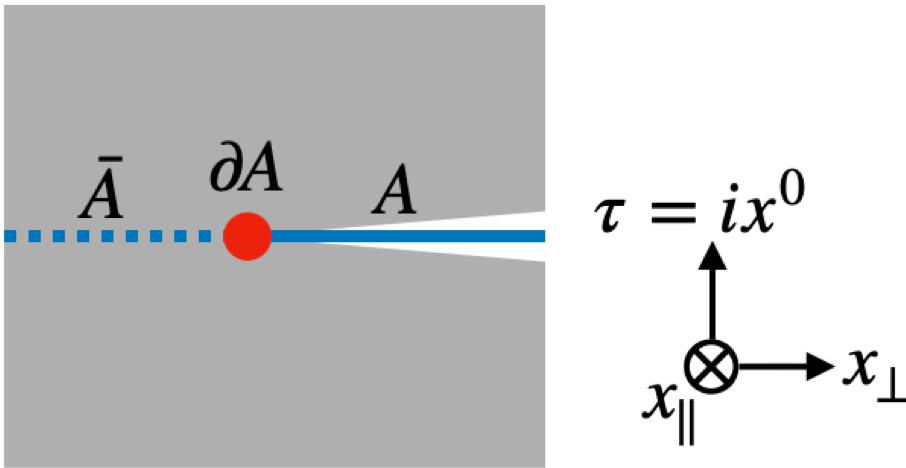}
		\end{minipage}
		\hspace*{0.05\linewidth}
		\begin{minipage}{0.3\hsize}
			\centering
			\includegraphics[width=\linewidth]{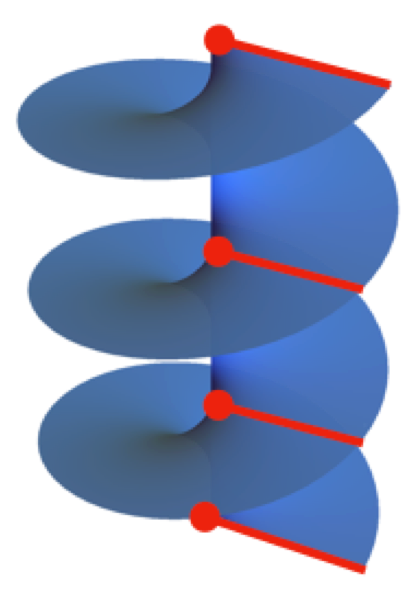}
		\end{minipage}
	\end{tabular}
	\caption{
		The Euclidean path integral representation of our reduced density matrix $\tilde{\rho}_A$ (left) and its $n$-fold cover $\tilde{\rho}_A^n$ ($n=3$) (right).
		}
	\label{fig:euclid-pathint}
\end{figure}
%%%%%%%%%%%%%%%%%%%%%%

Toward the path integral interpretation of EE \eqref{eq:replica}, let us define an unnormalized density matrix $\tilde{\rho}_\mathrm{tot}$ such that
$\rho_\mathrm{tot}=\tilde{\rho}_\mathrm{tot}/Z_1$, 
%\[\rho_\mathrm{tot}=\frac{\tilde{\rho}_\mathrm{tot}}{Z_1},\]
where $Z_1$ is a partition function of the total system on $\mathbb{R}^2\times\mathbb{R}^{d-1}$, 
where $\mathbb{R}^2$ is spanned by the Euclidean time $\tau=ix^0$ and $x_\perp$ with respect to $\partial A$ and the rest $\mathbb{R}^{d-1}$ is spanned by $x_\parallel$. 
In the \textbf{Euclidean path integral}\index{Euclidean path integral}, the unnormalized vacuum $\ket{\Psi}$ are prepared by the infinite Euclidean time evolution (Fig.\ref{fig:ket-vac}), i.e the wave function is given by
\[
\braket{\varphi}{\Psi}=\int^{\phi=\varphi} \mathcal{D}\phi e^{-S_E[\phi]},
\]
where $S_E[\phi]$ is the Euclidean action of the theory. 
Then, it defines the unnormalized reduced density matrix 
\[
\mel{\phi^\prime_A}{\tilde{\rho}_A}{\phi_A} := \Tr_{A} \tilde{\rho}_\mathrm{tot} = \int \mathcal{D}\phi_{\bar{A}} \braket{\Psi}{\phi_A, \phi_{\Bar{A}}} \braket{\phi^\prime_A, \phi_{\Bar{A}}}{\Psi}.
\]
Pictorically this means gluing the bra and ket vectors at $\Bar{A}$ (Fig.\ref{fig:euclid-pathint} to the left). 
$\Tr \tilde{\rho}_A^n$ is regarded as a partition function of the theory on $\Sigma_n\times\mathbb{R}^{d-1}$, where 
$\Sigma_n$ is an $n$-folded cover of a two-dimensional plane spanned by  %the Euclidean time 
$\tau$ and $x_\perp$ %, and thus, 
or equivalently a two-dimensional cone with an excess angle $2\pi(n-1)$.
The Euclidean path integral representation of the replicated reduced density matrix is shown in Fig.\ref{fig:euclid-pathint} to the right. The EE \eqref{eq:replica} can be rewritten as 
\begin{align}
	S_A &=-\pdv{n} \left.\left(\Tr \frac{\tilde{\rho}_A^n}{(Z_1)^n} \right)\right\vert _{n\rightarrow 1}\nonumber\\
	&=-\pdv{n}\left.\left(\frac{Z_n}{(Z_1)^n}\right)\right\vert _{n\rightarrow 1}\nonumber\\
	&=\left.\left(\frac{Z_n}{(Z_1)^n} \log Z_1 -\frac{1}{(Z_1)^n} Z_n \pdv{n} \log Z_n\right)\right\vert _{n\rightarrow 1}\nonumber\\
	&=\log Z_1 - \left.\left(\pdv{n} \log Z_n \right)\right\vert _{n\rightarrow 1}\nonumber\\
	%&=\left. \pdv{n} (n\log Z_1 -\log Z_n)\right \vert_{n\rightarrow 1} \nonumber\\
	&=  \left.\pdv{F_n}{n} \right|_{n\to1}-F_1
	\label{e:EE_n}
\end{align}
 in terms of the free energy  on $\Sigma_n \times \mathbb{R}^{d-1}$
\[
F_n:=-\log \Tr_A\tilde{\rho}_A^{\,n}.
\] 

\subsection{Heat kernel method}\label{sec:heat-kernel}
To compute the free energy $F_n$ directly, it seems we need to deal with quantum field theory (QFT) in curved spacetime, however, it is not always the case! %Since EE is defined on an original spacetime without a conical singularity, we expect it is expressed in terms of the original flat spacetime quantities. 
When the subregion is taken to be a half space as in the previous subsection, we can use the so-called \textbf{heat kernel method}\index{heat kernel method}~\cite{Vassilevich:2003xt} to show EE in free QFTs can be expressed in terms of the original flat space quantities. In the heat kernel method, we can compute $F_n$, the free energy on a cone with an excess angle $2\pi(n-1)$. Although we work in the Minkowski spacetime, it also works in a curved spacetime like black hole backgrounds by utilizing the short distance expansion called the heat kernel expansion \cite{Fursaev:1994in,Solodukhin:2011gn,Casini:2010kt,Lewkowycz:2012qr,Hertzberg:2010uv,Herzog:2013py}.

The heat kernel method is applicable for a one-loop effective action. In the saddle point approximation, one can expand the action to quadratic order in fluctuations around the saddle point $\phi=\phi_{cl}$. By performing the Gaussian integral, the one-loop partition function is given by
\begin{equation}
Z=e^{-S[\phi_{cl}]} {\det} ^{-1/2} \hat{\mathcal{D}} .
\label{eq:part-fn-general}
\end{equation}
$\hat{\mathcal{D}}$ is the inverse of the propagator as a differential operator. For simplicity, let us focus on a bosonic field theory with the bare mass $m_0$. $\hat{\mathcal{D}}$ is equal to
\begin{equation}
G^{-1} \equiv -\Box +m_0^2,
\end{equation}
where $\Box$ is the Laplace-Beltrami differential operator.
Then, the free energy apart from the classical contribution is given by
\begin{equation}
    F= \frac{1}{2}\log \det G^{-1} = \frac{1}{2} \tr \log G^{-1}.
\end{equation}
Since
\begin{equation}
    \log \frac{A^2}{\varepsilon^2} = - \int_{\varepsilon^2}^{\infty} \frac{\dd s}{s} e^{-s A^2}, \quad
    \varepsilon \ll 1
\end{equation}
in the Schwinger parametrization, the free energy is written as a worldline partition function, i.e.
\begin{equation}
    F = -\frac{1}{2} \int_{\epsilon^2}^{\infty} \frac{\dd s}{s} \tr e^{-s G^{-1}} =: -\frac{1}{2} \int_{\epsilon^2}^{\infty} \frac{\dd s}{s} e^{-s m_0^2} \tr K(s), 
    \label{eq:free-ene-heat}
\end{equation}
where 
\begin{equation}
    K(x,y,s)\equiv \mel{x}{K(s)}{y}=\mel{x}{e^{s\, \Box}}{y}
\end{equation}
is called the \textbf{heat kernel}\index{heat kernel}.\footnote{The reason for the name ``heat kernel'' is that it satisfies a local heat equation $(\partial_t - \Box)_x K(x,y,s)=0$ with the initial condition $K(x,y,0)=\delta^{d+1}(x-y)$.}
The free energy is regulated by the heat kernel regularization with the UV cutoff $\epsilon^{-1}$. 

\eqref{eq:free-ene-heat} tells us that all we need to calculate is the heat kernel. Our goal was to calculate EE via the free energy on a cone. Let us consider a cone $C_\alpha$ with a conical angle $2\pi\alpha$ and the heat kernel $K_\alpha (x,y,s)$ on it. For simplicity we focus on a two-dimensional cone without transverse directions. As the rotational symmetry follows from the Lorentz invariance, $K_\alpha$ only depends on the difference of the angle $\Delta\theta$. 
Omitting the radial direction, let us denote the kernel by
\[
K_{\alpha}(\Delta\theta, s) := K_\alpha(x,y,s)
\]
where $x=(r\cos\theta,r\sin\theta)$, $y=(r^\prime\cos\theta^\prime,r^\prime\sin\theta^\prime)$, and $\Delta\theta=\theta-\theta^\prime$.
By applying the \textbf{Sommerfeld formula}\index{Sommerfeld formula}, the heat kernel on the cone can be constructed from the flat-space ($2\pi$-periodic) one~\cite{10.1112/plms/s1-28.1.395}:
\begin{equation}
    K_\alpha (\Delta\theta,s)=K(\Delta\theta,s)+\frac{i}{4\pi\alpha}\int_{\Gamma} \dd w \cot\frac{w}{2\alpha} K(\Delta\theta +w, s),
\end{equation}
where $\Gamma=(-\pi- i \infty, -\pi+i\infty)\cup (\pi+i\infty, \pi-i\infty)$. 
Its derivation requires deformations in the complex plane. See e.g. Section 4.7 of \cite{Fursaev:2011zz} for details. It follows from the formula that the heat kernel on the cone is written in terms of the kernel in the flat space.

The heat kernel method relies on the expression of a Laplace-Beltrami type of differential operator. This indicates that it works well for free fields, whose Green function is diagonalizable in the momentum space. The explicit calculation gives $K_\alpha$ in terms of the Bessel function of the first kind and its trace is easily calculated~\cite{Casini:2010kt}. Although the modification of the Green function is permissible~\cite{Solodukhin:2011gn}, it is generally difficult to take interactions into account unless we make some approximation like a coincidence limit~\cite{Hertzberg:2012mn}. Thus, we need an alternative approach to it.

\subsection{Various approaches to entanglement entropy}
Other than the replica trick, there are several different approaches to the computation of EE~\cite{Casini:2009sr}. However, for our purpose, they have limited applicability, especially when one tries to extend them to interacting cases. Nevertheless, since it is instructive to introduce some approaches, we address results from three cases:
(i) free fields, (ii) a half-space subregion, and (iii) an interval in $(1+1)$-dimensional (perturbed) conformal field theories (CFTs). Computations in (i) and (ii) rely on other approaches than the replica trick. The derivation of EE in (iii) uses the replica trick, though, we introduce this case here as well since this uses the conformal invariance, different from our studies in the subsequent sections in this chapter.

The first case (i) can be dealt with by the \textbf{real-time formalism}\index{real-time formalism for entanglement entropy} of the Gaussian EE~\cite{PhysRevD.34.373,PhysRevA.67.052311,Casini:2009sr,Chen:2020ild}:\footnote{The equivalence to the replica trick is discussed in \cite{Casini:2010kt}.}
\begin{equation}
    S_A = \Tr_A \left[\left(\sigma+\frac{1}{2}\right)\log \left(\sigma+\frac{1}{2}\right) - \left(\sigma-\frac{1}{2}\right) \log \left(\sigma-\frac{1}{2}\right) \right],
\end{equation}
where $\sigma=\sqrt{XP}$, $\mel{x}{X}{y}=\expval{\phi(x)\phi(y)}$, and $\mel{x}{P}{y}= \expval{\pi(x)\pi(y)}$ for a bosonic free field. Refer~\cite{PhysRevA.70.052329,Katsinis:2017qzh,Bianchi:2019pvv,Lewkowycz:2012qr,Buividovich:2018scl} for examples.\footnote{The formula can be applied even for non-commutative fields such as the BFSS matrix model~\cite{Buividovich:2018scl}.} A similar formula is found for a free fermionic field as well~\cite{Herzog:2013py}. 
This formula is written in terms of the correlators due to the so-called Williamson theorem~\cite{10.2307/2371062,BB28192157}.
This formula is applicable to any subregions and explicitly written in terms of correlators in the original Minkowski spacetime. The derivation of the formula strongly relies on the Gaussian nature of the ground state. Any Gaussian states remain Gaussian after the partial trace and the reduced density matrix can be diagonalized into a sum of mode-wise thermal states. It is remarkable in the sense that one does not restrict the shape of the subregion and the information of the target state or correlators suffices. On the other hand, this is peculiar to Gaussian states and a generalization to interacting cases is difficult (although see~\cite{Chen:2020ild} for few-order perturbations). Furthermore, it is usually difficult to derive an analytic form after the integration $\Tr_A$. Even for a Gaussian QFT, the partial trace in the position space is much harder than that in the momentum space.

The second case (ii) can be dealt with by the explicit form of the \textbf{modular Hamiltonian}\index{modular Hamiltonian} $K_A = -\log \rho_A -\log Z$. $K_A$ of half space in the Minkowski spacetime is given by \cite{10.21468/SciPostPhysCore.2.2.007}
\begin{equation}
    K_A=2\pi \int_{x^0=0.x_\perp\ge 0} \dd[d-1]{x_\perp} x_1 T^{00}.
    \label{eq:modular-ham-half}
\end{equation}
This is a consequence of the \textbf{Bisognano-Wichmann theorem}\index{Bisognano-Wichmann theorem} in QFT \cite{Bisognano:1976za}. The modular Hamiltonian is actually a boosted Hamiltonian. This is nothing but the Unruh effect and EE becomes the thermal entropy with the Hamiltonian $K_A$. Refer to Appendix \ref{app:thermal-EE} for the equivalence between the EE and thermal entropy. EE is given by
\eqn{S_A=\expval{K_A}+\log Z.}
Thus, it is also written in terms of a correlator (the vacuum expectation value of the Hamiltonian).
However, the explicit evaluation is difficult since it involves a position-space integration and an explicit position dependence in the integrand. The analytic evaluation in this way becomes much more difficult in the presence of interactions.

In the third case (iii), we have an explicit formula for EE for an interval of length $l$ in a $(1+1)$-dimensional CFT with central charge $c$.\index{conformal field theory} This calculation relies on the replica trick and the conformal invariance, which will be discussed in Section \ref{sec:replica-cft}. 
%which we review in the next chapter.
%\mynote{maybe review the calculation of EE in CFT in the next chapter?}
This is known as the \textbf{Calabrese-Cardy formula}\index{Calabrese-Cardy formula}~\cite{Holzhey:1994we,Calabrese:2004eu}, in which EE is calculated from a two-point function of so-called twist operators. The result is
\begin{equation}
    S_A=\frac{c}{3}\log\frac{l}{\epsilon},
\end{equation}
where $l$ is the length of the subregion $A$ and $\epsilon$ is the lattice spacing, which is reciprocally related to the UV cutoff $\Lambda$. In a renormalized QFT perturbed around a CFT, since the only dimensionful coupling is the mass $m$, which is roughly the inverse of the correlation length $\xi$, EE for an interval $l\gg \xi\gg \epsilon$ is given by~\cite{Calabrese:2004eu}\footnote{Note that if we take the subregion $A$ to be a semi-infinite line, $S_A=\frac{c}{6}\log\frac{\xi}{\epsilon}$ since $A$ has only a single endpoint.}
\begin{equation}
    S_A=\frac{c}{3}\log\frac{\xi}{\epsilon}=-\frac{c}{6}\log(m\epsilon).
    \label{eq:EE-massive}
\end{equation}
This result is also obtained from a holographic calculation~\cite{Nishioka:2009un}. (Holographic EE will be discussed in the latter part of Section \ref{sec:CFT2}.)
Again, \eqref{eq:EE-massive} indicates EE is expressed in terms of correlators through the correlation length.
The advantage of this method is that we can compute EE of various shapes of subregions through the Euclidean path integral representation~\cite{Calabrese:2004eu,Calabrese:2009qy,Ruggiero:2018hyl,Hung:2014npa,Casini:2010kt}.
However, the obvious drawback of this formula is that we need to start from a CFT and study a perturbation from it~\cite{Rosenhaus:2014zza,Rosenhaus:2014woa,Rosenhaus:2014ula}. Higher-order contributions are usually very difficult to compute~\cite{Rosenhaus:2014zza}. It is not obvious for the theory far from the fixed point like a theory in the middle of the renormalization group (RG) flow. Furthermore, the higher dimensional generalization of twist operators is not well established.

\subsection{Some known approaches to entanglement entropy in interacting QFTs and remaining issues}
Compared to free or conformal field theories, we have little understanding of EE for general interacting QFTs, 
apart from exactly solvable cases \cite{Donnelly:2019zde}. 
In some supersymmetric theories, the localization method enables an exact calculation of the free energy and EE, 
\cite{Nishioka:2009un,Jafferis:2011zi,Pufu:2016zxm,Nishioka:2013haa,RevModPhys.90.035007}. 
EE in interacting theories is also 
discussed in perturbative \cite{Hertzberg:2012mn,Chen:2020ild,Rosenhaus:2014zza,Rosenhaus:2014woa}, nonperturbative \cite{PhysRevB.80.115122,Akers:2015bgh,Cotler:2015zda,Fernandez-Melgarejo:2020utg,Fernandez-Melgarejo:2021ymz,Bhattacharyya:2017pqq} 
or lattice \cite{Wang_2014,Buividovich:2008kq,Buividovich:2008gq,Itou:2015cyu,Rabenstein:2018bri} approaches. 
Nonperturbative studies have taken advantage of the large-$N$ analysis and the RG flow in the $O(N)$ vector model \cite{PhysRevB.80.115122,Whitsitt:2016irx,Akers:2015bgh,Hampapura:2018uho} or variational trial wave functions \cite{Cotler:2015zda,Fernandez-Melgarejo:2020utg,Fernandez-Melgarejo:2021ymz} or instanton formalism~\cite{Bhattacharyya:2017pqq}. 
These works have partly grasped the behavior of EE relevant to renormalization and beyond free theories. 
Despite these studies, there are several issues yet to be understood.

Since field theories contain infinitely many degrees of freedom, EE suffers from UV divergences and 
appropriate regularization and renormalization are necessary to obtain finite results.\footnote{If the theory is massless like the Gaussian CFTs, it may cause IR divergence. However, we consider a massive theory in this chapter so that the IR divergences are assumed to be absent.} %universal terms. 
For free theories, the UV-divergent EE can be regularized by suitably renormalizing parameters in the background gravity \cite{Cooperman:2013iqr,Barrella:2013wja,Taylor:2016aoi,Taylor:2020uwf,Liu:2012eea,Liu:2013una}.\footnote{There are some subtleties from dimensionality. See~\cite{Demers:1995dq,deAlwis:1995cia,Kim:1996bp}.}\footnote{It may sound puzzling that the UV divergence in EE is absorbed into gravitational couplings even if the theory lives in the flat space~\cite{Pang:2017grr}. Basically, this is because this divergence is regulated by a geometric cutoff; the tip of the cone is smoothed in the replica calculation.} This type of regularized EE is called renormalized EE and is consistent with the treatment for the black hole entropy \cite{Susskind:1994sm,Kabat:1995eq,Larsen:1995ax,Jacobson:1994iw,Fursaev:1994ea,Solodukhin:1995ak,Frolov:1996aj}. 
There are  additional UV divergences  in interacting field theories due to the non-Gaussianity of the vacuum wave function. They should be dealt with  the usual flat space renormalization. 
A perturbative treatment of 
this renormalization was discussed  \cite{Hertzberg:2012mn}.
In the following sections, we focus on the latter divergence. We attempt to separate the non-Gaussian effect due to interactions from the Gaussian contribution after the flat space renormalization. 

\subsection{Orbifold method}\label{s:ZM}
Toward the calculation of EE in interacting QFTs, we consider an analytic continuation of the replica trick, which we call the \textbf{orbifold method}\index{orbifold method}. This was first discussed for free fields in~\cite{Nishioka_2007,He:2014gva}. In addition, this method is also important for the calculation of EE in string theory~\cite{Dabholkar:1994ai,He:2014gva,Witten:2018xfj,Dabholkar:2022mxo}.
Unlike other methods, the orbifold enables us to calculate EE within the ordinary flat-space QFTs but with identifications. In the following, we review the method and how it works for free QFTs.

In the orbifold method, we analytically continue $n$ to $1/M$ with an integer $M$ to obtain the theory on the \textbf{orbifold}\index{orbifold} $\mathbb{R}^2/\mathbb{Z}_M\times \mathbb{R}^{d-1}$ instead of a cone. \eqref{e:EE_n} is then rewritten in terms of the free energy $F^{(M)}=F_{1/n}$ as\index{entanglement entropy}
\aln{
	S_A=-\frac{\partial \left(M F^{(M)}\right) }{\partial M}\bigg|_{M\to 1},
	\label{eq:EE_M}
}
provided that $M\in\mathbb{Z}_{>1}$ can be analytically continued to 1.
A state on the orbifold can be obtained by acting the $\mathbb{Z}_M$ projection operator \cite{Gersdorff_2008}, 
\[
\hat{P}=\sum_{m=0}^{M-1}\frac{ \hat{g}^{\, m} }{M},
\]
on a state in an ordinary  flat plane, 
where $\hat{g}$ is a $2\pi/M$ rotation operator {around the origin},
\aln{
\hat{g}\left|x,\bar{x},x_\parallel\right>=\left|e^{2\pi i/M}x,e^{-2\pi i/M}\bar{x},x_\parallel\right>.
} 
In the following discussion, we will call this $\mathbb{Z}_M$ rotation 
 $\hat{g}^m$ as an  $m$-\textit{twist}, where $ m \in \mathbb{Z}\mod M$.\index{twist} 
%Here 
$\bm{x}=(x,\bar{x})$ are complex coordinates for the perpendicular directions, $x=x_\perp+i\tau,\, \bar{x}=x_\perp-i\tau$. 
%Hereafter they are denoted by $\bm{x}$. 
%In order to perform the summation over $m$ from $1$ to $M-1$, %in a way such that $M$ is analytically continued to $\mathbb{R}$, 
%it is essential to turn it into a complex integral, picking up the residues $\{e^{2\pi i m/M}\}_{m=1, \cdots, M-1}$ \cite{doi:10.1063/1.531345}. By this treatment, the analytic continuation is unique due to the Carlson's theorem \cite{Casini_2009,Witten:2018xfj}.

Let us first consider a free real scalar field theory on the $\mathbb{Z}_M$ orbifold without a nonminimal coupling
to the curvature. Since scalar fields have no spin and are singlet under the spatial rotation, the $\mathbb{Z}_M$ action $\hat{g}$ on the internal space of the fields is trivial. Its explicit action is given as follows: 
	\eqn{
		\hat{g} [\phi (x)]=\phi (\hat{g} x)=\phi (\hat{g}\bm{x},x_\parallel)=\phi( e^{2\pi i/M} x_\perp, e^{-2\pi i/M} \bar{x}_\perp, x_\parallel).
	}
By the abuse of notation, both the $\mathbb{Z}_M$ action on the coordinates and fields are represented by $\hat{g}$.

%\medskip
%By using the orbifold method, EE can be easily calculated for free  theories \cite{Nishioka_2007}. 
%In the case of a real scalar field theory, 
The free energy on the $\mathbb{Z}_M$ orbifold takes the following form,
\aln{
F_\text{free}^{(M)}&=\frac{1}{2}\mathrm{Tr}\,\mathrm{log}[\hat{P}(-\Box+m^2_{0})]\\
&=\frac{1}{2}\int\frac{d^2\bm{k}}{(2\pi)^2}\frac{d^{d-1}k_\parallel}{(2\pi)^{d-1}}\mathrm{log}(k^2+m_{0}^2)\left<\bm{k},k_\parallel\right|\hat{P}\left|\bm{k},k_\parallel\right>. 
\label{eq:free-energy-free-scalar}
}
The trace of $\hat{P}$ is computed as follows.
The diagonal matrix element of $\hat{g}^m$ is given by
\aln{
		\ev{\hat{g}^m}{\bm{k},k_\parallel}&=\bra{{k}_\parallel}\ket{{k}_\parallel}\ev{\hat{g}^m}{k,\bar{k}} %\nonumber\\
	%	&=V_{d-1}\bra{\omega^{-m}k,\omega^m \bar{k}}\ket{ k,\bar{k}}
	%	\nonumber\\
         =(2\pi)^2V_{d-1}\delta(\omega^mk-k)\delta(\omega^{-m}\bar{k}-\bar{k}) ,
   %      \nonumber
     }
where $V_{d-1}=(2\pi)^{d-1} \delta^{d-1}(k_\parallel-k_\parallel)$ is a transverse $(d-1)$-dimensional volume 
 and $\omega=e^{2\pi i/M}$.
For $m\neq 0$, this becomes
 \aln{
		\ev{\hat{g}^m}{\bm{k},k_\parallel}=(2\pi)^2 V_{d-1} \delta^{2}(\bm{k}) \frac{1}{\omega^m-1}\frac{1}{\omega^{-m}-1}.
		\label{eq:twist-k}}
For $m\neq0$, it is proportional to $V_{d-1}$, the area of the boundary $\partial A$. On the other hand, for
$m=0$, it is proportional to 
$V_{d-1}\times\delta^2(\bm{0})\propto V_{d+1}$, the volume of the 
whole region of the path integral. % $A$
From this, we  see that  twisting a propagator with $m \neq 0$ 
constrains the normal components of the momentum zero, ${\bm k}=0$.
%In Appendix \ref{app:twist}, we will give an alternative interpretation in the position space to the twisted propagator. 

The summation over $m$ from $1$ to $M-1$ can be performed as follows.
Given a holomorphic function $f(z)$, its summation is given by
\begin{equation}
    \sum_{m=1}^{M-1} f \left(\omega^{m}\right)=\oint_C
\frac{\dd{z}}{2\pi i} p(z) f(z),
    \label{eq:complex-int}
\end{equation}
where
\eqn{p(z)\equiv \frac{Mz^{M-1}}{z^M -1}-\frac{1}{z-1}}
has simple poles at $\omega^m =e^{2\pi m  i/M} $ with $m=1, \cdots, M-1$. The integration contour $C$ is chosen to surround these poles \cite{doi:10.1063/1.531345} (Fig.\ref{fig:contour}).
\begin{figure}
	\centering
	\includegraphics[width=7cm,clip]{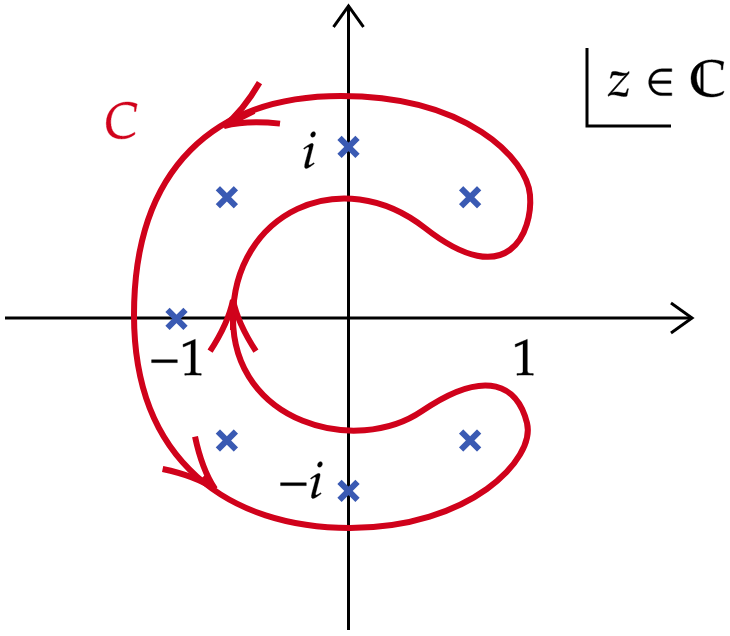}
	\caption{The contour $C$ (red curve) and simple poles of $p(z)$ (blue cross marks) in \eqref{eq:complex-int} ($M=8$, as an example).}
	\label{fig:contour}
\end{figure}
Then the summation of \eqref{eq:twist-k} is calculated as
\al{
	\sum_{m=1}^{M-1} \frac{1}{\omega^m-1}\frac{1}{\omega^{-m}-1}
	 %&    =\oint_C \frac{\dd{z}}{2\pi i} p(z) \frac{1}{z-1} \frac{1}{z^{-1}-1}\nonumber\\
	=\oint_{C_1}\frac{\dd{z}}{2\pi i} \frac{z p(z)}{(z-1)^2} 
%	&=\left.\pdv{z}\left(z p(z)\right)\right\vert_{z=1}\nonumber\\
	=\frac{M^2-1}{12},
	\label{eq:twist-sum}}
where $C_1$ is a counter-clockwise circle around $z=1$. 
It is written in a compact form,
\aln{
\sum_{m=1}^{M-1}\frac{1}{\sin^2\left(\frac{\pi m}{M}\right)}=\frac{M^2-1}{3}.
\label{e:summ}
}

Plugging \eqref{eq:twist-sum} into \eqref{eq:free-energy-free-scalar}, we obtain
%Thus, we have
\aln{
F_{\text{free}}^{(M)} 
&= \frac{1 }{2 M} \int \frac{d^2 \bm{k} \ d^{d-1}k_{\parallel}} {(2 \pi)^{d-1}}   \log (k ^2 + m^2_{0}) \left( \frac{V_{d+1}}{(2\pi)^2}+  V_{d-1}\frac{M^2-1}{12}  \delta^2(\bm{k}) \right) .
\label{eq:1-loop-F}
}
The first term  proportional to $V_{d+1}$ vanishes in the calculation of EE in \eqref{eq:EE_M}. 
We will see later that such property generally holds even in presence of interactions. 
On the other hand, %we have 
the second term is proportional to the area $V_{d-1}$ survives in \eqref{eq:EE_M}
due to an additional $M$-dependence of  $M^2 -1$. 
Consequently, 
EE for a free scalar theory is given by
\aln{
S_{\text{1-loop}}=-\frac{V_{d-1}}{12}\int^{1/\epsilon}\frac{d^{d-1}k_\parallel}{(2\pi)^{d-1}}\log\left[\bigl(\tilde{G}_0^\text{bdry}(k_\parallel)\bigr)^{-1}\epsilon^{2}\right].
\label{EE-1loop}
} 
%which obeys the area law of EE. %Here we have introduced a UV cutoff scale $\epsilon$ is introduced. 
Here we have introduced a UV cutoff given by the lattice spacing $\epsilon$. In the following, we will sometimes omit $1/\epsilon$ in the integral for simplicity. %, which is naturally identified as a lattice spacing for a lattice system, 
$\tilde{G}_0^{\text{bdry}}$ is the momentum space representation of the boundary propagator,
\[
\bigr(\tilde{G}_0^\text{bdry}(k_\parallel)\bigr)^{-1}=k_\parallel^2+m^2_{0}.
\]
This gives an effective squared mass on the two-dimensional plane with nonzero transverse momentum $k_\parallel$.
Note that the EE decreases as the mass increases. \eqref{EE-1loop} is consistent with the CFT result \eqref{eq:EE-massive} when $m=0$ (Gaussian fixed point).

The calculation for the scalar theory can be easily generalized to bosonic and fermionic higher spin theories \cite{He:2014gva}. 
In the bosonic case, a state is parametrized by $|{\bm x}, x_\parallel; s \rangle$, where $s$ is a spin of $SO(2)\subset SO(d+1)$ rotation around the $x_\parallel$ axis. 
Then the action of the two-dimensional rotation $\hat{g}$ is given by
\begin{equation}
	\hat{g}\left|x,\bar{x},x_\parallel;s \right>= e^{2\pi s i/M }
	\left|e^{2\pi i/M}x,e^{-2\pi i/M}\bar{x},x_\parallel; s \right>
	\label{eq:spin-twist}
\end{equation}
and the sum over $m \neq 0$ in \eqref{e:summ}
 is replaced by\footnote{When we sum over the internal space, the positive and negative spins appear in pairs. In total, the contribution from each spin is given by \eqref{e:summ-highers}. Details are explained in Appendix \ref{app:spinor}.} 
\aln{
\sum_{m=1}^{M-1}\frac{ \cos (\frac{2\pi ms}{M} ) }{\sin^2\left(\frac{\pi m}{M}\right)}=
\frac{1}{3} \left[M^2-1+6M^2\left( \left\{ s/M\right\}^2 -  \left\{ s/M\right\}\right) \right],
\label{e:summ-highers}
} 
where $\{ x \}$ is a fractional part of $x$. 
For fermionic generalizations, we need special care since $2\pi$ rotation
gives an extra minus sign, $\hat{g}^M=-1$, and it cannot be regarded as $\mathbb{Z}_M$ orbifold. 
To overcome this difficulty, the authors in  \cite{He:2014gva} take  an odd $M$ and consider 
$\hat{g}^2$ as the generator of $\mathbb{Z}_M$ orbifold on a double cover of the Riemann surface. Then, this works both for bosonic and fermionic higher spin fields. 
\eqref{eq:free-energy-free-scalar} with the internal basis for spins is\footnote{For gauge fields, we consider the free energy after fixing a gauge. See Appendix \ref{app:spin} for an example.}
\aln{
	F_\text{free}^{(M)}&=(-1)^F \frac{1}{2}\mathrm{Tr}\,\mathrm{log}[\hat{P}_{odd}(-\Box+m^2_{0})]\\
	&=\frac{1}{2}\sum_{\text{internal dof}} \int\frac{d^2\bm{k}}{(2\pi)^2}\frac{d^{d-1}k_\parallel}{(2\pi)^{d-1}}\mathrm{log}(k^2+m_{0}^2)\left<\bm{k},k_\parallel; s\right|\hat{P}_{odd}\left|\bm{k},k_\parallel; s\right>
}
with
\begin{equation}
	\hat{P}_{odd}= \sum_{m=1}^{M-1}\frac{\hat{g}^{2m}}{M},
\end{equation}
assuming the odd $M$. The summation over the internal degrees of freedom means we sum each spin contribution with weights counting its number of components. See Appendix \ref{app:spinor} and~\cite{He:2014gva} for details.

Another subtlety  in higher spin generalizations
in  analytical continuation of $M$ since \eqref{e:summ-highers} contains a
non-analytic function, $\left\{ s/M\right\}$, and we need to constrain the value of $s$ within $[-M, M]$ for fermions
or $[0,2M]$ for bosons. Thus the calculation of EE for higher spins than 3/2 may have subtlety in the orbifold method. 
For more details, see  \cite{He:2014gva}. 

\section{Orbifold method for interacting QFTs}\label{sec:orb-int}

\subsection{Interacting QFTs on an orbifold}\label{sec:int-orb}
In the previous section, we reviewed the one-loop calculation of EE by~\cite{Nishioka_2007,He:2014gva}. This method basically works for free fields but not for interacting cases. In~\cite{Iso:2021vrk,Iso:2021rop,Iso:2021dlj}, we extended the orbifold method\index{orbifold method} in a way applicable to interacting cases. To see this, we first discuss interacting QFTs on the $\mathbb{Z}_M$ orbifold in general. For simplicity, we focus on a scalar field theory, however, the higher spin generalization will be discussed in Section \ref{sec:spin-general}.

Consider, for a moment, a $\phi^4$ scalar field theory on  a $\mathbb{Z}_M$ orbifold without a non-minimal coupling
to the curvature.   
The action for a scalar field theory on $\mathbb{Z}_M$ orbifold\index{scalar field theory on $\mathbb{Z}_M$ orbifold} is given by 
\aln{
	I=	\int_{\mathbb{R}^2} \frac{d^2x}{M} \int_{\mathbb{R}^{d-1}} d^{d-1}x_\parallel \left[\frac{1}{2}
	%\phi \hat{P} 
	(\hat{P} \phi) \left(-\Box+m^2_{0} \right)
	(\hat{P}\phi )
	+V(\hat{P}\phi)\right]
	\label{eq:action} 
}
in terms of a field $\phi(x)$ in flat space $\mathbb{R}^2\times\mathbb{R}^{d-1}$ %and 
but with the projection operator,
where
\aln{
	\hat{P}\phi(\bm{x},x_\parallel):=\frac{1}{M}\sum_{m=0}^{M-1}\phi(\hat{g}^m\bm{x},x_\parallel).
}%--------------------------
In the following, we consider the $\phi^4$ theory, $V(\phi)=\frac{\lambda}{4}(\phi)^4$, for simplicity.
From the action \eqref{eq:action}, 
the inverse propagator of the orbifold theory in flat space is given by
\aln{
	\hat{G}_0^{-1\, (M)}= \frac{\hat{P}\hat{G}_0^{-1}\hat{P}}{M} =  \frac{\hat{P} \left(-\Box+m^2_{0} \right)   \hat{P}}{M} ,
	\label{propagator-orbifold}
}
%Using the relation
%\[
%\hat{G}_0^{-1\, (M)} \frac{1}{M} \hat{G}_0^{(M)}=\hat{G}_0^{(M)} \frac{1}{M} \hat{G}_0^{-1\, (M)}=\hat{P},
%\]
and  the propagator, which satisfies $\hat{G}_0^{-1\, (M)} \hat{G}_{0}^{(M)}=\hat{P} $, is then written as
\aln{
	G_{0}^{(M)} (x,y) &= M\bra{x}(\hat{P} \hat{G}_0 \hat{P})\ket{y} 
	=\sum_{m=0}^{M-1} G_0(\hat{g}^m x,y) ,
	% \nonumber \\
	%&=  \sum_{n=0}^{M-1} 
	%\nonumber
	\label{eq:green}
}
where 
\aln{ 
	G_0(\hat{g}^m x, y) %= \bra{\hat{g}^n x} G_0 \ket{y} 
	&=  \int\frac{d^{d+1} p}{(2\pi)^{d+1}}\frac{e^{i p \cdot(\hat{g}^m x-y)}}{p^2+m^2_{0}} %\nonumber\\&
	= \int\frac{d^2\bm{p}}{(2\pi)^2}\frac{d^{d+1} p_\parallel}{(2\pi)^{d+1}}\frac{e^{i \bm{p} \cdot(\hat{g}^m \bm{x}-\bm{y})+ip_\parallel\cdot(x_\parallel-y_\parallel)}}{p^2+m^2_{0}}.
}
The $\mathbb{Z}_M$ rotation on $y$ 
has been eliminated since a projection operator $\hat{P}$ commutes with $\hat{G}_0$.
From the identity ${\bm p}\cdot \hat{g}^m {\bm x}=\hat{g}^{-m} {\bm p} \cdot {\bm x}$, we see that
the flow-in momentum from the propagator at a vertex $x$ is given by the twisted momentum\index{twisted momentum} $(\hat{g}^{-m} {\bm p}, \ p_\parallel)$. 
In the momentum space representation, the propagator is written as 
\aln{
	\left<\bm{p},p_\parallel\right|G^{(M)}\left|\bm{q},q_\parallel\right>=\sum_{m=0}^{M-1}\frac{1}{p^2+m^2_{0}}(2\pi)^{d+1}\delta^2(\hat{g}^m\bm{p}-\bm{q})\delta^{d-1}(p_\parallel-q_\parallel)
}
with $-m$ redefined as $m$. In the position space, the twisted propagator is interpreted as a ``pinned'' correlator across the entangling surface $\partial A$. In particular, if there is only a single twist, we can make a following replacement:
\begin{equation}
    G_0(\hat{g}^mx-y)\rightarrow\frac{1}{4\sin^2\frac{m\pi}{M}}\int\frac{d^{d-1}k_\parallel}{(2\pi)^{d-1}}
    \tilde{G}_0^\text{bdry}(k_\parallel) {e^{ik_\parallel\cdot r_\parallel}}
    \delta^2(\bm{X})
    \equiv\frac{1}{4\sin^2\frac{m\pi}{M}}G_0^{\text{bdry}}(r_\parallel)\delta^2(\bm{X}).
    \label{e:efftwistprop-body}
\end{equation}
See Appendix \ref{app:twist} for details. The interpretation in the position space is in a reasonable agreement with the intuitive picture of entanglement; a quantum correlation between $A$ and $\bar{A}$.

%On the other hand, 
The interaction vertex in the Euclidean $\phi^4$ theory 
read off from the action \eqref{eq:action} is
\begin{equation}
    -3!\frac{\lambda}{M}.
	\label{eq:vertex-M}
\end{equation}
%This is 
The $x$ integration gives 
the ordinary momentum conserving delta functions
\aln{ \delta^2 \left(\sum_i \hat{g}^{m_i} \bm{p}_i \right)\ \delta^{d-1} \left(\sum_i p_{i\parallel} \right)
	\label{momentum-conserving-deltafuntion}
}
with twisted flow-in momenta.
%Note that the projection operators 
%in the interacting term in the action Eq.(\ref{eq:action} )
%can be absorbed into those in the propagators by using the relation $\hat{P}^2=\hat{P}$.

%%%%%%%%%%%%%%%%%%%%%%%%%%%%%%%%%%%%%%%%%%%%%%%%
\subsection{Redundancies in twisted diagrams}
Due to the $\mathbb{Z}_M$ symmetry, not all twists in a Feynman diagram are independent. The twisted part of \eqref{momentum-conserving-deltafuntion}
\begin{align}
    \delta^2 \left(\sum_i \hat{g}^{m_i} \bm{p}_i \right) 
    &= \delta^2 \left(\hat{g}^{m_{i_0}}\bm{p}_{i_0} + \sum_{i\neq i_0} \hat{g}^{m_i} \bm{p}_i \right)\\
    &= \delta^2 \left(\bm{p}_{i_0} + \sum_{i\neq i_0} \hat{g}^{m_i-m_{i_0}} \bm{p}_i \right).
\end{align}
Note that a twist $\hat{g}^n$ is just a phase factor. If there is no $\hat{g}^{i_0}$ factor in the diagram other than the delta function, then we can redefine the twist $m_i-m_{i_0}\rightarrow m_i$ for $i\neq i_0$; the $i_0$-twist is eliminated. The summation over $i_0$ gives an overall factor $M$. Fig.\ref{fig:phi43loops} is an example of such diagrams. Note that $\tilde{G}(\hat{g}^m \bm{p}, p_\parallel)=\tilde{G}(\bm{p}, p_\parallel)$ since the propagator only depends on the squared momentum. We can also interpret Fig.\ref{fig:phi43loops} as a kind of the one-loop diagram (Fig.\ref{fig:phi43loops-2}) and the only independent twist is $m_1+m_2+m_3-m_4$. It indicates a one-loop diagram has only one independent twist. 

After all, we can eliminate the redundant twists by using the invariance of the vertices under $\mathbb{Z}_M$ rotations. 
This procedure reminds us of gauge fixing in an ordinary gauge theory. In the next subsection, we will show that this analogy works well in the investigation and that we can extract independent twists systematically. We call this methodology \textbf{$\mathbb{Z}_M$ (lattice-like) gauge theory on Feynman diagrams}\index{$\mathbb{Z}_M$ gauge theory on Feynman diagrams}\index{$\mathbb{Z}_M$ lattice-like gauge theory on Feynman diagrams}. 

%%%%%%%%%%%%%%%%%%%%%%%%%%%%%%%%%%
\begin{figure}[t]
	\centering
	\subfloat{\includegraphics[width=7cm,clip]{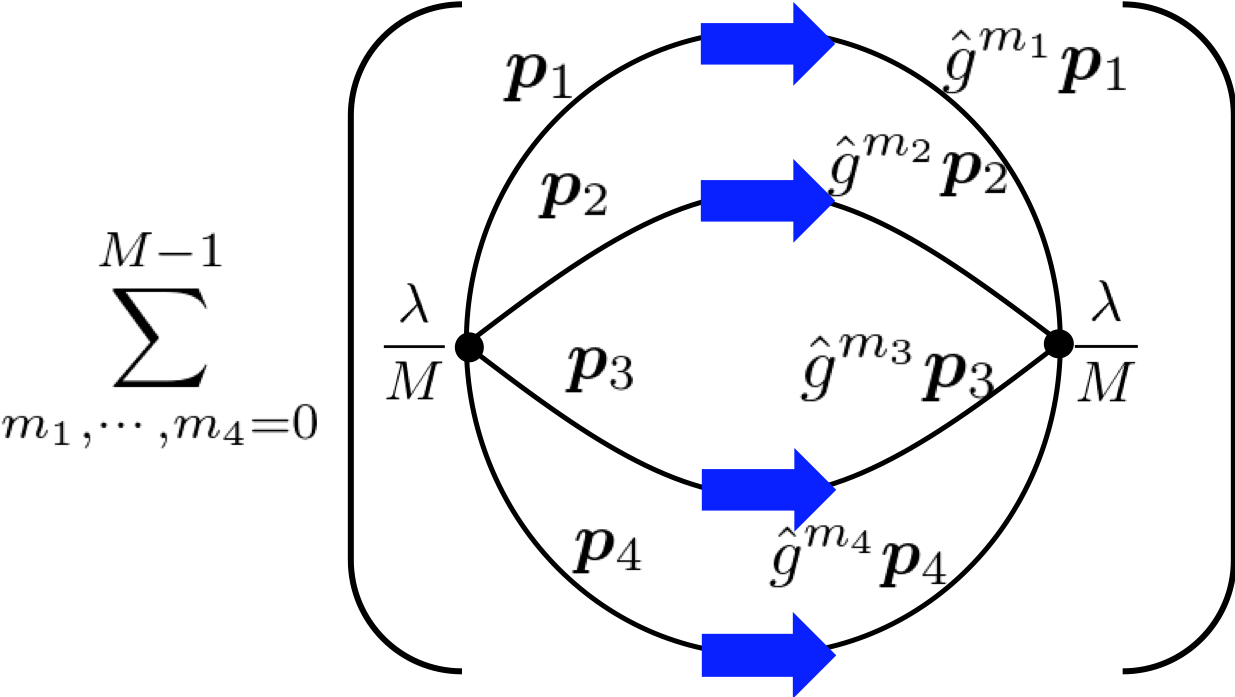}\label{fig:bubble1}}
	\hspace{1cm}
	\subfloat{\includegraphics[width=7cm,clip]{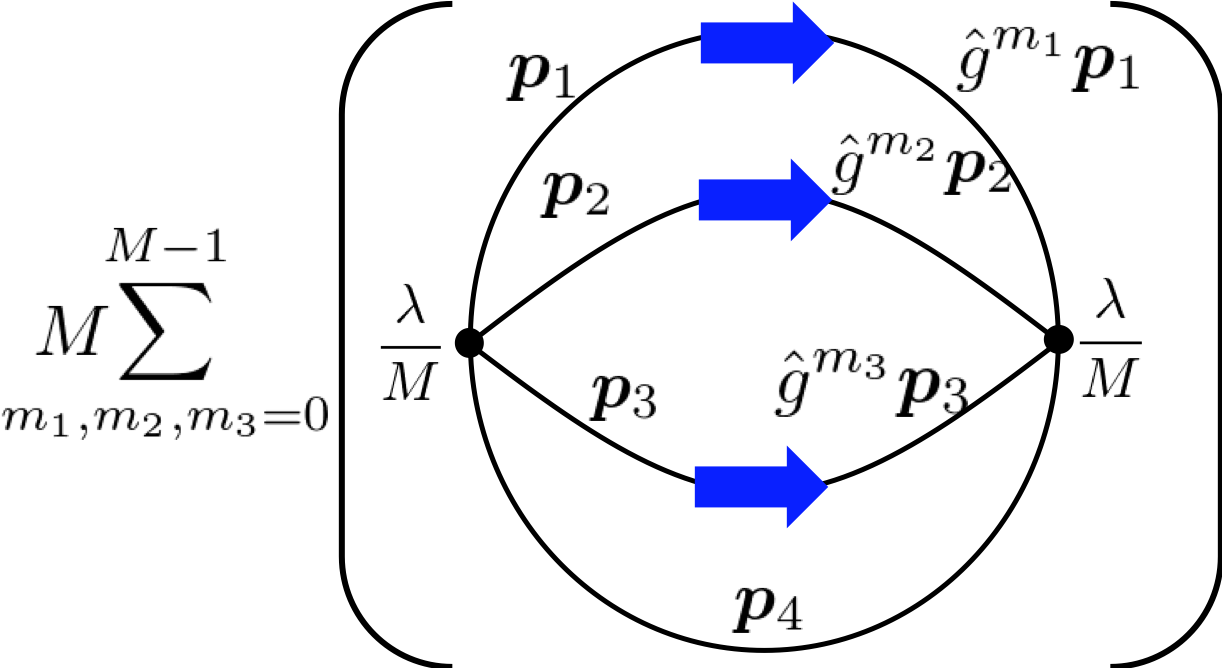}\label{fig:bubble2}}
	\caption{Two equivalent configurations of twists in the three-loop diagram. 
	%Red dashed lines denote the propagators with the momentum twisted. 
	Blue arrows denote twisted momenta with a twist $m_i$. 
The bottom propagator in the right is made untwsited by a $\mathbb{Z}_M$ rotation at a vertex.}
	\label{fig:phi43loops}
\end{figure}
%%%%%%%%%%%%%%%%%%%%%%%%%%%%%%%%%%
%%%%%%%%%%%%%%%%%%%%%%%%%%%%%%%%%%
\begin{figure}[t]
	\centering
	\includegraphics[width=9cm,clip]{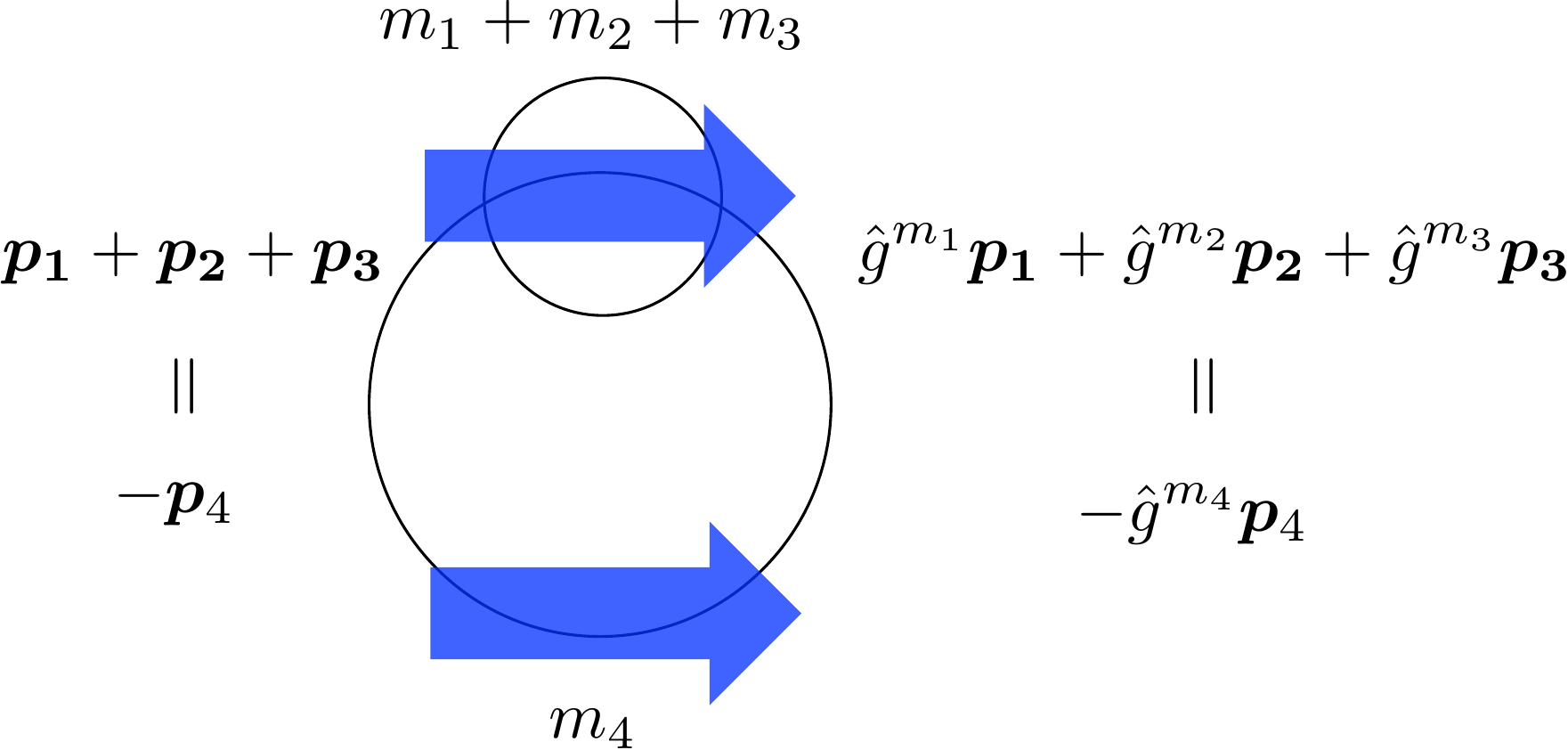}
	\caption{A diagram equivalent to Fig.\ref{fig:phi43loops}, where the three twisted propagators and two vertices are regarded as a single twisted object with $m_1+m_2+m_3$. The redundancy indicates there is only a single independent twist $m_1+m_2+m_3-m_4$ for this `one-loop' diagram.}
	\label{fig:phi43loops-2}
\end{figure}
%%%%%%%%%%%%%%%%%%%%%%%%%%%%%%%%%%

%This type of redundancies leads us to regard twisted diagrams as the \textit{$\mathbb{Z}_M$ lattice-like gauge theory on Feynman diagrams}.

\subsection{$\mathbb{Z}_M$ lattice-like gauge theory on Feynman diagrams}\label{s:ZM-gauge}

In this section, we interpret the orbifold method in interacting field theories as a gauge theory to calculate the free energy
of the $\mathbb{Z}_M$ orbifold.  
Namely,  assign  $\mathbb{Z}_M$ twists on  each link (i.e., on a propagator) and  
define  $\mathbb{Z}_M$ gauge transformations  on each vertex, and take a summation
over all the twists modulo $\mathbb{Z}_M$ gauge transformations. 
Then, a gauge-invariant  configuration of twists  is characterized 
by a set of fluxes of twists on each plaquette of each Feynman diagram. 
%Within this framework, we can easily prove the area law of  EE.
This framework is not only telling us a correct prescription to compute calculable ($\sim$dominant) contributions to EE, but also useful in proving the area law to all orders, which has not been done for general interacting QFTs before our work.

%%%%%%%%%%%%%%%%%%%%%%%%%%%%%%%%%%%%%%%%
\begin{figure}[h]
	\begin{tabular}{c}%prevent line break
		\hspace*{0.15\linewidth}
		\begin{minipage}{0.3\hsize}%align figs horizontally
			\centering
			\includegraphics[width=\linewidth]{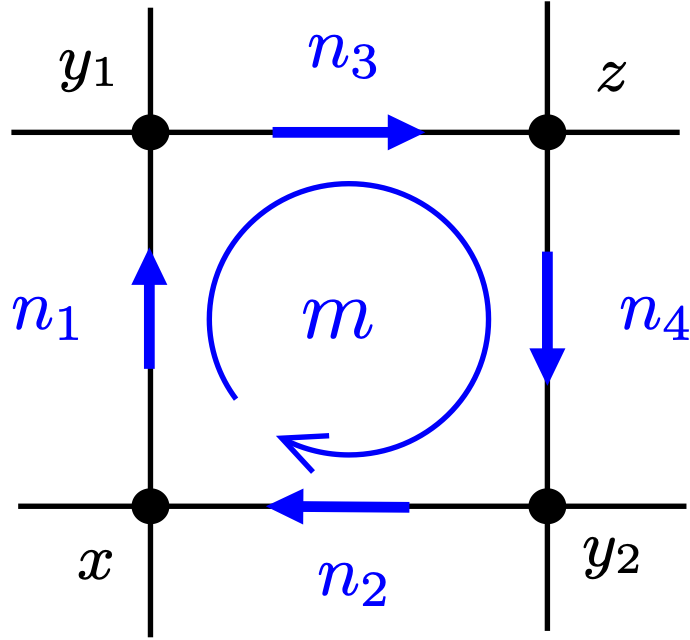}
		\end{minipage}
	\hspace*{0.1\linewidth}
		\begin{minipage}{0.35\hsize}
			\centering
			\includegraphics[width=\linewidth]{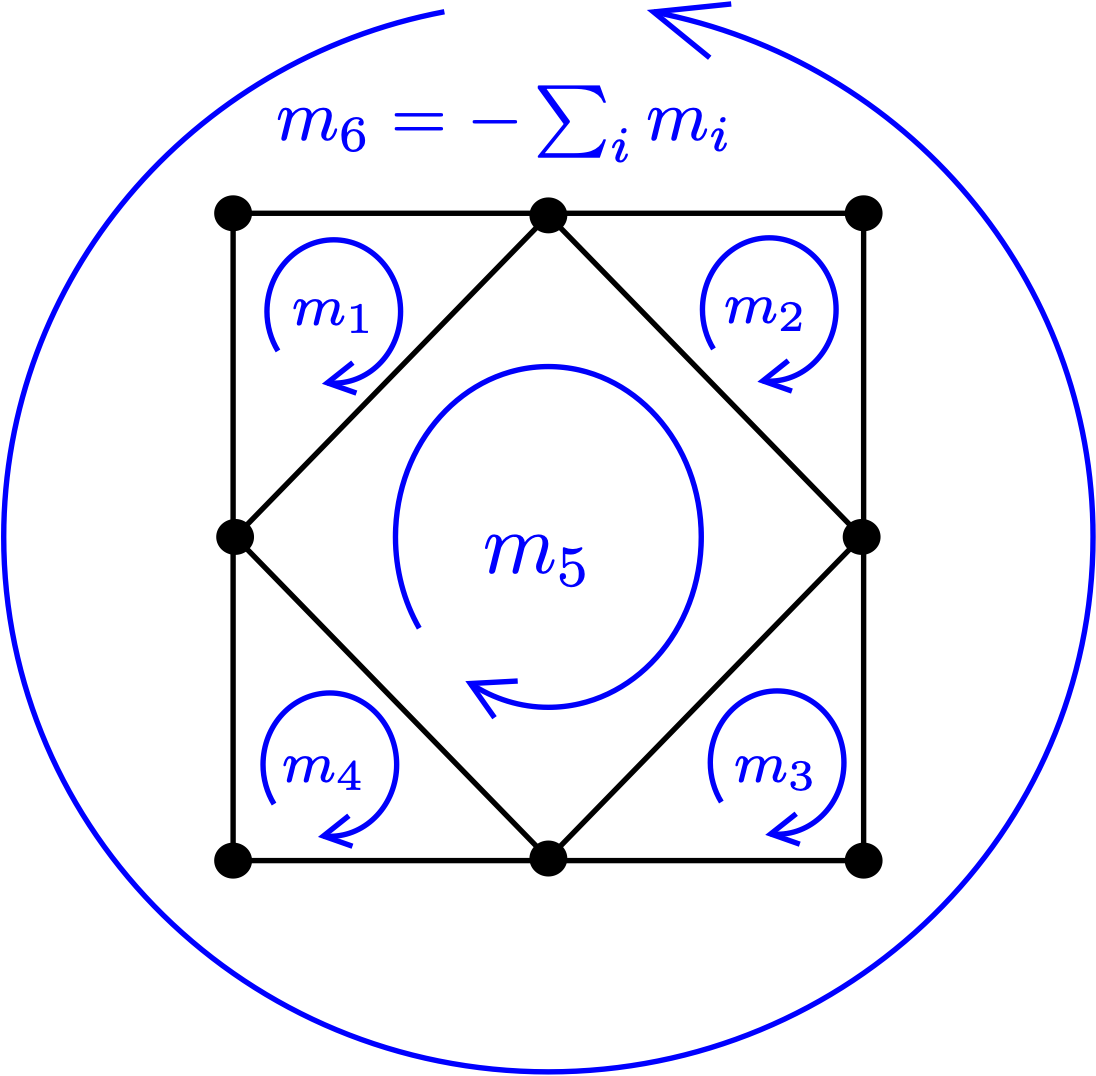}
		\end{minipage}
	\end{tabular}
\caption{
$\mathbb{Z}_M$ gauge theory on Feynman diagrams: $\{n_i\}$ are twists on links (propagators), and $m=\sum_i n_i \mod M$ is 
a flux of twists around the plaquette and invariant under  $\mathbb{Z}_M$ gauge transformations on vertices.
The right figure is a set of $\mathbb{Z}_M$ invariant fluxes of twists on plaquettes.}
\label{Fig1}
\end{figure}
%%%%%%%%%%%%%%%%%%%%%%%%%%%%%%%%%%%%%%%%%%%
On a $\mathbb{Z}_M$ orbifold, each propagator in a Feynman diagram is twisted as in \eqref{eq:green}. 
The notion of $\mathbb{Z}_M$ gauge symmetry appears since we can rotate away 
some of the twists of propagators by the $\mathbb{Z}_M$ gauge transformations on vertices of the Feynman diagram.
In general, a rotation of the coordinates at the vertex $x$ in the left figure of Fig.\ref{Fig1} by $2\pi l/M$ shifts $n_1$ by $l$, and $n_2$ by $-l$; 
therefore, {the sum} of twists around a plaquette\index{plaquette} $m=\sum_{i=1}^4 n_i \mod M$, which we call a \textit{flux}\index{flux}, is invariant under $\mathbb{Z}_M$ rotations at vertices.
In other words, we can classify
independent configurations of twists
up to  $\mathbb{Z}_M$ transformations in terms of 
$\mathbb{Z}_M$ fluxes of twists on plaquettes. 
For instance, for a Feynman diagram on the right of Fig.\ref{Fig1}, a set of $L (=5)$ fluxes of twists on each plaquette specifies an independent configuration of twists in a $\mathbb{Z}_M$-invariant way. Since the independent configurations are characterized by twist numbers of plaquettes, the number of independent twists coincides with the number of loops $L$. We can assign the complement twist number to the outer circle ($m_6=- \sum_{i=1}^5 m_i$ in the right figure of Fig.\ref{Fig1}). Its minus sign comes from the opposite direction of the twist when the diagram is put on a sphere. As the twist of the outer loop is not an independent one, we will omit writing it in the following.

In the language of gauge theory, twists correspond to gauge fields. Vertices in a Feynman diagram are topologically connected ``sites'' on a lattice.  
Since a twist on a propagator is defined between the two vertices, it can be regarded as a link variable associated with the relative phase of the vertices. Then, the %invariance
$\mathbb{Z}_M$ rotation on each vertex is interpreted as a local change of the phase. It is nothing but a gauge transformation, 
but the angle  is restricted to $2\pi m/M$ with $m=0,\cdots,M-1$. 
As a result, it is understood as the \textbf{$\mathbb{Z}_M$ (lattice-like) gauge theory on Feynman diagrams}\index{$\mathbb{Z}_M$ gauge theory on Feynman diagrams}\index{$\mathbb{Z}_M$ lattice-like gauge theory on Feynman diagrams}. 
A flux in a plaquette is invariant under $\mathbb{Z}_M$ rotations at vertices and characterizes  distinct configurations. It is a counterpart of the Wilson loop, a gauge-invariant object in a gauge theory. 
%A difference is that the internal space in the present case is phases of the coordinates or momenta. 
%Accordingly, a twist of a plaquette is also regarded as a twist of the corresponding loop momentum. 
A flux of twists is defined as a sum of the twists  of propagators in a counterclockwise direction along a plaquette. 
The flux is defined modulo $M$; i.e. $-m$ flux is equivalent to $M-m$ flux.

%Appendix \ref{appenarea} for area law.

\section{Perturbative calculation of entanglement entropy}\label{sec:pert-EE}
The procedure %to calculate a contribution to 
to calculate EE  is straightforward in principle. 
(1) Perform momentum integrations of each bubble Feynman diagram 
with a fixed configuration of twists, i.e. fluxes $\{m_l\}$ for plaquettes. (2) Sum up them over all the twist configurations  of $m$'s to obtain a bubble diagram on the $\mathbb{Z}_M$ orbifold.
(3) Sum all the connected bubble diagrams to obtain the free energy on the $\mathbb{Z}_M$ orbifold $F^{(M)}$. (4) From the $M$-dependence of the free energy, one can obtain EE from \eqref{eq:EE_M}.

Among various flux configurations, we have a trivial bubble with no twist $m_1=m_2=\cdots=0$. The free energy without any twists is nothing but that in flat space $F^{(1)}$, which is proportional to the volume $(2\pi)^{d+1}\delta^{d+1}(0)=V_{d+1}$. The $M$-dependence in an arbitrary bubble diagram is given by
\begin{equation}
    \left(\frac{1}{M}\right)^{N_V} M^{N_P-L}=\frac{1}{M}
\end{equation}
in terms of the number of vertices $N_V$, propagators (edges) $N_P$, and loops $L$. 
The first factor comes from vertices \eqref{eq:vertex-M}. The second factor comes from the summation over the redundant twists. The equality follows from the Euler's polyhedral formula. Remarkably, since the free energy of the zero twist configuration is proportional to $1/M$, it does not contribute to EE \eqref{eq:EE_M}! In fact, this identity plays an important role in proving the \textbf{area law}\index{area law} of EE. In Appendix \ref{appenarea}, we proved the area law of EE, R\'{e}nyi entropy\index{R\'{e}nyi entropy}, and capacity of entanglement~\cite{PhysRevLett.105.080501,deBoer:2018mzv}\index{capacity of entanglement} of half space in arbitrary QFTs with local interactions to all orders.

Since the configuration of trivial twists does not contribute to EE, we are interested in configurations where at least one of the twists are nonzero. In the following sections, we denote a part in the free energy contributing to EE by $\tilde{F}^{(M)}$, i.e.
\begin{equation}
    F^{(M)}=\frac{F^{(1)}}{M}+\tilde{F}^{(M)}.
\end{equation}
The evaluation of  Feynman diagrams with nonvanishing twists is in general very involved due to multiple discrete summations (or contour integrals). 
Thus our strategy is,  instead of considering general configurations of twists, 
to focus on dominant contributions to EE.
In particular, we consider two specific types of configurations, contributions from a single {\it twisted propagator} or {\it twisted vertex}, and the all-order resummation of them. %In terms of fluxes, these diagrams contain multiple fluxes.
We will later discuss in Section \ref{sec:wilson-rg} why they will give dominant contributions to EE
and how  the rest of contributions are incorporated in the Wilsonian renormalization picture. %\footnote{Although in Sec. \ref{s:prop} and \ref{s:vert}, flux configurations that cannot be attributed to a single twist of either a propagator or a vertex remain uncalculated, we can address this issue via the Wilsonian renormalization group, which will be discussed in Sec. \ref{s:discussion}.}

\subsection{One-loop calculation revisited}

%%%%%%%%%%%%%%%%%%%%%%%%%%%%%%%%%%%%%%%%
\begin{figure}
	\centering
	\includegraphics[width=0.25\linewidth]{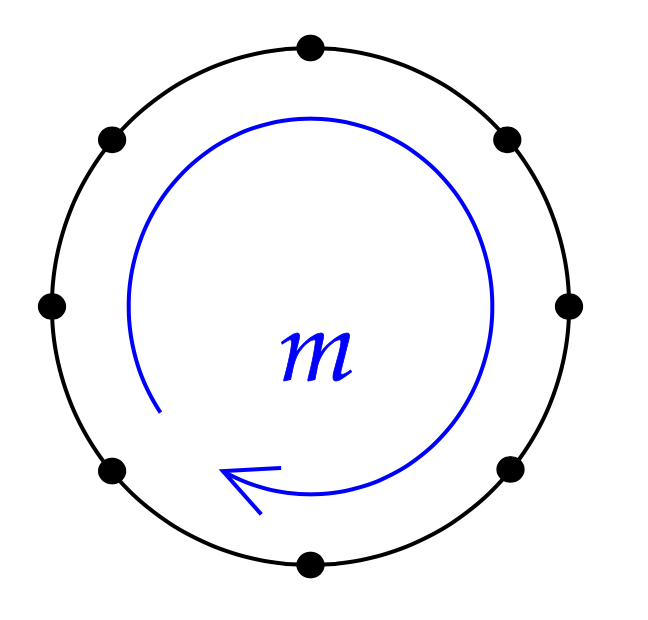}
	\caption{
		There is a single twist for a one-loop diagram.
		The number of independent twists is unchanged
		even if we divide  the propagator into multiple connecting pieces as drawn in the figure.
	}
	\label{Fig2}
\end{figure}
%%%%%%%%%%%%%%%%%%%%%%%%%%%%%%%%%%%%%%%%%%%

%%%%%%%%%%%%%%%%%%%%%%%%%%%%%%%%%%%
We already computed EE at the one-loop level in Section \ref{s:ZM}, however, let us reconsider it from the $\mathbb{Z}_M$ gauge theory perspective. A one-loop bubble diagram (Fig.\ref{Fig2}) is a simple example of the $\mathbb{Z}_M$ invariant configuration. The one-loop diagram has only one plaquette and the $\mathbb{Z}_M$ invariant twist is given by the flux {$m$}.
This property is essential for the later discussion to prove our main result, in which EE can be written as a sum of one-loop type diagrams of various composite operators.

In the perturbative approach, the one-loop bubble is given by the summation of the Feynman diagrams with {“two-point vertices,”}
\aln{
		{F}^{(M)}_\text{1-loop}
		= \frac{1}{2}\Tr \log [\hat{P}G_0^{-1}\hat{P}]
		=\frac{1}{2}\int_{{\epsilon^2}}\frac{ds}{s}\Tr e^{-sG_0^{(M)}/M}
		=\frac{1}{2}\int_{\epsilon^2}\frac{ds}{s}\sum_{n=0}^{\infty}\frac{(-s)^n}{n!}\Tr\left[\left(G_0^{(M)}\frac{1}{M}\right)^n\right].
\label{e:1loopexpand}
} 
%They are exceptional 
These diagrams are exceptional in the sense that they are composed of %consist of 
a {single} chain of the propagators {connected by the two-point vertices of $(1/M)$}.\footnote{%-----------------------
	{Reflecting the orbifold action in (\ref{eq:action}), the path integral measure is given by 
		\aln{
			\int \mathcal{D}(\delta\phi)\, e^{-\frac{1}{2M}\int d^{d+1}x(\delta\phi)^2}=1
		} 
		so that the 
		one-loop part of the free energy
		is given by $(1/2)\mathrm{Tr}\,\mathrm{ln}( \hat{P} G_0^{-1} \hat{P}) = (1/2)\mathrm{Tr}\,\mathrm{ln}(M(G_0^{(M)})^{-1})$
		rather than $(1/2)\mathrm{Tr}\,\mathrm{ln}((G_0^{(M)})^{-1})$. 
		See \eqref{propagator-orbifold}. 
		This is responsible for the coefficient $(1/M)$ of the “two-point vertex”  in \eqref{e:1loopexpand}. 
	}
} %--------------------
%(Fig.\ref{Fig2}). 
There is only one configuration with an $m$-flux  for the center plaquette. 
It can be regarded as a twist of a {\it single} propagator among
$n$ propagators in the expansion \eqref{e:1loopexpand}.
 %\footnote{%-------------------
%Note that they are in the same kind of configurations as Fig.\ref{f:proptwist} with the complementary twist assigned to the outer circle. 
%} %--------------------- 
From the viewpoint of operators, it stems from the idempotency of the projection: $\hat{P}^2=\hat{P}$. Thus, the correct, twisted one-loop free energy is
\begin{equation}
    {F}^{(M)}_\text{1-loop}
	= \frac{1}{2M} \int_{\epsilon^2}\frac{ds}{s}\sum_{n=0}^{\infty}\frac{(-s)^n}{n!}\mathrm{Tr}\,[G_0^{n-1}\,G_0^{(M)}].
	\label{eq:gauge-fix-1-loop}
\end{equation}
%Therefore, $G_0^{(M)}$'s in Eq.(\ref{e:1loopexpand}) can  be replaced with an untwisted propagator $G_0$ 
%except one at each diagram. 

It is very important to notice that \eqref{e:1loopexpand} is not equal to the naive free energy obtained by summing every single twisted propagator\index{twisted propagator} in the diagram,
\begin{equation}
\begin{split}
    F^{(naive)\, (M)}_\text{1-loop}&\equiv 
    \sum_{k=0}^{n-1}
	\frac{1}{2}\int_{\epsilon^2}\frac{ds}{s}\sum_{n=0}^{\infty}\frac{(-s)^n}{n!}\Tr\left[
	\left(G_0\frac{1}{M}\right)^{k}\left(G_0^{(M)}\frac{1}{M}\right)\left(G_0\frac{1}{M}\right)^{n-k-1}\right] \\
	&=\frac{1}{2M}\int_{\epsilon^2}\frac{ds}{s}\sum_{n=0}^{\infty}\frac{(-s)^n}{n!} n\Tr\left[
	G_0^{n-1}\, G_0^{(M)}\right] \\
	&\neq {F}^{(M)}_\text{1-loop}.
\end{split}
\label{eq:1-loop-overcount}
\end{equation}
This is an overcounting of the diagrams by a factor $n$ for a one-loop bubble with $n$ propagators. To obtain the correct result, we need to divide the naive twisted bubble with $n$ propagators by its multiplicity $n$. This way of calculation corresponds to performing summation first and then dividing by the gauge volume. This is in contrast to the calculation \eqref{eq:gauge-fix-1-loop}, where we fix a particular gauge. In this section, we discuss perturbative contributions by fixing a gauge, i.e. specify where to twist in the diagram. However, later we employ the ``divide-by-multiplicity'' method\index{divide-by-multiplicity method} to resum the perturbative contributions to all orders.

%Since twists do not affect transverse $(d-1)$-dimensional space, we write only the 2-dimensional 
%momenta (or coordinates) in the following, unless otherwise noted. 
%The free energy with twist $n$ is easily calculated \cite{Nishioka_2007} by noting 
%that $\bra{\bm{k}} \hat{g}^n \ket{\bm{k}}= (2\pi)^2 \delta^2 (\bm{k}) /4 \sin^2(n \pi /M)$ for $n \neq 0$. 

\subsection{Propagator contributions}\label{s:prop}
Among various configurations of  twists, we first focus on the configurations that a single propagator is twisted. 
Consider a configuration where  two plaquettes with a nonvanishing flux of twists 
share a  propagator 
and their fluxes are given by $m$ and $-m$ respectively.
For such a diagram, both of the fluxes  can be attributed to the 
$m$-twist of the shared propagator  (Fig.\ref{f:proptwist}) and we can 
interpret such a flux configuration as a twist of the propagator.
%%%%%%%%%%%%%%%%%%%%%%%%%%%%%%%%%%%%%%%%
\begin{figure}[t]
\centering
\includegraphics[width=9cm]{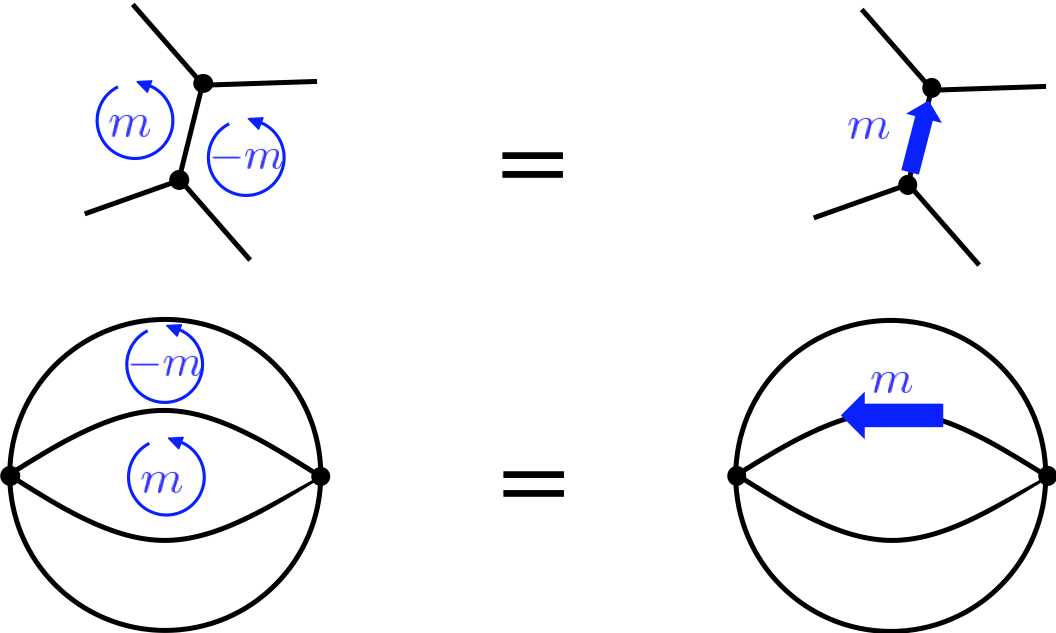}
\caption{
If fluxes of plaquettes straddling a shared propagator are given by $m$ and $-m$, 
such a configuration is interpreted  as a  twist of the shared %common
propagator.   
The upper figures show a relevant part with the  twisted propagator  in general diagrams.
%denote a part of such a diagram  are inserted in general diagrams.
}
\label{f:proptwist}
\end{figure}
%%%%%%%%%%%%%%%%%%%%%%%%%%%%%%%%%%%%%%%%%%
The contributions to EE from such a class of diagrams are then understood as propagator contributions. %, and 
We will investigate it both in the perturbative and nonperturbative approaches.

%%%%%%%%%%%%%%%%%%%%%%%%%%%%%%%%%%%%%%%%
\begin{figure}[t]
\centering
\includegraphics[width=0.35\linewidth]{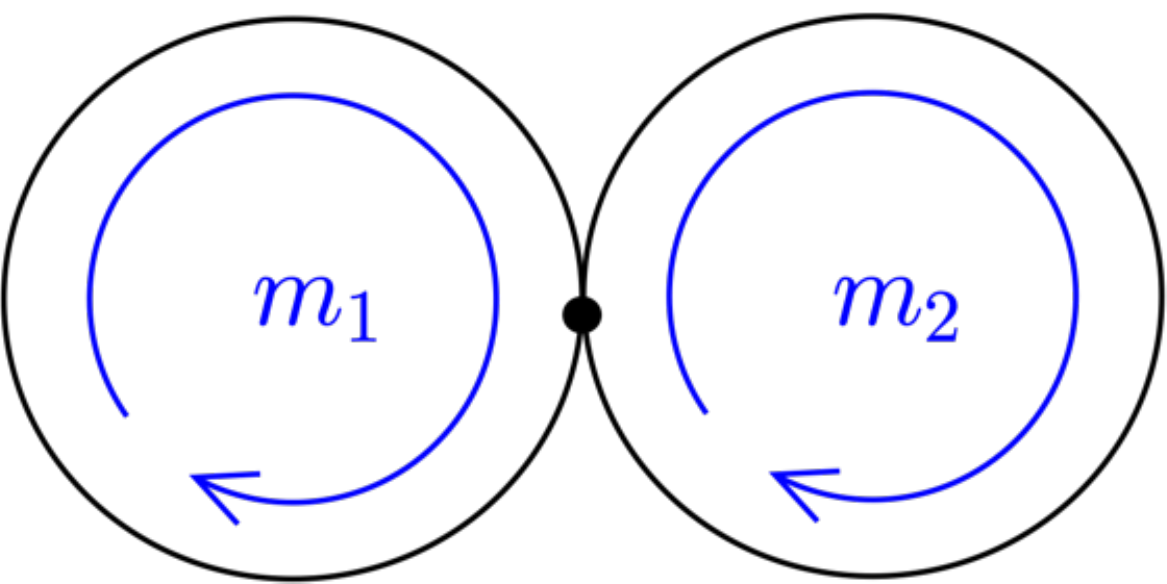}
\caption{
A 2-loop diagram with twist $(m_1, m_2)$. %The twist $(m,  \pm m)$ can be interpreted 
%as a twist of the 4-point vertex by decomposing it into two 3-point vertices. % by $\delta^{d+1}(x-y)$.
}
\label{Fig3}
\vspace{15pt}
\includegraphics[width=9cm]{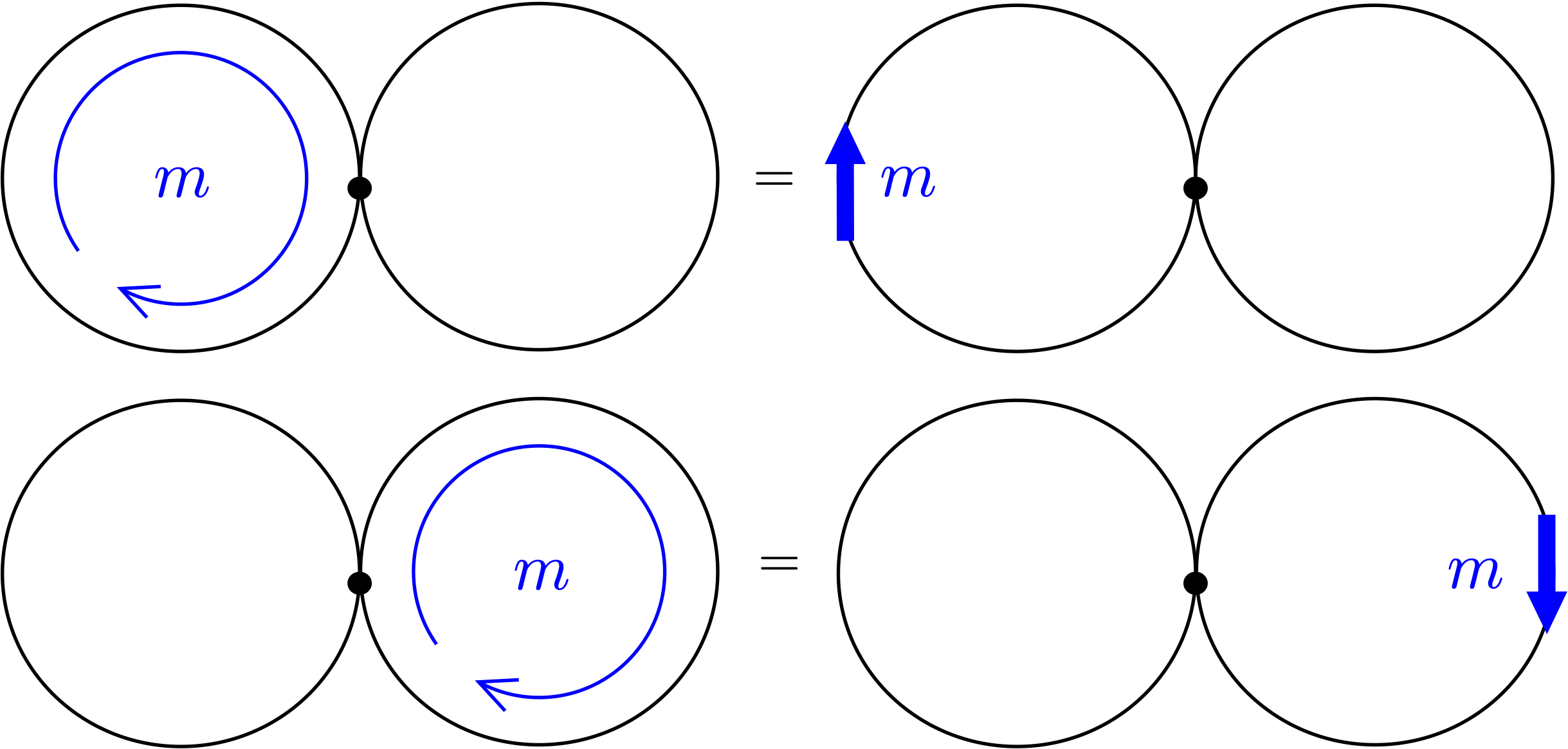}
\caption{
 2-loop diagram with twists $(m_1, m_2)=(m,0), (0,m)$ with $m\neq0$ (left). 
 They are interpreted as a twist of the corresponding propagators (right).
}
\label{Fig3-1}
\end{figure}
%%%%%%%%%%%%%%%%%%%%%%%%%%%%%%%%%%%%%%%%%%%

For a one-loop diagram (Fig.\ref{Fig2}), since the omitted outer flux is $-m$, the diagram is also a part of the propagator contributions. 

Next, we study contributions to EE from multi-loops. Flux configurations of the 2-loop figure-eight diagram are characterized by
 twists $(m_1, m_2)$ on the two plaquettes (Fig.\ref{Fig3}). 
Its contribution to the free energy is given by 
\aln{
\tilde{F}_{\text{2-loop}}^{(M)}= 
\sum_{m_1,m_2} \frac{3 \lambda}{4M} \int d^{d+1} x \ G_0 (\hat{g}^{m_1} x, x) G_0(\hat{g}^{ m_2} x, x) .
\label{8figure}
}

Specific configurations of twists, $(m, 0)$ and $(0, m)$ with $m\neq 0$, correspond to a twist of each propagator (Fig.\ref{Fig3-1}). They consist of the two-loop propagator contributions. By using (\ref{e:summ}) and (\ref{e:efftwistprop-body}), 
their contributions to the free energy and EE are computed respectively as
\aln{
\tilde{F}_\text{2-loop, prop}^{(M)}&=2\times\frac{3\lambda}{4M}\sum_{m=1}^{M-1}\int d^{d+1}x\frac{1}{4\sin^2\theta_m}G_0^{\text{bdry}}(0)\delta^2(\bm{x})G_0(0)\nonumber\\
&=V_{d-1}\frac{3\lambda(M^2-1)}{24M}G_0(0)G_0^{\text{bdry}}(0),\\
S_\text{2-loop, prop}&=-\frac{V_{d-1}}{12}G^{\text{bdry}}_0(0)\,[3\lambda G_0(0)]\nonumber\\
&=
-\frac{V_{d-1}}{12}\int\frac{d^{d-1}k_\parallel}{(2\pi)^{d-1}}\tilde{G}_0^{\text{bdry}}(k_\parallel)[3\lambda G_0(0)]. 
\label{e:S2loopprop}
}
Note that the repulsive (positive $\lambda$) interaction lessens EE. 
It is consistent with an expectation that the degrees of freedom must be reduced by introducing a positive $\lambda$
(otherwise the system becomes unstable) interaction. 

\eqref{e:S2loopprop} indicates that this contribution to EE can be attributed to the  mass renormalization 
to the one-loop contribution of \eqref{EE-1loop}:
\begin{gather}
S_\text{1-loop}+S_\text{2-loop, prop}=
-\frac{V_{d-1}}{12}\int^{1/\epsilon}\frac{d^{d-1}k_\parallel}{(2\pi)^{d-1}}\mathrm{ln}\,\left[\bigl(\tilde{G}_1^{\text{bdry}}(k_\parallel)\bigr)^{-1}\epsilon^{2}\right],
\\
\tilde{G}_1^{\text{bdry}}(k_\parallel)=\frac{1}{k_\parallel^2+m_{r1}^2},~~~~
m_{r1}^2=m^2_{0}+ 3\lambda G_0(0).
\nonumber
\label{e:renormalizem}
\end{gather}
The above equalities hold up to $O(\lambda^1)$. This was also suggested in~\cite{Hertzberg:2012mn} to $O(\lambda^1)$. 
Although all the higher order calculation of the propagator contributions can be done in the same manner order by order, we will later see these contributions can be resummed to all orders in a nonperturbative way.

\subsection{Vertex contributions}\label{sec:vertex-pert}
Besides the propagator contributions, there exist a new class of contributions to EE, namely vertex contributions. Such configurations with a simple interpretation are given by 
a set of flux configurations that {\it straddles a vertex} instead of a propagator. %If we took a special gauge and assigned twists on  particular links of Feynman diagrams, we would not find them since they are hidden in twisting multiple links. 
Consider a diagram with twists given schematically in Fig.\ref{f:verttwist}.
%%%%%%%%%%%%%%%%%%%%%%%%%%%%%%%%%%%%%%%%%%  
\begin{figure}[t]
	\centering
	\includegraphics[width=10cm]{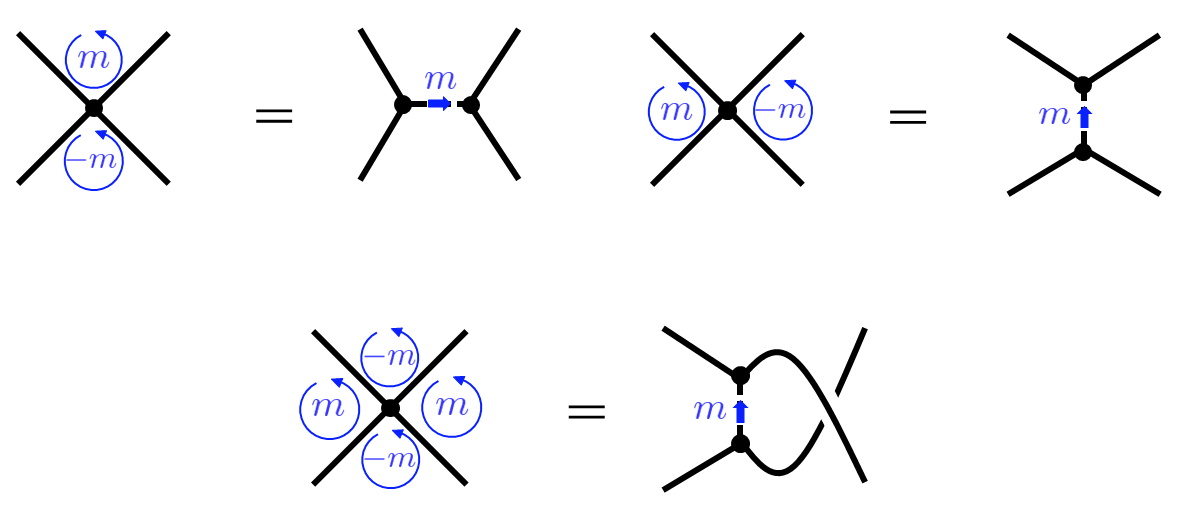}
	\caption{
Twisting a vertex: these three types of 
configurations can be attributed to a twist of a vertex. The dotted lines in the figures on the right-hand sides are delta functions to open the vertex. 
The twist of a vertex is interpreted as a twist of the dotted propagator. 
Each set of figures represent the three channels, 
$s$-channel (upper left figures), $t$-channel (upper right figures) and
$u$-channel (lower figures) respectively.
	}
	\label{f:verttwist}
\end{figure}
%%%%%%%%%%%%%%%%%%%%%%%%%%%%%%%%%%%%%%%%%%
%There, 
In these configurations, 
plaquettes with nonvanishing fluxes of twists 
meet at a vertex, and there are three types of such configurations. 
 We can interpret these configurations as a configuration of a single twisted vertex\index{twisted vertex} in the $s$, $t$, and $u$-channel respectively.
 This interpretation can be 
  realized by ``opening'' the vertex with a delta function. 
   For example, the four-point vertex can be rewritten as
\aln{
	\frac{\lambda}{4}\int d^{d+1}x\,\phi(x)^4=\frac{\lambda}{4}\int d^{d+1}x\,d^{d+1}y\,\phi(x)^2\phi(y)^2\delta^{d+1}(x-y).
}
Then,  we can understand a twisted vertex as %means 
an opened vertex with a twist on the separated two coordinates as
\aln{
	\frac{\lambda}{4}\int d^{d+1}x\,d^{d+1}y\,\phi(x)^2\phi(y)^2\delta^{d+1}(\hat{g}^mx-y).
} 
%The upper and lower figures in Fig.\ref{f:verttwist} correspond to the $s$- and $u$-channel openings of the vertex, respectively. 
%The $t$-channel opening is given in the same manner
%as the $s$-channel. 
The upper left, upper right, and lower figures in Fig.\ref{f:verttwist} correspond to the $s$, $t$, and $u$-channel openings of the vertex, respectively. 
As we have 
demonstrated for a single twisted propagator \eqref{e:efftwistprop-body}, 
we can replace the twisted delta function (to be exact, its two-dimensional part) in the diagram as 
\aln{
	\delta^{2}(\hat{g}^{m}\bm{x} - \bm{y}) =
	e^{\cot \theta_n \hat{R}_{\bm{X}} /2} \frac{ \delta^2 (\bm{X})}{4 \sin^2\theta_m}  \rightarrow \frac{\delta^2 (\bm{X})}{4 \sin^2\theta_m}. 
}
The \textit{twisted vertex}\index{twisted vertex} is thus 
interpreted as a vertex symmetrically splitted with two loose ends and also  with its center coordinate being fixed at the boundary.

Let us evaluate these vertex contributions up to the 3-loop level. The 2-loop vertex contributions stem
 from the figure-eight diagram with two types of configurations of twists, as shown in Fig.\ref{Fig3-2}. 
%%%%%%%%%%%%%%%%%%%%%%%%%%%%%%%%%%%%%%%%
\begin{figure}[t]
	\centering
	\includegraphics[width=0.8\linewidth]{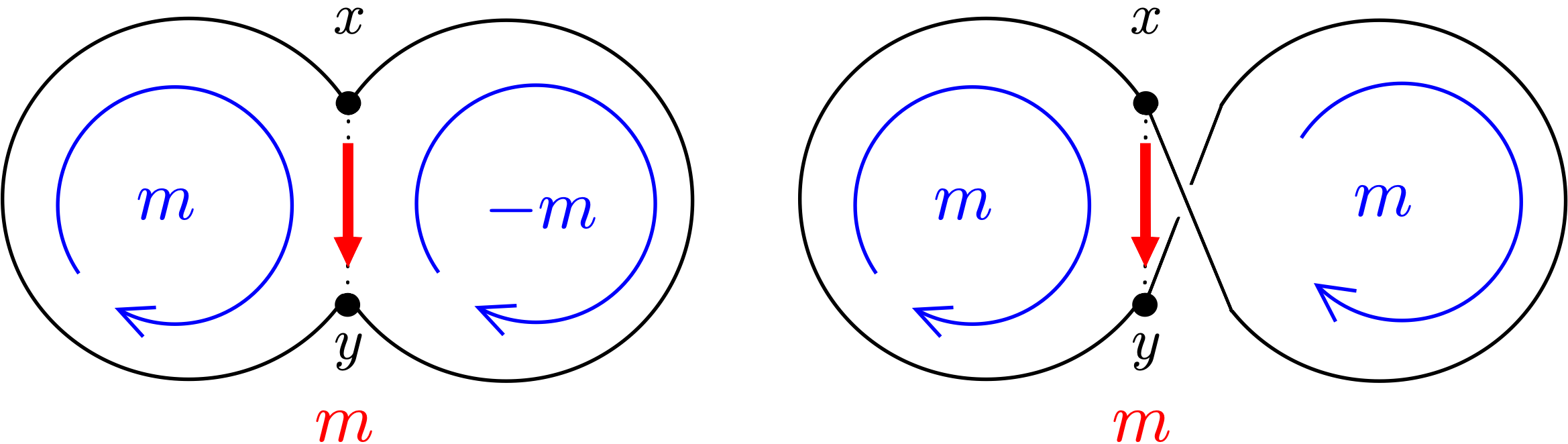}
	\caption{
		{
			2-loop figure-eight diagrams with twists $(m_1, m_2)=(m,\mp m)$. These flux configurations of  twists can be interpreted 
			as a twist of the 4-point vertex by decomposing it into two 3-point vertices. % by $\delta^{d+1}(x-y)$.
		}
	}
	\label{Fig3-2}
\end{figure}
%%%%%%%%%%%%%%%%%%%%%%%%%%%%%%%%%%%%%%%%%%%
Note that the configuration of the $s$-channel opening is absent in the figure-eight diagram
because the vertex in the figure-eight diagram is surrounded by essentially three plaquettes, %: 
two circles, and one outer circle.
% and in the $s$-channel diagram two pairs of plaquettes with simultaneous twists adjacent to each other, necessary for a vertex twist, is absent.  
Their contributions to the free energy and EE are calculated as
\aln{
	\tilde{F}^{(M)}_{\text{2-loop.vert}}&=2\times\sum_{m=1}^{M-1}\frac{3\lambda}{4M}\int d^{d+1}x\,d^{d+1}y\,G_0(x-y)^2\delta^{d-1}(x_\parallel-y_\parallel)\delta^2(\hat{g}^m\bm{x}-\bm{y})\nonumber\\
	&=V_{d-1}\lambda\frac{M^2-1}{8M}\int d^2\bm{r}G_0(\bm{r},0)^2,\\
	S_{\text{2-loop, vert}}&=-V_{d-1}\frac{\lambda}{4}\int d^2\bm{r}G_0(\bm{r},0)^2.
	%\nonumber\\&=-V_{d-1}\frac{\lambda}{4}\int d^{d+1}r\,G_0(\bm{r},0)^2\delta^{d-1}(r_\parallel).
	\label{e:EE2loopvert}
}
%In the scond line of Eq.(\ref{e:EE2loopvert}), we have rewritten the contribution for later use. 
{The vertex correction to EE is negative for repulsive (positive $\lambda$) interaction.} 
Note that in the real $\phi^4$ theory, %the dotted lines are indistinguishable between different channels
different channels are indistinguishable 
and a summation of different channels give just an additional  numerical factor in front. 
In the next section, we will consider an extended model in which a different channel gives a different 
type of contribution.

The 3-loop contributions come from two diagrams shown in Fig.\ref{f:3-loop-3} and Fig.\ref{Fig5}. %One of them is illustrated in Fig.\ref{f:3-loop-3}. 
%There, 
For a diagram illustrated in Fig.\ref{f:3-loop-3}, the vertex contributions stem from the four configurations: $(m_1,m_2,m_3)=(m,\pm m,0)$, $(0,m,\pm m)$. We see them as $t$- and $u$-channel opening of the two vertices. $s$-channels are absent because each vertex is surrounded by two plaquettes and one outer circle, not four independent ones. The contributions from these configurations to the free energy and EE are given by 
\aln{
		\tilde{F}^{(M)}_{\text{3-loop, vert1}}&=4\times\biggl(-\frac{9\lambda^2}{4M}\sum_{m=1}^{M-1}\int d^{d+1}x_1\,d^{d+1}x_2\,d^{d+1}y\,G_0(x_1-x_2)\,G_0(x_1-y)\nonumber\\*
		&\hspace{6cm}\times G_0(x_2-y)\,G_0(0)\,\delta^{d+1}(\hat{g}^mx_1-x_2)\biggr)\nonumber\\*
		&=-V_{d-1}\frac{3\lambda^2(M^2-1)}{M}\int d^2\bm{x}\,d^2\bm{y}\,d^{d-1}r_\parallel\ G_0(2\bm{x},0)\,G_0(\bm{x}-\bm{y},r_\parallel)\nonumber\\*
		&\hspace{8cm}\times G_0(\bm{x}+\bm{y},r_\parallel)\,G_0(\bm{0},0),\nonumber\\*
		&=-V_{d-1}\frac{3\lambda^2(M^2-1)}{4M}\int d^2\bm{r}\,d^2\bm{s}\,d^{d-1}r_\parallel\ G_0(\bm{r},0)\,G_0(\bm{r}-\bm{s},r_\parallel)G_0(\bm{s},r_\parallel)\,G_0(\bm{0},0),\\*
		S_{\text{3-loop, vert1}}&=V_{d-1}\,\frac{3}{2}\lambda^2\int d^2\bm{r}\,d^2\bm{s}\,d^{d-1}r_\parallel\ G_0(\bm{r},0)\,G_0(\bm{r}-\bm{s},r_\parallel)\,G_0(\bm{s},r_\parallel)\,G_0(\bm{0},0),\nonumber\\*
		&=V_{d-1}\,\frac{3}{2}\lambda^2\int d^{d+1}r\,G_0(r)\,G_0(0)\left[\int d^2\bm{s}\,G_0(\bm{s},0)\,G_0(\bm{s}-\bm{r},r_\parallel)\right].
	\label{e:EE3loopvert2}
}
%%%%%%%%%%%%%%%%%%%%%%%%%%%%%%%%%%%%%%%%
\begin{figure}[t]
	\includegraphics[width=6cm]{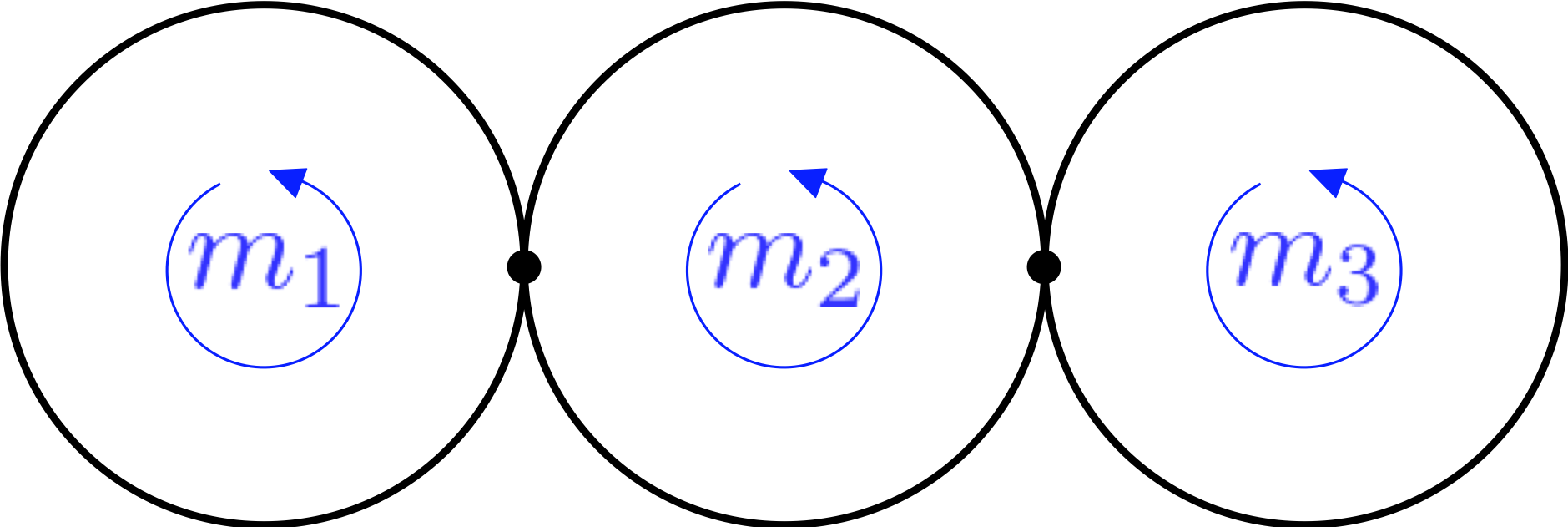}
	\centering
	\caption{
		A 3-loop diagram. Four types of flux configurations,  $(m_1,m_2,m_3)=(m,\pm m,0)$, $(0,m,\pm m)$,
		can be interpreted as twisting vertices. 
		Opening vertices are done in the same manner as in the 2-loop diagrams.  
	}
	\label{f:3-loop-3}
\end{figure}
%%%%%%%%%%%%%%%%%%%%%%%%%%%%%%%%%%%%%%%%%%%

Another 3-loop diagram is given by the leftmost diagram in Fig.\ref{Fig5}.
%%%%%%%%%%%%%%%%%%%%%%%%%%%%%%%%%%%%%%%%
\begin{figure}[t]
\centering
	\begin{tabular}{c}%prevent line break
		\hspace*{-0.05\linewidth}
		\begin{minipage}{0.45\hsize}%align figs horizontally
			\centering
			\includegraphics[width=\linewidth]{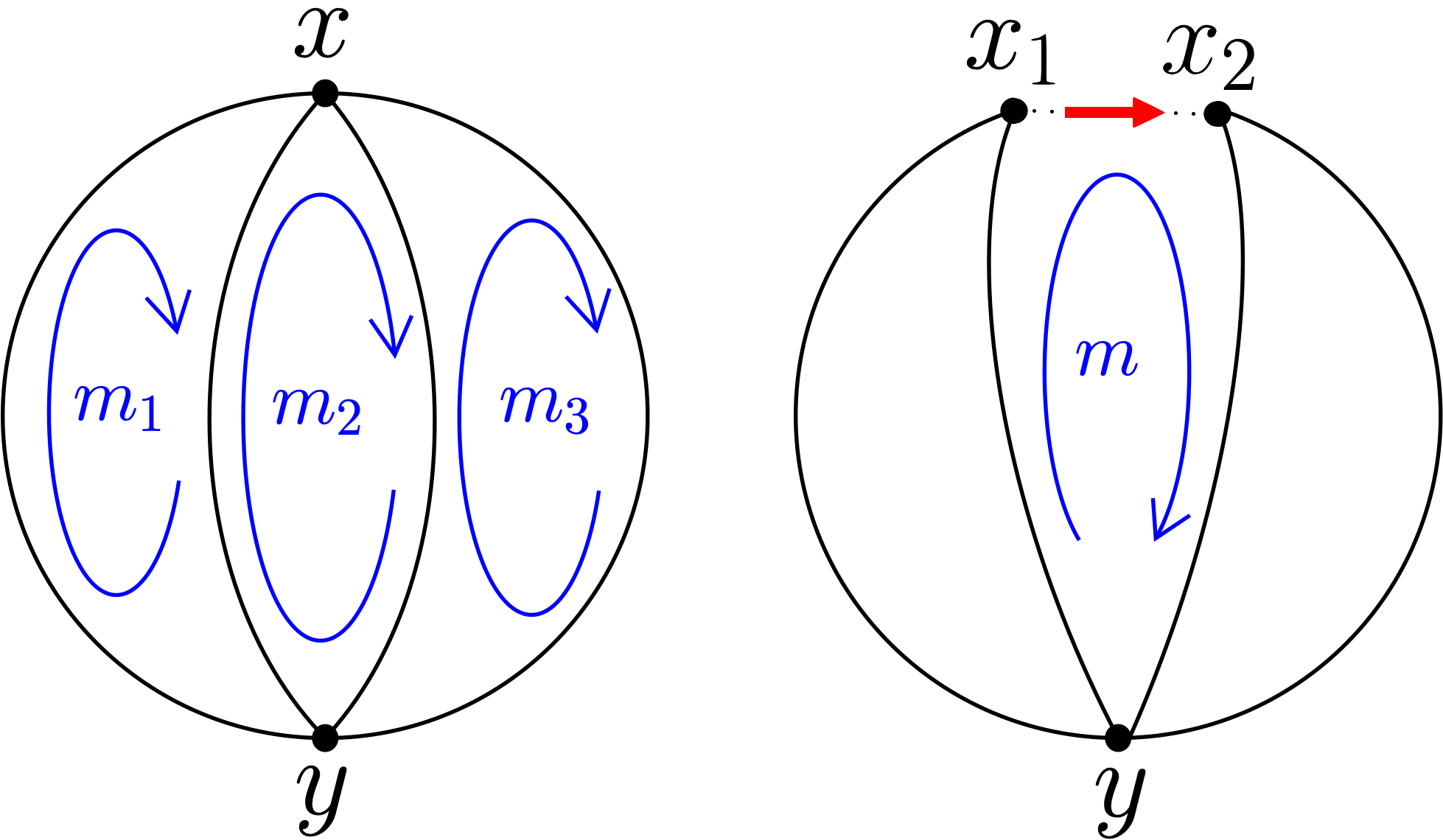}
		\end{minipage}
		\hspace*{0.03\linewidth}
		\begin{minipage}{0.45\hsize}
			\centering
			\includegraphics[width=\linewidth]{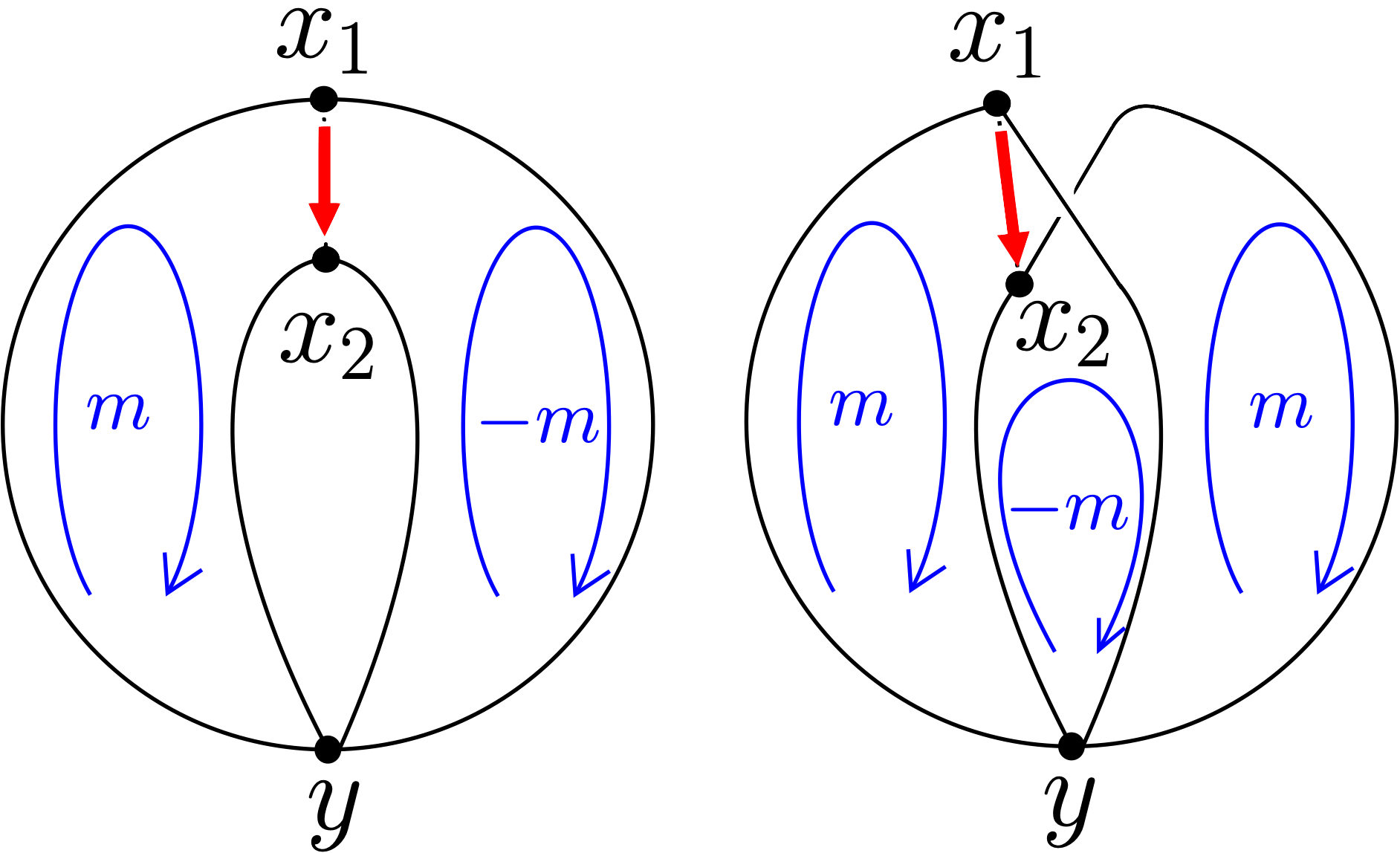}
		\end{minipage}
	\end{tabular}
	%\centering
	%\includegraphics[width=0.8\linewidth]{Fig5.png}
	\caption{
		Another 3-loop diagram with twists $(m_1, m_2, m_3)$ {(leftmost)}.
		A particular configuration $(0,m,0)$ corresponds to twisting a vertex, as well 
		as $(m,0,-m)$ and $(m,-m,m)$ (three diagrams on the right).  They generate a twist in the delta function  
		$\delta^2({\bm x}_1-{\bm x}_2)$. We can also open the vertex at ${\bf y}$ instead of ${\bm x}$, 
		and they have two different interpretations of twisting vertices, analogous to the propagator case (Fig.\ref{Fig2}).
		These vertex contributions should not be double-counted. 
		% of twists, $(m,0,\pm m)$ and $(0, m, 0)$,  corresponding to twisting a vertex (right).
	}
	\label{Fig5}
\end{figure}
%%%%%%%%%%%%%%%%%%%%%%%%%%%%%%%%%%%%%%%%%%%
The following three types of configurations of twists correspond to twists of a vertex:  $(0,m,0)$, $(m,0,-m)$, and $(m,-m,m)$. 
We can assign a flux of twist $-m$, $0$, and $-m$ on the outer circle of the plaquette respectively. 
They are equivalent to the $t$-, $s$- and $u$-channel opening of the vertex.   
In this diagram,  we face a similar problem as we discussed for a one-loop diagram; the failure of one-to-one correspondence. 
There are two ways to attribute the flux configurations to twisting either an upper or lower vertex. 
These two attributions are not independent and 
we can only twist one of them.  
These three channels give the same contributions in the $\phi^4$ theory. %and 
Then, the corresponding 3-loop contributions from Fig.\ref{Fig5} are computed as
\aln{
		\tilde{F}^{(M)}_{\text{3-loop, vert2}}&=3\times\biggl(-\frac{3\lambda^2}{4M}\sum_{m=1}^{M-1}\int d^{d+1}x_1\,d^{d+1}x_2\,d^{d+1}y\,G_0(x_1-y)^2G_0(x_2-y)^2
		%   \nonumber\\* &\hspace{8.4cm}\times
		\delta^{d+1}(\hat{g}^mx_1-x_2)\biggr)\nonumber\\* 
		&=-V_{d-1}\frac{3\lambda^2(M^2-1)}{4M}\int d^2\bm{x}\,d^2\bm{y}\,d^{d-1}		r_\parallel\,G_0(\bm{x}-\bm{y},r_\parallel)^2\,G_0(\bm{x}+\bm{y},r_\parallel)^2,\nonumber\\
		&=-V_{d-1}\frac{3\lambda^2(M^2-1)}{16M}\int d^2\bm{r}\,d^2\bm{X}\,d^{d-1}r_\parallel\,G_0(\bm{r},r_\parallel)^2\,G_0(\bm{X},r_\parallel)^2,\\
		S_{\text{3-loop, vert2}}&=V_{d-1}
		\frac{3\lambda^2}{8} \int d^2\bm{r}\,d^2\bm{s}\,d^{d-1}r_\parallel\,G_0(\bm{r},r_\parallel)^2\,G_0(\bm{s},r_\parallel)^2,\nonumber\\
		&=V_{d-1}
		\frac{3\lambda^2}{8}
		\int d^{d+1}rG_0(r)^2\left[
		\int d^2\bm{X}\,G_0(\bm{X},r_\parallel)^2\right].
	\label{e:EE3loopvert}
}
In contrast to the twisting of propagators, both of the contributions of \eqref{e:EE3loopvert2} and \eqref{e:EE3loopvert} 
essentially originate from the non-Gaussianity of the vacuum. 
We also emphasize the importance of the covariant viewpoint as $\mathbb{Z}_M$ gauge theory on Feynman diagrams.  
If we take a special gauge and assign twists on specific links (propagators), we could not find out vertex contributions since they are hidden in the configurations with multiple twisted links.

While \eqref{e:EE3loopvert2} can be interpreted as a contribution from the figure-eight diagram with the renormalized propagator \eqref{e:renormalizem}, \eqref{e:EE3loopvert} cannot be absorbed into the renormalization of the propagator nor the vertex. 
The situation is different from the propagator contributions, nonetheless, it is consistent with the ordinary renormalization structure in another viewpoint. 
In the following, we will show that the above vertex contributions can be %identified to 
summarized as those from renormalized two-point functions of composite operators. 

\section{Nonperturbative EE via resummation}\label{sec:nonpert-EE}
In the previous section, we perturbatively evaluated propagator and vertex contributions to EE. In this section, we address how to resum these contributions to all orders and obtain nonperturbative expressions in terms of renormalized quantities.

\subsection{Resummation of propagator contributions by two-particle irreducible (2PI) formalism}\label{sec:prop-2PI}
%%%%%%%%%%%%%%%%%%%%%%%%%%%%%%%%%%%%%%%%
\begin{figure}[b]
\centering
\includegraphics[width=4cm]{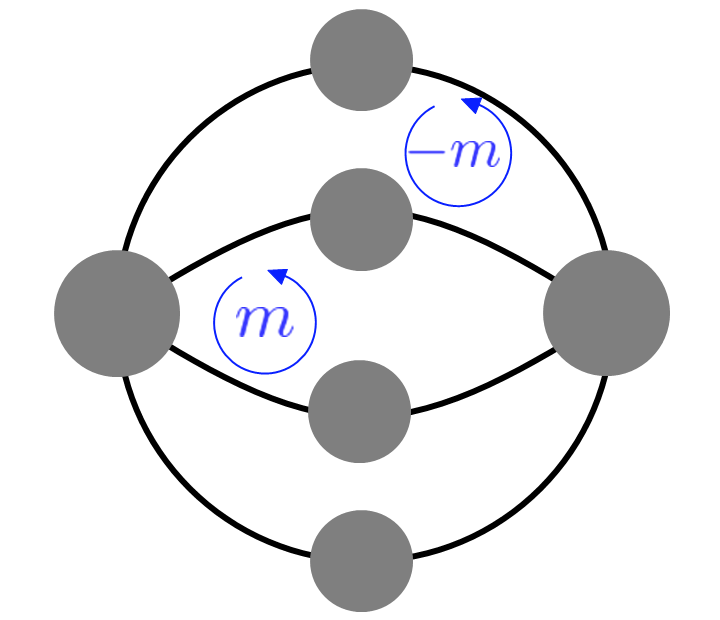}
\caption{
An example of a configuration of fluxes $(m,-m)$ that has multiple  interpretations 
in terms a twist of a {\it bare} propagator. 
The gray blobs %denote 
represent 1PI subdiagrams. 
This configuration can be interpreted as a twist of  either 
one of the two shared (bare) propagators, but not both. }
\label{f:ambiguity1}
\end{figure}
%%%%%%%%%%%%%%%%%%%%%%%%%%%%%%%%%%%%%%%%%%%

One might suspect whether  a configuration of fluxes like Fig.\ref{f:proptwist} has a one-to-one correspondence
to a configuration of  a twist of the propagator   in general diagrams. 
Indeed, we need careful treatment for particular diagrams. 
Consider a diagram like Fig.\ref{f:ambiguity1} 
where two  plaquettes with nonzero fluxes meet at two or more propagators. 
In this case,  the configuration of fluxes $(m,-m)$ 
corresponds to a twist of either propagator, but not to both. We already saw the simplest example Fig.\ref{Fig2} in the one-loop calculation as well as \eqref{eq:divide-by-mult-prop} in the propagator contributions.
Fig.\ref{f:ambiguity1} shows that such a configuration of fluxes can be interpreted as a twist of a single, full propagator. This is what we confirmed perturbatively up to the first order in \eqref{e:renormalizem}. It is not a trivial fact, but physically natural since EE is a measure of entanglement among microscopic degrees of freedom and should be related to the low-energy observables through renormalization. This observation motivates us to pursue the following analysis that EE (or at least its universal term) is expressed in terms of renormalized correlation functions in the \textbf{two-particle irreducible (2PI) formalism}\index{two-particle irreducible formalism|see  2PPI formalism }\index{2PI formalism}~\cite{PhysRevD.10.2428,Berges:2004yj}.

The 2PI formalism allows us to systematically study a relationship between renormalization of propagators and EE. 
Combined with the orbifold analysis, we will confirm that the Gaussian contributions to EE are completely expressed in terms of the renormalized two-point function of the fundamental field in the following. 

The \textbf{2PI effective action}\index{2PI effective action} is given, in addition to the classical action, by
\aln{
\Gamma[G]  &{=F[G]=-\log Z} \nonumber\\
&=  -\frac{1}{2} \tr \log G + \frac{1}{2} \tr(G_0^{-1} G-1) +\Gamma_2[G],
\label{2PIEA}
}
where $G$ is a full propagator, namely, a renormalized two-point function\index{renormalized two-point function}. $\Gamma_2[G]$ (or $\Phi$ in some literature) is minus the sum of connected 2PI bubble diagrams which consist of the full propagators $G$'s as internal lines. We assume that the one-point function vanishes: $\left<\phi\right>=0$. In this formalism, $G$ is determined self-consistently by its nonperturbative equation of motion called a \textbf{gap equation}\index{gap equation} or \textbf{Dyson equation}\index{Dyson equation}: \aln{
\frac{\delta \Gamma[G]}{\delta G}=0~~~\Leftrightarrow~~~G^{-1}=G_0^{-1}+2\frac{\delta \Gamma_2}{\delta G}[G].
\label{e:gapeq} 
} 
With the solution to \eqref{e:gapeq}, $G=\bar{G}[G_0]$, $\Gamma[\bar{G}]$ coincides with the 1PI free energy\index{1PI free energy}.
\begin{equation}
    \Sigma \equiv -2\frac{\delta \Gamma_2}{\delta G}[G=\bar{G}]
    \label{eq:self-energy}
\end{equation}
is nothing but the 1PI self-energy\index{self-energy}. 

In the 2PI analysis, 
since  $G(x,y)$ itself is composed of  propagators as  internal loop corrections, 
we  distinguish the following two types of twistings.
The first type of twistings is denoted by $\delta_m G(x,y)$, which
represents a variation of the internal structure induced by twisting. 
The second type is simply given by $G(\hat{g}^mx,y)$, which represents the
twisting of the full propagator in the same way as previously. 
Namely, the projection operator $\hat{P}$ is acted from outside. 
We will show that the first type of twistings is canceled by the gap equation.
Moreover, we will prove that 
there are further cancellations among 2PI diagrams and the second term of the 2PI effective action in \eqref{2PIEA}. 
The gap equation is responsible for  the cancellations, but special care is necessary 
for such diagrams in Fig.\ref{f:ambiguity1}. 

First, let us see that twistings inside the full propagators  are canceled and
contributions from $\delta_m G(x,y)$ vanish. 
It is simply because of the gap equation;
\aln{
\Gamma[\bar{G}]_\text{prop,int}&=\sum_{m=1}^{M-1}\int d^{d+1}x\,d^{d+1}y\,\frac{1}{2}\left(-\bar{G}^{-1}+G_0^{-1}+2\frac{\delta\Gamma_2}{\delta G}[\bar{G}]\right)_{yx}\delta_mG(x,y)\nonumber\\
&=0.
}
Thus we can safely forget about the internal structure of the full propagator. 

Next, we look at the twisting of the full propagator itself.  
As expected, most configurations with a single twisted propagator
 are canceled due to the gap equation, except for diagrams like Fig.\ref{f:ambiguity1}
where a configuration of fluxes of $(m,-m)$ can be attributed to twisting  one of the propagators
straddled by the plaquettes. 
In the 2PI formalism,  such diagrams  are included only in the first term in \eqref{2PIEA}
 because all  diagrams with such property are not 2PI (see Fig.\ref{f:ambiguity1}) and not included in other terms.\footnote{
The second term is not 2PI, but $G_0^{-1}$ is a local operator and it is sufficient to twist the propagator $G$ in the trace. }
 Then, we can separately consider contributions from the first term 
 and those from the second and third term in \eqref{2PIEA}. 

From the first logarithmic term of \eqref{2PIEA}, it is straightforward to see that 
we have 
\aln{
S^{\text{2PI}}_{\text{prop, ext,1}} =- \frac{V_{d-1} }{12} 
 \int^{1/\epsilon} \frac{d^{d-1}k_{\parallel}} {(2 \pi)^{d-1}}  
  \log \left[ \tilde{G}^{-1}(\bm{0}; k_\parallel) \epsilon^2 \right],
\label{EE-1loop2PI,1}
}
where $\tilde{G}(\bm{k}; k_\parallel)$ is a Fourier transform of the
renormalized Green function, $G(\bm{x}; x_\parallel)$. $\tilde{G}(\bm{0}; k_\parallel)$ is a renormalized counterpart of $\tilde{G}^{\text{bdry}}$. 
Note that,  though
$\tilde{G}_0^{\text{bdry}}(k_\parallel)$ describes a propagation in a $(d-1)$-dimensional theory, 
the renormalization of $\tilde{G}(\bm{0}; k_\parallel)$ itself is performed in 
the $(d+1)$-dimensional space, as in \eqref{e:renormalizem}.
\eqref{EE-1loop2PI,1} is in the same situation as the previous one-loop calculation of EE. Using the gap equation \eqref{e:gapeq}, \eqref{eq:self-energy}, the difference from the free scalar case \eqref{EE-1loop} is given by
\begin{align}
    S^{\text{2PI}}_{\text{prop, ext,1}} - S_{\text{1-loop}} &=- \frac{V_{d-1} }{12} 
 \int^{1/\epsilon} \frac{d^{d-1}k_{\parallel}} {(2 \pi)^{d-1}}  
  \log \left[ 1-\tilde{G}_0 \Sigma\right](\bm{0}; k_\parallel) \nonumber\\
  & = \frac{V_{d-1} }{12} 
 \int^{1/\epsilon} \frac{d^{d-1}k_{\parallel}} {(2 \pi)^{d-1}}  
 \sum_{n=1}^\infty \frac{1}{n}
  \left[ \tilde{G}_0 \Sigma\right]^n (\bm{0}; k_\parallel).
  \label{eq:divide-by-mult-prop}
\end{align}
The factor $1/n$ comes from the divide-by-multiplicity method. This prevents the overcounting from a naive chain of the bare propagators and self-energies with a twist.

Other contributions to EE follow
 from the second and third terms in \eqref{2PIEA}. Their contributions to EE are given by
\aln{
S^{\text{2PI}}_\text{prop,ext,2+3}&=\sum_{m=1}^{M-1} 
\int d^{d+1}x \,d^{d+1}y\,
\left (\frac{1}{2} G_0^{-1} +  \frac{\delta \Gamma_2}{\delta G}[\bar{G}] \right)_{yx} \bar{G}(\hat{g}^{m}x, y)\nonumber\\
&= \sum_{m=1}^{M-1} 
\int d^{d+1}x d^{d+1}y
\left (\frac{1}{2} \bar{G}^{-1} \right)_{yx} \bar{G}(\hat{g}^{m}x, y)%G(\hat{g}^{n}x, y).
\label{trivialEE}
}
Since this is nothing but a twist of $\tr (G^{-1} G) /2$, i.e. just a variation of unity, $S^{\text{2PI}}_\text{prop,ext,2+3}$ is a trivial constant and %we can drop it
can be dropped. 

By combining %the above results 
\eqref{EE-1loop2PI,1} and \eqref{trivialEE}, we obtain the contribution to EE from twisting a propagator in terms of %from 
the renormalized two-point function nonperturbatively:
\aln{
S^{\text{2PI}}_{\text{prop}} =- \frac{V_{d-1} }{12} 
 \int^{1/\epsilon} \frac{d^{d-1}k_{\parallel}} {(2 \pi)^{d-1}}  
  \log \left[ \tilde{G}^{-1}(\bm{0}; k_\parallel) \epsilon^2 \right]. 
\label{e:EE2PIprop}
}
%This completes a proof of the 
Previously we made a conjecture that the total propagator contribution to EE could be represented as renormalization of the propagator. 
The above argument completes the proof. %Note that 
The Gaussian contribution %from the Gaussian nature 
is all %organized 
summarized in the above form.\footnote{%-------------------------------------------
When we  compare \eqref{e:EE2PIprop} to the ordinary perturbative calculation, 
since all the diagrams in \eqref{2PIEA} are written in terms of the full propagator $G$, 
we have to expand  each diagram in the comparison. 
As a result, diagrams consisting of $G_0$'s are included in all the three terms in \eqref{2PIEA} 
and the correct coefficients can be obtained by taking all these terms into account.
\label{fn:pertnonpert}
 } %-------------------------------- 
Note that it is consistent with the leading order result of perturbative calculations in \cite{Hertzberg:2012mn,Chen:2020ild}.

\subsection{Resummation of vertex contributions -- auxiliary field method}\label{sec:auxiliary}
In contrast to the propagator contribution, which is written in terms of renormalized propagators, the vertex contribution does not seem to be simply written in terms of renormalized vertices as discussed in Section \ref{sec:vertex-pert}. Even if it comes from the renormalized vertices, it is difficult to study the vertex contribution similarly as the propagator contribution since the $n$-PI formalism is not known for a general $n$. Thus, we take a different route, namely, interpreting the opened vertices as corresponding auxiliary fields and calculating EE from their propagators.

%%%%%%%%%%%%%%%%%%%%%%%%%%%%%%%%%%%%%%%
In order to formulate the ``opening of a vertex'' more systematically, it is instructive to consider a model where opening each vertex leads to distinct $s$-, $t$- and $u$-channels. One of such models is described by two complex scalars, whose action is given by 
%written as follows:
\aln{
I=\int\frac{d^{d+1}x}{M}\left[\sum_{i=1}^2\bar{\phi}_i(-\Box+m^2_{0})\phi_i+\frac{\lambda}{4}(\bar{\phi}_1\phi_1)(\bar{\phi}_2\phi_2)\right].
\label{e:actionphiphi}
}
Here and in the following, $\mathbb{Z}_M$ projections on fields are written implicitly. 
%symbols to denote the fields such as $\phi_1$ implicitly means a projected fields $\hat{P}\phi_1$. 
Each vertex contribution involves three configurations of twists as mentioned in Fig.\ref{f:verttwist}. 
It is now almost clear that 
each twist of a vertex in $s$, $t$, and $u$-channels 
can be regarded as a twist of the propagator of the corresponding auxiliary field\index{auxiliary field}. 
 With the auxiliary field, the action has %The field has 
a three-point interaction vertex and reproduces the original four-point one when integrated out.

%In a correspondence
Corresponding to the above three ways for the opening of vertices, 
we can rewrite the action \eqref{e:actionphiphi} into the following three forms:
\aln{
I_s&=\int\frac{d^{d+1}x}{M}\left[\sum_{i=1}^2\bar{\phi}_i(-\Box+m^2_{0})\phi_i+c_1c_2+i\frac{\sqrt{\lambda}}{2}c_1(\bar{\phi}_2\phi_2)+i\frac{\sqrt{\lambda}}{2}c_2(\bar{\phi}_1\phi_1)\right],
\label{e:actions}\\
I_t&=\int\frac{d^{d+1}x}{M}\left[\sum_{i=1}^2\bar{\phi}_i(-\Box+m^2_{0})\phi_i+\bar{d}d+i\frac{\sqrt{\lambda}}{2}\bar{d}\phi_1\phi_2+i\frac{\sqrt{\lambda}}{2}d\bar{\phi}_1\bar{\phi}_2\right],
\label{e:actiont}\\
I_u&=\int\frac{d^{d+1}x}{M}\left[\sum_{i=1}^2\bar{\phi}_i(-\Box+m^2_{0})\phi_i+\bar{d'}d'+i\frac{\sqrt{\lambda}}{2}\bar{d'}\,\bar{\phi}_1\phi_2+i\frac{\sqrt{\lambda}}{2}d'\,\bar{\phi}_2\,\phi_1\right].
\label{e:actionu}
} 
We have introduced three pairs of auxiliary fields: real scalars $(c_1, c_2)$, and complex scalars $(d, \bar{d}), (d', \bar{d'})$.\footnote{%---------------------------
The path integral contour for them should be chosen so that %as to make 
the partition function is convergent and thus the apparent violation of the reality or boundedness %boundness
in the above actions does not produce pathology.
} %-------------------------------------
Of course, each of \eqref{e:actions}-\eqref{e:actionu} is equivalent to \eqref{e:actionphiphi} after integrating the auxiliary fields out.
Consequently, if we sum up the bubble diagrams from  all three models, 
we will encounter an overcounting at the level of free energy. 
However, when we consider configurations of twists, 
there is a one-to-one correspondence between vertex contributions of three channels and 
propagator contributions of each auxiliary field in these three models. 
In this sense, as far as a single twist of vertices is concerned, 
the vertex contributions we consider can be regarded as 
the propagator contributions from these three auxiliary fields.
As in Fig.\ref{Fig5}, 
four-point vertex contributions to EE with the  flux configurations, 
% $(m,0,-m,0)$, $(0,m,0,-m)$ or $(m,-m,m,-m)$
$(0,m,0)$, $(m,0,-m)$, or $(m,-m,m)$,  corresponds 
 to a propagator contribution of the associated auxiliary fields
 given by  \eqref{e:actions}-\eqref{e:actionu}, respectively, for any bubble diagrams of the action \eqref{e:actionphiphi}.

Every vertex in the bubbles generated by \eqref{e:actionphiphi} gets the contributions 
from the three channels\footnote{Figure-eight diagram is an exception and there is no $s$-channel.}.  
They coincide respectively with the contributions from a twisted propagator in the equivalent diagrams generated either by \eqref{e:actions}-\eqref{e:actionu}. 
Here we have the same problem of the one-to-one correspondence 
between fluxes of twists in the plaquettes and twists of vertices, 
as mentioned in the previous subsection (Fig.\ref{Fig5}). 
In terms of the auxiliary fields, this problem is easily resolved by using the same logic
as in the propagator contributions. 
2PI diagrams do not have this kind of problem, and only one-loop diagrams of the auxiliary fields
need care. See Fig.{\ref{f:ambiguity2}} as an example. 
As a result, the problem  is translated into the same problem for the twisted propagator of the auxiliary field. 
%%%%%%%%%%%%%%%%%%%%%%%%%%%%%%%%%%%%%%%%%%  
\begin{figure}[t]
\centering
\includegraphics[width=10cm]{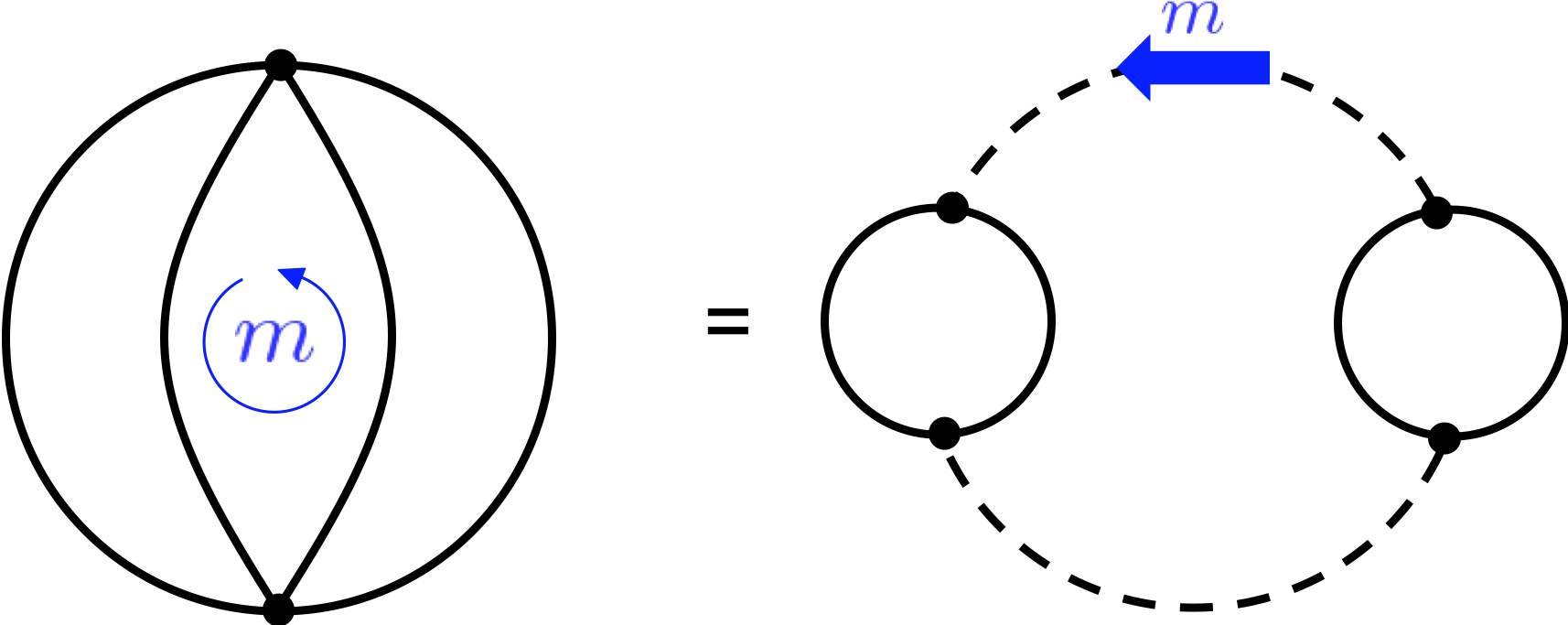}%{ambiguity2_ver2.pdf}
\caption{The right non-2PI diagram is obtained by opening two vertices in the left in terms of the auxiliary fields.
%An example of configurations with the one-to-one correspondence problem
%for a twist of a vertex. 
We can regard a flux of the center plaquette
as a twist of either upper or lower vertex, but not both.  
In terms of the auxiliary field,  
it is nothing but the phenomena  explained in Fig.\ref{f:ambiguity1}.
% in the non-2PI diagrams (right). 
}
\label{f:ambiguity2}
\end{figure}
%%%%%%%%%%%%%%%%%%%%%%%%%%%%%%%%%%%%%%%%%%

The above observation leads us to express EE in the 2PI formalism with the auxiliary fields. 
Although we cannot rewrite the action itself by using all the auxiliary fields simultaneously, 
the vertex contributions to the free energy and EE can be written as a %direct 
sum of the contributions from these three. 
%It is because of the one-to-one correspondence explained above. 
The result is given by\footnote{%------------------------------------
Diagrams with tadpoles (one-point functions)  are cancelled  due to to the equation of motion. 
Namely, in calculating the 1PI free energy, 
 an appropriate source term is introduced 
depending on $M$ so that the equation of motion is always satisfied. 
} %--------------------------------------
\aln{
S^{\text{2PI}}_{\text{vert}}=&- \frac{V_{d-1} }{12} 
 \biggl(\int^{1/\epsilon} \frac{d^{d-1}k_{\parallel}} {(2 \pi)^{d-1}}  
 \tr \log \left[\tilde{G}_{c}^{-1}(\bm{0}; k_\parallel)   \right]
  \nonumber\\
 %&\hspace{1.6cm}
 &+2\int^{1/\epsilon} \frac{d^{d-1}k_{\parallel}} {(2 \pi)^{d-1}}  
  \log \left[ \tilde{G}_d^{-1}(\bm{0}; k_\parallel) \right] %\nonumber\\
% & \hspace{1.6cm}
 +2\int^{1/\epsilon} \frac{d^{d-1}k_{\parallel}} {(2 \pi)^{d-1}}  
  \log \left[ \tilde{G}_{d^\prime}^{-1}(\bm{0}; k_\parallel) \right]\biggr).
\label{e:EE2PIvertC}
}
Here, $(\tilde{G}_{c})_{ij}$, $\tilde{G}_d$ and $\tilde{G}_{d'}$ is the Fourier transformations of the two-point functions $\left<c_i(x)c_j(y)\right>$, $\left<d(x)\bar{d}(y)\right>$, and $\left<d'(x)\bar{d'}(y)\right>$ %, respectively. 
and the first, second, and third terms in \eqref{e:EE2PIvertC} 
represent the vertex contributions from the $s$-, $t$- and $u$-channel openings, respectively. The coefficients ``2'' in the second and third lines come from the fact that $(d,\bar{d})$ and $(d', \bar{d'})$ are complex fields. 
$(c_1,c_2)$ are real fields, but its propagator is written as a $2\times2$ matrix and has two degrees of freedom.  
The $\tr$ is the trace taken over this $2\times2$ matrix. 

%%%%%%%%%%%%%%%%%%%%%%%%%%%%%%%%%%%%%

The above model is simple in the sense that the auxiliary field of each $s$, $t$, and $u$-channel is different and 
the correspondence between twisting a vertex and twisting  propagator  of each auxiliary field is clear. 
Let us then consider a less easy (though seemingly easier) case, namely the $\phi^4$ theory with a single real scalar. 
The action written with an auxiliary field $c$  takes the following form:
\aln{
I_{stu}=\int\frac{d^{d+1}x}{M}\left[\frac{1}{2}\phi(-\Box+m^2_{0})\phi+\frac{1}{2}c^2+i\sqrt{\frac{\lambda}{2}}\,c\,\phi^2\right]. 
\label{e:actionc}
}
In order to reproduce the vertex contributions to EE in the 
original  $\lambda\phi^4/4$ theory, we need to sum all the contributions from the three different channels for $c$. 
If we use the above action,  the free energy in flat space can be reproduced, 
but not the free energy of the orbifold theory. 
Thus, we cannot use the renormalized two-point function of $c$ via $\log G_c^{-1}$ to express the correct
amount of vertex contributions to EE. 
EE in $\phi^4$ theory 
 is neither expressed by a single auxiliary field $c$ nor by triple copies of it
because the three channels coincide and  get mixed among them. 

%%%%%%%%%%%%%%%%%%%%%%%%%%%%
%%%%%%%%%%%%%%%%%%%%%%%%%%%%

\section{Unified formula for propagator and vertex contributions and its generalization to various interactions and higher spins}\label{sec:unified}
In Section \ref{sec:prop-2PI}, we resum the propagator contribution by the 2PI formalism. In Section \ref{sec:auxiliary}, we resum the vertex contribution by regarding it as propagator contributions of auxiliary fields. There are two unsatisfactory points. 1) We need to introduce an auxiliary field for
each channel of opening a vertex appropriately. This makes the analysis theory-dependent. Especially, this method does not work for the $\phi^4$ theory with a single real scalar. 2) The propagator and vertex contributions are treated separately in a different manner. To resolve these issues, we write EE in terms of composite operators and introduce a notion of a \textit{generalized} 1PI\index{generalized 1PI} to unify the framework of the propagator and vertex contributions. See also Appendix \ref{appencomp} for a detailed proof. These formulation allows us to consider various models including derivative interactions and higher-spin fields.

\subsection{Vertex contribution as propagator contribution of composite operators}
The expression of EE from each auxiliary field \eqref{e:EE2PIvertC} has a remarkable interpretation. Note that we can regard the auxiliary fields as degrees of freedom of composite operators\index{composite operators}:
\aln{
c_1\sim \bar{\phi}_2\phi_2,~&~~c_2\sim \bar{\phi}_1\phi_1,\\
d\sim\bar{\phi}_1\bar{\phi}_2,~&~~\bar{d}\sim\phi_1\phi_2,\\
d'\sim\bar{\phi}_2\,\phi_1,~&~~\bar{d'}\sim\bar{\phi}_1\phi_2.
}
They are justified in various ways,  for instance, the vacuum expectation values of both sides coincide. From this viewpoint, \eqref{e:EE2PIvertC} indicates that vertex contributions are in fact understood as propagator contributions of the composite operators. 
 From the actions \eqref{e:actions}, \eqref{e:actiont}, and \eqref{e:actionu}, the propagators of auxiliary fields are 
written in terms of correlation functions of the above composite operators as
\aln{
\tilde{G}_{cij}  &=(\sigma_x)_{ij} -\frac{\lambda}{4} \tilde{G}_s(\bm{k},k_\parallel)_{ij},
\\
\tilde{G}_{d} &=1 -\frac{\lambda}{4} \tilde{G}_t(\bm{k},k_\parallel),
\\
\tilde{G}_{d^\prime} &=1- \frac{\lambda}{4} \tilde{G}_u(\bm{k},k_\parallel),
}
where 
\aln{
\tilde{G}_s(\bm{k},k_\parallel)_{ij}&=\int d^2\bm{r}\,d^{d-1}r_\parallel\,e^{-i(\bm{k}\cdot\bm{r}+ik_\parallel\cdot r_\parallel)}\left<[\bar{\phi}_j \phi_j](\bm{r};r_\parallel)~[\bar{\phi}_i\phi_i](\bm{0};0)\right>,\\
\tilde{G}_t(\bm{k},k_\parallel)&=\int d^2\bm{r}\,d^{d-1}r_\parallel\,e^{-i(\bm{k}\cdot\bm{r}+i k_\parallel\cdot r_\parallel)}\left<[\bar{\phi}_1\bar{\phi}_2](\bm{r};r_\parallel)~[\phi_1\phi_2](\bm{0};0)\right>,\\
\tilde{G}_u(\bm{k},k_\parallel)&=\int d^2\bm{r}\,d^{d-1}r_\parallel\,e^{-i(\bm{k}\cdot\bm{r}+ik_\parallel\cdot r_\parallel)}\left<[\bar{\phi}_2\phi_1](\bm{r};r_\parallel)~[\bar{\phi}_1\phi_2](\bm{0};0)\right>.
}
The square brackets  $[{\cal O}]$ represent the normal ordering of an operator ${\cal O}$. 
The resulting contributions to EE, including both of those from the propagators and vertices, are given by 
\aln{
S^{\text{2PI}}_{\text{prop\&vert}}=&- \frac{V_{d-1} }{6} 
 \biggl( \sum_{i=1}^2\int^{1/\epsilon} \frac{d^{d-1}k_{\parallel}} {(2 \pi)^{d-1}}  
  \log \left[ \tilde{G_{\phi_i}}^{-1}(\bm{0}; k_\parallel) \epsilon^2 \right]
  \nonumber\\
 & \hspace{1.6cm}
 {-} \frac{1}{2} \int^{1/\epsilon} \frac{d^{d-1}k_{\parallel}} {(2 \pi)^{d-1}}  
 \tr  \log \left[ \sigma_x -\frac{\lambda}{4}\tilde{G_s}(\bm{0}; k_\parallel) \right]
  \nonumber\\
 & \hspace{1.6cm}
{-} \phantom{\frac{1}{2}} 
 \int^{1/\epsilon} \frac{d^{d-1}k_{\parallel}} {(2 \pi)^{d-1}}  
  \log \left[ 1-\frac{\lambda}{4} \tilde{G_t}(\bm{0}; k_\parallel)  \right]
  \nonumber\\
 & \hspace{1.6cm}
{-} \phantom{\frac{1}{2}} 
 \int^{1/\epsilon} \frac{d^{d-1}k_{\parallel}} {(2 \pi)^{d-1}}  
  \log  \left[1-\frac{\lambda}{4}\tilde{G_{u}}(\bm{0}; k_\parallel) \right] \biggr) ,
\label{e:EE2PIC}
} 
where the $\tr$ in the second line is a trace over the $ 2 \times 2$ matrix.

%%%%%%%%%%%%%%%%%%%%%%%%%%%%%%%%%%%%%%%%%%%%%%%%
%%%%%%%%%%%%%%%%%%%%%%%%%%%%%%%%%%%%%%%%%%%%%%%%
Remarkably, this composite operator approach works for the $\phi^4$ theory with a single real scalar while the auxiliary field approach failed to give a correct contribution to EE as we discussed at the end of Section \ref{sec:auxiliary}.
As the previous observation indicates, we will now focus on the following correlation function,
	\aln{
		G_{\phi^2\phi^2}(x-y):= \left<\,[\phi^2](x)\, [\phi^2](y)\,\right>.
		\label{eq:composite-green}
	}
	The vertex contributions to EE in the $\phi^4$ theory is expected to be given by 
	\aln{
		S^{\text{2PI}}_{\text{vert}}=
		\frac{V_{d-1}}{12}
		\int^{1/\epsilon}\frac{d^{d-1}k_\parallel}{(2\pi)^{d-1}}\mathrm{log}\,\left[1-\frac{3}{2}\lambda \,\tilde{G}_{\phi^2\phi^2}(\bm{0},k_\parallel)
		\right] .
		\label{e:EEcomp}
	}
	Here, the coefficient $-3\lambda/2$ is understood as $-\lambda/4\times 6$ %the former of which is 
	where $-\lambda/4$ is the coefficient in front of the interaction vertex (the same coefficient as in \eqref{e:EE2PIC}) 
	and the coefficient 6 is the combinatorial factor for separating four $\phi(x)$'s into a pair of  two $\phi(y)$'s. 
	The unity in the logarithm in \eqref{e:EEcomp} means that 
	the composite operator does not have any new degrees of freedom  in the free field limit and does not contribute to EE. 
	The overall factor is not $1/6$ but $1/12$ since the composite operator is real. 
%	To summarize, We expect that Eq.(\ref{e:EEcomp}) is a consistent counterpart to Eq.(\ref{e:EE2PIvertC}) (or equivalently, the second term and the followers in Eq.(\ref{e:EE2PIC})), which describes the vertex correction in the $\phi^4$ theory and $1-(3\lambda/2)G_{stu}$ should correspond to the full propagator of the necessary degrees of freedom for opening vertices, even though we cannot equip it at the level of action. 

Since we cannot introduce the auxiliary field and use the conventional  2PI formalism, we do not know at this point how to justify that the above expression of \eqref{e:EEcomp} gives the correct vertex contribution to EE. This will be proved by considering a one-loop diagram consisting of 1PI diagrams for the composite operator (generalized 1PI diagrams). We will discuss this in the next section \S\ref{s:vertexEE} and Appendix \ref{appencomp}. 
Instead, we will perturbatively check its correctness up to  $\lambda^2$ in the following. 
The two-point function of the composite operator can be evaluated as 
	\aln{
		G_{\phi^2\phi^2}=2A-6\lambda A^2-12\lambda B+O(\lambda^2),
		\label{e:comppropexpand}
	}  
	where
	\aln{
		A &:=G_{0}(x-y)^2, \nn
		B& :=\int d^{d+1}z\,G_{0}(x-y)G_{0}(x-z)G_{0}(z-y)G_{0}(0).
	}
	In \eqref{e:comppropexpand},  
	the product of operators represents a convolution; 
	\[XY(x-y)=\int d^{d+1}z X(x-z)Y(z-y).\] 
	By substituting \eqref{e:comppropexpand} into \eqref{e:EEcomp}, and using the identity
	\[\int\frac{d^{d-1}k_\parallel}{(2\pi)^{d-1}}\tilde{f}(\bm{0},k_\parallel)=\int d^2\bm{r}f(\bm{r},0),\]
	we can expand \eqref{e:EEcomp} up to $O(\lambda^2)$ as 
	\aln{
		S^{\text{2PI}}_{\text{vert}}&=\frac{V_{d-1}}{12}\int d^{2}\bm{r}\left[\log\,\left(1-\frac{3}{2}\lambda G_{\phi^2\phi^2}\right)\right](\bm{r},0)\nonumber\\
		&=\frac{V_{d-1}}{12}\int d^{2}\bm{r}\left[-\frac{3}{2}\lambda G_{\phi^2\phi^2}-\frac{9}{8}\lambda^2 G_{\phi^2\phi^2}^2\right](\bm{r},0)
		+ {\cal O}(\lambda^3)
		\nonumber\\
		&=\frac{V_{d-1}}{12}\int d^{2}\bm{r}\left[-3\lambda A+18\lambda^2B+\frac{9}{2}\lambda^2 A^2\right](\bm{r},0)
		+ {\cal O}(\lambda^3) \nonumber\\
		&=-\frac{V_{d-1}}{4}\lambda\int d^2\bm{r}\,G_{0}(\bm{r},0)^2\nonumber\\
		&~~+\frac{3V_{d-1}}{2}\lambda^2\int d^2\bm{r}\,d^2\bm{s}\,d^{d-1}r_\parallel\, G_0(\bm{r},0)\, G_0(\bm{r}-\bm{s},r_\parallel)\,G_0(\bm{s},r_\parallel)G(\bm{0},0)\nonumber\\
		&~~+\frac{3V_{d-1}}{8}\lambda^2\int d^2\bm{r}\,d^2\bm{s}\,d^{d-1}r_\parallel\, G_{0}(\bm{r},r_\parallel)^2\,G_0(\bm{s},r_\parallel)^2
+ {\cal O}(\lambda^3)	.}
	These three terms indeed coincide with \eqref{e:EE2loopvert}, \eqref{e:EE3loopvert2}, and \eqref{e:EE3loopvert}, respectively.

\subsection{Unification of propagator and vertex contributions using generalized 1PI diagrams}
\label{s:vertexEE}
In the previous section, we discussed the composite operator approach works well over the auxiliary field method. In the same manner as the propagator contribution, we need to replace a vertex in the bubble diagrams with a twisted one to compute EE. When we sum up these twisted diagrams, we need to take care of multiplicities by the divide-by-multiplicity method. Just as we did in \eqref{eq:divide-by-mult-prop}, we show that the composite operator approach is indeed correct to all orders in Appendix \ref{appencomp}.
The key is to replace the usual 1PI self energy with a \textbf{generalized 1PI diagram}\index{generalized 1PI}, the self-energy for a composite operator.

As shown in Fig.\ref{f:G-1PI}, 
the Green function of the composite operator can be written as
\aln{
G_{\phi^2\phi^2} = \,&\Sigma^{(g)}_{\phi^2\phi^2}
+\left(\frac{-3\lambda_4 }{2} \right)(\Sigma^{(g)}_{\phi^2\phi^2})^2 + 
\left(\frac{-3\lambda_4 }{2} \right)^2 (\Sigma^{(g)}_{\phi^2\phi^2})^3 + \cdots
\nn
=\, & \frac{\Sigma^{(g)}_{\phi^2\phi^2}} {1-  \left(\frac{-3\lambda_4 }{2} \right) \Sigma^{(g)}_{\phi^2\phi^2}},
\label{G=sumofg1PI}
}
where $\Sigma^{(g)}_{\phi^2\phi^2}$ is the 1PI self-energy of $[\phi^2]$  %and 
{in a \textit{generalized} sense.} We call it g-1PI\index{{g-1PI}}. 
Namely,  the quantity with the superscript $(g)$ 
 does not contain a diagram like Fig.\ref{f:g1PI}
that is separable by cutting an arbitrary vertex in the middle. 
We call such a diagram a beads diagram:
%that is 
1PI in the ordinary sense but not in the generalized sense. Thus these beads diagrams 
are not included in g-1PI diagrams. %$\Sigma^{(g)}$.
%%%%%%%%%%%%%%%%%%%%%%%%%%%%%%%%%%%%%%%%
\begin{figure}[h]
\centering
\includegraphics[width=\linewidth]{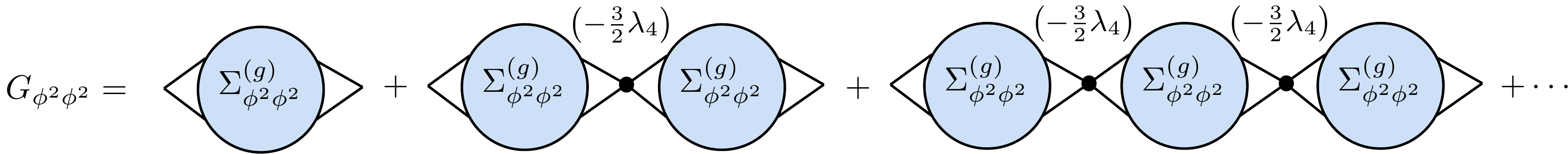}
\caption{
A Green function of a composite operator can be written in terms of the generalized self-energy $\Sigma^{(g)}_{\phi^2\phi^2}$,
which is 1PI with respect to the propagator of the composite operator at the vertex as well as the fundamental field. 
}
\label{f:G-1PI}
\end{figure}
%%%%%%%%%%%%%%%%%%%%%%%%%%%%%%%%%%%%%%%%%%%
%%%%%%%%%%%%%%%%%%%%%%%%%%%%%%%%%%%%%%%%%%  
\begin{figure}[h]
\centering
\includegraphics[width=7cm]{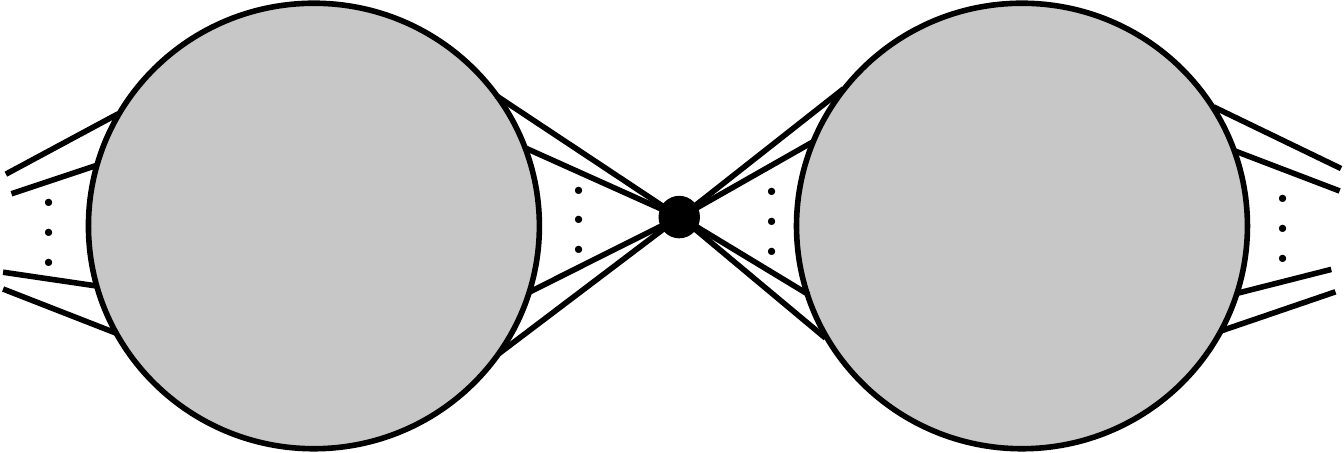}%{ambiguity2_ver2.pdf}
\caption{
Beads diagram: this diagram is NOT  1PI in the {\it generalized} sense 
since it is separable by cutting the  propagator at the opened vertex.
}
\label{f:g1PI}
\end{figure}
%%%%%%%%%%%%%%%%%%%%%%%%%%%%%%%%%%%%%%%%%%

From the diagrammatic analysis (see Appendix \ref{appencomp}), we can show that
the vertex contribution to EE is written in terms of a correlation function of composite operators. Since we consider various interactions in this section, here and in the following, we denote the coupling constant of the $\phi^4$ theory as $\lambda_4$, i.e. the interaction term in the Lagrangian is given by ${\cal L}_{pot}=\lambda_4 \phi^4 / 4$.
In the case of the $\phi^4$ theory, the vertex contribution is given by
\aln{
		S_{\text{vertex}}= 
		\frac{V_{d-1}}{6}
		\int^{1/\epsilon}\frac{d^{d-1}k_\parallel}{2(2\pi)^{d-1}}\mathrm{log}\,\left[1 - \frac{3}{2}\lambda_4 \,
		G_{\phi^2 \phi^2}(\bm{0},k_\parallel)
		\right] ,
		\label{e:EEcomp-ph4}
	}
where $G_{\phi^2\phi^2}$ is given by \eqref{eq:composite-green}.
In the following, the cutoff $\epsilon$ is not explicitly written for notational simplicity, 
as it can be recovered by dimensional analysis. 
The coefficient $-3\lambda_4 /2$ is a product of  $-\lambda_4 /4$ %\times 6$ 
and $6$, 
where $-\lambda_4/4$ is the coefficient in front of the interaction vertex 
and the coefficient $6$ is a combinatorial factor for separating four $\phi(x)$'s into 
{a pair of $\phi^2(x)$ and $\phi^2(y)$.} %a pair of  two $\phi(y)$'s. 

By using \eqref{G=sumofg1PI},
we can rewrite \eqref{e:EEcomp-ph4} as
\aln{
S_{\text{vertex}}= 
		 - \frac{V_{d-1}}{6}
		\int^{1/\epsilon}\frac{d^{d-1}k_\parallel}{2(2\pi)^{d-1}}\mathrm{log}\,\left[1 -  \left( -\frac{3}{2}\lambda_4 \right) \,
		 \Sigma^{(g)}_{\phi^2\phi^2}
		\right] .
		\label{e:EEcomp2}
}
In the following equations including \eqref{e:EEcomp2}, the argument $(\bm{k}=\bm{0},k_\parallel)$ of the integrand for the $k_\parallel$ integral is implicit.
Now we can write 
both the propagator and vertex contributions in \eqref{e:EE2PIprop} and \eqref{e:EEcomp2}
in a unified matrix form as
\aln{
S_{EE}^{(\phi^4)} = - \frac{V_{d-1}}{6}
		\int^{1/\epsilon}\frac{d^{d-1}k_\parallel}{2(2\pi)^{d-1}}\mathrm{tr}\,\mathrm{log}
	\left[\hat{G}_0^{-1} - \hat\lambda \hat{\Sigma}^{(g)} \right],
	\label{eq:EE_phi4}
}
where
\aln{
  \hat{G}_0 =&   \left( \begin{array}{cc} G_{0} & 0 \\0 & 
  1  \end{array}\right)  , 
  \   \ 
  \hat\lambda = \left(\begin{array}{cc}1  & 0   \\ 0  & -3 \lambda_4/ 2 \end{array}\right) ,   
  \   \
 \hat{\Sigma}^{(g)} =    \left(\begin{array}{cc}  \Sigma^{(g)}    &{0}   %\red{\epsilon^{2}}
  \\ 0  & \Sigma^{(g)}_{\phi^2\phi^2}\end{array}\right) .
  \label{GLS1}
 }

In the following, we generalize these results to include higher-point vertices whose composite operators %will be
are mixed in a complicated way. 
It is important to note that the form of \eqref{eq:EE_phi4} is convenient for a unified description in the following discussions, but it is 
always possible to go back to the form like \eqref{e:EEcomp-ph4}, where 
the vertex contributions are written  in terms of the ordinary renormalized propagators without the superscript ${}^{(g)}$. 
Also, note that all the single twist contributions from a vertex can be written in the above one-loop type formula, \eqref{e:EEcomp-ph4} or \eqref{eq:EE_phi4}.

\subsection{$\phi^6$ scalar field theory}
In the subsequent sections, we extend the analysis of vertex contributions to EE from the $\phi^4$ interaction to more general cases.
First, let us consider the $\phi^6$ interaction,
\aln{
{\cal L}_{\mathrm{pot}}=\frac{\lambda_6}{6 } \phi^6.
}
In this case, we  have  two types of vertex configurations\footnote{
{
The $\phi^6$ interaction will induce $\phi^4$ interaction by contracting
two $\phi$'s, but in this section, we simply set it zero by renormalization
and do not consider contributions to EE from such diagrams as the vertex contributions at this stage.
A model containing both of $\phi^4$ and $\phi^6$
interaction vertices are studied in the next section. 
%It should be understood that the $\phi^6$ interaction vertex is normal-ordered and 
%the renormalized $\phi^4$ interaction vertex is set to zero here. A theory with both of $\phi^4$ 
%and $\phi^6$ interaction vertices are considered in the next section.
%We take contributions with the singular four-point vertex $-(5\lambda_6/2)\,G_0(0)\,\phi^4$ into $\Sigma^{(g)}$.
 }
}
 as drawn in Fig.\ref{f:6point-vertex}
%%%%%%%%%%%%%%%%%%%%%%%%%%%%%%%%%%%%%%%%%%  
\begin{figure}[t]
\centering
\includegraphics[width=7cm]{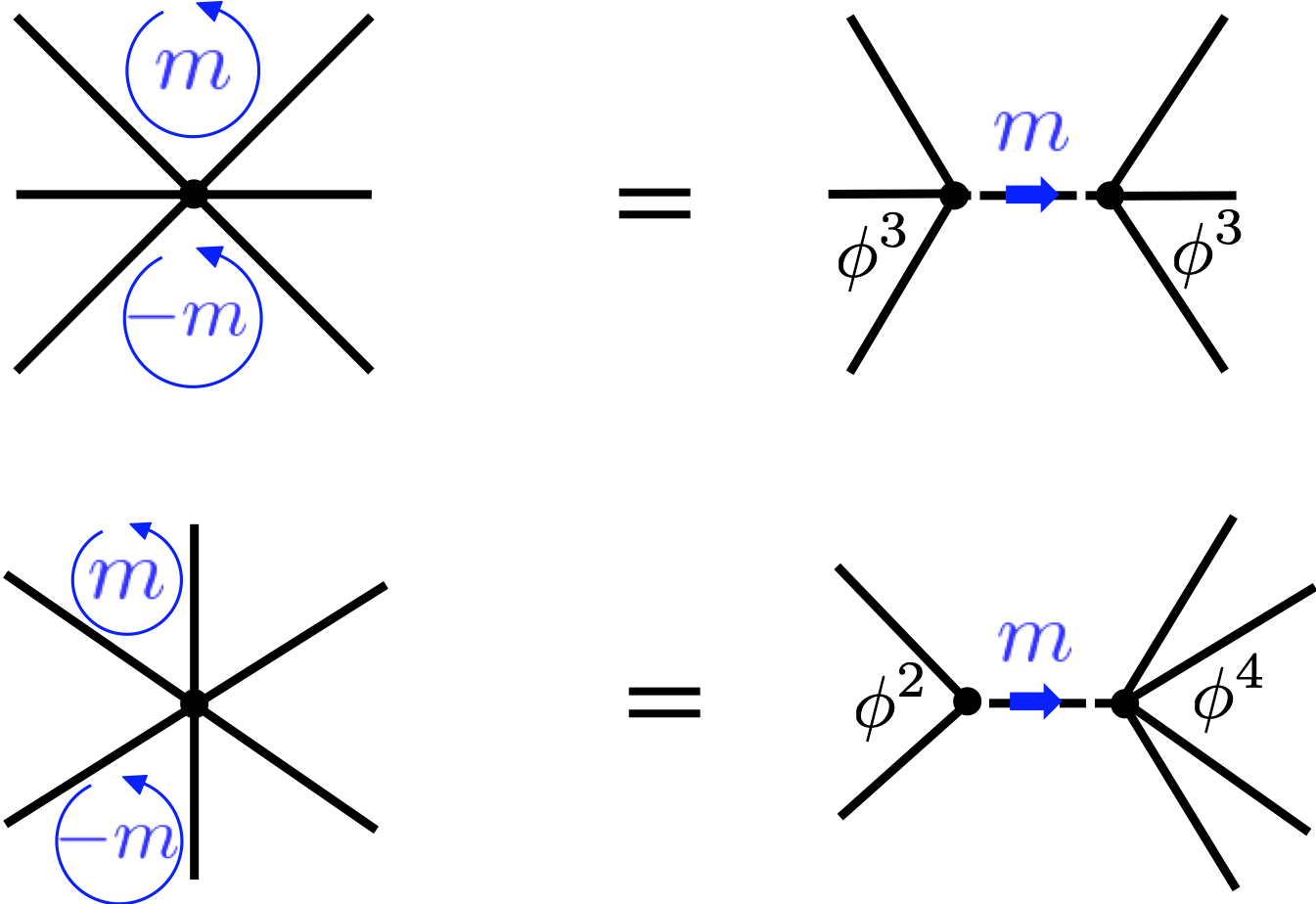}%{ambiguity2_ver2.pdf}
\caption{
Two different composite operators appear by
opening the $\phi^6$ vertex. Each flux configuration
corresponds to twisting the propagator of the respective composite operator. 
}
\label{f:6point-vertex}
\end{figure}
%%%%%%%%%%%%%%%%%%%%%%%%%%%%%%%%%%%%%%%%%%
and need to introduce three types of composite operators, $\phi^2$, $\phi^4$, and $\phi^3$, 
to extract all the vertex contributions to EE. Since the theory has $\mathbb{Z}_2$ invariance under $\phi \rightarrow -\phi$, 
the $\mathbb{Z}_2$-even operators,  $\phi^2$ and $\phi^4$, are mixed with themselves 
while the $\mathbb{Z}_2$-odd operator $\phi^3$ is mixed with the fundamental field $\phi$.
Therefore,  the propagator 
contribution in \eqref{e:EE2PIprop} needs a  modification. 

First, let us consider the modified propagator contributions in the $\phi^6$ theory. 
%%%%%%%%%%%%%%%%%%%%%%%%%%%%%%%%%%%%%%%%%%  
\begin{figure}[t]
\centering
\includegraphics[width=5cm]{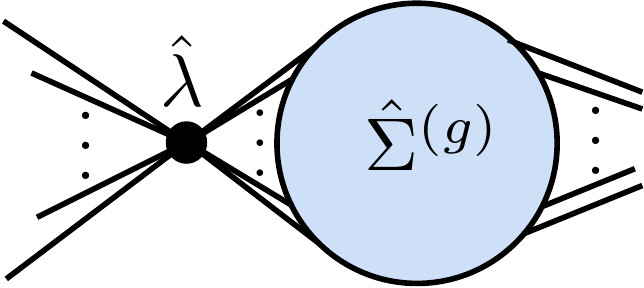}%{ambiguity2_ver2.pdf}
\caption{
A Schwinger-Dyson type diagram to represent  mixings between different operators. All possible composite operators
are assigned to each dotted part. $\hat\lambda$ is a matrix-valued vertex  
and $\Sigma^{(g)}$ is a generalized 1PI (g-1PI) self-energy with respect to
the composite operators. 
}
\label{f:13mixing}
\end{figure}
%%%%%%%%%%%%%%%%%%%%%%%%%%%%%%%%%%%%%%%%%%
Such contributions come from one-loop type diagrams of
mixed correlations of $\phi$ and $\phi^3$ operators. 
They are given by (Fig.\ref{f:13mixing})
\aln{
 S_{\mathbb{Z}_2\text{-odd}}
 =   &  - \frac{V_{d-1} }{6} 
 \int \frac{d^{d-1}k_{\parallel}} {2(2 \pi)^{d-1}} \tr \log   
 \left[\hat{G}_0^{-1} - \hat\lambda\hat{\Sigma}^{(g)}  \right],
 \label{EEphi6-Z2odd}
 }
 where %$\tilde{G}_g^{-1} = \tilde{G}_g^{(0)-1} -\Sigma_g$  and
 \aln{
  \hat{G}_0 =&   \left( \begin{array}{cc} G_0  & 0 \\0 & 
  1  \end{array}\right)  , 
  \   \ 
 \hat\lambda= \left(\begin{array}{cc}1  & 0   \\ 0  & -10 \lambda_6  /3 \end{array}\right) ,   
  \   \
  \hat{\Sigma}^{(g)} =    \left(\begin{array}{cc}  \Sigma^{(g)}  \, & \Sigma^{(g)}_{\phi\phi^3}
  \\ \Sigma^{(g)}_{\phi^3\phi} & \Sigma^{(g)}_{\phi^3\phi^3}\, \end{array}\right) .
  \label{GLS}
 }
 It is a natural generalization of \eqref{GLS1} including an operator mixing. 
The diagonal component of $\hat{G}_0$ is the bare propagators of $\phi$ and $\phi^3$ operators, respectively.
$\hat\lambda$ is a matrix whose matrix element 
represents the coefficients of opening the $\phi^6$ vertex.
{The coefficient for $\phi^3$ to $\phi^3$ in $\hat\lambda$ is given by $1/6 \times _6 C_3$.} 
$\hat{\Sigma}_{}^{(g)} = \hat{\Sigma}_{}^{(g)} ({\bm k}=\bm{0}, k_\parallel)$ is the $\mathbb{Z}_2$-odd
g-1PI\index{{g-1PI}} function.\footnote{
{In the $\mathbb{Z}_2\text{-odd}$ set of operators, 
the $[\phi^5]$ operator does not appear in the  mixing, though
the $\phi^6$ vertex can be decomposed into $\phi$ and $\phi^5$. 
%It is because  the g-1PI diagram is 1PI in the ordinary sense and 
% a diagram with $\langle \,\phi\,[\phi^5]\,\rangle$ is not 1PI.
It is because a diagram with $\langle\, \phi\, [\phi^5]\, \rangle$ is not 1PI while the g-1PI is 1PI as well in the ordinary sense. } 
}  
Namely, it consists of 1PI diagrams that do not contain beads diagrams shown
 in Fig.\ref{f:g1PI}.
 Such a generalization of the 1PI concept is mandatory since, in calculating 
the vertex contributions to EE,  
we need to open a vertex to take account of various  channel contributions 
and  special care of the beads diagram in Fig.\ref{f:g1PI} is necessary. 
This is the reason why we have generalized the concept of 1PI.

The above discussions can be straightforwardly extended to the contributions from 
$\mathbb{Z}_2$-even operators,  $\phi^2$ and $\phi^4$. 
This case is simpler because the bare Green function is unity; $G^{(g)}=1$. 
Then, we have the same matrix form
\aln{
S_{\mathbb{Z}_2\text{-even}} = &  - \frac{V_{d-1} }{6} 
 \int^{1/\epsilon} \frac{d^{d-1}k_{\parallel}} {2(2 \pi)^{d-1}} \tr \log (\hat{1}-
 \hat\lambda 
 \hat{\Sigma}^{(g)} ),
 \label{EEphi6-Z2even}
}
where,  in this case,  matrices are given by
\aln{
 \hat\lambda
  = &
 \left(\begin{array}{cc}0 & -5 \lambda_6/2  \\-5 \lambda_6/{2} & 0\end{array}\right), \ \ 
 \hat{\Sigma}^{(g)} =
\left(\begin{array}{cc}\Sigma^{(g)}_{\phi^2\phi^2}&  \Sigma^{(g)}_{\phi^2\phi^4}
\\ \Sigma^{(g)}_{\phi^4\phi^2}& \Sigma^{(g)}_{\phi^4\phi^4}  \end{array}\right).
\label{X}
}
The coefficient comes from $5/2 =1/6 \times {}_6 C_2$. 
It is a $2 \times 2$ matrix generalization of \eqref{e:EEcomp-ph4}.
The g-1PI self-energy $\hat\Sigma^{(g)}$ does not contain 
beads diagrams, especially diagrams connected by 
the $\phi^6$ vertex decomposed into $\phi^2$ and $\phi^4$. 
%\footnote{\red{ Due to the $\mathbb{Z}_2$ invariance, 
%beads diagrams connected by $\phi^6$ vertex decomposed into
% two $\phi^3$'s vanish in the $\mathbb{Z}_2$ even sector.}}

Note that  EE of \eqref{EEphi6-Z2odd} and \eqref{EEphi6-Z2even} written in terms of the g-1PI functions 
can be rewritten  in terms of the renormalized correlation functions
as in the $\phi^4$ case of  \eqref{e:EEcomp-ph4}.
The only difference is that we now have operator mixings and the relationship becomes more complicated. 
Let us explicitly check it for the $\mathbb{Z}_2$-odd case of \eqref{EEphi6-Z2odd}.
It is rewritten as
\aln{
S_{\mathbb{Z}_2\text{-odd}}&=\frac{V_{d-1}}{6}\int^{1/\varepsilon}\frac{d^{d-1}k_\parallel}{2(2\pi)^{d-1}}\mathrm{tr}\,\mathrm{ln}\left(\hat{G}_0\,\frac{1}{\hat{1}-\hat{\lambda}_{}\hat{\Sigma}_{}^{(g)}\hat{G}_0}\right)\nonumber\\
&=\frac{V_{d-1}}{6}\int^{1/\varepsilon}\frac{d^{d-1}k_\parallel}{2(2\pi)^{d-1}}\mathrm{tr}\,\mathrm{ln}\left(\hat{G}_0+\hat{G}_0\hat{\lambda}_{}\hat{\Sigma}_{}^{(g)}\hat{G}_0+\hat{G}_0\hat{\lambda}_{}\hat{\Sigma}_{}^{(g)}\hat{G}_0\,\hat{\lambda}_{}\hat{\Sigma}_{}^{(g)}\hat{G}_0+\cdots\right).%\nonumber\\
%&=\frac{V_{d-1}}{6}\int^{1/\varepsilon}\frac{d^{d-1}k_\parallel}{2(2\pi)^{d-1}}\mathrm{tr}\,\mathrm{ln}\left[\hat{G}_0+\lambda_0\hat{G}_0\hat{\Sigma}_{}^{(g)}\hat{G}_0\left(1+\lambda_{}\hat{\Sigma}_{}^{(g)}\hat{G}_0
%+(\lambda_{}\hat{\Sigma}_{}^{(g)}\hat{G}_0)^2+\cdots\right)\right].
%&=\frac{V_{d-1}}{6}\int^{1/\varepsilon}\frac{d^{d-1}k_\parallel}{2(2\pi)^{d-1}}\mathrm{tr}\,\mathrm{ln}\left[\tilde{G}\right],
}
Writing the inside of the parenthesis as $\tilde{G}$, its matrix elements are given by
\aln{
(\tilde{G})_{11}&=G_0+G_0\Sigma^{(g)}G_0+G_0\Sigma^{(g)}_{\phi\phi^3}\left(-\frac{10}{3}\lambda_6\right)\Sigma^{(g)}_{\phi^3\phi}G_0+G_0\Sigma^{(g)}G_0\Sigma^{(g)}G_0+\cdots,\\
(\tilde{G})_{12}&=G_0\Sigma^{(g)}_{\phi\phi^3}+G_0\Sigma^{(g)}G_0\Sigma^{(g)}_{\phi\phi^3}+G_0\Sigma^{(g)}_{\phi\phi^3}\left(-\frac{10}{3}\lambda_6\right)\Sigma^{(g)}_{\phi^3\phi^3}+\cdots,\\
(\tilde{G})_{21}&=\left(-\frac{10}{3}\lambda_6\right)\left(\Sigma^{(g)}_{\phi^3\phi}G_0+\Sigma^{(g)}_{\phi^3\phi}G_0\Sigma^{(g)}G_0+\Sigma^{(g)}_{\phi^3\phi^3}\left(-\frac{10}{3}\lambda_6\right)\Sigma^{(g)}_{\phi^3\phi}G_0+\cdots,\right),\\
(\tilde{G})_{22}&=1+\left(-\frac{10}{3}\lambda_6\right)\left(\Sigma^{(g)}_{\phi^3\phi^3}+\Sigma^{(g)}_{\phi^3\phi}G_0\Sigma^{(g)}_{\phi\phi^3}+\Sigma^{(g)}_{\phi^3\phi^3}\left(-\frac{10}{3}\lambda_6\right)\Sigma^{(g)}_{\phi^3\phi^3}+\cdots,\right).
}
We can explicitly see that the sum of g-1PI's in each matrix element is combined into 
 the ordinary 1PI functions {$\Sigma$'s}, and hence can be written by the 
 renormalized correlation functions as
\aln{
(\tilde{G})_{11}&=G_0+G_0\Sigma G_0+G_0\Sigma G_0\Sigma G_0+\cdots    %\nonumber\\
=G,\\
(\tilde{G})_{12}&=(G_0+G_0\Sigma G_0+G_0\Sigma G_0\Sigma G_0+\cdots)\Sigma_{\phi\phi^3} %\nonumber\\
=G_{\phi\phi^3},\\
(\tilde{G})_{21}&=\left(-\frac{10}{3}\lambda_6\right)\Sigma_{\phi^3\phi}(G_0+G_0\Sigma G_0+G_0\Sigma G_0\Sigma G_0+\cdots) 
%\nonumber\\ &
=-\frac{10}{3}\lambda_6G_{\phi^3\phi},\\
(\tilde{G})_{22}&=1+\left(-\frac{10}{3}\lambda_6\right)\left(\Sigma_{\phi^3\phi^3}+\Sigma_{\phi^3\phi}G\Sigma_{\phi\phi^3}\right)
%\nonumber\\ &
=1-\frac{10}{3}\lambda_6G_{\phi^3\phi^3}.
}
As a result, \eqref{EEphi6-Z2odd} can be summarized as
\aln{
 S_{\mathbb{Z}_2\text{-odd}}
 =   &  \frac{V_{d-1} }{6} 
 \int^{1/\epsilon} \frac{d^{d-1}k_{\parallel}} {2(2 \pi)^{d-1}} \tr \log   
 \left[ \tilde{I}+ \hat\lambda_{}\hat{G}\right],
 \label{eq:gen-to-ren}
}
where
\aln{
\tilde{I}=\left(\begin{array}{cc}
0 & 0 \\
0 & 1
\end{array}\right)%,~~~
, 
\   \ 
{\hat\lambda= \left(\begin{array}{cc}1  & 0   \\ 0  & -10 \lambda_6  /3 \end{array}\right) ,   }
\   \
\hat{G}=\left(\begin{array}{cc}
G                   & G_{\phi\phi^3} \\
G_{\phi^3\phi}  & G_{\phi^3\phi^3}
\end{array}\right).
}
The same discussion can be applied to \eqref{EEphi6-Z2even}.
This gives  an alternative, unified formula for EE in terms of the renormalized Green functions. 
%%%%%%%%%%%%%%%%%%%%%%%%%%%%%%%%%%%
\subsection{$\phi^4 +\phi^6$  theory and further generalizations}
Let us generalize a bit more and consider a case when the Lagrangian contains two interaction terms
\aln{
{\cal L}_{pot}=\frac{\lambda_4}{4 } \phi^4 + \frac{\lambda_6}{6 } \phi^6 .
}
As in the $\phi^6$ theory, we need to consider three composite operators, $\phi^2$, $\phi^4$, and $\phi^3$,
in order to take into account contributions to EE from these vertices. 
Again, we have $\mathbb{Z}_2$ invariance and EE is a sum of $\mathbb{Z}_2$-even and odd contributions. 
The $\mathbb{Z}_2$-odd contribution is given by
\aln{
 S_{\mathbb{Z}_2\text{-odd}}
 =  & - \frac{V_{d-1} }{6} 
 \int^{1/\epsilon} \frac{d^{d-1}k_{\parallel}} {2(2 \pi)^{d-1}} \tr \log   
 \left[ 
 \hat{G}_0^{-1} -
 % \right.     \nn    &  \left.
\left(\begin{array}{cc}1  & 0%- \lambda_4  
\\ 0 %  - \lambda_4 
  & -10 \lambda_6/3 \end{array}\right)
 \left(\begin{array}{cc}  \Sigma^{(g)}_{}  & \Sigma^{(g)}_{\phi\phi^3}
  \\  \Sigma^{(g)}_{\phi^3\phi}  & \Sigma^{(g)}_{\phi^3\phi^3} \end{array}\right) 
 \right],
  \label{EEphi46-Z2odd}
 }
 where $\hat{G}_0$ is the same as in \eqref{GLS}
 while $\mathbb{Z}_2$-even contribution is given by 
 \aln{ 
S_{ \mathbb{Z}_2\text{-even}} = & - \frac{V_{d-1} }{6} 
 \int^{1/\epsilon} \frac{d^{d-1}k_{\parallel}} {2(2 \pi)^{d-1}} 
 \tr \log \left[ \hat{1}-
 \left(\begin{array}{cc} -3\lambda_4/2 & -5 \lambda_6/2  \\-5 \lambda_6/2 & 0\end{array}\right)
\left(\begin{array}{cc} \Sigma^{(g)}_{\phi^2\phi^2}  &  \Sigma^{(g)}_{\phi^2\phi^4}
\\ \Sigma^{(g)}_{\phi^4\phi^2}  & \Sigma^{(g)}_{\phi^4\phi^4} \end{array}\right)  
 \right]  .
 \label{EEphi46-Z2even}
}

Now a generalization to e.g.  $\phi^{2n}$ vertices with higher $n$ is evident.
The propagator and vertex contributions to EE are unified to be written in a  matrix form as 
\eqref{EEphi6-Z2odd}:
\aln{
 S_{EE}
 =   & - \frac{V_{d-1} }{6} 
 \int^{1/\epsilon} \frac{d^{d-1}k_{\parallel}} {2(2 \pi)^{d-1}} \tr \log   
 \left[ (\hat{G}_0^{-1} -  \hat\lambda \hat{\Sigma}^{(g)}  ) ({\bm k}=0, k_\parallel) \right].
 \label{EE-generalform2}
 }
 The size of matrices becomes larger as a larger number of operators are mixed and 
each set of mixed operators forms a block diagonal component. 
$\hat{G}_0$ is a diagonal matrix whose entry is mostly 1 except the fundamental field. 
$\hat{\lambda}$ represents a mixing among operators via vertices while $\hat{\Sigma}^{(g)}$ 
represents amputated correlators of all the fundamental and composite operators. 
The notion of the g-1PI is also extended to exclude all the beads diagrams 
constructed by all the vertices along with the ordinary non-1PI diagrams. 
This form of EE contains all the contributions from the propagators and the vertices. {We provide the derivation in Appendix \ref{appencomp}.}

An essential point is that we can rewrite \eqref{EE-generalform2} in terms of the renormalized correlation functions in the same manner as in {\eqref{eq:gen-to-ren}} as %Eq.(\ref{EEphi6-Z2odd}),
\aln{
S_{EE}=\frac{V_{d-1}}{6}\int^{1/\varepsilon}\frac{d^{d-1}k_\parallel}{2(2\pi)^{d-1}}\mathrm{tr}\,\mathrm{ln}\left(\tilde{I}+{\hat{\lambda}} %\lambda
 \hat{G}\right).
\label{EEren}
}
Here, $\tilde{I}=\text{diag}(0,1,\cdots,1)$, $\hat{G}$ is the matrix form of the correlators of operators, and we have arranged the elements of the matrices so that the first line and first column involve the fundamental field $\phi$. 
The size of the matrices is finite as far as there is a finite number of vertices. 
In the $\phi^n$-theory,
 we need to consider only the composite operators $[\phi^j]$ with $j\leq n-2$, which appear to open vertices. %The asymmetric structure with respect to the fundamental and composite operators reflects the Gaussian and non-Gaussian nature of EE.

%%%%%%%%%%%%%%%%%%%%%%%%%%%%%%%%%%%
\subsection{Derivative interactions}
Special care is necessary for generalizations with derivative interactions 
since composite operators with Lorentz indices appear. 
Let us consider the following interaction as an example, 
\aln{
{\cal L}_{pot}= \frac{\lambda_{{\partial}}}{4}(\phi \partial \phi)^2 . 
}
In this case,  the two types of scalar composite operators,  $[\phi^2]$ and $[(\partial \phi)^2]$,  as well as a
spin-1 operator $[\phi \partial_\mu \phi]$ appear from an opened vertex. 
Since the spin-1 operator does not mix with either $\phi$ or $[\phi^2]$ 
or $[(\partial \phi)^2]$, we can separately study its contribution to EE.  
Thus we have three block-diagonal sectors.

The spin-0 sectors can be treated as  before. 
Thus let us focus on the spin-1 sector. The formula \eqref{EE-generalform2} gets a bit modified since EE of 
% a spinning field is different from a scalar field because $\mathbb{Z}_M$ twist induces the rotation
% of the internal spin 
 a spinning field is different from that of a scalar field due to the rotation
of the internal spin induced by $\mathbb{Z}_M$ twist 
 and hence an extra phase appears in evaluating EE \cite{He:2014gva}. 
 The operator $J_\mu:= [\phi \partial_\mu \phi]$ is decomposed into its two-dimensional part ${\bm J}$
and $(d-1)$-dimensional part $J_i$. The latter is a scalar on %$d=2$ space-time 
{the two-dimensional spacetime normal to the boundary} and can be treated 
as in \ref{EE-generalform2}. 
On the other hand, the contribution to EE from the 2-dimensional vector ${\bm J}$ is modified. 
From (2.21) in \cite{He:2014gva}, the coefficient of EE is proportional to 
\aln{
c_{\rm eff}^{\rm boson}(s)=\frac{1}{4} \frac{\partial J(s, M)}{\partial M} \Big|_{M=1} =\frac{1}{6} -\frac{|s|}{2}
}
for a bosonic field with spin {$s$}%$2s$
. This coefficient $c_{\rm eff}$ replaces the coefficient of
$1/6$ in front of \eqref{e:EE2PIprop}. 
Thus for $(d+1)$-dimensional vector $J_\mu$, the total coefficient is 
given by $(d-1)/6+2(1/6-1/2)=(d-5)/6$.\footnote{See Appendix \ref{app:spinor} for more detail.}
Therefore the propagator and vertex contributions to 
EE with this derivative interaction is given by either of the following two forms, 
\aln{
 S_{EE}
 =   & - %V_{d-1} 
 {\frac{V_{d-1}}{6}}
 \int^{1/\epsilon} \frac{d^{d-1}k_{\parallel}} {2(2 \pi)^{d-1}} \tr \left( S \log   
 \left[ \hat{G}_0^{-1}-  \hat{\lambda} \hat{\Sigma}^{(g)}  \right] \right)\nonumber\\
 = & %V_{d-1} 
 {\frac{V_{d-1}}{6}}
 \int^{1/\epsilon} \frac{d^{d-1}k_{\parallel}} {2(2 \pi)^{d-1}} \tr \left( S \log   
 \left[ \tilde{I}+  \hat{\lambda} \hat{G}  \right] \right),
 \label{EE-derivative} 
 }
 where 
 \aln{
S&= %\frac{1}{6} 
\left(\begin{array}{cccc}1 & 0 & 0 & 0 \\0 & 1 & 0 & 0 \\0 & 0 & 1 & 0 \\0 & 0 & 0 & (d-5) \end{array}\right), \ 
\tilde{I}= \left(\begin{array}{cccc} 0 & 0 & 0 & 0 \\0 & 1 & 0 & 0 \\0 & 0 & 1 & 0 \\0 & 0 & 0 & 1\end{array}\right),
\ \ 
\hat{\lambda } = \left(\begin{array}{cccc}1 & 0 & 0 & 0 \\0 & 0 & -\lambda_{{\partial}} & 0 \\0 & -\lambda_{{\partial}} & 0 & 0 
\\0 & 0 & 0 & -\lambda_{{\partial}}/2\end{array}\right), \nn 
 \hat{G}^{}  &= \left(\begin{array}{cccc}
 {G} & 0 & 0 & 0 \\
 0 & G_{\phi^2\phi^2}  & G_{\phi^2(\partial\phi)^2} & 0 \\
 0 &  G_{(\partial\phi)^2\phi^2} & G_{(\partial\phi)^2(\partial\phi)^2} & 0 \\
 0 & 0 & 0 & G_{(\phi\partial_\mu\phi)(\phi\partial^\mu\phi)}\end{array}\right) .
}
$S$ is an additional coefficient due to the spin. Here we have summed over $(d+1)$-dimensional vector contributions,
but generally speaking,  it is more convenient to write a matrix corresponding to 
each irreducible representation of the 2-dimensional rotation with spin $s$. 
According to \cite{He:2014gva}, the coefficient $c$ for fermions with odd half-integer spin $s$ is given by
$c_{\rm eff}^{\rm fermion}(s) = -1/3$. Thus if we treat each 2-dimensional spin component as an independent
field, the diagonal component of the matrix $S{/6}$ is given by $c_{\rm eff}^{\text{boson/fermion}}(s)$
for each spin $s$ field. 
%\red{It is straightforward to rewrite Eq.(\ref{EE-derivative}) in terms of the renormalized two-point functions as Eq.(\ref{EEren}), taking account $S$.
% It is less trivial whether }

\subsection{Interactions with higher spin fields -- Generalities}\label{sec:spin-general}
General interactions involving higher spin fields similarly add an additional coefficient due to spins. 
A twisted propagator with a spin-$s$ field $\varphi_s$ is accompanied with a rotation in the internal space:
\aln{
G_{\varphi_s\,0}^{(M)}(x,y)=\sum_{m=0}^{M-1}e^{-2i\theta_m\mathcal{M}^{(s)}_{1,d+1}}\,G_{\varphi_s\,0}(\hat{g}^mx-y). 
}
Here, $\mathcal{M}_{1,d+1}^{(s)}$ is one of the generators of $SO(d+1)$ in the spin-$s$ representation, which drives a rotation on a plane spanned by $x_\perp$ (1-direction) and $\tau$ ($(d+1)$-direction). For example, the propagator for a Dirac fermion is given by
\aln{
G_{\text{fermion}}^{(M)}(x,y)=\sum_{m=0}^{M-1}e^{\theta_m \gamma_1\gamma_{d+1}}\int\frac{d^2\bm{k}}{(2\pi)^2}\frac{d^{d-1}k_\parallel}{(2\pi)^{d-1}}\frac{i\bm{k}\cdot\bm{\gamma}+ik_\parallel\cdot\gamma_\parallel-m_{0}}{k^2+m^2_{0}}e^{i(\bm{k}\cdot\hat{g}^m\bm{x}-\bm{k}\cdot\bm{y}+k_\parallel\cdot(x_\parallel-y_\parallel))}
}%Note that 
%If 
with $\bm{\gamma}=(\gamma_1,\gamma_{d+1})$ and $\gamma_\parallel=(\gamma_2,\cdots,\gamma_{d-2})$.\footnote{See Appendix \ref{app:spinor} for details.} 

In a bubble diagram, each propagator has such an additional rotational factor. However, since an interaction vertex is rotationally invariant, 
it is still invariant under $\mathbb{Z}_M$ rotation and consequently invariant under
 an overall twist of the adjacent propagators.\footnote{%---------------------------------
A simple example is a vertex in the $U(1)$ gauge theory, $(\gamma_\mu)_{\alpha\beta}$. It has one vector field and two spinor fields
 and is invariant under simultaneous rotations of the fields. 
} %---------------------- 
Suppose that we have a multi-point vertex of fields with spins $s_q$ ($q=1,2, \cdots$) and the coefficient is given by $C_{i_1i_2\cdots}$.
The $\mathbb{Z}_M$ invariance of the vertex is written as 
\aln{
C_{i_1i_2\cdots}\delta^2(\bm{p}_1+\bm{p}_2\cdots)&=(e^{2i\theta_m \mathcal{M}_{1,d+1}^{(s_1)}})_{i_1}^{~j_1}(e^{2i\theta_m \mathcal{M}_{1,d+1}^{(s_2)}})_{i_2}^{~j_2}\cdots C_{j_1j_2\cdots}\delta^2(\hat{g}^{m}(\bm{p}_1+\bm{p}_2\cdots)).
}
By decomposing each field  into irreducible representations of $SO(2)$, this simply means that  a sum of $SO(2)$ spins vanish
at each vertex. Due to the invariance, the basic framework of $\mathbb{Z}_M$ gauge theory on Feynman diagrams is not
changed. Namely, we can classify $\mathbb{Z}_M$ invariant configurations of twists in terms of  fluxes in plaquettes as before.
The additional phase associated with spins can be calculated by taking a special gauge of $\mathbb{Z}_M$ fluxes because of their gauge invariance.
Furthermore, for fermions, we have to replace the twist operator $\hat{g}$ with $\hat{g}^2$ due to the anti-periodic boundary condition. 
In this case, $M$ should be considered as an odd integer. 

As far as the contributions from the propagators and vertices are concerned, 
it is sufficient to consider a  twist of a particular propagator of fundamental or composite operators. For the propagator contributions from a general bosonic or fermionic field $\varphi_s$ with spin $s$, we can formally write down the free energy:
\aln{
\tilde{F}^{\text{2PI}}_{\varphi_s,\text{prop}}&=\frac{V_{d-1}}{2M}\sum_{m=1}^{M-1}\frac{1}{4\sin^2\theta_m}\mathrm{Tr}\left[e^{2i\theta_m \mathcal{M}^{(s)}_{1,d+1}}\,\int\frac{d^{d-1}k_\parallel}{(2\pi)^{d-1}}\mathrm{log}\,\tilde{G}_{\varphi_s}(\bm{0}, k_\parallel)\right]~~(\text{for bosons}),
\label{e:freeenergyspinbos}\\
\tilde{F}^{\text{2PI}}_{\varphi_s,\text{prop}}&=-\frac{V_{d-1}}{2M}\sum_{m=1}^{M-1}\frac{1}{4\sin^22\theta_m}\mathrm{Tr}\left[e^{4i\theta_m \mathcal{M}^{(s)}_{1,d+1}}\,\int\frac{d^{d-1}k_\parallel}{(2\pi)^{d-1}}\mathrm{log}\,\tilde{G}_{\varphi_s}(\bm{0}, k_\parallel)\right]~~(\text{for fermions}).
\label{e:freeenergyspinfer}
}
``$\mathrm{Tr}$'' here represents the trace over the internal space. For vertex contributions, we need to pay further attention as they involve composite operators. When one considers a general composite operator such as $:\!\varphi_s\varphi'_{s'}\!:$, it is generically in a reducible representation of $SO(d+1)$. 
We should first decompose it into irreducible components, each of which corresponds to a different composite operator, then assign the rotational factor due to spins for each representation.

If we reduce \eqref{e:freeenergyspinbos} and \eqref{e:freeenergyspinfer} to the free field cases, we can easily evaluate the trace because both the rotational factor and $\tilde{G}_{\varphi\,0}(\bm{0};k_\parallel)$ are diagonalized in the basis of the eigenstates for $SO(2)$. 
The resulting EEs coincide with those in \cite{He:2014gva}. 
On the other hand, for interacting cases, $\tilde{G}_{\varphi}(\bm{0};k_\parallel)$ has off-diagonal components 
and we need to take a trace  of {the product of the rotational factor and the matrix-valued logarithmic terms} in a nontrivial way. While the computation becomes technically cumbersome, we can still conclude that the non-Gaussian part in EE is understood as contributions from renormalized two-point functions of composite operators while the Gaussian part is  a contribution from the fundamental fields.

\section{Multiple twist contributions}\label{sec:num-mult}
\subsection{Numerical computation for a figure-eight diagram}
So far we have focused on contributions from a single twist of either a propagator or a vertex. For a complete evaluation of EE in generic QFTs, we also need to consider multiple twist configurations. Naively, each summation of twists gives an $(M^2-1)$ factor and in total $O\left((M^2-1)^2\right)$ for multiple twist configurations. Then, these do not contribute to EE. However, they cannot be independent in general; this expectation fails. 
In this section, we present a numerical result for the figure-eight diagram (Fig.\ref{Fig3}) with nonzero $m_1$ and $m_2$ and a fixed integer $M$ to see how the multiple twist contribution behaves as a function of $M$.

Performing the integration of \eqref{8figure}, we have
 \aln{
&\phantom{=}  \int d^{d+1} x \ G_0(\hat{g}^{m_1} x, x) G_0(\hat{g}^{m_2} x, x)     \nonumber \\
& = 
{
\frac{V_{d-1}  }{16 \pi}
 \int \frac{d^{d-1}k_\parallel d^{d-1}p_\parallel}{(2\pi)^{2(d-1)}}% \frac{d^{d-1}p_\parallel }{(2\pi)^{d-1}}
	\frac{1}{\sin^2\theta_{m_1}  M_{k_\parallel}^2 - \sin^2\theta_{m_2} M_{p_\parallel}^2}
	\log\left(\frac{\sin^2\theta_{m_1} M_{k_\parallel}^2}{\sin^2\theta_{m_2} M_{p_\parallel}^2}  \right).
	    }
\label{8fig-twist}
}
%%%%%%%%%%%%%%%%%%%%%%%%%%%%%%%%%%%%%%%%
\begin{figure}[h]
\centering
\includegraphics[width=0.75\linewidth]{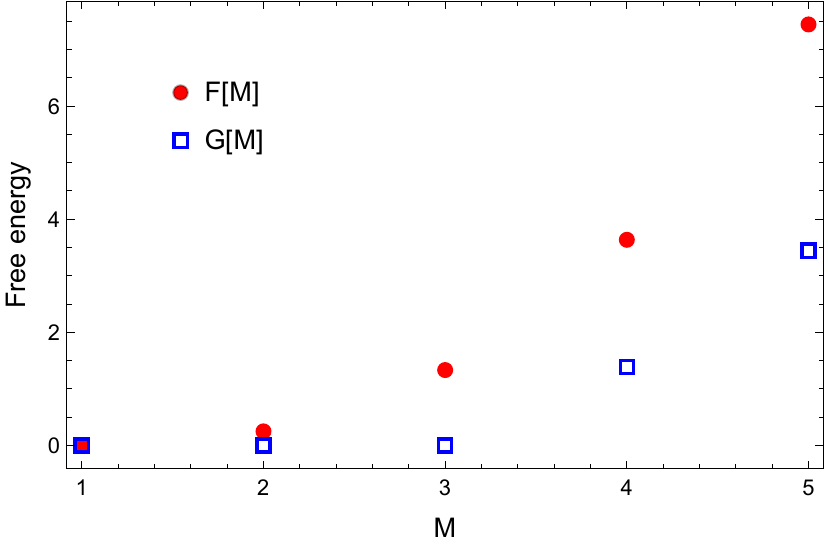} % 
\caption{
$F(M)$ is a sum of the %coefficients 
{integrand} of \eqref{8fig-twist} over  $m_1, m_2 = 1 \cdots M-1$ {for $d=1$}.
 Vertex contributions  $(m, \pm m)$ are subtracted in $G(M)$.
}
\label{8fig}
\end{figure}
%%%%%%%%%%%%%%%%%%%%%%%%%%%%%%%%%%%%%%%%%%%
EE is obtained by the analytical continuation of $M$ and calculating the coefficient of the first derivative at $M=1$. 
To see the behavior of $M$-dependence of \eqref{8fig-twist}, let us focus on the $d=1$ case for simplicity.
Summation over {nonzero} $m_1$ and $m_2$ can be 
explicitly evaluated and plotted in Fig.\ref{8fig}. 
$F(M)$ in Fig.\ref{8fig} is a sum of the integrand of \eqref{8fig-twist} over 
$m_1, m_2 =1, \cdots, M-1$. They include 2-loop vertex corrections {$(m_1,m_2)=(m,\pm m)$}. 
$G(M)$ is plotted without the vertex corrections. %$(m, \pm m)$. 
If we can simply interpolate the free energy to continuous $M$ near $M=1$, 
the first derivative seems to be dominated by the vertex contributions.  
Of course, it is not sufficient but we expect that 
EE of the figure-eight diagram is
dominantly given by twisting the propagators, $(m,0)$ and $(0,m)$, 
and the vertex $(m, \pm m)$. As long as we calculate the multiple twist contributions for a fixed integer $M$, we cannot be sure whether these contributions are absent in EE or not. In the next section, we consider the EE in the effective field theory in the IR limit instead of the explicit computation of the multiple twist contributions. In this limit, we expect these multiple twist contributions become less and less dominant.

\section{Wilsonian RG and entanglement entropy in effective field theories}\label{sec:wilson-rg}
We discussed that the non-Gaussian contributions to EE can be understood in terms of the two-point functions of composite operators. 
As explained in Appendix \ref{app:twist}, a twisted propagator is pinned with loose ends reflecting quantum correlations between two spatial regions. 
From this observation, it is tempting to expect that 
the vertex contributions to EE from composite operators reflect the emergent vertices in the low-energy effective theory.
Indeed, in the framework of the \textbf{Wilsonian RG}\index{Wilsonian renormalization group}, the effective action (EA) changes as the energy scale is changed, and the EA contains infinitely many vertices. 
Thus EE would also follow the same RG flow. This is similar to the statement of the entropic $c$-function, where the central charge $c$ in a $d=2$ CFT or analogous quantities in higher dimensions is defined from EE. However, the usual entropic $c$-theorem is a weak version, i.e. the $c$-function is compared between the UV and IR fixed point CFTs and not in the middle.\footnote{The existing proof of monotonicity of the entropic $c$-function along the RG flow relies on the strong subadditivity of EE. Since this is the monotonicity with respect to the length scale of the subregion, not the UV cutoff scale itself, it is not clear if this really proves the strong version of the $c$-theorem far from the conformal fixed points.} Our formula of EE allows us to compute from the Lagrangian and the QFT is not limited to CFTs. While it is interesting to study the monotonicity of EE using our formula, there are remaining, not evaluated contributions from multiple twists. Therefore, we consider in this section if there is any regime along the RG flow such that these multiple twist contributions are negligible. In this section, we give a conjecture that 
the IR part of EE is exhausted by summing all
the vertex contributions (together with propagator contributions)
constructed from the IR Wilsonian effective action. 
%%%%%%%%%%%%%%%%%%%%%%%%%%%%%%%%%%%%%%%%%
\subsection{More properties of vertex {contributions} to entanglement entropy}
First, %let 
note that  the leading order term of the vertex contribution corresponding to 
composite operators $\{{\cal O}_n\}$ %in Eq.(\ref{EE-derivative}) 
 is perturbatively given by expanding the logarithm as
\aln{
S_{\rm vertex} &= \frac{V_{d-1}}{12} \int  \frac{d^{d-1}k_{\parallel}} {(2 \pi)^{d-1}} \tr \hat\lambda   \hat\Sigma^{(g)} ({\bm k}=0, k_\parallel)
\nn
&= \frac{V_{d-1}}{12} \int  d^{2}{\bm r} \sum_{m,n} (\hat\lambda)_{mn}   \langle{\cal O}_n(-\frac{{\bm r}}{2}, x_\parallel=0) 
{\cal O}_m (\frac{{\bm r}}{2}, x_\parallel=0) \rangle^{(g)} .
}
The integral in the second line reflects the property of a twisted propagator
that its center coordinate is pinned at the boundary $x_\parallel=0$ with two loose ends. 
%For an operator, e.g.,  
{For instance, when ${\cal O}=[\phi^2]$,} the leading perturbative term %in this correlator 
is given by using the renormalized  propagator $G$ of the fundamental field $\phi$ as 
\aln{
S_{\rm vertex} \sim 
\frac{V_{d-1}}{6} \int  d^{2} {\bm r} \left(-\frac{3\lambda_4}{2}\right)  \ G({\bm r},0)^2 .
\label{EE-G2}
}
$-3\lambda_4/2$ is the component of $\hat{\lambda}$ which associates $[\phi^2]$ to $[\phi^2]$.
If we consider an operator such as ${\cal O}=[\phi^n]$, %with a larger $n$, 
the integrand is proportional to $G({\bm r},0)^n$ and decays faster for a larger $n$.  
This means that at least perturbatively, 
higher-dimensional composite operators tend to %provide less contributions 
contribute less to EE. 

Another important point to note in \eqref{EEren},  particularly for its vertex part, is that
if some composite operators in $\hat{\lambda} \hat{G}$ dominates $1$ in the logarithm in a strong coupling region, 
the contribution from the composite operator can be approximated as
\aln{
\mathrm{tr}\,\log (\hat{1}+ \hat\lambda \hat{G}) \sim   \mathrm{tr}\,\log (\hat{G})
} 
%Here we have assumed a simple case that there is no mixing of operators and we do not need to worry about noncommutativity of $\lambda$ and $\Sigma_g$. 
up to a constant depending on the coupling constant. 
%The first term of $\ln (\lambda)$ is just a constant and can be dropped \red{正定値性...?}. %The second term is analogous to the propagator contribution in  Eq.(\ref{e:EEprop}), if $\langle {\cal O}\,{\cal O} \rangle_g$ is identified with $\hat{G}^{-1}$.
Then, EE can be written as a logarithm of renormalized correlators similar to the fundamental field. 
There is no explicit dependence on the coupling constant {other than the overall factor} and
its dependence is only given through the renormalization of correlators. 

%%%%%%%%%%%%%%%%%%%%%%%%%%%%%%%%
\subsection{Wilsonian RG and entanglement entropy: free field theories}
Now we discuss the issue of other contributions to EE besides the propagators and vertices. 
For this purpose, it is convenient to utilize the concept of the 
 \textbf{Wilsonian RG}\index{Wilsonian renormalization group}  to the effective field theory in the IR region \cite{Wilson:1973jj,Polchinski:1983gv}.\footnote{A modern approach for the Wilsonian RG is given by the functional RG, also known as the exact RG, method \cite{Wetterich:1992yh, Morris:1993qb}.}
In the Wilsonian RG, we first divide the momentum domain into low and high regimes.
Schematically, 
\aln{
{%\bm 
	k} \in [0, \Lambda] = [0, e^{-t}\Lambda] + [e^{-t}\Lambda, \Lambda]
}
with $t>0$ and then, integrate quantum fluctuations over the high regimes. 
%A modern approach for the Wilsonian RG is given by the functional RG method \cite{Wetterich:1992yh, Morris:1993qb}.
Then, we rescale the momentum ${%\bm 
	k} \rightarrow {%\bm 
	k^\prime} = e^{t}{%\bm 
	k}$  so that 
${%\bm 
	k^\prime} \in [0, \Lambda]$. In this procedure, the original parameters in the action
are renormalized, e.g., 
\aln{
m \rightarrow m^\prime, \hspace{5mm} \lambda_4 \rightarrow \lambda_4^\prime. 
}
%and,  
In addition, new interaction terms appear, e.g., in {the $\phi^4$ theory in} $(3+1)$ dimensions, %al theory, such as
\aln{
\lambda_6  \frac{\phi^6}{\Lambda^2},  \hspace{5mm}  \lambda_{{\partial}} \frac{(\phi \partial \phi)^2}{\Lambda^2}, 
\hspace{5mm}  \lambda_8  \frac{\phi^8}{\Lambda^4}  \cdots .
}

%%%%%%%%%%%%%%%

First, let us look at what happens for a free theory. For a free scalar field with a mass $m$, EE is simply given by 
\aln{
S_{\text{EE}}(\Lambda) & =- \frac{V_{d-1} }{12} 
 \int^{\Lambda} \frac{d^{d-1}k_{\parallel}} {(2 \pi)^{d-1}}  
  \log \left[ (k_\parallel^2 +m^2) / \Lambda^2 \right] .
  \label{e:EEprop-free}    
 }
 By integrating the high momentum region, nothing happens except fluctuations of that region
 are discarded: 
 \aln{
 S_{\text{EE}} (e^{-t}\Lambda) & =- \frac{V_{d-1} }{12} 
 \int^{e^{-t}\Lambda} \frac{d^{d-1}k_{\parallel}} {(2 \pi)^{d-1}}  
  \log \left[ e^{2t}(k_\parallel^2 +m^2) / \Lambda^2 \right] .
 }
 Then, we rescale the momentum as $k^\prime = e^t k$ to obtain 
 \aln{
  S_{\text{EE}}^\prime (\Lambda)  =- 
 \frac{V_{d-1} }{12} 
 \int^{\Lambda} \frac{d^{d-1}k^\prime_{\parallel}}{e^{(d-1)t} \times (2 \pi)^{d-1}}  
  \log \left[ (k^{\prime 2}_\parallel + e^{2t} m^{ 2}) / \Lambda^2 \right] .
  \label{EE-IR-free}
 }
Of course, for a free field theory, it is equal to \eqref{e:EEprop-free} with the integration range $[0, e^{-t}\Lambda]$. 
For an interacting theory, it is different since high and low momentum {modes} are entangled. 
We continue the integration over high momentum modes until $e^{-t}\Lambda = m$. 
Then, EE is given by \eqref{e:EEprop-free} with the integration range $[0,  m]$.
It gives the  IR part of the EE at the scale $m$, and
the discarded parts in higher momentum are  UV cut-off dependent.  
By performing the momentum integration, 
{the EE at the scale $m$} is now given by 
\aln{
  S_{\text{EE}}^{\rm IR} (m) 
  	& {\equiv - \frac{V_{d-1} }{12} 
  	\int^{m} \frac{d^{d-1}k_{\parallel}} {(2 \pi)^{d-1}}  
  	\log \left[ (k_\parallel^2 +m^2) / \Lambda^2 \right] } \\
  & =
 \frac{N_{\rm eff} V_{d-1} }{12}
  \,  m^{d-1}\log \left[  \tilde\Lambda^2/ m^2 \right] ,
  \label{EE-IR-free2}
 }
where $\tilde{\Lambda}$ is proportional to the  UV cutoff 
as $\tilde{\Lambda}= \Lambda \exp[\Phi(-1,1,\frac{d+1}{2})/2] /\sqrt{2}$. $\Phi(z,s,\alpha)\equiv \sum_{n=0}^\infty z^n/(n+\alpha)^s$ is the Lerch transcendent.
For example, in $d=3$, it is given  by $\tilde{\Lambda}={e^{1/2} \Lambda}/{2}$. 
%{for even space-time dimensions. }
{\eqref{EE-IR-free2} coincides with the ordinary universal term in even spacetime dimensions.} 
 $V_{d-1}$ is the area of the boundary and 
\aln{
 N_{\rm eff}=  \left(\frac{1}{2}\right)^{d-1} \frac{1}{ \pi^{(d-1)/2} \Gamma((d+1)/2)} 
} 
is the effective {number of} degrees of freedom that can contribute to  EE {in the IR}. %, and 
%{we have absorbed some constants into $\tilde\Lambda$.
%\red{For $d=3$, it is given as $\tilde{\Lambda}={e^{1/2} \Lambda}/{2}$.} 
The result of \eqref{EE-IR-free2} 
%shows
{indicates} that the universal part of EE originates in the quantum correlations
of fields whose length scale is larger than the typical correlation length $\xi=1/m$ of the system.
$S_{\text{EE}}^{\rm IR} (m)$ becomes larger for smaller masses $m$. %}
%%%%%%%%%%%%%%ここまで%%%%%%%%%%%%%%%%%%%%%%%%%%%%%%%%%%%%%%%%%%%%%
%%%%%%%%%%%%%%%%%%%%%%%%%%%%%%%%
\subsection{Wilsonian RG and entanglement entropy: interacting field theories}
In the free case, the Wilsonian RG  can extract the {IR} behavior of EE that is independent of the UV
cutoff. 
In the Wilsonian RG,  quantization is gradually performed from high momentum to low, and
 in the IR limit, all fluctuations are integrated out so that
all the loop effects are incorporated in  the \textbf{Wilsonian effective action (EA)}\index{Wilsonian effective action}\index{effective action}\index{EA |see effective action}.\footnote{Strictly speaking, the EA explained here is a 1PI EA instead of Wilsonian EA. The tree-level diagrams gives the exact free energy in the former while we need to still perform the path integral in the latter. Nevertheless, the difference must be negligible in the IR limit since we are considering massive QFTs and they always have a finite IR cutoff. See~\cite{Burgess:2007pt} for further details.}
{The Wilsonian EA} becomes more and more complicated as radiative corrections are gradually taken into account. 
Thus we can expect that all the contributions to EE are encoded in the Wilsonian EA.
We conjecture that EE is given by a sum of all the {propagator and} vertex contributions in the Wilsonian EA.  

%%%%%%%%%%%%%%%%%%%%%%%%%%%%%%%%%%%%%%%%%%  
\begin{figure}[t]
\centering
\hspace*{-1.2cm}
\includegraphics[width=\linewidth]{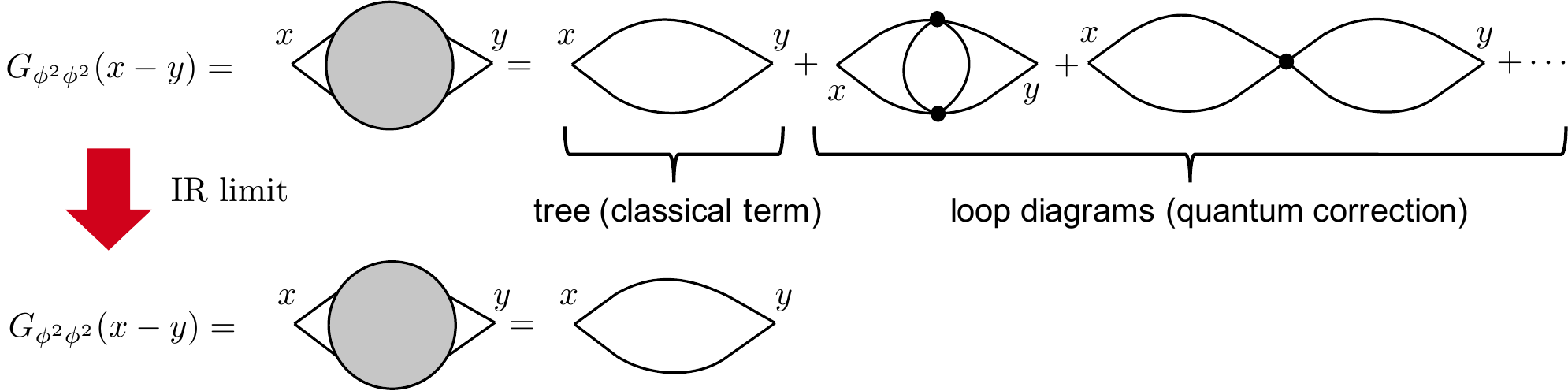}%{ambiguity2_ver2.pdf}
\caption{
The upper figure shows Feynman diagrams constituting $G_{\phi^2\phi^2}$.
Only the leading diagram survives in the IR limit. 
}
\label{f:vertexWilson}
\end{figure}
%%%%%%%%%%%%%%%%%%%%%%%%%%%%%%%%%%%%%%%%%%
In the following, we focus on the IR limit of the Wilsonian EA. 
Let us recall a simple case in \eqref{e:EEcomp2}. 
The correlator $G_{\phi^2\phi^2}$ is graphically given by the upper figure of 
Fig.\ref{f:vertexWilson} and
the first term is  given by \eqref{EE-G2}. 
This diagram is present even in the IR limit where all the fluctuations are integrated out since
it is simply connected by the propagators of the fundamental field. 
The other terms vanish in the IR limit of the Wilsonian RG
since they are quantum corrections to the first classical term. 
After all the fluctuations are integrated out, further quantum corrections should be absent because
such effects are already absorbed in the Wilsonian EA. 
%Thus we expect that  the vertex contributions of, e.g.,  the composite operator
%$\phi^2$ are 
Thus we expect that the vertex contributions of the composite operator, e.g. $\phi^2$, are drastically simplified in the IR limit in which we can 
replace the Green function $G_{\phi^2\phi^2}$ by the leading diagrams as shown in the lower figure of Fig.\ref{f:vertexWilson}.
After all, the vertex contribution in \eqref{e:EEcomp2} becomes 
\aln{
S_{\text{vertex}}= 
		 - \frac{V_{d-1}}{12}
		\int^{\Lambda}\frac{d^{d-1}k_\parallel}{(2\pi)^{d-1}}\mathrm{log}\,\left[1 + 
		3 \lambda_4  
		%\int d^2 {\bm r} G({\bm r},0)^2
	       \int\frac{d^{d+1}p}{(2\pi)^{d+1}}G(p)G(-\bm{p},k_\parallel-p_\parallel)   
		%\phi^2 \phi^2のcontractionから2通り
		\right] 
		\label{e:EEvertexRG}
}
in the IR limit. The coupling constant $\lambda_4$ is the renormalized one since it is a coefficient of 
Wilsonian EA in the IR limit.\footnote{
We have already taken quantum fluctuations into account and eliminated the UV divergences in coupling constants and observables in the IR limit, but another UV divergences appear in the calculation of EE since 
we need to sum all the momentum modes. It is also necessary even for the free theory 
and indeed we extracted the IR universal part by subtracting cutoff dependent terms. }
As in the free case, we separate the vertex contributions into  %universal and non-universal parts
{IR and UV parts}.
The IR part is defined similarly by restricting the integration range from $k_\parallel \in [0, \Lambda]$
to $[0, m]$.  
Instead, we may integrate  up to $1/\xi_{\phi^2}$
where  $\xi_{\phi^2}$ is the correlation length of the operator $[\phi^2]$. 
The difference is a matter of definition of the IR universal part of EE and we need a 
precise prescription to subtract the cutoff dependent terms in EE. 
For example, we may  take a variation with respect to the mass $m$ and then integrate to obtain the universal 
part of EE. In this definition, we need to know how $\xi_{\phi^2}$ and $m$ are related. 
%It is now under investigation.
{This issue should be explored in future.}

In general, of course, we need to take operator mixings into account
but the generalization is straightforward. 
%We will investigate more detailed behaviors of EE in the infrared limit of the Wilsonian EA for a concrete model.
The final question is whether there are contributions to EE other than the  vertex contributions in the Wilsonian EA. 
In the formulation of EE based on the $\mathbb{Z}_M$ gauge theory on Feynman diagrams, vertex contributions are only a part of all the contributions to EE. But, in the IR limit of Wilsonian RG, all the quantum fluctuations are integrated out and  we would not need to evaluate loop diagrams: all the Feynman diagrams are tree diagrams. 
Thus, the vertex contributions, as well as the propagator contributions, to EE must suffice {for} the IR behavior of EE.

\chapter{Holography, AdS/CFT, and AdS/BCFT}\label{ch:2}
\renewcommand{\thesection}{\thechapter.\arabic{section}}
%!TEX root = ../thesis.tex
%*******************************************************************************
%****************************** Second Chapter *********************************
%*******************************************************************************

%Entanglement in holography
%\chapter{My second chapter}

\ifpdf
    \graphicspath{{Chapter2/Figs/Raster/}{Chapter2/Figs/PDF/}{Chapter2/Figs/}}
\else
    \graphicspath{{Chapter2/Figs/Vector/}{Chapter2/Figs/}}
\fi

\renewcommand{\thesection}{\thechapter.\arabic{section}}
\setcounter{section}{0}
\textit{This chapter reviews holography from a bottom-up perspective by inspecting both sides. The purpose of this chapter is to make this dissertation self-contained about holography with a sufficient number of references and let readers prepare for the next chapter.
Starting from the original proposal of the AdS/CFT correspondence, we review the bottom-up AdS$_3$/CFT$_2$ correspondence and the AdS/BCFT correspondence by adding a boundary. We also review holographic entanglement entropy formula for each. Useful reviews on the AdS/CFT correspondence include~\cite{Kaplan:2016,Banerjee:2018,1020000782220175744,Harlow:2014yka,Hartman:2015}. We review CFT and some results are presented without a derivation to avoid unnecessary complications. See \cite{Di_Francesco_1997} for more details; in particular, \cite{BB29063466} for higher-dimensional CFTs and \cite{BB19618269} for two-dimensional CFTs.}\\

In this chapter, we provide an extensive review of holography, mainly focusing on the bottom-up approaches by inspecting the gravity side and field theory side, respectively. 
Section \ref{sec:overview-hol} gives the outline of holography. It also explains the original proposal by Maldacena. While his conjecture is top-down, we avoid using string theory as much as possible and only discuss the parameter dependence and its consequence.
Section \ref{sec:AdS-spacetime} and Section \ref{sec:asympt-AdS} review various asymptotically anti-de Sitter (AdS) spacetimes. In particular, the coordinate transformations among them will be important in the next chapter.
Section \ref{sec:CFT-gen} reviews conformal field theory (CFT) in general dimensions.
Based on the last three sections, Section \ref{sec:AdS-CFT-general-dim} presents the (traditional) AdS/CFT correspondence between the large-$N$ gauge theory and Einstein gravity on AdS. The topics included here are: the dictionary between these two; some examples of deformations; the parameter dependence of operators; the reason why interactions are negligible in AdS; the geodesic approximation to the two-point function of heavy operators; holographic method for stress tensor; holographic entanglement entropy. The last three subsections Section \ref{sec:geodesic-approx}, \ref{sec:hol-EM-FG}, and \ref{sec:HEE-CFT} are relevant to Chapter \ref{ch:2-2}.
Next, in Section \ref{sec:AdS3/CFT2}, we focus on the AdS$_3$/CFT$_2$ case and provide explanations on why these dimensions are special and important. Features we discussed are: uniqueness of three-dimensional gravity solution; all solutions with a constant negative cosmological constant are locally AdS; eliminating higher curvature corrections; symmetry enhancement in CFT$_2$; conformal blocks; operator product expansion (OPE) with the stress tensor; replica trick in CFT$_2$; first law of entanglement entropy. Various results from this section, in particular the CFT part, will be used in Chapter \ref{ch:2-2}.
Then, in Section \ref{sec:hol-CFT2} we define two-dimensional holographic CFTs, which are often characterized differently compared to higher-dimensional cases (Section \ref{sec:def-holCFT}). Then, in the second subsection Section \ref{sec:detail-holCFT}, we discuss its consequence in more detail. This involves the result of the semiclassical conformal blocks, which will be used in Chapter \ref{ch:2-2}. Furthermore, the Ryu-Takayanagi formula is demonstrated in a simple example in the last subsection Section \ref{sec:HEE-2dCFT}.
The last two sections are applications and generalizations of the AdS/CFT correspondence, which will be very important in Chapter \ref{ch:2-2}. In Section \ref{sec:local-op-quench}, we review the gravity dual of the local operator quench in CFT$_2$. This gives us the foundation for the discussion in Chapter \ref{ch:2-2}. We avoid using CFT techniques, instead, we derive the black hole threshold and explain the reason why heavy operators are dual to black hole from thermodynamics, in particular, the eigenstate thermalization hypothesis (ETH).
In Section \ref{sec:AdS-BCFTgen}, we explain the AdS/BCFT correspondence, i.e. the AdS/CFT correspondence with boundaries. We review the boundary CFT (BCFT) briefly in the first subsection Section \ref{sec:BCFT} and then we define the AdS/BCFT correspondence (Section \ref{sec:AdS-BCFT-state}), where the end-of-the-world (EOW) brane is introduced. Finally, in Section \ref{sec:HEE-BCFT}, we discuss how the RT formula is extended in this case. The remaining last subsection addresses some known subtleties regarding the AdS/BCFT correspondence.
In Appendix \ref{app:worldline}, we show QFT two-point function can be written as the transition amplitude of a relativistic point particle.
In Appendix \ref{app:geodesic-approx}, we explicitly confirm the geodesic approximation gives the correct CFT two-point function.
In Appendix \ref{app:grav-action}, we present a standard AdS gravity action including the boundary term and counterterms. We discuss the meaning of counterterms in the context of holographic renormalization.
In Appendix \ref{app:conf-ward}, we derive the conformal Ward identity. 
In Appendix \ref{app:replica-calc-EM}, we calculate the energy-momentum tensor in the replicated space.

\section{Holography -- overview}\label{sec:overview-hol}
\subsection{Generalities}
\textbf{Holographic principle}\index{holographic principle}, or in short \textbf{holography}\index{holography}, is a duality between a $(d+1)$-dimensional quantum gravity defined on a manifold $\Sigma$ and a $d$-dimensional QFT on its boundary $\partial\Sigma$~\cite{tHooft:1993dmi,Susskind:1994vu}.\footnote{\textbf{Quantum gravity}\index{quantum gravity} here means a perturbative quantum theory of gravity on a fixed background like a perturbative string theory. We are still yet to achieve a nonperturbative formulation of quantum gravity, in which the background independence of quantum gravity is manifest.} The gravitating spacetime is often called a \textit{bulk}\index{bulk} while the dual QFT is said to live on the \textit{boundary}\index{boundary}. It simply means the dual QFT is defined on a fixed background spacetime, which coincides with the boundary of the bulk. 
Holography is motivated by the study of black hole, where its degrees of freedom counted by entropy comes from its area rather than its volume. This is known as the \textbf{Beckenstein-Hawking entropy}\index{Beckenstein-Hawking entropy}. We will briefly revisit this using a microscopic analysis based on holography in Section \ref{sec:ETH-Cardy}. For more details, see~\cite{BHentropy} and references therein.

A concrete realization was provided first by Maldacena, now known as the \textbf{AdS/CFT correspondence}\index{AdS/CFT correspondence} or \textbf{gauge/gravity duality}\index{gauge/gravity duality}~\cite{Maldacena:1997re,Gubser:1998bc}.\footnote{See \cite{Aharony:1999ti,Witten:1998qj} and \cite{Klebanov:2000me} for early-time reviews and lecture notes.} In the AdS/CFT correspondence, a gravitational theory (e.g. string theory) on $(d+1)$-dimensional asymptotically anti-de Sitter (AdS) spacetime is dual to a $d$-dimensional CFT on the boundary.\footnote{\textit{Caveat}: Recently, people found an ensemble average of \textit{many} theories %which are CFTs in low energy 
is dual to the \textbf{Jackiw-Teitelboim (JT) gravity}\index{Jackiw-Teitelboim gravity} in a two-dimensional (nearly) AdS spacetime. Without averaging, we do not have a semiclassical geometric dual~\cite{Saad:2021rcu}. Some people point out we need some type of random CFTs for pure AdS${}_3$ gravity~\cite{Cotler:2020ugk,Chandra:2022bqq,Belin:2020hea}. % on the torus times interval 
In this dissertation, we do not consider any ensembles.} 
The original proposal by Maldacena states that \textbf{$\boldsymbol{\mathcal{N}=4}$ super-Yang-Mills (SYM)}\footnote{$\mathcal{N}=4$ means we have four spinor supercharges.} (i.e. non-Abelian gauge) \textbf{theory}\index{$\mathcal{N}=4$ super-Yang-Mills theory} with the $SU(N)$ gauge group in four dimensions is dual to \textbf{type-IIB superstring theory on AdS${}_5\times$S${}^5$}\index{superstring theory}\index{string theory |see superstring theory }.\footnote{The vanishing conformal anomaly of the string worldsheet CFT requires the \textbf{Ricci-flat condition}\index{Ricci-flat condition}. This is not contradicting with having AdS since the Ricci-flatness is imposed for a non-compact six-dimensional cone over the five-dimensional compactified space~\cite{Govindarajan:2022}.}
Later, the original proposal has been extended to various dimensions, collectively denoted as the \textbf{AdS${}_{d+1}$/CFT${}_d$ correspondence}. In this chapter, we mainly focus on the lower dimensional case, namely, AdS${}_3$/CFT${}_2$ correspondence and its extension. The readers may wonder why we consider such a low-dimensional case although our spacetime dimension is four. In fact, this low-dimensional setup leads to various simplifications retaining rich, nontrivial spacetime structures. It is a good playground to investigate holography in detail and guess higher-dimensional extensions.

We will first %sketch 
give a brief overview of %the derivation of 
the original proposal by Maldacena. Then, we introduce the AdS${}_3$/CFT${}_2$ correspondence and illustrate its difference from the simplification. Finally, we consider its extension by introducing boundaries.

\subsection{Maldacena's proposal}
In this section, we will briefly review the original proposal of the AdS/CFT correspondence by Maldacena~\cite{Maldacena:1997re}.\footnote{Here we focus on the limits of parameters and do not follow the actual derivation from D-branes\index{D-branes} as they do not play an important role in our discussion for the AdS${}_3$/CFT${}_2$ correspondence.} The CFT is given by $\mathcal{N}=4$ $SU(N)$ SYM\index{$\mathcal{N}=4$ super-Yang-Mills theory} in four dimensions, whose gauge field part of the action is given by
\begin{align}
    \mathcal{L}_{SYM} &=-\frac{1}{2g_{YM}^2} F_{\mu\nu a} F^{\mu\nu}_b \tr [t^a t^b ]\\
    &=-\frac{1}{4g_{YM}^2} F_{\mu\nu a} F^{\mu\nu}_a. 
\end{align}
This is conjectured to be dual to type-IIB superstring\index{type-IIB superstring theory} on AdS${}_5\times S^5$, whose worldsheet action (in the Nambu-Goto form) is given by
\begin{equation}
    I_\text{string}=-\frac{1}{2\pi\alpha^\prime}\int\dd[2]{\sigma} \sqrt{-h},
\end{equation}
where $h$ is the determinant of the worldsheet metric and $\sigma$'s are the worldsheet coordinates. The prefactor $1/(2\pi\alpha^\prime)\equiv m_s^2 /(2\pi)$ equals the string tension\index{tension (string)}. The gravitational sector of its ten-dimensional effective action is given by
\begin{equation}
    I_{SUGRA}=\frac{1}{(2\pi)^7 \alpha^{\prime\, 4}}\int\dd[10]{x} \sqrt{-g_{10}} e^{-2\phi} \left[\mathcal{R}_{10} +O(\alpha^\prime)\right],
    \label{eq:sugra}
\end{equation}
where $g_{10}$ is the ten-dimensional metric determinant from the \textbf{supergravity}\index{supergravity} solution and its Ricci scalar is given by $R_{10}$. $\phi$ is the dilaton\index{dilaton} field, whose vacuum expectation value $\phi_0$ gives string coupling constant $g_s \sim e^{\phi_0}$. In this AdS${}_5$/CFT${}_4$ case, which is derived from D3-branes, the on-shell dilaton value is constant over AdS~\cite{Aharony:1999ti}. Thus, it can be treated as a true constant independent of coordinates. $O(\alpha^\prime)$ contains higher curvature terms like $\mathcal{R}_{10}^2$.

%Since the stringy effect is suppressed in the planar limit, the description in the string theory side can be well approximated by supergravity, in which the Newton's constant $G_N\equiv \kappa^2 \sim (g_s \alpha^{\prime\, 2})^2$. The AdS spacetime appears as a solution in supergravity. In the next subsection, we will briefly explain why the AdS spacetime shows up by considering the string dynamics near the so-called D-branes.

%\subsection{Why AdS?}
%To understand why the gauge theory is dual to the string theory on AdS${}_5\times S^5$ in Maldacena's orignal proposal, we need to introduce D-branes~\cite{Dai:1989ua,Horava:1989ga}.\footnote{Note that this is just to explain the original proposal. In the lower dimensions like AdS${}_3$/CFT${}_2$, symmetry rules them all and the correspondence can be explained without bringing D-branes.} See~\cite{Duff:1999rk,Polchinski:1996na,1997KJ00004709097,BA78793344,Hashimoto:2009,BB23240886,BA56959443} for reviews. A D$p$-brane is a spatially $p$-dimensional soliton in string theory. An open string can end on the D$p$-brane with the Dirichlet condition fixing its location for the rest codimension $p+1$. The vacuum expectation value of the dilaton $\phi$ in the effective D-brane action (known as the Dirac-Born-Infeld action) is related to the string coupling constant as $g_s\sim e^{\phi}$. Thus, the mass density of the D-brane is proportional to $1/g_s$.\footnote{This is why the D-branes are often referred as nonperturbative objects.} Thus, when $g_s$ is very small, the D-brane is treated as an non-dynamical, heavy source.

To gain an intuition of why this correspondence can hold, one very enlightening piece of evidence %for this duality 
is symmetry. The $\mathcal{N}=4$ SYM is invariant under the \textbf{conformal group}\index{conformal group} $SO(4,2)$ as we can see from its vanishing $\beta$-function\index{$\beta$-function} in the 't Hooft limit (explained later)
\begin{align}
    \beta(g_{YM})& \equiv \mu \dv{g_{YM}}{\mu}\\
    &=-\frac{g_{YM}^3 N}{48\pi^2}\left(11-2n_f -\frac{1}{2}n_s\right)\\
    &=0, \quad (\because n_f=4,\ n_s=6)
\end{align}
where $n_f$ is the number of Weyl fermions and $n_s$ is the number of real scalars.
AdS${}_5$ also exhibits the same symmetry. It is the only solution with the $SO(4,2)$ isometry, in which the metric does not change locally under its group. %The additional extra dimensions come from supersymmetry. 
Besides, we need a five-dimensional compactified spacetime in addition for a consistent string theory. $S^5$ has a global symmetry $SO(6)$. This is reflected in the fact that $\mathcal{N}=4$ SYM has a global $SU(4)\simeq SO(6)$ $R$-symmetry\index{$R$-symmetry}.\footnote{These symmetries combined with 16 supercharges constitute a superconformal group $PSU(2,2|4)$.}

This correspondence is achieved by taking two limits, \textbf{the large-$
\bm{N}$ ('t Hooft) limit}~\cite{tHooft:1973alw}\index{large-$N$ limit}\index{'t Hooft limit} and the \textbf{large 't Hooft coupling}\index{'t Hooft coupling} $\lambda\equiv g_{YM}^2 N=g_s N$. In the large-$N$ limit, it suppresses the quantum effect due to the worldsheet branching controlled by $g_s$ in string theory (gravity) while it suppresses non-planar diagrams, which cannot be put on a sphere, in gauge theory.\footnote{This limit is also good for D-branes\index{D-branes}. Since their mass is proportional to $1/g_s$, they become non-dynamical as $g_s\rightarrow 0$.} This is a saddle-point approximation valid under $N\rightarrow \infty$ or $g_{YM}\rightarrow 0$ with $\lambda$ fixed, suppressing the higher-genus terms. In the large 't Hooft coupling limit, it suppresses the massive modes peculiar to string oscillations controlled by $\alpha^\prime$ in both classical and quantum string theory\footnote{Since $\alpha^\prime=l_s^2$, a square of string length is dimensionful, its suppression must be compared to some characteristic length scale like the curvature radius of spacetime. In this case, this is given by the \textbf{AdS radius}\index{AdS radius} $R$, which appears from supergravity. See the next footnote.} while it leads to a strongly-coupled ($\lambda\rightarrow \infty$) theory in gauge theory, in which diagrams with many vertices nonperturbatively contribute.\footnote{In string theory, this corresponds to the low energy limit for an observer at \textit{infinity}. Near a stack of D3-branes, the supergravity solution becomes AdS. All modes are blueshifted as they approach the center of AdS due to its hyperbolicity. This results in the full string theory in AdS.} By taking these two limits simultaneously, it suggests a duality between semiclassical gravity like Einstein gravity and strongly-coupled large-$N$ gauge theory. 

Indeed, these two limits lead to the semiclassical approximation of Einstein gravity. Since any higher curvature contributions from the $\alpha^\prime$ correction are suppressed in the large $\lambda$ limit and the dilaton is constant, the effective supergravity\index{supergravity} action \eqref{eq:sugra} becomes \textbf{Einstein gravity}\index{Einstein gravity}
\begin{equation}
    \frac{1}{16\pi G_N^{(10)}}\int \dd[10]{x}\sqrt{-g_{10}} \mathcal{R}_{10},
    \label{eq:sugra-2}
\end{equation}
where the ten-dimensional Newton's constant\index{Newton's constant} scales as $G^{(10)}_N\sim g_s^2 \alpha^{\prime\,4}$. Next, let us consider its dimensional reduction to five dimensions. The ten-dimensional metric (near the D3-branes) is AdS${}_5\times S^5$:
\begin{equation}
    ds^2_{10} = \frac{\bar{r}^2}{R^2}\eta_{\mu\nu}dx^\mu dx^\nu +\frac{R^2}{\bar{r}^2}d\bar{r}^2 +R^2 d\Omega^2_5,
    \label{eq:AdS5}
\end{equation}
where $\eta_{\mu\nu}=\mathrm{diag}(-1,1,1,1)$ and the AdS radius is related to the parameters of the theory as $R=\sqrt[4]{\lambda \alpha^{\prime\, 2}}$ . Since the radius of the compactified $S^5$ is $R$, the gravity action \eqref{eq:sugra-2} after the compactification is
\begin{equation}
    \frac{R^5 \Omega_5}{16\pi G_N^{(10)}}\int \sqrt{-g_5} (\mathcal{R}_5 +\cdots),
\end{equation}
where the five-dimensional Newton's constant is $G_N^{(5)}\equiv G_N^{(10)}/(R^5 \Omega_5)=g_s^2 \alpha^{\prime\,4}/(R^5 \Omega_5)$. Since
\begin{equation}
    \frac{G_N^{(5)}}{R^3}  \sim g_s^2 \left(\frac{\alpha^{\prime\, 2}}{R^4}\right)^2 = \frac{1}{N} \quad \left(\because \frac{R^4}{\alpha^{\prime\, 2}}=\lambda=g_s N \right) ,
\end{equation}
the large-$N$ limit is equivalent to the valid \textbf{semiclassical approximation}\index{semiclassical} $G_N^{(5)}\ll R^3$.

The AdS/CFT correspondence is very powerful and promising as we could investigate a very difficult parameter regime in gravity such as quantum gravity correction from a finite-$N$ analysis and highly stringy correction from a perturbative analysis.
For further information, refer to to~\cite{Aharony:1999ti} and related literature.

If one does not restrict the bulk theory to the Einstein gravity, the dual CFTs are not necessarily strongly coupled; an ensemble average of free, Narain CFTs\index{Narain CFTs} and the $U(1)$ gravity~\cite{Afkhami-Jeddi:2020ezh,Maloney:2020nni}\index{$U(1)$ gravity}; higher-spin gravity\index{higher-spin gravity} in AdS${}_4$ (Vasiliev theory\index{Vasiliev theory}) and $O(N)$ vector model~\cite{Klebanov:2002ja}\index{$O(N)$ vector model}.\footnote{It is worth noting that the three-dimensional Vasiliev theory is dual to the $W_N$ minimal model~\cite{Gaberdiel:2010pz}\index{$W_N$ minimal model}.}

%In particular, the large $N$ limit~\cite{tHooft:1973alw} of

%In the following, we review the AdS spacetime and CFT respectively, then introduce some important relations in the AdS${}_3$/CFT${}_2$ correspondence.

\section{Anti-de Sitter (AdS) spacetime}\label{sec:AdS-spacetime}
In this section, we define the AdS spacetime and introduce some useful choices of its coordinate systems. \textbf{AdS}\index{anti-de Sitter spacetime}\index{AdS} is the maximally symmetric\index{maximally symmetric} spacetime with constant negative curvature. We will review what ``maximally symmetric'' means later in Section \ref{sec:AdS3}. To see the global nature of AdS, it is useful to employ the \textbf{embedding formalism}\index{embedding formalism}, i.e. embedding AdS into one higher dimension. AdS${}_{d+1}$ is a codimension-one hypersurface in a $(d+2)$-dimensional flat spacetime with two temporal directions $\mathbb{R}^{2,d}$:
\begin{align}
    ds^2 & = -dX_0^2 - dX_{d+1}^2 + \sum_{i=1}^d dX_i^2 \label{eq:AdS-metric}\\
    -R^2 & = -X_0^2 - X_{d+1}^2 + \sum_{i=1}^d X_i^2,
\end{align}
where $R$ is the AdS radius\index{AdS radius}, which is related to the Ricci scalar $\mathcal{R}$ and the cosmological constant $\Lambda$ as\footnote{We omit the dimension index of curvatures hereinafter.}
\begin{equation}
    \mathcal{R}=-\frac{d(d+1)}{R^2}=\frac{2(d+1)}{d-1} \Lambda,
\end{equation}
which is indeed a negative constant.
The embedding coordinates are useful to compute coordinate-independent quantities such as geodesic length (Appendix \ref{sec:Embedd}).
After the Wick rotation $X_{d+1}^{(E)}\equiv iX_{d+1}$, the Euclidean AdS\index{Euclidean AdS} is a hyperbolic space $H^{d+1}$, a space with the constant negative curvature $\mathcal{R}$. $H^{d+1}$ is invariant under rotation $SO(1,d+1)$. Back to the original signature, AdS is invariant under $SO(2,d)$. Its generators are given by $L^A_B=X^A\pdv{X_B}-X^B\pdv{X^A}\ (A,B=0,\cdots,d+1)$.

The AdS asymptotic boundary\index{asymptotic boundary} is obtained by sending $X^A$'s homogeneously to infinity. More precisely, it is given by the $\epsilon\rightarrow 0$ limit with a fixed null projective cone\index{null projective cone}, whose coordinates are $P^A\equiv \epsilon X^A$.
It is convenient to introduce various coordinate systems so that the location of the asymptotic boundary becomes more intuitively understandable and the boundary topology is fixed.
\subsection{Global coordinates}
In \textbf{global coordinates}\index{global coordinates}, the embedding coordinates are parametrized as
\begin{equation}
    \begin{split}
        X_0 &= R \cosh \tilde{\rho} \cos \tau\\
        X_{d+1} &= R \cosh \tilde{\rho} \sin \tau\\
        X_i &= R \sinh\tilde{\rho}\, \Omega_i \quad (i=1,\cdots,d),
    \end{split}
    \label{eq:global-coords}
\end{equation}
where $\Omega_i$ is the spherical coordinate of $S^{d-1}$ and $\sum_{i=1}^d \Omega_i^2 =1$. To cover the entire hyperboloid once, we need $\rho\ge 0$ and $\tau\in [-\pi,\pi)$. However, we consider the universal covering of AdS by taking $\tau\in (-\infty,\infty)$ to obtain a causal spacetime, avoiding closed timelike curves. This is commonly referred to as \textbf{global AdS}\index{global AdS}.

The AdS metric \eqref{eq:AdS-metric} in global coordinates \eqref{eq:global-coords} is 
\begin{equation}
    ds^2=R^2\left(-\cosh^2\tilde{\rho} d\tau^2+d\tilde{\rho}^2 + \sinh^2\tilde{\rho} d\Omega^2_{d-1}\right).
\end{equation}
The asymptotic boundary is at $\tilde{\rho}\rightarrow \infty$, where the metric diverges. The boundary topology is $\mathbb{R}_\tau \times S^{d-1}$.

The constant-$\tau$ surface is a hyperbolic space. Defining $d\Theta=d\tilde{\rho}/\sinh\tilde{\rho}$, i.e. $\tan(\Theta/2)=\tanh(\tilde{\rho}/2)$, we can put the entire AdS on a finite disk. This is known as the \textbf{Poincar\'e disk}\index{Poincar\'e disk} representation.

It is also common to reparametrize $\tilde{\rho}$ by $r=R\sinh\tilde{\rho}$. 
The corresponding embedding coordinates are
\begin{equation}
    \begin{split}
        X_0 &= \sqrt{R^2+r^2} \cos \tau\\
        X_{d+1} &= \sqrt{R^2+r^2}  \sin \tau\\
        X_i &= r\, \Omega_i \quad (i=1,\cdots,d).
    \end{split}
    \label{eq:global-coords2}
\end{equation}
The transformed metric is
\begin{equation}
    ds^2=-(r^2+R^2) d\tau^2 + \frac{R^2}{r^2+R^2}dr^2 + r^2 d\Omega^2_{d-1}.
    \label{eq:global-static-coords}
\end{equation}
This is sometimes called \textbf{static coordinates}\index{static coordinates}. The asymptotic boundary is given by $r\rightarrow\infty$.

\subsection{Conformal coordinates}
Just like the Penrose diagram of Minkowski spacetime, we can put the entire AdS to a finite region by $r=\tan \tilde{\theta}$ $(0\le\tilde{\theta}<\pi/2)$. The metric after this coordinate transformation is
\begin{equation}
    ds^2=\frac{1}{\cos^2\tilde{\theta}}(-d\tau^2+d\tilde{\theta}^2+\sin^2\tilde{\theta} d\Omega^2_{d-1}).
\end{equation}
These are called \textbf{conformal coordinates}\index{conformal coordinates}. The AdS boundary is now located at a finite distance, $\tilde{\theta}=\pi/2$. After a local conformal transformation removing the Weyl factor $\frac{1}{\cos^2\tilde{\theta}}$, the resulting topology is that of a solid cylinder, whose boundary is again $\mathbb{R}_\tau \times S^{d-1}$. This means AdS acts like a finite box with a spatial boundary. It makes AdS special compared to Minkowski spacetime, which is an infinite box with null infinities. For example, a large Schwarzschild black hole in AdS has a positive heat capacity and is thermodynamically stable while it is not in Minkowski; most importantly, AdS has a spatial boundary where we can define a physically reasonable `clock' and the dual field theory. A Cauchy surface is not enough as initial data to describe the entire AdS; the boundary condition at its asymptotic boundary must be specified.

\subsection{Poincar\'e coordinates}
We can parametrize the embedding coordinates differently. In \textbf{Poincar\'e coordinates}\index{Poincar\'e coordinates}, they are expressed as
\begin{equation}
    \begin{split}
        X_0 &= R \frac{\alpha^2+z^2+\bm{x}^2-t^2}{2\alpha z}\\
        X_{d+1} &= R \frac{t}{z}\\
        X_{\bar{i}} &= R \frac{x^{\bar{i}}}{z} \quad (\bar{i}=1,\cdots,d-1) \\
        X_d & = R \frac{-\alpha^2+z^2+\bm{x}^2-t^2}{2\alpha z},
    \end{split}
    \label{eq:poincare-coords}
\end{equation}
where $\bm{x}^2=\sum_{\bar{i}} (x^{\bar{i}})^2$ and $x^{\bar{i}}$ is Cartesian coordinates. $\alpha$ parametrizes a particular AdS isometry $SO(2,d)$. We have already seen Poincar\'e coordinates in a slightly different form. \eqref{eq:poincare-coords} is equivalent to the first five-dimensional part of \eqref{eq:AdS5} by rewriting $\tilde{r}=R^2/z$.

Since $X_0-X_d=R\alpha/z\ge 0$, Lorentzian Poincar\'e AdS only covers half of the global one. $t\rightarrow \pm \infty$ slices are causal horizons\index{causal horizon}. Note that despite Lorentzian Poincar\'e AdS does not cover the entire global AdS, Euclidean Poincar\'e AdS ($t\rightarrow t_E=it$) is the same space as Euclidean global AdS ($\tau\rightarrow \tau_E=i\tau$)~\cite{Aharony:1999ti}.

The Poincar\'e AdS has the metric
\begin{equation}
    ds^2= R^2\frac{-dt^2+d\bm{x}^2+dz^2}{z^2}.
    \label{eq:poincare-metric}
\end{equation}
The asymptotic boundary is now located at $z\rightarrow 0$. The boundary topology is $\mathbb{R}^{1,d-1}$. This metric is particularly useful as various geometric calculations are simplified.

In Poincar\'e coordinates, the scale invariance is manifest. The metric \eqref{eq:poincare-metric} is invariant under
\begin{equation}
    t\rightarrow \lambda t;\ \bm{x}\rightarrow \lambda \bm{x};\ z\rightarrow \lambda z.
\end{equation}
This means rescaling $t$ and $\bm{x}$ requires the same rescaling in the radial coordinate $z$. This leads to the idea of the scale/radius duality in the AdS/CFT correspondence, where the IR cutoff\index{IR cutoff} of AdS $z\ge\epsilon$ is equivalent to the UV cutoff\index{UV cutoff} of CFT $\Lambda=1/\epsilon$ at its asymptotic boundary with the topology $\mathbb{R}^d$~\cite{Susskind:1998dq}.

\section{Asymptotically AdS spacetime}\label{sec:asympt-AdS}
The AdS/CFT correspondence can also deal with \textbf{asymptotically AdS spacetimes}\index{asymptotically AdS spacetimes}. They approach pure AdS near the asymptotic boundary. In particular, when $d+1=3$, we impose particular fall-off/boundary conditions for them~\cite{Brown:1986nw,Henneaux:1985tv}. In general dimensions, the asymptotically AdS spacetimes are described by the Fefferman-Graham coordinates\index{Fefferman-Graham coordinates}, which will be explained in Section \ref{sec:hol-EM-FG}.
Let us review some important examples for now as asymptotically AdS spacetimes include many interesting geometries as we review in the following.

\subsection{Schwarzschild-AdS black hole}
A Schwarzschild black hole in AdS (\textbf{Schwarzschild-AdS black hole}\index{Schwarzschild-AdS black hole}) is described by~\cite{Nozaki:2013wia,Jahn:2017xsg,Charmousis:2009}
\begin{equation}
    ds^2=-\left(r^2+R^2-\frac{M_{d+1}}{r^{d-2}}\right)d\tau^2+\frac{R^2}{r^2+R^2-\frac{M_{d+1}}{r^{d-2}}}dr^2+r^2d\Omega_{d-1}^2,
    \label{eq:AdS-Sch-BH}
\end{equation}
where the mass parameter $M_{d+1}$ is related to the black hole ADM mass\index{ADM mass} $m$ as~\cite{Nozaki:2013wia,Jahn:2017xsg,Ashtekar:1999jx}
\begin{equation}
   M_{d+1}=\frac{16\pi G_N^{(d+1)}R^2}{(d-1)\Omega_{d-1}} m =\frac{8\Gamma\left(\frac{d}{2}\right)G_N^{(d+1)} R^2}{(d-1)\pi^{d/2-1}} m.
    \label{eq:mass-BH}
\end{equation}
The event horizon is defined from the outer radius $r=r_+$ satisfying $r^2+R^2-\frac{M_{d+1}}{r^{d-2}}=0$. 
Only in $d+1=3$, we need $M>R^2$ for the Schwarzschild-AdS black hole solution.\footnote{This three-dimensional black hole with a negative cosmological constant is called the Ba{\~{n}}ados-Teitelboim-Zanelli (BTZ) black hole~\cite{Banados:1992wn}.}
By inspecting the periodicity of $\tau_E=i\tau$ near the horizon, we can find the Hawking temperature\index{Hawking temperature}~\cite{Harlow:2014yka,Hartman:2015}
\begin{equation}
    \beta_H=\frac{1}{T_H}=\frac{4\pi R^2 r_+}{(d-2)R^2 + d r_+^2}.
\end{equation}
When $d+1=3$, 
\begin{equation}
    \beta_H=2\pi R^2/r_+=2\pi R^2/\sqrt{M-R^2} \quad (M\equiv M_3).
    \label{eq:BH-temp-3d}
\end{equation}

In the large mass limit, the horizon approaches the asymptotic boundary and the black hole becomes planar. The \textbf{planar Schwarzschild-AdS black hole}\index{planar Schwarzschild-AdS black hole} is given by (rewriting $r\rightarrow \tilde{r}$ from \eqref{eq:AdS-Sch-BH})
\begin{align}
    ds^2 & = -\left(\tilde{r}^2-\frac{M_{d+1}}{\tilde{r}^{d-2}}\right)d\tau^2+\frac{R^2}{\tilde{r}^2-\frac{M_{d+1}}{\tilde{r}^{d-2}}}d\tilde{r}^2+\tilde{r}^2d\Omega_{d-1}^2 \\
    & = \frac{R^2}{z^2}\left(-\frac{R^{2d}-M_{d+1}z^d}{R^{2d-2}} d\tau^2 + R^2 \frac{R^{2d-2} dz^2}{R^{2d}-M_{d+1}z^d} + R^2 d\Omega_{d-1}^2\right) \quad (\tilde{r}=R^2/z).
\end{align}

When the black hole mass is small enough, a different phase called \textbf{thermal AdS}\index{thermal AdS} is preferred. This can be understood by comparing each gravitational free energy. Such a thermodynamic phase transition is known as the \textbf{Hawking-Page transition}\index{Hawking-Page transition}~\cite{Hawking:1982dh}. The metric of thermal AdS is the same as the pure AdS, however, Euclidean time has a periodicity corresponding to the inverse temperature. We will not go into details here.

In higher dimensions $d>3$, a horizon can be non-spherical. Such a solution is called a black $p$-brane\index{black $p$-brane}. $p=0$ corresponds to black hole; $p=1$ corresponds to black string\index{black string}, etc.

We can also consider rotating and/or charged black holes. They are important in the program of microstate counting~\cite{Strominger:1996sh}. A charged black hole is called the Reissner-Nordstr{\o}m black hole\index{Reissner-Nordstr{\o}m black hole} and a rotating black hole is called the Kerr black hole\index{Kerr black hole}. A charged, rotating black hole is called the Kerr-Newman black hole\index{Kerr-Newman black hole}. 

\subsection{Conical deficit geometry}
When $d+1=3$, $r^{d-2}=1$. In such a case, there can be no solution for the horizon. This happens when $M\equiv M_3 <R^2$. This case is interpreted as a massive particle rather than a black hole.\footnote{We restrict to three dimensions just for the later purpose. Of course, there exists a conical deficit solution itself in higher dimensions~\cite{Bayona:2010sd}.}
The metric after taking into account the gravitational backreaction of this massive particle is given by
\ba
ds^2=-(r^2+R^2-M)d\tau^2+\frac{R^2}{r^2+R^2-M}dr^2+r^2d\theta^2,   \label{GBmet}
\ea
where $\theta$ has the periodicity $2\pi$ as $-\pi\leq \theta<\pi$. This is the same as \eqref{eq:AdS-Sch-BH} when $d+1=3$ but with $M<R^2$.
The mass parameter $M$ is related to the mass $m$ via
\ba
M=8G_N R^2 m  \label{relam}
\ea
as in the case of black hole \eqref{eq:mass-BH}. Note that the temperature \eqref{eq:BH-temp-3d} becomes pure imaginary, in this case, $M-R^2<0$.

The geometry (\ref{GBmet}) can be transformed into the metric\footnote{This can also be worked out directly from a particular parametrization of embedding coordinates~\cite{Banerjee:2018}.}
\ba
ds^2=-(\ti{r}^2+R^2)d\ti{\tau}^2+\frac{R^2}{\ti{r}^2+R^2}d\ti{r}^2+\ti{r}^2d\ti{\theta}^2,
\label{deficitg}
\ea
via the map
\ba
\ti{\tau}=\chi \tau,\ \ \ti{\theta}=\chi \theta,\ \  \ti{r}= \frac{r}{\chi},
\quad\text{where}\quad
\chi=\s{\frac{R^2-M}{R^2}}.
\label{teha}
\ea
Even though this looks like a global AdS$_3$, there is a deficit angle\index{conical deficit geometry} at $r=0$ since the periodicity of the new spatial coordinate is  $\ti{\theta}$ is 
\ba
%-\pi \s{\frac{R^2-M}{R^2}}\leq \ti{\theta}<\pi \s{\frac{R^2-M}{R^2}}.
-\chi \pi \leq \ti{\theta}<\chi \pi, \quad \chi<1.
\label{tildtheta}
\ea

\section{Conformal field theory (CFT)}\label{sec:CFT-gen}
In this section, we review some basic properties of CFTs. Important properties which will not be addressed here include CFTs as UV/IR fixed points of renormalization group flow. For these directions, see~\cite{RevModPhys.90.035007} and references therein.

\subsection{Conformal symmetry}
A $d$-dimensional \textbf{conformal field theory (CFT)}\index{conformal field theory}\index{CFT |see conformal field theory } is highly constrained by its \textbf{conformal symmetry}\index{conformal symmetry} $SO(2,d)$. Its generators are
\begin{equation}
    J_{AB}=
    \mqty(0 & D & \frac{1}{\sqrt{2}}(P_\mu-K_\mu) \\ -D & 0 & \frac{1}{\sqrt{2}}(P_\mu+K_\mu) \\ \frac{1}{\sqrt{2}}(P_\mu-K_\mu) & \frac{1}{\sqrt{2}}(P_\mu+K_\mu) & L_{\mu\nu})
    ,
    \label{eq:conf-gen}
\end{equation}
where 
\begin{equation}
    \begin{split}
    D &= x^\mu \pdv{x^\mu} \qq{(dilatation)}\\
    P_\mu &= -i \pdv{x^\mu} \qq{(translation)}\\
    K_\mu &= -i \left(2x_\mu x^\nu \pdv{x^\nu} - x^\nu x_\nu \pdv{x^\mu}\right) \qq{(special conformal transformation)}\\
    L_{\mu\nu} &= i\left(x_\mu \pdv{x^\nu} - x_\nu \pdv{x^\mu}\right) \qq{(Lorentz transformation $\mathfrak{so}(1,d-1)$)}
    \end{split}
    \label{eq:gen-diff}
\end{equation}
as differential operators\index{special conformal transformation}\index{dilatation}. $\mu,\nu=0,\cdots, d$ are Lorentz indices for the $d$-dimensional flat spacetime where the CFT lives.

$J_{AB}$ \eqref{eq:conf-gen} satisfies the $\mathfrak{so}(2,d)$ commutation relation
\begin{equation}
    \comm{J_{AB}}{J_{CD}}=i (\eta_{AD} J_{BC} + \eta_{BC} J_{AD} - \eta_{AC} J_{BD} - \eta_{BD} J_{AC}),
\end{equation}
where $\eta_{AB}=\mathrm{diag}(-1,-1,1,\cdots,1)$, two temporal coordinates $(-1)$'s and $d$ spatial coordinates $1$'s.
We can easily check this by directly plugging \eqref{eq:gen-diff} and \eqref{eq:conf-gen} into the commutator. Note that $P_\mu$ and $K_\mu$ are transformed as vectors under $L_{\nu\rho}$ and $\comm{K_\mu}{P_\nu}$ can be calculated from the Jacobi identity.

Since the special conformal transformation is a composition of inversion-translation-inversion\index{inversion}, a CFT is naturally defined on spacetime including infinity. This is the reason why we can map an infinite plane to a sphere and vice versa. (The bulk counterpart will be the coordinate transformation between Poincar\'e and global AdS.) All the conformal transformations preserve the `shape' of the transformed object; angles are locally preserved. In other words, a conformal transformation is a \textbf{Weyl transformation}\index{Weyl transformation}, $g_{\mu\nu}(x)\rightarrow \Lambda (x) g_{\mu\nu}(x)$. This is the reason why they are `conformal' transformations.

\subsection{Primary operator}
In CFT, it is useful to work in the eigenbasis of dilatation $D$. An operator $O_\Delta$ is called a \textbf{primary operator}\index{primary operator} or in short, \textbf{primary}\index{primary} when
\begin{equation}
    D\, O_\Delta(0) \ket{0}= \Delta\, O_\Delta(0)\ket{0},\quad K_\mu \, O_\Delta (0)\ket{0}=0.
    \label{eq:primary-d-dim}
\end{equation}
The eigenbasis $\ket{O_\Delta}\equiv O_\Delta(0)\ket{0}$ is called a \textbf{primary state}\index{primary state}.
$\Delta\ge 0$ is the \textbf{scaling/conformal dimension}\index{scaling dimension}\index{conformal dimension}. This is because a primary operator transforms under scaling as
\begin{equation}
    O(x)\rightarrow O'(\lambda x)= \lambda^{-\Delta} O(x).
    \label{eq:scaling}
\end{equation}
(When $d=2$, as there is a separation into left and right moving sectors, the scaling dimension is given by the sum of holomorphic and anti-holomorphic \textbf{conformal weights}\index{conformal weight}: $\Delta=h+\bar{h}$.) $O_\Delta(0) \ket{0}$ is called a primary state. A primary operator and a primary state are in a one-to-one correspondence because of the state/operator correspondence\index{state/operator correspondence}, i.e. dilatation can shrink arbitrary path-integrated region defining a state to a point defining an operator.

$K_\mu$ acts like a lowering operator. $P_\mu$ acts like a raising operator:
\begin{equation}
    D\left(P_\mu\, O_\Delta (0)\right)\ket{0} = (\Delta+1)\left(P_\mu\, O_\Delta(0)\right)\ket{0}.
\end{equation}
States with $P_\mu$'s acting on a primary state, $P_\mu P_\nu \cdots O(0) \ket{0}$, are called \textbf{descendant states}\index{descendant states}.

\subsection{Correlation functions}
From conformal symmetry, correlation functions are highly constrained. Since correlators of descendant operators can be obtained by differentiating those of primary operators, we focus on primary operators. Furthermore, for simplicity, we consider only scalars. For more details such as cases involving higher spin fields, see~\cite{Di_Francesco_1997,BB29063466,BB19618269}.

A two-point function of primary operators is completely fixed by conformal symmetry as
\begin{equation}
    \ev{O_{\Delta_1} (x) O_{\Delta_2} (y)} = \frac{\delta_{\Delta_1 \Delta_2}}{\abs{x-y}^{\Delta_1+\Delta_2}}
    \label{eq:CFT-2pt}
\end{equation}
up to normalization. This can be derived by using scaling \eqref{eq:scaling} and translational invariance. Similarly, a three-point function is also fixed except for the operator product expansion coefficient described below.

Higher-point functions are not fully determined by symmetry. However, its functional form is still constrained. For example, a four-point function is written as
\begin{equation}
    \ev{O_1(x_1)O_2(x_2)O_3(x_3)O_4(x_4)}=\frac{g(u,v)}{x_{12}^{\Delta_1+\Delta_2}x_{34}^{\Delta_3+\Delta_4}} \qty(\frac{x_{24}}{x_{14}})^{\Delta_{12}} \qty(\frac{x_{14}}{x_{13}})^{\Delta_{34}}
    \quad (O_i\equiv O_{\Delta_i}),
    \label{eq:4pt-fn}
\end{equation}
where $x_{ij}\equiv \abs{x_i-x_j}$, $\Delta_{ij}=\Delta_i-\Delta_j$, $u$ and $v$ are \textbf{cross ratios}\index{cross ratios} defined as
\begin{equation}
    u=\frac{x_{12}^2 x_{34}^2}{x_{13}^2 x_{24}^2},\quad v=\frac{x_{14}^2 x_{23}^2}{x_{13}^2 x_{24}^2},
    \label{eq:cross-ratios}
\end{equation}
and $g(u,v)$ is the only component that is not determined from symmetry. In CFT, an \textbf{operator product expansion (OPE)}\index{operator product expansion}\index{OPE |see operator product expansion}
\begin{equation}
    O_1(x_1) O_2 (x_2) = \sum_{O_k \in \mathrm{primaries}} c_{12k} (x_{12}, i P_\mu) O_k (x_2)
\end{equation}
is possible as a consequence of the state/operator correspondence. $c_{12k} (x_{12}, i P_\mu)$ is the OPE coefficient\index{OPE coefficient}, which is completely determined from conformal symmetry~\cite{BB29063466}. By applying the OPE to the four-point function \eqref{eq:4pt-fn}, we obtain
\begin{equation}
    g(u,v) = \sum_{O_k} c_{12k} c_{34k} g_{\Delta_k, s_k}^{\Delta_{12},\Delta_{34}}(u,v),
\end{equation}
where $s_k$ is the spin of the operator $O_k$ appearing in the intermediate state of OPE. $g_{\Delta_k, s_k}^{\Delta_{12},\Delta_{34}}(u,v)$ is called the \textbf{conformal block}\index{conformal block}.\footnote{Note that one can always map four points onto some two-dimensional plane by a suitable conformal transformation. Then, we can take a complex coordinate on it such that $x_1=0$, $x_2=1$, $x_3=z$, $x_4=\infty$. The cross ratios are related to $z$ simply by $u=\abs{z}^2$ and $v=\abs{1-z}^2$.} In the above expansion, we applied OPE to $O_1\cdot O_2$ and $O_3\cdot O_4$ but expanding in a different channel $O_1\cdot O_3$ and $O_2\cdot O_4$ should also give the same result. This is known as the \textbf{crossing symmetry}\index{crossing symmetry} and the constraint equation is called the \textbf{conformal bootstrap equation}\index{conformal bootstrap equation}. The conformal bootstrap\index{conformal bootstrap} is a program for investigating the conformal block using the constraint equations.

\section{AdS/CFT correspondence in general dimensions}\label{sec:AdS-CFT-general-dim}
In this section, we discuss some basic statements and known results of the AdS/CFT correspondence\index{AdS/CFT correspondence} in general dimensions. 

\subsection{GKP-Witten relation}
The \textbf{AdS/CFT correspondence}\index{AdS/CFT correspondence} is a duality between theories. The \textbf{Gubser-Klebanov-Polyakov-Witten (GKP-Witten) relation}\index{GKP-Witten relation} claims the gravitational (on-shell) partition function of the AdS spacetime perfectly matches the corresponding CFT (which we call \textbf{holographic CFT}\index{holographic CFT}) partition function~\cite{Gubser:1998bc,Witten:1998qj}:\footnote{There is no $\sqrt{-g}$ since the CFT lives on flat spacetime.}\footnote{Although we focus on the zero-temperature case, a finite-temperature case can be also dealt with by considering a (mixed-state) black hole~\cite{Witten:1998zw}. Its purification is the TFD state \eqref{eq:TFD} and dual to the \textbf{two-sided eternal black hole}\index{two-sided eternal black hole}~\cite{Maldacena:2001kr}. This is the \textbf{Hartle-Hawking(-Israel) wave functional}\index{Hartle-Hawking state}\index{Hartle-Hawking(-Israel) wave functional} prepared by a Euclidean path integral~\cite{Hartle:1976tp,Israel:1976ur}.}
\begin{equation}
    Z_{AdS}[\phi_i|\phi_i(z=\epsilon)\propto\phi_{0\, i}] = \expval{e^{i\int \dd[d+1]{x} O_i(x) \phi_{0 i}(x)}}_{CFT},
    \label{eq:GKPW}
\end{equation}
where $i$ denotes the Lorentz index $\mu$ (e.g. $J_\mu$) or the flavor index (e.g. $\phi_i$ in the $O(N)$ vector model).\footnote{This is different from the color index (e.g. $SU(N)$ for $\mathcal{N}=4$ SYM. Since $\phi_i$ in gravity is gauge singlet, $O_i$ must also be singlet like $\Tr F^2$.} $\phi_{0\, i}$ is the external source field and specifies the boundary condition of the bulk field at the asymptotic boundary. We will later make a more precise statement including the prefactor. 
By differentiating the free energy with respect to $\phi_{0i}$, we can obtain correlators of $O_i$.
In the following, we mostly consider single-flavor cases.\\
\emph{Caveat: Although \eqref{eq:GKPW} is written in the Lorentzian signature, there is a subtlety in this signature compared to the Euclidean one~\cite{Marolf:2004fy}. Precisely speaking, states should be prepared in the Euclidean signature and they are analytically continued to the Lorentzian signature at the time-reflection symmetric slice~\cite{Wakeham:2022wyx}.}

For simplicity, let us restrict to a scalar operator.
An operator $O$ in a holographic CFT couples to a scalar field $\phi$ in the bulk \eqref{eq:GKPW}. Recall that $\Delta$ gives the scaling dimension of the operator. In the AdS/CFT correspondence, the bulk field $\phi\sim z^\Delta$ as it approaches the asymptotic boundary $z\rightarrow 0$ ($r\rightarrow 0$) in the same way as the CFT operator $O$ does. By plugging this into the bulk equation of motion (assuming free QFT in the bulk\footnote{This assumption will be justified later in Section \ref{sec:int-AdS}.}), the conformal dimension\index{conformal dimension} $\Delta$ of the CFT operator $O$ must be either
\begin{equation}
    \Delta_{\pm}\equiv\frac{d+1}{2}\pm\sqrt{(mR)^2+\qty(\frac{d+1}{2})^2} \ \Leftrightarrow\ (mR)^2=\Delta_\pm (\Delta_\pm -(d+1)),
    \label{eq:delta-mass}
\end{equation}
where $m$ is the mass of the bulk scalar field $\phi$.
The conformal dimension $\Delta$ is constrained by the \textbf{unitarity bound}\index{unitarity bound}, e.g. $\Delta\ge (d-1)/2$ for a scalar operator\footnote{This comes from the non-negativity of norms.}~\cite{Andrade:2011dg,BB29063466}. %We have $\Delta\ge (d-1)/2$ for a scalar operator.\footnote{From the conformal block expansion, the operator saturating the unitarity bound is generalized free~\cite{Ohl:2012bk}.} 
When the mass is above the \textbf{Breitenlohner-Freedman (BF) stability bound}\index{Breitenlohner-Freedman stability bound}\index{BF bound |see Breitenlohner-Freedman stability bound}~\cite{Breitenlohner:1982jf,Breitenlohner:1982bm}
\begin{equation}
    (mR)^2 \ge -\qty(\frac{d+1}{2})^2,
    \label{eq:BF-bound}
\end{equation}
we can take either $\Delta_-$ or $\Delta_+$ as $\Delta$ without violating the unitarity bound.\footnote{Note that unless $d=1$, bulk fields within the BF bound are tachyonic. When $d=1$, the operator dual to the bulk massless scalar can have a conformal dimension $\Delta=\Delta_-=0$. This is the only operator satisfying the BF bound.} 
This allows various possibilities for boundary conditions, however, let us focus on the case $\Delta=\Delta_+$, which does not violate the unitarity bound for any values of the mass.
As the bulk field approaches the asymptotic boundary, we have a linear combination of the two independent solutions (as the equation of motion is second-order)
\begin{equation}
    \phi(z,x)\rightarrow z^{\Delta_-} A(x) + z^{\Delta_+} B(x),\quad z\rightarrow 0.
    \label{eq:asympt-field-boundary}
\end{equation}
The first term dominates as $z\rightarrow 0$. This is referred to as the \textbf{non-normalizable mode}\index{non-normalizable mode}. The second term is referred to as the \textbf{normalizable mode}\index{normalizable mode}.
The boundary condition is now fixed as
\begin{equation}
    \phi(z=\epsilon)=\epsilon^{\Delta_-}\phi_0,
    \label{eq:bc1}
\end{equation}
where we put the UV cutoff as $z=\epsilon \ll R$ (the metric is fixed by the Dirichlet boundary condition).
This means we have
\begin{equation}
    A(x)=\phi_0(x).
    \label{eq:bc2}
\end{equation}
The non-normalizable mode\index{non-normalizable mode} diverges as $z\rightarrow 0$. It cannot be quantized and should be treated as a source
\begin{equation}
    \Delta S = \int \dd[d+1]{x} O(x) \phi_{0}(x).
    \label{eq:source-ads}
\end{equation}
Such a modification of the theory is called \textbf{non-normalizable deformation}\index{non-normalizable deformation}.
Given the boundary condition, \eqref{eq:bc1} or \eqref{eq:bc2}, we can write the bulk field $\phi(x,z)$ by dressing with the \textbf{bulk-to-boundary propagator}\index{bulk-to-boundary propagator} $K_\Delta (z,x;y)$\footnote{This is different from the smearing function in the global bulk reconstruction. See Section 3.2 in~\cite{Harlow:2018fse}.}
\begin{equation}
    \phi(x,z)=\int \dd[d+1]{y} K_\Delta (z,x;y) A(y),
\end{equation}
where the bulk-to-boundary propagator is obtained as
\begin{equation}
    K_\Delta (z,x;y)= \frac{\Gamma(\Delta)}{\pi^{(d+1)/2}\Gamma(\Delta-(d+1)/2)}\qty(\frac{z}{z^2+(x-y)^2})^\Delta
\end{equation}
from the Green function~\cite{Aharonov:1988xu,ammon_erdmenger_2015,1020000782220175744}.
Under the non-normalizable deformation \eqref{eq:source-ads}, the bulk-to-boundary propagator $K(z,x;y)$ is completely determined. By expanding $\phi(x,z)$ around $z=0$, we obtain
\begin{equation}
    B(x)=\ev{O(x)}_{A(x)},
\end{equation}
where $\ev{O(x)}_{A(x)}$ is the one-point function in the presence of the source $A(x)$.

This is not the end of the story. Since there are two independent solutions \eqref{eq:asympt-field-boundary}, we can actually choose the coefficient of the normalizable mode\index{normalizable mode} \textit{independently} by adding another normalizable mode $z^{\Delta_+} \tilde{B}(x)$~\cite{nastase_2015}. Since the normalizable mode can be quantized and does not deform the theory at the asymptotic boundary, it corresponds to deforming a state without modifying the theory, i.e. an excited state with finite energy. In particular, we call it a \textbf{local operator quench}\index{local operator quench}\index{quench} if the state is excited by a local (external) operator~\cite{Nozaki:2014hna,Nozaki:2016mcy}.\footnote{Other types of local quenches include \textbf{joining/splitting quenches}\index{joining/splitting quenches}, which involve merging and splitting states. They have been studied both in CFTs~\cite{Calabrese:2007mtj} and holography~\cite{Shimaji:2018czt}. We may also consider a \textbf{global quench}\index{global quench}, where a ground state is suddenly evolved by a different Hamiltonian~\cite{Calabrese:2007rg,Calabrese:2016xau,Calabrese:2006rx,Hartman:2013qma}.}

In the absence of the source, $\phi\rightarrow z^\Delta B(x)$ as $z\rightarrow 0$. $B(x)$ is finite even at the asymptotic boundary and is directly related to the one-point function if the source is present. This indicates the dual CFT operator can be alternatively defined from the bulk field at the operator level as
\begin{equation}
    O(t,x)=\lim_{\epsilon\rightarrow 0} \left.\frac{\phi(t,z,x)}{z^\Delta}\right|_{z=\epsilon}
    \label{eq:dict}
\end{equation}
instead of the correlators from the GKP-Witten relation~\eqref{eq:GKPW}.
This is known as the \textbf{extrapolate dictionary}\index{extrapolate dictionary} or the \textbf{Banks-Douglas-Horowitz-Martinec (BDHM) formula}\index{BDHM formula}~\cite{Banks:1998dd}. The same correlators can be obtained either from the generating functional in the GKP-Witten relation \eqref{eq:GKPW} or the extrapolate dictionary \eqref{eq:dict}. For its proof, see~\cite{Harlow:2011ke} as well as Section 12.1 in~\cite{Kaplan:2016}.\footnote{It is notable these two approaches can give a different answer in the dS/CFT correspondence~\cite{Strominger:2001pn,Witten:2001kn}, an analytical-continued version of the AdS/CFT correspondence~\cite{Harlow:2011ke}.}

After all, the GKP-Witten relation \eqref{eq:GKPW} holds beyond the vacuum case by evaluating the right-hand side with a nontrivial excited state (but in the original theory). The problem is what is the gravity dual to evaluate the on-shell action in the saddle-point approximation of the left-hand side. The geometry is asymptotically AdS but not pure AdS due to the backreaction from the excitation. In the later sections, we discuss such a case explicitly in lower dimensions.

\subsection{Deformation from CFT}
The source term can deform the CFT action to make the theory non-conformal. The bulk geometry can be drastically different depending on the dimensionality of the deformed operator. 
When $\Delta<d$, $O$ is relevant\index{relevant}; the deformation becomes stronger as we go deep into the IR. This means we still have an asymptotically AdS spacetime, however, the bulk interior is different from the pure AdS. From \eqref{eq:delta-mass}, this case corresponds to a tachyonic bulk field ($m^2<0$).\footnote{A tachyonic field is allowed above the BF stability bound \eqref{eq:BF-bound}.} Recently, relevant deformations are important in discussing traversable wormholes\index{traversable wormholes}. See an original proposal by Gao, Jafferis, and Wall~\cite{Gao:2016bin} as well as its boundary CFT analog~\cite{May:2020tch}.
When $\Delta=d$, $O$ is marginal\index{marginal}.\footnote{Recently, it has been discussed that a marginal deformation can add strongly-coupled matter fields in the bulk without being suppressed as $G_N\rightarrow 0$~\cite{Apolo:2022pbq}.} Finally, when $\Delta>d$, $O$ is irrelevant\index{irrelevant}; the deformation is strong in the UV. This implies the UV theory must be renewed. In the bulk, it means the deformation grows as we approach the asymptotic boundary and the asymptotic geometry can be very distinct from pure AdS. %While such an irrelevant deformation changes the UV theory, which is usually the starting point, one might think such an irrelevant deformation is useless. Actually it is not. 
To list a few examples, the \textbf{$\bm{T\overline{T}}$ deformation}\index{$T\overline{T}$ deformation} proposed to be dual to a finite cutoff AdS, in which the CFT lives on a boundary at a finite distance $z>\epsilon$~\cite{McGough:2016lol,Cottrell:2018skz}; a mass deformation of the $\mathcal{N}=4$ SYM by the Konishi operator\index{Konishi operator} $\Tr(\phi^I\phi^I)$ ($I$ is the $SO(6)$ index), whose conformal dimension $\Delta\sim m_s R\sim R/{\alpha'} \sim\lambda^{1/4}$ (from \eqref{eq:delta-mass}) is large in the usual holographic limit~\cite{Aharonov:1988xu}.\footnote{The Konishi operator does not belong to the short representation of the superconformal algebra. Thus, it is not dual to a supergravity field around pure AdS and is expected to correspond to a massive mode in the string spectrum.}

\subsection{Generalized free field}
From the extrapolate dictionary \eqref{eq:dict}, CFT correlators can be calculated from bulk field correlators. To calculate the bulk correlators, we need to consider QFT on AdS, which is what we expect from holography as an effective description of quantum gravity. When the bulk QFT is free, the CFT correlator must obey Wick's theorem\index{Wick's theorem}, in which any multi-point functions factorize into a product of two-point functions. This is remarkable from the CFT perspective! We know a strongly-coupled theory does not usually have such a simple factorization. In the large-$N$ theory,\footnote{For large-$N$ theories, \cite{Moshe:2003xn} provides an excellent, comprehensive review.} it can be shown that all the non-factorized terms are suppressed as $N\rightarrow\infty$ for any operators with $\Delta\sim O(N^0)$ and the number of such operators is $O(N^0)$~\cite{El-Showk:2011yvt,Harlow:2018fse}. For example, see Section 6.6 in~\cite{Kaplan:2016} for an illustration. Wick's theorem holds when $N$ is strictly infinite.\footnote{Note that there is also a subtle difference in the strictly infinite $N$ versus a large but finite $N$ from the von Neumann algebra of QFT~\cite{Chandrasekaran:2022eqq,Leutheusser:2021qhd,Ghosh:2017gtw}. See~\cite{Witten:2021jzq,Witten:2018zxz} for reviews.} This is known as the \textbf{large-$\bm{N}$ factorization}\index{large-$N$ factorization}.

The restriction $\Delta\sim O(N^0)$ is very important. These operators dual to free bulk fields are called \textbf{generalized free}\index{generalized free} fields. The corresponding CFT sector is called the generalized free CFT. Note that the generalized free CFT cannot be a complete theory alone. What about $\Delta\sim O(N)$? Such operators cannot be dual to free bulk fields on pure AdS. They are heavy enough to backreact on the geometry. The bulk is no longer pure AdS but an asymptotically AdS like a black hole. In the next section, we will see an explicit example of such situations in the AdS${}_3$/CFT${}_2$ correspondence.

Since there is a large `gap' between the low-energy operators with $\Delta\sim O(N^0)$ and high-energy (stringy) ones with $\Delta\sim O(N)$, holographic CFTs are also called the \textbf{large-$\bm{N}$ \emph{gapped} CFTs}\index{large-$N$ gapped CFTs}.\footnote{In addition, we need the theory to be strongly-coupled to have a large gap~\cite{Hartman:2015,Wakeham:2022wyx}.} This is also referred to as a \textbf{sparse (low-lying) spectrum}\index{sparse spectrum}.\footnote{This terminology is mostly used in two-dimensional holographic CFTs, in which sparseness is more precisely defined~\cite{Hartman:2014oaa,Hartman:2013mia}.} 
This is expected from the Hawking-Page transition\index{Hawking-Page transition}~\cite{Hawking:1982dh}, in which the thermodynamic entropy counting the number of microstates jumps from $O(N^0)$ to $O(N)\sim O(1/G_N)$. For more details, see Section 1.1.4 of \cite{Wakeham:2022wyx}.

\subsection{Interactions in AdS}\label{sec:int-AdS}
In the above discussion, we have neglected interactions among the bulk fields without justifications. (Gravitational interactions are an exception. They are already taken into account as backreaction.) We can show that interacting terms are suppressed in the small-$G_N$ (large-$N$) limit for the semiclassical\index{semiclassical} gravity. Let us assume there is no dimensionful coupling constant other than mass and Newton's constant. Recalling that the mass dimension of a scalar field $\phi$ is $(d-1)/2$ and that of the gravitational constant $G_N$ is~$-(d-1)$, the action is given by~\cite{Kajuri:2020vxf}
\begin{equation}
    I=\frac{1}{16\pi G_N} \int\dd[d+1]{x} \sqrt{-g} \left[R + G_N \left((\nabla \phi)^2+m^2 \phi^2\right) + \lambda_3 G_N^{3/2} \phi^3 + \cdots + \lambda_k G_N^{(k+2)/2} \phi^k +\cdots \right],
\end{equation}
where $\lambda_k$ denotes a dimensionless coupling constant for the $k$-th interaction. Although we restricted to a scalar QFT for simplicity, a similar expansion can be done for other fields. From the above equation, we can straightforwardly see the effect of interactions is suppressed in the small-$G_N$ (i.e. large-$N$) limit. Thus, as long as we focus on the leading order in $G_N$, we can neglect interactions in AdS.

\subsection{Geodesic approximation}\label{sec:geodesic-approx}
The AdS/CFT correspondence tells us the CFT two-point function is calculated from the bulk two-point function. In the \textbf{worldline formalism}\index{worldline formalism} (cf. Appendix \ref{app:worldline}), it can be written as a propagation of a relativistic point particle. Since its action is proportional to mass, which is dual to the conformal dimension, the two-point function of heavy fields\index{heavy} $\Delta\gg 1$ can be well approximated by the saddle point, the geodesics. In this limit, $\Delta\approx mR$ from \eqref{eq:delta-mass}. It follows that the Euclidean two-point function of bulk fields is given by~\footnote{The prefactor is determined from the normalization of the Green's function. The Klein-Gordon operator in AdS is essentially equal to the quadratic Casimir operator of $SO(2,d)$, from which we can compute the Green's function with this prefactor. For more details, see Section 3.2 of~\cite{Penedones:2016voo}.}
\begin{equation}
    \ev{\phi(X_A)\phi(Y_B)}\approx \frac{\Gamma(\Delta)}{\Gamma(\Delta+1-d/2)}e^{-\Delta d(A,B)/R},
    \label{eq:geodesic-approx}
\end{equation}
where $d(A,B)$ is the geodesic distance between $A$ and $B$ in Euclidean AdS (Appendix \ref{sec:Embedd}). This is called the \textbf{geodesic approximation}\index{geodesic approximation}~\cite{Balasubramanian:1999zv,Louko:2000tp}. Appendix \ref{app:geodesic-approx} shows that this correctly reproduces the CFT two-point function.

\subsection{Holographic stress tensor}\label{sec:hol-EM-FG}
Since the energy-momentum tensor\index{energy-momentum tensor} comes from the variation with respect to the metric perturbation, we first need to introduce coordinates describing an asymptotically AdS spacetime, namely, the \textbf{Fefferman-Graham (FG) coordinates}\index{Fefferman-Graham coordinates}\index{FG coordinates |see Fefferman-Graham coordinates}~\cite{AST_1985__S131__95_0,Fefferman:2007rka}.

Following the notation in~\cite{Nozaki:2013wia}, an asymptotically AdS spacetime\index{asymptotically AdS spacetimes} is expressed by the following metric\footnote{It is often expressed as 
\begin{equation}
    \dd{s}^2= R^2 \frac{\dd{\rho}^2}{4\rho^2}+\frac{g_{ab}(\rho,x)\dd{x}^a\dd{x}^b}{\rho}
    \label{eq:metric-FG}
\end{equation}
by taking $\rho=z^2/R^2$~\cite{ammon_erdmenger_2015,1020000782220175744}.}:
\begin{equation}
    \dd{s}^2=R^2\frac{\dd{z}^2 + g_{ab}(z,x) \dd{x}^a \dd{x}^b}{z^2} \equiv R^2 \frac{\dd{z}^2}{z^2}+\gamma_{ab}\dd{x}^a \dd{x}^b,
\end{equation}
where the asymptotic boundary is located at $z\rightarrow 0$. The metric perturbation behaves around $z=\epsilon\rightarrow 0$ like
\begin{equation}
    g_{ab}(z,x)=\eta_{ab}+t_{ab}(x)\, z^d +O(z^{d+1}).
    \label{eq:FG-expansion}
\end{equation}
Coefficients of the FG expansion \eqref{eq:FG-expansion} can be calculated order by order from the Einstein's equation.

By plugging the FG expansion into the action, we can compute the counterterm action to remove divergences as $\epsilon\rightarrow 0$. This is known as \textbf{holographic renormalization}\index{holographic renormalization}. For more details, see Appendix \ref{app:grav-action}. The entire gravity action is given by the sum of the Einstein-Hilbert term\index{Einstein-Hilbert term}, the Gibbons-Hawking boundary term\index{Gibbons-Hawking boundary term}, and the counterterm at the asymptotic boundary. By taking its variation with respect to the metric, we obtain the expectation value of the energy-momentum tensor (we omit the bracket)\footnote{We define the energy-momentum tensor by \begin{equation}
    T_{ab}=-\frac{2}{\sqrt{-\gamma}}\frac{\delta I}{\delta \gamma^{ab}},\quad T^{ab}=\frac{2}{\sqrt{-\gamma}}\frac{\delta I}{\delta \gamma_{ab}}.
    \label{eq:EM-tensor-def}
\end{equation}
}:
\begin{equation}
    T_{ab}=\lim_{z\rightarrow 0} \qty(\frac{R^{d-1}}{8\pi G_N z^{d-1}}(K_{ab}-\gamma_{ab}K)-\frac{2}{\sqrt{-\gamma}}\frac{\delta I_{ct}}{\delta \gamma^{ab}}).
\end{equation}
The contribution from the Einstein-Hilbert action vanishes from the equation of motion. The counterterm contribution depends on the spacetime dimension. At least up to $d=4$, they are a linear combination of the induced metric $\gamma_{ab}$ and the Einstein tensor for the induced metric. When the asymptotic boundary is flat like the FG coordinates, which we are interested in now, the Einstein tensor (proportional to the boundary cosmological constant) and higher curvature terms vanish. This yields a universal expression
\begin{equation}
    T_{ab}=\lim_{z\rightarrow 0} \qty[\frac{R^{d-1}}{8\pi G_N z^{d-1}}\qty(K_{ab}-\gamma_{ab}K - \frac{d}{R}\gamma_{ab})].
\end{equation}
By plugging \eqref{eq:metric-FG} with \eqref{eq:FG-expansion} into the above expression, we obtain the \textbf{holographic stress tensor}\index{holographic stress tensor}
\begin{equation}
    T_{ab}=\frac{(d+1)R^d}{16\pi G_N} t_{ab}.
    \label{eq:hol-EM-tensor}
\end{equation}
From the FG expansion, $t_{ab}$ should have mass dimension $d$. This can be confirmed from the dimensional analysis as follows. First, $T_{ab}$ integrated over a time slice of the asymptotic boundary must have mass dimension of energy.\footnote{
Let us denote the bulk Cauchy surface by $\Sigma$ and its asymptotic boundary counterpart by $\partial\Sigma$.
Given the Arnowitt-Deser-Misner (ADM) decomposition\index{ADM decomposition} of the metric on $\partial$AdS $\gamma_{ab}\dd{x}^ab\dd{x}^b = - N^2 \dd{t}^2 + \sigma_{ij} (\dd{x}^i + N^i \dd{t})(\dd{x}^j + N^j \dd{t})$, where $N$ is a lapse and $N^i$ are shift functions and $\sigma_{ij}$ is the space metric on $\partial \Sigma$,
the ADM energy\index{ADM energy} is defined as
\begin{equation}
    \int_{\partial \Sigma}\dd[d-1]x \sqrt{\sigma} N \epsilon,
\end{equation}
where $\epsilon$ is a timelike unit normal vector to $\partial\Sigma$~\cite{Balasubramanian:1999re}.}
This means $[T_{ab}]=M^{-1-(d-1)}=M^d$. Recalling $[G_N]=M^{-(d-1)}$, \eqref{eq:hol-EM-tensor} yields $[t_{ab}]=M^d$. We will use \eqref{eq:hol-EM-tensor} later for a consistency check of holography.

\subsection{Holographic entanglement entropy -- formula}\label{sec:HEE-CFT}
In the AdS/CFT correspondence, EE can be calculated in the bulk as the area of an extremal surface homologous to the asymptotic boundary subregion~\cite{Ryu:2006bv,Ryu:2006ef,Hubeny:2007xt}.\footnote{If there are several extremal surfaces, one chooses the minimum-area surface.} The original proposal is called the \textbf{Ryu-Takayanagi (RT) formula}\index{Ryu-Takayanagi formula}\index{RT formula |see Ryu-Takayanagi formula }~\cite{Ryu:2006bv,Ryu:2006ef} and derived in~\cite{Lewkowycz:2013nqa}. This states
\begin{equation}
    S_A = \min_{\gamma} \frac{\mathrm{Area}(\gamma)}{4G_N},
    \label{eq:HEE-RT}
\end{equation}
where the minimization is taken over all possible codimension-two spacelike surfaces such that $\partial\gamma=\partial A$ and homologous to the subregion $A$. The minimal surface $\gamma$ is called the \textbf{RT surface}\index{RT surface}\index{Ryu-Takayanagi surface} and the region enclosed by the RT surface with the asymptotic boundary is called the \textbf{entanglement wedge}\index{entanglement wedge}~\cite{Headrick:2014cta,Jafferis:2015del}. The RT formula works only when the spacetime is static or the RT surface is evaluated at the time reflection symmetric time slice; otherwise one needs to use its covariant version, \textbf{Hubeny-Rangamani-Takayanagi (HRT) formula}\index{HRT formula |see Hubeny-Rangamani-Takayanagi formula }\index{Hubeny-Rangamani-Takayanagi formula}~\cite{Hubeny:2007xt}. When the bulk matter fields are also present, we need to extremize both the area term and the matter entropy at the same time. These two terms are called \textbf{generalized entropy}\index{generalized entropy}~\cite{Bekenstein:1972tm} and extremizing over the generalized entropy is the most general scheme so far.\footnote{In lower dimensions, there is some recent progress such as islands or replica wormholes~\cite{Almheiri:2019qdq,Penington:2019kki}.} This is called \textbf{quantum extremal surface prescription}\index{quantum extremal surface prescription}~\cite{Engelhardt:2014gca}. These holographic calculations of EE are collectively referred to as \textbf{holographic EE (HEE)}\index{holographic entanglement entropy}\index{HEE |see holographic entanglement entropy }.
A series of related references can be found in Section 2 of~\cite{Bousso:2022ntt}.
Also, there are several equivalent formulations such as bit thread\index{bit thread}~\cite{Freedman:2016zud,Headrick:2017ucz,Agon:2018lwq} or the maximin construction\index{maximin construction}~\cite{Wall:2012uf,Akers:2019lzs}.

%\subsection{Shockwave}
%Shockwave can be created both by heavy and light operators. In the large $c$ limit, the latter gravitional backreaction is suppressed. To create shockwave from light operators, the wave packet must be complexified. To calculate these correlation functions, we should not use OPE but instead work in the Regge limit. For more detail, see~\cite{Afkhami-Jeddi:2017rmx} (See also the first paper discussing the CFT dual for a shockwave~\cite{Horowitz:1999gf}).

\section{Toward AdS${}_3$/CFT${}_2$ correspondence}\label{sec:AdS3/CFT2}
In this dissertation, we focus on low-dimensional holography, in particular, the \textbf{AdS${}_3$/CFT${}_2$ correspondence}\index{AdS${}_3$/CFT${}_2$ correspondence}. It is a good %playground
laboratory for higher-dimensional holography 
%starting point 
as it is nontrivial enough to contain black holes, branes, etc. while it is much simpler than higher-dimensional cases.\footnote{In higher dimensions, supersymmetry\index{supersymmetry} constrains the possible forms of interactions~\cite{McGreevy:2009xe}. In the AdS${}_3$/CFT${}_2$ correspondence, supersymmetry is not mandatory as the enhanced conformal symmetry is strong enough for computations in many cases.} By utilizing the (enhanced) symmetry, we can compute various quantities analytically, and obtain universal results independent from the details of the theory. In this section, we start with inspecting AdS and CFT separately (neglecting the effect of compactified spacetimes) although it is worth noting that the Poincar\'e AdS${}_3$ can appear from string theory by taking the near horizon limit of the D1-D5 system~\cite{David:2002wn}.

\subsection{AdS${}_3$}\label{sec:AdS3}

Three-dimensional gravity is simpler than higher-dimensional gravity as there are no propagating degrees of freedom like graviton (except at the asymptotic boundary, depending on the boundary condition), i.e. the Weyl tensor\index{Weyl tensor} vanishes. This implies the Riemann tensor is completely fixed by the Ricci tensor. As a consequence, the Riemann tensor is always written as
\begin{equation}
    \mathcal{R}_{\mu\nu\rho\sigma}=
    \Lambda (g_{\mu\rho}g_{\nu\sigma}-g_{\mu\sigma}g_{\nu\rho})\quad (\mathcal{R}_{\mu\nu}=2\Lambda g_{\mu\nu}),
    \label{eq:max-sym}
\end{equation}
where $\Lambda$ is a real-valued constant representing the cosmological constant.
A spacetime satisfying \eqref{eq:max-sym} is called \textbf{maximally symmetric spacetime}\index{maximally symmetric}.

What is special about three-dimensional gravity is that these maximally symmetric spacetimes are \emph{unique solutions} to Einstein's equation. They are classified into \textbf{AdS} ($\Lambda<0$), \textbf{dS} ($\Lambda>0$), and \textbf{Minkowski} ($\Lambda=0$). For holography, we consider AdS. Since Einstein's equation determines the local structure, all three-dimensional gravity solutions with $\Lambda<0$ are locally AdS. Nevertheless, there are nontrivial as we have seen in Section \ref{sec:asympt-AdS} since their global structure can be different by identifications (quotienting)~\cite{Banados:1992wn,Banados:1992gq}.

Since all solutions with $\Lambda<0$ are locally AdS, we can always perform a coordinate transformation so that the metric becomes that of pure AdS. This greatly simplifies various calculations. Furthermore, this indicates AdS isometry\index{AdS isometry} $SO(2,d)$ is enhanced. Indeed, the Killing vector fields of the asymptotic symmetry generate an infinite-dimensional algebra~\cite{Brown:1986nw,Balasubramanian:1998sn}. %This is crucial for the CFT${}_2$ at the asymptotic boundary as we see in the next section.

One may worry about suppressing higher curvature terms from stringy corrections. However, as we have discussed, three-dimensional gravity can be entirely written as non-Riemann curvatures \eqref{eq:max-sym} and they can be removed via field redefinitions and renormalization of $\Lambda$, resulting in pure Einstein-Hilbert action with cosmological constant and \textbf{(gravitational) Chern-Simons action}\index{gravitational Chern-Simons action}\footnote{In the presence of matter fields, there can be a gauge Chern-Simons term. Supersymmetry\index{supersymmetry} can relate the coefficient of these Chern-Simons terms to the cosmological constant. See references in~\cite{Gupta:2007th}.}~\cite{Gupta:2007th,Witten:2007kt,David:2007ak}
\begin{equation}
    I_{CS}=K \int \dd[3]{x} \epsilon^{\mu\nu\rho} \qty(\frac{1}{2} \Gamma^\tau_{\mu\sigma} \partial_\nu \Gamma^\sigma_{\rho\tau} +\frac{1}{3} \Gamma^\tau_{\mu\sigma} \Gamma^\sigma_{\nu\kappa} \Gamma^\kappa_{\rho\tau} ),
\end{equation}
where $K$ is a constant (typically of order $O(\alpha')$).
When the AdS radius is constant, the variation of the Chern-Simons action vanishes. Thus, it does not play any role as long as we stick to the on-shell action.

\subsection{CFT${}_2$}\label{sec:CFT2}
In the following subsections, we review some basic properties of two-dimensional CFTs. These notes are based on~\cite{BB19618269,1020000782220175744} but all the basic materials are also available in standard textbooks like~\cite{Di_Francesco_1997,Rangamani_2017}.

%\subsection{Conformal transformation in $(1+1)$ dimensions}
Two-dimensional CFT is very special. By taking a complex coordinate for $\mathbb{R}^2\simeq \mathbb{C}$ together with $\cup \{\infty\}$ (known as $\mathbb{CP}^1$), the holomorphic ($z$) and anti-holomorphic ($\bar{z}$) parts of any conformal transformations\index{conformal transformations} (more precisely, the Weyl transformation\index{Weyl transformation}) are factorized. %\footnote{For consistency throughout this dissertation, we denote the complex coordinates on plane by $w$ instead of $z$, which is more commonly used (e.g.~\cite{BB19618269}).} 
Thus, all holomorphic/anti-holomorphic coordinate transformations are conformal. This fact greatly simplifies the calculation of correlation functions among arbitrary spacetime points and on various path integral geometries.
In the following discussion, we mostly focus on the holomorphic sector unless noted.

The miraculous property of two-dimensional CFT is not only the factorization. The conformal symmetry is enhanced to infinite-dimensional \textbf{Virasoro symmetry}\index{Virasoro symmetry}. Originally, conformal transformations are just $SO(2,2)\simeq SL(2,\mathbb{C})$. These are \textbf{\emph{global} conformal transformations}\index{global conformal transformations}, where they are holomorphic on the entire plane. In particular, finite global conformal transformations on $\mathbb{CP}^1$ are written as the \textbf{M\"obious transformation}\index{M\"obious transformation}
\begin{equation}
    z%= x+it_E 
    \mapsto z'=\frac{az+b}{cz+d},
    \label{eq:global-conf-trf}
\end{equation}
where $a,b,c,d\in\mathbb{C}$ and $ad-bc=1$. On the other hand, we can consider \textbf{\emph{local} conformal transformations}\index{local conformal transformations} $z\mapsto z'=f(z);\quad \bar{z}\mapsto \bar{z}'=\bar{f(z)}\equiv \bar{f}(\bar{z})$. They are responsible for the enhancement of symmetry.

\subsubsection{Operators}
Now, let us take a closer look at the operator transformation.
As we have mentioned, since conformal transformations are factorized into holomorphic and anti-holomorphic sectors, we can also treat a (primary) operator\index{primary operator} \emph{as if} it is factorized into a product $O(z,\bar{z})=O(z)\bar{O}(\bar{z})$ as long as we consider a conformal transformation.\footnote{We do not mean the operator is really factorized. It is just a convenient choice for calculations.} The conformal dimension\index{conformal dimension} is also split into holomorphic and anti-holomorphic \textbf{conformal weights}\index{conformal weights}: $\Delta=h+\bar{h}$. For operators with zero spins, $s=h-\bar{h}=0$ thus $h=\Delta/2$.

In two-dimensional CFT, a primary operator\index{primary operator} is an operator which transforms as\footnote{We sometimes call $\qty(\frac{dz'}{dz})^h$ a conformal/Weyl factor\index{conformal factor}\index{Weyl factor}.}
\begin{equation}
    O(z)=O'(z')\qty(\frac{dz'}{dz})^h\quad (h:\text{conformal weight of } O)
    \label{eq:primary-2d}
\end{equation}
under an \emph{arbitrary} holomorphic function (i.e. local conformal transformation). When \eqref{eq:primary-2d} holds only for global conformal transformations, $O$ is called a \textbf{quasi-primary operator}\index{quasi-primary operator}.
\eqref{eq:primary-2d} reproduces \eqref{eq:scaling}, the transformation rule for primary operators in CFT${}_d$. Note that \eqref{eq:primary-d-dim} is the necessary and sufficient condition for $O$ to be primary in CFT${}_{d\ge 3}$ but not in CFT${}_2$. They are only the necessary condition. Quasi-primary operators do satisfy \eqref{eq:primary-d-dim} but not primary since by definition \eqref{eq:primary-2d} does not hold for them unless $f(z)$ is a global conformal transformation.

Let us consider when the local conformal transformation\index{local conformal transformations} is infinitesimal, $z' \equiv z+ \epsilon^z(z)$, where $\epsilon^z(z)$ is an infinitesimally small holomorphic function. If we choose $\epsilon^z(z)=\epsilon l_n \equiv \epsilon z^{n+1}$, i.e.
\begin{equation}
    z' = z+\epsilon z^{n+1}\ (n\in\mathbb{Z}),
\end{equation}
where $\epsilon\ll 1$ is an infinitesimal parameter, the transformation law \eqref{eq:primary-2d} becomes
\begin{equation}
    \delta_{\epsilon} O (z) \equiv O'(z')-O(z) = \epsilon\qty[ z^{n+1} \partial_z O(z) +h(n+1) z^n O(z) ].
    \label{eq:infinitesimal-trf}
\end{equation}

\subsubsection{Virasoro algebra}
From the infinitesimal local conformal transformations \eqref{eq:infinitesimal-trf}, we can construct an algebra generating them. They are called the \textbf{Virasoro algebra}\index{Virasoro algebra}. The constituting operators are \textbf{Virasoro operators}\index{Virasoro operators} $\{L_n\}$, which follow \eqref{eq:infinitesimal-trf}:
\begin{equation}
    \comm{L_n}{O(z)}= z^{n+1} \partial_z O(z) +h(n+1) z^n O(z).
    \label{eq:infinitesimal-trf2}
\end{equation}
Since the infinitesimal transformation of translation corresponds to $n=-1$, that of rotation and dilatation corresponds to $n=0$, and that of special conformal transformation corresponds to $n=1$, $L_{\pm 1}$ and $L_0$ generate global conformal transformations.

We can now examine the Virasoro algebra by considering how a primary operator transforms with the action $\comm{L_m}{L_n}$. One would expect the commutation relation is given by
\begin{equation}
    \comm{L_m}{L_n} \stackrel{?}{=} (m-n) L_{m+n}. 
\end{equation}
Is this the Virasoro algebra? No! This is known as the Witt algebra\index{Witt algebra}. The Virasoro algebra is actually its central extension. We know adding c-numbers to the element does not change \eqref{eq:infinitesimal-trf2}. This means we are free to add a c-number on the right-hand side of the commutation relation. 
%To decide the central extension term, recall that an infinitesimal transformation must be generated by a conserved charge. Since the transformation we are considering is the Weyl transformation, the Noether charge should be given by an integral of the energy-momentum tensor.
To determine the central extension term, we need to know which operator in CFT generates the variation of fields under conformal transformations. From the \textbf{conformal Ward identity}\index{conformal Ward identity} (Appendix \ref{app:conf-ward}), we know
the energy-momentum tensor\index{energy-momentum tensor} generates a conformal transformation:
\begin{equation}
    \delta_\epsilon O(z_i) = \frac{1}{2\pi i} \oint_{z_i} \dd{z} \epsilon^z (z) T(z) O(z_i).
    \label{eq:CWI-body}
\end{equation}
Plugging in \eqref{eq:CWI-body} a specific transformation $\epsilon^z(z)=\epsilon z^{n+1}$ should give us the Virasoro operator $L_n$ in terms of the energy-momentum tensor. If we always take the radial ordering, we have $\delta_\epsilon O = L_n O$ for the choice of $\epsilon$. Equating these two, they are related at the operator level as
\begin{equation}
    L_n = \frac{1}{2\pi i} \oint \dd{z} z^{n+1} T(z),\quad \forall n\in \mathbb{Z}.
    \label{eq:virasoro-em}
\end{equation}
This indicates $\{L_n\}$ are nothing but the Laurent expansion coefficients of $T(z)$. The correct commutation relation defining the Virasoro algebra is derived from the commutator of the energy-momentum tensor. The details will be discussed later in this subsection. For a moment, we just introduce Virasoro algebra\index{Virasoro algebra}; the commutation relation is
\begin{equation}
    \comm{L_m}{L_n}= (m-n)L_{m+n} +\frac{c}{12}m(m^2-1)\delta_{m+n,0}.
\end{equation}
$c$ is called the \textbf{central charge}\index{central charge}\index{$c$}. It is known that $c$ counts the number of degrees of freedom, thus it is a monotonically decreasing function along the RG flow ($c$-theorem\index{$c$-theorem})~\cite{Zamolodchikov:1986gt}.\footnote{For example, $c=N$ for the $N$ free boson system.} We will comment on this  aspect more in the next section.

\subsubsection{States}
Since $L_0$ corresponds to dilatation, the (quasi-)primary state\index{primary state} $\ket{O}\equiv O(0)\ket{0}$ should satisfy \eqref{eq:primary-d-dim}. In Virasoro algebra, this is translated to
\begin{equation}
    L_0 \ket{O}= L_0 O \ket{0}= (\comm{L_0}{O}+OL_0)\ket{0} = \lim_{z\rightarrow 0}(z\partial_z O(z) +h O(z)) \ket{0} = h\ket{O}
\end{equation}
and
\begin{equation}
    L_{1}\ket{O}=0. \label{eq:spe-conf-2d}
\end{equation}
Note that the CFT vacuum $\ket{0}$ is invariant under the global conformal transformation\index{global conformal transformation}, i.e. $L_0\ket{0}=0$. In fact, \eqref{eq:spe-conf-2d} can be extended to $\forall n\ge 1$.

In CFT${}_2$, it is convenient to use the \textbf{radial quantization}\index{radial quantization}, in which operators are quantized according to radially growing time on the complex plane. By a conformal transformation $z=e^{iw}$, the origin $z=0$ is mapped to $w=i\infty$; the infinity $z=\infty$ is mapped to $w=-i\infty$. If we parametrize $w$ by $w=\theta-i\tau_E$, we have the following correspondence:
\begin{align}
    z=0\ \text{(origin at plane)} &\Leftrightarrow \tau_E = -\infty \ \text{(infinite past on cylinder)} \\
    z=\infty\ \text{(infinity at plane)} &\Leftrightarrow \tau_E = \infty \ \text{(infinite future on cylinder)}.
\end{align}
Note that since $z=e^{iw}$ has a $2\pi$-periodicity along $\theta$, $w$-coordinates live on a cylinder $\mathbb{R}\times S^1$.
From the correspondence, it is manifest that the radial quantization is equivalent to the usual quantization on the cylinder. The ket state is defined by the Euclidean path integral from infinite past on the cylinder, i.e. $z=0$ in the original coordinates as we have seen. On the other hand, the bra state is defined from infinite future, i.e. $z=\infty$. By mapping to its reciprocal, we can define the bra vector at the origin. Back to the original coordinates, we need the conformal factor corresponding to inversion:
\begin{equation}
    \bra{O}\equiv\lim_{z\rightarrow \infty} z^{2h} \bra{0} O(z).
\end{equation}
Note that $\bra{0}$ satisfies $\bra{0}L_n=0$ for $\forall n \le -1$. This can be understood from holomorphicity around $z=\infty$. From \eqref{eq:virasoro-em}, if there is no operator insertion, the residue at $z=\infty$ becomes zero whenever $n \le -1$. A similar argument holds for $\ket{0}$ by discussing $L_n$ around $z=0$.

\subsubsection{Primary operators}
Correlators of primary operators\index{primary operator} are similar to the higher dimensional cases. The two-point function of primary operators is
\begin{equation}
    O(z)O(w)\sim \frac{1}{(z-w)^{2h}}.
    \label{eq:2pt-2d}
\end{equation}
The three-point function is completely fixed except for the OPE coefficients.

As we have seen, the four-point function is governed by two cross ratios. This is also true for the two-dimensional case.
We can always move three points to a desired location by a global conformal transformation\index{global conformal transformation}. The canonical choice of four-point correlator is
\begin{equation}
    \ev{O_1(z_1,\bar{z}_1)O_2(z_2,\bar{z}_2)O_3(z_3,\bar{z}_3)O_4(z_4,\bar{z}_4)}\mapsto G(z,\bar{z}) \equiv \ev{O_4(\infty)O_3(z,\bar{z})O_2(1)O_1(0)}.
\end{equation}
This is possible by
\begin{equation}
    z\mapsto \frac{z_3-z_1}{z_3-z_4}\frac{z-z_4}{z-z_1}\equiv z'.
\end{equation}
Indeed, this is a global conformal transformation\index{global conformal transformation} \eqref{eq:global-conf-trf}. $z$ is related to the cross ratios \eqref{eq:cross-ratios}, $\abs{z'}^2=u^{-1}$ and $\abs{1-z'}^2=v^{-1}$.\footnote{One can remove the inverse by an additional global conformal transformation.}
Furthermore, using OPE twice, we can expand into conformal blocks\index{conformal block} as we have discussed earlier. The difference from previous cases is that the holomorphic and anti-holomorphic sectors are factorized~\cite{Belavin:1984vu,Caputa:2014vaa}:
\begin{equation}
    G(z,\bar{z})=\sum_p a_p \mathcal{F}(c,h_p,h_i,1-z) \mathcal{F}(c,\bar{h}_p,\bar{h}_i, 1-\bar{z}),
    \label{eq:conf-block-gen}
\end{equation}
where $p$ labels the intermediate operator, $h_p$ is its conformal weight, $h_i$ collectively denotes all the conformal weights of external operators $O_{1,2,3,4}$, and $a_p$ is a product of OPE coefficients $c_{12p}c_{34p}$. $\mathcal{F}$ is known as the \textbf{Virasoro conformal block}\index{Virasoro conformal block} or \textbf{conformal partial wave}\index{conformal partial wave}.

\subsubsection{Energy-momentum tensor}
%Ward identity
Other important correlators are those involving the energy-momentum tensor (also known as stress (energy) tensor)\index{energy-momentum tensor}\index{stress (energy) tensor}.
The trace of the energy-momentum tensor vanishes classically due to the Weyl invariance\index{Weyl invariance}. However, at the quantum level, it can have a finite value depending on the spacetime curvature
\begin{equation}
    \ev{T^\mu_\mu}\equiv \ev{-g^{\mu\nu}\frac{2}{\sqrt{-g}}\frac{\delta I}{\delta g^{\mu\nu}}}=g^{\mu\nu}\frac{2}{\sqrt{-g}}\frac{\delta W}{\delta g^{\mu\nu}} =\frac{1}{2\pi}\frac{c}{12}\mathcal{R},
    \label{eq:weyl-anomaly}
\end{equation}
where $W=-\log Z$.
This is known as the \textbf{trace} or \textbf{Weyl anomaly}\index{trace anomaly}\index{Weyl anomaly}. It is proportional to the central charge\index{central charge} $c$.\index{$c$} From this, it follows that
\begin{align}
    \ev{\delta I} &= \frac{1}{2} \int \dd[d]{x} \sqrt{-g} \ev{T^{\mu\nu}} \delta g_{\mu\nu}\\
    & = \int \dd[d]{x} \sqrt{-g} \sigma (x) \frac{c}{24\pi}\mathcal{R}(x)
    \label{eq:weyl-anomaly-2d}
\end{align}
for a \textbf{Weyl transformation}\index{Weyl transformation} $\delta g_{\mu\nu}= 2 \sigma(x) g_{\mu\nu}(x)$.\footnote{It is worth noting that when $\sigma(x)$ is constant, \eqref{eq:weyl-anomaly-2d} is proportional to $c$ times the Euler number\index{Euler number} $\chi$ since $\int \sqrt{-g} \mathcal{R}=4\pi \chi$. This is common for four dimension cases, which will be briefly commented in Section \ref{sec:hol-CFT2}.}
Thus, even if the trace part of the energy-momentum tensor vanishes on the flat spacetime, it can be nonzero after a local conformal transformation, which distorts the metric.
%For example, a thermal CFT corresponding to the Euclidean path integral on a cylinder has a negative (one-point function of) energy-momentum tensor. This is nothing but the Casimir energy. More explicitly, it is a consequence of Schwarzian in the transformation rule for the energy-momentum tensor.

A slight generalization of \eqref{eq:weyl-anomaly} leads to
\begin{equation}
    \ev{T^\mu_\mu X }=-\frac{c}{12} \mathcal{R}\ev{X}.
\end{equation}
for some field $X$. By rewriting the curvature in terms of the metric perturbation, we can derive OPEs\index{operator product expansion} for $T(z)O(z')$, $T(z)T(z')$, etc~\cite{BB19618269}. The correlation function made of multiple energy-momentum tensors can be generated by taking a variation with respect to the metric. Without derivations, we list a few important OPEs related to the energy-momentum tensor:
\begin{align}
    T(z)O(w) &\sim \frac{h}{(z-w)^2}+\frac{\partial_w O(w)}{z-w}\qq{($O$: primary operator)} \label{eq:TO-OPE}\\
    T(z) T(w) &\sim \frac{c/2}{(z-w)^4} + \frac{2T(w)}{(z-w)^2} + \frac{\partial_w T(w)}{z-w},\label{eq:TT}
\end{align}
where $\sim$ indicates the expansion when $z\approx w$ and all regular terms are omitted. We may define the central charge with \eqref{eq:TT}.
From \eqref{eq:TT}, we know the energy-momentum tensor is not primary but quasi-primary since its two-point function does not obey \eqref{eq:2pt-2d}. Thus, the transformation rule for the energy-momentum tensor also deviates from \eqref{eq:primary-2d}. As a local conformal transformation is generated by the energy-momentum tensor, we can read off its finite transformation rule from \eqref{eq:TT}:
\begin{equation}
    \langle T(w) \rangle = \left(\frac{dw^\prime}{dw}\right)^2 \left( \langle T'(w^\prime) \rangle -\frac{c}{12}\{w;w^\prime\} \right),
    \label{eq:EM-trf-rule}
\end{equation}
where
\ba
\{w;w^\prime\}\equiv\frac{\de_{w^\prime}^3 w}{\de_{w^\prime} w} - \frac{3}{2} \left(\frac{\de_{w^\prime}^2 w}{\de_{w^\prime} w}\right)^2
\ea
is the \textbf{Schwarzian derivative}\index{Schwarzian derivative}\index{Schwarzian |see Schwarzian derivative}.\footnote{The Schwarzian satisfies
\begin{equation}
    \{w;z\}=-\qty(\frac{dw}{dz})^2 \{z;w\}
\end{equation}
so that \eqref{eq:EM-trf-rule} can be reversed.} The Schwarzian term represents the deviation from the transformation law of primary operators, i.e. the Schwarzian derivative vanishes for global conformal transformations. From this, even if the energy-momentum tensor vanishes on the plane,\footnote{This is because the vacuum on the plane satisfies
\[L_{n\ge -1}\ket{0}=0 = \bra{0} L_{n\le 1} \Leftrightarrow \mel{0}{L_{\forall n}}{0}=0.\] 
This implies $\mel{0}{T(z)}{0}=0$ since $T(z)=\sum_n L_n/z^{n+2}$.} it may not vanish after a local conformal transformation. One example is $z=e^{iw}$, mapping from plane to cylinder. From \eqref{eq:EM-trf-rule}, we have
\begin{equation}
    \ev{T'(w)}=-z^2 \ev{T(z)}+\frac{c}{24} = \frac{c}{24}.
\end{equation}
Then, the energy density on a cylinder is\footnote{Here the circumference is $2\pi$ since $\Re w\sim \Re w +2\pi$. If the circumference is taken to be $2\pi R$, we need a further rescaling $w\mapsto w'=R w$. Then, the total vacuum energy density becomes $-\dfrac{c}{12}\dfrac{1}{2\pi R^2}$ and the total energy is given by $-\dfrac{c}{12R}$.}
\begin{equation}
    \ev{T_{\tau\tau}}=\ev{T_{ww}}+\ev{\bar{T}_{\bar{w}\bar{w}}}= -\frac{1}{\pi}\ev{T(w)} 
    =-\frac{1}{2\pi}\frac{c}{12},
    \label{eq:CFT-vac-ene}
\end{equation}
where we used $\ev{T(w)}=-2\pi\ev{T_{ww}}=\ev{\bar{T}(\bar{w})}=-2\pi \ev{\bar{T}_{\bar{w}\bar{w}}}$ for the second equality. The energy density is negative as we expect from the Casimir energy\index{Casimir energy}.

\subsection{Replica trick in two-dimensional CFTs}\label{sec:replica-cft}
As we have discussed in Section \ref{sec:replica}, EE can be calculated from $\Tr\rho_A^n$. In CFT, it is useful to reduce the problem to the calculation of correlators, which are highly constrained by symmetry. In fact, $\Tr\rho_A^n$ can be mapped to a two-point function on a single complex plane after a local conformal transformation. The following review is based on~\cite{Calabrese:2004eu,Calabrese:2009qy,Casini:2009sr}.

Instead of considering an $n$-sheeted Riemann surface, let us consider replica fields\index{replica fields} on a single complex plane. We consider the case when the subregion is an interval, whose endpoint is $z=u$ and $z=v$. Then, $\Tr\rho_A^n$ is interpreted as a Euclidean partition function on a single complex plane with $n$ replica fields $\{\phi^{(k)}\}_{k=1,\cdots,n}$ with the boundary conditions
\begin{equation}
\begin{cases}
    \phi^{(k)}(e^{2\pi i}(z-u)) &= \phi^{(k+1)}(z-u)\\
    \phi^{(k)}(e^{-2\pi i}(z-v)) &= \phi^{(k+1)}(z-v)
\end{cases}
\quad (k=1,\cdots,n),
\end{equation}
where $\phi^{(n+1)}\equiv\phi^{(1)}$.
These boundary conditions can be imposed as a monodromy\index{monodromy} around operators inserted at each endpoint $z=u,v$. Such operators are called \textbf{twist operators}\index{twist operators}, $\sigma_n$ and $\tilde{\sigma}_n$. Using them, it follows that
\begin{equation}
    \Tr\rho_A^n = \ev{\sigma_n(u)\tilde{\sigma}_n(v)}.
\end{equation}
To compute this two-point function, we need to know the conformal dimension of twist operators. Recall that conformal weights can be read from the OPE with the energy-momentum tensor $T(z)$ \eqref{eq:TO-OPE}. There are $n$ replica fields and each has the corresponding $T(z)$. Since the theory is CFT${}^n/\mathbb{Z}_n$,\footnote{We need to divide by $\mathbb{Z}_n$ since we assume all the replicas are identical (the replica symmetry). Note that if we take $n\rightarrow 0$ instead of $n\rightarrow 1$ to compute free energy, it is known that the replica symmetry breaking (RSB) happens in the spin-glass theory~\cite{mezard1987spin,sherrington1998spin,castellani2005spinglass}. For interested readers in quantum gravity community, see~\cite{Engelhardt:2020qpv} for a recent discussion in the context of gravitational path integral.} we can just consider a single energy-momentum tensor and multiply by $n$.\footnote{This means the central charge of $n$-replicated theory is $cn$.} In the original picture, it is just a single insertion of the energy-momentum tensor in the $n$-sheeted Riemann surface $\Sigma_n$, we have
\begin{equation}
    \ev{T(z)}^{\text{(with replicas)}}_{\mathbb{C}}\equiv \frac{\ev{T(z) \sigma_n (u)\tilde{\sigma}_n (v)}_{\mathbb{C}}}{\ev{\sigma_n (u)\tilde{\sigma}_n (v)}_{\mathbb{C}}} = n \ev{T(z)}_{\Sigma_n},
    \label{eq:replica-EM-equation}
\end{equation}
where we denoted the complex coordinates on $\Sigma_n$ by $w$. To compute the right-hand side, we need to find a conformal map from $\Sigma_n$ to $\mathbb{C}$. The glued interval can be mapped to a semi-infinite line by moving one endpoint to the origin and the other endpoint to infinity:
\begin{equation}
    z\mapsto \frac{z-u}{z-v}.
\end{equation}
The resulting Euclidean path integral looks like the right of Fig.\ref{fig:euclid-pathint}. Then, we can just take the $n$-th root so that the periodicity around the origin is $2\pi$. To summarize, the conformal map
\begin{equation}
    z\mapsto \qty(\frac{z-u}{z-v})^{1/n}\equiv w
\end{equation}
maps $\Sigma_n$ to $\mathbb{C}$.
Applying this conformal map to \eqref{eq:EM-trf-rule}, we obtain
\begin{equation}
    \ev{T(z)}_{\Sigma_n} = \frac{c}{24}\qty(1-\frac{1}{n^2}) \qty(\frac{u-v}{(z-u)(z-v)})^2.
    \label{eq:replica-EM}
\end{equation}
For a more detailed calculation, see Appendix \ref{app:replica-calc-EM}.

Plugging \eqref{eq:replica-EM} into \eqref{eq:replica-EM-equation}, we obtain
\begin{equation}
    \ev{T(z) \sigma_n (u)\tilde{\sigma}_n (v)}_{\mathbb{C}} = \frac{c}{24}\qty(n-\frac{1}{n}) \qty(\frac{u-v}{(z-u)(z-v)})^2 \ev{\sigma_n (u)\tilde{\sigma}_n (v)}_{\mathbb{C}}.
\end{equation}
By taking $z\rightarrow u$ or $v$, we can compare the equation to the OPE \eqref{eq:TO-OPE} to find 
\begin{equation}
    \Delta_n= h_n + \bar{h}_n = \frac{c}{12}\qty(n-\frac{1}{n})
\end{equation}
for both $\sigma_n$ and $\tilde{\sigma}_n$.

Finally, EE\index{entanglement entropy} \eqref{eq:replica} is calculated as
\begin{align}
    S_{\mathrm{EE}} &= -\lim_{n\rightarrow 1} \pdv{n} \left[\Tr_A\rho_A^n\right] \\
    &= -\lim_{n\rightarrow 1} \pdv{n} \ev{\sigma_n(u)\tilde{\sigma}_n(v)}\\
    & = -\lim_{n\rightarrow 1} \pdv{n} (l/\epsilon)^{-\frac{c}{6}\qty(n-\frac{1}{n})} \\
    &= \frac{c}{3}\log\frac{l}{\epsilon},
    \label{eq:calabrese-EE}
\end{align}
where $\abs{u-v}=l$. This completes the derivation of EE of a single interval in CFT${}_2$, known as the \textbf{Calabrese-Cardy formula}\index{Calabrese-Cardy formula}.

Although we have calculated EE on an infinite line, we can compactify it to $S^1$ with a finite circumference $L$ by a conformal map $z\mapsto \tan (\pi z/ L)$. EE is given by~\cite{Calabrese:2004eu}
\begin{equation}
    S_{\mathrm{EE}}=\frac{c}{3}\log \qty[\frac{L}{\pi\epsilon}\sin\qty(\frac{\pi l}{L})].
\end{equation}
If one increases the interval size from $l=0$, EE grows from zero and becomes maximal at $l=L/2$, when the interval size is just a half of the total size. Then, EE decreases and eventually goes to zero at $l=L$, which is expected from the pure state nature. Such behavior is called a \textbf{Page curve}\index{Page curve}.

\subsection{First law of entanglement entropy}\label{sec:first-law}
While EE is an information-theoretic quantity, it can be related to a physical quantity, energy. A primitive example is thermodynamic entropy. It can be regarded as EE by considering a TFD state. We know the first law of thermodynamics relates entropy to energy: $\Delta S = \beta \Delta E$. 
This idea can be generalized to entanglement entropy as well. 
The Bisognano-Wichmann theorem\index{Bisognano-Wichmann theorem}~\cite{Bisognano:1976za} tells us the modular Hamiltonian for a half-space is the boost Hamiltonian \eqref{eq:modular-ham-half}
\begin{equation}
    K_A=2\pi \int_{x^0=0.x_\perp\ge 0} \dd[d-1]{x_\perp} x_1 T^{00}.
\end{equation}
This indicates entanglement entropy is thermodynamic entropy for an accelerated observer (\textbf{Unruh effect}\index{Unruh effect}). Can we go beyond this specific subregion choice? The answer is yes for CFTs. In CFTs (not limited to two dimensions), one can conformally map a half space to a spherical region.
This is known as the \textbf{Casini-Huerta-Myers (CHM) map}\index{Casini-Huerta-Myers map}~\cite{Casini:2011kv}. 

While these arguments can relate EE itself to the canonical energy, we need to consider their difference from a reference state like the vacuum to discuss the \textbf{first law of EE}\index{first law of entanglement entropy}. Originally, the first law of EE was proposed in~\cite{Bhattacharya:2012mi} based on the holographic calculation. Their conjecture was 
\begin{equation}
    \Delta S_A = \frac{2\pi}{d+1} l \Delta E_A
    \label{eq:first-law-hol}
\end{equation}
for a sufficiently small subregion $A$. $\Delta$ means a difference between the vacuum and a slightly excited state. $S_A$ is EE for a spatial subregion $A$ on a $(d-1)$-dimensional time slice and $E_A$ is the canonical energy defined as
\begin{equation}
    \Delta E = \int_A\dd[d-1]{x} \ev{T_{00}}.
\end{equation}
While this seems to be true, it was not clear whether we can discuss the first law without taking a small subregion limit. 

Later, a more general version of the first law was discussed by~\cite{Blanco:2013joa} based on quantum information. More specifically, their conclusion stems from the non-negativity of \textbf{relative entropy}\index{relative entropy}, \begin{equation}
    S(\rho+\delta\rho||\rho) = -\tr \qty[(\rho+\delta\rho)[\log\rho - \log(\rho+\delta\rho)]] ,
\end{equation}
and the first law was generalized for generic states $\rho$ and its perturbation $\delta\rho$ to first order.

We will briefly review the information-theoretic first law by~\cite{Blanco:2013joa}, then see how it is consistent with \eqref{eq:first-law-hol} proposed by~\cite{Bhattacharya:2012mi} by using the modular Hamiltonian of a spherical subregion~\cite{Casini:2011kv}. In the following, we omit the subregion subscript $A$.

The \textbf{modular Hamiltonian}\index{modular Hamiltonian} is defined as $K=-\log\rho$. We denote its expectation value under the variation of the state by
\begin{equation}
    \Delta\ev{K} \equiv \tr\qty[(\delta\rho)K].\label{eq:MH-var}
\end{equation}
EE is defined as $S=\tr(\rho K)$, thus its variation is given by
\begin{align}
    \Delta S &\qty(= -\tr\qty[(\rho+\delta\rho)\log(\rho+\delta\rho)-\rho\log\rho])\nonumber\\
    & =\tr \qty[(\delta\rho) K] + \tr (\rho\Delta K) \nonumber \\
    & = \Delta\ev{K} +\ev{\Delta K}.\label{eq:EE-var}
\end{align}
The relative entropy is given by
\begin{align}
    S(\rho+\delta\rho||\rho) &= \tr\qty[(\rho+\delta\rho)\qty(\log(\rho+\delta\rho)-\log\rho)]\\
    &= \tr (\rho\Delta K) + \tr \qty[(\delta\rho) \Delta K].
\end{align}
From \eqref{eq:MH-var} and \eqref{eq:EE-var},
\begin{align}
    S(\rho+\delta\rho||\rho) &= \Delta S - \Delta\ev{K} + \tr \qty[(\delta\rho) \Delta K]\\
    &= \Delta S - \Delta\ev{K} + O\qty((\delta\rho)^2).
\end{align}
Since the relative entropy is non-negative (see~\cite{nielsen_chuang_2010} for its proof using Klein's inequality\index{Klein's inequality}), it must be positive for both positive and negative $\delta\rho$ to the first order. This means the first-order relative entropy $\Delta S - \Delta\ev{K}$ must vanish. This is the \textbf{first law of EE}\index{first law of entanglement entropy}. 

In general, the modular Hamiltonian is very complicated, however, using a special conformal transformation as well as translation, one can map a half space to a spherical subregion. \eqref{eq:modular-ham-half} is mapped to~\cite{Casini:2011kv}
\begin{equation}
    K=2\pi \int_A\dd[d-1]{x} \frac{l^2-r^2}{2l} T_{00},
\end{equation}
where $l$ is the radius of the spherical subregion $A$. Now we would like to evaluate this with respect to an excited state and the vacuum. Since the vacuum energy on a plane is zero, $\ev{K}=\Delta\ev{K}$. From the first law of EE, we have
\begin{equation}
    \Delta S = 2\pi \int_A \dd[d-1]{x} \frac{l^2-r^2}{2l} \ev{T_{00}}.
\end{equation}
When $\ev{T_{tt}}$ is approximately constant, we can evaluate the above integral. This happens when the subregion size is sufficiently small. 
Under this approximation,
\begin{equation}
    \Delta S \approx 2\pi \ev{T_{00}} \int_A \dd[d-1]{x} \frac{l^2-r^2}{2l} = 2\pi \frac{\Delta E}{|A|} \frac{l|A|}{d+1} =\frac{2\pi}{d+1} l \Delta E,
    \label{eq:first-law-rel}
\end{equation}
where we introduced the canonical energy $\Delta E \equiv \int_A\dd[d-1]{x} \ev{T_{00}}\approx |A| \ev{T_{00}}$. $|A|$ is the volume of the spherical subregion, i.e. $|A|=l^{d-1}\Omega_{d-2}/(d-1)$.
\eqref{eq:first-law-rel} is nothing but \eqref{eq:first-law-hol}.
The first law of EE is often useful for a consistency check for holographic duality as EE and energy-momentum tensor are calculated separately (both in the bulk and boundary).

\section{Holographic CFT${}_2$}\label{sec:hol-CFT2}

Holographic CFT is characterized by a large degree of freedom and strong coupling. When its spacetime dimension is two, it is more useful to characterize it by the large central charge $c\gg 1$ and a sparse spectrum~\cite{Hartman:2014oaa}. This is a necessary and sufficient condition for two-dimensional holographic CFTs. This enables us to discuss quite generic features of holography from CFT without assuming the existence of gauge symmetry.% Such a simple characterization is possible because of the strong symmetry enhancement in two dimensions. 

\subsection{Defining holographic CFT$_2$}\label{sec:def-holCFT}
\textbf{Holographic CFT}\index{holographic CFT} is said to be \textbf{large-$\bm{c}$}\index{large-$c$} and \textbf{sparse}\index{sparse spectrum}. In this subsection, we explain what these jargons mean and why they appear as conditions for holographic theories.
\subsubsection{Large $c$}
The intuitive explanation why the large-$c$ limit corresponds to the large-$N$ limit is that both count the number of degrees of freedom of the theory. This becomes clear by considering how the four-dimensional central charge $a$\index{a} analogous to the two-dimensional central charge $c$ is related to $N$. Recall that $c$ appears as a coefficient in the Weyl anomaly\index{Weyl anomaly} \eqref{eq:weyl-anomaly-2d} in CFT${}_2$. To find $a$, it is reasonable to similarly consider the Weyl anomaly.
In four dimensions, it is characterized by two parameters $a$ and $c$ called the Weyl anomaly coefficients:
\begin{equation}
    \delta S = \int \dd[4]{x} \sqrt{-g} \sigma (a \mathrm{Euler} - c \mathrm{Weyl}^2) \quad (\delta g_{\mu\nu}=2\sigma g_{\mu\nu}),
\end{equation}
where $\mathrm{Euler}=\mathcal{R}_{\mu\nu\rho\sigma}^2-4\mathcal{R}_{\mu\nu}^2+\mathcal{R}^2$ is the Euler density\index{Euler density}, whose volume integral equals the Euler number\index{Euler number}, and $\mathrm{Weyl}^2 = \mathcal{R}_{\mu\nu\rho\sigma}^2-2\mathcal{R}_{\mu\nu}^2+\mathcal{R}^2/3$ is the squared Weyl tensor\index{Weyl tensor}, which multiplied by $\sqrt{-g}$ is conformal invariant. The \textbf{$a$-theorem}\index{$a$-theorem}~\cite{Cardy:1988cwa,Komargodski:2011vj} tells us that $a$ is a monotonically decreasing function along the RG flow like $c$ in two dimensions. Furthermore, the coefficient of the universal, logarithmic term of EE is proportional to $a$ as the two-dimensional EE is proportional to $c$. Thus, it is reasonable to identify $a$ as the number of degrees of freedom of a four-dimensional CFT. Since the $a=(N^2-1)/4$ for the $\mathcal{N}=4$ $SU(N)$ SYM, the large-$N$ limit is equivalent to large-$a$. Thus, the corresponding limit in CFT$_2$ should be the large-$c$ limit.

Similar to the case of large $N$, the large-$c$ limit validates the semiclassical approximation. For the AdS${}_3$/CFT${}_2$ correspondence, it can be shown that the asymptotic symmetry\footnote{Since the bulk diffeomorphism\index{diffeomorphism} is not a gauge symmetry at the asymptotic boundary, it constitutes the global asymptotic symmetry. Thus, diffeomorphism at the asymptotic boundary maps a physical state in the dual CFT to another inequivalent one.} generator of AdS${}_3$ constitutes two copies of Virasoro algebra (corresponding to holomorphic and anti-holomorphic sectors)~\cite{Brown:1986nw}. See Section 3.3 of~\cite{Banerjee:2018} for a review. From the equivalence, we can read off a relation between (both chiral and anti-chiral) Virasoro central charge\index{central charge} $c$\index{$c$} and three-dimensional Newton's constant~$G_N$:
\begin{equation}
    c(=\bar{c})=\frac{3R}{2G_N}.
    \label{eq:brown-henneaux}
\end{equation}
This is the celebrated \textbf{Brown-Henneaux central charge}\index{Brown-Henneaux central charge}\index{Brown-Henneaux relation}. From this, it is apparent the large-$c$ limit is equivalent to the small-$G_N$ limit ($G_N\ll R$).

\subsubsection{Sparse spectrum}
The density of low-lying states with $\Delta<\frac{c}{12}\sim O(c)$ is given by~\cite{Hartman:2014oaa}\footnote{Note that $E=\Delta-c/12$ in~\cite{Hartman:2014oaa}.}
\begin{equation}
    \rho\qty(\Delta<\frac{c}{12}) \lesssim e^{2\pi\Delta}<O(c)
\end{equation}
while $\rho(\Delta\ge c/12)\sim O(e^c)$ (black hole microstates). We call such a spectrum sparse\index{sparse spectrum}. For a sparse spectrum, we generally expect a strongly-coupled theory otherwise there are too many light operators.\footnote{For example, the free scalar theory has infinitely many conserved integer spin currents~\cite{Maldacena:2011jn,Kirilin:2018qpy}.} 
Sparseness comes from the fact the thermodynamic entropy in the microcanonical ensemble has a phase transition from $O(1)$ to $O(c)$ at $\Delta=c/12$. This is exactly what we mentioned in Section \ref{sec:AdS-CFT-general-dim} regarding the gap between generalized free fields and others. We will revisit the calculation of the thermodynamic entropy from an interplay between the microcanonical and canonical ensemble in Section \ref{sec:local-op-quench}.

\subsection{Consequence of large $c$ and sparseness}\label{sec:detail-holCFT}
\subsubsection{Operators}
As we have discussed in Section \ref{sec:AdS-CFT-general-dim}, we can classify operators of holographic CFTs based on their order in $c$.
Operators with $\Delta=O(c^0)\equiv O(1)$ are \textbf{generalized free}\index{generalized free}, whose correlation functions are factorized into a product of two-point functions. We will later see that the factorization property can be explicitly confirmed from the large-$c$ conformal block. The other class of operators has a conformal dimension that scales with the central charge $c$, i.e. $\Delta/c$ is held fixed. We further classify them into two subclasses: (i) \textbf{light operators}\index{light operators} such that $1 \ll \Delta \ll c$ and (ii) \textbf{heavy operators}\index{heavy operators} such that $1 \ll \Delta \sim O(c)$.\footnote{Sometimes we refer to both generalized free and light operators as simply `light'.} An example of light operators is a twist operator $\sigma_n$ around $n=1$ since its conformal dimension vanishes at $n=1$ although it is proportional to $c$. %We will discuss this in detail in Section \ref{sec:HEE-2dCFT}.

\subsubsection{Large-$c$ conformal block}
In the large-$c$, sparse CFTs, we can say a bit more about the conformal blocks \eqref{eq:conf-block-gen} appearing in the expansion of a four-point function. There are mainly two approaches~\cite{Fitzpatrick:2014vua}: (i) direct summing over the intermediate states using the so-called \textbf{graviton basis}\index{graviton basis} and (ii) the \textbf{monodromy method}\index{monodromy method}. %We are interested in (ii) for our purpose later.  
We briefly review them in the following. A more detailed summary is provided in Section 4.2 of~\cite{KusY:2021} and Section 1 of~\cite{Piatek:2021aiz}.

The first method (i) can deal with generalized free fields or %not-so-heavy 
light fields, i.e. %$h_i/c\rightarrow 0$ 
$h_i/c\ll 1$ while $c\rightarrow \infty$. One can focus on either the \textbf{identity (vacuum) conformal block}\index{identity conformal block}\index{vacuum conformal block}~\cite{Fitzpatrick:2014vua}, where the intermediate operator $O_p$ is identity or its descendants, or the \textbf{global conformal block}\index{global conformal block} (also known as semiclassical conformal block\index{semiclassical conformal block})~\cite{Dolan:2000ut,Dolan:2003hv}, where we only consider $SL(2,\mathbb{C})$ (global conformal transformation) descendants, corresponding to neglecting backreaction and taking only AdS isometry into account. They are expanded in the graviton ``basis'' and the identity Virasoro block\index{identity Virasoro block} has been calculated in this regime~\cite{Fitzpatrick:2014vua,Caputa:2014vaa}.\footnote{The reason we quoted ``basis'' is that these graviton states are orthogonal to each other only when $c\rightarrow \infty$.} 
%Thus, the conformal block is dominated by an identity intermediate state in the large-$c$ limit. It is called the vacuum or identity conformal block and takes a following universal form (See the first term of (11) in~\cite{Zamolodchikov:1984eqp}\footnote{In the paper, the central charge is denoted by $C$.}):
%\begin{equation}
%    F_O(b|z)\simeq z^{\Delta_b-2\Delta_O} {}_2F_1 (\Delta_b,\Delta_b,2\Delta_b,z),
%\end{equation}
%where ${}_2F_1(a,b,c,z)$ is the hypergeometric function.
%For its application to entanglement entropy calculated from the $n$-sheeted Riemann surface, see~\cite{Caputa:2014vaa}.

The second method (ii) is what we focus on later~\cite{Asplund:2014coa,Hartman:2013mia}. It deals with the large-$c$ limit with $h_i/c$ and $h_p/c$ fixed. Although it has not been proven rigorously, it is believed that in this limit the Virasoro conformal block exponentiates~\cite{1987TMP....73.1088Z,Belavin:1984vu}:
\begin{equation}
    \mathcal{F}(c,h_p,h_i,x)\approx \exp\qty[-\frac{c}{6} f\qty(\frac{h_p}{c}, \frac{h_i}{c}, x)].
\end{equation}
One piece of evidence for this exponentiation is that this happens in the semiclassical limit of the Liouville theory. 

Assuming the exponentiation, we need to sum over these exponentials in the four-point function. At least when $z\rightarrow 1$, the largest contribution is dominated by the identity block, i.e. $O_p=\bm{1}\Rightarrow a_p=1$. %(and its descendants). 
Since the identity operator corresponds to the vacuum state via the operator/state correspondence, this observation is called the vacuum block dominance.
The problem is reduced to solving for $f_0(h_i/c,x)\equiv f(0,h_i/c,x)$.
When we discuss EE, the conformal dimension of the twist operator becomes light. Then, it enables us to perform perturbation theory in $\Delta_n/c$ $(n\rightarrow 1)$ but nonperturbatively for the other external operator appearing in the conformal block. This is known as the \textbf{heavy-heavy-light-light (HHLL) identity block}\index{heavy-heavy-light-light block}\index{HHLL block}~\cite{Fitzpatrick:2014vua,Asplund:2014coa}.

From the conformal Ward identity \eqref{eq:CWI} for a so-called level-2 null field $\psi(z)$, we obtain a second-order ordinary differential equation for $\psi(z)$. By analyzing the monodromy property around its singularity, i.e. the choice of OPE channels, we can obtain the functional form of the conformal block. See Section 5.2 of~\cite{Banerjee:2018} for a brief review and see appendices of~\cite{Fitzpatrick:2014vua} for comprehensive reviews on conformal blocks. 

After all, the holomorphic part of the HHLL correlator ($h_1=h_4=h_O$ and $h_2=h_3=h_n \ll c$) becomes~\cite{Asplund:2014coa}
\begin{align}
    G_n(z) &\approx \mathcal{F}\qty(c,0,\frac{h_O}{c},\frac{h_n}{c},z)\\
    &\approx \exp\qty[-\frac{c}{6}f_0\qty(\frac{h_O}{c},\frac{h_n}{c},1-z)]\\
    &= \qty[\frac{z^{1-\alpha_O}(1-z^{\alpha_O})^2}{\alpha_O^2}]^{-h_n+O\qty((n-1)^2)},
\end{align}
where
\begin{equation}
    \alpha_O = \sqrt{1-\frac{24 h_O}{c}}.
    \label{eq:alpha-O-monod}
\end{equation}
The first approximation comes from picking identity only and the second approximation is about the exponentiation, which should be accompanied by an $O(c^0)$ correction in the exponent.
Instead of explaining all the details of the monodromy method deriving this result, we will give another explanation for why  \eqref{eq:alpha-O-monod} appears from a different perspective in the next section (Section \ref{sec:local-op-quench}). Surprisingly, this is related to switching between microcanonical and canonical ensembles.

%\subsubsection{Holographic stress tensor}
%Variation with respect to metric

%read out from metric perturbation

%In the presence of local operator quench,

\subsection{Holographic entanglement entropy -- demonstration}\label{sec:HEE-2dCFT}
In Section \ref{sec:HEE-CFT}, we discussed the holographic entanglement entropy\index{holographic entanglement entropy} proposal, in particular, the RT formula\index{RT formula |see Ryu-Takayanagi formula } \eqref{eq:HEE-RT}.
As a simple example, let us demonstrate how the minimal surface in AdS${}_3$ leads to EE in CFT${}_2$ \eqref{eq:calabrese-EE}. The codimension-two minimal surface in AdS${}_3$ is nothing but geodesics. Suppose the subregion is an interval from $x=-l/2$ to $x=l/2$. Since the AdS geodesics is given by $z^2+x^2=(l/2)^2$ in Poincar\'e coordinates, its length $\gamma_A$ is
\begin{align}
    \gamma_A &= 2\int_\epsilon^{l/2}\dd{z} \frac{ds}{dz} \\
    &= 2R\int_\epsilon^{l/2}\dd{z} \frac{\sqrt{1+x'(z)^2}}{z}\\
    &= 2R\int_\epsilon^{l/2}\dd{z} \frac{1}{z}\sqrt{1+\frac{z^2}{(l/2)^2-z^2}}\\
    &=2R\frac{l}{2}\int_\epsilon^{l/2}\dd{z} \frac{1}{z\sqrt{(l/2)^2-z^2}}\\
    &= 2R\log \frac{l}{\epsilon} +O(1).
\end{align}
Then, from the Brown-Henneaux relation\index{Brown-Henneaux relation} \eqref{eq:brown-henneaux}, holographic EE in the leading order in $\epsilon$ is
\begin{equation}
    \frac{\gamma_A}{4G_N}=\frac{c}{3}\log\frac{l}{\epsilon}.
\end{equation}
This is indeed vacuum EE of an interval $[-l/2,l/2]$ in CFT${}_2$ \eqref{eq:calabrese-EE}.

In other coordinates, the geodesic distance can be a bit complicated. In such a case, it is sometimes clearer to use the embedding space formalism \index{embedding space formalism}. See Appendix \ref{sec:Embedd}.

\section{Holographic local operator quench}\label{sec:local-op-quench}
\subsection{Gravity dual of heavy primary operator}
In this section, we consider the gravity dual of a heavy operator insertion in CFT. Although we consider only the AdS${}_3$/CFT${}_2$ correspondence here for simplicity, higher dimensional discussions are also available~\cite{Horowitz:1999gf}.

First, let us consider a gravity dual of the CFT primary state $\ket{O_\Delta}$. Identifying the radial quantizing direction as a global time in global coordinates, the state is time-translation invariant. Thus, it is an energy eigenstate in global coordinates. Since we want to consider a nontrivial gravitating spacetime, we assume the conformal dimension is large enough $\Delta\sim O(c)$. Then, the dual bulk geometry must be backreacted from pure AdS. Furthermore, since the state is rotationally invariant, the bulk geometry must be something in the center of AdS. We already know such bulk geometries. In Section \ref{sec:asympt-AdS}, we discussed the conical deficit (massive particle) geometry\index{conical deficit geometry} \eqref{GBmet} and the Schwarzschild-AdS black hole\index{Schwarzschild-AdS black hole} \eqref{eq:AdS-Sch-BH}. Both have a massive object in the center of AdS and are rotationally invariant.

How can we justify they are really dual to heavy primary states? The key observation is the geodesic approximation\index{geodesic approximation} (Section \ref{sec:AdS-CFT-general-dim}). Since we are considering a heavy state, the geodesic approximation must be valid. The ket/bra primary state is located at the past/future infinity. Thus, the geodesic connecting these infinities is just a straight line going up in the global coordinates! The straight line should be placed in the center to preserve the rotational symmetry although it can be moved by AdS isometry.\\
\emph{Caveat: One might worry whether a nice gravity dual exists for a single energy eigenstate $\ket{E=\Delta-c/{12R}}$. What we should really do is averaging over microstates with the same energy. However, followed by the eigenstate thermalization hypothesis (ETH)\index{eigenstate thermalization hypothesis}\index{ETH |see eigenstate thermalization hypothesis } explained later, the deviation from the mean value is exponentially suppressed.}

The conical deficit geometry and Schwarzschild-AdS black hole have a transition point $M_{crit}=R^2$. Beyond this point, a horizon is formed. In terms of the particle mass, it equals $m_{crit}=1/(8G_N)$. From the holographic dictionary \eqref{eq:delta-mass}, mass is related to the conformal dimension\index{conformal dimension} of the dual operator as $\Delta=mR$.\footnote{Since $\Delta\sim O(c)$, we dropped a subleading $O(1)$ constant.} Using this relation, the critical conformal dimension is $\Delta_{crit}=R/(8G_N)=c/12$. This is known as the \textbf{black hole threshold}\index{black hole threshold}. The black hole temperature \eqref{eq:BH-temp-3d} is related to the mass and conformal dimension by
\begin{equation}
    \beta=\frac{2\pi R^2}{\sqrt{M-R^2}}=\frac{2\pi R}{\sqrt{8G_N m-1}}= \frac{2\pi R}{\sqrt{\frac{12\Delta}{c}-1}}.
    \label{eq:BH-temp-conf-dim}
\end{equation}

How can we understand the transition from CFT? One way is just comparing various quantities from AdS and CFT such as the energy-momentum tensor and EE~\cite{Nozaki:2013wia}. This approach will be one of the main calculations in this chapter.
%our work on BCFTs as well~\cite{Kawamoto:2022etl} (Section \ref{sec:quench-BCFT}-). 
In this section, we try to clarify the physical meaning of the transition and the relation between $\Delta$ and $\beta$ \eqref{eq:BH-temp-conf-dim} without being mystified by a cumbersome calculation.

\subsection{Eigenstate thermalization hypothesis and Cardy formula}\label{sec:ETH-Cardy}
It is mysterious at first glance that both a finite-temperature CFT and an energy eigenstate can have the same gravity dual -- black hole. This suggests that in the large-$c$ (small-$G_N$) saddle point approximation, these CFT states are indistinguishable. The \textbf{eigenstate thermalization hypothesis (ETH)}\index{eigenstate thermalization hypothesis}\index{ETH |see eigenstate thermalization hypothesis } helps us to discuss this on firm ground. It has been argued that ETH holds for large-$c$, sparse CFTs in the thermodynamic limit~\cite{Basu:2017kzo}\footnote{Here we only consider \emph{global} ETH while recent studies focus on the subsystem ETH (ETH within a subregion)~\cite{Lashkari:2016vgj,Lashkari:2017hwq}.} (see~\cite{Karlsson:2021duj} for the large-$N$ CFTs and~\cite{Faulkner:2017hll} for analysis in the finite-$c$ regime). 

ETH\index{eigenstate thermalization hypothesis}\index{ETH |see eigenstate thermalization hypothesis } is a conjecture in which (almost)\footnote{Strong ETH refers to all energy eigenstates while weak ETH refers to almost all energy eigenstates.} all energy eigenstates look identical to the corresponding (microcanonical) thermal states when they are probed by few-body operators~\cite{PhysRevA.43.2046,Srednicki:1995pt}. (See~\cite{DAlessio:2015qtq,Srednicki_1999} for reviews.)
More explicitly, it postulates that a matrix element of a local operator $\mathcal{O}$ in the energy eigenbasis is almost diagonal up to a nonperturbative correction, i.e.
\begin{equation}
    \mel{E_i}{\mathcal{O}}{E_j}=\mathcal{O}(E_i)\delta_{ij} + e^{-S\qty(\frac{E_i+E_j}{2})/2}\, f(E_i,E_j)\, R_{ij},
    \label{eq:ETH}
\end{equation}
where $\mathcal{O}(E)\equiv \mel{E}{\mathcal{O}}{E}$ and $f(E,E')$ are some real smooth functions and $R_{ij}$ is a complicated, complex-valued $O(1)$ function. Since the second term is $O(e^{-S/2})$, it is exponentially suppressed in the semiclassical limit $S\sim O(1/G_N)\sim O(c)$. $\mathcal{O}(E)$ is related to the equilibrium value in a canonical ensemble $\ev{\mathcal{O}}_\beta$ by
\begin{equation}
    \ev{\mathcal{O}}_\beta \equiv \frac{\tr \qty(e^{-\beta H} \mathcal{O})}{\tr e^{-\beta H}}= \frac{\int\dd{E} \rho(E) e^{-\beta E} \mathcal{O}(E)}{\int\dd{E} \rho(E) e^{-\beta E}} + O\qty(e^{-S/2}),
    \label{eq:ETH2}
\end{equation}
where we used \eqref{eq:ETH} in the last equality. $\rho(E)$ is the density of states within a small energy window around $E$ and related to the microcanonical (Boltzmann) entropy as\footnote{Precisely speaking, the delta function must be smeared to avoid divergence.}
\begin{equation}
    \rho(E)\equiv e^{S(E)}=\sum_{i} \delta(E-E_i).
\end{equation}
After all, ETH postulates any local operators in the energy eigenbasis are almost diagonal and their elements are given by the microcanonical average.\footnote{From this, it immediately follows that its long-time average thermalizes (although this does not explain all thermalization phenomena), i.e. the long-time averaged one-point function obeys the microcanonical ensemble and the variance after averaging is suppressed exponentially.} In other words, combining \eqref{eq:ETH} and \eqref{eq:ETH2}, 
\begin{equation}
    \mel{E}{\mathcal{O}}{E}=\ev{\mathcal{O}}_{\beta(E)} +O(e^{-S/2}),
\end{equation}
where $\beta(E)$ is determined from $\beta$ such that $\ev{H}_\beta=E$.

Now we would like to find explicit formulae for the microcanonical thermodynamic entropy in CFT${}_2$ (also known as the Cardy formula) and the relation between $\beta$ and $E$ (canonical and microcanonical ensemble).
As we have discussed, the partition function is calculated as
\begin{equation}
    Z(\beta)=\sum_i e^{-\beta E_i}= \sum_i \int \dd{E} \delta(E-E_i) e^{-\beta E} = \int \dd{E} \rho(E) e^{-\beta E}= \int \dd{E} e^{S(E)} e^{-\beta E}.
\end{equation}
In the equilibrium case, this partition function is calculated in the saddle point approximation: Suppose the saddle point is $E=E_\ast(\beta)$, the partition function is approximated as\footnote{We slightly changed the definition for the free energy as $F=-\frac{1}{\beta}\log Z$ as usual in thermodynamics from $F=-\log Z$ in Chapter \ref{ch:1}.}
\begin{equation}
    e^{-\beta F(E_\ast)}=Z(\beta)=e^{S(E_\ast)-\beta E_\ast}.
\end{equation}
By taking logarithms on both sides, we obtain the familiar Legendre transform in thermodynamics bridging microcanonical and canonical ensembles. The saddle point equation defining $E_\ast(\beta)$ is given by
\begin{equation}
    S'(E_\ast)=\beta.
    \label{eq:saddle-ETH}
\end{equation}
(Note that in the original setup, the parameter we give first is the energy $E_\ast$ of the energy eigenstate, not the inverse temperature $\beta$. In this case, we should regard \eqref{eq:saddle-ETH} as a defining equation for $\beta=\beta(E_\ast)$.)

To proceed with the saddle-point analysis, we need a functional form of $S(E)$. For a moment, we focus on the low and high temperature (energy) limits with $c$ fixed (\textbf{Cardy limit}\index{Cardy limit}). Since we are interested in the $O(c)$ transition in entropy, it must be sufficient to compare those two extreme limits.\footnote{This is actually different from the usual large-$c$ limit, where $\Delta\sim O(c)$ is fixed. Nevertheless, the same logic applies to the usual limit and in the holographic case, we have a sharp transition at $\Delta=E+c/12=c/12$. See Section 4.2 of~\cite{Hartman:2014oaa} for more details.}
In a two-dimensional CFT, we can easily calculate the partition function in the high-temperature limit from the low-temperature limit using the modular invariance\index{modular invariance}.\footnote{Although it has been much less explored, modular invariance indeed exists in higher dimensions as well. \cite{Belin:2016yll} argues that the modular invariance put a stringent constraint on a possible holographic CFT.} The following analysis is based on Section 25.3 of~\cite{Hartman:2015}.

In the low-temperature limit $\beta\rightarrow \infty$, the partition function is dominated by the vacuum contribution:
\begin{equation}
    Z(\beta)\approx e^{-\beta E_{vac}},
\end{equation}
where $E_{vac}=-\frac{c}{12 R}$ \eqref{eq:CFT-vac-ene}.
Since CFT${}_2$ has modular invariance, the partition functions in the low-temperature limit and high-temperature limit are related to each other by $\beta\leftrightarrow 4\pi^2 R^2/\beta$. The partition function in the high-temperature limit ($\beta\rightarrow 0$) is dominated by
\begin{equation}
     Z(4\pi^2 R^2/\beta)\approx \exp\qty[- \frac{4\pi^2 R^2}{\beta}E_{vac}]=\exp\qty[\frac{\pi^2 c R}{3\beta}].
\end{equation}
From these, the thermodynamic entropy\index{thermodynamic entropy}\footnote{Surprisingly, the thermodynamic entropy in the Cardy limit equals the Beckenstein-Hawking entropy\index{Beckenstein-Hawking entropy} for the BTZ black hole~\cite{Strominger:1997eq}.} and energy are calculated as follows:
\begin{equation}
    (S(\beta),E(\beta))=
    \begin{cases}
    \qty(O(1), -\dfrac{c}{12 R}) \quad & (\beta\rightarrow\infty) \\
    \qty(\dfrac{2\pi^2 c R}{3\beta}, \dfrac{\pi^2 c R}{3\beta^2}) \quad & (\beta\rightarrow 0)
    \end{cases}
    ,
    \label{eq:cardy-ene}
\end{equation}
where
\begin{equation}
    S=(1-\beta\partial_\beta)\log Z;\quad E=-\partial_\beta \log Z.
\end{equation}
By eliminating $E$ from \eqref{eq:cardy-ene}, we obtain
\begin{equation}
    S(E)=
    \begin{cases}
    O(1),\quad &\beta\rightarrow \infty\\
    2\pi \sqrt{\dfrac{c}{3}RE} \quad &\beta\rightarrow 0
    \end{cases}
    .
\end{equation}
This is equivalent to solving the saddle-point equation \eqref{eq:saddle-ETH}.
In the Cardy limit, where $E\rightarrow\infty$ with $c$ fixed, we can clearly see the phase transition of entropy as expected. $S(E)\sim 2\pi \sqrt{\frac{c}{3}RE}$ in the high energy limit is known as the \textbf{Cardy formula}\index{Cardy formula}.

Our goal was to express the canonical inverse temperature $\beta$ in terms of energy $E$ in a microcanonical ensemble. By inverting \eqref{eq:cardy-ene}, we have
\begin{equation}
    \beta=\pi\sqrt{\frac{cR}{3E}}.
\end{equation}
In CFT${}_2$, $E_R=L_0-c/24=h-c/24,\ E_L=\bar{L}_0-\bar{c}/24$, and $E=E_L+E_R$. As the operator (state) we are considering is non-chiral, $E_L=E_R$. On the energy eigenstate $\ket{O_\Delta}$, we have $E=\Delta-c/12$. Finally, we obtain
\begin{equation}
    \beta=\pi\sqrt{\frac{cR}{3\qty(\Delta-\frac{c}{12})}} = \frac{2\pi R}{\sqrt{\frac{12\Delta}{c}-1}}.
\end{equation}
This is exactly equal to \eqref{eq:BH-temp-conf-dim}. 

To summarize, from ETH, we could explicitly show the CFT heavy primary state is the energy eigenstate dual to the microcanonical black hole\index{microcanonical black hole} and its energy is related to the canonical temperature through the Legendre transformation. When $\Delta\le c/12$, the canonical ensemble is no longer applicable as the temperature becomes imaginary. (This is what we have already discussed in the conical deficit geometry!) The relation of the temperature and conformal dimension can be purely derived from the conformal block\index{conformal block} as well~\cite{Asplund:2014coa}.\footnote{\textit{Caveat:} The Cardy limit used here and the large-$c$ limit are actually slightly different. The high energy regime in the Cardy limit is much higher than the black hole regime. Nevertheless, it can be shown the lower bound for its validity is extended and this comparison is justified~\cite{Hartman:2014oaa}. This led to the microscopic derivation of the black hole entropy by Strominger~\cite{Strominger:1997eq}.} The underlying principle for this miraculous matching is the ETH.%\footnote{ETH indicates that there might be Note that there is another realization of global quench by regulating a boundary state via an Euclidean time evolution $e^{-\tau H}\ket{B}$ (We will define the boundary state in Section \ref{sec:boundary-state}). While the single-sided black hole dual to a heavy primary contains the entire black hole interior, the gravity dual of the regulated boundary state terminates on a brane inside the horizon. This indicates the boundary state masks the details of the black hole microstate as the brane.}

\subsection{Gravity dual of local operator quench}
Writing the two-dimensional coordinates as $(t,x)$, we define a local operator quench\index{local operator quench} by a state 
\begin{equation}
    e^{-\alpha H}O(x=0, t_E=0)\ket{0} \quad (t_E=it),
    \label{eq:local-quench-state}
\end{equation}
where the conformal dimension of the operator is $O(c)$ and $\alpha$ is a regularization parameter~\cite{Nozaki:2014hna,He:2014mwa,Caputa:2014vaa,Asplund:2014coa}. This is a locally excited state smeared over $O(\alpha)$ width. Its gravity dual was discussed by~\cite{Nozaki:2013wia}. The key is again the geodesic approximation\index{geodesic approximation}. The unnormalized partition function $\ev{Oe^{-2\alpha H}O}$ is a two-point correlator separated by $2\alpha$. In the Euclidean Poincar\'e AdS, its gravity dual is a massive particle propagating along $z^2+t_E^2=\alpha^2$. By a Wick rotation, it gives the trajectory 
\begin{equation}
    z^2-t^2=\ap^2,\ \ \ x=0,
    \label{eq:traj-massive-part-local}
\end{equation}
of a massive particle dual to the time-evolved state $e^{-iHt}e^{-\alpha H}O(0)\ket{0}$. The mass of the particle is again related to $O$'s conformal dimension by $\Delta\simeq mR$. Recalling the gravity dual of a heavy primary state, the dual geometry of a local operator quench must be an infalling massive particle or an infalling black hole along \eqref{eq:traj-massive-part-local}.

In general, the backreaction of a bulk infalling particle in the Poincar\'e AdS can be obtained by pulling back a spacetime with a massive static particle in global AdS \cite{Horowitz:1999gf}.
Consider the AdS$_3$ in the Poincar\'e metric \eqref{eq:poincare-coords}, 
\ba
ds^2=R^2\left(\frac{dz^2-dt^2+dx^2}{z^2}\right).   \label{Pmet}
\ea
This is transformed into the global AdS$_3$  with the metric \eqref{eq:global-static-coords}
\ba
ds^2=-(r^2+R^2)d\tau^2+\frac{R^2}{r^2+R^2}dr^2+r^2d\theta^2,  \label{Gmet}
\ea
via the following map:\footnote{Here we identified the angular coordinates as $\Omega_1=\sin\theta$ and $\Omega_2=\cos\theta$ in \eqref{eq:poincare-coords}.}
\ba
\s{R^2+r^2}\cos\tau &=& \frac{R\ap^2+R(z^2+x^2-t^2)}{2\ap z},\no
\s{R^2+r^2}\sin\tau &=& \frac{Rt}{z},\no
r\sin\theta &=& \frac{Rx}{z},  \no
-r\cos\theta &=& \frac{-R\ap^2+R(z^2+x^2-t^2)}{2\ap z}.  \label{mappo}
\ea
We chose the range $-\pi \le \theta<\pi$ for the spatial coordinate of the global AdS. The parameter $ \alpha$ is an arbitrary real number, corresponding to a particular isometry of AdS. In this map, 
the time slice $t=0$ in Poincar\'e AdS is mapped into $\tau=0$ in global AdS.

Using this map \eqref{mappo} we can map the asymptotically global, conical deficit geometry \eqref{GBmet}
\ba
ds^2=-(r^2+R^2-M)d\tau^2+\frac{R^2}{r^2+R^2-M}dr^2+r^2d\theta^2,  
\ea
into an asymptotically Poincar\'e AdS background~\cite{Horowitz:1999gf}.
As identified in~\cite{Nozaki:2013wia}, this is a time-dependent background with a local operator insertion at 
$t=x=0$ i.e. given by the local operator quench state \eqref{eq:local-quench-state}. 
The parameter $\alpha$ in \eqref{mappo} is exactly equal to the regularization parameter in \eqref{eq:local-quench-state} such that the particle trajectory \eqref{eq:traj-massive-part-local} to the center of AdS $r=0$ for any $\tau$.
Alternatively, the entire metric can be transformed into the pure one \eqref{deficitg} by a coordinate rescaling \eqref{teha}. Instead, it will make the periodicity of the angular coordinate to be $2\pi$ to $2\chi\pi<2\pi$. In this way, the conical deficit around the massive particle is manifest.

 Moreover, as can be found by taking the AdS boundary limit $r\to \infty$ and $z\to 0$, the above transformation (\ref{mappo}) corresponds to the following  conformal transformation in the CFT dual
\ba
t \pm x=\ap\tan\left(\frac{\tau\pm \theta}{2}\right).  \label{confmap}
\ea
This indeed maps local operators at $t_E=it=\pm\alpha$, $x=0$ to $\tau_E=i\tau=\pm\infty$, $\theta=0$, where primary operators for the bra and ket states are located.

Holographic EE is evaluated by calculating geodesics in this deficit background~\cite{Asplund:2014coa}. See Appendix \ref{sec:Embedd} for the geodesic calculation in general. We will perform almost identical calculations in the presence of a brane later (Section \ref{sec:connected-ent}). Note that we need to choose an appropriate OPE channel when we approximate the four-point function by the vacuum block. This corresponds to the choice of the geodesic windings around the heavy operators when the twist operators come close to each other~\cite{Asplund:2014coa}. This is an important point in the calculation in the Euclidean signature. We should be also careful about the choice of the branch when we analytically continue the result to the Lorentzian signature. There is also subtlety in the order of limits\index{order of limits}, $c\rightarrow \infty$ and $n\rightarrow 1$. This is discussed in Section 2.2 of~\cite{Hartman:2013mia} and Section 5.2 of~\cite{Fitzpatrick:2015zha}.

\section{AdS/BCFT correspondence}\label{sec:AdS-BCFTgen}
We are interested in the \textbf{AdS/BCFT correspondence}\index{AdS/BCFT correspondence}~\cite{Takayanagi2011,Fujita:2011fp,Karch:2000gx} as it is capable of describing various spacetimes as \textbf{brane worlds}\index{brane world}~\cite{Randall:1999ee,Randall:1999vf}. In particular, we can realize the big-bang/big-crunch universe described by the Friedman-Robertson-Walker (FRW) metric\index{ Friedman-Robertson-Walker metric} on the brane dual to the so-called regulated boundary state. 
This is why the AdS/BCFT correspondence is of interest to the author of this dissertation. 
The gravity picture was already addressed by~\cite{Maldacena:2004rf} even before the discovery of the AdS/BCFT correspondence or the Ryu-Takayanagi formula. After a decade, \cite{Hartman:2013qma} first proposed this picture realizing it in the AdS/BCFT correspondence. Later on, this setup has been investigated in detail, e.g. the SYK model\index{SYK model}~\cite{Kourkoulou:2017zaj} and the double trace deformation~\cite{Almheiri:2018ijj}.
Recently, Raamsdonk and his collaborators clarified the induced metric on the brane takes the form of FRW metric describing the big-bang/big-crunch universes in general dimensions and tension (a parameter in the AdS/BCFT correspondence)~\cite{Cooper:2018cmb,VanRaamsdonk:2021qgv}. Further extensions to address gravity localization with an AdS-Reissner-Nordstr\"om black hole~\cite{Antonini:2019qkt}.

Another reason for the importance of the AdS/BCFT correspondence is that the \textbf{brane-world  holography} predicts that the dynamics of \textbf{end-of-the-world (EOW) branes}\index{end-of-the-world brane}\index{EOW brane |see end-of-the-world brane}\index{ETW brane |see end-of-the-world brane} is dual to that of quantum gravity \cite{Randall:1999ee,Randall:1999vf,Gubser:1999vj,Karch:2000ct}. These two different interpretations of the EOW branes: the AdS/BCFT and the brane world. They are expected to be equivalent. This is manifest in the calculation of entanglement entropy, i.e. the holographic entanglement entropy formula for AdS/BCFT in the presence of EOW branes \cite{Takayanagi2011,Fujita:2011fp} takes the identical form as the \textbf{island formula}\index{island formula} \cite{Penington:2019npb,Almheiri:2019psf,Almheiri:2019hni,Almheiri:2019qdq,Penington:2019kki} which computes the entanglement entropy in the presence of gravity. This correspondence can be regarded as an equivalence between a BCFT and a CFT coupled to gravity, as systematically studied recently in \cite{Suzuki:2022xwv} (see also \cite{Numasawa:2022cni,Kusuki:2021gpt} for related studies from conformal field theoretic approaches). In this way, these ideas of EOW branes are deeply connected at the bottom and their further understanding is expected to be a key ingredient to %fully understand 
quantum gravity.

Before stating the claim of the AdS/BCFT correspondence, we first review boundary conformal field theory (BCFT). Then, we discuss the proposed gravity dual to BCFTs and give an example of holographic EE in BCFTs. We will also make several remarks specific to the AdS/BCFT correspondence.
Although we focus on the bottom-up construction of the AdS/BCFT correspondence in this dissertation, it is worth noting that there exists a construction of the AdS/BCFT correspondence from string theory~\cite{Fujita:2011fp,VanRaamsdonk:2021qgv,Martinec:2022ofs,Reeves:2021sab}.

\subsection{Boundary conformal field theory (BCFT)}\label{sec:BCFT}
\textbf{Boundary CFT (BCFT)}\index{boundary conformal field theory}\index{BCFT |see boundary conformal field theory} is CFT with conformal boundaries. It was originally considered by Cardy~\cite{Cardy:1984bb,Cardy:1986gw}. For reviews, see~\cite{Cardy:2004hm,Recknagel:2013uja,BB05998456}. BCFT is constructed from CFT by adding boundaries without breaking conformal symmetry as much as possible. In the following, let us consider a BCFT defined on the upper half plane (UHP) $\Im z\ge 0$ of a complex plane $z\in\mathbb{C}$. The boundary extending on the real axis $\Im z=0$ obviously breaks the translational symmetry perpendicular to the boundary. But the introduced boundary does not break conformal symmetries parallel to the boundary. In general, if a $p$-dimensional \textbf{conformal defect}\index{conformal defect} is introduced to (Euclidean) CFT, i.e. \textbf{defect CFT (DCFT)}\index{defect conformal field theory}\index{DCFT |see defect conformal field theory}\footnote{When two or more different CFTs are glued at some interfaces, they are called interface CFTs (ICFT). When all the CFTs are the same, ICFT reduces to DCFT.}, the conformal symmetry becomes partially broken:
\begin{equation}
    SO(1,d+1)\rightarrow SO(1,p+1)\times SO(d-p).
\end{equation}
When $p=d-1$, the DCFT becomes BCFT and the defect is called the \textbf{conformal boundary}\index{conformal boundary}.

\subsubsection{Boundary condition and doubling trick}
Since the boundary should be mapped to itself under any infinitesimal local conformal transformations, the mapping function must be a real, analytic function on the boundary. Furthermore, the reality condition for the field transformation implies
\begin{equation}
    T(z)=\bar{T}(\bar{z}) \Leftrightarrow T_{x t_E}=0 \label{eq:conf-bc}
\end{equation}
on the boundary ($\Im z=0$) from \eqref{eq:CWI} and its anti-holomorphic counterpart.
This implies no energy flux flow across the boundary. \eqref{eq:conf-bc} is called the \textbf{conformal boundary condition}\index{conformal boundary condition}.

From the Schwarz reflection principle\index{Schwarz reflection principle}, a holomorphic function $f(z)$ on UHP which is real-valued on the real axis can be extended to the entire complex plane by
\begin{equation}
    f(\bar{z})=\overline{f(z)}.
\end{equation}
Since infinitesimal conformal transformations are real on the real axis, they can be extended to the entire complex plane. This is known as the \textbf{doubling trick}\index{doubling trick}~\cite{Cardy:2004hm}. For example,
\begin{equation}
    \ev{O(z)}_{\mathrm{UHP}}=\ev{O(z)O(\bar{z})}^{\mathrm{chiral}}_{\mathbb{C}}\propto \frac{1}{(2\Im z)^{2h}}
\end{equation}
for an ambient\index{ambient}\footnote{Here we used ``ambient'' to denote that the operator is not located on the boundary ($\Im z>0$). Usually, such an operator is called bulk primary, however, we avoid this terminology as it may be confusing with the bulk in holography.} primary operator $O(z)$.
As for the conformal Ward identity\index{conformal Ward identity} for $n$ primary fields on UHP, the contour for $T(z)$ in UHP is reflected in lower half plane (LHP). Since the mirrored contour is clockwise, one can merge it with the original contour, which is counterclockwise. Then, the big single contour surrounds $2n$ (reflection-symmetric) chiral primary fields, i.e.
\begin{equation}
\begin{split}
    &\left\langle T(z) \prod_j O(z_j,\bar{z}_j) \right\rangle_{\mathrm{UHP}} \\
    = &\sum_j \left(
    \frac{h}{(z-z_j)^2} + \frac{1}{z-z_j} \de_{z_j} + \frac{\bar{h}}{(z-\bar{z}_j)^2} + \frac{1}{z-\bar{z}_j} \de_{\bar{z}_j}
    \right) \left \langle \prod_j O(z_j,\bar{z}_j) \right \rangle_{\mathrm{UHP}} .
\end{split}
\label{yyy}
\end{equation}

In free theory, the Dirichlet and Neumann boundary condition satisfy the conformal boundary condition \eqref{eq:conf-bc}, however, their two-point functions and energy-momentum tensor, which can be obtained from the conformal Ward identity\index{conformal Ward identity} in UHP \eqref{yyy}, are different since a two-point function on UHP becomes a chiral four-point function on $\mathbb{C}$ via the doubling trick and it is not completely fixed by conformal symmetry. This will be explicitly demonstrated in the latter part of Section \ref{sec:op-exc}.

\subsubsection{Boundary operators}
The conformal boundary condition \eqref{eq:conf-bc} yields the holomorphic sector equal to the anti-holomorphic sector on the boundary. Then, there is only one (chiral) Virasoro algebra for the operators on the boundary.
By inserting a \textbf{boundary operator}\index{boundary operator}, one can change the boundary condition. For example, a primary state can change the boundary condition $a$ on $\Re z <0, \Im z=0$ to $b$ on $\Re z>0, \Im z=0$. Such an operator is also called the \textbf{boundary (condition) changing operator}\index{boundary (condition) changing operator}.

The boundary operator is unavoidable when we discuss the radial quantization. %The dilatation becomes a symmetry only on the boundary. 
The dilatation on the boundary generates time translation on an infinite strip. The (ADM) energy\index{ADM energy} in global coordinates should be read from the eigenvalue of the dilatation operator, equal to the conformal dimension\index{conformal dimension} (defined as twice the chiral dimension\index{chiral dimension} also known as the boundary scaling dimension $h$\index{boundary scaling dimension}). Since any ambient states are generated from a dilatation on the boundary state, any ambient operators can be expanded as a linear combination of eigenstates of dilatation, i.e. boundary operators. This is known as the \textbf{boundary operator expansion (BOE)}\index{boundary operator expansion}\index{BOE |see boundary operator expansion }.\footnote{This is nothing but a consequence of the operator/state correspondence\index{operator/state correspondence}.} The BOE coefficients\index{BOE coefficients} are related to the ambient OPE coefficients\index{OPE coefficients} via the calculation of correlation functions.
Note that when a boundary operator is involved in correlation functions, it should not be doubled when the doubling trick\index{doubling trick} is applied~\cite{Recknagel:2013uja}. For more details, see~\cite{Sully:2020pza,Reeves:2021sab,Biswas:2022xfw} and Section 3.1 of~\cite{Wakeham:2022wyx}.

\subsubsection{Boundary entropy}\label{sec:boundary-state}
In the radial quantization picture, the Euclidean path integral on a unit disk $D$ with a boundary condition\index{boundary condition} $a$ defines a transition amplitude
\begin{equation}
    g_a \equiv \braket{0}{B_a},
\end{equation}
where $\ket{B_a}$ is a \textbf{boundary state}\index{boundary state} satisfying the boundary condition,
\begin{equation}
    \left.\qty(T(z)-\bar{T}(\bar{z}))\right |_{z\in \partial D} \ket{B_a} = 0.
\end{equation}
The \textbf{boundary entropy}\index{boundary entropy} is defined as $S_{bdy}^{(a)}=\log g_a$.
It is named `entropy' as it arises as a constant offset in thermodynamic entropy from the cylinder amplitude (known as the `overlap') due to the boundary~\cite{Affleck:1991tk,Cardy:2004hm}.

When we compute EE for an interval of length, the same contribution appears in addition to the usual size-dependent term from the Calabrese-Cardy formula. Provided the boundary is located at $x=0$ in the two-dimensional BCFT and the subregion $A$ is an interval $[0,l]$, EE\index{entanglement entropy} is given by~\cite{Calabrese:2004eu}
\begin{equation}
    S_A=\frac{c}{6}\log \frac{2l}{\epsilon} + S_{bdy}.
    \label{eq:BCFT-EE}
\end{equation}
Since EE in the BCFT is related to that in the chiral CFT\index{chiral CFT} via the doubling, the length dependence must be $2l$ while the overall factor must be half of \eqref{eq:calabrese-EE}.
Since $S_{bdy}$ is just a constant, we mostly focus on the first term in this dissertation.

%\subsubsection{Boundary states}\label{sec:boundary-state}

\subsection{Statement of the AdS/BCFT correspondence}\label{sec:AdS-BCFT-state}
Writing the two-dimensional coordinates as $(t,x)$ of the CFT, we define our BCFT by restricting the spacetime in the right half plane $x>x_0$. Its gravity dual can be found by inserting an end-of-the-world brane (EOW or ETW brane) appropriately in the 
three-dimensional asymptotically AdS. This prescription is called the AdS/BCFT correspondence~\cite{Takayanagi2011,Fujita:2011fp}. In the following, we consider the gravity sector only, i.e. constant matter fields.
For a discussion involving bulk matter fields, see \cite{Fujita:2011fp,Suzuki:2022xwv,Kusuki:2022ozk} (see also~\cite{Biswas:2022xfw,Miyaji:2022dna} for a nontrivial matter or defect on the brane, which may resolve some issues regarding intersecting branes). The review in this section is mostly based on~\cite{Fujita:2011fp}.

The \textbf{AdS/BCFT correspondence} states that the $d$-dimensional BCFT on a manifold $\mathcal{M}$ is dual to the $(d+1)$-dimensional asymptotically AdS spacetime on a manifold $\mathcal{N}$ bounded by an \textbf{EOW brane}\index{end-of-the-world brane}\index{ETW brane |see end-of-the-world brane }\index{EOW brane |see end-of-the-world brane } $Q$. The BCFT lives on the asymptotic boundary $\mathcal{M}$ such that the EOW brane ends on the BCFT boundary, $\partial\mathcal{M}=\partial\mathcal{Q}$, when there is no other boundary to end. The EOW brane bounds the AdS spacetime such that $\partial \mathcal{N}=\mathcal{M}\cup Q$. See Fig.1 of~\cite{Fujita:2011fp} for the illustration. Note that $Q$ is not necessarily located on the asymptotic boundary. This means $Q$ is a gravitating system in contrast to the boundary in the conventional AdS/CFT correspondence. The boundary data in the BCFT roughly corresponds to the brane data in the AdS.

The (Lorentzian) gravity action is given by
\begin{equation}
    I=\frac{1}{16\pi G_N } \int_{\mathcal{N}}\dd[d+1]{x} \sqrt{-g} (\mathcal{R}-2\Lambda) + \frac{1}{8\pi G_N}\int_Q \dd[d]{x} \sqrt{-h} (K-{\cal T}),
    \label{eq:action-AdS/BCFT}
\end{equation}
where $h_{ab}$ is the induced metric on $Q$, $K_{ab}=\nabla_{a} n_b$ is the \textbf{extrinsic curvature}\index{extrinsic curvature} on $Q$ ($n^a$ is a normal vector on $Q$), and $\cal T$ is the \textbf{tension}\index{tension (brane)} (cosmological constant) of the brane $Q$.\footnote{We can say this comes from a matter action $I_Q$ with the energy-momentum tensor
\begin{equation}
    T_{ab}^Q=\frac{2}{\sqrt{-h}}\frac{\delta I_Q}{\delta h_{ab}}=\frac{\cal T}{8\pi G_N}.
\end{equation}
} 
We omitted the Gibbons-Hawking term\index{Gibbons-Hawking term} at the asymptotic boundary and the counterterm\index{counterterm} (cf. Appendix \ref{app:grav-action}). $\Lambda$ is the AdS cosmological constant, which is related to the AdS radius $R$ as $\Lambda=-d(d-1)/(2R^2)$.

By varying the action \eqref{eq:action-AdS/BCFT}, the bulk part of the action can be eliminated by the equation of motion after the integration by parts. The remaining part is given by
\begin{equation}
    \delta I = \frac{1}{16\pi G_N}\int_Q \dd[d]{x} \sqrt{-h} (K_{ab}-K h_{ab}+{\cal T}h_{ab})\delta h^{ab}.
    \label{eq:action-variation-bcft}
\end{equation}
While we impose the Dirichlet boundary condition $\delta h_{ab}=0$ at the asymptotic boundary, we impose the \textbf{Neumann boundary condition}\index{Neumann boundary condition} on $Q$ so that the brane profile is dynamically determined.\footnote{It is also possible to impose different boundary conditions. See~\cite{Chu:2021mvq} and references therein. Here we follow the original proposal by Takayanagi.} The Neumann boundary condition read from \eqref{eq:action-variation-bcft} is given by
\begin{equation}
    K_{ab}-K h_{ab} = - {\cal T} h_{ab}.
    \label{eq:EOW}
\end{equation}
By taking its trace, we obtain
\begin{equation}
    K=\frac{d}{d-1} {\cal T}.
\end{equation}
Combined with \eqref{eq:EOW}, it implies
\begin{equation}
    K_{ab}=\frac{1}{d-1} {\cal T} h_{ab}.
    \label{eq:tension-K}
\end{equation}

To find the solution, it is useful to work in the \textbf{Gaussian normal coordinates}\index{Gaussian normal coordinates}
\begin{equation}
    ds^2 = \pm dn^2 + h_{ab} dx^a dx^b,
\end{equation}
where $\{x^a\}$ are coordinates on $Q$.
The sign of the first term is $-$ if $Q$ is spacelike while $+$ if $Q$ is timelike.\footnote{The spacelike brane corresponds to the pre/post-section in CFT by a boundary state~\cite{Numasawa:2016emc,Antonini:2022sfm,Akal:2021dqt,Akal:2020wfl,Cooper:2018cmb,VanRaamsdonk:2021qgv}.} In our cases, we focus on a timelike $Q$ so the sign is $+$.
The extrinsic curvature in the Gaussian normal coordinates is known to be~\cite{Wald:1984rg}
\begin{equation}
    K_{ab}=\frac{1}{2}\frac{\partial h_{ab}}{\partial n}.
    \label{eq:ex-curv-gauss}
\end{equation}

From the Poincar\'e metric \eqref{Pmet}, we can bring the metric into
\begin{equation}
    ds^2=d\rho^2+\cosh^2 \frac{\rho}{R} \qty(R^2 \frac{\eta_{ab}dx^a dx^b}{y^2})
    \label{eq:metric-slicing-AdS}
\end{equation}
by the following AdS slicing:\footnote{Note that $\rho>0$ corresponds to $x<x_0$ region due to the sign. We chose this sign since we want to parametrize the allowed bulk region as $\rho=-\infty$ to some point. This is opposite to the setup in~\cite{Takayanagi2011,Fujita:2011fp}.}
\begin{equation}
    z=\frac{y}{\cosh(\rho/R)},\quad x-x_0= - y \tanh \frac{\rho}{R},
    \label{eq:rho-coords}
\end{equation}
where $y\ge 0$ and $\rho\in(-\infty,\infty)$.
In these coordinates, we can show $\rho=\rho_\ast$ (constant) solves the Neumann boundary condition \eqref{eq:EOW}. 
Then, \eqref{eq:metric-slicing-AdS} is indeed in Gaussian normal coordinates with
\begin{equation}
    h_{ab}=\frac{R^2}{y^2} \cosh^2 \frac{\rho}{R} \eta_{ab}.
\end{equation}
We can find its extrinsic curvature \eqref{eq:ex-curv-gauss} to be
\begin{equation}
    K_{ab}=\frac{1}{R} \tanh\qty(\frac{\rho_\ast}{R})h_{ab},
\end{equation}
from which the tension $\cal T$ and the brane location $\rho=\rho_\ast$ are related to each other from \eqref{eq:tension-K} as
\begin{equation}
    {\cal T}R = (d-1)\tanh \frac{\rho_\ast}{R}.
\end{equation}

In the original Poincar\'e coordinates \eqref{Pmet}, the brane profile $\rho=\rho_\ast$ is given by
\begin{equation}
    x-x_0=-\lambda z, \quad \lambda=\sinh\frac{\rho_\ast}{R}.
    \label{eomb}
\end{equation}
The tension\index{tension (brane)} $\cal T$ is related to the slope $\lambda$ as
\ba
{\cal T}R=(d-1)\frac{\lambda}{\s{1+\lambda^2}}.
\ea
In the following discussions, we set $x_0=0$ for simplicity unless noted.
The EOW brane (\ref{eomb}) corresponds to the following profile in the global AdS$_3$ (\ref{Gmet}):
\ba
r\sin\theta=-\lambda R,
\ea
which is depicted in Fig.\ref{mapfig2}. Note that the brane always extends to the antipodal point $\theta=\pi$ in this time-independent setup. A brane cannot bend more because of the holographic $g$-theorem, which can be proven by imposing the \textbf{null energy condition}\index{null energy condition} on $Q$~\cite{Fujita:2011fp}.

\begin{figure}
  \centering
  \includegraphics[width=12cm]{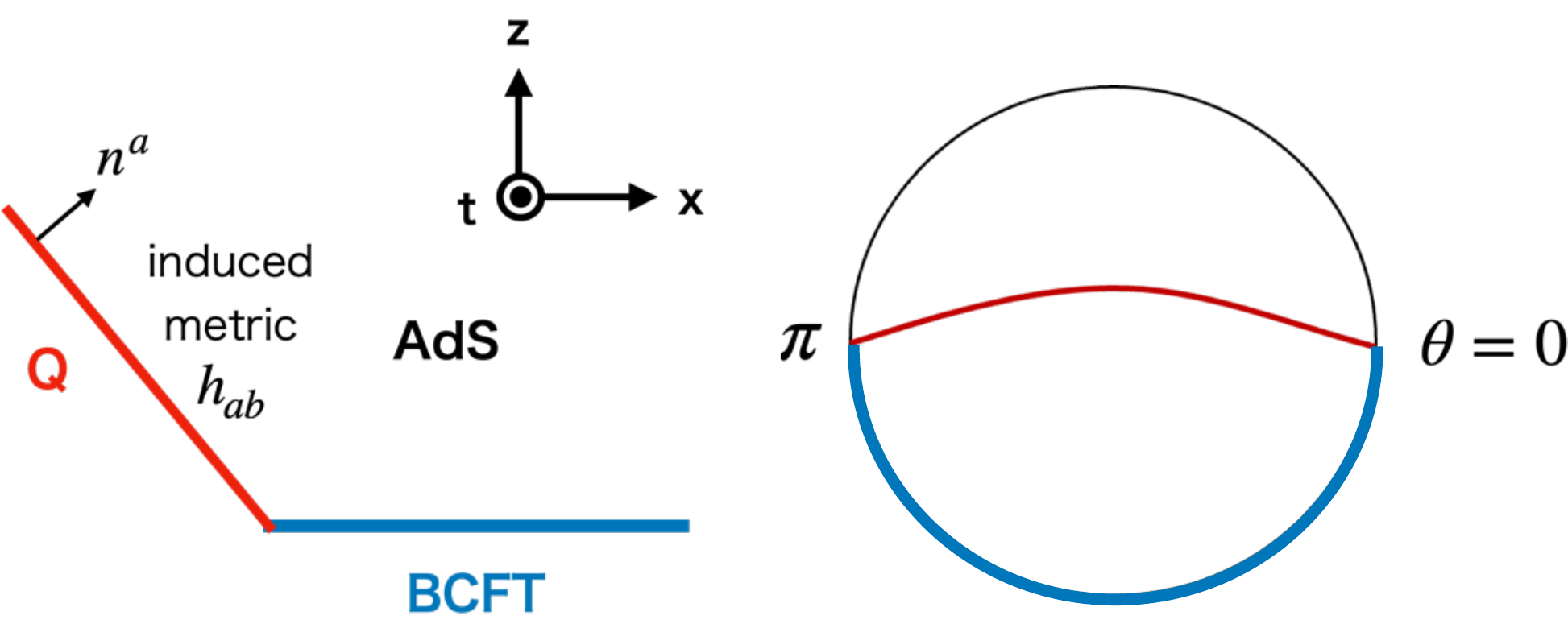}
  \caption{A constant time slice of the gravity dual of BCFT (the blue curve) in the Poincar\'e coordinates (left) and global coordinates (right). The EOW brane is drawn as a red curve.}
\label{mapfig2}
\end{figure}

\subsection{Holographic entanglement entropy in the AdS/BCFT correspondence}\label{sec:HEE-BCFT}
When CFT has no boundary, the RT formula \eqref{eq:HEE-RT} tells us EE is given by the length of the minimal geodesic. The essence does not change even after introducing boundaries, however,
we need to take into account geodesics which end on the EOW brane \cite{Takayanagi2011,Fujita:2011fp}. This leads to multiple candidates of geodesics for holographic EE (Fig.\ref{fig:HEE-AdSBCFT}). One is the ``connected geodesic'' $\Gamma^{con}_{AB}$, which literary connects these two boundary points $A$ and $B$. The other is the ``disconnected geodesics'' $\Gamma^{dis}_{AB}$ which consist of two disjoint pieces, one connects the boundary point $A$ and a point on the EOW brane, and the other connects $B$ and ends on the EOW brane. The net result for the holographic entanglement entropy\index{holographic entanglement entropy} is given by taking the minimum of these two contributions,
\begin{equation}
    S_{AB} ={\rm Min} \left\{ S^{con}_{AB},  S^{dis}_{AB}\right\},
    \label{eq:HEE-BCFT}
\end{equation}
where
\begin{equation}
    S^{con}_{AB}=\frac{L(\Gamma^{con}_{AB})}{4G_N}, \quad
    S^{dis}_{AB}=\frac{L(\Gamma^{dis}_{AB})}{4G_N}.
    \label{eq:HEE-phase}
\end{equation}
$L(\Gamma)$ denotes the length of the geodesic $\Gamma$.
By this modification to the RT formula, the result agrees with the BCFT calculation \eqref{eq:BCFT-EE} when the state is vacuum~\cite{Calabrese:2004eu}.\footnote{In the case of non-vacuum states, matching the BCFT result with gravity dual calculations is very nontrivial and needs fine-tuning~\cite{Sully:2020pza}.}
Each phase can be computed in the BCFT as a chiral four-point function of twist operators. A particular choice of OPE channels corresponds to each phase.

\begin{figure}
  \centering
  \includegraphics[width=12cm]{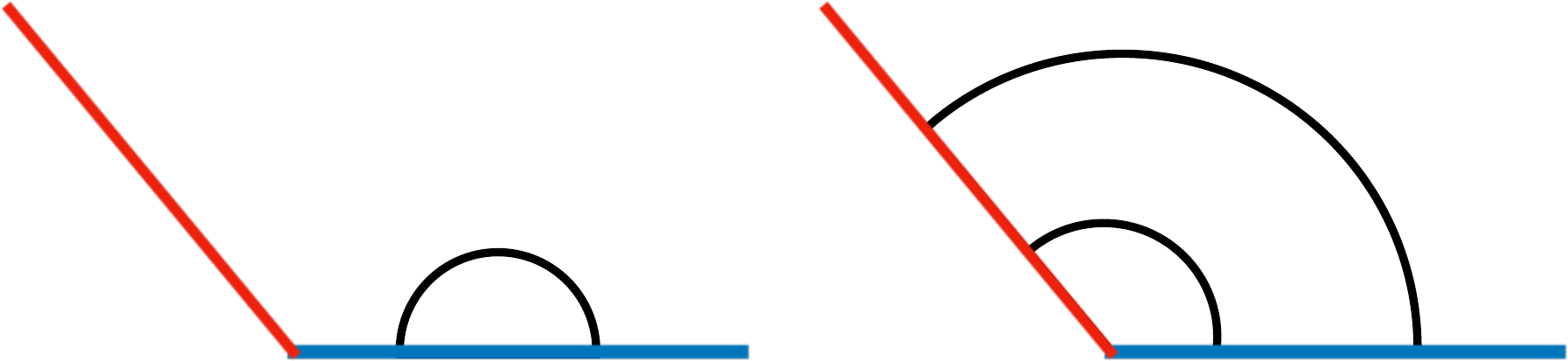}
  \caption{Two phases of holographic EE in the AdS/BCFT correspondence: connected (left) and disconnected (right). Since the smaller one gives EE, a small subregion corresponds to the connected geodesics and a large subregion corresponds to the disconnected geodesics.}
\label{fig:HEE-AdSBCFT}
\end{figure}

In the following, we consider a case when the subregion $I$ is an interval, whose boundary point is $x=0$ and $x=l$, with the BCFT boundary at $x=0$. In this case, the disconnected geodesic is just a single curve emanating from $x=l$.
We restrict $\cal T$ to a certain parameter region where the connected geodesic does not have an intersection with the EOW brane as it may cause a problem~\cite{Kusuki:2022wns}.
We demonstrate below that the holographic calculation in AdS$_3$ reproduces the BCFT result \eqref{eq:BCFT-EE}. The geodesic is given by $z^2+x^2=l^2$. This is translated to the AdS slicing coordinates \eqref{eq:rho-coords} as
\begin{equation}
    y=l.    
\end{equation}
Then, the UV cutoff $z=\epsilon$ on the (disconnected) geodesic is equivalent to
\begin{equation}
    \rho_\infty = R\, \mathrm{arccosh}\frac{l}{\epsilon} = R \log \frac{2l}{\epsilon} \quad (l/\epsilon\rightarrow\infty).
\end{equation}
Since the geodesic is exactly lying on the constant-$y$ (and $t$) surface, holographic EE is easily calculated as
\begin{equation}
    S_A=\frac{1}{4G_N}\int_{-\rho_\infty}^{\rho_\ast}d\rho = \frac{\rho_\ast+\rho_\infty}{4G_N}= \frac{c}{6}\log\frac{l}{\epsilon}+ S_{bdy},
\end{equation}
where we used the Brown-Henneaux relation\index{Brown-Henneaux relation} \eqref{eq:brown-henneaux}. This is always smaller than the connected entropy $\frac{c}{3}\log \frac{l}{\epsilon}$ \eqref{eq:calabrese-EE}.
Thus, the holographic EE perfectly reproduces \eqref{eq:BCFT-EE} by defining the boundary entropy\index{boundary entropy} as
\begin{equation}
    S_{bdy}=\frac{\rho_\ast}{4G_N}=\frac{c}{6}\mathrm{arctanh}(R{\cal T})= \frac{c}{6}\mathrm{arcsinh}\lambda.
    \label{eq:bdy-ent-formula}
\end{equation}
This definition can be confirmed by other definitions of the boundary entropy like the disk partition function.
Geometrically, the boundary entropy\index{boundary entropy} comes from the length of the geodesic to the left of the center up to the brane. This is perfectly consistent with the intuition of the boundary entropy counting the boundary degrees of freedom. The location of the brane with no $S_{bdy}$ corresponds to the center and the boundary degrees of freedom is describing the region going beyond the center.

\subsection{Subtleties regarding the AdS/BCFT correspondence}
In this subsection, we briefly review some known subtleties which will not be mentioned in detail in the following sections. Despite these subtleties, the AdS/BCFT correspondence shows universal results regarding EE. Thus, we focus on its gravitational sector only in this dissertation.

%\subsubsection{Which BCFT is holographic?}
It is not trivial which BCFT has a holographic dual a priori. In~\cite{Sully:2020pza,Reeves:2021sab,Wakeham:2022wyx}, it has been studied what is necessary for a BCFT to have a gravity dual in the spirit of~\cite{Heemskerk:2009pn} (see a conjectured definition for holographic BCFTs given in~\cite{Reeves:2021sab}). In fact, holographic BCFTs\index{holographic BCFTs} are very special among BCFTs. According to these papers, both the ambient and boundary channel vacuum-block dominance are required and it leads to a particular sparse spectrum. Furthermore, the conformal bootstrap assuming vacuum dominance yields a particular randomness in the OPE coefficients from ETH \eqref{eq:ETH} as well as a similar sparseness condition~\cite{Kusuki:2021gpt} (see also~\cite{Kusuki:2022wns}). 

The specialty of (purely gravitational) holographic BCFTs also stems from the vanishing one-point function~\cite{Suzuki:2022xwv}. 
This is because the excitation we consider is only the metric in gravity. This means the only nonzero expectation value (one-point function) is the stress tensor~\cite{Horowitz:1999gf} and otherwise zero:
\[\ev{O}=0\quad (O\neq T_{\mu\nu}).\]
From the geodesic approximation, the one-point function should be described by a geodesic extending from the asymptotic boundary to somewhere else. Unless the field is sourced by the EOW brane (cf. Section 8.1 of~\cite{Fujita:2011fp}), the one-point function should vanish. When the EOW brane sources the matter field, we need to add the matter field on the brane. This is highly model-dependent. Said conversely, the validity of the geodesic approximation constrains the field dependence on the brane~\cite{Kastikainen:2021ybu}.

%\subsection{Double holography}

%\subsection{Constructing AdS/DCFT or AdS/ICFT}

\chapter{Holographic local operator quench in BCFTs}\label{ch:2-2}
\renewcommand{\thesection}{\thechapter.\arabic{section}}
%!TEX root = ../thesis.tex
%*******************************************************************************
%****************************** Second Chapter *********************************
%*******************************************************************************

%Entanglement in holography
%\chapter{My second chapter}

\ifpdf
    \graphicspath{{Chapter2/Figs/Raster/}{Chapter2/Figs/PDF/}{Chapter2/Figs/}}
\else
    \graphicspath{{Chapter2/Figs/Vector/}{Chapter2/Figs/}}
\fi

\renewcommand{\thesection}{\thechapter.\arabic{section}}
\setcounter{section}{0}
\textit{This chapter follows my own work with members at YITP, Tadashi Takayanagi, who is my host faculty during the atom-type fellow, Tomonori Ugajin, whom I mainly performed CFT calculations with, and their graduate students Taishi Kawamoto and Yu-ki Suzuki~\cite{Kawamoto:2022etl}. 
%For the notation consistency, we will denote time in Poincar\'e coordinates by $t$ and its Wick rotation by $t_E$ while time in global coordinates by $\tau$.}
%
%\textit{
It is recommended to read Chapter \ref{ch:2} before this chapter, yet, it is still possible to gain comprehension by consulting referred sections as needed.
The author of this dissertation has contributed to realizing and justifying the modified relation between the deficit angle and the conformal dimension of the inserted operator, the overall BCFT calculation including the energy-momentum tensor and entanglement entropy, and writing the corresponding part of the paper.}\\

In Chapter \ref{ch:1}, we considered vacuum entanglement entropy in massive, interacting QFTs, which are non-conformal. In this chapter, we begin with holography but add additional factors: conformal boundaries and excitation. This leads to a time-dependent setup in contrast to the vacuum case while we can make use of a part of conformal symmetries and holographic formulae. However, a naive extension of the original AdS/CFT correspondence involves several problems. We provide a correct prescription for the gravity dual of the operator local quench (excitation) in BCFTs and argue it perfectly resolves the problems.

This chapter is organized as follows.
In Section \ref{sec:quench-BCFT}, we explain the setup of holographic local operator quench in BCFTs. 
In Section \ref{sec:coords-trasform}, we present our gravity dual of the local operator quench by introducing a localized excitation in the AdS/BCFT. This is achieved by a combination of coordinate transformations.
We also describe the folding brane problem and the induced boundary problem due to a naive application of the AdS/CFT dictionary. We propose a resolution to eliminate an induced boundary by rescaling.
In Section \ref{sec:m-less-r}, we calculate its holographic stress tensor in this model and show the result matches a naive expectation using a correct prescription. We also examine a coordinate transformation that helps us to identify its BCFT dual. Importantly, we show our prescription modifies the black hole threshold so that the folding problem is avoided.
In Section \ref{sec:holoEE}, we compute the holographic entanglement entropy in our holographic local quench model. In Section \ref{sec:BCFT-EM}, we provide the BCFT description of our local operator quench and compute the stress tensor. We compare the results in the gravity dual with those in the BCFT. In Section \ref{sec:BCFT-EE-calc}, we compute entanglement entropy in the BCFT and observe a consistent result with the identified gravity dual.
%
%Appendix
In Appendix \ref{sec:Embedd}, we compute the geodesic distance in the embedding space formalism.
In Appendix \ref{app:curve}, we derive \eqref{disco} by analyzing the branch, from which we obtain the correct black hole threshold.
In Appendix \ref{app:branch}, we perform the monodromy analysis to determine the branch of cross ratios in \eqref{eq:conn-EE-CFT} and \eqref{eq:dis-next-br}.

\section{Setup}\label{sec:quench-BCFT}
In this chapter, we consider timelike boundary cases but with a large and inhomogeneous excitation. This is because studies of EOW branes so far have mainly been limited to holographic setups at zero or finite temperatures. This raises a basic question of whether the holographic duality of AdS/BCFT works successfully in time-dependent backgrounds. It is also interesting for our ultimate goal of realizing our universe since this setup is time-dependent and some nontrivial backreactions like a conical defect and black hole are involved.

We focus on an analytical model where the excitation is created by a massive particle in  three-dimensional AdS geometry. Via the AdS/BCFT duality, this is dual to a local operator quench\index{local operator quench} in the holographic two-dimensional BCFT:\footnote{
In this section, our two-dimensional BCFT has either a single boundary at $x=0$ or two boundaries (one at $x=0$ and the other at a specific time-dependent location $x=Z(t)$). Although in the former case, in which $H|0\lb=0$ holds, the definition (\ref{LOS}) is equivalent to the usual definition of the local quench $|\Psi(t)\lb=e^{-itH}e^{-\ap H}O(x=x_a)|0\lb$, we need to modify this for the latter case, in which $H|0\lb\neq 0$, as the time evolution is not unitary but isometry as in \cite{Cotler:2022weg}.
As we will discuss later, \eqref{LOS} is the correct definition for the local quench dual to what is discussed in this chapter. (For an advanced reader: precisely speaking, the operator $O$ must be treated as a boundary operator and its chiral dimension is related to the bulk mass~\cite{Kusuki:2022ozk}. However, our result is shown to be consistent with this treatment.)
%Since we employ the path-integral description later, we will not get into details of this.
}
\ba
|\Psi(t)\lb=e^{-itH}O(x=x_a,t_E=-\alpha)|0\lb,  \label{LOS}
\ea
where $H$ is the CFT Hamiltonian, $t_E=it$, and $\ap$ is a regularization parameter \cite{Nozaki:2014hna,He:2014mwa,Caputa:2014vaa,Asplund:2014coa} (refer also to \cite{Guo:2015uwa} for an earlier analysis of 
local operator quench in BCFTs). Our gravity model is obtained by introducing an EOW brane\index{end-of-the-world brane} in the holographic local quench model~\cite{Nozaki:2013wia} discussed in Section \ref{sec:local-op-quench}. The shape of the EOW brane is deformed by the local excitation via the gravitational backreaction, which gives novel dynamics of the local quench in the BCFT. Via the double holography, this model is also closely related to the local quench in a two-dimensional gravity studied in \cite{Goto:2020wnk} as a model of black hole evaporation.
We analyze the system both from the gravity dual calculations and from the direct computations in BCFT. 

Our local quench model is analytically tractable both from the gravity side and the BCFT side. The former can be found by finding the correct asymptotically AdS spacetime with an EOW brane and a massive particle. This gravity dual geometry allows us to calculate the holographic stress tensor \cite{Balasubramanian:1999re} and holographic entanglement entropy \cite{Ryu:2006bv,Ryu:2006ef,Hubeny:2007xt}.  
On the other hand, the latter can be analyzed via  direct conformal field theoretic computations by employing a suitable conformal map. We will note that a careful choice of the coordinate system in the gravity dual is important when we compare the results of the former with those of the latter. Eventually, we will find complete matching between the gravity dual results and the BCFT ones in the large-$c$ limit.

\section{Coordinate transformations}\label{sec:coords-trasform}
In this section, we perform coordinate transformations to find a gravity dual of CFT with both boundaries and a heavy excitation.
\subsection{Recap: Coordinates for holographic local quench} 
Let us briefly recap what we have seen in Section \ref{sec:asympt-AdS} and \ref{sec:local-op-quench}.
By \eqref{mappo}, we can derive the geometry dual to the local operator quench in BCFT by pulling back the backreacted global coordinates \eqref{GBmet}
\ba
ds^2=-(r^2+R^2-M)d\tau^2+\frac{R^2}{r^2+R^2-M}dr^2+r^2d\theta^2,   
\ea
where $\theta$ has the periodicity $2\pi$ as $-\pi\leq \theta<\pi$,
to the asymptotically Poincar\'e AdS describing the infalling particle (Fig.\ref{mapfig}). 
The line $x=0$ at the boundary $z=0$ in Poincar\'e AdS is corresponds to $\theta=0$. Since the excitation is at $\tau=\pm \ap$ on the asymptotic boundary $z=0$, the massive particle is dual to a local operator excitation given by (\ref{LOS}) with $x_a=0$, where the mass $m$ is related to the conformal dimension $\Delta$ of the primary operator $O(x)$ via 
$\Delta\simeq mR$ as identified in \cite{Nozaki:2013wia}.

\begin{figure}
  \centering
  \includegraphics[width=8cm]{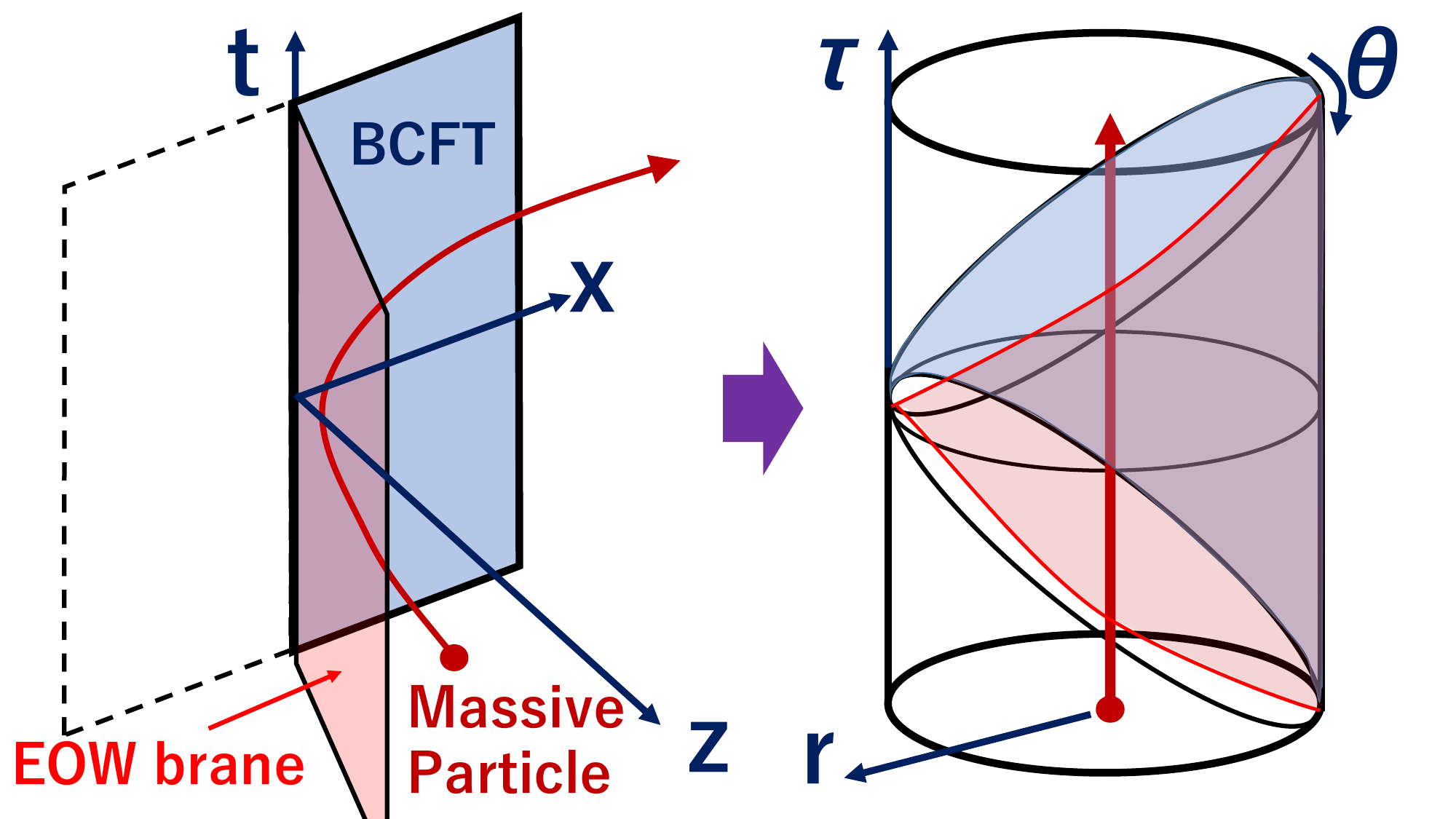}
  \caption{A sketch of the coordinate transformation from the Poincar\'e AdS into a global AdS in the 
presence of a massive particle (the red arrow) and an EOW brane (the red surface).}
\label{mapfig}
\end{figure}

\subsection{Profiles of EOW branes}\label{sec:m-zero}
We are interested in the effects of the presence of a boundary in the aforementioned local quench, i.e. the profile of the EOW brane and its dynamics. 
In the following sections, we focus on the case ${\cal T}\neq 0$ so that the falling particle is off the EOW brane.
Since the backreacted metric due to the local quench is locally pure global AdS in the tilded global coordinates \eqref{deficitg},
it is convenient to use the map \eqref{teha}. For $0<M<R^2$, it rescales to the global AdS with a conical deficit \eqref{deficitg}
\ba
ds^2=-(\ti{r}^2+R^2)d\ti{\tau}^2+\frac{R^2}{\ti{r}^2+R^2}d\ti{r}^2+\ti{r}^2d\ti{\theta}^2.
\ea
The new angular coordinate $\ti{\theta}$ has a deficit angle $2\pi\chi$ as we have seen in \eqref{tildtheta}. For $M>R^2$, the geometry (\ref{GBmet}) describes a BTZ black hole, where the horizon is situated as $r=\s{M-R^2}$.

Recall that the EOW brane (\ref{eomb}) corresponds to the following profile in the global AdS$_3$ (\ref{Gmet}):
\ba
r\sin\theta=-\lambda R,
\ea
where we set the boundary location $x_0=0$.
We can specify the brane profiles in the time-dependent geometry by applying the chain of diffeomorphisms. We begin with the profile of the EOW brane in the transformed geometry (\ref{deficitg}),
\ba
\ti{r}\sin \ti{\theta}=-\lambda R.  \label{profti}
\ea
Thus for $0\leq M<R^2$, the EOW brane profile is found in the original coordinates as 
\ba
r\sin\left(\chi\theta\right)=-\lambda \s{R^2-M}.
\label{beoms}
\ea
%Thus 

For $\lambda>0$, the EOW brane extends from $\theta=0$ to $\theta=-\frac{\pi}{\chi} \simeq 2\pi-\frac{\pi}{\chi}$ as
 depicted in the middle panel of Fig.\ref{deformationfig}. When $M=0$, the brane intersects with the asymptotic boundary at $\theta=0$ and $\theta =\pi$. As we increase the mass  of the bulk particle, the brane profile is eventually  bent, and the coordinate distance between the two endpoints gets closer.
 
 For $\lambda<0$, the EOW brane extends from $\theta=0$ to $\theta=\frac{\pi}{\chi}\simeq \frac{\pi}{\chi}-2\pi$ as shown in Fig.\ref{deformationfig2}. In contrast to the $\lambda>0$ cases, the dual gravity region does not include the falling particle but is affected by its backreaction.
 
 %In order to have a non-trivial effect due to the massive particle, we assume 
 %$ {\cal T}>0$ below. This is because the center $r=0$, where the massive particle is situated, is included in the gravity dual spacetime only when $ {\cal T}>0$. 
 %Moreover, 
 We would like to note that if $M$ is large enough such that $M>\frac{3}{4}R^2$, then the EOW brane gets folded as in the right panel of Fig.\ref{deformationfig} and the gravity dual does not seem to make sense. As discussed in \cite{Cooper:2018cmb} (see also \cite{Miyaji:2021ktr,Geng:2021iyq}), this kind of folded solution implies that the hard wall approximation is no longer valid in the presence of strong backreaction. Nevertheless, as we will see later, suitable treatment of the AdS/BCFT yields $M$ is always less than or equal to $\frac{3}{4}R^2$, given a conformal dimension less than the black hole threshold.

\begin{figure}
  \centering
  \includegraphics[width=12cm]{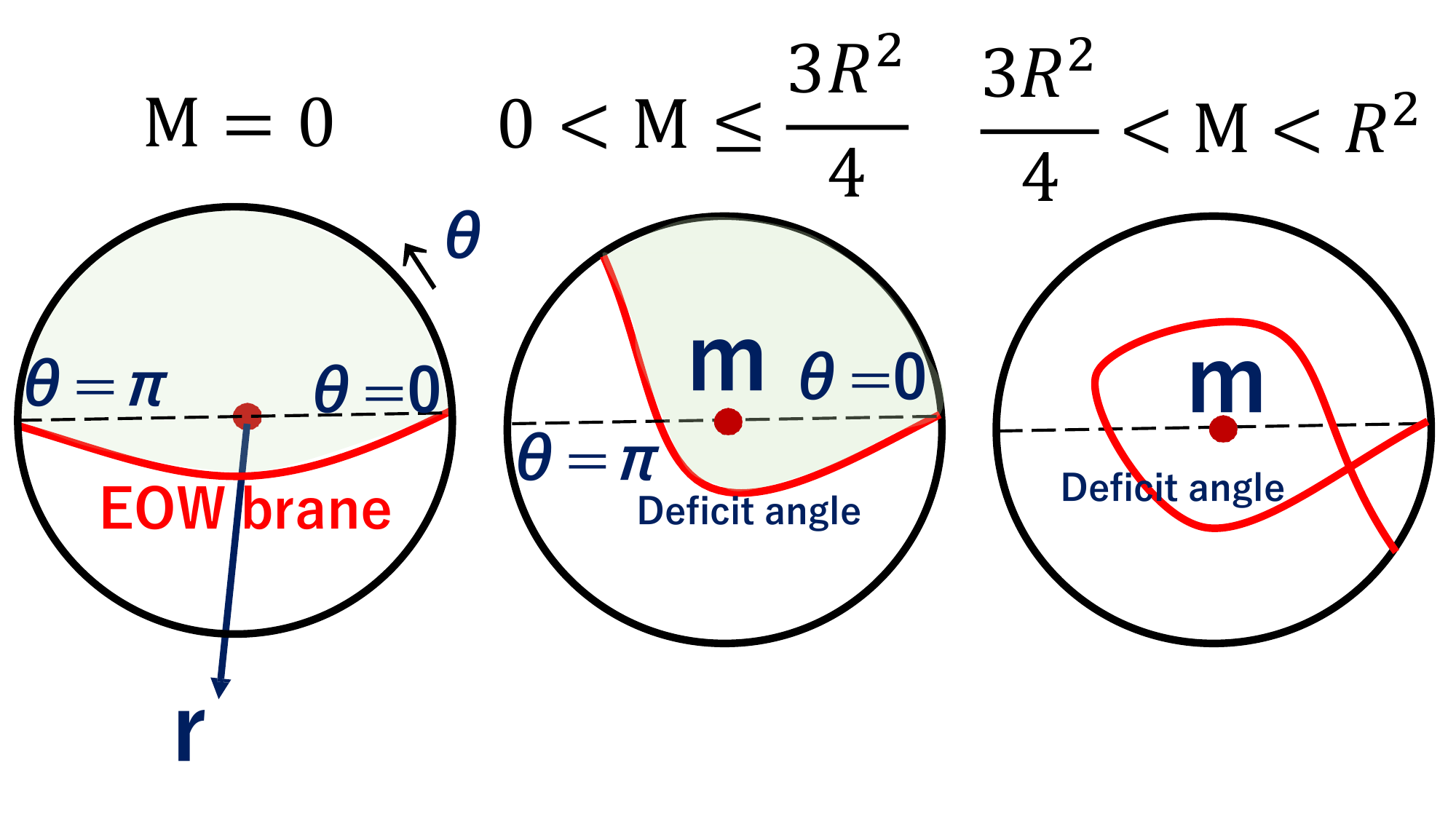}
  \caption{Cross sections at constant $\tau$ for the backreacted geometry with a mass and a positive tension $\lambda>0$. We depicted the EOW brane as red curves. The light green regions are the gravity duals in the AdS/BCFT. Though for $M>\frac{3}{4}R^2$, the EOW brane gets folded and the gravity dual does not make sense, we do not need this range of the mass when we consider the BCFT dual as we explain around (\ref{relasd}).}
\label{deformationfig}
%\end{figure}
%
%\begin{figure}
%  \centering
  \includegraphics[width=12cm]{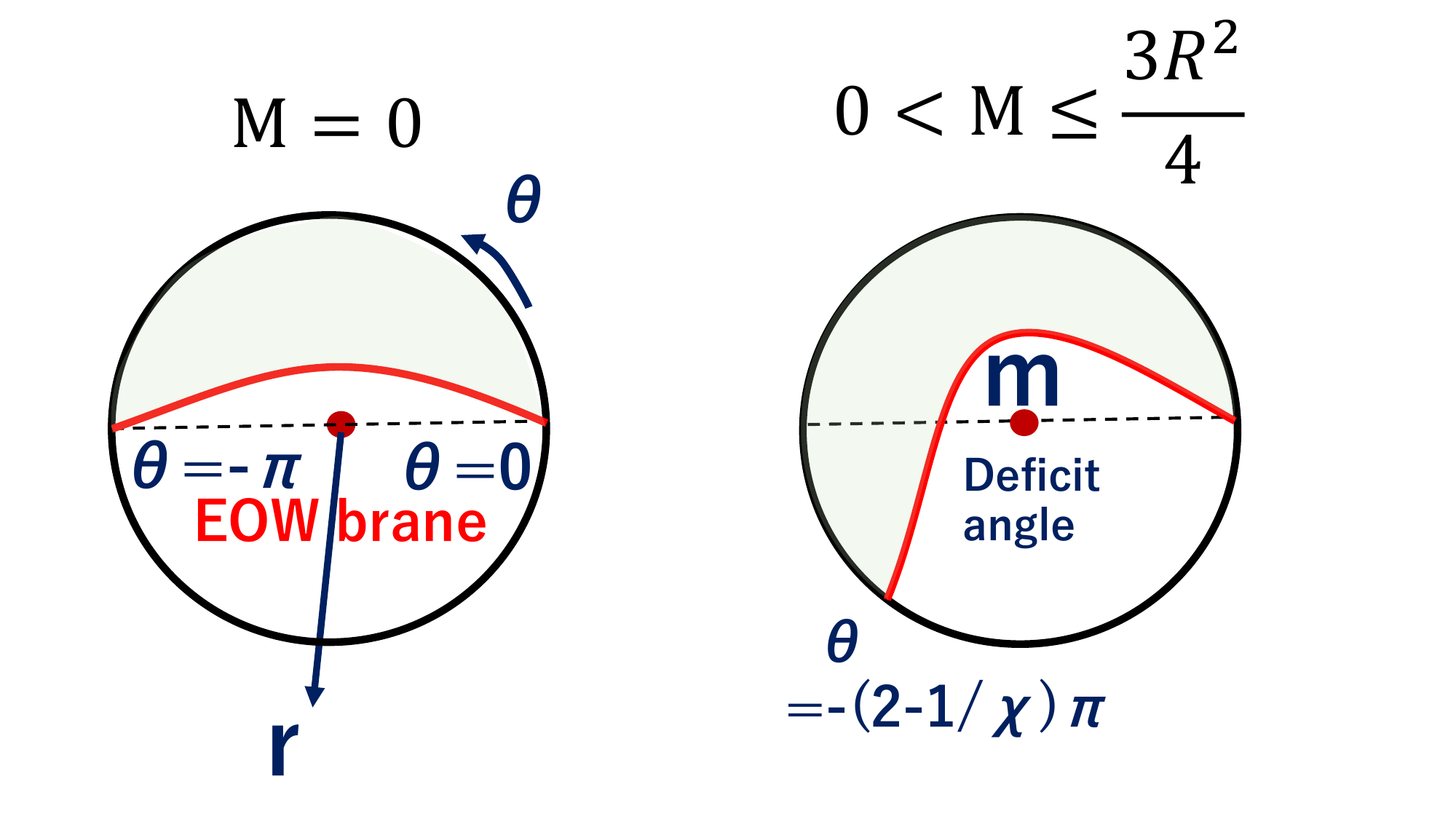}
  \caption{Cross sections at constant $\tau$ for the backreacted geometry with a mass and a negative tension $\lambda<0$. We depicted the EOW branes as red curves. The light green regions are the gravity duals of the BCFT.}
\label{deformationfig2}
\end{figure}

\begin{figure}
\centering
    \includegraphics[width=5cm]{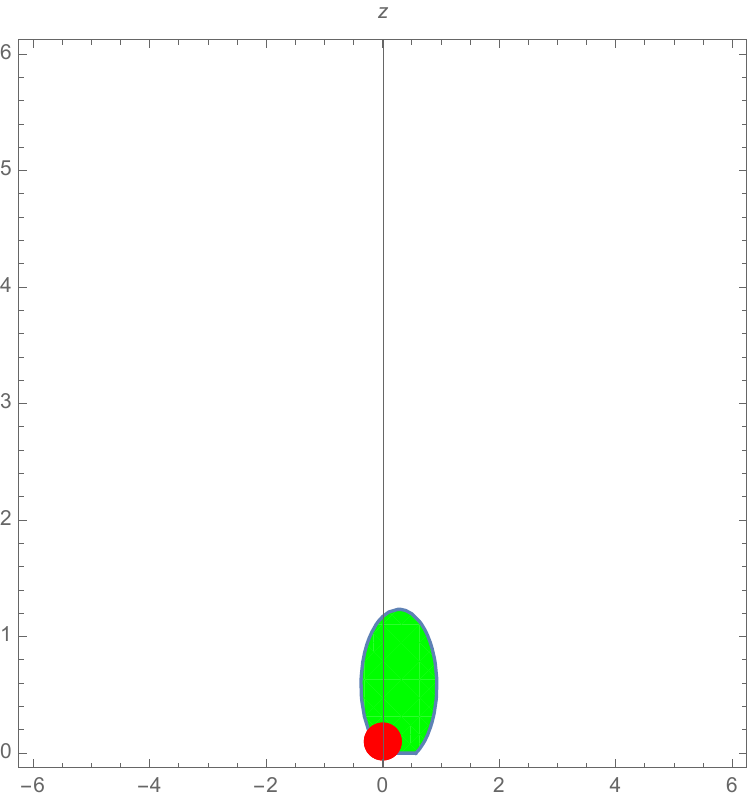}
    \vspace{1cm}
    \includegraphics[width=5cm]{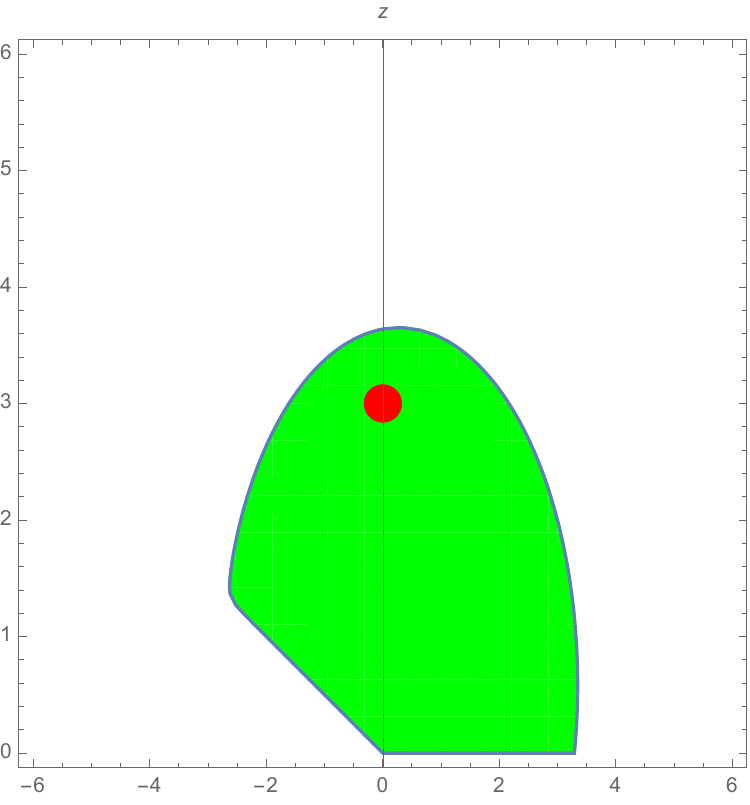}
    \vspace{1cm}
     \includegraphics[width=5cm]{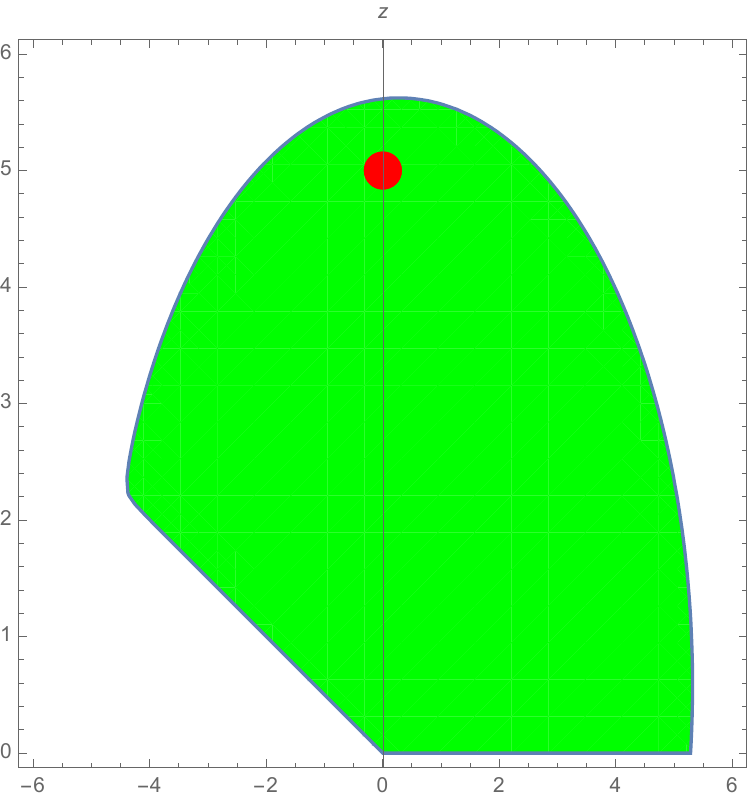}
  \caption{Time evolution of an EOW brane and a massive particle in the Poincar\'e coordinates  with  $\alpha=0.1$, $\chi=0.9$, and $\lambda=2$. The vertical line is $z$ axis (the bulk direction) and the horizontal line is $x$ axis (the boundary spatial direction). The green region represents the gravity dual region of the BCFT  and the red dot represents the position of the massive particle. The EOW brane attaches at the boundary at $x=0$ and $x=Z(t)$. We chose the time to be $t=0,3$ and $5$ in the left, middle and right panels. As time evolves the EOW brane probes deeper in the bulk.}
  \label{timeetw135}
\end{figure}

By applying the map (\ref{mappo}), assuming $0<M<\frac{3}{4}R^2$, %and ${\cal T}>0$ , 
we find that the BCFT dual to the Poincar\'e patch has two boundaries which are the two intersections of  the AdS boundary $z=0$ and the EOW brane: $x=0$ and 
\ba
x=\pm\left( \frac{\ap}{\gamma}+\s{\ap^2\left(1+\frac{1}{\gamma^2}\right)+t^2} \right)\equiv \pm Z(t), \quad Z(t) >0,\label{disco}
\ea
where 
\ba
\gamma=\tan\left(\frac{\pi}{\chi}\right). 
\ea
The $+$ sign corresponds to the $\lambda>0$ case while the $-$ sign corresponds to the $\lambda<0$ case. The BCFT lives on the spacetime defined by $0\le x \le Z(t)$ for $\lambda>0$ and $\{ 0\le x, x \le Z(t) \}$ for $\lambda<0$. In the following discussions, we focus on the $\lambda>0$ case, though, the $\lambda<0$ case is treated in the same manner.

We plot the time evolution of a massive particle and the EOW brane, see Fig.\ref{timeetw135}. We can see that the EOW brane bends toward the conformal boundary due to the massive particle.

In summary, the BCFT lives on the spacetime defined by $0\le x \le Z(t)$ (for $\lambda>0$).
This geometry at the AdS boundary is depicted in the left panel of Fig.\ref{deformationfig}. Note that  the location of the second boundary is time-dependent and at late times it almost expands at the speed of light. As we will see later, this geometry naturally arises from a conformal transformation in a two-dimensional BCFT.

\begin{figure}[h]
  \centering
  \includegraphics[width=8cm]{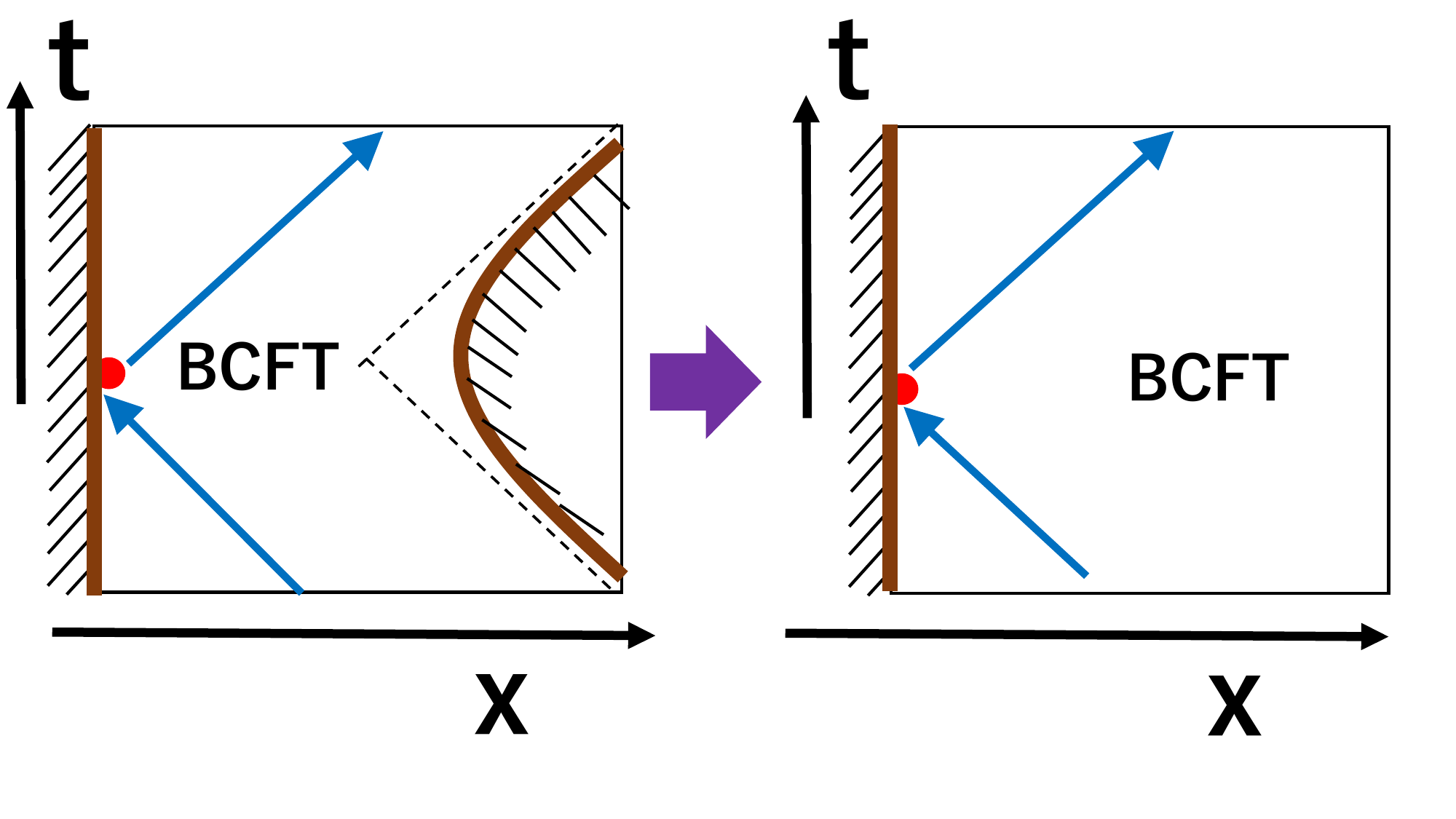}
  \caption{The left picture sketches the setup of BCFT in the presence of the two boundaries: $x=0$ and (\ref{disco}). We may try to remove the right boundary by a coordinate transformation (right picture).}
\label{setupsbdyfig}
\end{figure}

In principle, it is also possible to shift the location of the EOW brane for  $x_0<0$ in (\ref{eomb}).
It is given by the surface
\ba
\frac{\ti{r}}{R}\sin\theta+\frac{x_0}{R\ap} (\s{R^2+\ti{r}^2}\cos\ti{\tau}+\ti{r}\cos\ti{\theta})+\lambda=0.
\ea
This is dual to a local operator quench (\ref{LOS}) at $x_a=-x_0$, by shifting the location of the boundary from $x=x_0$ to $x=0$. Since this gives a complicated time-dependent spacetime, we will not discuss this further. 

On the other hand, for $M>R^2$, the spacetime (\ref{GBmet}) describes a BTZ black hole. We can still analytically continue the expression (\ref{teha}) and (\ref{profti}), we obtain
\ba
r\sinh\left(\frac{\s{M-R^2}}{R}\theta\right)=-\lambda \s{M-R^2}.
\ea
This EOW brane extends from the AdS boundary to the black hole horizon. We will not get into this black hole setup in more detail as our comparison with the CFT result can be done with the deficit angle geometry.

\subsection{Coordinate transformation}
As we have found in the previous section, 
for $0<M<\frac{3}{4}R^2$, the AdS/BCFT setup is given by the boundary surface (\ref{beoms})
in the deficit angle geometry (\ref{GBmet}). Via the coordinate transformation 
(\ref{mappo}), it is mapped into the asymptotically Poincar\'e AdS geometry whose boundaries consist of two segments $x=0$ and $x=Z(t)$ (\ref{disco}) as depicted in the right of Fig.\ref{setupsbdyfig}. The appearance of the second boundary may not be surprising because the massive particle in the center of the global AdS attracts the EOW brane toward the center and this backreaction bends the brane such that its intersection with the AdS boundary gets shifted toward the first boundary $x=0$.  
Originally, however, we have intended a local operator quench in a BCFT on a half place (the right panel of Fig.\ref{setupsbdyfig}), 
instead of the region surrounded by two boundaries 
(the left one of  Fig.\ref{setupsbdyfig}).

To resolve this issue, we would like to perform the following rescaling of the global AdS coordinates:
\ba
\theta'=\eta\theta,\ \ \tau'=\eta\tau,\ \ r'=r/\eta.  \label{rescaleex}
\ea
Applying this to the metric \eqref{GBmet} gives
\ba
ds^2=-(r'^2+R^2-M')d\tau'^2+\frac{R^2}{r'^2+R^2-M'}dr'^2+r'^2d\theta'^2,
\ea
where
\ba
M'=\frac{M}{\eta^2}+R^2\left(1-\eta^{-2}\right).  \label{newmass}
\ea
Note that this is equivalent to 
\begin{equation}
    \chi=\sqrt{\frac{R^2-M}{R^2}}\rightarrow \frac{\chi}{\eta}=\sqrt{\frac{R^2-M^\prime}{R^2}}.
    \label{eq:chi-transform}
\end{equation}
Since the asymptotically AdS region surrounded by the surface $Q$ (\ref{beoms}) is 
given by $0<\theta<\left(2-\frac{1}{\chi}\right)\pi$, if we choose
\ba
\eta\geq \eta_0\equiv \frac{1}{2-{1}/{\chi}},  \label{mineta}
\ea
then the range of the new angular coordinate $\theta'$ takes 
$0<\theta'<\theta'_{max}$ with $\theta'_{max}\geq \pi$  on the asymptotically AdS boundary.\footnote{Notice that even when $\lambda<0$, this rescaling by $\eta\geq\eta_0$ takes the second boundary $\theta=-\left(2-\frac{1}{\chi}\right)\pi$ to $\theta^\prime\leq-\pi$. Thus, by the same coordinate transformation, the BCFT with a negative tension becomes a half space.} When $\eta=\eta_0$, we have $\theta'_{max}=\pi$.

Therefore, if we apply the coordinate transformation\footnote{Here we mean $r'\sin\theta'=\frac{Rt'}{z'}$ for example.} 
(\ref{mappo}) with $(r,\theta,\tau)$ and $(z,x,t)$ replaced with 
$(r',\theta',\tau')$ and $(z',x',t')$ with $\eta\geq\eta_0$, then the resulting asymptotically Poincar\'e AdS, given by the coordinate $(z',x',t')$,
includes only a single boundary $x'=0$ since the Poincar\'e patch only covers the regime $-\pi\le\theta^\prime< \pi$.  While the falling particle trajectory remains the same $z'=\s{t'^2+\ap^2}$, the shape of EOW brane gets modified. In this way,  we can realize the gravity setup dual to the local operator quench on a half plane by choosing $\eta\geq \eta_0$. Especially when $\eta=\eta_0$, the entire BCFT and its gravity dual region is covered even after the Poincar\'e patch. We will justify this correspondence after the rescaling by comparing the bulk calculation with that in the holographic CFT in the subsequent sections.

\section{Holographic stress tensor}\label{sec:m-less-r}

One way to check the validity of the dual CFT interpretation is to compute the holographic stress tensor \eqref{eq:hol-EM-tensor}~\cite{Balasubramanian:1999re}.
 The value of the holographic stress tensor in our setup before we perform the coordinate transformation (\ref{rescaleex})
 is exactly the same as that without the boundaries, which was computed in 
\cite{Nozaki:2013wia}:
\ba
T_{--} (M)=\frac{M\ap^2}{8\pi G_N R\left((t-x)^2+\ap^2\right)^2},\ \ 
T_{++} (M)=\frac{M\ap^2}{8\pi G_N R\left((t+x)^2+\ap^2\right)^2}, \label{eflux}
\ea
where $x^\pm=x\pm t$.
The mass $M$ is related to the mass $m$ of the particle via (\ref{relam}).
One might expect that the conformal dimension of the dual operator $O(x)$ for the local quench is given by 
$mR$ via the familiar correspondence rule. However, this is not completely correct due to the reason we will explain soon later. Instead, we introduce $\Delta_{AdS}$ to distinguish this from the correct conformal dimension $\Delta_O$ of the dual primary operator in the BCFT and write as follows:
\ba
mR\simeq \Delta_{AdS}.  \label{dimfo}
\ea

Indeed, the result of energy fluxes may look confusing at first because the holographic expression  (\ref{eflux}) looks identical to that without any EOW brane inserted. For a local operator quench in the presence of a boundary, we actually expect that the energy fluxes will be doubled due to the mirror charge effect as explained in Fig.\ref{doublefluxfig}. As we will confirm from the CFT calculation later, even though the boundary $x=0$ produces the doubled flux, the presence of the other boundary (\ref{disco}) reduces the energy fluxes, which is analogous to the Casimir effect.

This also gives another motivation for performing the previous coordinate transformation (\ref{rescaleex}) as this removes the extra boundary from the dual BCFT.  The energy flux after the transformation is simply given by (\ref{eflux}) with $M$ replaced with $M'$ in (\ref{newmass}). This leads to 
a class of asymptotically Poincar\'e AdS solutions with an EOW brane which is specified by the two parameters 
$M$ and $\eta$. The energy flux is obtained by replacing $M$ in (\ref{eflux}) with $M'$ given by (\ref{newmass}), which is a monotonically increasing function of $\eta$.  Note that 
$M$ is related to $\Delta_{AdS}$ via (\ref{dimfo}) and (\ref{relam}):\footnote{It is useful to note that $\chi=\sqrt{\frac{R^2-M}{R^2}}=\sqrt{1-12\frac{\Delta_{AdS}}{c}}$ from \eqref{relationm}.}
\ba
\frac{M}{R^2}=12\frac{\Delta_{AdS}}{c}, \label{relationm}
\ea
which leads to the following holographic stress tensor in the BCFT language:
\ba
&& T_{\pm\pm} (M^\prime)=s_{AdS}(\Delta_{AdS},\eta)\cdot\frac{\ap^2 }{\pi \left((t\pm x)^2+\ap^2\right)^2},\no
&& s_{AdS}(\Delta_{AdS},\eta)\equiv\frac{\Delta_{AdS}}{\eta^2}+\frac{c}{12}
\left(1-\frac{1}{\eta^2}\right).
\label{efluxx}
\ea
It is also useful to note that the energy density $T_{tt}$ is given by\footnote{In previous chapters, we denoted the energy density by $T_{00}$. In this chapter, we denote it by $T_{tt}$ to emphasize that it is the energy with respect to time in global coordinates.} 
\ba
&& T_{tt}=T_{++}+T_{--}=2s_{AdS}(\Delta_{AdS},\eta)\cdot H(t,x),\no
&& H(t,x)=\frac{\ap^2\left((t^2+x^2+\ap^2)^2+4t^2x^2\right)}{\pi\left((x^2-t^2-\ap^2)^2+4\ap^2x^2\right)^2}.\label{energyd}
\ea

On the other hand, the other parameter $\eta$ describes the degrees of freedom of gravitational excitations dual to those of descendants in the BCFT and this affects the shape of EOW brane. In other words, $\eta$ corresponds to a conformal transformation.
Since the cylinder coordinates and the plane coordinates are related by the conformal transformation 
(\ref{confmap}), changing $\eta$ induces further conformal transformation:
\ba
t'\pm x'=\ap\tan\left(\frac{\tau'\pm \theta'}{2}\right)=\ap\tan\left(\frac{\eta(\tau\pm \theta)}{2}\right).
\label{conf-map-eta}
\ea
The stress tensor is generated from this via the Schwarzian derivative term.
This is dual to the descendant excitations of the two-dimensional CFT. Note that if the two solutions $(M_1,\eta_1)$ and $(M_2,\eta_2)$ satisfies
\ba
\frac{R^2-M_1}{\eta_1^2}=\frac{R^2-M_2}{\eta_2^2},\label{relpa}
\ea
or equally $\chi_1/\eta_1=\chi_2/\eta_2$, then 
the stress tensors become identical. This means that the asymptotic metric near  the AdS boundary is the same.  However they are different globally,  because of the different deficit angle due to the massive particle.  

In particular, when $\eta=\eta_0$ (\ref{mineta}), 
the energy flux gets minimized among those dual to a geometry with a single boundary $x=0$ and thus we expect $\eta=\eta_0$ corresponds to the BCFT with a local operator excitation without any other excitations. In this case, the stress tensor is given by (\ref{efluxx}) with
$s_{AdS}(\Delta_{AdS},\eta_0)$ takes 
\ba
s_{AdS}(\Delta_{AdS},\eta_0)=\frac{c}{3}\s{1-\frac{12\Delta_{AdS}}{c}}
\left(1-\s{1-\frac{12\Delta_{AdS}}{c}}\right).
\label{efluxy}
\ea
We argue that in this case $\eta=\eta_0$, $s_{AdS}$ is directly related to the conformal dimension of 
the primary $O(x)$ in the dual BCFT by
\ba
s_{AdS}(\Delta_{AdS},\eta_0)=2\Delta_O.
\label{eq:AdS EM tensor}
\ea
This is because this setup is the BCFT defined on the right half plane $x>0$ and the energy flux should be simply twice of that in the same CFT without any boundary. We will show that this is true from the explicit CFT calculation in \S\ref{sec:EMtensor-holCFT}. It is also useful to note that 
this relation between $\Delta_O$ and $\Delta_{AdS}$ can be expressed as follows: 
\ba
2\s{1-\frac{12\Delta_{AdS}}{c}}-1=\s{1-\frac{24\Delta_O}{c}},\label{relasd}
\ea
and when $\Delta_O,\Delta_{AdS}\ll c$, we find 
\ba
\frac{\Delta_{AdS}}{c}=\frac{\Delta_O}{c}+3\left(\frac{\Delta_O}{c}\right)^2+\ddd.
\ea
It is also helpful to note that the standard range of primary operator below the black hole threshold given by $0<\Delta_O<\frac{c}{24}$ corresponds\footnote{Notice that $\Delta_O$ is the total conformal dimension i.e. the sum of chiral and anti-chiral conformal dimension.} to the range $\frac{1}{2}<\s{1-\frac{12\Delta_{AdS}}{c}}=\chi<1$. This is consistent with the previous observation\footnote{This bound is equivalent to $\Delta_{AdS}<c/16$ given in \cite{Geng:2021iyq} by replacing $\Delta_{\mathrm{bcc}}$ in their paper with $\Delta_{AdS}$.} that the EOW brane configuration makes sense only when $0<M<\frac{3}{4}R^2$  in order to avoid the self-intersection of the brane (folding problem).\footnote{The self-intersection problem also appears in higher dimensions \cite{Fallows:2022ioc}. Although our analysis focus on $d=2$, a similar analysis may circumvent the problem likewise.}
One important difference, which has been confused previously, is that this bound $\Delta_O<\frac{c}{24}$ is exactly equivalent to the black hole threshold for $2\Delta_O$ through the correct relation \eqref{relasd}. Indeed, the dimension gets actually doubled in the presence of the boundary due to the mirror effect as we will explain from the BCFT viewpoint in section 4. Notice also that we can in principle extend to the heavier excitation with $\Delta_O\geq \frac{c}{24}$ by analytically continuing the formula (\ref{relasd}), where $\Delta_{AdS}$ and therefore the mass $M$ gets complex valued.

\begin{figure}
  \centering
  \includegraphics[width=8cm]{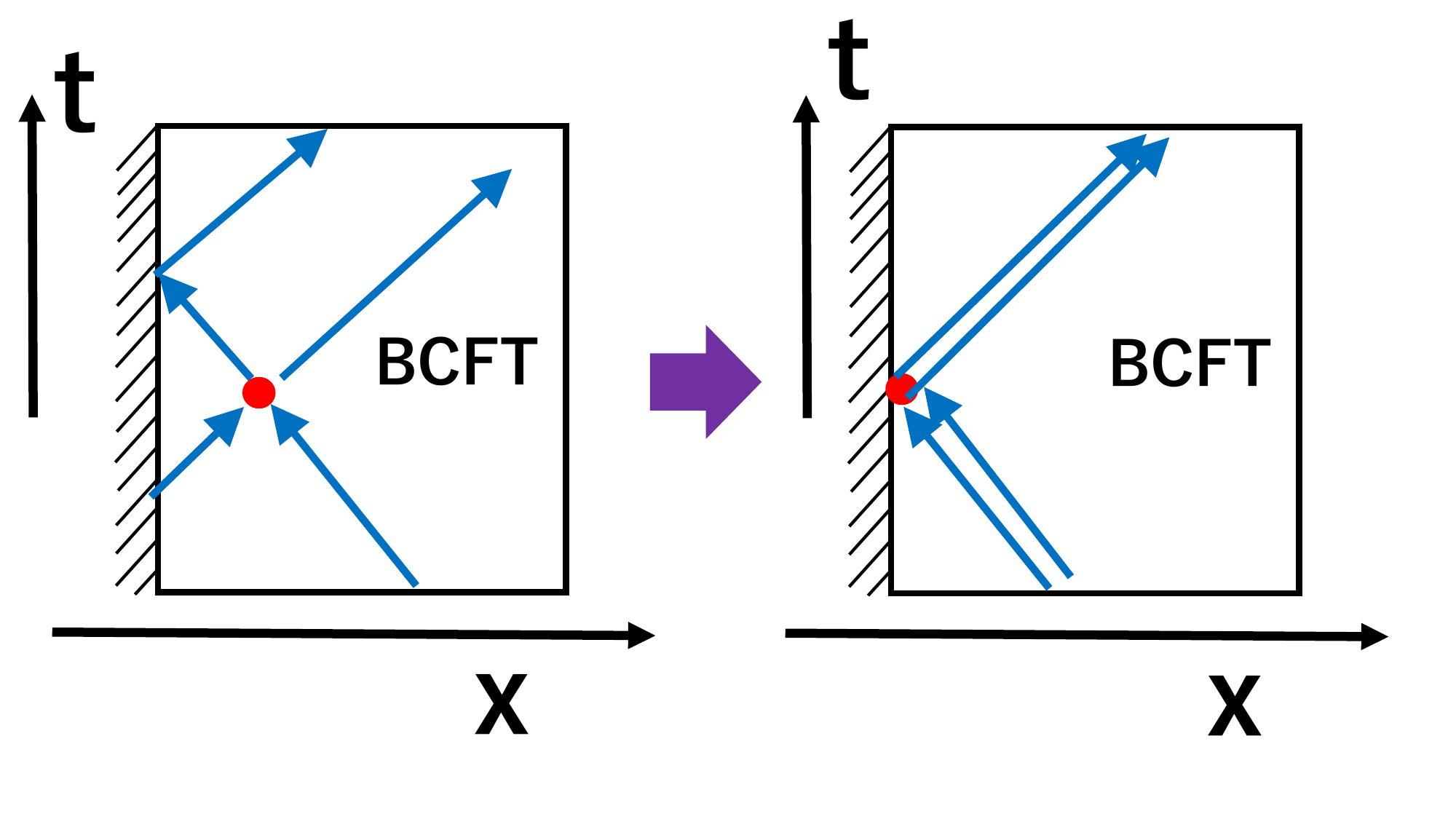}
  \caption{The energy fluxes from an excitation (red point) in the presence of 
  a boundary. When the excitation coincides with the boundary, the energy flux is doubled as in the right picture.}
\label{doublefluxfig}
\end{figure}

\section{Time evolution of holographic entanglement entropy}\label{sec:holoEE}

In this section, we calculate the time evolution of the holographic entanglement entropy in our gravity dual of the local operator quench in the BCFT. We can analytically obtain results since the holographic setup is related to a global AdS$_3$ with an EOW brane and a deficit angle through the chain of coordinate transformations, as we saw in the previous section. 

Our goal is to calculate the holographic entanglement entropy in the asymptotically Poincar\'e AdS$_3$ background. On this boundary, we define a subsystem $I$ to be an interval $I$ between two points $A$ and $B$ at each time $t$. We write the spatial coordinate of $A$ and $B$ as $x_A$ and $x_B$, respectively, assuming $x_{A}< x_{B}$. Then we follow the time evolution of the entanglement entropy $S_{AB}$ in our BCFT as a function of the boundary time $t$. 

The entanglement entropy $S_{AB}$ is defined as usual by first introducing the (time-dependent) reduced density matrix $\rho_{AB}(t)$ by tracing out the complement $I^c$ of the interval $I=[x_A,x_B]$ from the density matrix for the operator local quench state (\ref{LOS}):
\ba
\rho_{AB}(t)=\mbox{Tr}_{I^c}\left[|\Psi(t)\lb\la\Psi(t)|\right].
\label{eq:reduced-rho}
\ea
The entanglement entropy is defined by the von-Neumann entropy as a function of time $t$:
\ba
S_{AB}(t)=-\mbox{Tr}[\rho_{AB}(t)\log\rho_{AB}(t)].
\ea

\subsection{Holographic entanglement entropy in AdS/BCFT}
As we have reviewed in the latter part of Section \ref{sec:hol-CFT2},
the holographic entanglement entropy in AdS$_3/$CFT$_2$ is computed by the length of geodesics which anchor the boundary points $A:(z=\epsilon,x=x_{A},t)$ and $B:(z=\epsilon,x=x_{B},t)$, where $\epsilon$ is the UV cutoff surface~\cite{Ryu:2006bv,Ryu:2006ef,Hubeny:2007xt}.
We discussed in Section \ref{sec:HEE-BCFT} that the formula is modified in the AdS/BCFT setup so that the geodesics can end on the EOW brane. There are two phases and we need to take the minimum of those~\eqref{eq:HEE-BCFT} (Fig.\ref{HEEfig}).

Below we separately compute these two contributions, then find the actual value of the entropy.  We initially  do not perform the rescaling (\ref{rescaleex}), i.e. we set $\eta=1$, where the gravity dual is given by the BCFT on $0<x<Z(t)$, where 
$Z(t)$ is given in (\ref{disco}). Later, we will extend our analysis to $\eta>1$ case.

\begin{figure} 
  \centering
  \includegraphics[width=6cm]{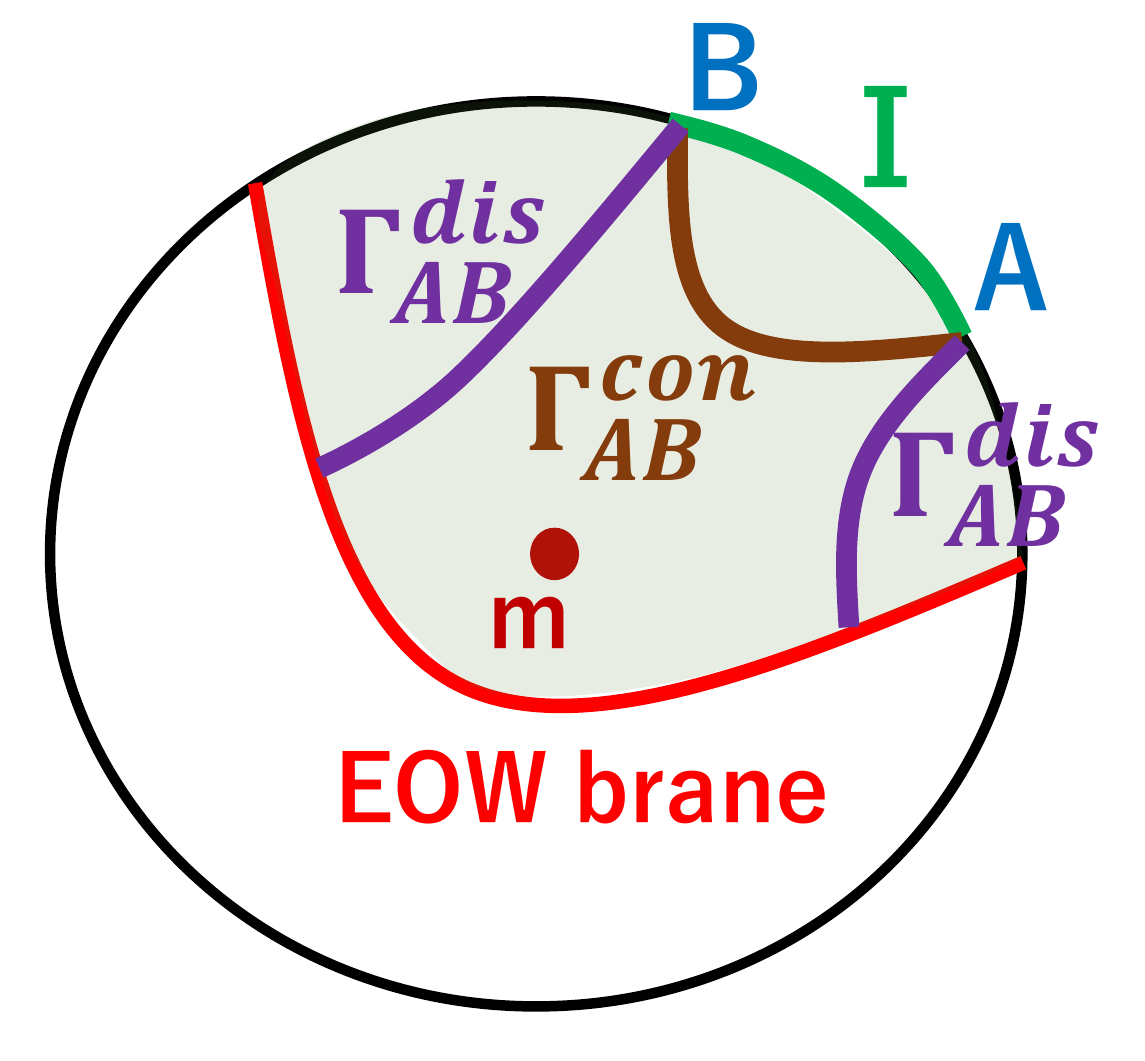}
  \caption{A sketch of calculation of holographic entanglement entropy in AdS/BCFT with local quench. We showed that the connected geodesic $\Gamma^{con}_{AB}$ (brown curve) and the disconnected geodesics $\Gamma^{dis}_{AB}$ (purple curves) for a boundary subsystem given by an interval $I$ between $A$ amd $B$.}
\label{HEEfig}
\end{figure}

\subsection{Connected entropy}\label{sec:connected-ent}

First, we would like to calculate the connected entropy $S^{con}_{AB}$.
As we review in Appendix \ref{sec:Embedd}, an efficient  way to compute the length of a geodesic in $AdS_{3}$ is using its embedding to flat space $\mathbb{R}^{2,2}$, where we derive the detailed formula. Below we show only the  results. The geodesic length connecting two boundary points $ A: (\theta_{A}, \tau_{A}, r_{A})$  and $B: (\theta_{B}, \tau_{B}, r_{B})$ in the metric \eqref{GBmet} is obtained by rescaling  (\ref{AppendixA EE}) by \eqref{teha}. The corresponding holographic entanglement entropy reads
\begin{equation}
    S^{con}_{AB} = \frac{c}{6}\log{\left[ \frac{2r_A r_B}{R^2 \chi^2}(\cos(\chi(\tau_A-\tau_B))-\cos(\chi(\theta_A-\theta_B))\right]},
    \label{connected HEE}
\end{equation}
where
\ba
\chi = \sqrt{\frac{R^2-M^2}{R^2}}= \s{1-\frac{12\Delta_{AdS}}{c}},
\ea
as we have introduced \eqref{teha} and $c$ is the central charge $c= \frac{3R}{2G}$ of the dual CFT \eqref{eq:brown-henneaux}.

Although we  presented the formula for the holographic entanglement entropy in global coordinates, ultimately we are interested in its expression in the Poincar\'e coordinates where the setup of the local quench is  introduced.  By restricting the coordinate transformation (\ref{mappo}) at the AdS boundary $z=\ep$, we obtain the boundary map:
\ba
&& e^{i\tau}=\frac{\ap^2+x^2-t^2+2i\ap t}{\s{(x^2-t^2)^2+2\ap^2(t^2+x^2)+\ap^4}},\no
&& e^{i\theta}=\frac{\ap^2-x^2+t^2+2i\ap x}{\s{(x^2-t^2)^2+2\ap^2(t^2+x^2)+\ap^4}},\no
&& r=\frac{R}{2\ap \ep}\s{(x^2-t^2)^2+2\ap^2(t^2+x^2)+\ap^4}. \label{thrmap}
\ea
Mapping the end points $A=(x_{A},t_{A},\epsilon)$ and $B$ in the Poincar\'e coordinates by the above transformation to those in the global AdS coordinates, we can calculate the geodesic length from the formula 
(\ref{connected HEE}). Note that we need to choose $A$ and $B$ within the region of BCFT i.e. $0<x_A<x_B<Z(t)$. Refer to Fig.\ref{fig:regionplot} for an explicit plot.

\begin{figure}[t]
    \centering
    \includegraphics[width = 60mm]{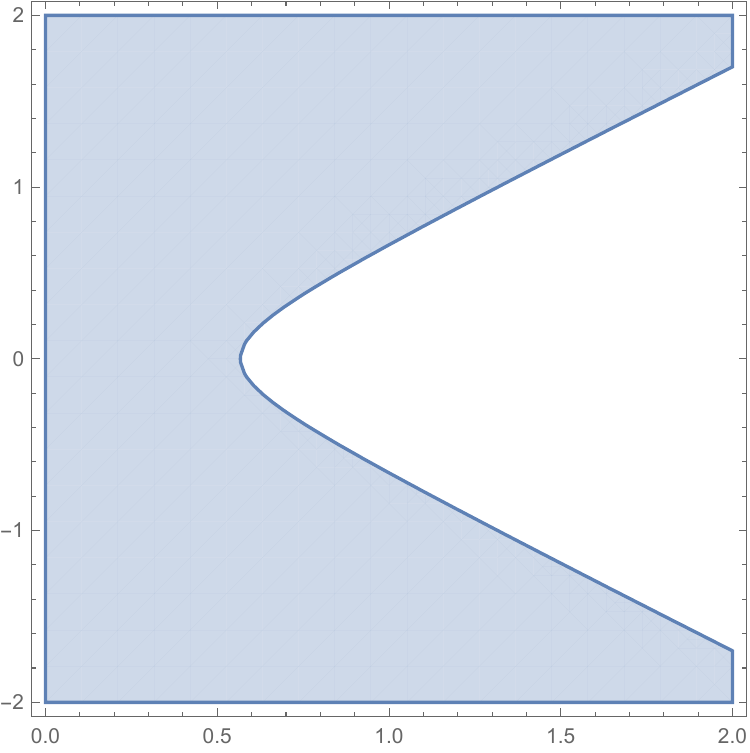}
    \hspace{1cm}
    \includegraphics[width = 60mm]{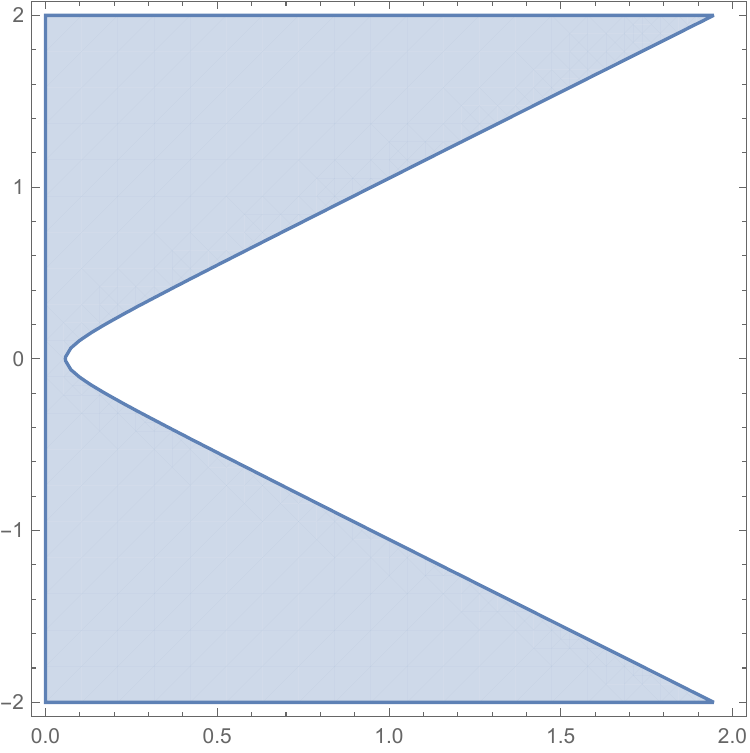}
    \caption{Plot of the physical region $0<x<Z(t)$ of our BCFT in $(x,t)$ plane. We chose $\chi=0.9$ in the left plot and $\chi=0.6$ in the right plot.
    Again we chose $R=1$ and $\ap=0.1$.}
    \label{fig:regionplot}
\end{figure}

We plot the time evolution of the connected entanglement entropy with a given interval for $\chi=0.9$ and $\chi=0.6$ in Fig.\ref{fig:connected entropy}. 
We can see that there is a peak when the shock wave from the falling particle hits the center of the interval $I$ i.e. $\s{\ap^2+t^2}\simeq \frac{x_A+x_B}{2}$ \cite{Nozaki:2013wia}. Moreover the final value is given by the vacuum entanglement entropy \cite{Holzhey:1994we}
\begin{equation}
    S_{AB}= \frac{c}{3}\log\left[ \frac{x_B-x_A}{\epsilon}\right].
\end{equation}

It is also useful to examine the first law of entanglement entropy \cite{Bhattacharya:2012mi,Blanco:2013joa}, which states that the growth of the entanglement entropy $\Delta S_{AB}$, defined by the difference of the entanglement entropy of an excited state and that of the CFT vacuum, is directly related to the energy density $T_{tt}$ in the small subsystem limit $|x_A-x_B|\to 0$ via
\ba
T_{tt}(x_A,t)=\lim_{|x_A-x_B|\to 0}\frac{3 }{\pi|x_A-x_B|^2}\cdot\Delta S_{AB}(x_A,x_B,t).  \label{firstL}
\ea
In this short subsystem limit, we find after some algebra that our holographic entanglement entropy (\ref{connected HEE}) behaves as 
\ba
&& \Delta S^{con}_{AB}=\frac{c}{6}\log{\left[ \frac{\cos(\chi(\tau_A-\tau_B))-\cos(\chi(\theta_A-\theta_B))}{\chi^2\left(\cos(\tau_A-\tau_B)-\cos(\theta_A-\theta_B)\right)}\right]}\no
&&\simeq \frac{c}{18}(1-\chi^2)H(t,x_A)(x_A-x_B)^2,
\ea
where $H(t,x)$ is defined in (\ref{energyd}). Thus, we can confirm that the first law (\ref{firstL}) perfectly reproduces the stress tensor (\ref{efluxx}) and (\ref{energyd}) at $\eta=1$.

\begin{figure}[h]
    \centering
    \includegraphics[width = 70mm]{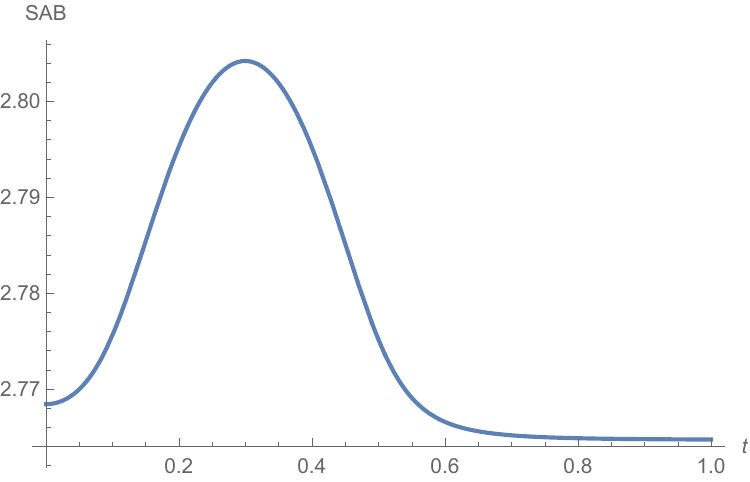}
    \hspace{1cm}
    \includegraphics[width = 70mm]{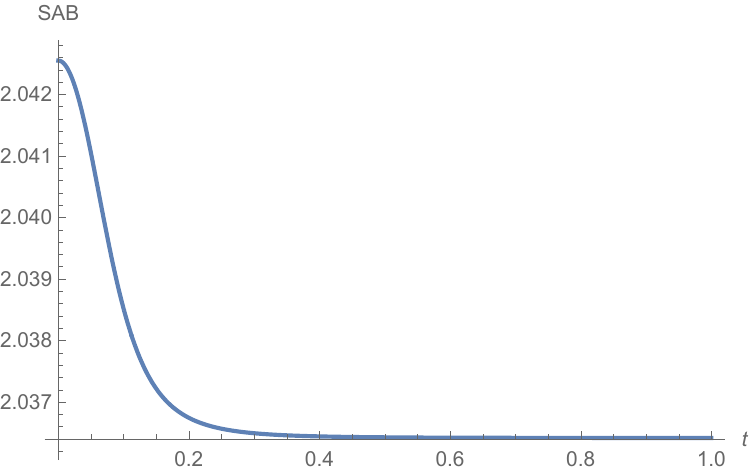}
    \caption{Time evolution of the connected entanglement entropy $S^{con}_{AB}$. In the left plot, we chose $\chi=0.9$ and $(x_A,x_B)=(0.1,0.5)$. In the right one, we chose $\chi=0.6$ and $(x_A,x_B)=(0.005,0.05)$. 
    In both, we took $R=c=1$, $\ep=0.0001$ and $\ap=0.1$.
    }
    \label{fig:connected entropy}
\end{figure}

\subsection{Disconnected entropy}

Next, we will consider the disconnected contribution $S^{dis}_{AB}$. We will work in the tilde coordinates again. This time, we consider the geodesics that stick to the EOW brane perpendicularly. The formula for the disconnected entropy is obtained in the tilde Poincar\'e coordinates, where the EOW brane is simply given by 
the plane $\tilde{x}=-\lambda \tilde{z}$. In %the latter 
this tilde coordinates, we can employ the known result \cite{Takayanagi2011,Fujita:2011fp}
\ba
S^{dis}_{AB}=\frac{c}{6}\log\frac{2l}{\ep}+S_{bdy},
\ea
for each of the disconnected geodesics, where $S_{bdy}$ is the boundary entropy given by 
$ \frac{c}{6}\sinh^{-1}\lambda$ \eqref{eq:bdy-ent-formula}. 
However, we have to keep in mind that in this tilde Poincar\'e coordinates we need to care about the periodicity of $\theta$ due to the conical defect. A careful consideration results in the following expression,
\begin{equation}\label{disconnected HEE}
  S^{dis}_{AB}= \frac{c}{6} \log{\left(\frac{2r_A}{R\chi} \sin(\chi\theta_A^{\text{min}})\right)}+\frac{c}{6} \log{\left(\frac{2r_B  }{R\chi} \sin(\chi\theta_B^{\text{min}})\right)}+\frac{c}{3}\sinh^{-1}\lambda,
\end{equation}
where $\theta^{\text{min}}$ is defined as the smaller angle measured from $x=0$ or $x=Z(t)$:
\begin{equation}
    \theta^{\text{min}}  = \min{\left[\theta,\left(2-\frac{1}{\chi}\right)\pi - \theta\right] }.
\end{equation}
We note that we should take the minimum of $\theta$ because in the disconnected case we have two extremal values of the geodesics due to the deficit angle at the center.

The resulting holographic entanglement entropy is plotted in Fig.\ref{fig:Sdis1} ($\chi=0.9$) and Fig.\ref{fig:Sdis2} ($\chi=0.6$). Under the time evolution, $S^{dis}_{AB}$ gets initially increasing since the EOW brane extends toward the inner region and 
the disconnected geodesics get longer. We can also note a peak around the time $t\simeq \s{x_B^2-\ap^2}$ since the falling particle crosses the disconnected geodesics which extend from $B$.
However, the actual holographic entanglement entropy is dominated by $S^{con}_{AB}$ after some critical time. In early time regime, $S^{dis}_{AB}$ is smaller and this is analogous to the setup of global quantum quenches
\cite{Calabrese:2005in}. We can also understand the plots of holographic entanglement entropy with respect to the subsystem size. It grows initially and reaches a maximum in a middle point, during which $S^{con}_{AB}$ dominates. After that it starts decreasing, and is eventually dominated by $S^{dis}_{AB}$. If we choose $x_A=0$, then we will end up with $S^{dis}_{AB}=0$ and this may be a sort of the Page curve behavior, because the total system is defined as an interval $0<x<Z(t)$.

\begin{figure}[h]
    \centering
    \includegraphics[width = 70mm]{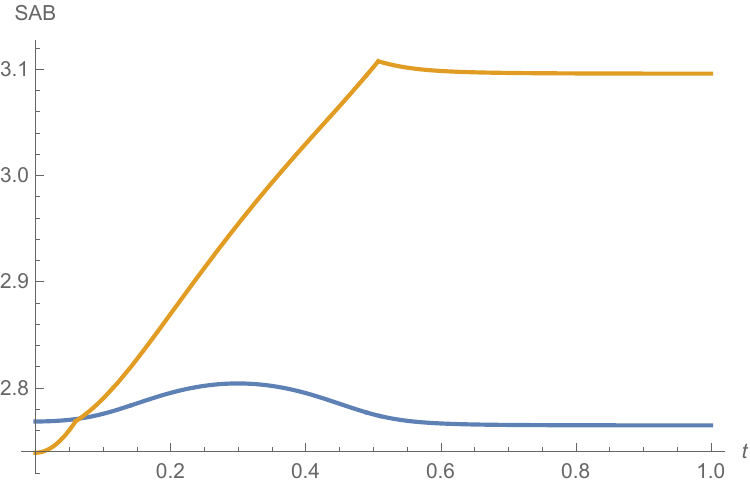}
     \hspace{0.5cm}
    \includegraphics[width = 70mm]{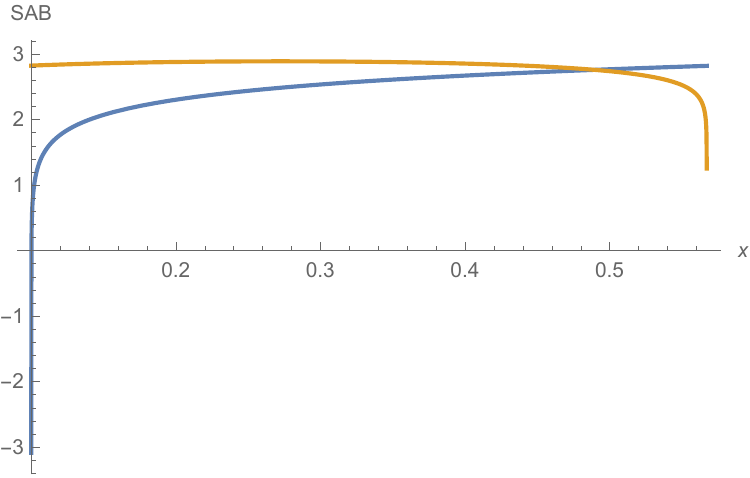}
    \caption{Plots of connected entropy $S^{con}_{AB}$ (blue curves) and disconnected entropy $S^{dis}_{AB}$ (orange curves) for $\chi=0.9$. The left panel shows the time evolution of them for the interval $(x_A,x_B)=(0.1,0.5)$.
    The right one describes their behaviors as functions of 
    $x$ when we chose the subsystem to be $(x_A,x_B)=(0.1,x)$ at $t=0$. In both, we took $R=c=1$, $\ep=0.0001$, $\lambda=1$ and $\ap=0.1$.
    }
    \label{fig:Sdis1}
\end{figure}

\begin{figure}[h]
    \centering
    \includegraphics[width = 70mm]{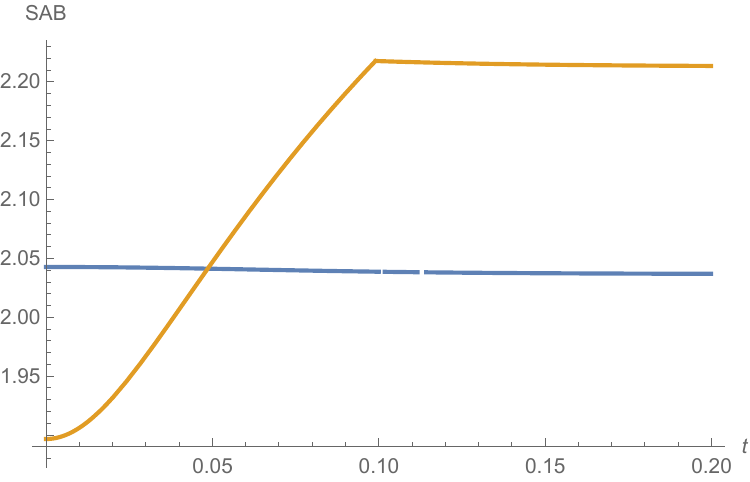}
    \hspace{0.5cm}
    \includegraphics[width = 70mm]{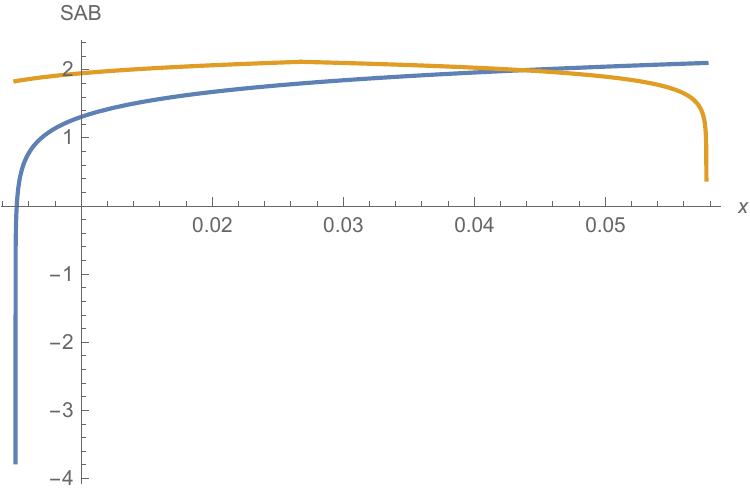}
    \caption{Plots of connected entropy $S^{con}_{AB}$ (blue curves) and disconnected entropy $S^{dis}_{AB}$ (orange curves) for $\chi=0.6$. The left panel shows the time evolution of them for the interval $(x_A,x_B)=(0.005,0.05)$.
    The right one describes their behaviors as functions of 
    $x$ when we chose the subsystem to be $(x_A,x_B)=(0.005,x)$ at $t=0$. In both, we take $R=c=1$, $\ep=0.0001$, $\lambda=1$ and $\ap=0.1$.}
    \label{fig:Sdis2}
\end{figure}

\subsection{Consideration of the parameter $\eta$}

As we have seen, We can introduce the parameter $\eta$ in addition to the mass parameter $M$, via the coordinate transformation
 (\ref{rescaleex}). We can calculate the holographic entanglement entropy with $\eta\neq 1$ by shifting $\chi$ into $\chi\vert_{M\rightarrow M^\prime} = \chi/\eta$ in (\ref{connected HEE}) and (\ref{disconnected HEE}). This allows us to realize a gravity dual of local operator quench on a half plane $x>0$, by pushing the second boundary to $x=\infty$. For example, it is straightforward to confirm that the first law relation (\ref{firstL}) perfectly reproduces the stress tensor (\ref{efluxx}) for any $\eta$.
 
 When  $\eta=\eta_0$, we plotted the behavior of the holographic entanglement entropy in Fig.\ref{fig:etaonee}. This setup is expected to be dual to the BCFT only with the excitation by a local operator at $x=0$ and $t=0$.  The profile of $S^{con}_{AB}$ is qualitatively similar to that at $\eta=1$. On the other hand, $S^{dis}_{AB}$ at $\eta=\eta_0$ has two peaks at $t=\s{x_A^2-\ap^2}$
 and  $t=\s{x_B^2-\ap^2}$, which is because the falling massive particle crosses each of the disconnected geodesic. By taking the minimum,  $S^{dis}_{AB}$ is favored. This matches the BCFT dual because one part of the entangled pair created by the local excitation is reflected at the boundary $x=0$ and merges with the other of the pair. Since both parts come together to the subsystem, the entanglement entropy for the subsystem does not increase except that there is a width $\ap$ of flux of excitations. This means that the entanglement entropy can increase only at the end points $A$ and $B$ for a short time of order $\ap$. This fits nicely with the behavior of $S^{dis}_{AB}$.

\begin{figure}[ttt]
    \centering
    \includegraphics[width = 70mm]{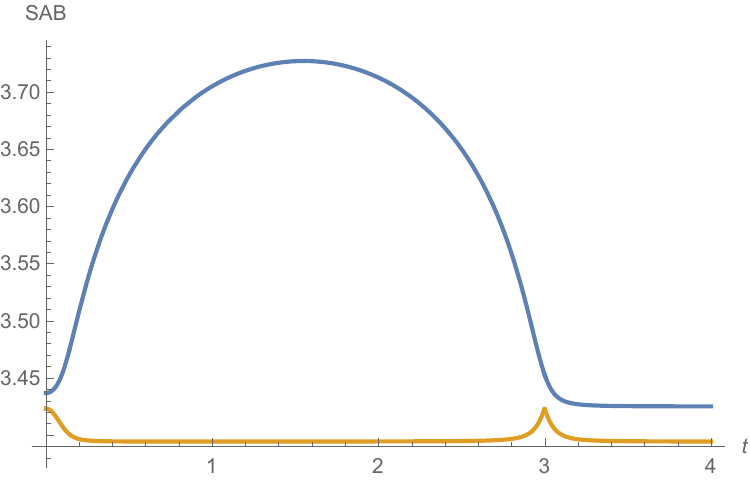}
    \hspace{0.5cm}
    \includegraphics[width = 70mm]{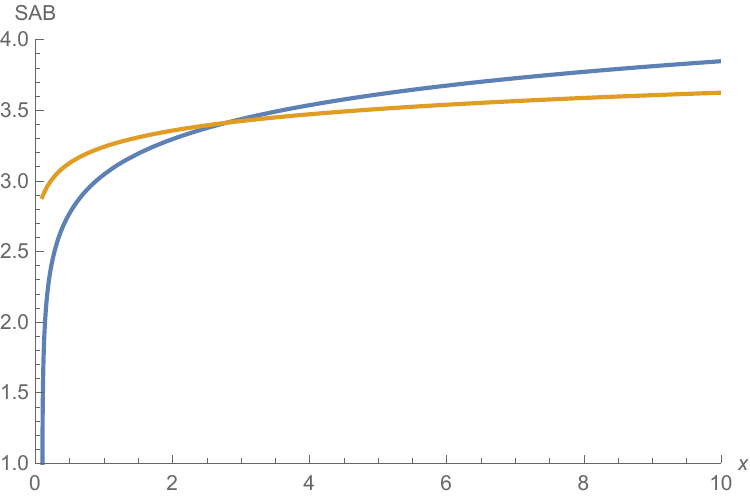}
    \caption{Plots of $S^{con}_{AB}$ (blue) and $S^{dis}_{AB}$ (orange)   for $\eta=\eta_0(=1.125)$ and $\chi=0.9$.    The left panel shows the time evolution of them for the interval $(x_A,x_B)=(0.1,3)$.
    The right one describes their behaviors as functions of 
    $x$ when we chose the subsystem to be $(x_A,x_B)=(0.1,x)$ at $t=0$. In both, we took $R=c=1$, $\ep=0.0001$, $\lambda=1$ and $\ap=0.1$.}
    \label{fig:etaonee}
\end{figure}

As we noted in (\ref{relpa}), there is one parameter family of $(M,\eta)$ which gives the same stress tensor.  As such an example, we can consider the case $\chi=0.99$ and $\eta=1.1\eta_0$, which should have the same energy flux in the case $\chi=0.9$ and $\eta=\eta_0$. Indeed, we plotted in Fig.\ref{fig:etaoneea}, the connected entropy $S^{con}_{AB}$ is precisely identical to that in Fig.\ref{fig:etaonee}.
This confirms that the three-dimensional metric in a neighborhood of the AdS boundary coincides. However, this is not actually physically equivalent because the global structure, especially the monodromy around the massive particle is different. This is simply because the mass of the particle is determined by $M$.  Therefore, if we consider a geodesic which goes around the particle, its geodesic length depends on the value of $M$. On the other hand if we consider a geodesic which does not go around the particle, its length does not depend on $M$. This monodromy clearly affects only the disconnected geodesics. This explains the reason why $S^{dis}_{AB}$ in Fig.\ref{fig:etaoneea} is more enhanced than that in  Fig.\ref{fig:etaonee}. In the latter, the disconnected geodesic length is reduced due to the larger deficit angle.
By taking a minimum between them, the resulting holographic entanglement entropy gets largely increased in the early time regime. In the BCFT side, this enhancement is due to the descendant excitations.

\begin{figure}[ttt]
    \centering
    \includegraphics[width = 70mm]{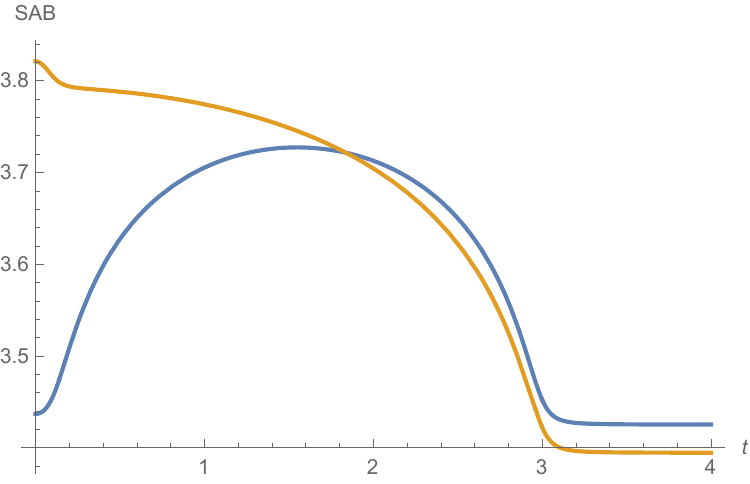}
    \hspace{0.5cm}
    \includegraphics[width = 70mm]{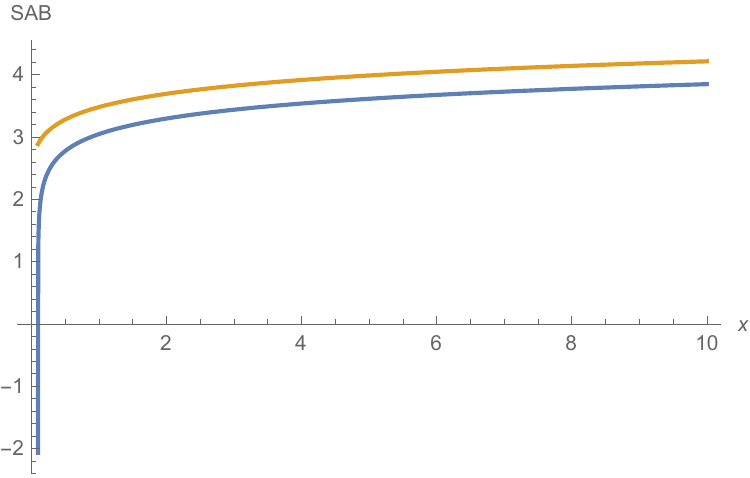}
    \caption{Plots of $S^{con}_{AB}$ (blue) and $S^{dis}_{AB}$ (orange)   for $\eta=1.1\eta_0(\simeq 1.12375)$ and $\chi=0.99$.  The graph of 
    $S^{con}_{AB}$ is identical to that in Fig.\ref{fig:etaonee}, while $S^{dis}_{AB}$ is not.  The left panel shows the time evolution of them for the interval $(x_A,x_B)=(0.1,3)$. 
    The right one describes their behaviors as functions of 
    $x$ when we chose the subsystem to be $(x_A,x_B)=(0.1,x)$ at $t=0$. In both, we took $R=c=1$, $\ep=0.0001$, $\lambda=1$ and $\ap=0.1$.}
    \label{fig:etaoneea}
\end{figure}

 \section{The BCFT calculation -- stress tensor}\label{sec:BCFT-EM}
 
 In the last section we studied  the local quench process in the presence of boundaries through the dual gravitational setup. 
 In this section, we study the same quench process by purely boundary  CFT means. In particular, we perform the BCFT calculations of the stress tensor.
 In the calculations of the previous section, we made use of the fact that the bulk geometry with an infalling particle is related to the pure $AdS_{3}$  via a simple diffeomorphism. In this section and the next we study the quantities of our interest directly in the BCFT side by relating the  Euclidean setup with the time-dependent boundary to the upper half plane via a conformal mapping.  This makes explicit the one-to-one correspondence between  each step of calculations on the gravity side and the CFT side.

 Since the BCFT setup involves an additional boundary $x=Z(t)$ \eqref{disco}, we consider the BCFT between these two boundaries. This setup is  related to the upper half plane (UHP), as we summarize in Fig.\ref{setupsbdyfig}). To see this, we conformally map the original region %, which is a semi-disc-like region (SD) in the Euclidean signature,
 to a wedge-like region (the third figure in Fig.\ref{setupsbdyfig}), which is further mapped to the UHP. %We denote $\alpha=R e^\beta$ in the following calculations.
 
 \begin{figure}
    \centering
    \hspace*{-1.8cm}
    \includegraphics[width = 1.2\linewidth]{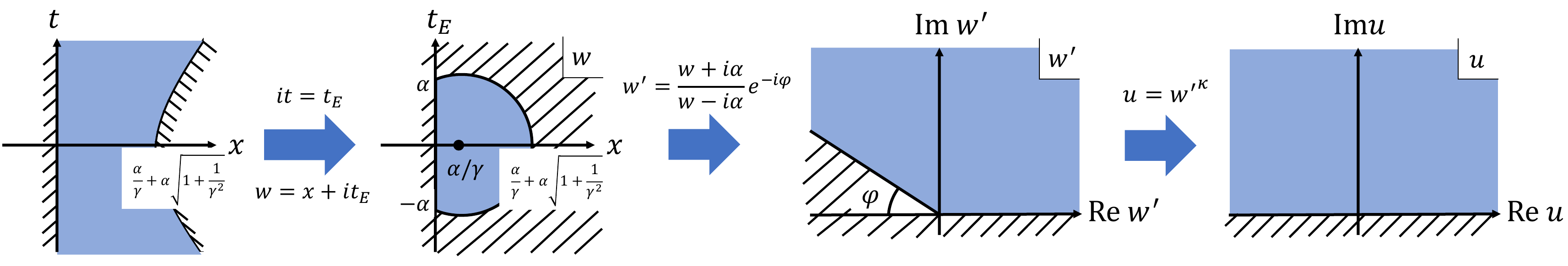}
    \caption{A sequence of conformal maps from the original semi-disk region to the upper half plane: The leftmost figure shows the original Lorentzian region of the BCFT. The second leftmost figure shows the corresponding region on the $w$-plane after Wick rotation. The third figure from the left shows the wedge-like region on the $w^\prime$-plane obtained by a global conformal transformation. The rightmost figure shows the upper half plane (UHP) on the $u$-plane obtained by a further local conformal transformation. A similar sequence of conformal transformations is discussed in \cite{Geng:2021iyq}.}
    \label{fig:conf-map}
\end{figure}
 
\subsection{Conformal map to the upper half plane}
%\color{red}Request: Since I want to use $u$ as a complex variable, I want to write the ingoing/outgoing null coordinates as $x^{+}$ and $x^{-}$ instead of $v=t+x$ and $u=t-x$ in \eqref{eflux}\color{black}

The region on which we define the BCFT is surrounded by two boundaries. The left boundary is static
\ba
x = 0,
\ea
whereas the right boundary $x=Z(t)$ is time dependent, given by
\ba
\left(x-\frac{\alpha}{\gamma}\right)^2 = \alpha^2\left(1+\frac{1}{\gamma^2}\right)+t^2.
\label{bcft-region1}
\ea
This corresponds to the leftmost figure in Fig.\ref{fig:conf-map}.

A BCFT  on this region can be efficiently studied by mapping the region to the UHP. To do this, we perform the Wick rotation $it=t_E$ to make the hyperbolic boundary into a semi-circle. \eqref{bcft-region1} is now given by
\ba
\left(x-\frac{\alpha}{\gamma}\right)^2+t_E^2 =
\alpha^2\left(1+\frac{1}{\gamma^2}\right)
\ea
and is depicted in the second leftmost figure in Fig.\ref{fig:conf-map}. We denote this region by SD (semi-disk). The semicircle intersects with $x=0$ at $t_{E} =\pm \alpha$. 

We then  perform two conformal transformations to map the SD to the UHP. First, we consider a global conformal transformation such that in the complex coordinates $w=x+ i t_{E}, \bar{w}=x- i t_{E}$,  one of the intersections of two boundaries  $w=i\alpha$ is mapped to the infinity, the other $w=-i\alpha$ is mapped to the origin. We also require  the circle boundary comes  to the real axis. Such a map is given by
\ba
w^\prime=\frac{w+i\alpha}{w-i\alpha}e^{-i\varphi},
\label{conf-map-w}
\ea
where
\ba
\tan \varphi=\gamma.
\ea
The resulting region is shown in the second rightmost figure of Fig.\ref{fig:conf-map}. %We denote this wedge region by wedge.
Although there is an ambiguity of choosing $\varphi$, we choose $0\le \varphi <\pi$ so that the $\gamma\rightarrow 0 \Leftrightarrow M\rightarrow 0$ limit covers the whole region as expected. As we discussed in Section \ref{sec:m-less-r} (and \eqref{restrictionM}), we have $1/2< \chi=\sqrt{(R^2-M^2)/R^2} \le 1$.\footnote{Note that due to this condition, we always have a finite region after this conformal transformation.} This determines the value of  $\varphi$ uniquely as
\ba
\varphi=\pi\left(\sqrt{\frac{R^2}{R^2-M}}-1\right)=\pi\left(\frac{1}{\chi}-1\right).
\label{eq:vp}
\ea

We then map the wedge region to UHP. This is achieved by
\ba
u=w^{\prime\, \kappa}
\label{conf-map-u}
\ea
by suitably choosing $\kappa$.
In order to obtain a unique $\kappa$, we need to specify the branch. Let us take $0\le \arg u \le \pi$. Then, since $0\le \varphi<\pi$, $\kappa$ is determined as
\ba
\kappa (\pi-\varphi)=\pi\quad \Leftrightarrow\quad \kappa=\frac{1}{1-\frac{\varphi}{\pi}}.
\label{eq:kappa}
\ea
Note that we simply find $\kappa=\eta_0$, where $\eta_0$ was introduced as a critical rescaling parameter in (\ref{mineta}). The conformal map (\ref{conf-map-u}) creates a deficit angle at $w=0$ and $w=\infty$, which indeed corresponds to the intersection of the deficit angle in the AdS$_3$ with the AdS boundary in our gravity dual.

Finally, we comment on the case with $\lambda<0$. Even in this case, the same conformal transformation can be used. Since the original BCFT region with $\lambda<0$ is given by the complement region with the parity in the $x$ direction reversed in Fig\ref{fig:conf-map}, the resulting region becomes the lower half plane instead of the UHP. Therefore, the calculation of the stress tensor and entanglement entropy in the subsequent sections is exactly the same for $\lambda<0$.

%\subsection{Explicit map from $(t,x)$ to the UHP ($u$)}
%In this subsection, we discuss a map between the original coordinates $(t,x)$ and that of UHP ($u$) in order to calculate entanglement entropy. As we restrict the argument of each complex variable, we inspect it in a careful manner.

%We consider $w=x+i t_E, \bar{w}=x-it_E$.%, where $t_E=it$ and $\delta$ is a UV cutoff for the local quench. When we insert an operator, we take $\delta=\ap$. Momentarily we just absorb $\delta$ into the real part of $t_E$.

%The conformal map is given by
%\ba
%u=w^{\prime\, \kappa}=\left(\frac{w+i\alpha}{w-i\alpha} e^{-i\varphi}\right)^\kappa,
%\label{eq:conf-map}
%\ea
%where $\kappa=1/(1-\varphi/\pi)$ and $0\le \varphi/\pi=\sqrt{R^2/(R^2-M)}-1 < 1$.
%%%%%%%%%%%%%%%%%%%%%%%%%%%%
%Followings are probably wrong since $\bar{w}$ is not a mere conjugation.
%%%%%%%%%%%%%%%%%%%%%%%%%%%%
%Let us denote $w=re^{i\zeta}$, where $r^2=(x-t)^2+\delta^2$. For each complex variable, $w$, $w^\prime$, and $u$, we choose a principle value for their arguments. We would like to find $\arg w^{\prime}$ within $(-\pi,\pi]$.

\subsection{Stress tensor without a local operator}
In this section, we  compute  the one-point function of the stress tensor at $(w,\bar{w})=(x-t,x+t)$ in the SD and intend to compare it with the  holographic result \eqref{eflux}. We only take the boundary effect into account in this section and no local operator is being inserted. The conformal transformations (Fig.\ref{fig:conf-map}) yield
\ba
\langle T(w) \rangle_\mathrm{SD}=
\left(\frac{dw^\prime}{dw}\right)^2
\langle T(w^\prime) \rangle_\mathrm{wedge}
\ea
%global conformal transformation has a vanishing Schwarzian
and
\ba
0=\langle T(u) \rangle_\mathrm{UHP}=\left(\frac{dw^\prime}{du}\right)^2 \left( \langle T(w^\prime) \rangle_\mathrm{wedge} -\frac{c}{12}\{u;w^\prime\} \right),
\ea
where
\ba
\{u;w^\prime\}=\frac{\de_{w^\prime}^3 u}{\de_{w^\prime} u} - \frac{3}{2} \left(\frac{\de_{w^\prime}^2 u}{\de_{w^\prime} u}\right)^2
\ea
is the Schwarzian derivative.
From this we can read off the holomorphic part of the stress tensor as
\ba
\langle T(w) \rangle_\mathrm{SD}=\frac{c}{6}\ap^2 (\kappa^2-1) \frac{1}{((x-t)^2+\ap^2)^2},
\ea
%where $w=x+it_E=x-t$, $\ap=Re^\beta$, and $\kappa$ is given in \eqref{eq:kappa}.

Next, we would like to compare it with the holographic result \eqref{eflux} for $ T_{--}$. The holomorphic part of the stress tensor  computed above is related to the holographic convention by  $\langle T(w) \rangle = -2\pi T_{--}$.\footnote{The stress tensor is here defined as $T_{\mu\nu}=\frac{2}{\sqrt{|g|}}\frac{\delta W}{\delta g^{\mu\nu}}$ \eqref{eq:EM-tensor-def}, where $W$ is the free energy defined via the partition function $Z=e^{-W}$. The convention here follows \cite{Di_Francesco_1997} and same as \cite{Nozaki:2013wia}.} 
%Since $w=x-t=-u$ and $\bar{w}=x+t=v$, $T_{ww}=T_{uu}$.
Therefore, we find the energy flux reads
\ba
  T^{{\rm BCFT}}_{\pm\pm}=-\frac{c}{12}(\kappa^2-1)\cdot \frac{\ap^2}{\pi((x\pm t)^2+\ap^2)^2}.  \label{eq:EMtensor-bdy}
\ea
Note that the coefficient is negative and this is because the existence of two boundaries leads to the Casimir energy. The functional form of (\ref{eq:EMtensor-bdy}) agrees with the holographic result (\ref{eflux}). We will see in the next subsection, the coefficient also perfectly agrees with the holographic result.

%%%%  I omitted the plot of energy flux as it is rather trivial %%%%%%%
%%%%                                                            %%%%%%%
%Plots of $T_{--}/c$ in terms of $t-x$ with varying parameters $\ap> \epsilon$ and $\chi= %\sqrt{\frac{R^2-M}{R^2}}>1/2$ are shown in Fig.\ref{fig:plots-EMtensor1}.
%\begin{figure}
%    \centering
%    \begin{minipage}{0.9\textwidth}
%    \centering
%    \includegraphics[width=\textwidth]{EMplot1.pdf}
%    \end{minipage}\hfill
%    \begin{minipage}{0.9\textwidth}
%    \centering
%    \includegraphics[width=\textwidth]{EMplot2.pdf}
%    \end{minipage}
%    \caption{Plots of $T_{uu}$ \eqref{eq:EMtensor-bdy}, which only takes the boundary %effect from the EOW brane into account with respect to $u=t-x$. The top-left, top-right, %bottom-left, bottom-right plots are when $(\chi,\ap)=(0.9,1.0),(0.9,0.1),(0.55,1.0)$ and %$(0.55,0.1)$, respectively.}
%    \label{fig:plots-EMtensor1}
%\end{figure}

\subsection{Stress tensor with a local excitation}\label{sec:op-exc}
In the previous subsection, we computed the one-point function of the stress tensor without any insertions of local operators. In this subsection, we evaluate the one-point function of the stress tensor of an excited state by a primary operator $O(x)$, in the BCFT setup.  We denote  its conformal weight  $(h,\bar{h})$. For simplicity, we focus on a scalar operator, i.e. $h=\bar{h}=\Delta_O/2$. %where $\Delta=1+\sqrt{m^2 R^2 +1}$. %In the present setup, the inserted local operator should be irrelevant and not deform the theory as it gives a finite energy excitation. For the normalizable deformation to be dominant, we require $\Delta\gg 1$. 
Such a state has the following form,
\begin{align}
    |\Psi (t)\rangle &= \mathcal{N} e^{-itH} e^{-\ap H} O(t=0,x=\ti{\epsilon}) e^{\ap H} |0\rangle \nonumber\\
    &= \mathcal{N} e^{-itH} O(t_E=-\ap,x=\ti{\epsilon}) |0\rangle, %\nonumber\\
    %&=\mathcal{N} O(-t+i\ap,x=\ti{\epsilon}) |0\rangle, 
\end{align}
where  $\mathcal{N}$ is a  normalization coefficient to ensure $\langle \Psi (t) | \Psi(t) \rangle=1$, and $\alpha$ plays a role of a UV regulator. Note that as mentioned in the footnote 1, the UV regulator $\ap$ for the quench is written as the imaginary time position of the operator while the (Lorentzian) time evolution cannot as $H|0\rangle\neq 0$ due to the existence of the time-dependent boundary.\footnote{This is apparent if we consider the Euclidean path integral representation of the state.} The spatial location of the local operator is arbitrary, however  for a comparison with the holographic setup,  we place it at $x=\ti{\epsilon} \ll 1$, as we are interested in the limit where the local operator is inserted at the boundary $x=0$. 
Thus the location of the operator insertion on the Euclidean SD is given by $w_2=\ti{\epsilon}+it_{E\, 2}=\ti{\epsilon}-i\ap$. Similarly, the bra state is given by
\begin{align}
    \langle \Psi (t)| &= \mathcal{N} \langle 0 | e^{\ap H} O (t=0,x=\ti{\epsilon}) e^{-\ap H} e^{itH}  \nonumber\\
    &=\mathcal{N}  \langle 0 | O(t_E=\ap,x=\ti{\epsilon}) e^{itH}.
\end{align}
We denote the insertion point corresponding to this as $w_1=\ti{\epsilon}+ i\ap$. Let us consider measuring the holomorphic energy momentum tensor at $w=x-t$, that is, discussing its time evolution in the Heisenberg picture.
By employing  the series of conformal transformations, we have
\begin{align}
    %&
    &\langle \Psi (t=0) |T(w=x-t)|\Psi(t=0)\rangle \nonumber\\
    =&\frac{\langle O(w_1,\bar{w}_1) T(w) O(w_2,\bar{w}_2)\rangle_{\mathrm{SD}}}{\langle O(w_1,\bar{w}_1) O(w_2,\bar{w}_2)\rangle_{\mathrm{SD}}} \nonumber\\
    =&
    \left( \frac{dw^\prime}{dw} \right)^2 %\nonumber\\
    %&\times 
    \left[
    \left( \frac{du}{dw^\prime} \right)^2 \frac{\langle O(u_1,\bar{u}_1) T(u) O(u_2,\bar{u}_2) \rangle_{\mathrm{UHP}}}{\langle O(u_1,\bar{u}_1) O(u_2,\bar{u}_2) \rangle_{\mathrm{UHP}}}
    +\frac{c}{12} \{u;w^\prime\}
    \right],
\end{align}
where $u_i=u(w_i)$ and $\bar{u}_i=\bar{u}(\bar{w}_i)$.
To compute the first term in the brackets, we employ the conformal Ward identity in the UHP~\cite{Cardy:2004hm}
\begin{align}
    &\left\langle T(u) \prod_j O(u_j,\bar{u}_j) \right\rangle_{\mathrm{UHP}} \nonumber\\
    &= \sum_j \left(
    \frac{h}{(u-u_j)^2} + \frac{1}{u-u_j} \de_{u_j} + \frac{\bar{h}}{(u-\bar{u}_j)^2} + \frac{1}{u-\bar{u}_j} \de_{\bar{u}_j}
    \right) \left \langle \prod_j O(u_j,\bar{u}_j) \right \rangle_{\mathrm{UHP}} .\label{yyy}
\end{align}
%Alternatively, we can compute the expectation value from a three point function since $\langle TO_1O_2 O_3 O_4\rangle\sim \langle TO_1O_2\rangle \langle O_3 O_4\rangle + \langle O_1O_2\rangle \langle T O_3 O_4\rangle$ using the large-$c$ factorized form $\langle O_1O_2 O_3 O_4\rangle \sim \langle O_1O_2\rangle \langle O_3 O_4\rangle$ as discussed below.

This expression is further simplified by using the doubling trick~\cite{Cardy:1984bb,Recknagel:2013uja}. This  trick relates  a correlator on the UHP  to a chiral correlator on the entire plane $\mathbb{C}$.  For example, we have
\ba
\langle O(u_{1},\bar{u}_{1})  O(u_{2},\bar{u}_{2})\rangle_{\mathrm{UHP}} =\langle O(u_1) O(\bar{u}_1) O(u_2) O(\bar{u}_2) \rangle_{\mathbb{C}}^{c}.
\ea
The superscript $c$ means the correlator is chiral.

Using this relation as well as the detailed expression of the conformal map,  we get
%since $\frac{du}{dw^\prime}=\kappa u^{1-1/\kappa}, \frac{dw^\prime}{dw}=-2i\ap u^{1/\kappa}/(w^2+\ap^2)$, and $\{u;w^\prime\}=-\frac{\kappa^2-1}{2} u^{-2/\kappa}$, we obtain
\begin{align}
    & \langle \Psi (t=0) |T(w=x-t)|\Psi(t=0)\rangle \nonumber\\
    =& -\frac{4\ap^2 \kappa^2}{(w^2+\ap^2)^2} \frac{1}{\langle O(u_1) O(\bar{u}_1) O(u_2) O(\bar{u}_2) \rangle_{\mathbb{C}}^{c}} \nonumber\\
    & \left[
    u^2 \sum_{i=1}^2 
    \left\{ h \left( \frac{1}{(u-u_i)^2} + \frac{1}{(u-\bar{u}_i)^2} \right) + \frac{1}{u-u_i}\de_{u_i} + \frac{1}{u-\bar{u}_i}\de_{\bar{u}_i} 
    \right\}
    -\frac{c}{24} \left(1-\frac{1}{\kappa^2}\right)
    \right] \nonumber\\
    & \langle O(u_1) O(\bar{u}_1) O(u_2) O(\bar{u}_2) \rangle_{\mathbb{C}}^{c}.
    \label{em-tensor-op}
\end{align}

\subsubsection{Stress tensor in holographic CFT} \label{sec:EMtensor-holCFT}

Let us evaluate the stress tensor in the holographic CFT dual to our AdS/BCFT model. In our analysis of  the dual gravitational setup, we assumed that the EOW brane does not have any source to the  bulk scalar field dual to the local operator $O(u)$ 
(for the analysis in the presence of non trivial scalar field configuration, refer to \cite{Suzuki:2022xwv}).  This means that  its  one-point function in the UHP vanishes: 
$\langle O(u,\bar{u}) \rangle_{\mathrm{UHP}} =0$. Therefore  $\langle O(u_1) O(\bar{u}_1) O(u_2) O(\bar{u}_2) \rangle_{\mathbb{C}}^{c}$  is factorized into the product of two point functions $\la O(u_1)O(u_2)\lb \cdot \la O(\bar{u}_1)O(\bar{u}_2)\lb$ for any parameter region of $u_1$ and $u_2$. This leads to the conclusion
\ba
\langle O(u_1) O(\bar{u}_1) O(u_2) O(\bar{u}_2) \rangle_{\mathbb{C}}^{c}
\propto \frac{1}{|u_1-u_2|^{\Delta_O} |\bar{u}_1-\bar{u}_2|^{\Delta_O}}.
\label{holcorb}
\ea
Notice that this is equivalent to the chiral two-point function with the doubled conformal dimension $\tilde{\Delta}\equiv 2\Delta_O$. (Fig.\ref{fig:ope}) This can be understood as follows.
The primary operators $O$ at the origin and the infinity constituting the ket and bra states are doubled via the doubling trick, up to the normalization. Let us denote these doubled states as $|\Psi\rangle \rightarrow |\Psi \Psi^\ast\rangle$ and $\langle \Psi| \rightarrow \langle\Psi \Psi^\ast|$.
Then, the equivalence between the four-point correlator and the two-point correlator with a doubled conformal dimension implies that we should regard the local operator and its mirror, $|\Psi\Psi^\ast\rangle$ and $\langle \Psi\Psi^\ast |$, as composite operators $|\tilde{\Psi}\rangle$ and $\langle \tilde{\Psi}|$, where tilded operators have a double conformal dimension $\tilde{\Delta}=2\Delta_O$ than the original ones. As we will see, this identification is necessary for the first law of entanglement entropy and the correct correspondence with the holographic calculation.

\begin{figure}
    \centering
    \includegraphics[width=15cm]{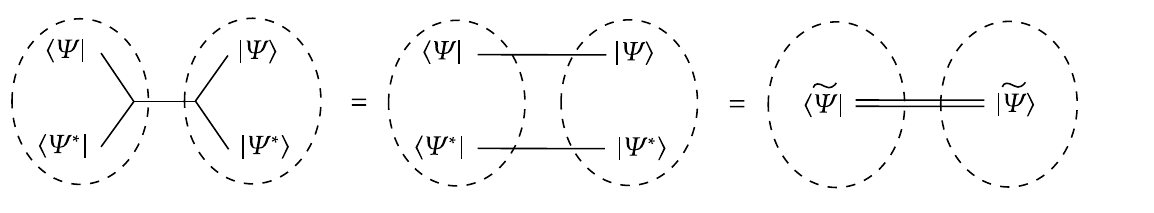}
    \caption{Left: The two-point function in the BCFT equals to the chiral four-point function. Middle: Since the holographic dual in concern has a vanishing one-point function, the contractions of operators in dashed circles cannot contribute to the correlator on their own. Right: The primary operator $\Psi$ and its mirror operator $\Psi^\ast$ should be regarded as a single operator $\tilde{\Psi}$, whose conformal dimension is $\tilde{\Delta}=2\Delta_O$.
    }
    \label{fig:ope}
\end{figure}

%Note also that 
Since we take the limit $\ti{\ep}\to 0$ which leads to
$w_1\to-i\ap$ and $w_2\to i\ap$ in the $w$-plane, we find $u_1\to 0$ and $u_2\to \infty$. 
Thus, this leads to the following result at any time $t$ by plugging (\ref{holcorb}) into 
(\ref{em-tensor-op}):
\ba
\la T(w)\lb=\frac{c\ap^2 (\kappa^2-1)}{6(w^2+\ap^2)^2}
-\frac{4\Delta_O \ap^2 \kappa^2}{(w^2+\ap^2)^2}.
\ea
This leads to the stress tensor in Lorentzian signature via
$T_{--}=-\frac{1}{2\pi}T(w)$:
\ba
&& T_{--}=s_{BCFT}(\Delta_O)\cdot \frac{\ap^2}{\pi((x-t)^2+\ap^2)^2},\no
&& s_{BCFT}(\Delta_O)\equiv-\frac{c}{12}(\kappa^2-1)+2\kappa^2\Delta_O
=\kappa^2 \tilde{\Delta} +\frac{c}{12}(1-\kappa^2).  \label{EMgencft}
\ea
Notice that the part $-\frac{c}{12}(\kappa^2-1)$ in $s_{BCFT}$ can be regarded as the Casimir energy part which is negative and the other part $2\kappa^2\Delta_O$ corresponds to the local operator excitation. It is also useful to note that when the local excitation is light $\Delta_O\ll c$, we find
\ba
s_{BCFT}(\Delta_O)\simeq \Delta_O.
\ea

Now we would like to compare the BCFT result (\ref{EMgencft}) with the gravity dual result (\ref{efluxx}) at $\eta=1$. It is straightforward to see that these two are identical $s_{BCFT}(\Delta_O)=s_{AdS}(\Delta_{AdS},1)=\Delta_{AdS}$ via the relation (\ref{relasd}). Furthermore, by solving \eqref{EMgencft} for $\tilde{\Delta}=2\Delta_O$ and using the fact $\kappa=\eta_0$, we obtain\footnote{It is worth to note that $s_{BCFT}$ and $s_{AdS}|_{\eta=\eta_0}$ as a function of $\tilde{\Delta}$ and $\Delta_{AdS}$ respectively are inverse to each other.}
\begin{equation}
    2\Delta_O=\tilde{\Delta}=s_{AdS}(s_{BCFT}(\Delta_O),\eta_0) =s_{AdS}(\Delta_{AdS},\eta_0).
\end{equation}
This is nothing but \eqref{eq:AdS EM tensor}. In this way, we can perfectly reproduce the previous holographic stress tensor from the present BCFT approach.

\subsubsection{Stress tensor in the free scalar CFT}

It is also helpful to compare our energy flux in the holographic CFT with that in a free scalar CFT.  Consider the $c=1$ CFT of an uncompactified scalar in two dimensions. The real scalar field $\phi(u,\bar{u})$ has the operator product expansion (OPE) in the presence of the boundary (the real axis in the $u$ coordinates) with the Neumann boundary condition:
\ba
\la \phi(u,\bar{u}) \phi(u',\bar{u}') \lb_N=-\log |u-u'|^2-\log |u-\bar{u}'|^2.
\ea
For the Dirichlet boundary condition, we have 
\ba
\la \phi(u,\bar{u}) \phi(u',\bar{u}') \lb_D=-\log |u-u'|^2+\log |u-\bar{u}'|^2.
\ea

For the local primary operator we choose,
\ba
O(u_1,\bar{u}_1)=e^{ik\phi(u_1,\bar{u}_1) },\ \ \ \ O^\dagger(u_2,\bar{u}_2)=e^{-ik\phi(u_2,\bar{u}_2) },
\ea
both of which have the conformal dimension $\Delta_O=2h=k^2$.

The two-point functions for these  boundary conditions read
\ba
&& \la O(u_1,\bar{u}_1)O^\dagger(u_2,\bar{u}_2)\lb_N
=\frac{|u_1-\bar{u}_1|^{\Delta_O}|u_2-\bar{u}_2|^{\Delta_O}}
{|u_1-u_2|^{2\Delta_O}|u_1-\bar{u}_2|^{2\Delta_O}},\no
&& \la O(u_1,\bar{u}_1)O^\dagger(u_2,\bar{u}_2)\lb_D
=\frac{|u_1-\bar{u}_2|^{2\Delta_O}}{|u_1-u_2|^{2\Delta_O}|u_1-\bar{u}_1|^{\Delta_O}|u_2-\bar{u}_2|^{\Delta_O}}.\label{twoptsc}
\ea

The expectation value of the stress tensor $T(u)=-\frac{1}{2}\de_u\phi\de_u\phi$ 
is evaluated for each boundary condition as follows
\ba
&& \la T(u)\lb_N=\frac{\Delta_O}{2}\left(\frac{1}{u-u_1}+\frac{1}{u-\bar{u}_1}
-\frac{1}{u-u_2}-\frac{1}{u-\bar{u}_2}\right)^2,\no
&& \la T(u)\lb_D=\frac{\Delta_O}{2}\left(\frac{1}{u-u_1}-\frac{1}{u-\bar{u}_1}
-\frac{1}{u-u_2}+\frac{1}{u-\bar{u}_2}\right)^2.
\ea
They perfectly agree with those obtained from the two point functions (\ref{twoptsc}) by applying the conformal Ward identity (\ref{yyy}). 

By taking our limit $u_1\to 0$ and $u_2\to \infty$, we obtain 
\be
\la T(u)\lb_N\simeq \frac{2\Delta_O}{u^2},\ \ \ \ \ 
\la T(u)\lb_D\simeq 0.
\ee
They take different values when compared with the previous holographic BCFT result (obtained by plugging \eqref{holcorb} in \eqref{em-tensor-op} and take the $u_1\rightarrow 0$ and $u_2\rightarrow \infty$ limits):
\ba
 \la T(u)\lb_{Hol}\simeq \frac{\Delta_O}{u^2}.  \label{holbcft}
\ea

Finally, by using the conformal map of the stress tensor 
(\ref{em-tensor-op}), or explicitly 
\ba
\la T(w)\lb=-\frac{4\ap^2\kappa^2}{(w^2+\ap^2)^2}\left[u^2\la T(u)\lb-\frac{c}{24}(1-\kappa^{-2})\right],
\ea
we obtain the physical stress tensor $\la T(w)\lb$. In this way we found that the energy flux is sensitive to both the types of CFTs and the boundary conditions.

\subsection{Stress tensor after rescaling by $\eta$}
The stress tensor in the holographic CFT matches the holographic computation given in \eqref{efluxx} even after the rescaling by an arbitrary $\eta$ \eqref{rescaleex}.

The rescaling by $\eta$ maps $\chi$ and $\kappa=\eta_0$ to $\chi/\eta$ and $\kappa_{\eta}\equiv\kappa/\eta$ respectively since $M$ in $\chi$ becomes $M^\prime$ \eqref{newmass} and the location of the induced boundary $\theta={\pi}/{\kappa}$ becomes $\theta^\prime=\eta{\pi}/{\kappa}$.\footnote{Note that even though $\kappa_{\eta=1}\equiv\kappa$ is written as a function of $\chi$ (\eqref{eq:vp} and \eqref{eq:kappa}), $\kappa_\eta$ does not equal to $\kappa|_{\chi\rightarrow\chi/\eta}$. This is because $\kappa$ is defined as a location of the boundary, which is directly rescaled by $\eta$, not a mere function of $\chi$.} This is consistent with \eqref{Identifyap}.

After replacing $\kappa$ in \eqref{EMgencft} by $\kappa_\eta={\kappa}/{\eta}$,
we exactly reproduce \eqref{efluxx} through \eqref{eq:AdS EM tensor}.

Alternatively, this can be explicitly confirmed by the conformal mapping \eqref{conf-map-eta}
\begin{equation}
    w_{\eta}=\ap\tan\left(\eta \arctan \left(\frac{w}{\ap}\right)\right),
    \quad
    \bar{w}_{\eta}=\ap\tan\left(\eta \arctan \left(\frac{\bar{w}}{\ap}\right)\right)
\end{equation}
where $w_{\eta}=x^\prime-t^\prime$, the $w$ coordinates after the rescaling by $\eta$ (same for the anti-holomorphic one).
Then,
\begin{align}
    \langle T(w_\eta) \rangle &= \left(\frac{dw}{dw_\eta}\right)^2 \langle T(w) \rangle +\frac{c}{12} \{w; w_\eta \} \nonumber\\
    &=\left[\frac{\kappa^2}{\eta^2}\tilde{\Delta}+\frac{c}{12}\left(1-\frac{\kappa^2}{\eta^2}\right)\right]\frac{-2\ap^2}{(w_\eta^2+\ap^2)^2}.
\end{align}
This is nothing but $\kappa\rightarrow\frac{\kappa}{\eta}$ in \eqref{EMgencft} as mentioned above.

In particular, when $\eta=\eta_0=\kappa$,
\begin{equation}
    s_{BCFT}(\Delta_O)=2\Delta_O
\end{equation}
and \eqref{eq:AdS EM tensor} is readily confirmed.

\section{The BCFT calculation -- entanglement entropy}\label{sec:BCFT-EE-calc}
In this section, we calculate entanglement entropy in the holographic BCFT of concern. In the first part, we compute the vacuum entanglement entropy with the induced boundary $x=Z(t)$ \eqref{disco} as well as the original boundary $x=0$. Then, in the second part, we discuss entanglement entropy with both the boundaries and a local operator insertion. Finally, we comment on the behavior of the entanglement entropy after the generic rescaling by $\eta$ \eqref{rescaleex}. We will see that all of these calculations perfectly reproduce the holographic results obtained in section 3.

\subsection{Entanglement entropy without a local operator}
%\textbf{Caution:} Same with the previous section. But for entanglement entropy, if the time difference between inserted points and intervals are large, probably we can just factorize the effect of inserted operators themselves like $\log <O^\dagger O>$. But still it contributes to the time evolution.
%Taking only the effect of the brane boundary
In this section, we  compute the vacuum entanglement entropy on the SD, i.e. only take the (induced) boundary effect into  account, and compare it with the holographic result discussed in Section \ref{sec:holoEE}. By using the conformal map  %between $w$ and $u$ given by \eqref{conf-map-w} and \eqref{conf-map-u} 
(Fig.\ref{fig:conf-map}), entanglement entropy of an interval $I=[x_A,x_B]$ ($0<x_A<x_B<Z(t)$) at time $t$ is given by
\begin{equation}
    S_{AB\ \mathrm{no\, op.}} =\lim_{n\rightarrow 1} \frac{1}{1-n} \log \mbox{Tr} \rho_{AB\ \mathrm{no\, op.}}^n,
\end{equation}
where
\begin{align}
    \mbox{Tr}\rho_{AB\ \mathrm{no\, op.}}^n &=\langle 0 | \sigma_n (w_A,\bar{w}_A) {\sigma}_n (w_B,\bar{w}_B) |0\rangle_{\mathrm{SD}} \nonumber \\
    &=\left(
    \left.\frac{du}{dw}\right\vert_{u=u_A} \left.\frac{d\bar{u}}{d\bar{w}}\right\vert_{\bar{u}=\bar{u}_A}
    \left.\frac{du}{dw}\right\vert_{u=u_B} \left.\frac{d\bar{u}}{d\bar{w}}\right\vert_{\bar{u}=\bar{u}_B}
    \right)^{\Delta_n /2}
    \langle 0 | \sigma_n(u_A,\bar{u}_A) {\sigma}_n(u_B,\bar{u}_B) |0\rangle_\mathrm{UHP}
\label{twist-correlator}
\end{align}
in terms of the reduced density matrix defined in \eqref{eq:reduced-rho}. $\Delta_n=\frac{c}{12}\left(n-\frac{1}{n}\right)$ is the conformal weight of the twist operators $\sigma_n$. Each twist operator is inserted at the endpoints of the interval $I$: $w_{A,B}=x_{A,B}+it_{E}=x_{A,B}-t$, $\bar{w}_{A,B}=x_{A,B}+t$. The $u$'s are coordinates after the conformal transformation, i.e. $u_{A,B}=u(w_{A,B})$ and $\bar{u}_{A,B}=\bar{u}(\bar{w}_{A,B})$.

To compute a correlation function on the UHP, we use the doubling trick~\cite{Cardy:1984bb,Recknagel:2013uja}:
\ba
\langle \sigma_n(u_A,\bar{u}_A) {\sigma}_n(u_B,\bar{u}_B) \rangle_\mathrm{UHP}=\langle \sigma_n(u_A) \sigma_n(\bar{u}_A) {\sigma}_n(u_B) {\sigma}_n(\bar{u}_B) \rangle_\mathrm{\mathbb{C}}^{c}.
\label{doubling}
\ea
%$\sim$ indicates this equivalence holds at the level of differential equations. 
Assuming the CFT is holographic, this four-point function can be evaluated via the vacuum conformal block in the large central charge limit~\cite{Hartman:2013mia,Asplund:2014coa}. The two possible OPE channels correspond to the connected and disconnected geodesics in the holographic calculation.

%To evaluate this four-point function, we employ the HHLL approximation assuming $n\rightarrow 1$ in the large-$c$ expansion followed by~\cite{Hartman:2013mia,Asplund:2014coa}. The four-point function \eqref{doubling} is calculated via the conformal block expansion. In the large-$c$ limit, the dominant intermediate states are given by the identity operator and its descendants. If we focus on the early and late time behavior, descendants become less contributing and hence the four-point function can be factorized as a product of two-point functions. The dominant contribution is given by either the $s$-channel or the $t$-channel expansion. They are given by
%\ba
%\mathrm{tr}\rho_A^n=
%\left|
%\left(\frac{2\ap}{1-\frac{\vp}{\pi}}\right)^{2} 
%\frac{u_1 u_2}{(w_1^2+\ap^2) (w_2^2+\ap^2)}
%\right|^{\Delta_n}
%\begin{dcases*}
%\frac{1}{|u_1-\bar{u}_1|^{\Delta_n} |u_2-\bar{u}_2|^{\Delta_n} }&($s$-channel)\\
%\frac{1}{|u_1-u_2|^{\Delta_n} |\bar{u}_1-\bar{u}_2|^{\Delta_n} }&($t$-channel)
%\end{dcases*},
%\label{renyi-bdy}
%\ea
%\footnote{This expression is up to an overall factor irrelevant to entanglement entropy.} 

%\subsubsection{Entanglement entropy in the absence of local operator}

%By using the conformal map (\ref{conf-map-w}) and  (\ref{conf-map-u}), we can calculate the entanglement entropy for an interval $A=[x_A,x_B]$ at time $t$ on the $w$-plane (note $w=x+i\tau=x-t$).  Note again that $x_A$ and $x_B$ can take values in the range $0<x_A<x_B<Z(t)$. We assume that there is no operator inserted and the CFT is a holographic one. Then 
The entanglement entropy reads $S_{AB\ \mathrm{no\, op.}}=\min \{S^{con}_{AB\ \mathrm{no\, op.}},S^{dis}_{AB\ \mathrm{no\, op.}}\}$, where
\ba
&& S^{con}_{AB\ \mathrm{no\, op.}}=\frac{c}{6}\log\frac{|u_A-u_B|^2}{\ep^2|u'_A||u'_B|},\no
&& S^{dis}_{AB\ \mathrm{no\, op.}}=\frac{c}{6}\log\frac{|u_A-\bar{u}_A||u_B-\bar{u}_B|}{\ep^2|u'_A||u'_B|}+2S_{bdy},
\label{eq:EEwoop}
\ea
where $u_{A,B}=u(w_{A,B})$ and  $u'_{A,B}=\left.\frac{du}{dw}\right|_{w=w_{A,B}}$. $S_{bdy}$ is the boundary entropy and $\ep$ is the UV cutoff. 

We plotted the behavior of the entanglement entropy in Fig.\ref{fig:heevac}. The connected entanglement entropy $S^{con}_{AB}$ for an interval $[x_A,x_B]$ shows a dip around the time $t\sim\frac{x_A+x_B}{2}$. This can be interpreted as the shock wave of negatve energy flux emitted by the falling massive object in AdS, whose trajectory looks like 
$x\sim t$. On the other hand, the disconnected entanglement entropy $S^{dis}_{AB}$ increases until the massive falling particle crosses the second minimal surface when $t\sim x_B$.

The rightmost plot is proportional to the behavior of the negative energy flux at $t=0$ as predicted by the first law of entanglement entropy (\ref{firstL}). Indeed we can confirm that in the small subsystem size  limit $|x_A-x_B|\to 0$, we have 
\ba
 \Delta S^{con}_{AB}\simeq \frac{c}{18}(1-\kappa^2)H(t,x_A)(x_A-x_B)^2,
\ea
which is completely  in accord with the result for the  energy momentum tensor (\ref{eq:EMtensor-bdy}).\footnote{To derive this, it is useful to see $H(t,x)=\frac{\ap^2}{\pi}\left(\frac{1}{(w^2+\ap^2)^2}+\frac{1}{(\bar{w}^2+\ap^2)^2}\right)$. Note that since $\Delta S_{AB}^{con}=-\frac{c}{72}|x_A-x_B|^2\left(\{u_A\;w_A\}+\{\bar{u}_A\;\bar{w}_A\}\right)$, the first law holds for any conformal map $u(w)$. Thus, even after the rescaling by $\eta$, it continues to hold.}

\begin{figure}
      \includegraphics[width=5.25cm]{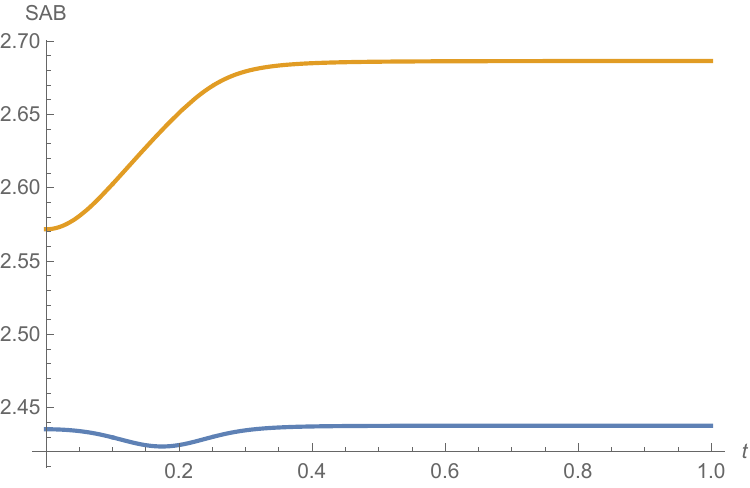}
      \includegraphics[width=5.25cm]{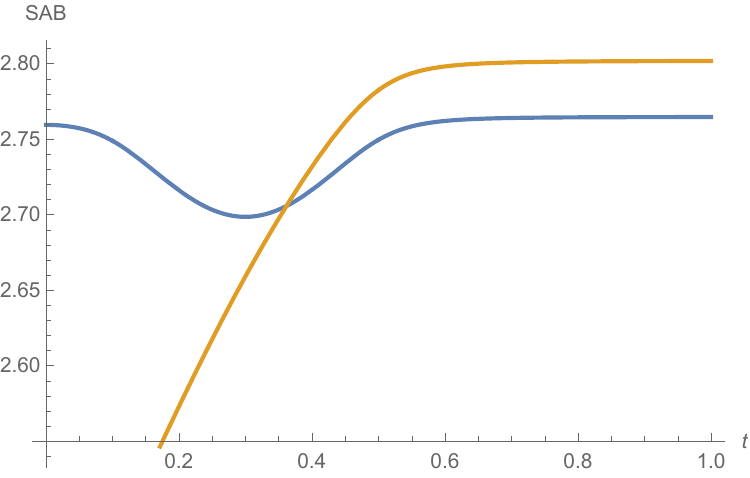}
         \includegraphics[width=5.25cm]{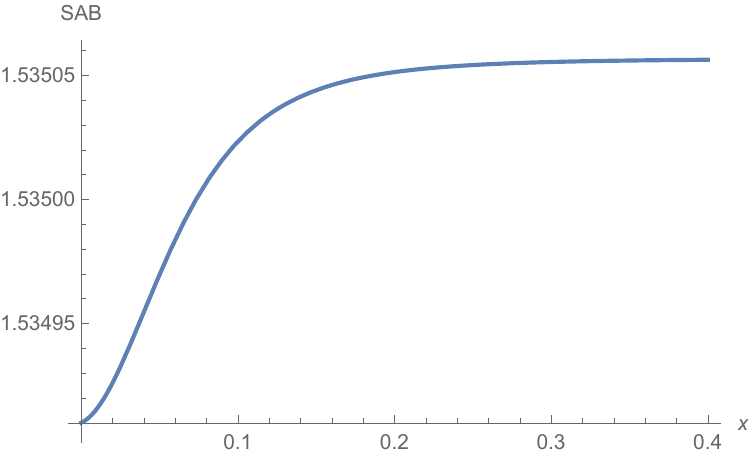}
        \caption{Plots of the entanglement entropy \eqref{eq:EEwoop} in the BCFT.
        We chose $\ap=0.1$,  $\chi=0.9$, and $\ep=0.0001$, where we have $Z(0)=\frac{\ap}{\tan\frac{\vp}{2}}\simeq 0.567$. In the left and middle plot, 
        we chose $I=[0.1,0.25]$ and $I=[0.1,0.5]$, respectively and we showed the entanglement entropy as a function of the time $t$. The blue and orange graph describe $S^{con}_{AB}$ and $S^{dis}_{AB}$ (we set $S_{bdy}=0$), respectively. In the right panel we plotted $S^{con}_{AB}$ at $t=0$ for the subsystem $I=[x,x+0.01]$ as a function of $x$. In these plots, we set $c=1$ and $S_{bdy}=0$.
        }
    \label{fig:heevac}
\end{figure}

%In contrast to $S^{con}_{AB}$, the behavior of $S^{dis}_{AB}$ is qualitatively similar to the holographic result (Fig.\ref{fig:Sdis1}), in which the effect of the local operator is already taken into the account. We can expect the operator effect is not crucial for the disconnected phase and only gives a small peak around $x=x_B$.\footnote{Note that the magnitude of $S^{dis}_{AB}$ seems to be different from the holographic one. This causes the discrepancy of the time of the phase transition from the connected phase to the disconnected phase between the CFT results and holographic ones. One possibility is that the boundary entropy gets slightly modified by the inclusion of the local quench so that $S^{dis}_{AB}$ differs by a constant.}

\subsection{Entanglement entropy with a local excitation}
%After the doubling trick, we can analyze the six point function as discussed in p.14-16 of \cite{Caputa:2015waa}.

In this section, we discuss entanglement entropy of a single interval $I=[x_A,x_B]$, taking the effect of the local operator insertions as well as the induced boundary into the account. %Since there involves technical difficulties regarding higher-point functions in general, the more detailed analysis is kept for the future work. Here we intend to address the possible derivation of holographic results from the holographic CFT.

\subsubsection{Factorization of entanglement entropy}

Entanglement entropy is calculated from
\begin{align}
\hspace{-15pt}
    \mbox{Tr} \rho_{AB}^n &= \bra{\Psi}\sigma (w_A,\bar{w}_A) \sigma (w_B,\bar{w}_B) \ket{\Psi}_{\mathrm{SD}} \nonumber\\
    &=\left[(2\kappa \ap)^4 \frac{u_A\bar{u}_A u_B\bar{u}_B}{(w_A^2+\ap^2)(\bar{w}_A^2+\ap^2)(w_B^2+\ap^2)(\bar{w}_B^2+\ap^2)} \right]^{\Delta_n /2} \bra{\Psi}\sigma (u_A,\bar{u}_A) \sigma (u_B,\bar{u}_B) \ket{\Psi}_{\mathrm{UHP}} \nonumber\\
    &=\left[(2\kappa \ap)^4 \frac{u_A\bar{u}_A u_B\bar{u}_B}{(w_A^2+\ap^2)(\bar{w}_A^2+\ap^2)(w_B^2+\ap^2)(\bar{w}_B^2+\ap^2)} \right]^{\Delta_n /2} \bra{\Psi\Psi^\ast} \sigma (u_A) \sigma (u_B) \sigma (\bar{u}_A) \sigma (\bar{u}_B) \ket{\Psi\Psi^\ast}_{\mathbb{C}}^c.
    \label{eq:2twists}
\end{align}
%where we again used the conformal map from the SD to the UHP $u_i=u(w_i)_{i=A,B}$ given by \eqref{conf-map-w} and \eqref{conf-map-u} (Fig.\ref{fig:conf-map}). 
%In the following, we will denote $\sigma (u_i)$ as $\sigma_i$ and $\sigma (\bar{u}_i)$ as $\bar{\sigma}_i$.

\eqref{eq:2twists} is written in terms of the chiral eight-point function and we cannot obtain a simple analytic form even in the large central charge limit in general. 
However, when the subregion size is sufficiently small or large,
there exists two phases corresponding to the connected and disconnected geodesics in the holographic entanglement entropy. %The former should dominate whenever $x_A\ (x_B)\gg x_B-x_A$ while the latter should dominate whenever $x_A\ (x_B)\ll x_B-x_A$.
%Here, we focus on the two phases mentioned above and hypothesize the dominant OPE contribution for each phase. 
The holographic entanglement entropy predicts the connected entropy dominates at early time and the disconnected entropy dominates at late time (Fig.\ref{fig:Sdis1}, \ref{fig:Sdis2}). The former corresponds to $x_A\ (x_B)\gg x_B-x_A$ or equivalently, $u_A\sim u_B$ (and $\bar{u}_A\sim\bar{u}_B$), while the latter corresponds to $x_A\ (x_B)\ll x_B-x_A$ or equivalently, $u_{A,B}\sim \bar{u}_{A,B}$. In terms of the OPE channels, they correspond to
\begin{align}
&\tikzset{every picture/.style={line width=0.75pt}}
\begin{tikzpicture}[baseline={([yshift=-3.5ex]current bounding box.center)},x=0.75pt,y=0.75pt,yscale=-1,xscale=1]
\draw    (42.67,44.6) -- (63.67,75.6) ;
\draw    (42.67,104.6) -- (63.67,75.6) ;
\draw    (63.67,75.6) -- (103.67,75.6) ;
\draw    (103.8,51.4) -- (83,25.4) ;
\draw    (103.8,51.4) -- (103.67,75.6) ;
\draw    (103.8,51.4) -- (126.6,25.8) ; 
\draw    (103.67,75.6) -- (205.67,75.6) ;
\draw    (205.8,51.4) -- (185,25.4) ;
\draw    (205.8,51.4) -- (205.67,75.6) ; 
\draw    (205.8,51.4) -- (228.6,25.8) ;
\draw    (265.31,106.13) -- (244.08,75.29) ;
\draw    (264.86,46.13) -- (244.08,75.29) ;
\draw    (244.2,75.4) -- (205.67,75.6) ;
\draw (6,32.6) node [anchor=north west][inner sep=0.75pt]    {$\langle \Psi |$};
\draw (1,92.6) node [anchor=north west][inner sep=0.75pt]    {$\langle \Psi ^{\ast } |$};
\draw (59.6,2.6) node [anchor=north west][inner sep=0.75pt]    {$\sigma ( u_{A})$};
\draw (166.6,2.6) node [anchor=north west][inner sep=0.75pt]    {$\sigma (\bar{u}_{A})$};
\draw (115.6,2.6) node [anchor=north west][inner sep=0.75pt]    {$\sigma ( u_{B})$};
\draw (221.4,2.6) node [anchor=north west][inner sep=0.75pt]    {$\sigma (\bar{u}_{B})$};
\draw (269.4,33.2) node [anchor=north west][inner sep=0.75pt]    {$|\Psi \rangle $};
\draw (270.8,93.2) node [anchor=north west][inner sep=0.75pt]    {$|\Psi ^{\ast } \rangle $};
\end{tikzpicture}
\quad\text{(connected phase)}
\intertext{and}
&
\tikzset{every picture/.style={line width=0.75pt}}
\begin{tikzpicture}[baseline={([yshift=-3.5ex]current bounding box.center)},x=0.75pt,y=0.75pt,yscale=-1,xscale=1]
\draw    (42.67,44.6) -- (63.67,75.6) ;
\draw    (42.67,104.6) -- (63.67,75.6) ;
\draw    (63.67,75.6) -- (103.67,75.6) ;
\draw    (103.8,51.4) -- (83,25.4) ;
\draw    (103.8,51.4) -- (103.67,75.6) ;
\draw    (103.8,51.4) -- (126.6,25.8) ; 
\draw    (103.67,75.6) -- (205.67,75.6) ;
\draw    (205.8,51.4) -- (185,25.4) ;
\draw    (205.8,51.4) -- (205.67,75.6) ; 
\draw    (205.8,51.4) -- (228.6,25.8) ;
\draw    (265.31,106.13) -- (244.08,75.29) ;
\draw    (264.86,46.13) -- (244.08,75.29) ;
\draw    (244.2,75.4) -- (205.67,75.6) ;
\draw (6,32.6) node [anchor=north west][inner sep=0.75pt]    {$\langle \Psi |$};
\draw (1,92.6) node [anchor=north west][inner sep=0.75pt]    {$\langle \Psi ^{\ast } |$};
\draw (59.6,2.6) node [anchor=north west][inner sep=0.75pt]    {$\sigma ( u_{A})$};
\draw (166.6,2.6) node [anchor=north west][inner sep=0.75pt]    {$\sigma ( u_{B})$};
\draw (115.6,2.6) node [anchor=north west][inner sep=0.75pt]    {$\sigma (\bar{u}_{A})$};
\draw (221.4,2.6) node [anchor=north west][inner sep=0.75pt]    {$\sigma (\bar{u}_{B})$};
\draw (269.4,33.2) node [anchor=north west][inner sep=0.75pt]    {$|\Psi \rangle $};
\draw (270.8,93.2) node [anchor=north west][inner sep=0.75pt]    {$|\Psi ^{\ast } \rangle $};
\end{tikzpicture}
\quad\text{(disconnected phase).}
\end{align}

\vspace{0.8cm}

As we discussed in Fig.\ref{fig:ope}, we can regard $|\Psi\Psi^\ast\rangle$ and $\langle \Psi\Psi^\ast |$ as excited states created by a single primary operator $|\tilde{\Psi}\rangle$ and $\langle \tilde{\Psi}|$, whose conformal dimension is $\tilde{\Delta}=2\Delta_O$  as we are considering the BCFT dual of the spacetime without the bulk matter profile and there is no source for the matter field on the EOW brane.
The limits $x_A\ (x_B)\gg x_B-x_A$ and $x_A\ (x_B)\ll x_B-x_A$ imply $\sigma(u_A) \sigma(u_B)\sim \mathbf{1}$ and $\sigma(u_A) \sigma(\bar{u}_A)\sim \mathbf{1}$ respectively. Then, the orthogonality of the two point function in any CFTs leads~\cite{Caputa:2015waa}
\begin{align}
    \bra{\Psi\Psi^\ast} \sigma (u_A) \sigma (u_B) \sigma (\bar{u}_A) \sigma (\bar{u}_B) \ket{\Psi\Psi^\ast}_{\mathbb{C}}^c 
    &= \bra{\tilde{\Psi}} \sigma (u_A) \sigma (u_B) \sigma (\bar{u}_A) \sigma (\bar{u}_B) \ket{\tilde{\Psi}}_{\mathbb{C}}^c \nonumber\\
    &= \sum_\ap \bra{\tilde{\Psi}} \sigma (u_A) \sigma (u_B) \ket{\ap}^c_{\mathbb{C}} \!\bra{\ap} \sigma (\bar{u}_A) \sigma (\bar{u}_B) \ket{\tilde{\Psi}}_{\mathbb{C}}^c \nonumber\\
    &\approx 
    \bra{\tilde{\Psi}} \sigma (u_A) \sigma (u_B) 
    \ket{\tilde{\Psi}}_{\mathbb{C}}^c
    \bra{\tilde{\Psi}}
    \sigma (\bar{u}_A) \sigma (\bar{u}_B) \ket{\tilde{\Psi}}_{\mathbb{C}}^c 
    \label{eq:conn-HHLL}
\end{align}
for the connected entropy and
\begin{align}
    \bra{\Psi\Psi^\ast} \sigma (u_A) \sigma (u_B) \sigma (\bar{u}_A) \sigma (\bar{u}_B) \ket{\Psi\Psi^\ast}_{\mathbb{C}}^c 
    &= \bra{\tilde{\Psi}} \sigma (u_A) \sigma (\bar{u}_A) \sigma (u_B) \sigma (\bar{u}_B) \ket{\tilde{\Psi}}_{\mathbb{C}}^c \nonumber\\
    &=\sum_\ap \bra{\tilde{\Psi}} \sigma (u_A) \sigma (\bar{u}_A) \ket{\ap}^c_{\mathbb{C}} \!\bra{\ap} \sigma (u_B) \sigma (\bar{u}_B) \ket{\tilde{\Psi}}_{\mathbb{C}}^c \nonumber\\
    &\approx 
    \bra{\tilde{\Psi}} \sigma (u_A) \sigma (\bar{u}_A)
    \ket{\tilde{\Psi}}_{\mathbb{C}}^c
    \bra{\tilde{\Psi}}
    \sigma (u_B) \sigma (\bar{u}_B) \ket{\tilde{\Psi}}_{\mathbb{C}}^c 
\end{align}
for the disconnected entropy, where $\sum_\ap \ket{\ap}\!\bra{\ap}=\mathbf{1}$. After all, when we take $n \rightarrow 1$ the chiral eight-point function factorizes into the product of the heavy-heavy-light-light (HHLL) correlator\index{heavy-heavy-light-light conformal block} in the connected and disconnected limits.

Note that \eqref{eq:conn-HHLL} is a chiral correlator multiplied by its conjugate. This equals to the nonchiral, ordinary HHLL correlator, which indicates that only the operator insertion affects EE in the connected phase and the boundary effect plays no role other than doubling the conformal dimension due to the mirror operator. This picture is consistent with the holographic result.

\subsubsection{Computation of the chiral identity conformal block}
In the previous subsection, we have seen the entanglement entropies in the connected/ disconnected limits are given in terms of the chiral HHLL correlator $\bra{\tilde{\Psi}}\sigma_n \sigma_n \ket{\tilde{\Psi}}_{\mathbb{C}}^c$. In this subsection, we employ the identity block approximation as we discuss entanglement entropy, where $n\rightarrow 1$
%In the large-$c$ limit, such a correlator only involves the identity and its descendants as intermediate states
\cite{Asplund:2014coa,Caputa:2015waa,Sully:2020pza}. We denote the chiral identity conformal block as $G_n(z)\equiv \bra{\tilde{\Psi}}\sigma(z) \sigma(1) \ket{\tilde{\Psi}}_{\mathbb{C}}^c \propto \langle {\tilde{\Psi}}(\infty) \sigma(z) \sigma(1) {\tilde{\Psi}}(0)\rangle ^c _\mathbb{C}$.\footnote{$z$ in this subsection is the cross ratio and obviously different from $z$ in the Poincar\'e coordinates.}

%To avoid complication due to those factors, 
We will hereon focus on $\Delta \left.S_{AB}^{(n)}\right|_{\mathrm{SD}}$, the difference of (R\'enyi) entanglement entropy from that of the vacuum on the SD region $S_{AB\ \mathrm{no\, op.}}$. The vacuum entanglement entropy in the SD (i.e. the entanglement entropy without a local operator but with the induced boundary) is already computed in the previous section \eqref{eq:EEwoop}.\footnote{Note that this is different from the vacuum entanglement entropy in the CFT without boundaries; $\Delta S_{AB}\neq \Delta S_{AB}|_{\mathrm{SD}}$. This distinction becomes crucial when we discuss the first law of entanglement.}
The conformal factors cancel out in $\Delta \left.S_{AB}^{(n)}\right|_{\mathrm{SD}}$
and is given by the logarithm of the ratio of the chiral identity conformal blocks: %ore precisely speaking, the difference of entanglement entropy of an excited state by $O$ and vacuum with the original boundary $x=0$ and the induced one \eqref{disco}, is given by
\begin{align}
    \Delta \left.S_{AB}^{(n)}\right|_{\mathrm{SD}}&=\frac{1}{1-n}\log \frac{\mbox{Tr}\left.\rho_{AB}^n\right|_{\mathrm{SD}}}{\mbox{Tr}\left.\rho_{AB,\mathrm{vac}}^n\right|_{\mathrm{SD}}} \nonumber\\
    &=\frac{1}{1-n}\log 
    \begin{dcases}
    %\frac{G_n(z)}{G_n^{(0)}(z)}\quad &\text{(i)}\\
    \left\vert\frac{G_n(z_{{con}})}{G_n^{(0)}(z_{{con}})}\right\vert^2\quad &\text{(connected phase)}\\
    \frac{G_n(z_A)}{G_n^{(0)}(z_A)}\frac{G_n(z_B)}{G_n^{(0)}(z_B)}\quad &\text{(disconnected phase)}
    \end{dcases}
    ,
    \label{eq:EE-op}
\end{align}
where $G_n^{(0)}(z)=\bra{0}\sigma(z) \sigma(1) \ket{0}_{\mathbb{C}}^c$ and the cross ratio $z_{con},z_{A,B}$ is respectively given by $u_B/u_A$ and $\bar{u}_{A,B}/u_{A,B}$.

Since the holomorphic and anti-holomorphic parts are factorized in the Virasoro block \cite{Fitzpatrick:2014vua}, the chiral identity block is just a square root of the usual identity block. The ratio is given by~\cite{Asplund:2014coa}
\begin{align}
    \frac{G_n(z)}{G_n^{(0)} (z)}&= \left(\frac{1}{|\ap_O|^2} |z|^{1-\ap_O} \left|\frac{1-z^{1-\ap_O}}{1-z}\right|^2 \right)^{-\Delta_n/2} %,\quad n\rightarrow 1 
    \nonumber\\
    &=\left| \frac{z^{\ap_O /2}-z^{-\ap_O /2}}{\ap_O (z^{1/2}-z^{-1/2})}\right|^{-\Delta_n},
    \label{eq:conf-block}
\end{align}
where $\ap_O=\sqrt{1-12\frac{\tilde{\Delta}}{c}}=\sqrt{1-24\frac{\Delta_O}{c}}$. Here we emphasize again that the conformal dimension in $\ap_O$ is twice as large as the one discussed in \cite{Asplund:2014coa,Hartman:2013mia} since the local operator in the HHLL correlator in concern is $\tilde{\Psi}$, whose conformal dimension is $\tilde{\Delta}=2\Delta_O$.

By rewriting $z=e^{i\omega}$ and $\bar{z}=e^{-i\bar{\omega}}$,\footnote{$\bar{\omega}$ is not the complex conjugate of $\omega$ but is defined from $\bar{z}$.} \eqref{eq:conf-block} reduces to
\begin{equation}
    \frac{G_n(z)}{G_n^{(0)} (z)}= \left( \frac{\sin \frac{\ap_O \omega}{2}}{\ap_O \sin \frac{\omega}{2}}
    \frac{\sin \frac{\ap_O \bar{\omega}}{2}}{\ap_O \sin \frac{\bar{\omega}}{2}} \right)^{-\Delta_n/2}.
    \label{eq:conf-block-sin}
\end{equation}

\subsubsection{Connected entropy}
When $z=u_B/u_A\equiv z_{con}$ (connected phase), after the analytical continuation to the Lorentzian time,
we have
\begin{align}
    z_{con}&=\left(\frac{x_A-t+i\ap}{x_A-t-i\ap}\frac{x_B-t-i\ap}{x_B-t+i\ap}\right)^\kappa. %\nonumber\\
    %&=\left[e^{2i\left(\arccos{f_-(x_1)} - \arccos{f_- (x_2)} \right)} \right]^\kappa.
\end{align}
%
%Since $0<\arccos{f_-(x_1)} - \arccos{f_- (x_2)}<\pi$,
%\begin{equation}
%    \omega_\mathrm{conn.}=2\kappa\left(\arccos{f_-(x_1)} - \arccos{f_- %(x_2)} \right).
%\end{equation}
%
Similarly, by replacing $t\rightarrow -t$,
\begin{align}
    \bar{z}_{con}&=\left(\frac{x_A+t+i\ap}{x_A+t-i\ap}\frac{x_B+t-i\ap}{x_B+t+i\ap}\right)^\kappa. %\nonumber\\
    %&=\left[e^{2i\left(\arccos{f_+(x_1)} - \arccos{f_+ (x_2)} \right)} \right]^\kappa.
\end{align}
%Since $0<\arccos{f_+(x_1)} - \arccos{f_+ (x_2)}<\pi/2$,
%\begin{equation}
%    \bar{\omega}_\mathrm{conn.}=2\kappa\left(\arccos{f_+(x_1)} - \arccos{f_+ (x_2)} -\pi \right).
%\end{equation}
%so that the argument is within $(-2\pi,0]$.
Using the coordinate transformation \eqref{thrmap}, the cross ratios can be written in terms of the global coordinates:
\begin{equation}
    z_{con}=\left(\frac{e^{i(\theta_A-\theta_B)}}{e^{i(\tau_A-\tau_B)}}\right)^\kappa, \quad \bar{z}_{con}=\left(e^{i(\theta_A-\theta_B)}e^{i(\tau_A-\tau_B)}\right)^\kappa.
\end{equation}
We can immediately read out $\omega$ and $\bar{\omega}$ from this.
Assuming the trivial branch, we can show
\begin{equation}
    \frac{G_n(z_{con})}{G_n^{(0)} (z_{con})}= \left[\frac{\cos(\ap_O \kappa(\tau_A-\tau_B))-\cos(\ap_O\kappa (\theta_A-\theta_B))}{\ap_O^2 \left(\cos(\kappa(\tau_A-\tau_B))-\cos(\kappa (\theta_A-\theta_B))\right)}
    \right]^{-\Delta_n/2}.
    \label{eq:conn-EE-CFT}
\end{equation}
by some simple trigonometric algebra.

This analysis leads to the final expression of the growth of connected entanglement entropy:
\ba
\Delta S_{AB}^{con,{CFT}}=\frac{c}{6}\log \left[\frac{|u_1||u_2|}{|w_1-w_2|^2|u'_1||u'_2|\ap_O^2}
\cdot \left|\left(\frac{u_1}{u_2}\right)^{\f{\ap_O}{2}}-\left(\frac{u_1}{u_2}\right)^{-\frac{\ap_O}{2}}\right|^2\right]. \label{CFTconEE}
\ea

On the other hand, our previous gravity dual result \eqref{connected HEE} can be rewritten as follows:
\ba
\Delta S^{con,AdS}_{AB}=\frac{c}{6}\log\left[\frac{\left|\left(\frac{u_1}{u_2}\right)^{\frac{\chi}{2\kappa}}-\left(\frac{u_1}{u_2}\right)^{-\frac{\chi}{2\kappa}}\right|^2}{\chi^2\left|\left(\frac{u_1}{u_2}\right)^{\frac{1}{2\kappa}}-\left(\frac{u_1}{u_2}\right)^{-\frac{1}{2\kappa}}\right|^2}\right].
\label{AdSconEE}
\ea

%Note that since
%\ba
%\frac{\chi}{\kappa}=2\s{1-\frac{12\Delta_{AdS}}{c}}-1\approx %1-\frac{12\Delta_{AdS}}{c}\approx %\sqrt{1-24\frac{\Delta_{AdS}}{c}},
%\label{eq:leading-approx}
%\ea
%the CFT calculation here and the holographic result matches to %the leading order in $\Delta_O/c$ if we take %$\Delta_{AdS}=\Delta_O$. This is consistent with %\eqref{eq:leading-approx}. 

The BCFT result (\ref{CFTconEE}) exactly agrees with the holographic one  by identifying the parameters as follows:
(\ref{AdSconEE}) 
\begin{equation}
    \ap_O=\frac{\chi}{\kappa}.   \label{Identifyap}
\end{equation}
Since $\kappa=\eta_0$, this is precisely what we have observed for $\chi$ under the rescaling by $\eta_0$ \eqref{eq:chi-transform}.
Thus, it implies $12\frac{\tilde{\Delta}}{c}=\frac{M^\prime}{R^2}$ when $\eta=\eta_0$. This is exactly equivalent to the relation we argued \eqref{eq:AdS EM tensor}. In this way, We have explicitly confirmed the perfect matching between the BCFT and the gravity dual by looking at entanglement entropy as well as the stress tensor.

To see the first law of entanglement, we need to subtract the vacuum contribution in the CFT given by
\begin{equation}
    S_{AB}^{{vac}}=\frac{c}{6}\log \frac{|x_B-x_A|^2}{\epsilon^2}
\end{equation}
from the connected entropy discussed above.\footnote{The vacuum contribution is equivalently given by the full contribution with $\Delta_O=0$ and $\kappa=1$.} What we computed in this subsection is $\Delta \left.S_{AB}\right|_\mathrm{SD}$, which is the full contribution $S_{AB}$ minus the contribution with the induced boundary but with no operator $S_{AB\ \mathrm{no\, op.}}$.
We can explicitly confirm that $\Delta S_{AB}%\equiv S^{con,CFT}_{AB}-S_{AB}^{{vac}}
=\Delta_O \left.S_{AB}\right|_\mathrm{SD} + S_{AB\ \mathrm{no\, op.}} -S_{AB}^{{vac}}$ in the short distance limit satisfies the first law with the stress tensor \eqref{EMgencft} within the current CFT.
Via the first law (\ref{firstL}), \eqref{CFTconEE} leads to the energy density 
\ba
T_{tt}=\frac{c}{\pi}\left[-\frac{1}{6}(\kappa^2-1)+4\kappa^2\frac{\Delta_O}{c}\right]\cdot H(t,x),
\ea
which perfectly agrees with the stress tensor (\ref{EMgencft}), which was directly computed from the CFT.

\subsubsection{Disconnected entropy}
When $z=\bar{u}/u\equiv z_{dis}$ (disconnected phase), after the analytical continuation to the Lorentzian time,
we have
\begin{align}
    z_{dis}=\bar{z}_{dis}^{-1}&=\left(\frac{x-t-i\ap}{x-t+i\ap}\frac{x+t-i\ap}{x+t+i\ap} e^{-2i\varphi} \right)^\kappa \nonumber\\
    %&=\left[e^{2i\left(-\arccos{f_-(x)} - \arccos{f_+ (x)} + \varphi\right)} \right]^\kappa.
\end{align}

Using the coordinate transformation \eqref{thrmap}, the cross ratio can be written in terms of the global coordinates. In contrast to the connected entropy, we leave the branch choice from $w^\prime$ to $u$ and the cross ratio unspecified. Then,\footnote{Although $z$ and $\bar{z}$ can independently have a different branch, we consider $\omega=\bar{\omega}$ here, which turns out to be the one corresponding to the holographic entanglement entropy.}
\begin{equation}
    z_{dis}=e^{2\kappa i(\theta+\varphi-l\pi)} e^{-2\pi m} \Rightarrow \omega_{dis}=2\kappa \left(\theta+\varphi-\pi l -\frac{\pi}{\kappa} m \right) \quad (l,m\in\mathbb{Z}).
\end{equation}
It follows that
\begin{align}
    &\frac{G_n(z_{dis})}{G_n^{(0)} (z_{dis})} \nonumber\\
    =&\left( \frac{
    \displaystyle\sin \frac{\ap_O \omega_{dis}}{2}
    }{
    \displaystyle\ap_O \sin \frac{\omega_{dis}}{2}
    }
    \frac{
    \displaystyle\sin \frac{\ap_O \bar{\omega}_{dis}}{2}
    }{
    \displaystyle\ap_O \sin \frac{\bar{\omega}_{dis}}{2}
    } \right)^{-\Delta_n/2} %\nonumber\\
    =\left(
    \frac{\sin \left(\ap_O \kappa \left(
    \displaystyle\theta -\pi l +\varphi - \frac{\pi}{\kappa} m  
    \right)\right)}{\ap_O \sin \left( 
    \displaystyle\kappa \left(\theta -\pi l +\varphi - \frac{\pi}{\kappa} m  \right)\right)}
    \right)^{-\Delta_n} 
    \label{eq:dis-first-br}\\
    =&\left( \frac{
    \displaystyle\sin \frac{\ap_O (-\omega_{dis})}{2}
    }{
    \displaystyle\ap_O \sin \frac{-\omega_{dis}}{2}
    }
    \frac{
    \displaystyle\sin \frac{\ap_O (-\bar{\omega}_{dis})}{2}
    }{
    \displaystyle\ap_O \sin \frac{-\bar{\omega}_{dis}}{2}
    } \right)^{-\Delta_n/2} %\nonumber\\
    =\left(
    \frac{\sin \left(\ap_O \kappa \left(
    \displaystyle \pi l -\theta -\varphi +\frac{\pi}{\kappa} m
    \right)\right)}{
    \displaystyle \ap_O \sin \left(\kappa \left(\pi l -\theta -\varphi +\frac{\pi}{\kappa} m \right)\right)
    }
    \right)^{-\Delta_n}.
    \label{eq:dis-next-br}
\end{align}
As we have discussed for the connected entropy, the identification (\ref{Identifyap}) $\ap_O=\chi/\kappa$ reproduces the holographic result for the disconnected entropy as well by choosing $(l,m)=(-1,1)$ for \eqref{eq:dis-first-br} and $(l,m)=(1,0)$ for \eqref{eq:dis-next-br}; each choice of the branch corresponds to $\theta^{min}=\theta$ and $(2-1/\chi)\pi -\theta$, respectively. $l$ corresponds to choosing an appropriate branch so that $0\le \arg u\le \pi$ is satisfied; This is nothing but choosing an appropriate $\theta^{min}$. To justify the branch choices of $m$ within the CFT, we need a careful consideration of the monodromy, which we discuss in Appendix \ref{app:branch}.

%where
%\begin{equation}
%    f_{\pm} (x)= \frac{x\pm t}{\sqrt{(x\pm t)^2 +\ap^2}}.
%\end{equation}
%By some calculations, we can show $-\pi<-\arccos{f_-(x)} - \arccos{f_+ (x)}<0$ for $\forall t\in \mathbb{R}, x>0$. Since $\varphi \ll 1$ in the large-c limit (sparse spectrum), this inequality is basically holds. \textcolor{red}{(But it slightly shifts the inequality. Doesn't it affect the final result?)} Since when we derived the conformal map \eqref{eq:conf-map}, we assumed the argument is within the fundamental domain $[0,2\pi]$, we should choose the branch so that this is satisfied. Consequently,
%\begin{equation}
%    \omega_\mathrm{discon.} (x)=2\kappa \left(\varphi-\arccos{f_-(x)} - \arccos{f_+ (x)} + \pi \right).
%\end{equation}
%Since $\bar{z}_\mathrm{discon.}=1/z_\mathrm{discon.}$, $\bar{\omega}_\mathrm{discon.}=-\omega_\mathrm{discon.}$.

\subsection{Entanglement entropy after rescaling by $\eta$}
As it has been mentioned in the previous sections, entanglement entropy after an arbitrary rescaling by $\eta$ \eqref{rescaleex} is simply given by the replacement $\chi\rightarrow\chi/\eta$. Since the entanglement entropy is written in terms of $\ap_O \kappa$, it becomes $\ap_O \kappa/\eta$ and this is consistent with $\chi$ becoming $\chi/\eta$ in holographic entanglement entropy.

\chapter{Entanglement in quantum many-body systems}\label{ch:3}
\renewcommand{\thesection}{\thechapter.\arabic{section}}
%!TEX root = ../thesis.tex
%*******************************************************************************
%****************************** Third Chapter **********************************
%*******************************************************************************
%\chapter{My third chapter}

%Entanglement in quantum many body systems
% **************************** Define Graphics Path **************************
\ifpdf
    \graphicspath{{Chapter3/Figs/Raster/}{Chapter3/Figs/PDF/}{Chapter3/Figs/}}
\else
    \graphicspath{{Chapter3/Figs/Vector/}{Chapter3/Figs/}}
\fi

\renewcommand{\thesection}{\thechapter.\arabic{section}}
\setcounter{section}{0}
\textit{This chapter contains both reviews and original materials. Useful reviews on tensor networks include~\cite{1130282270516945024,Bridgeman:2016dhh}. 
Section \ref{sec:tensor-network} is mostly reprinted from my second-year student research report and devoted to the review of tensor networks. Section \ref{sec:HED} and subsequent sections follow our original work with Hidetaka Manabe, who mainly performed numerical simulations, and Hiroaki Matsueda~\cite{Mori:2022xec}. The author of this dissertation has contributed to the original proposal of geometric entanglement distillation by pushing the bond cut surface to the minimal surface, revealing its relation to entanglement distillation from pseudo entropy, analytic calculations in multi-scale entanglement renormalization ansatz and matrix product states, proposing a numerical setup, leading the project as a first author, and writing the paper.}\\

So far we consider modifications of field theories by adding mass and interactions, or boundaries and excitations. In this chapter, we are motivated to reformulate holography itself based on tensor network from an operational perspective rather than making modifications. 
Our guiding principle is holographic entanglement entropy \eqref{eq:HEE-RT}, however, this formula itself only tells us about entanglement entropy (EE) and no more. Thus, it is interesting to consider if we can geometrically realize some information-theoretic operations to see the structure of entanglement shared among unit cells of quantum many-body states.
As a first step, we carefully examine the relationship between the geometry of tensor networks and the entanglement structure and interpret a certain geometric operation as entanglement distillation. We also consider a measure for how well the geometric operation can serve as entanglement distillation. Note that our target is arbitrary states represented by tensor networks and not limited to ground states.

In this chapter, we relate a geometric procedure using tensor networks to an operation of concentrating entanglement, namely, entanglement distillation~\cite{Mori:2022xec}. 
Section \ref{sec:intro-TN} provides our motivation to discuss tensor networks from an operational perspective toward a generalization of holography. We will also give a short summary of previous attempts and missing pieces. 
Section \ref{sec:tensor-network} is devoted to the review of tensor network methods. We provide definitions for matrix product states (MPS), projected entangled-pair states, and multi-scale entanglement renormalization ansatz (MERA). Next, in Section \ref{sec:HED}, we introduce our geometric distillation method using MERA as an example. We introduce reduced transition matrices and pseudo entropy to justify our proposal as an approximate entanglement distillation. We also propose a quantitative measure to evaluate how well this geometric approach works as entanglement distillation.
Then, in Section \ref{sec:HED-num} we numerically demonstrate our proposal in the random MERA. In the last section Section \ref{sec:HED-MPS}, we extend our proposal to an MPS in a canonical form and point out the structural similarity with MERA. Appendix \ref{app:TN-rep} explains the diagrammatic notations used in this chapter. 

\section{Motivation for the operational aspects of tensor networks}\label{sec:intro-TN}
Holographic duality was strongly motivated by the Ryu-Takayanagi (RT) formula \eqref{eq:HEE-RT}, where EE in holographic QFT is proportional to the area of the minimal surface in its gravity dual called the RT surface~\cite{Ryu:2006bv}. 
The RT formula is essentially a holographic extension of the famous Bekenstein-Hawking formula for black hole entropy~\cite{Bekenstein:1973ur,Bardeen:1973gs,Hawking:1975vcx}. A very important feature of a black hole is the presence of radiation of Hawking pairs inside and outside the event horizon. Then, the theory is described by the Bogoliubov transformation in superconductivity to connect both sides of the event horizon. Thus, the RT surface should be characterized by the condensation of entangled pairs from elementary objects, which have critical information about the holographic spacetime. Meanwhile, as we have seen in Section \ref{sec:EE-op-based}, EE can be defined operationally in quantum information. Entanglement entropy of a state asymptotically equals %to 
the number of extractable Einstein–Podolsky–Rosen (EPR) pairs via LOCC in the limit of a large number of state copies. This procedure of extraction is called entanglement distillation \eqref{eq:distillation}.

To further clarify the information-theoretic aspect of holography, it is important to understand how the RT formula is derived from the operation-based definition of EE. However, the previous derivation of the RT formula~\cite{Lewkowycz:2013nqa} relies on the state-based definition and the relation to the operational definition remains unclear (although see \cite{Bao:2018pvs} for some progress regarding one-shot entanglement distillation). {The} bit thread {formalism}, a mathematically equivalent formulation of the RT formula, suggests EPR pairs across the RT surface~\cite{Freedman:2016zud}. While this picture supports the operational definition in holography, the physical origin of the EPR pairs is still unknown in contrast to the case of a black hole.
Although the RT formula is derived in holography based on the gravitational path integral~\cite{Lewkowycz:2013nqa}, such an operational derivation  of the RT formula has not been accomplished (although see \cite{Bao:2018pvs} for some progress regarding one-shot entanglement distillation). The mathematically equivalent formulation of the RT formula called bit threads~\cite{Freedman:2016zud} supports this operational picture where EPR pairs are condensed across the RT surface, nevertheless, the physical origin of these pairs is not clear in contrast to the case of black hole.
In this chapter, we address this issue in terms of quantum operational techniques based on tensor networks, which we will describe in Section \ref{sec:tensor-network}.

A recent proposal for a better understanding of the RT surface stated that the maximally entangled states characterize the surface~\cite{Freedman:2016zud,Agon:2018lwq,Agon:2021tia,Rolph:2021hgz,Chen:2018ywy,Cui:2015pla}. The proposal suggests that the surface may emerge from entanglement distillation by a deformation of the boundary. 
One of the goals of this chapter is to provide a concrete method to achieve this procedure in the MERA and discuss a possible extension to other tensor networks such as MPS.
%\red{Intuitively vs. Precisely speaking -> should remedy the whole paragraph?}
%
%Motivated by the possible relationship with distillation and the minimal bond cut surface, we examine geometric operations in tensor networks with and without a holographic direction. 
In special circumstances, previous literature has established the relation between the discrete version of the RT formula and a (one-shot) entanglement distillation in tensor networks. These tensor networks are perfect~\cite{Pastawski:2015qua} or special tree tensor networks~\cite{Bao:2018pvs,Lin:2020yzf,Yu:2020zwk,Lin:2020ufd}. Using the isometric property of their composing tensors, we can show the state equals a collection of EPR pairs across the minimal surface via the so-called greedy algorithm.
However, these tensor networks are still inadequate to achieve conformally invariant states, which are usually assumed in holography. 
For instance, a correlation function in perfect tensor networks does not decay as the distance increases and its entanglement spectrum is flat. This is contradictory to the result for CFTs. Thus, we mainly focus on MERA in this chapter as it is known to efficiently approximate critical ground states. Furthermore, it has a capacity to express various wave functions via a variational optimization, which is also missing in the holographic tensor network toy models in previous literature. Despite MERA being neither a perfect nor a tree tensor network, our method enables us to discuss entanglement distillation in the MERA. Moreover, we claim that the methodology is also applicable to an MPS.
There is no direct bulk/boundary correspondence in MPS since it lives on the lattice of our target model. Nevertheless, when we define a partial system, a minimal bond cut surface can always be defined as the edge of the partial system. By appropriately distilling over each matrix, we can find a state close to the EPR pair. Our goal is to show analytical and numerical evidence for these procedures in relation to a minimal bond cut surface and EPR pairs.

\section{Tensor networks}\label{sec:tensor-network}
In this section, we briefly review the tensor network method, a geometrical representation of variational wave functions.\footnote{We focus on the Hamiltonian formulation here. But we can make a similar discussion in the Lagrangian formalism by considering a cut in the partition function to define a reduced density or transition matrix (which will be defined later).}
Refer to Appendix \ref{app:TN-rep} for the definition of the graphical expressions employed in this chapter.

\subsection{Quantum many-body system}
A \textbf{quantum many-body system}\index{quantum many-body system} is described by a linear combination of a tensor product of states at each site. \underline{In this report, we consider a \emph{finite}-dimensional Hilbert} \underline{space unless otherwise noted.} Each site $j$ on a lattice (composed of $N$ sites) is a $d$-level quantum system equipped with $d$ dimensional Hilbert space $\mathcal{H}_j$ (for instance $\frac{d-1}{2}$-spin system\footnote{The representation of the state is of course not restricted to SU(2).}) with the basis $\{\ket{s_j}\}_{\scaleto{s_j=1,2,\cdots d\mathstrut}{9pt}}$, called \emph{qudits}\footnote{They are called qubits in $d=2$ and qutrits in $d=3$.}. Then any \emph{pure} state $\rho_\Psi=\dyad{\Psi}$ of the lattice is given by
\begin{align}
| \Psi \rangle = \sum _ { s _ { 1 } ,\cdots  s _ { N } = 1 } ^ { d } \Psi ^ { s _ { 1 }  \cdots s _ { N }} \bigotimes _ { j } \left| s _ { j } \right\rangle
\qq{where} \bigotimes_s \ket{s_j}&\equiv\ket{s_1}\otimes\ket{s_2}\otimes\cdots\otimes\ket{s_N}
\label{eq:manybody}\\
&\equiv\ket{s_1 s_2 \cdots s_N}.\nonumber
\end{align}
$\ket{\Psi}$ is characterized by $d^N$ complex coefficients $\Psi _ { s _ { 1 }  \cdots s _ { N }} \in \mathbb{C}$. This is nothing but the many-body wave function in the basis $\{\bigotimes_s \ket{s_j}\}_j$. We can regard these coefficients as a rank $N$ tensor if the basis is fixed\footnote{Usually there is a U$(d^N)$ gauge redundancy of the basis transformation. If we stick to the site-basis representation of the state, then the state is already gauge fixed site by site thus we only need U$(d)$ gauge fixing.}.
As long as the basis is fixed and known, specifying this tensor determines the target pure state. We intend to decompose the tensor into some pieces of smaller tensors. This collection of tensors is called \textbf{tensor network (TN)} \index{tensor network}\index{TN}. The motivation of tensor network is originally an economical computation of wave functions of interest. As we can see from the rank of the tensor in \eqref{eq:manybody}, the computational cost of any observables in a many-body system, such as a norm, correlation functions and expectation values, scales exponentially with the system size. The purpose of the tensor network representation of wave functions (mostly in condensed matter physics and the lattice gauge theory) is to suppress the exponential cost to a polynomial order. 

Although a typical usage of tensor networks in condensed matter physics is reducing the cost in numerical computations, the philosophy of tensor networks is very useful in examining the entanglement structure even analytically.
A tensor network can be constructed by contracting internal bonds between tensors defined on each lattice site of our model system. The hidden degrees of freedom carried by the internal bonds represent how nonlocal quantum entanglement is shared between two distant sites. By controlling the dimension, we can sequentially increase the resolution of the variational optimization to obtain the true ground state. Mathematically, the network of the tensors is reformulated as projected entangled-pair states (PEPS)~\cite{Verstraete:2004cf}, in which we define maximally entangled states among artificial degrees of freedom on each bond and then {take} some physical mapping on each site. The well-known matrix product state (MPS) is a one-dimensional (1D) version of PEPS.
Since holography, in particular, the RT formula also suggests encoding the entanglement structure into some geometry, it has been discussed if
a certain class of tensor networks can simulate the structure of holography, e.g. see~\cite{Swingle:2009bg,Pastawski:2015qua,Hayden:2016cfa}. 
Our intention is to study tensor networks from an entanglement perspective to discuss a more general framework of holography, not limited to the aforementioned specific models.
For this purpose, we review several representative classes of tensor networks in the following subsections. We mostly focus on the $(1+1)$-dimensional systems for simplicity.

\subsection{Matrix product states (MPS)}\label{sec:mps}
In this section, the most fundamental class of tensor network called a \textbf{matrix product state (MPS)}~\cite{Ostlund:1995zz} is introduced.

The singular value decomposition (SVD) \eqref{eq:svd} can decompose each tensor into smaller tensors. By performing the SVDs $n$ times on a one-dimensional $n$-body state, we obtain a chain of matrices. This class of tensor network is called a \textbf{matrix product state (MPS)}.\index{matrix product state}\index{MPS |see matrix product state } By truncating the bond dimension, i.e., essentially applying low-rank approximations to each matrix, the MPS becomes an efficient representation to some one-dimensional quantum system. The procedure is as follows~\cite{Bridgeman:2016dhh}.
\begin{enumerate}
\item Consider a one-dimensional $n$-body state. Each external leg shows the degrees of freedom at each site.
\begin{equation}
| \psi \rangle = \sum _ { s _ { 1 } ,\cdots  s _ { N } = 0 } ^ { d-1 } \psi ^ { s _ { 1 }  \cdots s _ { N }} \bigotimes _ { j } \left| s _ { j } \right\rangle= \includegraphics[valign=c, clip, width=2.5cm]{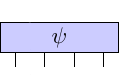}
\end{equation}

\item Perform the SVD ($\psi$ as a $d\times d^{N-1}$ matrix) from the left and separate the leftmost site from others:
\begin{equation}
\includegraphics[valign=c, clip, width=2.5cm]{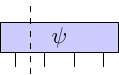}=\includegraphics[valign=c, clip, width=4.2cm]{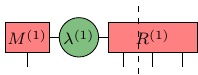}.
\end{equation}

\item In the same manner, perform the SVD ($R^{(1)}$ as a $dr^{(1)} \times d^{N-2}$ matrix) to separate the leftmost site in $R^{(1)}$ ($r^{(1)}$ is the dimension of the internal leg between $M^{(1)}$ and $R^{(1)}$):
\begin{equation}
\includegraphics[valign=c, clip, width=4.2cm]{mps3-mast.pdf}=\includegraphics[valign=c, clip, width=6.2cm]{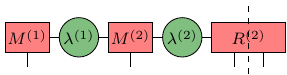}.
\end{equation}

\item Repeat the procedure above until all sites are separated:
\begin{align}
\includegraphics[valign=c, clip, width=6.2cm]{mps4.pdf}=&\includegraphics[valign=c, clip, width=7.8cm]{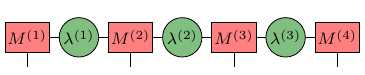}.\nonumber\\
&\hspace*{3.5cm}(N=4\ \mathrm{case})
\end{align}

\item Finally merge a singular value matrix $\lambda^{(i)}$ into each nearest matrix $M^{(i)}$ at each site\footnote{One could merge $\lambda^{(i)}$ into $M^{(i-1)}$ or $\sqrt{\lambda^{(i)}}$ into both $M^{(i-1)}$ and $M^{(i)}$ instead.}.
\begin{equation}
\includegraphics[valign=c, clip, width=6cm]{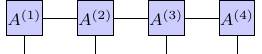}
\end{equation}

\item To reduce the exponential cost, we truncate the dimension of all the internal legs by the bond dimension $\chi \le \min_{i}(r^{(i)})$.
\end{enumerate}
Although in the procedure above, we performed the SVDs from the left, it does not matter where one starts the SVD from.

To sum up, the MPS approximation with the bond dimension $\chi$ of a $n$-body state (including the summation of external legs with the basis) is given by
\begin{equation}
\hspace*{-1cm}
\ket{\psi}_{1\cdots N}=\sum_{ s _ { 1 } ,\cdots  s _ { N } = 0 }^{d-1} \sum_{\alpha_1, \cdots \alpha_N=1}^\chi A\indices{^{(1)\, s_1}_{\alpha_1}} A\indices{^{(2)\, s_2\, \alpha_1} _{\alpha_2}} \cdots A\indices{^{(N-1)\, s_{N-1}\, \alpha_{N-1}} _{\alpha_N}} A^{(N)\, s_N\, \alpha_N} \ket{s_1 s_2 \cdots s_{N-1} s_N}_{1\cdots N},
\end{equation}
where $A\indices{^{(i)\, s_i\ \alpha} _{\beta}}$ is a projector at site $i$ from two virtual indices $\alpha$ and $\beta$ to a physical index $s_i$.
When the external indices are fixed, the wave function (a component of a state) is a matrix product of $A$'s This is why the state is called the \emph{matrix product} state. Strictly speaking, two $A$'s at two endpoints are vectors, not matrices. We often split the leftmost and rightmost vectors into sets of a matrix and a vector:
\begin{equation}
A\indices{^{(1)\, s_1}_{\alpha_1}}=:B_{\alpha_0} A\indices{^{(1)\, s_1\, \alpha_0} _{\alpha_1}},
\end{equation}
where $\bra{B_0}=\bra{\bm{e}}^{\alpha_0} B_{\alpha_0}$ is the leftmost \emph{boundary} covector, which specifies the left boundary condition. Similarly,
\begin{equation}
A^{(N)\, s_N\, \alpha_N} =:A\indices{^{(N)\, s_N\, \alpha_N} _{\alpha_{N+1}}} B^{\alpha_{N+1}},
\end{equation}
where $\ket{B_{N+1}}=B^{\alpha_{N+1}}\ket{\bm{e}}_{\alpha_{N+1}}$ is the rightmost \emph{boundary} vector, which specifies the right boundary condition.

In this way we can separate the boundary expressed as two vectors at the end from the bulk expressed as a chain of matrices with external legs $\{s_i\}$:
\begin{equation}
\ket{\psi}_{1\cdots N}=\sum_{ s _ { 1 } ,\cdots  s _ { N } = 0 }^{d-1} \sum_{\alpha_0, \cdots \alpha_{N+1}=1}^\chi \bra{B_0} A\indices{^{(1)\, s_1}} \cdots A^{(N)\, s_N} \ket{B_{N+1}} \ket{s_1 \cdots s_N}_{1\cdots N},
\end{equation}
where $\bra{B_0}$ and $\ket{B_{N+1}}$ are the boundary vectors with one virtual bond; $A\indices{^{(1)\, s_1}} \cdots A^{(N)\, s_N}$ are the bulk described by a product of matrices with virtual bonds, which are implicitly contracted, and an external leg (physical degrees of freedom) for each; $\ket{s_1 s_2 \cdots s_{N-1} s_N}$ is the bulk basis vector.\footnote{The terminology \textit{bulk} and \textit{boundary} here is obviously different from that in holography.}

\subsubsection{MPS with a mixed boundary state conditions}
So far we considered a state with an \emph{open} boundary condition. For the MPS with a general boundary condition, a probabilistic mixture of boundary states, i.e., a mixed boundary state $Q=\ket{\bm{e}}_{\alpha_{N+1}}Q\indices{^{\alpha_{N+1}}_{\alpha_0}}\bra{\bm{e}}^{\alpha_0}$, gives the most general setup for the MPS state. It is important to include this boundary operator $Q$ to the variational parameters. This makes the MPS representation nontrivial: For example, even if the system is translationally invariant, each tensor component is not translationally invariant. $Q$ consists of some important information about the state such as the quantum number.
\begin{definition}
The \textbf{matrix product state (MPS)}\index{matrix product state}\index{MPS |see matrix product state } with a boundary condition given by a boundary operator $Q$ is given by
\begin{equation}
\ket{\psi}_{1\cdots N}=\sum_{ s _ { 1 } ,\cdots  s _ { N } = 0 }^{d-1} \Tr_{\mathrm{bond}} (A\indices{^{(1)\, s_1}} \cdots A^{(N)\, s_N} Q) \ket{s_1 \cdots s_N}_{1\cdots N}.
\label{eq:mps}
\end{equation}
\end{definition}
This definition of the MPS naturally includes a state with an open boundary condition when $Q$ is a pure state, $Q=\ket{B_{N+1}}\bra{B_0}$, and a periodic boundary condition when $Q$ is a completely mixed state, $Q=\mathbf{1}$.

%normalization
\subsubsection{Normalization and canonical form in MPS}
When a tensor network contains \emph{no loop}, the normalization condition $\braket{\psi}=1$ can be made into a local condition. For example, the norm of a MPS with an open boundary condition $\ket{\psi}$ given in (\ref{eq:mps}) is
\begin{align}
&\phantom{=}\braket{\psi}_{1\cdots N}\nonumber\\
&=\sum_{ s _ { 1 } ,\cdots  s _ { N } = 0 }^{d-1} \bra{B_0^\ast} A\indices{^{(1) \dagger}_{s_1}} \cdots A\indices{^{(N) \dagger}_{s_N}} \ket{B_{N+1}^\ast} \bra{B_0} A\indices{^{(1)\, s_1}} \cdots A^{(N)\, s_N} \ket{B_{N+1}}\nonumber \\
&=\left(\bra{B_0} \otimes \bra{B_0^\ast}\right) \left(\sum_{ s _ { 1 } = 0 }^{d-1} A\indices{^{(1)\, s_1}} \otimes A\indices{^{(1) \dagger}_{s_1}}\right) \cdots \left(\sum_{ s _ { N } = 0 }^{d-1} A^{(N)\, s_N} \otimes A\indices{^{(N) \dagger}_{s_N}}\right) \left(\ket{B_{N+1}} \otimes \ket{B_{N+1}^\ast}\right).
\end{align}
Note $\dagger$ is taken for the physical indices only, otherwise it is just a complex conjugation. Graphically, the norm is first given by
\begin{equation}
\includegraphics[valign=c, clip, width=9cm]{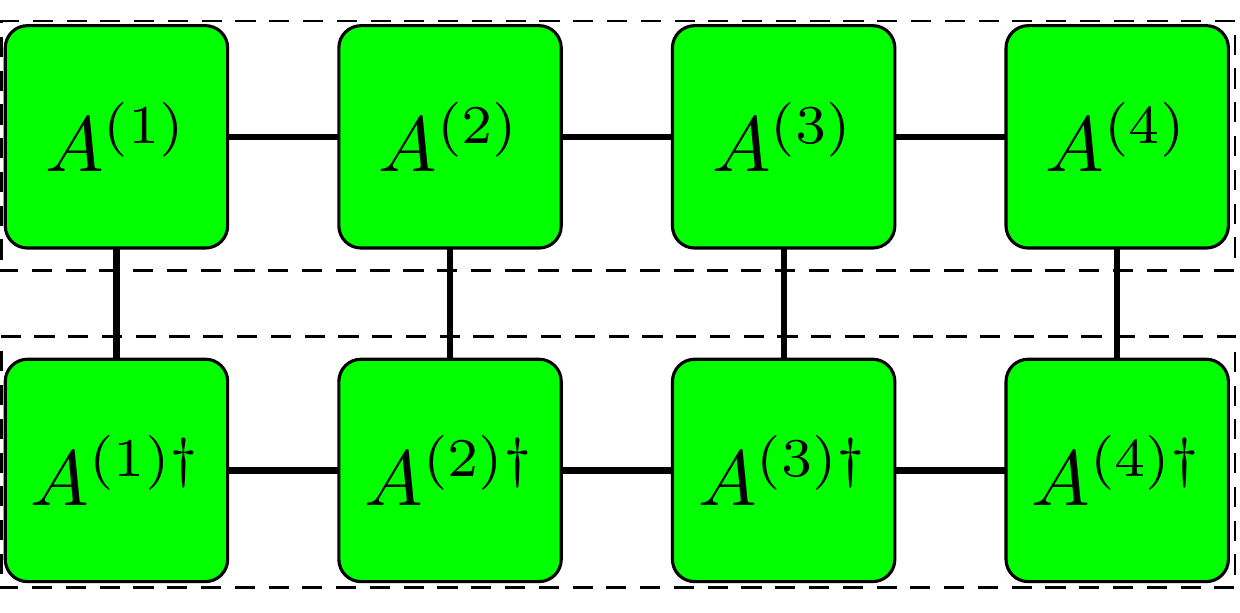}\ .
\end{equation}
(The boundary vectors are absorbed into the bulk matrices at the end.)
Then, we contracted physical indices first and took the trace of a matrix product of $\chi^2 \times \chi^2$ matrices $\{ \sum_{ s _ { i } = 0 }^{d-1} A\indices{^{(i)\, s_i}} \otimes A\indices{^{(i) \dagger}_{s_i}} \}$:
\begin{equation}
\includegraphics[valign=c, clip, width=9cm]{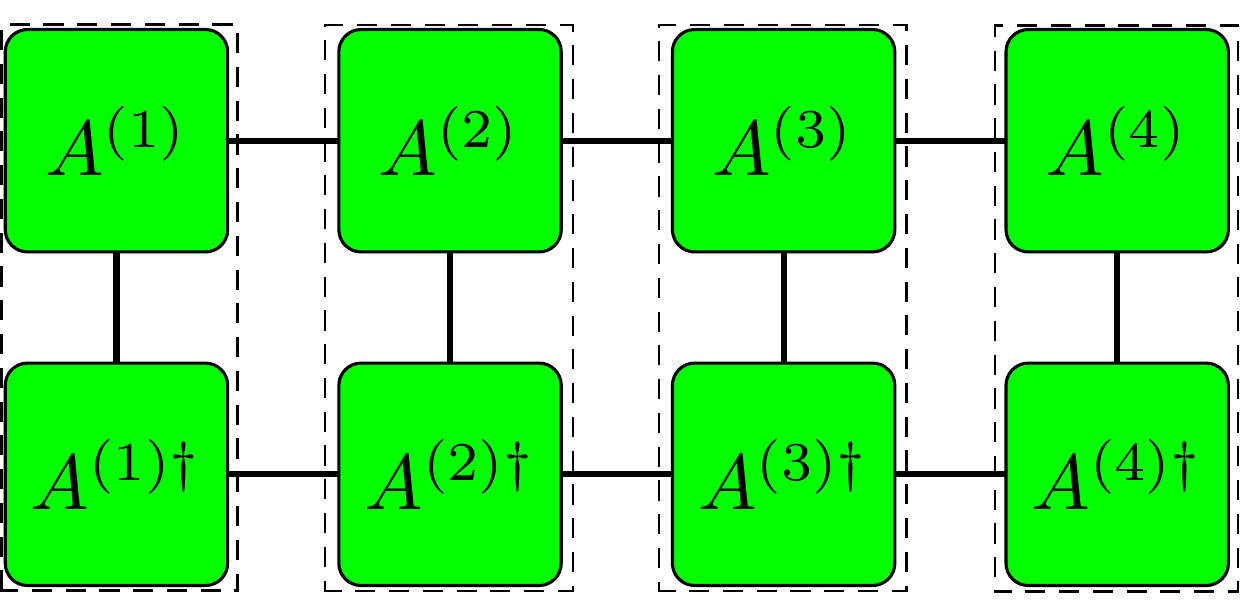}.
\end{equation}
%canonical form
We can compute the norm efficiently by gauge fixing, i.e., imposing a \emph{local} normalization condition on $\sum_{ s _ { i } = 0 }^{d-1} A\indices{^{(i)\, s_i}} \otimes A\indices{^{(i) \dagger}_{s_i}}$. We restrict to the case of an MPS with an \emph{open} boundary condition. Although it is necessary to use the SVD for the gauge fixing to a canonical form as follows, we cannot divide the tensor network in two if the tensor network has a loop. Since the periodic MPS involves a loop, no canonical form is possible.\index{canonical form}
\begin{itemize}
\item \underline{Left-canonical form}\\
We can perform the SVD on the leftmost vector as $d\times \chi$ matrix by the gauge redundancy. It gives
\begin{align}
\includegraphics[valign=c, clip]{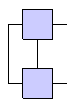}
&=\sum_{ s _ { 1 } = 0 }^{d-1} A\indices{^{s_1}_{\alpha}} \otimes A\indices{^{\dagger}_{s_1\, \beta}}\\
&=U\indices{_\alpha ^{\tilde{\gamma}}}\sigma_{\tilde{\gamma}} V\indices{^\dagger _{\tilde{\gamma}} ^{s_1}} V\indices{_{s_1} ^{\tilde{\rho}}} \sigma_{\tilde{\rho}} U\indices{^\dagger_{\tilde{\rho}\, \beta}}\\
&=\left(\sum_{\mathrm{bond}} \sigma_{\tilde{\alpha}}^2\right) \mathbf{1}_{\mathrm{bond}}=\left(\sum_{\mathrm{bond}} \sigma_{\tilde{\alpha}}^2\right)\includegraphics[valign=c, clip]{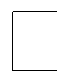}.
\label{eq:canonical}
\end{align}
Then we can `close the zipper' from left to right, iteratively using (\ref{eq:canonical}) every time. After the first SVD, the tensor contraction involves one additional internal leg contraction. We regard $A$ as $d\chi \times \chi$ matrix for the SVD. From the normalization condition $\braket{\psi}=1$, we can fix the numerical factor $\left(\sum_{\mathrm{bond}} \sigma_{\tilde{\alpha}}^2\right)^N=1$. As a result, we have the \textbf{left-canonical form}\index{left-canonical form}, in which each matrix is imposed the \textbf{left normalization}\index{left normalization} (\ref{eq:canonical}) as a gauge fixing condition (the \textbf{left gauge}\index{left gauge}). (\ref{eq:canonical}) implies $A$ after gauge fixing is an isometric tensor $A^\dagger A=\mathbf{1}$. Thus the left-canonical form is also known as the \textbf{left-isometric form}. By the QR decomposition\index{QR decomposition}, in which a matrix is decomposed into a product of an isometric matrix $Q$ and an upper-triangular matrix $R$, one can similarly show we can take the left-canonical form for any MPS\footnote{Usually, the QR decomposition is computationally faster than the SVD.}.

\item \underline{Right-canonical form}\\
Where one starts gauge fixing is arbitrary. If one starts from the rightmost matrix, one obtains the \textbf{right-canonical} or \textbf{right-isometric} form.\index{right-canonical form}\index{right-isometric form}
The \textbf{right gauge}\index{right gauge} is given by the \textbf{right normalization}:\index{right normalization}
\begin{equation}
\scalebox{-1}[1]{\includegraphics[valign=c, clip]{canonical1.pdf}}=\scalebox{-1}[1]{\includegraphics[valign=c, clip]{canonical2.pdf}}.
\end{equation}

\item \underline{Mixed-canonical form}\\
Similarly, one can start gauge fixing from the both ends. Finally one ends up one's contraction at a site or on a bond in the middle. The former is called the \textbf{(site-centered) mixed-canonical form}\index{site-centered mixed-canonical form}\index{mixed-canonical form} or \textbf{site-canonical form}\index{site-canonical form} and the latter is called the \textbf{(bond-centered) mixed-canonical form}\index{bond-centered mixed-canonical form} or \textbf{bond-canonical form}\index{bond-canonical form}~\cite{SCHOLLWOCK201196}. The site or bond in the middle is known as the \textbf{orthogonality center}\index{orthogonality center} or \textbf{center of orthogonality}\index{center of orthogonality}. The mixed-canonical form is commonly used in numerical computations.
\end{itemize}
Any canonical forms are of course equivalent under the gauge transformations (by matrix multiplications and matrix redefinitions). Gauge fix to these canonical forms drastically reduce the computational cost to solve the norm of a state, the expectation value of an operator, the correlation function, etc. Note that these gauges are partial gauge fixing. The unitarity is used in the derivation and thus there still exists a gauge freedom of local basis transformations by a unitary matrix. This remaining freedom is used to increase numerical stability.

%approximation error
The MPS representation reduces the computational cost instead of allowing some approximation errors due to the truncation by the bond dimension $\chi$. The error by the SVD is evaluated as follows. Consider the SVD of a matrix $\psi$:
\begin{equation}
\psi=M\indices{^{(1)\, s_1}_{\alpha_1}} \lambda_{\alpha_1} R\indices{^{(1)\, s_2 \cdots s_N\, \alpha_1}}.
\end{equation}
The approximation error is
\begin{align}
\norm{\psi - \sum_{\alpha_1=1}^\chi M\indices{^{(1)\, s_1}_{\alpha_1}} \lambda_{\alpha_1} R\indices{^{(1)\, s_2 \cdots s_N \, \alpha_1}} }&=\norm{\sum_{\alpha_1 = \chi +1}^{d^N} M\indices{^{(1)\, s_1}_{\alpha_1}} \lambda_{\alpha_1} R\indices{^{(1)\, s_2 \cdots s_N \, \alpha_1}} }\\
&< \norm{\sum_{\alpha_1 = \chi +1}^{d^N} M\indices{^{(1)\, s_1}_{\alpha_1}} \lambda_{\chi + 1} R\indices{^{(1)\, s_2 \cdots s_N \, \alpha_1}} }\\
&\le \lambda_{\chi+1}\times (\mathrm{const.}).
\end{align}
where the norm is defined by the Hilbert-Schmidt or Frobenius norm (with respect to the physical indices), i.e. $\| A \|\equiv \Tr A^\dagger A$. This error bound is known as \textbf{Eckart-Young theorem}\index{Eckart-Young theorem}\footnote{This result is further generalized to all unitarily invariant norms by Mirsky~\cite{10.1093/qmath/11.1.50}.}. The error is $\mathcal{O}(\lambda_{\chi+1})$. When $\lambda_{\chi+1}=0$, the result is error-free. In general, we set the bond dimension as the ansatz and do not know the exact values of singular values. Thus we extrapolate in $1/\chi$ and the convergence of the result should be checked.

%cost
How much cost is reduced by the MPS? By construction, we obtain the maximum Schmidt rank=the maximum number of singular values (parameters)=$d^{N/2}$, for an MPS when we decompose the middle tensor. Hence, the bond dimension $\chi$ is usually taken to be subexponential in $N$: $\mathit{O}(\mathrm{poly}(N))$ or $\mathit{O}(1)$. We evaluated the deviation by this truncation above and confirmed the error is small enough, however, does this truncation really suppress the exponential computational cost? The complete N-body wave function is specified by $d^N$ tensor components. Therefore a naive computation takes an exponential time with respect to the system size. In the MPS, each matrix has two virtual legs and one physical leg; each matrix has $\mathit{O}(\chi^2 d)$ parameters. Since the $N$-body MPS has $N$ matrices, the total number of parameters~cost is $\mathit{O}(N\chi^2 d)$. As $\chi$ is subexponential in N, the computational cost is indeed suppressed to be polynomial in $N$.

The MPS is based on the SVD between nearest two sites. Therefore the MPS is an efficient representation for a one-dimensional quantum system in which the nearest neighbor interaction (entanglement) is dominant.\footnote{We say a class of tensor network efficiently describes the ground state of a quantum system if the tensor network representation with a finite, $\mathit{O}(1)$ or $\mathit{O}(\mathrm{poly}(N))$, bond dimension well approximates the ground state with a small error than other classes of tensor network.} Indeed, EE of the ground state of a gapped system is roughly a constant~\cite{Evenbly_2011}, consistent with the area law.

\subsection{Projected Entangled-Pair States}
The \textbf{projected entangled-pair state (PEPS)}\index{projected entangled-pair state}\index{PEPS} is another type of tensor networks~\cite{Verstraete_2008,Kraus_2010,Corboz_2010}. For a spatially one-dimensional system, it equals MPS. In the following, we show how to construct a one-dimensional PEPS. One can readily see the PEPS is equivalent to the MPS from the final expression. Also, one will find no difficulty to generalize the procedure to the higher dimensional case; the PEPS in more than one dimension can be made in the same manner.

\begin{enumerate}
\item First, prepare a virtual, maximally entangled pair $| \phi \rangle _ { i , i + 1 } \in \mathcal{H}\indices{_\chi^{\otimes 2}}$ (e.g. singlet) over each nearest two sites.
\begin{align}
| \phi \rangle _ { i , i + 1 } &= \sum _ { \alpha = 1 } ^ { \chi } \frac { 1 } { \sqrt { \chi } } | \alpha \rangle _ { i } \otimes | \alpha \rangle _ { i + 1 },\phantom{aa}
\mathrm{up\ to\ a\ unitary\ transformation}\\
&=\includegraphics[clip,width=43pt,valign=c]{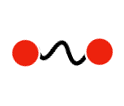}.
\end{align}

\item Then, apply a projector from two virtual degrees of freedom at each physical site onto a physical (real) degrees of freedom $P_i : \mathcal{H}\indices{_\chi^{\otimes 2}} \rightarrow \mathcal{H}_{d_i}$.
\begin{equation}
\ket{\psi}= \chi ^ { N/2 } \left( P _ { 1 } \otimes P _ { 2 } \otimes \cdots \otimes P _ { N } \right) \left(|\phi\rangle_{1,2}\otimes |\phi\rangle_{2,3}\otimes \cdots\otimes| \phi \rangle_{N,1}\right)
\end{equation}
(when the periodic boundary condition is imposed).
\begin{figure}[H]
\centering
\includegraphics[clip,width=0.5\linewidth]{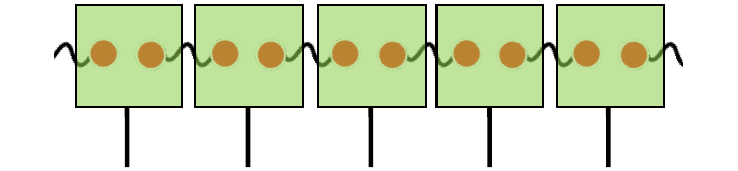}
\caption{The one-dimensional PEPS with physical degrees of freedom shown by orange dots}
\end{figure}

\item If one explicitly writes down the component of the projector, one can see the whole PEPS equals to the MPS.
\begin{equation}
\mathit{P} _ { j } = \sum _ { s_j = 1 } ^ { d } \sum _ { \alpha , \beta = 1 } ^ { \chi } \left({ A ^ { (j)\, s_j } }\right) _ { \alpha\beta } | s_j \rangle _ { j } \otimes \langle \alpha | _ { j } \otimes \langle \beta | _ { j },
\end{equation}
where $\mathrm{span}( \{\ket{s_j}_j\}) = \mathcal{H}_{d_j}$, a physical Hilbert space at site $j$, and $\mathrm{span}(\{\ket{\alpha}_j \}) =\mathcal{H\indices{_{\chi}}}$, a virtual Hilbert space at site $j$.
\end{enumerate}
Since the PEPS prepares the maximally entangled state between \emph{nearest-neighbor} sites and then `reduces' entanglement by a projector at each site, the PEPS formalism is suitable for a ground state of a Hamiltonian dominated by \emph{short-range} interactions.

Although we considered the finite PEPS so far, we can consider an infinitely long PEPS known as the \textbf{infinite PEPS (iPEPS)}\index{infinite projected entangled-pair state}\index{iPEPS}. If one chooses a fundamental unit cell, $ABC$ in an MPS for example, then one can repeat the sequence infinitely many times $\cdots ABCABCABC\cdots$ to get the iPEPS which describes the ground state of the thermodynamic limit of the system.

The PEPS construction is closely related to the idea mentioned in Section \ref{sec:intro-TN}: the condensation of entangled pairs at the RT surface in holography, except that the PEPS do not contain the extra holographic dimension. 
However, the multi-scale entanglement renormalization ansatz (MERA)~\cite{Vidal:2008zz} has an extra holographic dimension representing the coarse graining of information. Such tensor networks on a hyperbolic lattice 
with an extra dimension have been proposed to be toy models of holography~\cite{Swingle:2009bg,Swingle:2012wq,Beny:2011vh,Pastawski:2015qua,Yang:2015uoa,Czech:2015kbp,SinaiKunkolienkar:2016lgg,Hayden:2016cfa,Evenbly:2017hyg,Jahn:2017tls,Bhattacharyya:2017aly,Bao:2018pvs,Qi:2018shh,Milsted:2018san,Steinberg:2020bef,Jahn:2020ukq,Jahn:2021kti} and have facilitated an information-theoretic understanding of holography. Thus, we will examine the formulation of MERA in the next subsection.

\subsection{Multi-scale Entanglement Renormalization Ansatz (MERA)}\label{sec:mera}
The \textbf{multi-scale entanglement renormalization ansatz (MERA)}\index{multi-scale entanglement renormalization ansatz}\index{MERA} is a class of tensor network proposed by Vidal~\cite{Vidal:2008zz,vidal2009entanglement,Corboz_2009,Corboz_2010}. The MERA has an efficient expression for other systems which the MPS nor the PEPS is not suitable for. The MERA is depicted in Fig.\ref{fig:mera}. If one sees this tensor network \emph{from bottom to top}, the original sites in bottom (UV) get renormalized to top (IR) (in the sense of entanglement). In this sense, the MERA is a \textbf{real-space renormalization group flow}\index{real-space renormalization group}. This is why this network is called MERA. The MERA consists of three kinds of tensors: The blue square denotes the \textbf{disentangler}\index{disentangler}, a $\chi^2 \times\chi^2$ two-body unitary matrix
\begin{equation}
\includegraphics[valign=c,height=1.5cm]{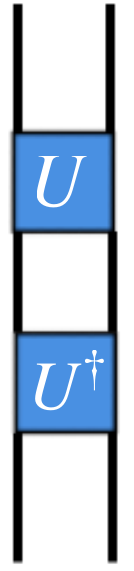}=\includegraphics[valign=c,height=1.5cm]{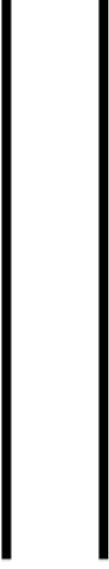}\ ,
\end{equation}
which cuts short-range entanglement between two sites when the MERA is seen from bottom to top\footnote{When the MERA is seen \emph{from top to bottom}, the disentangler is sometimes called the entangler. In this point of view, the MERA can be regarded as a \textbf{quantum circuit}\index{quantum circuit} when the isometries are unitarized with additional ancillae (reference states) $\ket{0}$. The MERA as a quantum circuit entangles a separable input $\ket{0}\otimes\cdots\otimes\ket{0}$ and outputs a nontrivially entangled state. Because of the tensor properties in the MERA, an operator expectation value or a correlation function is efficiently computable, in contrast to the usual quantum circuits, in which the contractions of unitaries within the causal cone of the operator(s) gives an exponential cost~\cite{PhysRevLett.112.240502}.}. The green triangle denotes an \textbf{isometry} (isometric tensor)\index{isometry}
\begin{equation}
\includegraphics[valign=c,height=1.5cm]{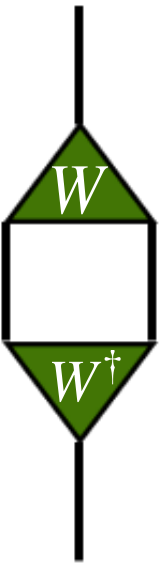}=\includegraphics[valign=c,height=1.5cm]{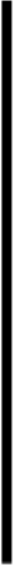}\ ,
\end{equation}
in which three sites are coarse-grained into one effective site. The red circle denotes the \textbf{top tensor}\index{top tensor} on top. From the other local normalization conditions of isometries and disentanglers, the top tensor is also normalized. It is important that all tensors are locally normalized so that the computational cost for the contractions in the calculation of the norm and correlation functions can be efficiently performed. This is contrary to the PEPS, in which the exact contraction takes an exponential cost.

Apparently, the MERA takes the \emph{long-range} entanglement into account by its hierarchical structure. The network is scale invariant in infinite MERA (scale invariant MERA) when all isometries/disentanglers are taken common. In the finite range MERA, the scale invariance breaks down at the correlation length, a special length scale associated with the number of layers. Also, the translational invariance is present by construction.
\begin{figure}[htbp]
\centering
    \includegraphics[width=0.7\linewidth]{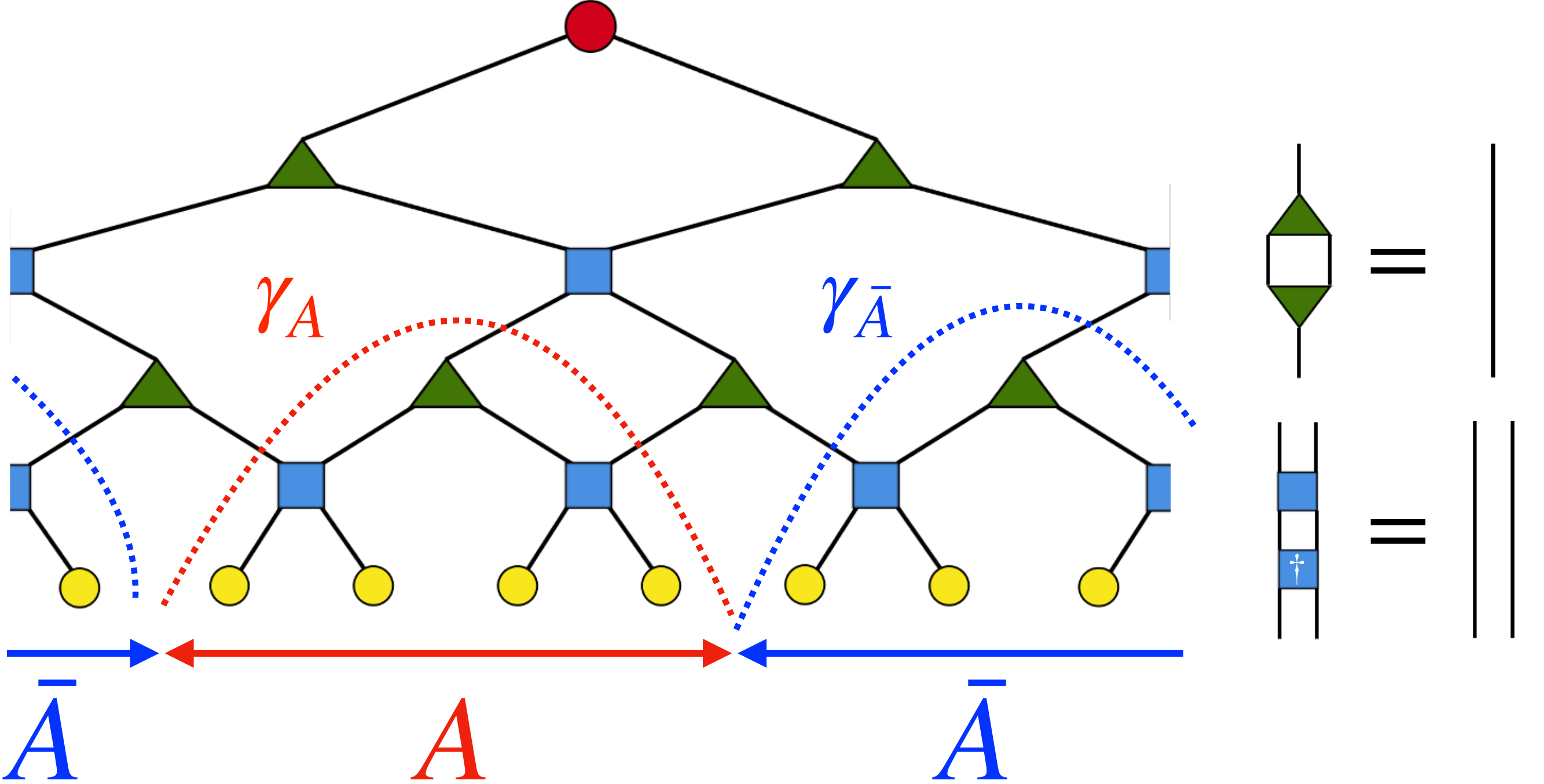}
    \caption{{A %periodic 
    MERA tensor network is composed of binary unitaries (blue squares), isometries (green triangles), and a top tensor (red circle). Yellow circles represent %unspecified 
    physical indices. $A$ and $\bar{A}$ denote a subregion and its complement{, respectively}. For this symmetric bipartition, both $\gamma_A$ and $\gamma_{\bar{A}}$ become minimal bond cut surfaces $\gamma_\ast$.}}
    \label{fig:mera}
\end{figure}

While the MPS and PEPS are efficient tensor network-represented states for a ground state with short-range entanglement, the MERA provides an efficient expression for a ground state with long-range entanglement, e.g. a ground state of a critical or gapless system.

Not only the MERA describes a gapless ground state, the emergent direction in the MERA is appealing. For example, although we deal with a spatially one-dimensional state, the MERA has a logarithmic scale extending vertically (Fig.\ref{fig:mera}). This reminds us of the holographic emergent spacetime and its hyperbolicity in the AdS/CFT.

\section{Entanglement distillation for {MERA}}\label{sec:HED}

In the following sections, we define geometric operations in a tensor network and relate it with entanglement distillation\index{entanglement distillation}. We focus on MERA in this and next sections.
We consider a binary MERA state $\ket{\Psi}$ represented by Fig.\ref{fig:mera}.
%As shown in Fig.\ref{fig:causal-cone}, isometric regions 
{The isometric regions shaded blue in Fig.\ref{fig:causal-cone}}
%{corresponding to a subregion $A$ and its complement $\bar{A}$}
%in the MERA 
are called future or exclusive causal cones~\cite{Vidal:2008zz,Evenbly:2007hxg,evenbly2013quantum,Czech:2015kbp}.
We denote them by {$\mathcal{C}(A)$ for a subregion $A$ and $\mathcal{C}(\bar{A})$ for the complement $\bar{A}$.}
%, respectively. 
{Their} edges %of them 
are denoted by $\gamma_A\equiv\partial\mathcal{C}(A)$ and $\gamma_{\bar{A}}\equiv\partial\mathcal{C}(\bar{A})$. We call the smaller one, a minimal bond cut surface $\gamma_\ast=\min(\gamma_A,\gamma_{\bar{A}})$. 
This surface $\gamma_\ast$ in {MERA} corresponds to the RT surface, a minimal surface in a holographic spacetime. From the PEPS perspective, there are EPR pairs across the surface. 
{Since isometries do not affect entanglement, the EPR pairs carry all of the entanglement of the state if all the projection tensors are isometries.} 
This is true for {a} perfect tensor {network}~\cite{Pastawski:2015qua}, which consists of isometries. %Isometries are maps which do not affect entanglement of the state. 
In contrast, 
{{MERA} has nontrivial projection degrees of freedom carried by each tensor.} 
%{Then,} 
{As a result,} this naive %picture
{view} of EPR pairs across the surface %is 
{becomes} subtle.
%for MERA as it has nontrivial projection degrees of freedom carried by each tensor like PEPS.

When the state is described by a MERA, given the fixed bond dimension $\chi$, entanglement entropy\index{entanglement entropy} satisfies the following inequality:
\begin{equation}
    S(\rho_A)\le (\text{\# of bond cuts by }\gamma_\ast)\times \log\chi.
    \label{eq:TN-EE}
\end{equation}
When \eqref{eq:TN-EE} is saturated, it is interpreted as a discrete version of the RT formula~\eqref{eq:HEE-RT}.

\begin{figure}[h]
    \centering
    \includegraphics[width=0.7\linewidth]{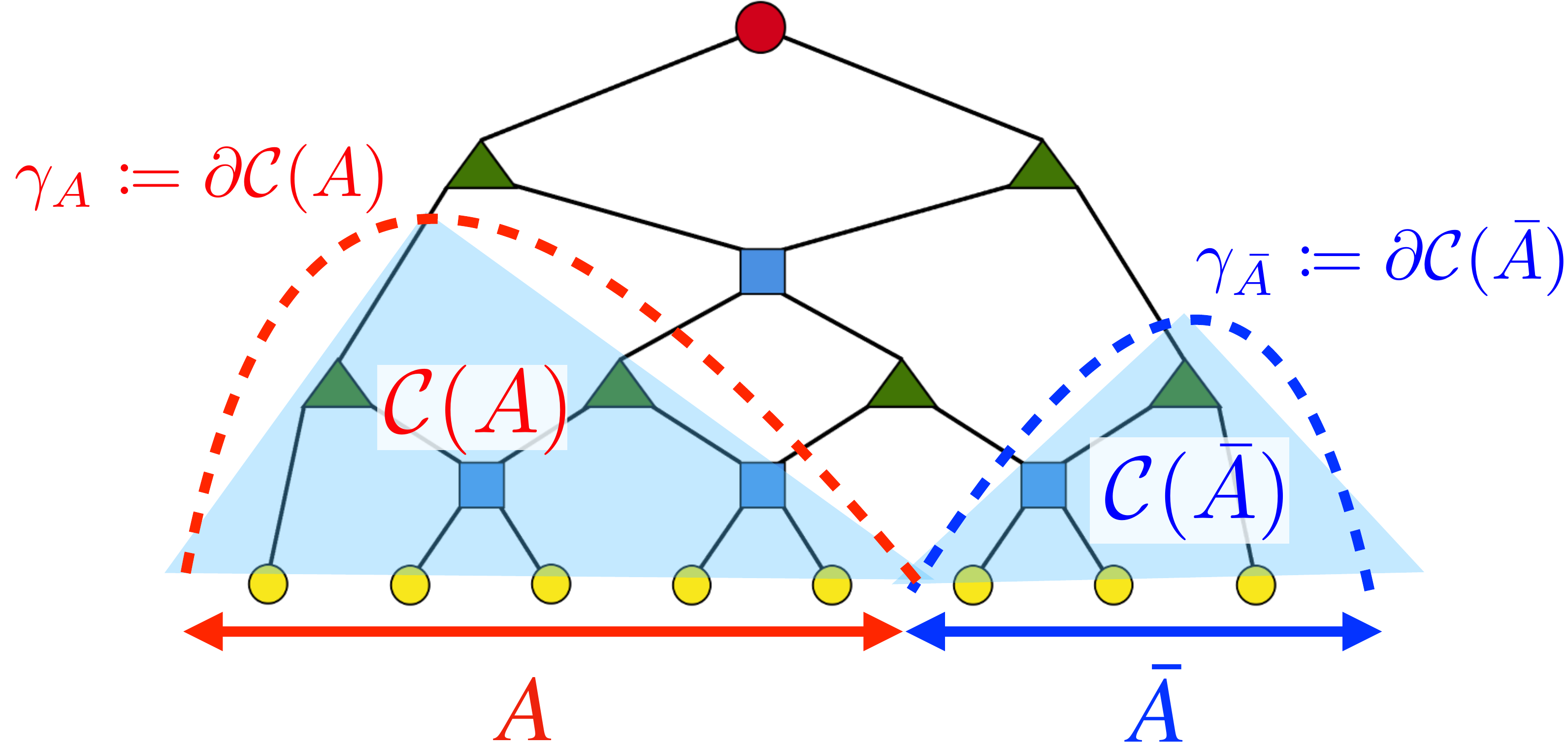}
    \caption{{{In a MERA tensor network}, the future or exclusive causal cone $\mathcal{C}(A)$ [$\mathcal{C}(\bar{A})$] of a subregion $A$ [$\bar{A}$] covers tensors that can %influence 
    {affect} only $A$ [$\bar{A}$] seen from the top to %the 
    {bottom}. The edge of $\mathcal{C}(A)$ [$\mathcal{C}(\bar{A})$] is called a causal cut~\cite{Czech:2015kbp} or a minimal curve~\cite{Swingle:2009bg,Swingle:2012wq} and is denoted by $\partial\mathcal{C}(A)$ [$\partial\mathcal{C}(\bar{A})$]. In the {aforementioned} example, %above, 
    the minimal bond cut surface $\gamma_\ast$ is given by $\gamma_{\bar{A}}$.}}
    %isometryとかの説明してない
    \label{fig:causal-cone}
\end{figure}

\begin{figure*}[t]
    \centering
    \includegraphics[width=\linewidth]{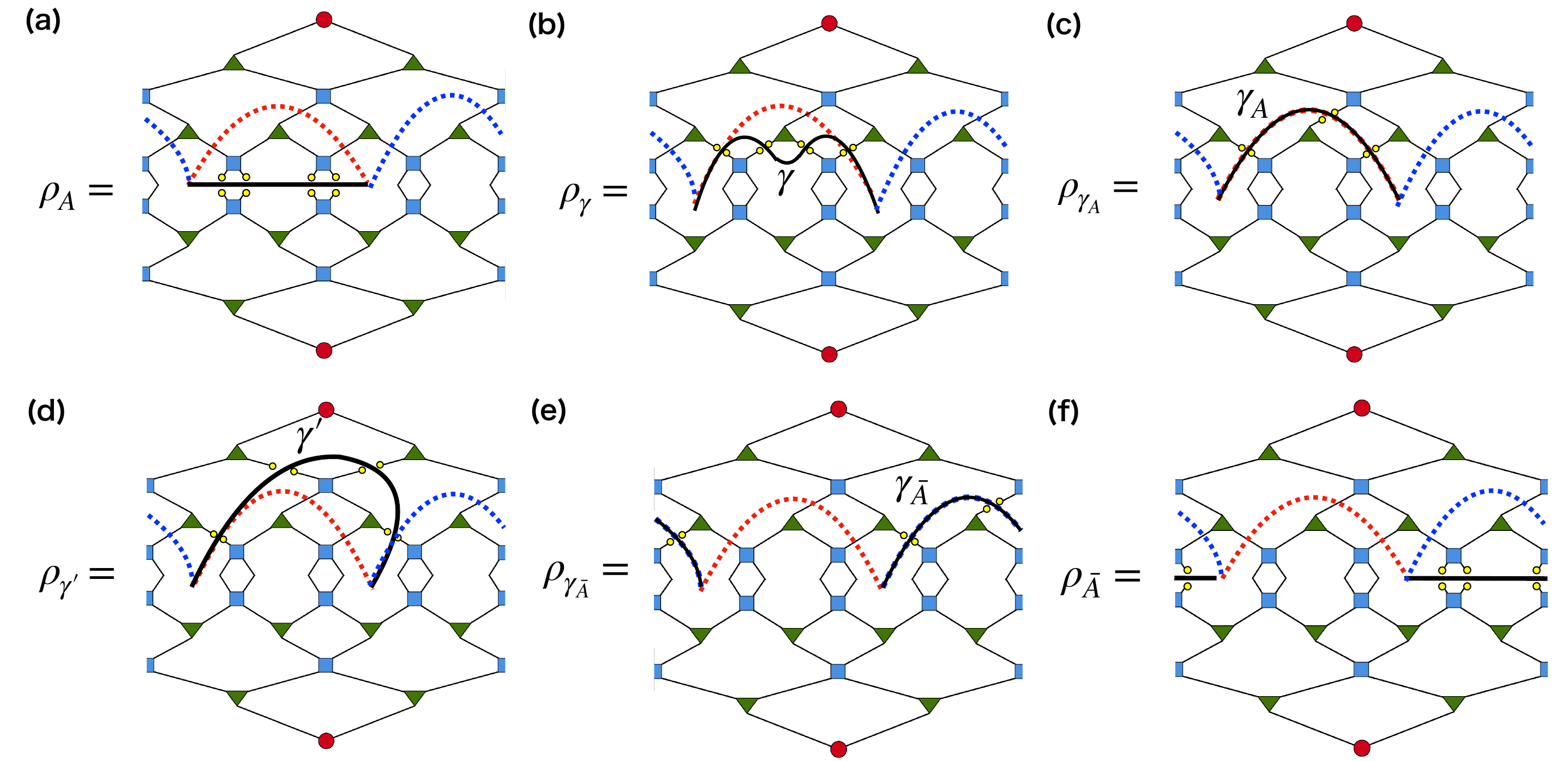}
    \caption{
    Reduced transition matrices corresponding to various foliations. %Yellow circles denote unspecified indices.
    %Given a subregion $A$, there exist two minimal bond cut surfaces $\gamma_A$ (red dashed curve) and $\gamma_{\bar{A}}$ (blue dashed curve) 
    {When the subsystem $A$ is half of the whole system,} there are two minimal bond cut surfaces $\gamma_\ast=\gamma_A, \gamma_{\bar{A}}$.
    (a) Cutting the physical bonds of $A$ in $\braket{\Psi}$ gives $\rho_A$.
    (b) The foliation $\gamma$ is pushed toward $\gamma_A$.
    (c) The foliation equals %to 
    $\gamma_A$.
    (d) The foliation is pushed toward the other minimal bond cut surface $\gamma_{\bar{A}}$.
    (e) The foliation reaches $\gamma_{\bar{A}}$.
    (f) Finally, the foliation cuts the physical bonds in $\bar{A}$ and it gives the reduced density matrix $\rho_{\bar{A}}$.
    %Green curves denote a family of foliations for the tensor network. When the foliation is $\gamma_1$ or $\gamma_6$, the distilled density matrix is given by $\rho_A$ or $\rho_{\bar{A}}$ respectively. This tensor network has two minimal bond cut (RT) surfaces $\gamma_2$ and $\gamma_3$. When we optimize a state using this, we expect the density matrices on the foliations becomes homogeneously close to the maximally mixed state as the foliation approaches to the RT surface.
    }
    \label{fig:HED}
\end{figure*}

%In the following, we present a protocol extracting a state close to the EPR pair from the original MERA state $\ket{\Psi}$ without losing entanglement. We denote a set of linear operators on a Hilbert space $\mathcal{H}$ by $\mathcal{L}(\mathcal{H})$. In step 1, we define an initial state. In step 2, we discuss geometric operations in tensor network. In step 3 and 4, we discuss how it would be related to entanglement distillation by considering a conservation of entanglement and the closeness to the EPR pair.

In the following, we first present a %\blue{natural} 
way to define a state on % a bond cut surface lying on 
{a surface across} internal bonds in {the} MERA. Then, such a state is shown to preserve the amount of entanglement %by appropriately choosing 
{with an appropriate choice of} a family of bond cut surfaces. 
%When \red{the degrees of freedom of the state} decreases as we change the location of the bond cut surface, we call it entanglement distillation. 
{As the minimal bond cut surface has the least number of bonds, we expect {the entanglement} per bond is concentrated to be {maximal}. Thus, we identify pushing a bond cut surface toward the minimal surface as entanglement distillation.} 
We quantify the process by examining the trace distance between each state and an EPR pair. %A set of linear operators on a Hilbert space $\mathcal{H}$ is denoted by $\mathcal{L}(\mathcal{H})$.

%\red{"by appropriately choosing a family of bond cut..." The subject is not human so it's weird.}

%\begin{enumerate}
    %\item 
    Given a MERA state $\ket{\Psi}$ (Fig.\ref{fig:mera}), its reduced density matrix $\rho_A$ for a subregion $A$ is obtained by cutting the physical bonds on $A$ in the norm $\braket{\Psi}$ as shown in Fig.\ref{fig:HED}~(a). 
    
    {In the following,} we consider a deformed surface $\gamma$ from $A$ such that the endpoints are common, $\partial\gamma=\partial A$. {This is a discrete version of the homology condition.} We call such a surface a \textit{foliation}\index{foliation}. {As an initial condition}, we have $\gamma=A$. A minimal bond cut surface $\gamma_\ast$ {equals} %to 
    a foliation with a minimum number of bond cuts, i.e. %. Thus,
    $\dim\mathcal{H}_\gamma\ge\dim\mathcal{H}_{\gamma_\ast}$, {where $\mathcal{H}_\gamma$ is the Hilbert space of bonds across $\gamma$}.
    
    %\item 
    Deforming $\gamma$ from $A$, we obtain a norm $\braket{\Psi}$ with bonds cut on $\gamma$. For example, when we choose a foliation $\gamma$ {as} shown in Fig.\ref{fig:HED} (b), the tensor network defines a reduced transition matrix\index{reduced transition matrix}~\cite{Nakata:2020luh}
    \begin{equation}
        \rho_\gamma = \tr_{\bar{A}}\left(\ket{\Psi(\gamma)} \bra{\Phi(\gamma)} \right)\in \mathcal{L}(\mathcal{H}_\gamma),
        \label{eq:red-trans}
    \end{equation}
    {where $\mathcal{L}(\mathcal{H})$ denotes a set of linear operators on a Hilbert space $\mathcal{H}$.}
    Fig.\ref{fig:ketbra} %gives
    {shows} the states $\ket{\Psi(\gamma)}\in \mathcal{H}_\gamma \otimes \mathcal{H}_{\bar{A}}$ and $\bra{\Phi(\gamma)}\in \mathcal{H}^\ast_\gamma \otimes \mathcal{H}^\ast_{\bar{A}}$. It immediately follows that $\braket{\Phi(\gamma)}{\Psi(\gamma)}=\tr \rho_\gamma=1$ for an arbitrary foliation $\gamma$. 
    $\bra{\Phi(\gamma)}$ and $\left|\Psi(\gamma)\right>$ are created by adding and removing tensors $M_\gamma$ bounded by $A$ and $\gamma$ in the tensor network representation:
    \begin{equation}
        \begin{aligned}
        \bra{\Phi(\gamma)}&=\bra{\Psi} M_\gamma,\\
        M_\gamma \ket{\Psi(\gamma)} &= \ket{\Psi}.
        \end{aligned}
        \label{eq:braket-rel}
    \end{equation}
    For example, if we consider a configuration {shown in} Fig.\ref{fig:HED}~(b), $M_\gamma=U_1\otimes U_2$ where $U_{1,2}$ are {shown} %given 
    in Fig.\ref{fig:ketbra}~(b). 
    
    Using the relation \eqref{eq:braket-rel}, we can show {that} any reduced transition matrices $\rho_\gamma$ have common positive eigenvalues with the original reduced density matrix $\rho_A$. This can be shown as follows. We denote $\tr_{\bar{A}} \left(\ket{\Psi(\gamma)} \bra{\Psi} \right)$ by $S_\gamma$ and the positive eigenvalues and eigenvectors of $\rho_\gamma$ are denoted by $\{\lambda_n\}_n$ and $\{\ket{n}_\gamma\}_n$. Then,
    \begin{equation}
        \rho_\gamma \ket{n}_\gamma = S_\gamma M_\gamma \ket{n}_\gamma = \lambda_n \ket{n}_\gamma.
        \label{eq:gamma}
    \end{equation}
    By multiplying $M_\gamma$ from {the} left, we obtain
    \begin{equation}
        M_\gamma S_\gamma M_\gamma \ket{n}_\gamma = \lambda_n M_\gamma \ket{n}_\gamma,
    \end{equation}
    whereas $M_\gamma S_\gamma = \tr_{\bar{A}} \left(M_\gamma\ket{\Psi(\gamma)} \bra{\Psi} \right)=\rho_A$ from \eqref{eq:braket-rel}. Since \eqref{eq:gamma} is by definition nonzero, $M_\gamma \ket{n}_\gamma \neq 0$. Therefore the positive eigenvalues of $\rho_A$ coincide with those of $\rho_\gamma$ for {an} arbitrary $\gamma$.
    
    %A reduced transition matrix for any foliation $\gamma$ is related to the original reduced density matrix $\rho_A$ via a similarity transformation
    %\begin{equation}
        %\rho_{\gamma}=M_\gamma^{{-1}}\rho_{A}M_\gamma,
    %    \blue{\rho_{\gamma}=M_\gamma^{+}\rho_{A}M_\gamma.}
    %    \label{eq:sim-trf}
    %\end{equation}
    %\blue{Since $\rank \rho_\gamma\neq \rank\rho_A$ in general, $M_\gamma$ is a rectangular matrix and $M_\gamma^{+}$ is a Moore-Penrose pseudo inverse matrix of $M_\gamma$, which is defined by
    %\begin{align*}
    %    M_\gamma^{+}&=\bar{V}\hat{m}^{+}\bar{U}^\dagger\\
    %    \hat{m}^{+}&=\mathrm{diag}(\hat{m}^{-1},0,0,\cdots),
    %\end{align*}
    %where the unitaries $\bar{U},\bar{V}$ and the singular value matrix $\hat{m}$ are defined from the singular value decomposition
    %\begin{equation}
    %    M_\gamma\equiv \bar{U}\hat{m} \bar{V}^\dagger.
    %\end{equation}
    %}
    %where 
    %\blue{In the tensor network representation,} the operators $M_\gamma$ and \blue{$M_\gamma^{+}$} represent addition and removal of tensors to create \blue{$\bra{\Phi(\gamma)}$ from $\bra{\Psi}$ and $\left|\Psi(\gamma)\right>$ from $\left|\Psi\right>$}, respectively. %\footnote{\blue{We emphasize that the inverse alone is not always well-defined. Togather with the tensor network state, its operation is determined as an addition or removal of tensors.}} 
    %For example, if we consider a configuration Fig.\ref{fig:HED}~(b), $M_\gamma=U_1\otimes U_2$ where $U_{1,2}$ are given in Fig.\ref{fig:ketbra}~(b). 

    \begin{figure}[h]
    \centering
    \includegraphics[width=0.5\linewidth]{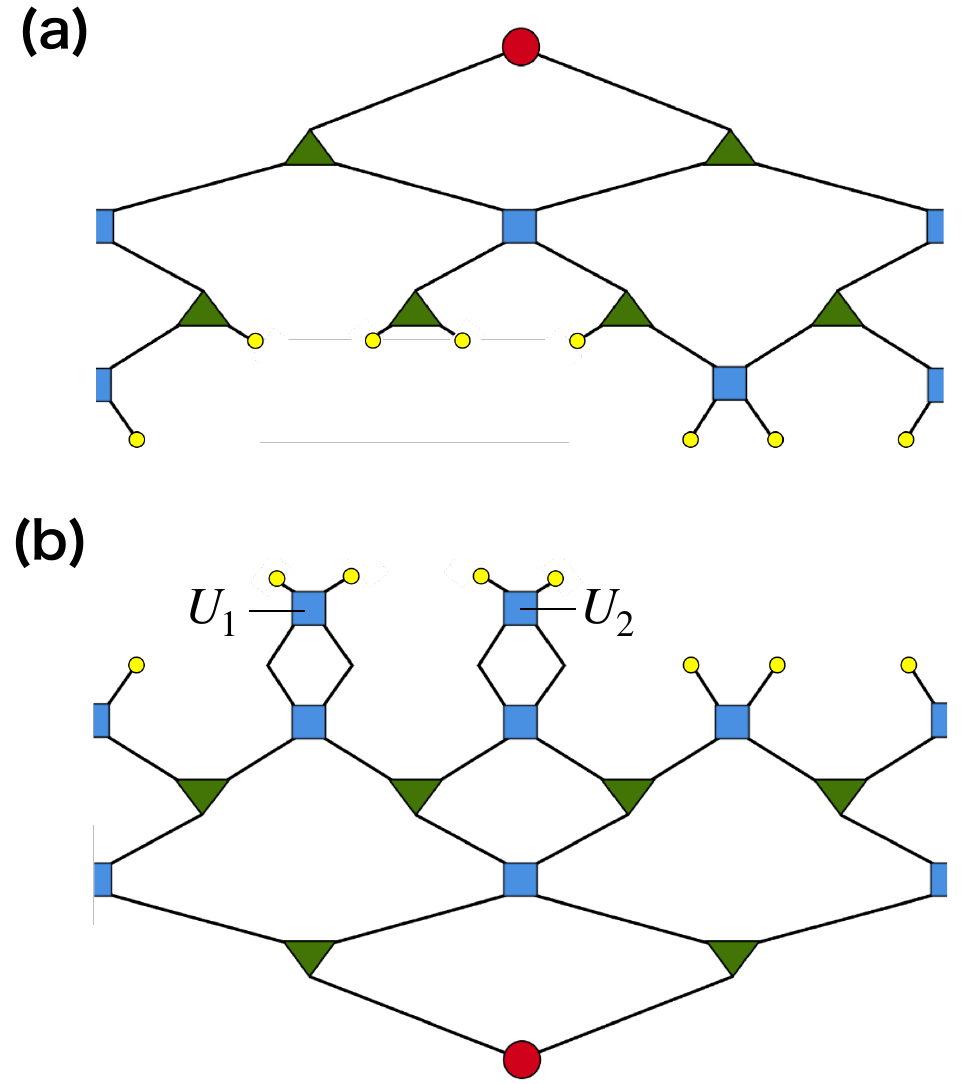}
    \caption{{When the foliation $\gamma$ is chosen as {shown in} Fig.\ref{fig:HED}~(b), $\ket{\Psi(\gamma)}$ is given by (a) and $\bra{\Phi(\gamma)}$ is given by (b). They are related to the original state $\ket{\Psi}$ by either removing or adding tensors $U_1\otimes U_2\in \mathcal{L}(\mathcal{H}_A)$.}}
    \label{fig:ketbra}
    \end{figure}
    
    %\item
    %The equivalence \eqref{eq:sim-trf} implies that any reduced transition matrices $\rho_\gamma$ have common semi-positive eigenvalues. Then, it follows that 
    {Since $\rho_A$ and $\rho_\gamma$ share common positive eigenvalues, it immediately follows that} 
    the von-Neumann entropy of a reduced transition matrix $S(\rho_\gamma)$ %as 
    {known} as \textbf{pseudo entropy}\index{pseudo entropy}~\cite{Nakata:2020luh}, {equals} %to 
    entanglement entropy:
    \begin{equation}
    S(\rho_\gamma)=S(\rho_A),\quad \forall \gamma \quad \mathrm{s.t.}\quad \partial\gamma=\partial A.
    \label{eq:pseudo}
    \end{equation}
    This identity is interpreted in two ways. %First, this 
    {The first} is a %sort 
    {type} of {the} bulk/boundary correspondence\index{bulk/boundary correspondence} like the RT formula. While the right-hand side represents the {entanglement} entropy of the boundary quantum state $\ket{\Psi}$, the left-hand side is given as a function of an operator in the bulk. %Second, this 
    {The second} is interpreted as a \textbf{conservation of entanglement}\index{conservation of entanglement} during the deformation of $\gamma$. From the PEPS perspective, $S(\rho_\gamma)$ effectively counts the amount of entanglement carried by bonds across $\gamma$. Then, the equality \eqref{eq:pseudo} indicates the amount of entanglement across each foliation is retained during the deformation of $\gamma$. 
    
    %\item
    Throughout the procedure, the number of %the 
    {bond} cuts at $\gamma$ changes and it minimizes at a minimal bond cut surface $\gamma_\ast$. Thus, the diluted entanglement over $\ket{\Psi}$ is concentrated into a smaller number of strongly entangled bonds across $\gamma_\ast$. Next, we evaluate the degree of this concentration in terms of the trace distance.
    
%\end{enumerate}

%Previous studies~\cite{Bao:2018pvs,Lin:2020yzf,Yu:2020zwk,Lin:2020ufd} discuss a distilled state {by cutting out a part of a tensor network state} in the light of surface/state correspondence~\cite{Miyaji:2015yva}. In contrast, we consider changing the location of the bond cut surface in the norm instead of a state. Nevertheless, this protocol becomes equivalent to the previous ones by contracting tensors in some cases such as tree~\cite{Bao:2018pvs}, perfect~\cite{Pastawski:2015qua}, or dual-unitary tensor networks~\cite{Bertini:2018fbz,Bertini:2019lgy}.

{Before moving on, let us comment on} the %relation 
%{similarity} and the difference 
{similarities and differences} between our procedure and previous proposals. In our procedure, we define %a state 
{a reduced transition matrix} on each foliation and identify pushing the foliation as entanglement distillation. %Comparing to
{Compared with} the previous studies of entanglement distillation in holography~\cite{Pastawski:2015qua,Bao:2018pvs}, %this 不明確
{pushing the foliation} can be regarded as a %kind 
{type} of \textbf{operator pushing}\index{operator pushing}. In~\cite{Pastawski:2015qua}, an operator pushing of %{some} 
{an} operator $O$ through an isometry\index{isometry} $V_{\mathrm{iso}}$ is defined by
\begin{equation}
    OV_{\mathrm{iso}}=V_{\mathrm{iso}}\Tilde{O},
\end{equation}
where $\Tilde{O}=V^\dagger_{\mathrm{iso}} O V_{\mathrm{iso}}$. %\blue{Meanwhile,} 
{While $O$ is usually state-independent,} 
in our procedure, the pushed operator is %now a state itself
{the reduced transition matrix defined from the state. The mapping between two reduced transition matrices $\rho_\gamma$ and $\rho_{\gamma^\prime}$ on the foliations $\gamma$ and $\gamma^\prime$ respectively is an operator pushing, i.e.}%:
\begin{equation}
    \rho_\gamma M = M \rho_{\gamma^\prime},
\end{equation}
where %$M=M_\gamma M_{\gamma^\prime}^{-1}$ from \eqref{eq:sim-trf}. 
{$M$ represents tensors bounded by $\gamma$ and $\gamma^\prime$.} 
{Although our procedure can be interpreted as a %kind 
{type} of operator pushing, one} %the
important difference is that $M$ is not necessarily isometric while $V_{\mathrm{iso}}$ was assumed to be isometric or unitary. 
{
This difference arises because our procedure deals with a reduced transition matrix rather than a state vector. For a state vector, the only operations %which 
{that} preserve entanglement entropy are isometry and unitary {ones}. This requirement severely restricts possible tensor network states. They must be composed of perfect tensors\index{perfect tensors}~\cite{Pastawski:2015qua} or dual unitaries\index{dual unitary}~\cite{Bertini:2018fbz,Bertini:2019lgy} or isometric tree tensor networks\index{tree tensor network}~\cite{Bao:2018pvs,Lin:2020yzf,Yu:2020zwk,Lin:2020ufd}. Such states can be distilled by removing the composing tensors %known as 
{by applying} a \textbf{greedy algorithm}\index{greedy algorithm}. %greedy algorithm 
%On the other hand,
In our procedure, we deal with a reduced transition matrix. The operations %which 
{that} preserve entanglement \eqref{eq:pseudo} are not limited to isometries. In this way, we can consider entanglement distillation\index{entanglement distillation} using reduced transition matrices on various bond cut surfaces in a more general tensor network like {MERA}, which has nonisometric $M$.} %Remarkably,} 
%Since one can remove the isometric degrees of freedom from the state vector as much as possible under LOCC, a complete distillation of EPR pairs across the minimal bond cut surface is easy to see in perfect tensor networks~\cite{Pastawski:2015qua} or at least the constituting tensors are dual unitaries~\cite{Bertini:2018fbz,Bertini:2019lgy}, i.e. unitaries which remain unitary even after a partial transposition. 
%This is called a greedy algorithm in~\cite{Pastawski:2015qua}. However, a general tensor networks on a hyperbolic lattice like MERA has nonisometric tensors as well. Then, the greedy algorithm only works partially. 
%In our approach, instead of deforming a state vector~\cite{Miyaji:2015yva,Bao:2018pvs,Lin:2020yzf,Yu:2020zwk,Lin:2020ufd}, we consider deforming a reduced density matrix to a reduced transition matrix by pushing a foliation inside the tensor network. 
{This} enables us to consider a state on an arbitrary bond cut surface even beyond the region a greedy algorithm can reach (called a \textbf{bipartite residual region}\index{bipartite residual region}~\cite{Pastawski:2015qua} or \textbf{causal shadow}\index{causal shadow}~\cite{Lewkowycz:2019xse} in the literature) while retaining the amount of entanglement $S(\rho_\gamma)$.
%When inside causal cone, it is equivalent to greedy algorithm.
% Or if the constituting tensors are made from dual-unitary tensors~\cite{Bertini:2018fbz,Bertini:2019lgy},

To evaluate how much entanglement is distilled from the original state $\ket{\Psi}$, we should quantify the closeness of a properly defined state across $\gamma$ to a maximally entangled state (the EPR pair).
However, since $\rho_\gamma$ is an operator, we cannot compare it with the EPR state %straightforwardly
{directly}. Thus, we define a distilled state on $\gamma$ as a state vector in $\mathcal{H}_\gamma\otimes\mathcal{H}_\gamma $ using the same idea with the purification (known as the \textbf{Choi state}\index{Choi state}),
\begin{equation}
    \ket{\rho_\gamma^{1/2}}\equiv {\mathcal{N}_\gamma} \sqrt{\dim\mathcal{H}_\gamma}
    (\rho_\gamma^{1/2}\otimes \mathbf{1})
    \ket{\mathrm{EPR}_\gamma},
    \label{eq:purif}
\end{equation}
where $\mathcal{N}_\gamma=\Big[{\tr(\rho_\gamma^{\dagger\, 1/2}\rho_\gamma^{1/2})}\Big]^{-1/2}$ and 
$\ket{\mathrm{EPR}_\gamma}=({\dim\mathcal{H}_\gamma})^{-1/2} \sum_{i=1}^{\dim\mathcal{H}_\gamma} \ket{i}\otimes\ket{i}$.
Then, we can define the closeness between the distilled {and EPR states} %state and the EPR state 
as the \textbf{trace distance}\index{trace distance} between them:
\begin{equation}
    D_\gamma\equiv \sqrt{1-\big|{\langle \mathrm{EPR}_\gamma | \rho_\gamma^{1/2}\rangle}\big|^2}.
    \label{eq:trace_distance}
\end{equation}
%Based on 
{On the basis of} this trace distance, we propose that the minimal bond cut surface $\gamma_\ast$ provides entanglement distillation such that $\ket{\rho_{\gamma}^{1/2}}$ becomes closest to the EPR pair $\ket{\mathrm{EPR}_\gamma}$ among other foliations $\gamma$.

For a later discussion, let us further rewrite \eqref{eq:trace_distance}. 
First, $\rho_A$ is represented by
\begin{equation}
    [\rho_A]_{IJ}=\sum_{\alpha^\prime=1}^r S_{I\alpha^\prime} \sigma_{\alpha^\prime}^2 S^\dagger_{\alpha^\prime J},
    \label{eq:svd-rho}
\end{equation}
where $S$ is an isometry, $\sigma$ is a singular value matrix\index{singular value matrix} of $\ket{\Psi}$, and $r$ is the Schmidt rank. Note that $r\le \dim\mathcal{H}_{\gamma_\ast}$. %Then, by combining \eqref{eq:svd-rho} with \eqref{eq:sim-trf}, we obtain
%\begin{equation}
%    \blue{\rho_\gamma=M_\gamma^{+} S\sigma^2 S^\dagger M_\gamma.}
%\end{equation}
%Using this, 
{Then, as the positive eigenvalues are common between $\rho_\gamma$ and $\rho_A$,} 
the inner product in $D_\gamma$ can be written as
\begin{align}
    \langle \mathrm{EPR}_\gamma | \rho_\gamma^{1/2}\rangle 
    %&=\frac{\mathcal{N}_\gamma}{\sqrt{\dim\mathcal{H}_{\gamma}}}
    %\tr \left(M_\gamma S\sigma S^\dagger M_\gamma^{-1} \right) \nonumber\\
    &=\frac{\mathcal{N}_\gamma}{\sqrt{\dim\mathcal{H}_{\gamma}}}
    \tr \rho_\gamma^{1/2} \nonumber\\
    &=\frac{\mathcal{N}_\gamma}{\sqrt{\dim\mathcal{H}_{\gamma}}}
    \sum_{\alpha^\prime=1}^r \sigma_{\alpha^\prime} \label{eq:trace-dist}\\
    &\le \frac{\mathcal{N}_\gamma}{\sqrt{\dim\mathcal{H}_{\gamma}}} r^{1/2} \sqrt{\sum_{\alpha^\prime=1}^r \sigma_{\alpha^\prime}^2} \nonumber\\
    &=\mathcal{N}_\gamma\sqrt{\frac{r}{\dim\mathcal{H}_{\gamma}}}.
    \label{eq:tr-dist-ineq}
\end{align}
The last line comes from the normalization $\tr\rho_A=1$. The inequality is saturated only when $\sigma\propto \mathbf{1}$.
We can further rewrite \eqref{eq:trace-dist} in terms of the $n$-th R\'enyi entropy\index{R\'enyi entropy}
\[
S_n\equiv \frac{1}{1-n}\log\tr\rho_A^n =\frac{1}{1-n}\log \sum_{\alpha^\prime=1}^r \sigma_{\alpha^\prime}^{2n}.
\]
Since $S_{1/2}=2\log \sum_{\alpha^\prime=1}^r \sigma_{\alpha^\prime}$, \eqref{eq:trace-dist} is rewritten as
\begin{equation}
    \big|{\langle \mathrm{EPR}_\gamma | \rho_\gamma^{1/2}\rangle}\big|^2=\frac{\mathcal{N}_\gamma^2}{{\dim\mathcal{H}_{\gamma}}}e^{{S_{1/2}}}.
    \label{eq:trace-dist-renyi}
\end{equation}
In any cases, the $\gamma$-dependence in $D_\gamma$ only appears through $\mathcal{N}_\gamma$ and $\dim\mathcal{H}_\gamma$.

When $\gamma\subset \mathcal{C}(A)\cup \mathcal{C}(\bar{A})$, $M_\gamma$ is either %{an} isometry or unitary
isometric or unitary. {Thus, we can apply a standard greedy algorithm\index{greedy algorithm} %applies
in this case. Since the tensors inside the causal cones are reduced to {an} identity after contractions,  a removal of tensors in $\mathcal{C}(A)\cup \mathcal{C}(\bar{A})$ from a state is equivalent to pushing the foliation in $\mathcal{C}(A)\cup \mathcal{C}(\bar{A})$. This means we can perform entanglement distillation {that is} perfectly consistent with the previous proposals. Let us see this from the view point of the trace distance.} By contracting isometries and unitaries in the MERA, this %implies
{indicates that}
    \begin{equation}
        \rho_\gamma^\dagger=\rho_\gamma
        \Rightarrow\mathcal{N}_\gamma=1.
        \label{eq:density-matrix}
    \end{equation}

From these expressions, the following statements can be derived for $\forall \gamma\subset \mathcal{C}(A)\cup\mathcal{C}(\bar{A})$. 
First, the inner product \eqref{eq:trace-dist} monotonically increases as we push $\gamma$ toward a minimal bond cut surface $\gamma_\ast$. This is because $\log\dim\mathcal{H}_\gamma$ is proportional to the number of bonds cut by $\gamma$. Then, from the definition \eqref{eq:trace_distance}, the trace distance monotonically decreases
\begin{equation}
    D_{\gamma^\prime} - D_\gamma < 0
    \label{eq:trace-dist-diff}
\end{equation}
as we push $\gamma$ to $\gamma^\prime$ toward a minimal bond cut surface $\gamma_\ast$. 
Second, 
\begin{equation}
    \gamma\neq\gamma_\ast \Rightarrow D_\gamma>0
    \label{eq:trace-pos}
\end{equation}
since from \eqref{eq:tr-dist-ineq}
\begin{equation}
    \langle \mathrm{EPR}_\gamma | \rho_\gamma^{1/2}\rangle \le \sqrt{\frac{r}{\dim\mathcal{H}_\gamma}}<1,
\end{equation}
where we used $r\le\dim\mathcal{H}_{\gamma_\ast}<\dim\mathcal{H}_\gamma$. The first statement \eqref{eq:trace-dist-diff} supports distilling a state closer to the EPR state by pushing $\gamma$ toward $\gamma_\ast$. The second statement \eqref{eq:trace-pos} indicates that we cannot have $D_\gamma=0$ (distillation of the EPR pair) unless $\gamma=\gamma_\ast$. The vanishing trace distance is equivalent to either a flat entanglement spectrum
\begin{equation}
    r=\dim\mathcal{H}_{\gamma_\ast}\qq{and}\sigma\propto\mathbf{1}
    \label{eq:schmidt}
\end{equation}
from \eqref{eq:tr-dist-ineq} or 
\begin{equation}
    S_{1/2}=\log\dim\mathcal{H}_{\gamma_\ast}
    \label{eq:renyi-1/2}
\end{equation}
from \eqref{eq:trace-dist-renyi}, which is expected from holography~\cite{Bao:2018pvs}.

\section{Numerical results for {random MERA}}\label{sec:HED-num}
In this section, we demonstrate the %above 
{aforementioned} procedure of entanglement distillation in a particular MERA called the random MERA by again calculating a trace distance. The \textbf{random MERA}\index{random MERA} is prepared with \textbf{Haar random unitaries}\index{Haar random unitaries} {$U^{\alpha\beta}_{\gamma\delta}$}, where each index runs from $1$ to $\chi$. Isometries are given by $W^\alpha_{\beta\gamma}=U^{\alpha\, 1}_{\beta\gamma}$ and the top tensor is given by $T_{\alpha\beta}=U^{11}_{\alpha\beta}$.

The random MERA\index{random MERA} is particularly suitable to %examine 
verify our {proposal of entanglement distillation}. Preceding studies have pointed out that a random tensor network can saturate \eqref{eq:TN-EE} in the large bond dimension limit, realizing a discrete version of the RT formula~\cite{Hayden:2016cfa,Swingle:2012wq,Qi:2018shh,Kudler-Flam:2019wtv}. %We would like to see if our method can really 
{The goal of our method is to} extract pure EPR pairs on $\gamma_\ast$ in this limit as expected from holography, {and more importantly, whether this way of entanglement distillation really works even in a finite bond dimension, which is less trivial}.

Numerical calculations have been done for $8$-site and $16$-site random MERAs with bond dimension $\chi$. {For the $8$-site MERA,} we choose a subregion $A$ and foliations as shown in Fig{s}.\ref{fig:mera} and %Fig.
\ref{fig:HED}. 
{For the $16$-site MERA, we choose $6$-site and $8$-site subregions and foliations at their minimal bond cut surfaces.} 
Then, we calculate the trace distance \eqref{eq:trace_distance} to %look at 
{investigate} the closeness between each state and the EPR state. We change the value of $\chi$ to see the %tendency 
{trend} of distillation in the large-$\chi$ limit. Tensor network contractions were performed using quimb\index{quimb}~\cite{Gray2018} and cotengra\index{cotengra}~\cite{Gray2021}.

\begin{figure}[h]
    \centering
    \includegraphics[width=0.7\linewidth]{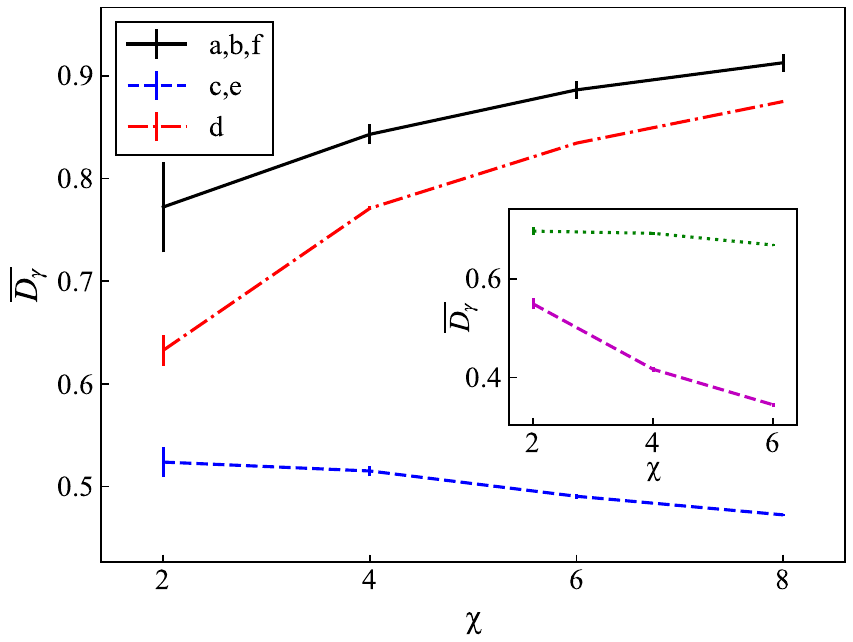}
    \caption{%The 
    {Random}-averaged trace distance  $\overline{D_\gamma}$ %between $\ket{\rho_\gamma^{1/2}}$
    %and the EPR state in the same Hilbert space on
    for each foliation in Fig.\ref{fig:HED} with respect to the bond dimension $\chi$ for the $8$-site random MERA. The inset is $\overline{D_\gamma}$ on the minimal bond cut surface for the $16$-site random MERA (green dotted line with an $8$-site subregion and purple dashed line with a $6$-site subregion).  %between $\ket{\rho_\gamma^{1/2}}$ associated to each foliation $\gamma$ in Fig.\ref{fig:HED} and the EPR state of the same rank.
    }
    \label{fig:trace_distance}
\end{figure}

Fig.\ref{fig:trace_distance} shows the random-averaged trace distance $\overline{D_\gamma}$ for each foliation $\gamma$ in the $8$-site random MERA. Each $\overline{D_\gamma}$ is calculated using {ten} samples. The trace distances for the states on %the 
foliation{s} (a) and (b) are the same due to the equivalence up to a unitary transformation on $\rho_\gamma$. The distances for the states on (a) and (f) are also the same %due to the reflection symmetry. 
{as $A$ and $\bar{A}$ are complement to each other in the pure state.}
It is the same for (c) and (e) {which are related via a common unitary transformation from (a) and (f), respectively}. The state on %the 
{foliation} (d) corresponding to neither $\rho_{A}$, $\rho_{\bar{A}}$ nor $\rho_{\gamma_\ast}$ has a trace distance in between others. {Note that any greedy algorithms can never reach %the 
{foliation} (d) but our method enables us to compute the trace distance even for such a case in a well-defined manner.} 
We can see the foliation  $\gamma=\gamma_\ast$ (c,e) exhibits the smallest trace distance among all the foliations for bond dimensions from 2 to 8. %This tendency agrees with the previous observation in the large-$\chi$ 
The trace distances for (c,e) monotonically decrease as the bond dimension increases, which is consistent with~\cite{Hayden:2016cfa}. %This tendency is
{These %tendencies
{trends} are} also seen in the situation of the $16$-site random MERA (Fig.\ref{fig:trace_distance} inset). This indicates this distillation procedure succeeds on the minimal bond cut surfaces $\gamma_\ast$ for each bond dimension {even when the bond dimension is not large}. 
%\blue{Foliations other than the minimal bond cut surfaces have clearly distinct trends for their trace distances. Those trace distances increase} 
{However,}
%On the other hand, 
the trace distance on the other foliations increases 
as we increase the bond dimension. In this way, the minimal bond cut surface can be characterized from the perspective of distillation.

The behavior in the large-$\chi$ limit can be analytically understood as follows. The previous study~\cite{Hayden:2016cfa} found
\begin{equation}
    \lim_{\chi\rightarrow\infty}\overline{S_n}=\log\dim\mathcal{H}_{\gamma_\ast}
\end{equation}
for a non-negative integer. Assuming its analytical continuation to $n=1/2$ 
\begin{equation}
    \lim_{\chi\rightarrow\infty}\overline{S_{1/2}}=\log\dim\mathcal{H}_{\gamma_\ast},
\end{equation}
holds as expected from holography~\cite{Hayden:2016cfa}, \eqref{eq:trace-dist-renyi} and the Jensen's inequality leads
\begin{align}
    \lim_{\chi\rightarrow\infty}\overline{\big|{\langle \mathrm{EPR}_{\gamma_\ast} | \rho_{\gamma_\ast}^{1/2}\rangle}\big|^2} &=
    \lim_{\chi\rightarrow\infty}\frac{1}{\dim\mathcal{H}_{\gamma_\ast}}\overline{\exp(S_{1/2})} \nonumber\\
    &\ge\lim_{\chi\rightarrow\infty}\frac{1}{\dim\mathcal{H}_{\gamma_\ast}}\exp({\overline{S_{1/2}}})
    = 1.
\end{align}
Since the inner product between normalized states is at most one, we can conclude the distilled state approaches the EPR state for a large bond dimension. Even at a finite $\chi$, the existence of a gap between the distance for $\gamma=\gamma_\ast$ and others is consistent with \eqref{eq:trace-dist-diff}.

\section{Entanglement distillation for MPS}\label{sec:HED-MPS}
Numerically we have seen our distillation method indeed works for the random MERA. To look for a possible extension to other classes of tensor networks{,} %than MERA, let us 
{we} focus on {MPS}, which belongs to a different criticality from {MERA}.

\begin{figure}[h]
    \centering
    \includegraphics[width=0.7\linewidth]{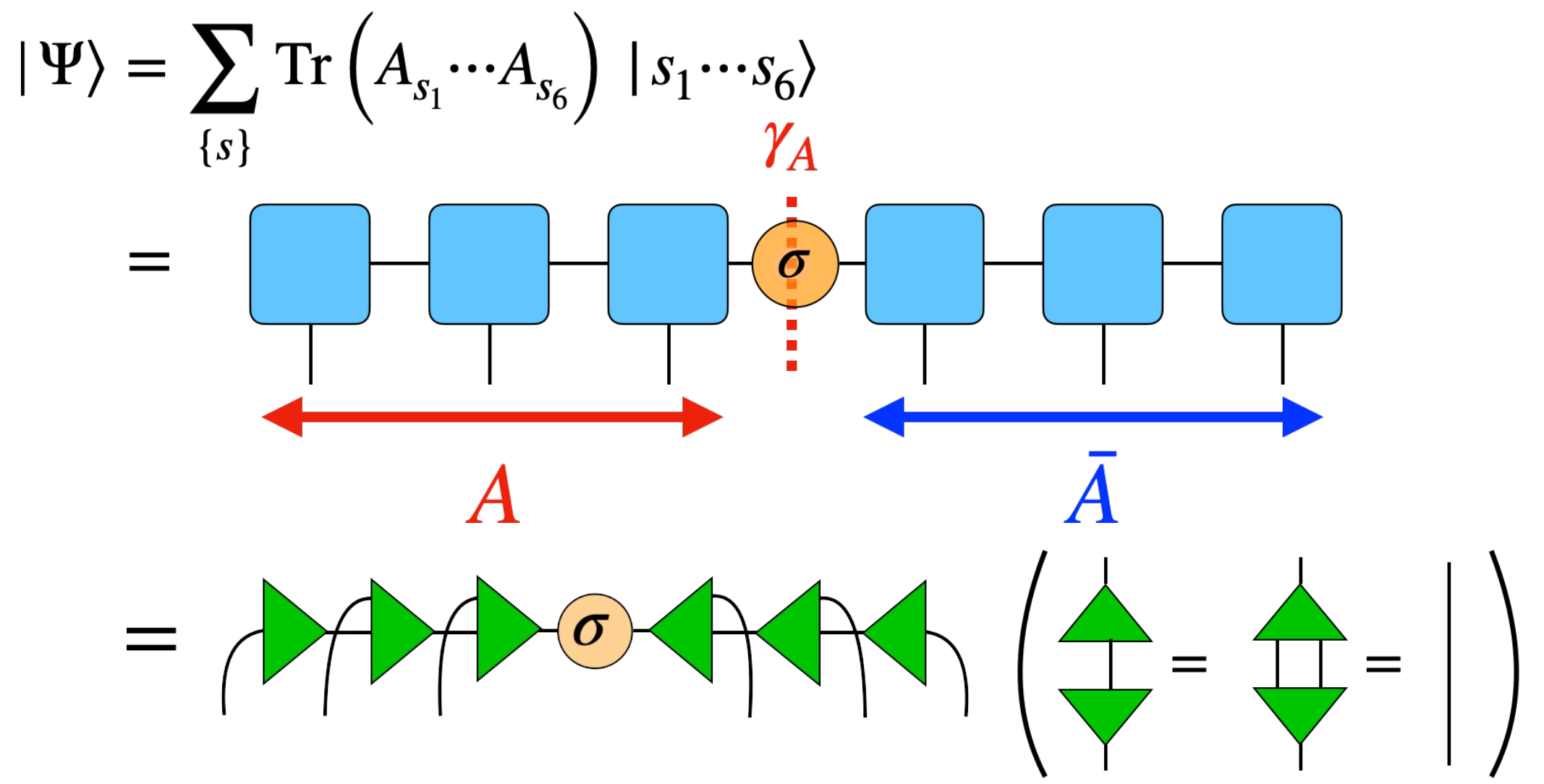}
    \caption{A matrix product state $\ket{\Psi}$ is shown in the first and second lines. The third line represents its mixed canonical form~\cite{SCHOLLWOCK201196} via a successive singular value decomposition for each matrix. Every matrix in the form is isometric (green triangle) and the singular value matrix $\sigma$ is placed in the center. The boundary between $A$ and $\bar{A}$ is denoted by $\gamma_A$.}
    \label{fig:MPS}
\end{figure}

%The 
{An} MPS\index{matrix product state} with open boundaries is shown in {the first and second lines in} Fig.\ref{fig:MPS}. For simplicity, we focus on the case when the subregion $A$ is on the left and its complement $\bar{A}$ is on the right. The boundary between $A$ and $\bar{A}$ is denoted by $\gamma_A$. The bond dimension for internal bonds is denoted by $\chi$. Since we can always transform an MPS in a mixed canonical form, the third line of Fig.\ref{fig:MPS} follows. In the mixed canonical form, every matrix is isometric except for the singular value matrix\index{singular value matrix} $\sigma$.

\begin{figure}[h]
    \centering
    \includegraphics[width=\linewidth]{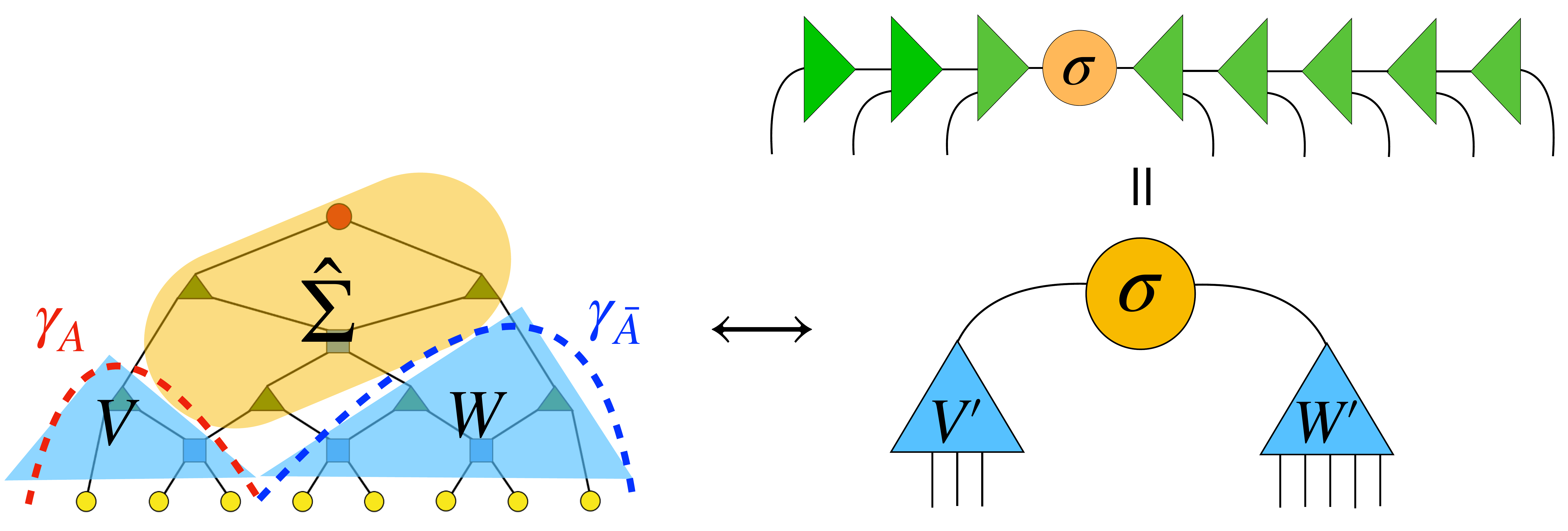}
    \caption{{MERA} can be divided into two isometries, $V$ and $W$, and the %rest
    {remaining} $\hat{\Sigma}${.} %\red{are}
    {This structure of MERA is} in analogy with {that of} MPS, whose isometry on the left (right) of $\sigma$ is collectively denoted by $V^\prime$ ($W^\prime$). {The lower right tensor network is equivalent to a so-called one-shot entanglement distillation tensor network discussed in~\cite{Bao:2018pvs,Lin:2020ufd}.}}
    \label{fig:MERA_SVD}
\end{figure}

Fig.\ref{fig:MERA_SVD} shows a structural similarity between {MPS} {in the form} and {MERA}. The isometric parts of {the} MPS, $V^\prime,W^\prime$, correspond to {those} of {the} MERA, $V,W$, in $\mathcal{C}(A),\mathcal{C}(\bar{A})$. The singular value matrix $\sigma$ in {the} MPS corresponds to $\hat{\Sigma}$ in {the} MERA (or {the} Python's lunch in a holographic context~\cite{Brown:2019rox}). From this {view}point, % of view, 
the MPS is not only another class of tensor networks than MERA, but a simpler model sharing a common isometric structure with {the} MERA.
In the following, we will consider the MPS analogue of the entanglement distillation in {MERA} and compare the results between {the two}. %MPS and MERA.

\begin{figure}
    \centering
    \includegraphics[width=0.5\linewidth]{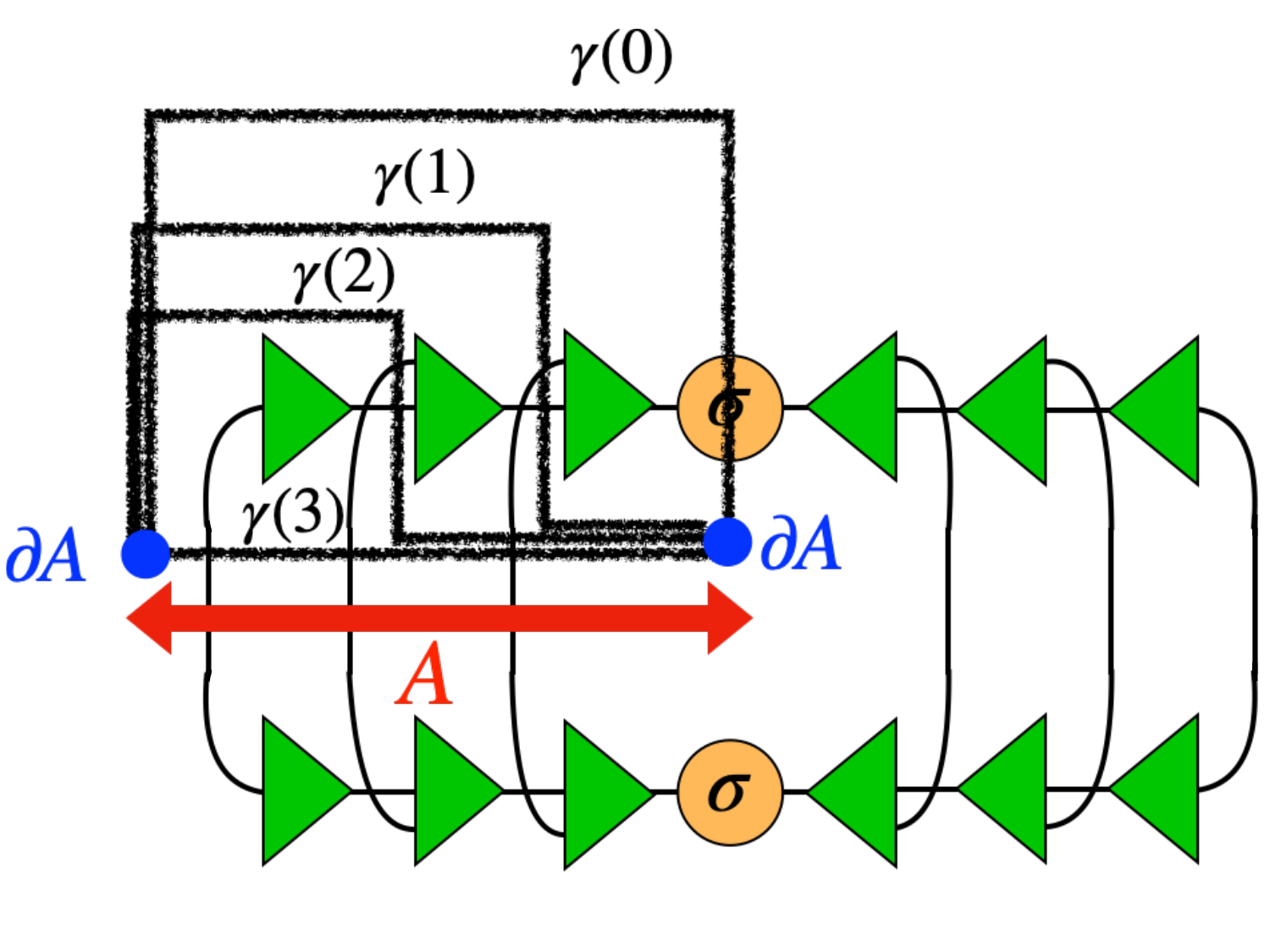}
    \caption{Foliations {interpolating the boundary subregion $A$ and the minimal bond cut surface $\gamma_A$} are denoted by $\{\gamma(\tau)\}$, where $\tau$ is an integer parametrizing each foliation.}
    \label{fig:MPS_fol}
\end{figure}

Through the correspondence in Fig.\ref{fig:MERA_SVD}, we can consider foliations\index{foliation} in {MPS} similar to %that 
{those} within $\mathcal{C}(A)$ in {MERA}. Fig.\ref{fig:MPS_fol} shows a family of foliations {$\{\gamma(\tau)\}_\tau$} in {the} MPS such that their endpoints are always fixed at the boundary of the subregion $\partial A$. Then, $\gamma_A$ can be characterized as a minimal bond cut surface, a foliation that cuts the minimum number of bonds in $\braket{\Psi}$. The foliations are chosen so that the location of the internal bond cut becomes monotonically closer to the minimal bond cut surface $\gamma_A$. The number $\tau$ specifies the number of matrices between the foliation and $\gamma_A$, parametrizing each foliation $\gamma(\tau)$. The previous discussion for the trace distance $D_\gamma$ only relies on the diagonalization of $\rho_A$ and a similarity transformation between $\rho_\gamma$. Thus, \eqref{eq:trace-dist-renyi} is %not restricted to MERA but is 
{also} {applicable} to {MPS}. %as well. 
Given the R\'enyi-$1/2$ entropy\index{R\'enyi entropy} $S_{1/2}$ of $\rho_A$, the trace distances\index{trace distance} are
\begin{align}
    D_{\gamma(\tau)}&=\sqrt{1-\frac{e^{S_{1/2}}}{\chi^{\tau+1}}},\quad \tau=0,1,2\\
    D_{\gamma(3)}&=D_{\gamma(2)}
\end{align}
since $\dim\mathcal{H}_{\gamma(\tau)}=\chi^{\tau+1}$ for $\tau=0,1,2$ and $\dim\mathcal{H}_{\gamma(3)}=\dim\mathcal{H}_{\gamma(2)}$. From this, it is apparent that the distance $D_\gamma$ decreases as the foliation approaches $\gamma_A$, i.e. $\tau$ decreases.
However, for the MPS case, it is more explicit to check {entanglement} distillation by following a state on each foliation.
The resulting reduced transition matrix on $\gamma(\tau)$ is represented by Fig.\ref{fig:MPS_red}. It can be written as 
\begin{equation}
    \rho_{\gamma{(\tau)}}= V_{\gamma{(\tau)}} \sigma^2 V_{\gamma{(\tau)}}^\dagger
\end{equation}
using an isometry $V_{\gamma{(\tau)}}$ composed of $\tau$ layers of isometries. Since only isometries act on the singular value matrix, the entanglement spectrum does not change {while the size of each $\rho_{\gamma(\tau)}$ decreases} during the distillation, which is in accordance with the necessary condition for entanglement distillation. When the foliation is a minimal bond cut surface ($\tau=0$), $V_\gamma$ becomes an identity matrix. This %concludes 
{indicates that} the distilled state via our procedure becomes diagonal $\sum_{\alpha=1}^\chi \sigma_\alpha \ket{\alpha\alpha}$ by removing isometric, redundant {degrees of freedom} from the original state.
In particular, the EPR state $\ket{\mathrm{EPR}_\chi}$ is distilled on $\gamma_A$ whenever $\sigma\propto\mathbf{1}$. Examples of such a state described by {the} MPS includes the {thermodynamic limit of the} valence-bond-solid state, i.e. the ground state of a gapped Hamiltonian called {the} AKLT model~\cite{Affleck:1987cy,Affleck:1987vf}. 

\begin{figure}
    \centering
    \includegraphics[width=0.4\linewidth]{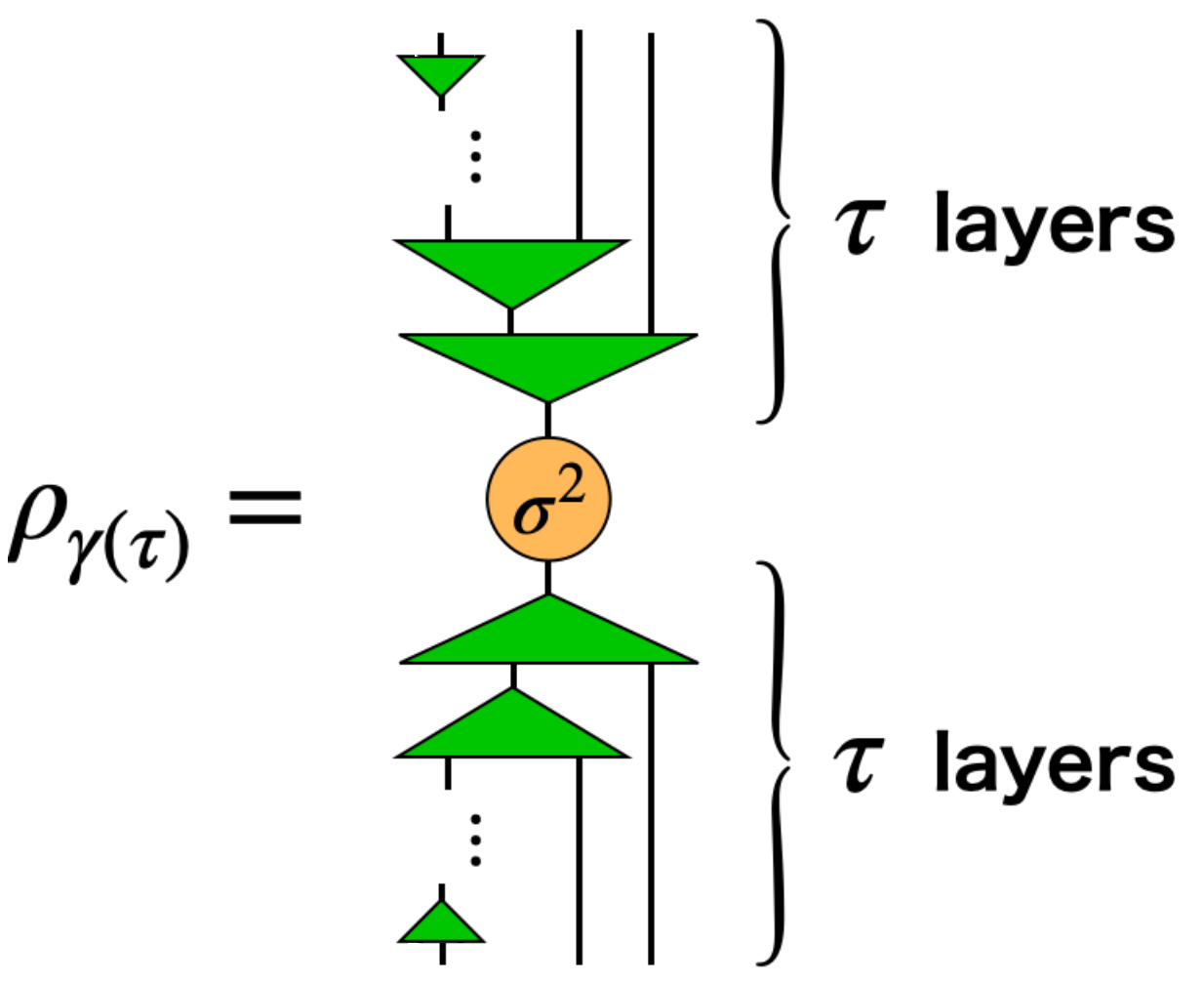}
    \caption{%The 
    {Reduced} transition matrix associated with the foliation $\gamma(\tau)$ in Fig.\ref{fig:MPS_fol}.}
    \label{fig:MPS_red}
\end{figure}

Despite the difference of criticality, the distillation in {MPS} has common features compared %to 
{with} that in {MERA}. In a general MERA, we discussed the monotonicity of the trace distance \eqref{eq:trace-dist-diff} toward $\gamma_\ast$. Furthermore, the distillation of the EPR state in {MERA} was equivalent to the flat entanglement spectrum \eqref{eq:schmidt}. All of these were %just 
shown for {the} MPS as well.

%After all, 
{Overall,} our distillation procedure in {MERA} can be extended to the MPS, where pushing the foliation toward the minimal bond cut surface corresponds to removing extra {degrees of freedom}, and a distillation of EPR pairs yields a flat entanglement spectrum.

%\section{Implementing tensor networks in quantum circuits}
%\mynote{Add if I have time}

%\subsection{Tensor network-inspired entanglement harvesting in quantum many-body computations}
%\mynote{Add if I have time}

%\input{Chapter3/appendix3}
\chapter{Conclusion and discussion}\label{ch:discussion}
% **************************** Define Graphics Path **************************
\ifpdf
    \graphicspath{{Conclusion/Figs/Raster/}{Conclusion/Figs/PDF/}{Conclusion/Figs/}}
\else
    \graphicspath{{Conclusion/Figs/Vector/}{Conclusion/Figs/}}
\fi

\section{Summary}
In Introduction, we asked three questions:
\begin{itemize}
\setlength{\parskip}{0pt} % 段落間
\setlength{\itemsep}{3pt} % 項目間
    \item[\underline{Q1.}] How can entanglement entropy be evaluated in generic QFTs? (Chapter \ref{ch:1})
    \item[\underline{Q2.}] Can we add any additional ingredients in the %standard
    AdS/CFT correspondence to realize more nontrivial spacetime and dynamics? How can we formulate it? (Chapter \ref{ch:2}, \ref{ch:2-2})
    \item[\underline{Q3.}] Can we generalize the formalism of holography itself in light of the entanglement structure? (Chapter \ref{ch:3})
\end{itemize}
We answered these questions in this dissertation with help of QFTs, holography, and tensor networks.
\begin{itemize}
\setlength{\parskip}{0pt} % 段落間
\setlength{\itemsep}{3pt} % 項目間
    \item[\underline{A1.}] We developed a new technique, the $\mathbb{Z}_M$ lattice-like gauge theory on Feynman diagrams, to compute entanglement entropy in generic interacting QFTs based on the orbifold method in free QFTs. This enables us to calculate the single twist contributions from flat-space propagators and vertices, where the latter contribution is identified by us for the first time. These contributions are resummed to all orders by considering the generalized 1PI diagrams. We found they are written in terms of the renormalized two-point functions of fundamental and composite operators in the end. The dominance of such contributions in entanglement entropy was discussed in light of the Wilsonian renormalization group.
    \item[\underline{A2.}] Yes, the AdS/BCFT correspondence can realize a nontrivial spacetime on the EOW brane. We studied a time-dependent setup by considering an operator local quench in BCFTs. 
    It has been pointed out that the ordinary prescription leads to a problem of a self-intersecting brane and the breaking down of the correspondence itself. However, we showed that the usual relation between the bulk and boundary is modified and our correct prescription solves the aforementioned problem by creating a black hole before the problem occurs. We checked our prescription by matching the energy-momentum tensor and entanglement entropy in the AdS and BCFTs through careful treatment of the boundary excitation.
    \item[\underline{A3.}] Probably yes. Given a tensor network representation of quantum many-body states, we established a systematic, quantitative way to investigate how geometry of tensor networks is related to entanglement from an operational perspective.
    Using the concept of pseudo entropy, we showed that the geometry is linked to entanglement distillation. This enables us to geometrically concentrate entanglement of the original state without being destroyed or created. We also evaluated this geometric entanglement distillation quantitatively and showed that the relation holds in various tensor networks.
\end{itemize}

%%%%%%%%%%%%%%%%%%%%%%%%

\section{Future work}
There are several possible directions. We list a few in the following subsections.
\subsection{Calculation of multiple twist contributions}
In Chapter \ref{ch:1}, we discussed single twist contributions. Computing multiple twist contributions in entanglement entropy is difficult since multiple complex integrals or sums in them make the analytical continuation intractable. However, there is room to relax this difficulty by considering the 'continuum' limit $M\rightarrow \infty$. Then, one could replace the summation over the twist $k_j\equiv 2\pi j/M$ by integration of $k$. This is just like a discrete Fourier transform to a continuous Fourier transform. However, this will only compute Reny\'i-0 (max) entropy, not entanglement entropy.\footnote{This calculation is still very ambiguous due to the infinite-dimensional nature of the QFT Hilbert space. Since entanglement entropy itself depends on the regularization scheme, it is not clear how to regulate the divergent integrals. One possibility is to use the same cutoff regularization as the free QFTs. Another possibility is to consider the problem in the original picture of the Riemann surface. The $M\rightarrow \infty$ limit corresponds to zero periodicity in the conical picture. This naively suggests that one can compute Reny\'i-0 entropy from a dimensional reduction. As we will discuss shortly, this regularization scheme dependence is basically equivalent to the choice of the smoothness parameter $\varepsilon$ and we expect the scheme dependence to be small enough for the calculation of entanglement entropy.} Interestingly, we have a formula relating entanglement entropy to a slightly modified version of the Reny\'i-0 entropy of the $n$ copies of the state. This is known as the \textbf{asymptotic equipartition property (AEP)}\index{asymptotic equipartition property}\index{AEP |see asymptotic equipartition property }~\cite{Tomamichel_2009,Wang:2021ptw,Akers:2020pmf}:
\begin{equation}
    \lim_{\varepsilon\rightarrow 0} \lim_{N\rightarrow \infty} \frac{1}{N} S_0^\varepsilon (\rho_A^{\otimes N}) = S(\rho_A).
    \label{eq:AEP}
\end{equation}
$S_0^\varepsilon$ is the \textbf{$\varepsilon$-smoothed max entropy}\index{$\varepsilon$-smoothed max entropy} is defined as the infimum of the Reny\'i entropy for an $\varepsilon$-separated reduced density matrix $\sigma$~\cite{Akers:2020pmf,Konig_2009,1365269,renner2005security}
\begin{equation}
    S_0^\varepsilon(\rho)=\inf_{\sigma\in \mathcal{B}^\varepsilon (\rho)} S_0 (\sigma),
    \label{eq:smooth-ent}
\end{equation}
where infimum ranges over all possible $\sigma$ within an $\varepsilon$-ball $\mathcal{B}^\varepsilon(\rho)$. The distance measure is usually a trace distance but other choices are also valid up to $O(\varepsilon)$~\cite{Akers:2020pmf}.
Intuitively, an $\varepsilon$-separated reduced density matrix is given by cutting off small eigenvalues less than $\varepsilon$. Then, it is straightforward to show such a state has an $O(\varepsilon)$ trace distance from the original reduced density matrix. 
Since \eqref{eq:AEP} involves multiple copies of the original state and a homogeneous cutoff $\epsilon$ on its eigenvalues, we anticipate we could calculate the left-hand side by promoting the theory to the $O(N)$ vector model for instance, and put a UV cutoff to the collective momentum. This smoothing procedure makes entanglement entropy UV finite even in QFTs (see \cite{Akers:2020pmf} and (11) in~\cite{Wilming:2018rvz}).
The remaining subtlety is the infimum in \eqref{eq:smooth-ent}. Since the state after such a momentum regularization is not guaranteed to be the infimum, we can only give an upper bound to entanglement entropy. Nevertheless, we expect the deviation to be only $O(\varepsilon)$ and in particular, negligible in the leading order in $G_N$ for holographic theories~\cite{Czech:2014tva,Bao:2018pvs}.

If we could resolve the issues related to the regularization scheme, we would be in principle able to compute entanglement entropy in generic interacting QFTs order by order. It is also interesting if this type of smoothed entropies can be also resummed to all orders and expressed in terms of renormalized quantities.

\subsection{Application of the AdS/BCFT correspondence to quantum cosmology}
The AdS/BCFT correspondence guides us to a holographic realization of the brane world, which might describe our universe~\cite{VanRaamsdonk:2021qgv,Antonini:2019qkt,Antonini:2022blk}. But how probable does our universe emerge on the brane? To answer this, we consider the \textbf{bubble nucleation}\index{bubble nucleation} in the five-dimensional spacetime due to \textbf{vacuum decay}\index{vacuum decay}~\cite{Coleman:1980aw} to realize our universe on the domain wall. In such models, we can explicitly compute the nucleation rate by the semiclassical analysis based on the Euclidean gravitational action. However, there remains a subtlety in the calculation due to the presence of gravity. By using the AdS/BCFT correspondence, we intend to study it from a UV-complete BCFT and possibly its quantum gravity correction as well. To realize a bubble in the AdS/BCFT setup, we glue two spacetimes together on the brane~\cite{Chen:2020uac,Chen:2020hmv,Anous:2022wqh,Grimaldi:2022suv,Baig:2022cnb}. The advantage of this construction method is the model independence by the \textbf{thin-wall approximation}\index{thin-wall approximation}. Hence, we expect we can obtain a universal expression for the nucleation rate independent of the details of theories.
The nucleation rate can be computed from the partition function of the BCFTs. Although similar analyses based on holography have been done (\cite{Barbon:2011zz,Bachas:2021fqo} for examples), they are either focusing on the gravitational calculation in the asymptotically AdS spacetimes or before the invention of the AdS/BCFT correspondence thus the CFT dual calculation was unclear.

What is even more exciting and of great relevance to our work~\cite{Kawamoto:2022etl} is the \textbf{catalytic effect}\index{catalytic effect} induced by black holes. It has been discussed that the existence of black holes catalyzes bubble nucleation. In this dissertation, we discussed a local operator quench in the AdS/BCFT correspondence, in which the conformal dimension of the operator is below the black hole threshold. We intend to extend this to the black hole regime and study the catalytic effect from BCFTs. These studies are work in progress with Naritaka Oshita and Issei Koga.

\subsection{Communication between brane and bulk through the AdS/BCFT correspondence and the intermediate picture}
The AdS/BCFT correspondence allows us to consider a brane world in a holographic setup. If we integrate out the bulk, we obtain a brane coupled to a bath at the asymptotic boundary. This picture is called a \textbf{brane perspective}\index{brane perspective} or the \textbf{intermediate picture}\index{intermediate picture}. In particular, if the CFT on the brane is also holographic, we can think three descriptions are equivalent: the bulk, the dual BCFT, and the intermediate picture. This doubly holographic setup has also been discussed in the context of black holes and the island formula\index{island formula} (AMMZ model\index{AMMZ model} \cite{Almheiri:2019hni,Almheiri:2019psf}).\footnote{For reviews of recent progress in quantum gravity, in particular black holes related to replica wormholes\index{replica wormholes}, island, and JT gravity\index{JT gravity}, see \cite{Almheiri:2020cfm,Maxfield:2021,Kundu:2021nwp,Sarosi:2017ykf,Chen:2021lnq} as well as the original articles, often referred to as the east coast paper~\cite{Almheiri:2019qdq} and the west coast paper~\cite{Penington:2019kki}.}\footnote{There is a proposal in string theory as well as these bottom-up constructions~\cite{Karch:2022rvr}.} Physics on the brane is obviously constrained by the bulk causality. Suppose an observer on the brane wants to communicate with the other observer on the asymptotic boundary to complete a certain quantum task between those two. The time for its completion is bounded by the bulk causality. Then, from double holography, we are interested in understanding it from the brane causality, which is more nontrivial due to the existence of the bulk shortcut and nonlocality in its effective theory~\cite{Omiya:2021olc}.

Furthermore, there are some cases where signals cannot meet each other in the intermediate picture while they can in the bulk. Then, how can the quantum task be possible in the intermediate picture? To answer this question, we think the ``\textbf{connected wedge theorem}\index{connected wedge theorem}'' proposed in \cite{May:2019yxi} and later extended in \cite{May:2019odp,May:2021zyu,May:2021nrl,May:2022clu} is useful. In these papers, quantum tasks whose inputs and outputs are at the asymptotic boundary are discussed. From \textbf{quantum cryptography}\index{quantum cryptography}, it is known that we need a certain amount of mutual information to complete the task whenever signals cannot meet each other. Holography tells us this is equivalent to a connected entanglement wedge~\cite{Headrick:2014cta,Jafferis:2015del} for some regions. This has been proven with the \textbf{focusing conjecture}\index{focusing conjecture} in AdS~\cite{May:2019odp}. These studies indicate that general relativity knows quantum information and vice versa. However, they are limited to cases with asymptotic observers. We intend to study how quantum tasks can be realized in holographic QFTs and also from the perspective of a brane observer~\cite{Neuenfeld:2021wbl,Neuenfeld:2021bsb,Omiya:2021olc,Suzuki:2022xwv,Almheiri:2019hni} by investigating whether the connected wedge theorem holds in this case. Mutual information is expected to be calculated by the RT formula, however, the standard prescription no longer applies as the subregions we need to consider on the brane and the asymptotic boundary are not connected by the usual entanglement wedge enclosed by geodesics. Since the connectedness of the entanglement wedge has more information compared to entanglement entropy, this direction is also intriguing from the perspective of understanding the entanglement structure. We are currently working in this direction with Beni Yoshida.\footnote{Following the submission of this thesis, our work mentioned here has been published~\cite{Mori:2023swn}.}

Since a brane can be behind the black hole horizon, we could also consider a case in which some portion of signals is scrambled inside the black hole for an evaporating black hole. This case is harder than the abovementioned case as it inevitably includes a backreaction to spacetime. We think the generic behavior can be captured from the \textbf{out-of-time-ordered correlators (OTOCs)}\index{out-of-time-ordered correlator}\index{OTOC |see out-of-time-ordered correlator } in CFTs~\cite{Shenker:2013pqa} as well as a measurement-device-independent quantum key distribution in quantum cryptography~\cite{Lo_2012}.
We strongly believe these information-theoretic analyses help solve all the research challenges and understand holographic spacetime from observers.

\subsection{Tensor network-based entanglement harvesting in quantum circuits and relativistic quantum information}
The geometric entanglement distillation presented in Chapter \ref{ch:3} can be implemented in quantum circuits\index{quantum circuits}~\cite{pirsa_22100111}. For example, a MPS can be written as unitary operators acting on some qubits and ancillary qubits with some postselections\index{postselection}\footnote{Note that postselections can be realized probabilistically by a projective measurement.} by maximally entangled states (Fig.\ref{fig:MPS-QC}). Using this implementation, we can rewrite the geometric distillation in MPS at the minimal surface (Fig.\ref{fig:MPS_fol}) as a unitary process with postselections on the ground state. See Fig.\ref{fig:MPS-QC-distil} for the case of $\chi=2$. Note that we do not need to express the ground state by MPS anymore. Instead, we need to optimize the unitaries so that the resulting state becomes close to the EPR state. Since this uses the ground state as a resource like \textbf{quantum annealing}\index{quantum annealing}~\cite{deFalco:1998an,Albash_2018} while the operations are \textbf{gate-based}\index{gate-based}~\cite{Deutsch1989QuantumCN}, it is a hybrid style of quantum computations assisted by tensor network, which has not been explored much.
\begin{figure}
    \centering
    \includegraphics[width=\linewidth]{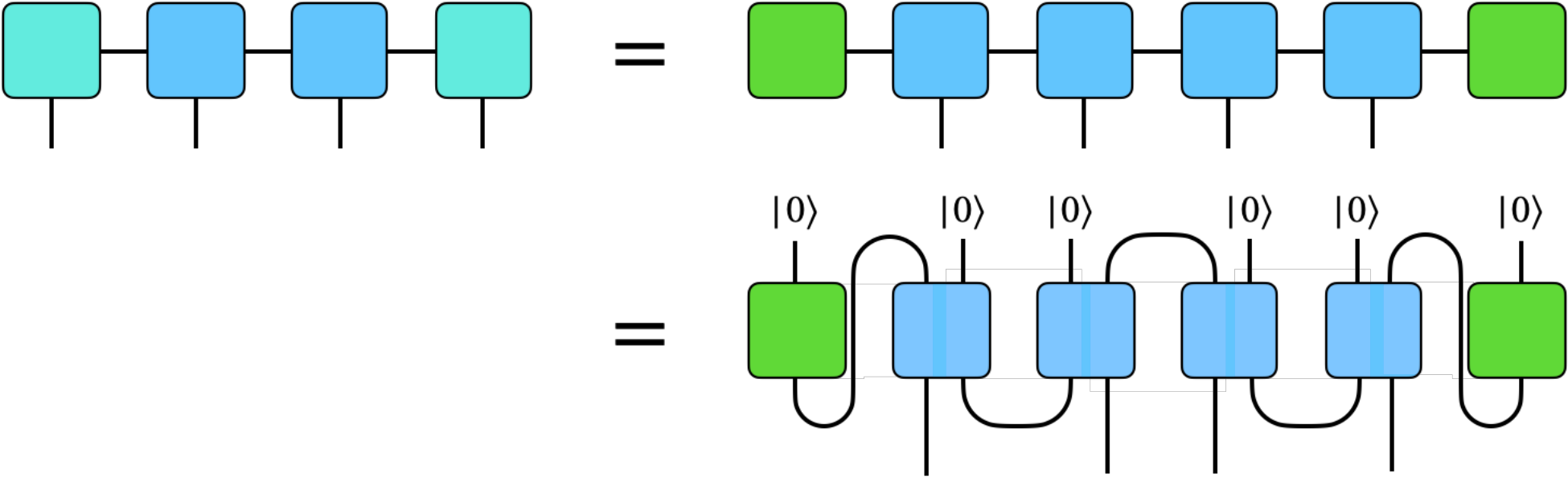}
    \caption{A MPS can be written as suitable one-qubit and two-qubit unitaries acting on and postselected by some ancillary qubits and maximally entangled states \eqref{eq:epr-tensor}.}
    \label{fig:MPS-QC}
\end{figure}
\begin{figure}
    \includegraphics[width=\linewidth]{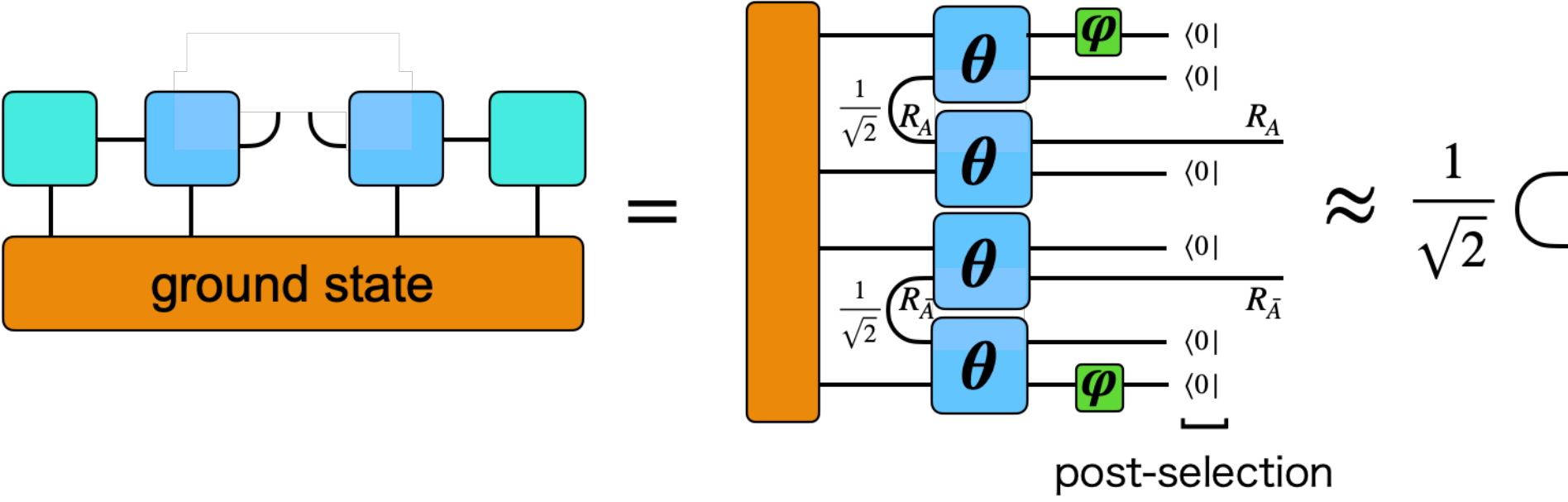}
    \caption{A geometric distillation at the minimal bond cut surface in MPS can be recast as a unitary evolution and postselections on the ground state. $R_A$ and $R_{\bar{A}}$ denote the Hilbert space of one of the qubits of the auxiliary EPR state within the subregion $A$ and $\bar{A}$, respectively. After the postselections, we have two qubits belonging to $R_A$ and $R_{\bar{A}}$ that are originally in the reference system.}
    \label{fig:MPS-QC-distil}
\end{figure}

The parameter optimization is expected to be completed by classical communications since the parameters are just c-numbers. The unitaries and postselections are local within each subregion. After all, we expect this indeed provides a variational LOCC algorithm for the entanglement extraction. As an entanglement distillation protocol, this is not efficient as it consumes two EPRs to generate one EPR. Instead, this operation transfers the internal bipartite entanglement in a many-body ground state to maximal entanglement between auxiliary qubits. Thus, we think this is a tensor network example of \textbf{entanglement harvesting}\index{entanglement harvesting} from a many-body ground state, which in most cases has been discussed in relativistic quantum information\index{relativistic quantum information} utilizing the \textbf{Unruh-de Witt-type detectors}\index{Unruh-de Witt(-type) detectors}~\cite{Reznik:2002fz,Salton:2014jaa,Henderson:2017yuv,Martin-Martinez:2015qwa,Reznik:2003mnx,VerSteeg:2007xs,Martin-Martinez:2013eda}. We hope this new perspective from tensor networks can help construct a better detector and discuss information-theoretic bounds for the amount of extractable entanglement. We are working in this direction with Erickson Tjoa.

\subsection{The continuum limit of geometric distillation}
Although we revealed the relation between geometry and entanglement distillation in Chapter \ref{ch:3}, the geometry is discrete. This is unsatisfactory from the original motivation of extending holography. Our universe is continuous (at least up to the Planck scale)! What would be the counterpart of entanglement distillation in tensor network? One possibility is to take the tensor network picture seriously. In particular, MERA has some interpretations as a series of discrete conformal transformations~\cite{Milsted:2018san,Milsted:2018vop,Milsted:2018yur}. Since the previous study considers a homogeneous conformal transformation, it is anticipating an inhomogeneous deformation like the \textbf{sine-square deformation (SSD)}\index{sine-square deformation}\index{SSD |see sine-square deformation }~\cite{10.1143/PTP.122.953}. Another possibility is an inhomogeneous \textbf{$T\overline{T}$ deformation}\index{$T\overline{T}$ deformation}~\cite{McGough:2016lol} as its holographic dual is expected to be a finite radius CFT obtained by 'pushing' the boundary into the bulk, which resembles the geometric distillation procedure in tensor networks.

\appendix
\chapter{Appendices for Chapter I}
\ifpdf
    \graphicspath{{Chapter1/Figs/Raster/}{Chapter1/Figs/PDF/}{Chapter1/Figs/}}
\else
    \graphicspath{{Chapter1/Figs/Vector/}{Chapter1/Figs/}}
\fi
\graphicspath{{./Chapter1/Figs/}}
\section{Quantum states}\label{app:quantum}
%\section{Appendix: Quantum states, fields, and operator algebras}\label{app:quantum}
%\mynote{I can add descriptions for quantum fields and operator algebra}
%\subsection{Quantum states}
Say $\mathcal{H}$ is $n$-dimensional complex vector space $\mathbb{C}^n$, a set of normalized states spans complex codimension-one unit sphere $S^{2n-1}$ in $\mathbb{C}^n$. Normalized \textbf{pure states}\index{pure states} are \textbf{rays}, obtained from $S^{2n-1}$ with \textit{suitable} identification i.e. the \textbf{projective Hilbert space} $P(\mathcal{H}_n)$.

\begin{definition}
	
	A (projective) \textbf{ray}\index{ray} [$\ket{\psi}$] is an equivalence class of state vectors (ket) in a complex Hilbert space $\mathcal{H}$ defined by
	\begin{equation}
		\ket{\psi}\sim \lambda \ket{\psi},\ \lambda\in \mathbb{C}.
	\end{equation}
	For a normalized ($\braket{\psi}{\psi}=1$) state, it is simply an equivalence class up to a global phase:
	\begin{equation}\label{eq:ray}
		\ket{\psi}\sim e^{-i\alpha}\ket{\psi},\ \alpha\in \mathbb{R}.
	\end{equation}
	Bra vectors can be defined in the same way but with a dual Hilbert space, $\mathcal{H^\ast}:\mathcal{H}\rightarrow \mathbb{C}$, instead of $\mathcal{H}$.
\end{definition}

\begin{definition}
	The \textbf{projective Hilbert space}\index{projective Hilbert space} $P(\mathcal{H})$ of a (complex) Hilbert space $\mathcal{H}$ is the set of rays. In terms of mathematics, it is \textbf{projectivization} of $\mathcal{H}$.
\end{definition}

Therefore, a pure state is not an element of $\mathcal{H}$. Nevertheless, throughout the thesis, we simply denote a pure state as $\ket{\psi}\in\mathcal{H}$. This abuse of notation is just for the sake of simplicity.

While pure states are commonly used in quantum mechanics and isolated quantum systems, quantum states are not limited to those. When one tries to describe a \textbf{mixed state}\index{mixed state}, a probabilistic mixture of pure states, the notion of \textbf{density matrices} is necessary. In the following, a set of linear operators on $\mathcal{H}$ is denoted by $\mathcal{L}(\mathcal{H})$.
\theoremstyle{definition}
\begin{definition}\label{def:dens}
	A \textbf{density matrix (operator)}\index{density matrix}\index{density operator|see {density matrix}} is an element of \textbf{state space} defined below.

	\textbf{State space}\index{state space} $\mathcal{S}(\mathcal{H})$ is a subset of all Hermitian operators\footnote{A Hermitian operator $A$ is defined as a linear operator on a complex Hilbert space s.t. $A^\dagger = A$} $\rho\in\mathcal{ L }(\mathcal{H})_h$ satisfying two following conditions:
	\begin{enumerate}
		\item $\rho\ge 0$,
		\item $\Tr \rho =1$.
	\end{enumerate}
In short,
\begin{equation}
	\mathcal{S}(\mathcal{H}):=\{\rho\in\mathcal{ L }(\mathcal{H})_h | \rho\ge 0, \Tr\rho =1\},
\end{equation}
where $\rho$ is called a \textbf{density matrix}.
\end{definition}

A density matrix can describe a probabilistic mixture (or statistical ensemble). This becomes clear from the alternative definition below.

\theoremstyle{theorem}
\begin{theorem}\label{thm:def-dens}
	\underline{\textbf{Definition of density matrix as a probabilistic mixture}}
	
	Instead of the Definition \ref{def:dens}, a density matrix\index{density matrix} can be alternatively defined as a probabilistic mixture (convex combination) of pure states (projection operators on each basis) $\ket{\phi_i}\bra{\phi_i}$:
	\begin{equation}\label{eq:dens}
		\rho=\sum_i p_i \ket{\phi_i}\bra{\phi_i},
	\end{equation}
	where $p_i \in \mathbb{R}_{\ge 0}$ is a probability of obtaining $\rho=\ket{\phi_i}\bra{\phi_i}$ after a basis measurement (one of the projective measurements); consequently, $\sum_i p_i =1$ and $\sum_i \ket{\phi_i}\bra{\phi_i} = \bm{1}$.
\end{theorem}

\begin{proof}
	$ $%for line break after "Proof."
	\begin{enumerate}[label=(\arabic*)]
	\item \underline{Definition \ref{def:dens}$\Rightarrow$Theorem \ref{thm:def-dens}}
	
	Since $\rho\ge 0$, $\rho\in\mathcal{S}(\mathcal{H}_d)$ can be expressed as $\rho=\sum_{i=1}^d p_i \ket{\phi_i}\bra{\phi_i},\ ^\forall p_i\ge 0$ (eigenvalue decomposition) with a complete basis set $\{\ket{\phi_i} \}_{i=1,\cdots d}$. The normalization, $\Tr \rho=1$, yields $\sum_{i=1}^d p_i =1$.
	
	\item \underline{Theorem \ref{thm:def-dens}$\Rightarrow$Definition \ref{def:dens}}
	
	For any vector $\ket{\xi}\in\mathcal{H}$,
	\begin{equation}
		\bra{\xi}\rho\ket{\xi}=\sum_i p_i |\braket{\phi_i}{\xi}|^2 \ge 0 \Rightarrow \rho\ge 0.
	\end{equation}
	Since $\ket{\phi_i}$ is normalized, $\Tr \rho=1$ is equivalent to $\sum_i p_i =1$.
	\end{enumerate}
\end{proof}

Under this definition (\ref{eq:dens}), $\rho$ is a \textbf{mixed state} if each probability is strictly less than 1. It is written as a strictly ({convex combination}) or nontrivial probabilistic combination of other pure states.\footnote{$\mathcal{S}(\mathcal{H})$ is a convex set of the real vector space $\mathcal{ L }(\mathcal{H})_h$.}

\textbf{Mixed states} are states other than pure states. Thus it is necessary to redefine a pure state as density matrix. The density matrix formalism naturally involves above mentioned pure states as $\rho=\ket{\psi}\bra{\psi}\in \mathcal{S}(\mathcal{H})$. Equivalent definitions for pure states are as follows:
\begin{enumerate}
	\item $\rho$ is a pure state: $\rho=\ket{\psi}\bra{\psi}$
	\item $\rho^2=\rho$
	\item The \textbf{purity} is unity: $\Tr \rho^2 =1$
	\item $\rho\in\mathcal{H}_d$ has one eigenvalue (probability) 1 and $(d-1)$-degenerate eigenvalues 0
	\item $\rho$ is an extreme point of $\mathcal{S}(\mathcal{H})$.
\end{enumerate}
From \#1, the pure density matrix is apparently a projector. (\#1$\Rightarrow$\#2) From the normalization, \#2$\Rightarrow$\#3. \#4 is understood as \#1 in (\ref{eq:dens}). \#5 follows from \textit{reductio ad absurdum}.

\section{Quantum measurements}\label{app:meas}
The most generic measurement process in quantum theories is known as \textbf{positive operator valued measure (POVM) measurements}\index{POVM measurements}. For its proof, refer \cite{nielsen_chuang_2010}. Here we intend to how quantum measurements are described in the POVM formalism.

%Given an observable $\mathcal{M}$, there is always a spectral decomposition such that
%\begin{equation}
%    \mathcal{M}=\sum_k m_k E_k,
%\end{equation}
%where $\{m_k\}_k$ is a set of eigenvalues
Quantum measurements are specified by a set of \textbf{POVM elements}\index{POVM element} $\{E_m\in \mathcal{L}(\mathcal{H})\}_m$. $m\in\mathcal{M}$ is a label for the measurement outcomes $\mathcal{M}$. The probability of outcome $m$ under a state $\rho$, $\mathrm{Prob}(m|\rho)$, is given by
\begin{equation}
    \Tr(E_m \rho).
\end{equation}
Thus, $E_m$ span all possible measurements, i.e. the completeness relation
\begin{equation}
    \sum_m E_m=\mathbf{1}
\end{equation}
holds,
and they are positive operators such that $E_m\ge 0$. $\sqrt{E_m}$ is called a \textbf{measurement operator}\index{measurement operator}.

In a \textbf{direct measurement}\index{direct measurement} of an observable $\hat{A}\in\mathcal{L}(\mathcal{H})_h$,
the probability of outcome $m$ is the probability of obtaining an eigenvalue $a_m$ under a state $\rho$. This follows the Born rule, 
\begin{equation}
    \mathrm{Prob}(m|\rho)=\Tr (\hat{P}_m \rho),
\end{equation}
where $\hat{P}_m$ is a projection operator onto the eigenspace of $\hat{A}$ with eigenvalue $a_m$. Note that the observable $\hat{A}$ has a spectral decomposition
\begin{equation}
    \hat{A}=\sum_m a_m \hat{P}_m.
\end{equation}
Since every POVM element is a projector in a direct measurement, %of some observable, 
this is classified as a \textbf{projection valued measure (PVM) measurement}\index{PVM measurement} or a \textbf{projective measurement}\index{projective measurement}.
When the eigenvalues are nondegenerate, the spectral decomposition equals the eigenvalue decomposition. Such a measurement is called the \textbf{basis measurement}\index{basis measurement}. 

On the other hand, when not all the POVM elements are projectors, this is not a PVM measurement. For example, in a qubit system, a measurement with the POVM elements
\begin{equation}
    E_1 = \dyad{0}, \quad E_2 = \dyad{+} = \frac{(\ket{0}+\ket{1})(\bra{0}+\bra{1})}{2}, \quad E_3=\mathbf{1}-E_1-E_2,
\end{equation}
is not a PVM measurement.

In an \textbf{indirect measurement}\index{indirect measurement}, the measurement process is truly described by the POVM formalism, not by the PVM formalism, when we focus on the target system.
This is actually a very common measurement in experiments, where we couple the target system to an ancillary system like a detector. What we really measure is observables for the ancilla (\textbf{meter observable}\index{meter observable}). The indirect measurement is described as follows. First, besides the target system $S$, prepare an ancillary system $R$. Then, the initial state $\rho\otimes\sigma$ for the entire system $S\cup R$ evolves by a unitary evolution operator $U(t)$,
\begin{equation}
    \rho\otimes\sigma\mapsto U(t) (\rho\otimes\sigma) U(t)^\dagger.
\end{equation}
Followed by the tensor network representation introduced in Appendix \ref{app:TN-rep}, the state after time $t$ is
\begin{equation}
    \includegraphics[height=3.9cm, valign=c, clip]{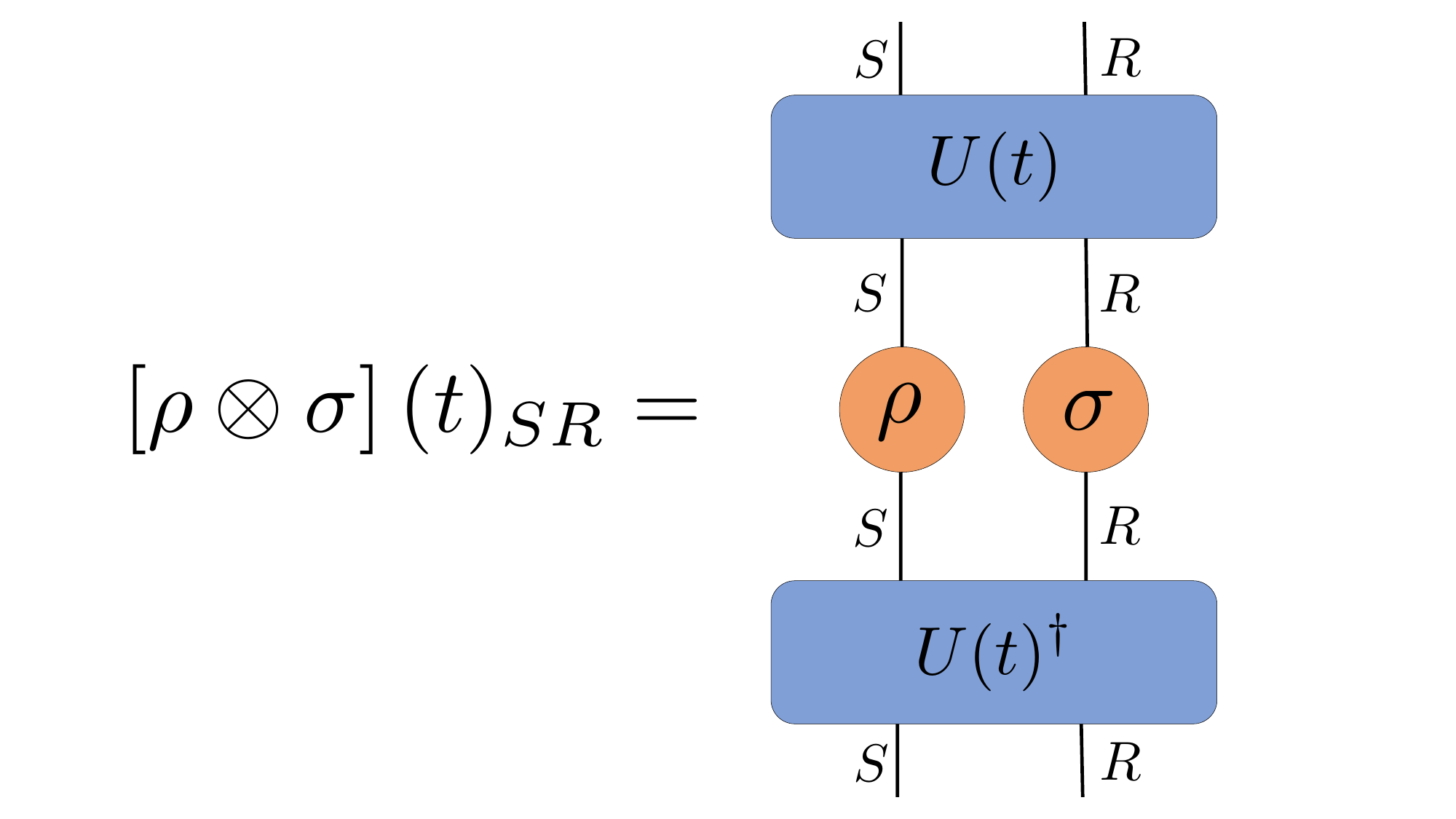}.
\end{equation}
Then, the probability of outcome $m$ upon the projective measurement of the ancilla $R$ is given by
\begin{equation}
    \includegraphics[height=6cm, valign=c, clip]{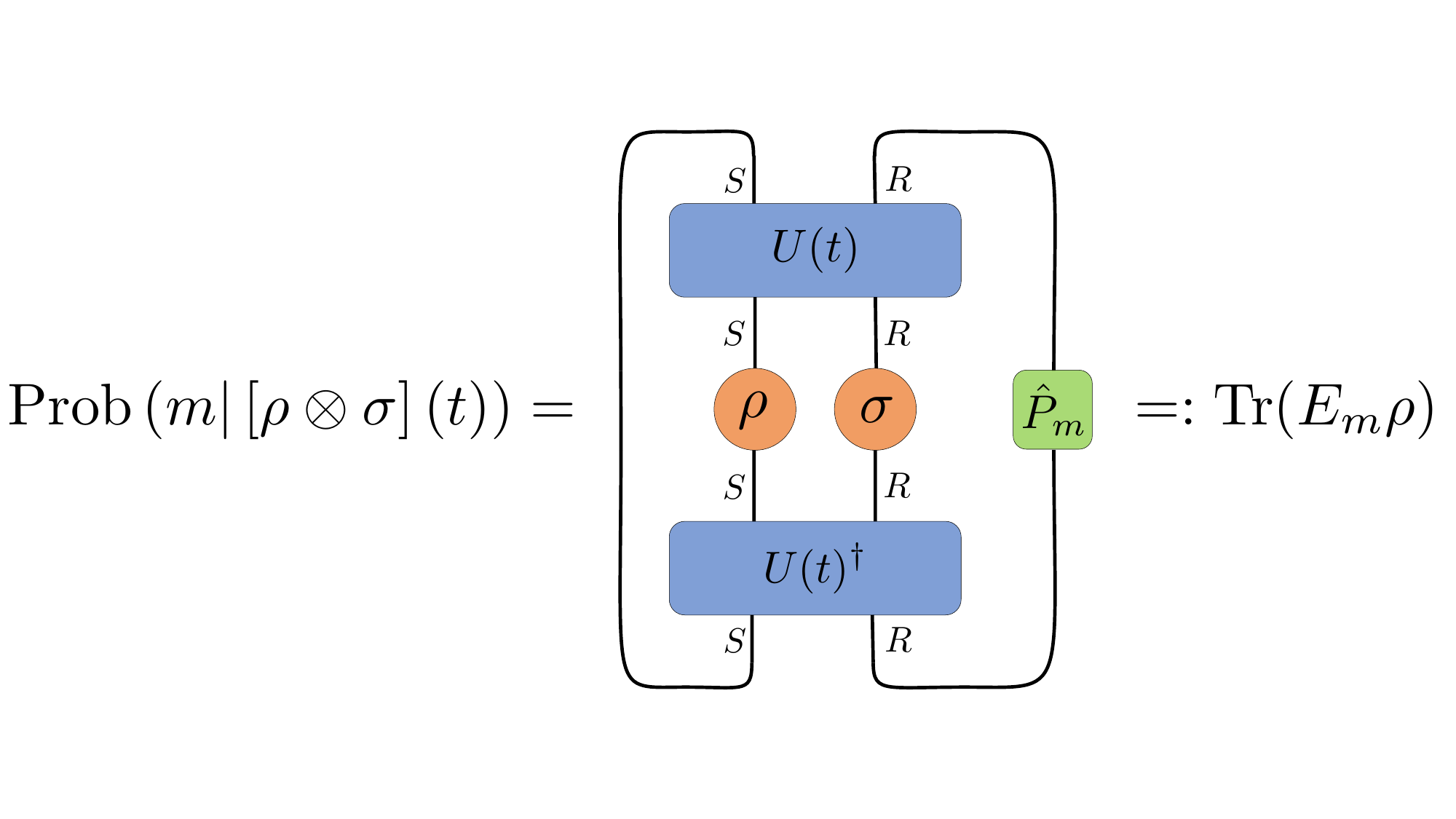}.
\end{equation}
Consequently, the POVM element is 
\begin{equation}
    \includegraphics[height=4.4cm, valign=c, clip]{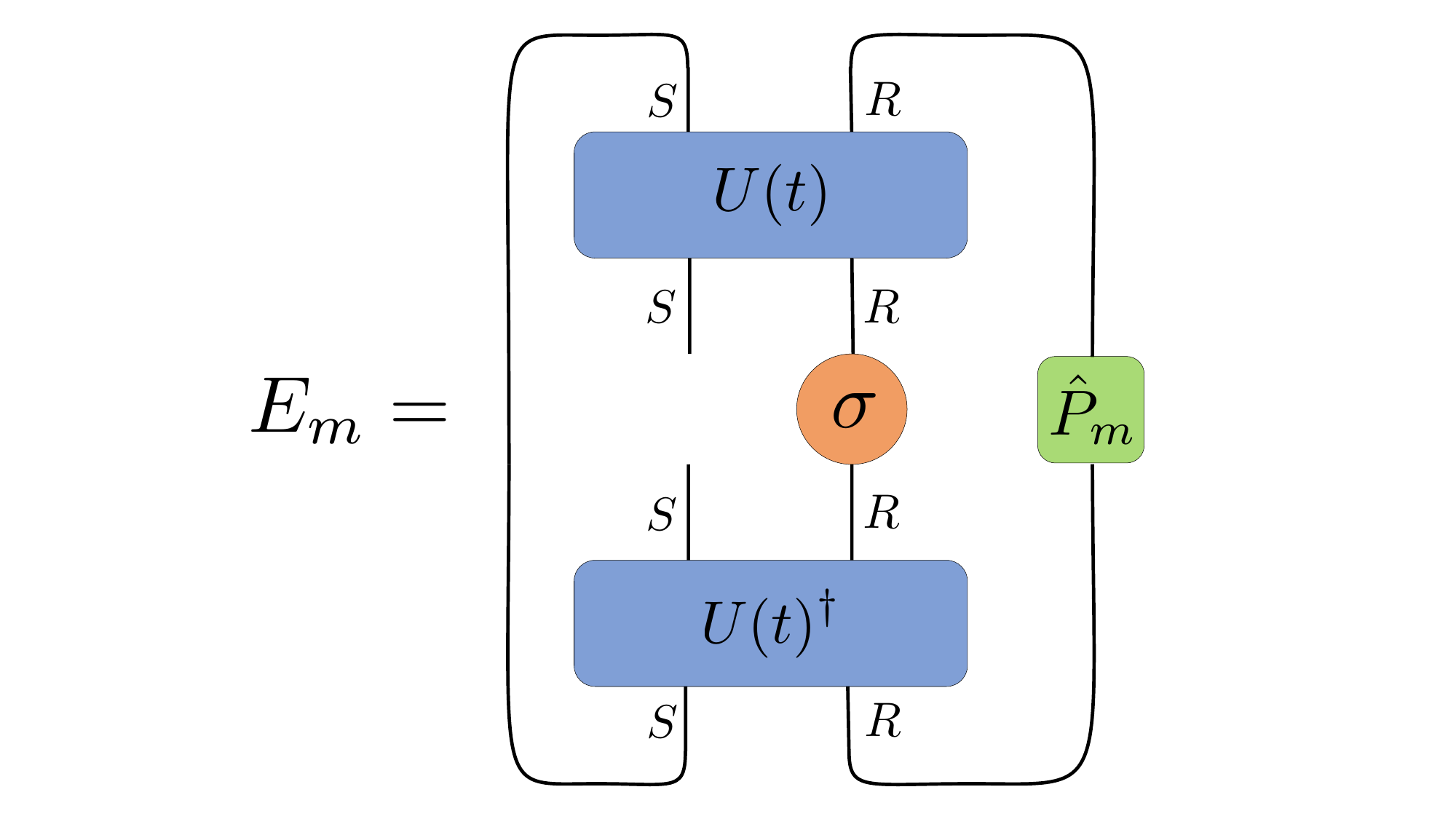}.
\end{equation}
Thus, indirect measurements are POVM measurements. Conversely, any POVM measurements can be realized by an indirect measurement~\cite{BB09394904}.

So far we have discussed the measurement process itself. What will be the state after the measurements? There are two types of measurements. One is the \textbf{selective measurement}\index{selective measurement}, where you obtain information on the outcome. In particular, if the POVM element is a projection operator $\hat{P}_m$ (i.e. a PVM measurement), the state after the measurement of outcome $m$ under a state $\rho$ is given by
\begin{equation}
    %\rho\mapsto 
    \rho_m\equiv \frac{\hat{P}_m \rho \hat{P}_m}{\Tr(\hat{P}_m\rho)}.
    \label{eq:proj-meas}
\end{equation}
This is also known as the von Neumann's projective measurement.\footnote{$\hat{P}_i \hat{P}_j= \hat{P}_i \delta_{ij}$ is implicity assumed. If not, the selective measurement is called the L\"{u}der's projective measurement.}
The other type of measurements is the \textbf{non-selective measurement}\index{non-selective measurement}, where you do not obtain information on the outcome. Instead, the state after the measurement process is given by a probabilistic mixture
\begin{equation}
    \rho\mapsto \rho^\prime = \sum_m \mathrm{Prob}(m|\rho) \rho_m.
    \label{eq:non-select}
\end{equation}

Note that POVM itself does not tell us states after the measurement. It only tells us about the probability distribution of obtaining outcome. We could have two different POVMs describing the same physical measurement process. To precisely discuss states after measurements and evolutions, we should better use the notion of CP-instruments\index{CP-instruments}, which will be explained in Appendix \ref{app:CP-inst}.

\nomenclature[z-quant]{POVM}{positive operator valued measure}
\nomenclature[z-quant]{PVM}{projection valued measure}

\section{Quantum channels}\label{app:quant-channel}
Any quantum evolutions $\Lambda$ must map a density matrix to another. In quantum information theory, it is called a \textbf{quantum channel}\index{quantum channel} in general. In the following, we discuss the necessary and sufficient conditions for a map to be a quantum channel and how it can be understood as a quantum evolution naturally. Note that quantum evolutions involve measurement processes described in the previous section as well \cite{R300000001-I031534480-00}.

\subsection{Completely positive trace-preserving (CPTP) maps}\label{sec:CPTP}
When we allow discarding a part of the system (e.g. environment) or measurements via a detector, the space of the initial density matrix $\mathcal{S}(\mathcal{H}_A)$ need not coincide with the final one $\mathcal{S}(\mathcal{H}_B)$. 
The addition and multiplication rules of probability lead to the affinity of $\Lambda$. %that is, %defined as follows.
\theoremstyle{definition}
\begin{definition}
\underline{Affinity of a quantum channel $\Lambda$}

For any $\rho,\sigma\in \mathcal{S}(\mathcal{H}_A)$ and $p\in [0,1]$, the \textbf{affinity}\index{affinity} of a quantum channel $\Lambda:\mathcal{S}(\mathcal{H}_A)\rightarrow \mathcal{S}(\mathcal{H}_B)$ means
\begin{equation}
    \Lambda(p\rho+(1-p)\sigma)=p\Lambda(\rho)+(1-p)\Lambda(\sigma)\in\mathcal{S}(\mathcal{H}_B).
    \label{eq:affine}
\end{equation}
This is also called a convex-linear map on the set of density matrices~\cite{nielsen_chuang_2010}.
\end{definition}

Although $\Lambda$ is defined only on the set of density matrices $\mathcal{S}(\mathcal{H}_A)$, one can find a natural extension to a set of linear operators $\mathcal{L}(\mathcal{H}_A)$ by the following procedure. 

First, we see any linear operators can be expressed as a linear combination of density matrices. Given an orthonormal basis $\{\ket{i}\}_i$ of $\mathcal{H}_A$, a linear operator $\hat{O}\in\mathcal{L}(\mathcal{H}_A)$ is written as
\begin{equation}
    \hat{O}= \sum_{i,j} O_{ij} \dyad{i}{j}. 
    \label{eq:op-basis}
\end{equation}
The problem is whether we can express $\dyad{i}{j}$ in terms of density matrices. Let us define four state vectors,
\begin{align}
    \ket{+}_{ij} &= \frac{1}{\sqrt{2}}(\ket{i}+\ket{j})\\
    \ket{-}_{ij} &= \frac{1}{\sqrt{2}}(\ket{i}-\ket{j})\\
    \ket{R}_{ij} &= \frac{1}{\sqrt{2}}(\ket{i}+i\ket{j})\\
    \ket{L}_{ij} &= \frac{1}{\sqrt{2}}(\ket{i}-i\ket{j}),
\end{align}
and the corresponding density matrices $\rho(\varphi_{ij})\equiv \ket{\varphi}_{ij} \bra{\varphi}_{ij}$, where $\varphi= +,-,R,L$. It can be straightforwardly shown that
\begin{equation}
    \dyad{i}{j}=\frac{\rho(+_{ij})-\rho(-_{ij}) + i \left[\rho(R_{ij})-\rho(L_{ij})\right]}{2}.
\end{equation}
By plugging this in \eqref{eq:op-basis}, any linear operators can be expanded as a (complex) linear combination of density matrices.

Second, this expansion can be used to naturally extend a quantum channel $\Lambda:\mathcal{S}(\mathcal{H}_A)\rightarrow\mathcal{S}(\mathcal{H}_B)$ to $\tilde{\Lambda}:\mathcal{L}(\mathcal{H}_A)\rightarrow\mathcal{L}(\mathcal{H}_B)$, a linear extension of $\Lambda$, as follows:
\begin{equation}
    \tilde{\Lambda}(\hat{O})\equiv \sum_{i,j} O_{ij} \frac{\Lambda(\rho(+_{ij}))- \Lambda(\rho(-_{ij})) + i \left[\Lambda(\rho(R_{ij}))-\Lambda(\rho(L_{ij}))\right]}{2}.
\end{equation}
Then, $\Tilde{\Lambda}$ becomes a linear map \eqref{eq:affine} with an arbitrary complex number $p$ while $\Lambda$ was only defined for $p\in [0,1]$.
Note that for $\rho\in\mathcal{S}(\mathcal{H}_A)$, $\Tilde{\Lambda}(\rho)=\Lambda(\rho)$ \cite{R300000001-I031534480-00}.

Based on the linear extension of the quantum channel, $\Tilde{\Lambda}$ is a \textbf{trace-preserving}\index{trace-preserving} map for any linear operators. This follows from the linearity of the trace operation and
\begin{equation}
    \Tr_B(\Tilde{\Lambda}(\rho))=\Tr_B({\Lambda}(\rho))=1=\Tr_A(\rho),\quad \forall \rho\in\mathcal{S}(\mathcal{H}_A).
\end{equation}
Furthermore, $\Tilde{\Lambda}$ maps a positive operator to another positive operator. Mathematically, it is called a \textbf{positive}\index{positive} map. The positivity is shown as follows. Any positive operator $\hat{A}$ can be written as $a\rho$, where $a=\Tr \hat{A}$ and $\rho=\hat{A}/a$. From the positivity, $\rho$ is a density matrix. From the affinity,
\begin{equation}
    \Tr[ \tilde{\Lambda}(\hat{A})] = a\Tr[ \Lambda(\rho)]\ge 0.
\end{equation}
In addition, we require that the positivity of $\Tilde{\Lambda}$ holds for a composite system for a consistent description of the quantum evolution. In other words, $\tilde{\Lambda}\otimes \mathbf{1}$ is a positive map for any linear operators on $\mathcal{H}_A\otimes \mathcal{H}_R$, where $\mathbf{1}$ denotes the identity map on the reference system $R$. The identity operation should not do anything, thus $\tilde{\Lambda}\otimes \mathbf{1}$ should be also a valid quantum channel that maps a density matrix to another. A map with this type of positivity is called a \textbf{completely positive (CP)}\index{completely positive}\index{CP|see completely positive } map.\footnote{Note that even if a map is positive, it is not necessarily completely positive. One such example is partial transpose, which is used to define entanglement negativity \cite{nielsen_chuang_2010,BB09394904}.}
Overall, a general quantum channel is a \textbf{completely positive, trace-preserving (CPTP) map}\index{completely positive trace-preserving}\index{CPTP|see completely positive trace-preserving }.\footnote{It is sometimes called a trace-preserving, completely positive (TPCP) map.}

\subsection{Kraus representation}
One of the useful representations of a CPTP map is known as the \textbf{Kraus representation}\index{Kraus representation}. $\tilde{\Lambda}:\mathcal{L}(\mathcal{H}_A)\rightarrow\mathcal{L}(\mathcal{H}_B)$ is CP if and only if there exists a \textbf{Kraus operator}\index{Kraus operator} $V_k:\mathcal{H}_A\rightarrow\mathcal{H}_B$ such that $\Tilde{\Lambda}$ has a \textbf{Kraus representation},
\begin{equation}
    \tilde{\Lambda} (\hat{O}) = \sum_k V_k \hat{O} V_k^\dagger, \quad \forall\hat{O}\in \mathcal{L}(\mathcal{H}_A).
\end{equation}
In addition to this, trace preservation requires 
\begin{equation}
    \sum_k V_k^\dagger V_k=\mathbf{1}_A.
\end{equation}
One can confirm a Kraus representation of $\tilde{\Lambda}$ is indeed CP by diagonalizing $V_k$. See \cite{BB09394904,R300000001-I031534480-00} for proof that any CPTP maps have a Kraus representation.

In the context of measurement processes, $V_k$ is a \textbf{measurement operator}\index{measurement operator}. For instance, the measurement operator for the projective measurement \eqref{eq:proj-meas} is $\hat{P}_m$.

\subsection{Stinespring representation}
At first sight, it is not so obvious whether the CPTP map or equivalently, a Kraus representation really describes a generic quantum evolution. In this section, we see this is indeed true by rewriting a Kraus representation as a unitary evolution of a composite system.

A Kraus representation is equivalent to a so-called \textbf{Stinespring representation} \cite{BB09394904,R300000001-I031534480-00},
\begin{equation}
    \Lambda(\rho)=\Tr_{AR} \left[U (\rho\otimes\sigma\otimes \rho_R) U^\dagger \right], 
\end{equation}
where $R$ is a reference system. $\sigma$ is an initial state of $B$. $\rho_R$ is an initial state of $R$, which can be taken to be pure. $U$ is a unitary evolution on the joint system $ABR$. The Stinespring representation indicates that any CPTP maps are interpreted as a unitary evolution of some joint system and vice versa. In the tensor network representation, it can be easily seen that Kraus and Stinespring representations correspond to each other \cite{Bridgeman:2016dhh}. Let us choose the initial state for $R$ a pure state $\rho_R=\dyad{0}$. Then,
\begin{equation}
\begin{aligned}
    \Lambda(\rho)&=\Tr_{AR} \left[U (\rho\otimes\sigma\otimes \rho_R) U^\dagger \right] \\[1.1ex]
    &= \includegraphics[height=3.5cm, valign=c]{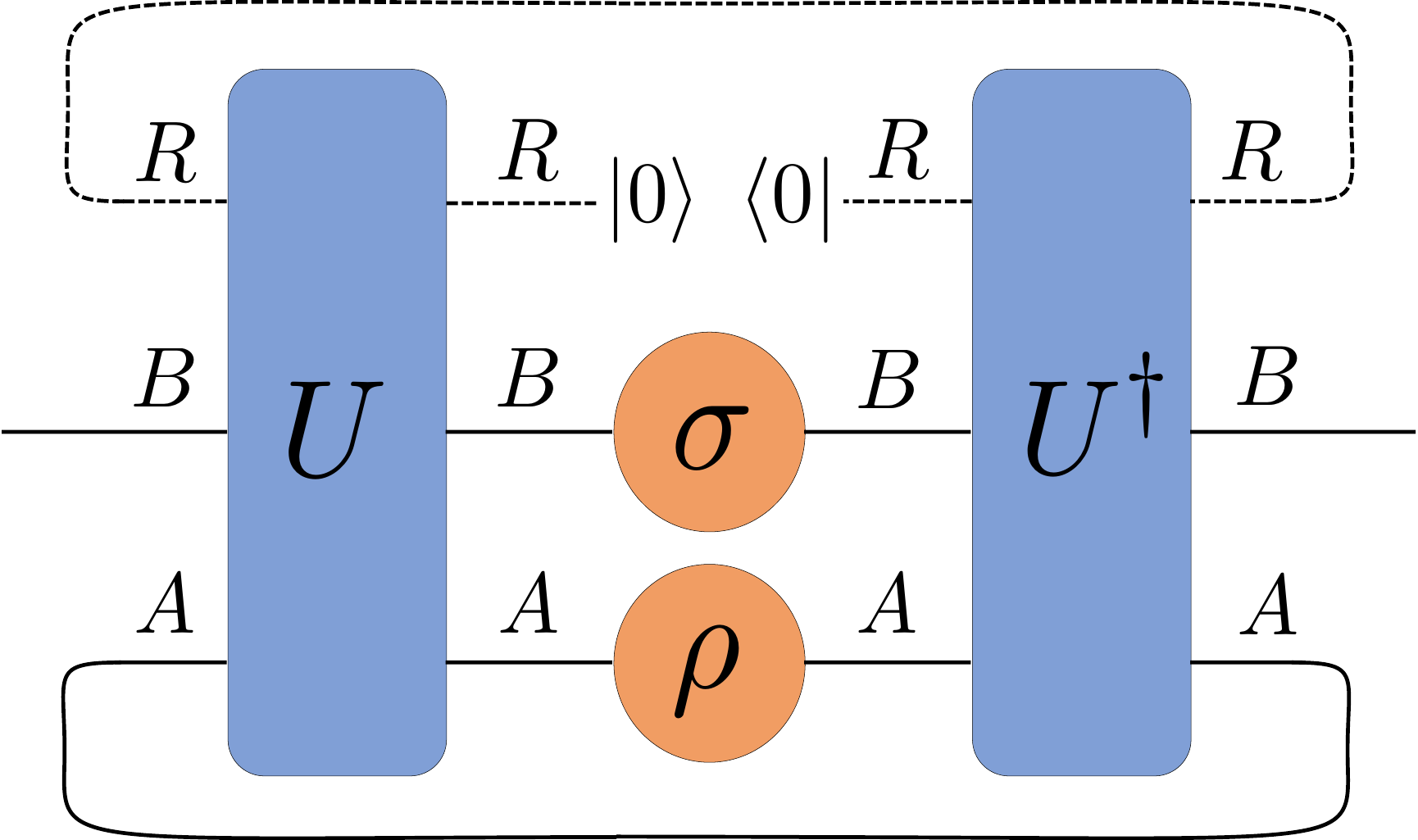} \\[1.1ex]
    &= \includegraphics[height=3.3cm, valign=c]{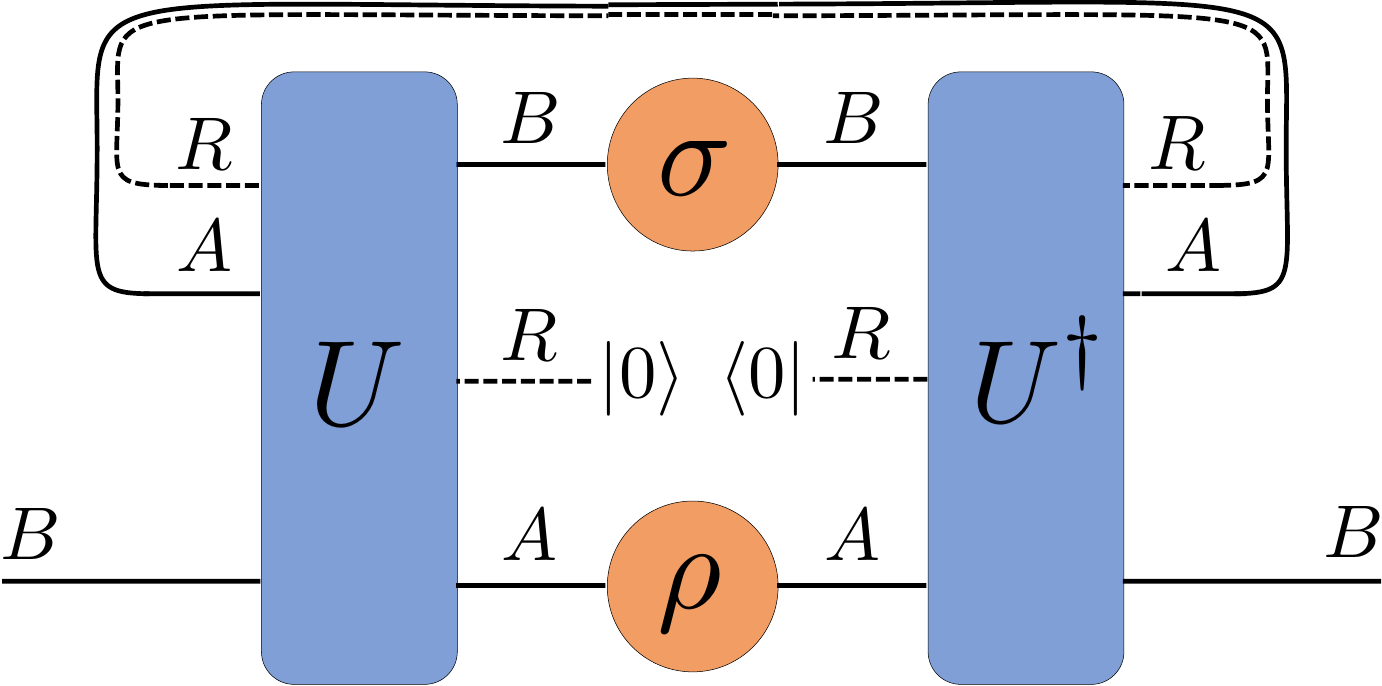} \\[1.1ex]
    &= \includegraphics[height=3.3cm, valign=c]{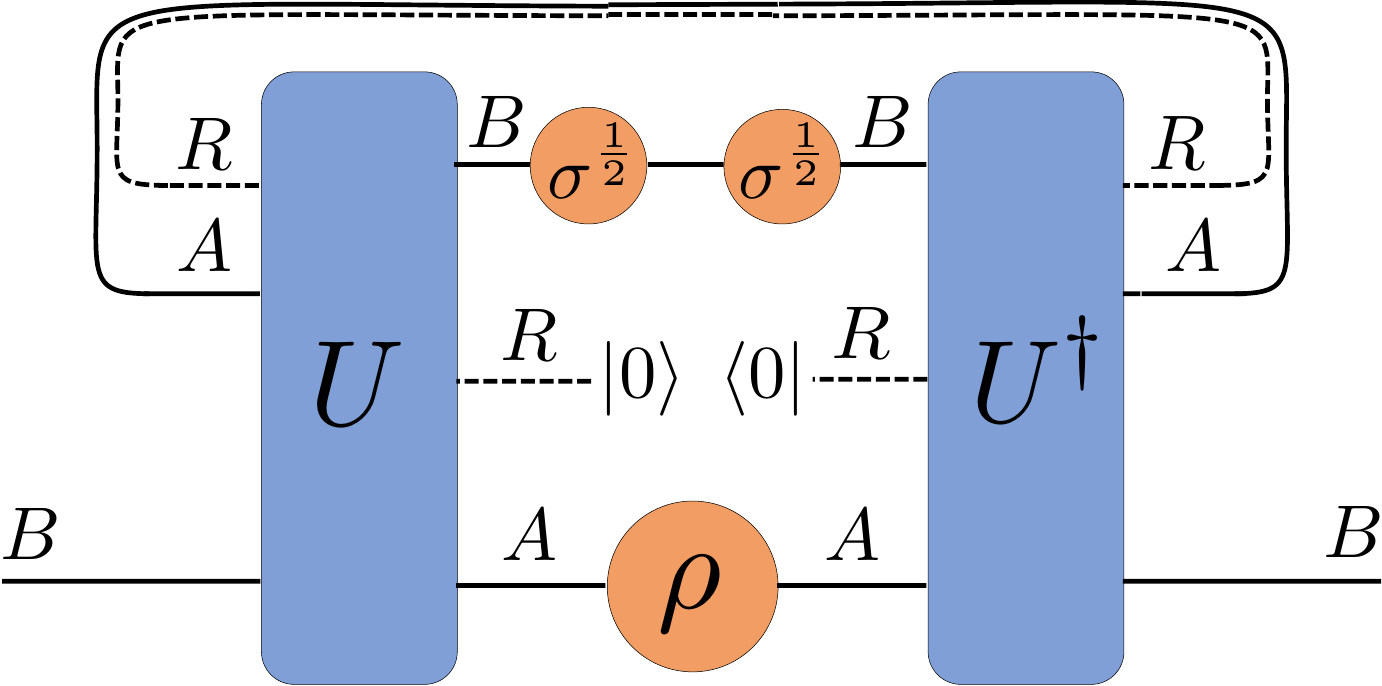} \\[1.1ex]
    &\equiv \sum_k V_k \rho V_k^\dagger.
    %= \includegraphics[height=3.3cm, valign=c]{stine4.pdf}.
\end{aligned}
\end{equation}
We changed the order of legs in the third line and we split the density matrix in half in the fourth line. From this tensor network representation, we can straightforwardly obtain the corresponding Kraus operator
\begin{equation}
    V_k = \includegraphics[height=4cm, valign=c,clip]{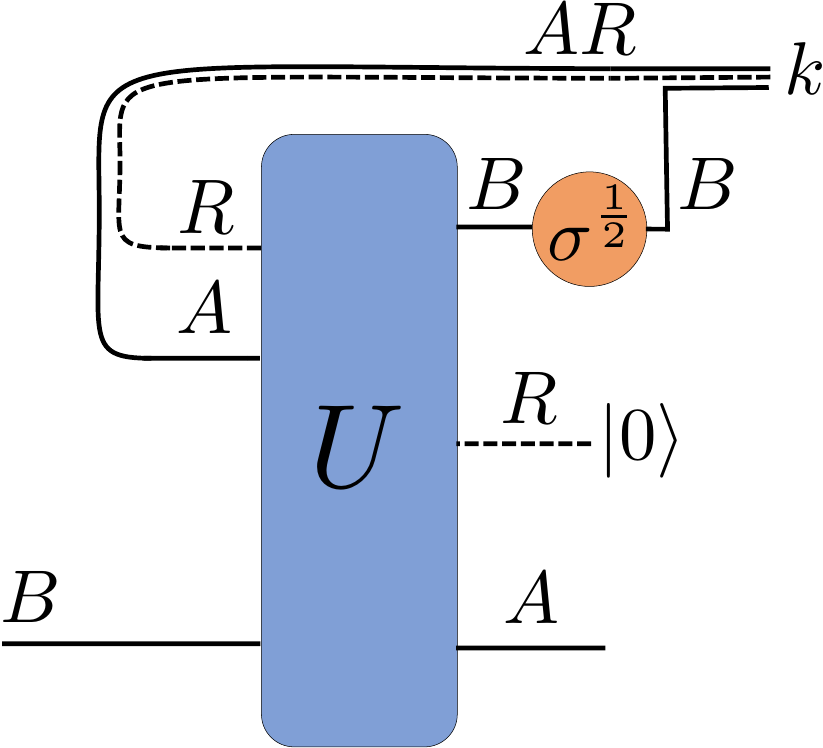}.
\end{equation}
The trace-preserving property is also evident since
\begin{equation}
\begin{aligned}
    \sum V_k^\dagger V_k &= \includegraphics[height=3.8cm, valign=c,clip]{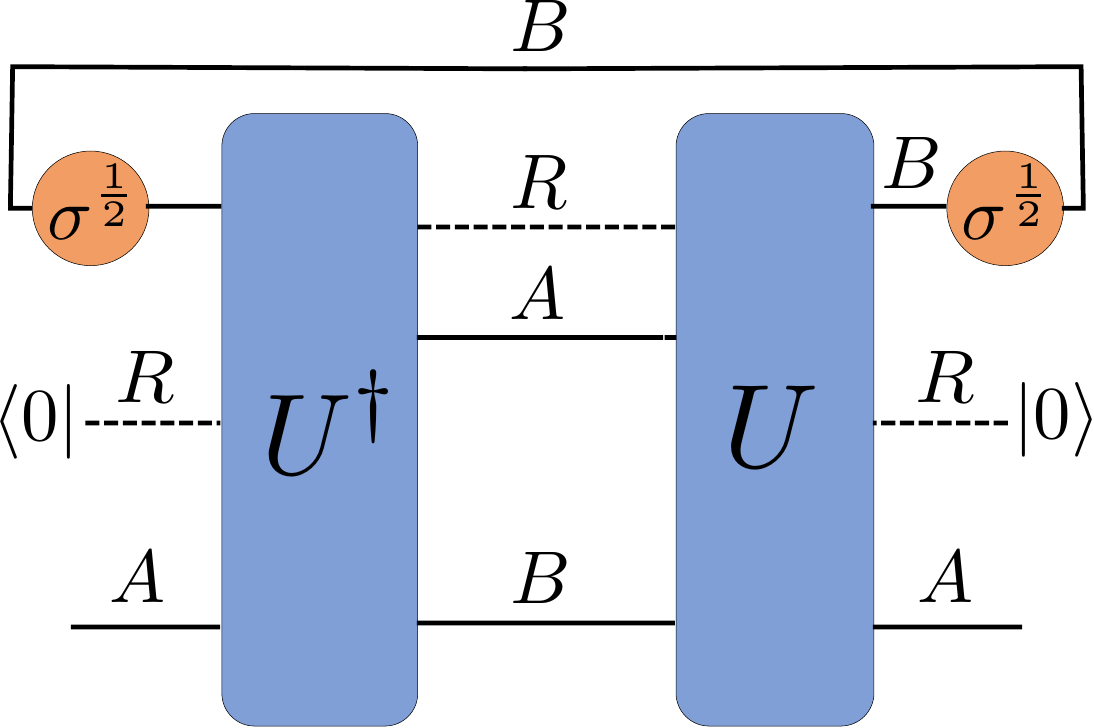}\\[1.5ex]
    &= \includegraphics[height=3.8cm, valign=c,clip]{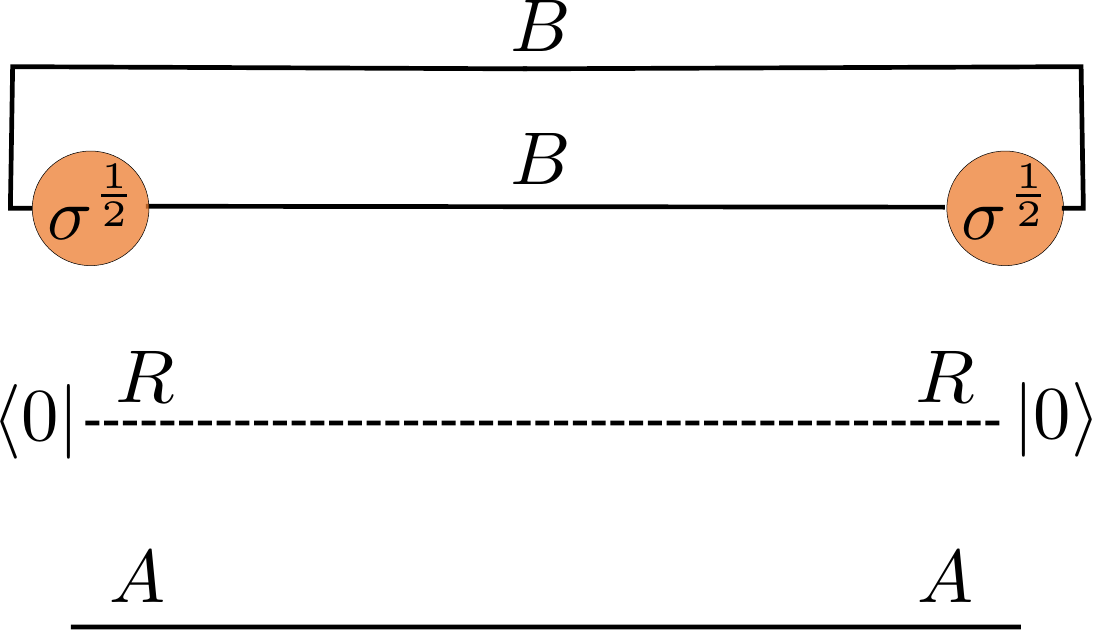}\\[1.5ex]
    &= \Tr\sigma \braket{0} \mathbf{1}_A= \mathbf{1}_A.
\end{aligned}
\end{equation}

\section{CP-instrument}\label{app:CP-inst}
We classified various measurement processes in Appendix \ref{app:meas}. We also briefly mentioned states after selective or non-selective measurements. To describe a state after an indirect measurement, or more generally, all possible measurements including time evolutions by CPTP maps, it is useful to introduce \textbf{CP-instruments} instead of a description by measurement operators.
For this purpose, it is essential to describe the measurement process as a CPTP map.

Suppose we have a CPTP map $\Lambda$ corresponding to a non-selective measurement \eqref{eq:non-select} whose outcomes are labeled by $\{m\}$. This indicates
\begin{equation}
    \Lambda=\sum_m \Lambda_m,
\end{equation}
where
\begin{equation}
    \Lambda_m(\rho)= \mathrm{Prob}(m|\rho) \rho_m.
\end{equation}
Note that even though $\Lambda$ is CPTP, $\Lambda_m$ is CP but not trace-preserving since $\Lambda_m(\rho)$ is an unnormalized density matrix. This type of maps is useful since it possesses affinity \eqref{eq:affine} (or linearity in its linear extension). Formally speaking, any measurement processes are described by a \textbf{CP-instrument}\index{CP-instrument}. It is a set of linear CP maps $\{\Lambda_m\}_m$ such that $\sum_m \Lambda_m$ is a CPTP map.

Using this, a non-selective measurement is specified by $\Lambda$ while a selective measurement is specified by $\Lambda_m$.

Since $\Lambda$ has a Kraus representation, an element of the CP-instrument $\Lambda_m$ also has the following decomposition
\begin{equation}
    \Lambda_m (\rho) = \sum_k V_k^{(m)} \rho V_k^{(m)\, \dagger}
\end{equation}
with $V_k^{(m)}:\mathcal{H}_A\rightarrow \mathcal{H}_B$
From the trace preservation,
\begin{equation}
    \mathbf{1}_A=\sum_m \sum_k V_k^{(m)\, \dagger} V_k^{(m)}. 
\end{equation}

%In this section, we explain so-called \textbf{CP-instruments}, where we can deal with all possible measurement processes with an arbitrary CPTP evolution in a coherent manner.

\chapter{Appendices for Chapter II}
\ifpdf
    \graphicspath{{Chapter1/Figs/Raster/}{Chapter1/Figs/PDF/}{Chapter1/Figs/}}
\else
    \graphicspath{{Chapter1/Figs/Vector/}{Chapter1/Figs/}}
\fi
\graphicspath{{./Chapter1/Figs/}}
\section{Entanglement entropy as thermal entropy}\label{app:thermal-EE}
When the subregion is a half space, the domain of dependence of the subregion is nothing but the \textbf{Rindler wedge}\index{Rindler wedge}. This choice of the subregion makes the computation of EE manageable as we can identify the replica parameter $n$ in the replica trick with the Rindler temperature $\beta$, which is a period of the imaginary time around the origin \cite{Nishioka_2007}. This is why EE of a half space is also known as the Rindler entropy\index{Rindler entropy}~\cite{Witten:2018xfj}. When the original Rindler wedge has the inverse temperature $\beta=2\pi$, its $n$-fold cover has the inverse temperature $n\beta$. The thermal entropy\index{thermal entropy} is given by
\begin{align}
S&=-\beta F+\beta E\nonumber \\
&=\log Z_1(\beta) +\beta \pdv{\beta} (\beta F)\nonumber\\
&=\log Z_1 (\beta)-\beta \pdv{\beta}\log Z_1(\beta) \nonumber\\
&=\left.\left[\log Z_n (\beta)-\beta \pdv{\beta}(n\beta)\pdv{(n\beta)}\log Z_n (\beta)\right]\right\vert_{n\rightarrow 1}.\nonumber
\intertext{Since $\beta$ is a constant,}
S&=\log Z_1 (\beta) -\left.\pdv{n} \log Z_n(\beta)\right\vert_{n\rightarrow 1}.
\end{align}
The last expression exactly equals that from the replica trick \eqref{e:EE_n}.

\section{The Euclidean partition function for various spin fields}\label{app:spin}
In this appendix, we calculate the 1-loop Euclidean partition function\index{1-loop Euclidean partition function} for scalar, vector, and fermionic fields, following~\cite{Kabat:1995eq}. The Euclidean partition function is given by \eqref{eq:part-fn-general}. We drop the saddle-point contribution in the following discussion. 

For a real scalar field with mass $m_0$,
\begin{equation}
    \hat{\mathcal{D}}=-\Box+m_0^2
\end{equation}
and 
\begin{equation}
    Z={\det} ^{-1/2} (-\Box+m_0^2).
\end{equation}

For the Abelian gauge field in the covariant gauge with an IR regulator $m_0$, the partition function is a product of two contributions. One comes from the spin-$1$ gauge field
\begin{equation}
    {\det}^{-1/2} \left[
    g^{IJ}(-\Box+m_0^2)-R^{IJ}+\left(1-\frac{1}{\xi}\right)\nabla^I\nabla^J
    \right] \quad (I,J=1,\cdots, d+1),
\end{equation}
where $R^{IJ}$ is the Ricci curvature tensor and $g^{IJ}$ is the metric tensor.
The other comes from the Faddeev-Popov ghost %given by
\begin{equation}
    {\det} \left[-\Box+m_0^2
    \right].
\end{equation}
When there is no curvature contribution, the partition function is just that of $(d-1)$ ($(d+1)$ without ghosts) scalar fields, corresponding to the $(d-1)$ ($(d+1)$ without ghosts) physical polarizations of the gauge field. When the curvature contributes, the free energy differs by the nonzero spin contribution.

For a fermionic field with mass $m_0$,
\begin{equation}
    Z=\int \mathcal{D}\psi\mathcal{D}\Bar{\psi} e^{-\int \Bar{\psi}(i\slashed{\partial}-m_0)\psi}=\det(i\slashed{\partial}-m_0).
\end{equation}
Since $(i\slashed{\partial}-m_0)$ is Hermitian,
\begin{align}
    Z&= {\det}^{1/2}\left[(-i\slashed{\partial}-m_0)(i\slashed{\partial}-m_0)\right] \nonumber\\
    &= {\det}^{1/2}\left[\frac{1}{2}\acomm{\gamma^I}{\gamma^J}\partial_I\partial_J +m_0^2\right].
\end{align}
Since the gamma matrices satisfy the following Clifford algebra in the Euclidean signature
\begin{equation}
    \acomm{\gamma^I}{\gamma^J}=-2\delta^{IJ},
\end{equation}
the partition function equals
\begin{equation}
    Z={\det}^{1/2}(-\Box+m_0^2).
\end{equation}

\nomenclature[a-dim]{$g_{\mu\nu}$}{a metric tensor}
\nomenclature[a-dim]{$R_{\mu\nu}$}{a Ricci tensor}
\nomenclature[g-x]{$\xi$}{the covariant gauge}                                             
% first letter G is for Greek Symbols
\nomenclature[g-g]{$\gamma^I$}{the gamma matrix}                                             
% first letter G is for Greek Symbols
\nomenclature[x-l]{$\Box$}{the Laplace-Beltrami differential operator} % first letter X is for Other Symbols
\nomenclature[a-dim]{$\hat{\mathcal{D}}$}{the quadratic coupling in the Euclidean action} % first letter X is for Other Symbols
\nomenclature[a-dim]{$G$}{a propagator}
\nomenclature[x-n]{$\nabla$}{the covariant derivative} % first letter X is for Other Symbols

\section{Spinor representation of $SO(2)\subset SO(d+1)$ and the number of degrees of freedom for each spin}\label{app:spinor}
In this appendix, we present how the Dirac fermion transforms under the $SO(2)\subset SO(d+1)$ rotation around the $x_\parallel$ axis. After the Wick rotation, we denote the gamma matrices for $x_\perp$, $x_\parallel$, $\tau$ by $\gamma_1$, $\gamma_2$ to $\gamma_d$, $\gamma_{d+1}$, respectively.
%on the two-plane perpendicular to $x_\parallel$ by $\gamma_1$ and $\gamma_{d+1}$. 

%Let us consider the $SO(2)\subset SO(3)$ case. We call $x_\perp, \tau, x_\parallel$ the x,y,z-direction, respectively. The gamma matrices are defined as $\gamma_I=i\sigma_I$ ($I=x,y,z$ in this case). The $\theta$-rotation matrix around the z-axis is given by
%\begin{equation}
%    e^{-\frac{i}{2}\theta \sigma_z}= e^{\frac{i}{2}\theta \gamma_x \gamma_y}=e^{i\theta \frac{\comm{\gamma_x}{\gamma_y}}{4}}.
%    \label{eq:spinor-trf}
%\end{equation}
%The last equality follows from $\acomm{\gamma_x}{\gamma_y}=0$.
%\eqref{eq:spinor-trf} is indeed the spinor representation of $SO(3)$.
The $SO(d+1)$ rotation is driven by
\begin{equation}
    \exp\left[{-\frac{i}{2}\theta_{IJ}S^{IJ}}\right],\quad S^{IJ}=\frac{i}{4}\comm{\gamma^I}{\gamma^J}=\frac{i}{2}\gamma^I\gamma^J.
\end{equation}
Then, for the $SO(2)\subset SO(d+1)$ rotation by $\theta$, it becomes
\begin{equation}
    \exp\left[{\frac{1}{2}\theta \gamma_1 \gamma_{d+1}}\right].
\end{equation}
By taking $\gamma_1=i\sigma_x$ and $\gamma_{d+1}=i\sigma_y$,\footnote{This is valid as long as we consider the 2-rotation.}
we obtain the familiar $SU(2)$ element
\begin{equation}
    \exp\left[{-\frac{i}{2}\theta \sigma_z}\right].
\end{equation}
%The eigenvalues of $\gamma_1 \gamma_{d+1}/2$ are $s=\pm 1/2$.

Note that when $\theta=2\pi n/M$,
\begin{equation}
    \exp\left[{\frac{\pi n}{M} \gamma_1 \gamma_{d+1}}\right].
    \label{eq:fermion-factor}
\end{equation}

In the calculation of EE, we need to count the number of components $N_{dof}$ for each spin after the irreducible decomposition~\cite{He:2014gva}.
$N_{dof}$ for fermions is the number of Majorana fermions. For real scalars, $N_{dof}$ is just the number of fields. For a $(d+1)$-dimensional U(1) gauge field without ghosts, $A_I$ transforms as a vector for $I=1,d+1$ and scalar for $I=2,\cdots,d$. Thus, $N_{dof}=d+1$ while spins are $s_1=1,s_{d+1}=-1$ and otherwise zero. For the spin-2 tensor $g_{IJ}$, it involves the rank-2 tensor $g_{ab}$, vectors $g_{ai}$, and scalars $g_{ij}$ under $SO(2)$, where $a=1,d+1$ and $i=1,\cdots,d$. The rank-2 tensor is further decomposed into the sum of a traceless symmetric part ($s=\pm 2$) and a scalar proportional to the identity ($s=0$). The vectors $g_{ai}$ have spins $s=\pm 1$ and $N_{dof}=d-1$ for each $s$. The scalars $g_{ij}$ have $N_{dof}=(d-1)d/2$. In total, the number of components for each spin is
\begin{equation}
    \abs{s}=2 : 2,\quad \abs{s}=1 : 2(d-1), \quad s=0: \frac{d^2-d+1}{2}.
\end{equation}

\section{Area laws for R\'{e}nyi entropy and entanglement capacity}
\label{appenarea}
In this appendix, we show that the area law\footnote{
	In general, the \textbf{area law}\index{area law} means $O(|\partial A|)$; the quantity scales at most as the area~\cite{Eisert:2008ur}.
} $S_{EE}\propto \mathrm{vol}(\partial A)$ ($\sim V_{d-1}$ in our setup) holds at the level of \textbf{R\'{e}nyi entropy}\index{R\'{e}nyi entropy} $S_n\equiv\frac{1}{1-n}\log \Tr \rho_A^n$. 
%We can apply the same discussion performed in our previous papers~\cite{Iso:2021vrk,Iso:2021rop} deriving the area law for EE.
To apply the orbifold method introduced in {Section \ref{s:ZM}}, we rewrite $S_{n=1/M}$ in terms of free energy $F^{(M)}$ on $\mathbb{R}^2/\mathbb{Z}_M\times \mathbb{R}^{d-1}$:
\begin{equation}
	S_{1/M}=\frac{1}{M-1}\left(F^{(1)}-MF^{(M)}\right).
	\label{eq:renyi}
\end{equation}
Based on the $\mathbb{Z}_M$ gauge theory on Feynman diagrams introduced in {Section \ref{s:ZM-gauge}}, the number of independent twists is given by the number of loops $L$ in the Feynman diagrams even though every propagator is originally twisted. Let us denote the number of initial twists or equivalently the number of propagators by $P$. Twists other than the independent ones can be eliminated by the redundancy at vertices. The number of such redundant twists is given by $P-L=V-1$, where the number of vertices is denoted by $V$. As a result, the $1/M$ factor from each vertex is almost canceled except one by the trivial summation for $(V-1)$ redundant twists. Furthermore, the overall momentum conservation yields the sum of $L$ twisted momenta equal to the original one. In short, any Feynman diagrams contributing to $F^{(M)}$ are expressed as
\begin{align}
	\frac{V_{d-1}}{M}\sum_{\{m\}}\int\prod_{l=1}^L\left[\frac{d^2\bm{p}_l}{(2\pi)^2}\right]I(\{\bm{p}\};\{m\})
	\delta^2 \left(\sum_{l=1}^L (1-\hat{g}^{m_l}) \bm{p}_l \right),
	\label{eq:area-law}
\end{align}
where  
$\sum_{\{m\}}$ is a summation over all twists; each from $0$ to $M-1$. $I(\{\bm{p}\};\{m\})$ is some function of momenta and twists.

When all $m$'s are zero, no momenta are twisted. Such diagrams constitute nothing but $F^{(1)}/M$. Although this contribution in $F^{(M)}$ is proportional to $V_{d+1}$ and seemingly violates the area law, it is canceled in $S_{1/M}$ in 
\eqref{eq:renyi}. Other configurations of twists include at least one nonzero twist. As a result, the argument of the delta function in \eqref{eq:area-law} is always nonzero and it combined with $I$ carries a nontrivial dependence in $M$ after the summation over twists. Unless an explicit calculation is done, we do not know the precise $M$ dependence of \eqref{eq:area-law}
or $S_{1/M}$. Nevertheless, since terms contributing to $S_{1/M}$ always have nonzero arguments of the delta function, there is no more volume factor other than $V_{d-1}$. 
If $M$ can be analytically continued to $M=1/n$, this completes the proof of the area law\index{area law} for R\'{e}nyi entropy $S_{n}$. 

The proof above only depends on the technique of Feynman diagrams. Thus the area law for R\'{e}nyi entropy is proven for any locally interacting QFTs, given a half space as a subregion.

It is worthwhile to note that the area law\index{area law} for R\'{e}nyi entropy immediately implies the area law for the \textbf{capacity of entanglement}\index{capacity of entanglement}~\cite{PhysRevLett.105.080501,deBoer:2018mzv},
\begin{align}
	C_A&\equiv \lim_{n\rightarrow 1} n^2 \pdv[2]{n}\log \Tr \rho_A^n\nonumber\\
	&=\left.\pdv[2]{n}\left[(1-n)S_n\right]\right\vert_{n\rightarrow 1}\nonumber\\
	&={2}\left.\pdv{S_{1/M}}{M}\right\vert_{M\rightarrow 1}
\end{align}
as well as EE since $C_A$ is linear in R\'{e}nyi entropy. Since $C_A$ is alternatively written as the fluctuation of the modular Hamiltonian $-\log\rho_A$, it is more sensitive to the change of dominant contributions in the replicated geometry and recently discussed in the context of the black hole evaporation~\cite{Kawabata:2021hac,Okuyama:2021ylc,Kawabata:2021vyo}. It is interesting if we can compute such quantities in interacting theories and follow the behavior of higher orders in $M$.

Although the area law itself is intuitive for physicists as entanglement across the boundary $\partial A$ should be dominant for any local QFTs, %its proof
{the proof of this} is difficult; a general proof is known only for gapped systems in $(1+1)$ dimensions~\cite{Hastings:2007iok}; %Since EE at the one-loop level is given in terms of the propagator, 
the area law at the one-loop level is derived through a modified propagator~\cite{Solodukhin:2011gn,Nesterov:2010yi}. It is remarkable that we can show the area law of both EE and R\'{e}nyi entropy in any locally interacting theories to all orders in the perturbation theory.

As a further generalization, it is intriguing to relax several assumptions and see how the EE and R\'{e}nyi entropy deviates from the area law. In our setup, $\partial A$ is smooth, the interactions are local, and the system is translationally invariant.  %It is known in some cases that if the entangling surface $\partial A$ has a singular geometry, the area law gets violated (see~\cite{Bueno:2019mex} for example). 
%It is known in some cases not satisfying the above features, the area law gets violated.
{Some cases are known where the above features are not satisfied and the area law is violated.} For example, when the entangling surface $\partial A$ has a singular geometry, a logarithmic correction appears (see~\cite{Bueno:2019mex} for an example). 
{For (non-)Fermi liquid theories~\cite{Ogawa:2011bz}, another logarithmic violation of the area law is known.} For nonlocal~\cite{Shiba:2013jja} or non-translationally invariant~\cite{Vitagliano:2010db,Ram_rez_2014} systems, the volume law instead of the area law of EE has been confirmed. {To see the transition from the area law to the volume law, Lifshitz theories~\cite{He:2017wla,MohammadiMozaffar:2017nri,Gentle:2017ywk} might be an interesting playground as it possesses nonlocal features in some limit.}

\section{Interpretation of a twisted propagator in position space}
\label{app:twist}
In Section \ref{sec:int-orb}, we discussed a twisted propagator\index{twisted propagator} in the position and momentum spaces. The physical interpretation of the twisted propagator becomes clearer in the position space compared to the momentum space.
We demonstrate below that a twisted propagator is  pinned at the boundary. 
For this purpose, 
it is convenient to introduce the center-of-mass and relative coordinates: $X=(x+y)/2$, $r=x-y$. 
A twisted propagator with $m\neq 0$ in the position space is written as
\aln{
G_{0} (\hat{g}^m x -y)  &=  
G_{0} ( \hat{g}^{m/2} \bm{x} - \hat{g}^{-m/2} \bm{y} ; r_\parallel) % \nonumber \\ 
= G_0 (\cos \theta_m \bm{r} + 2 \sin \theta_m (\epsilon \bm{X}) ; r_\parallel) 
\nonumber \\
&= e^{\cot \theta_m \hat{R}_{\bm{X}} /2} G_0 (2 \sin \theta_m  \bm{X} ;  r_\parallel) \nonumber\\
&=\frac{e^{\cot \theta_m \hat{R}_{\bm{X}} /2} }{4 \sin^2 \theta_m}  
  \int \frac{ d^{d-1} k_\parallel}{(2\pi)^{d-1}}
	\frac{	e^{ i k_\parallel \cdot r_\parallel}}{\left(-\partial^2_{\bm{X}}/4 \sin^2 \theta_m \right) + M^2_{k_\parallel}} 
\delta^2 (\bm{X}) ,
\label{twistedpropagator}
}
where 
\aln{
[\epsilon \bm{X}]_i &= \sum_j \epsilon_{ij}X_j \qq{($\epsilon_{ij}$: {the two-dimensional Levi-Civita symbol}),} \\
\hat{R}_{\bm{X}}&=\bm{r}\cdot(\epsilon \partial_{\bm{X}}), \  \
\theta_m= \frac{m\pi}{M},  \  \
M_{k_\parallel}^2 =k_\parallel^2+m^2_{0}.
}
\eqref{twistedpropagator} can be written in  a derivative expansion  on the delta function
with respect to $\partial^2_{\bm{X}}/M^2_{k_\parallel}$. 
When we consider a diagram with a single twist $m$ on a  %either of the 
propagator, it is formally written as
\aln{
\int d^{d+1}x\, d^{d+1}y\, G_0(\hat{g}^mx-y)F(r),\qq{where} r=x-y.
}
The integrand other than the twisted propagator only depends on $r$ due to the translational invariance.
%It is due to the translation invariance in the part other than the twisted propagator. 
With \eqref{twistedpropagator} and the partial integration, we can drop all the $\partial_{\bm{X}}$ in the expression. 
Therefore, in this case, we can replace the propagator in the diagram as
\aln{
G_0(\hat{g}^mx-y)&\rightarrow\frac{1}{4\sin^2\theta_m}\int\frac{d^{d-1}k_\parallel}{(2\pi)^{d-1}}
\frac{e^{ik_\parallel\cdot r_\parallel}}{M_{k_\parallel}^2}
\delta^2(\bm{X})\nonumber\\
&=\frac{1}{4\sin^2\theta_m}G_0^{\text{bdry}}(r_\parallel)\delta^2(\bm{X}).
\label{e:efftwistprop}
}
%Here, 
$G_0^{\text{bdry}}$ is an ordinary propagator but 
its propagation is restricted only to the directions parallel to the boundary. Note that the area law of EE $S_A\propto V_{d-1}$ can be naturally understood from the integration of the free transverse direction.

Now the physical meaning of \eqref{e:efftwistprop} is clear. Since the boundary of the subregion rests at the origin of the orbifold, %the twisted propagator has its 
the midpoint $\bm{X}$ of the twisted propagator is constrained on the boundary. Note that the propagator itself is not trapped on the boundary since the relative coordinate ${\bm r}$ is not constrained at all.   Rather, ${\bm r}$-dependence 
completely disappears from the twisted propagator.
Hence, it can be seen as a ``pinned propagator'' with  two loose endpoints on $A$ and $\bar{A}$.
This shows that the twisted propagator reflects a 
correlation between two points that are symmetrically distant from the boundary (Fig.\ref{f:pinnedprop}). 
In this sense, we can identify contributions to EE from a single twisted propagator as the
quantum correlation of two-point functions.  
%%%%%%%%%%%%%%%%%%%%%%%%%%%%%%%%%%%%%%%%
\begin{figure}[h]
\centering
\includegraphics[width=8cm]{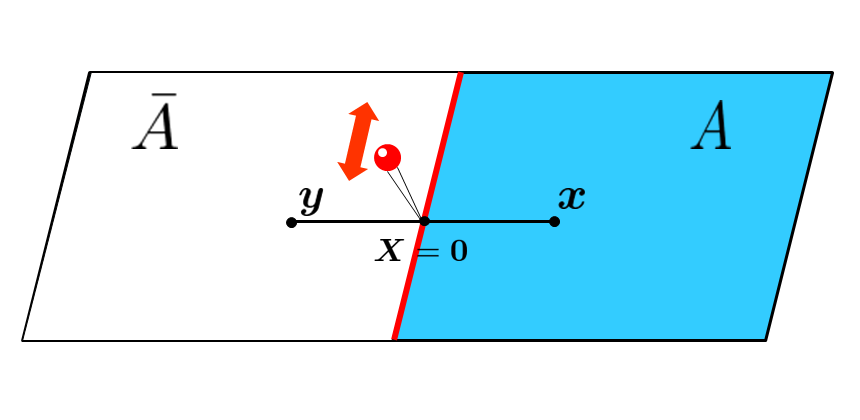}
\caption{
An illustrative picture of a propagator pinned on the boundary. Its midpoint $\bm{X}=(\bm{x}+\bm{y})/2$ is constrained on the boundary $\partial A$
while two end points move freely.
}
\label{f:pinnedprop}
\end{figure}
%%%%%%%%%%%%%%%%%%%%%%%%%%%%%%%%%%%%%%%%%%%

%%%%%%%%%%%%%%%%%%%%%%%%%%%%%%%%%%%%%%%%%%%

%%%%%%%%%%%%%%%%%%%%%%%%%%%%%%%%%%%%%%%%%%%%%
\section{Proof of the EE formulae of the vertex contributions}
\label{appencomp}
In Section \ref{s:vertexEE} and its subsequent sections, we have used the general formula for the vertex contributions\index{vertex contributions} to EE, such as \eqref{e:EEcomp}. 
In this appendix, we prove that this formula gives all the contributions of a single vertex twist. 
In the case of the propagator contributions, the general formula is given by {\eqref{e:EE2PIprop}} 
and the proof that all the single twist contributions are summarized by  the 1-loop expression is given 
in the 2PI framework as discussed in Section \ref{sec:prop-2PI}. 
For the vertex contributions, when auxiliary fields can be introduced, the proof is same, however, 
in general cases when various channels in the opened vertices are mixed, we need a different proof.
In this appendix, we give a diagrammatic proof. 

First, let us remind of the redundancies of assigning the flux $m$ of the plaquette %in the 1-loop diagram in Fig.\ref{Fig2}
to a twist of the propagators in the 1-loop diagram in {Fig.\ref{Fig2}}. In this case, due to the $\mathbb{Z}_M$ gauge invariance at each vertex 
connecting propagators, the flux can twist only one of the propagators; not more than one like \eqref{eq:1-loop-overcount}, 
and this gives the coefficient $1/n$  in the expansion of {\eqref{e:EE2PIprop}}. 
The same happens for the vertex contributions. 
The configurations illustrated in {Fig.\ref{f:verttwist}} are interpreted as the  
vertex contributions to EE, but  a similar redundancy will occur when 
the corresponding composite operators form a 1-loop type diagram. 
Thus, in order for the proof, we will take the following two steps: 
(i) summing all the vertex contributions as if all of them are independent
and  then, (ii) taking account of the redundancies to obtain the correct vertex contributions. 
This two-step proof  shows that only the 1-loop type contributions in {\eqref{EE-generalform2}} survive. 
The proof is similar to the one based on the 2PI formalism.

Let us begin with the fundamental relation between the free energy and an $n$-point vertex. 
Suppose that we have an $n$-point interaction vertex whose action is given by
\ga{
	(\text{action})=\frac{1}{2}\int d^{d+1}x\,\phi G_0^{-1}\phi+\cdots+
	\frac{\lambda_n}{n} \int \prod_{i=1}^n d^{d+1}x_i 
	\, V_{n0}(x_1,\cdots,x_n)\phi(x_1)\cdots\phi(x_n),\\
	V_{n0}(x_1,\cdots,x_n)=\int d^{d+1}y\,\prod_{i=1}^{n}\delta^{d+1}(y-x_{i}).
}
Then, we have the equation
\aln{
	\frac{\delta F}{\delta V_{n0}(x_1,\cdots,x_n)}=\frac{\lambda_n}{n}\langle\phi(x_1)\cdots\phi(x_n)\rangle ,
	\label{nptfunc} 
} 
where  $F$ is the free energy and the right-hand side is the exact $n$-point function multiplied by the coupling constant. 

In order to evaluate the EE contributions from twisting vertices, let us first sum all the contributions
as if  they were  independent.  This can be done by taking a variation 
of bubble diagrams (free energy) with respect to the tree-level interaction vertex, 
and then reconnecting the endpoints by a set of free propagators as in the leftmost figure in Fig.\ref{compredundancy}.
%%%%%%%%%%%%%%%%%%%%%%%%%%%%%%%%%%%%%%%%%%  
\begin{figure}[b]
	\centering
	%\hspace*{-1.2cm}
	\includegraphics[width=15cm]{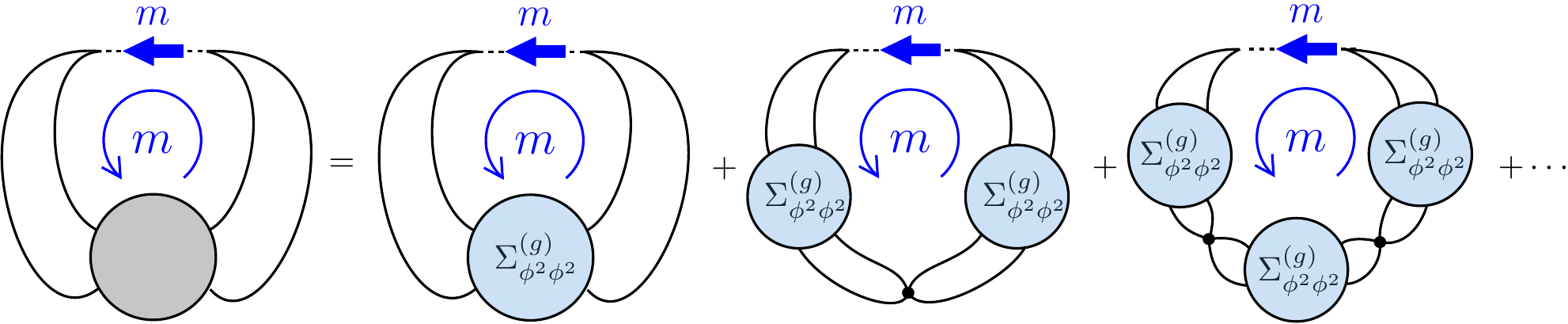}
	\caption{Graphical expression of \eqref{EEvertnaivek} and \eqref{compexpand} for $k=2$  in the $\phi^4$ theory. 
		%An example for twist configurations which would correspond to \eqref{compexpand} with $[\phi^2]$ in the $\phi^4$ theory. 
		The gray blob on the left-hand side is the exact four-point correlation function. The dotted line denotes a
		twisted delta function to open the vertex. The twist of the diagram is given by the flux $m$ in the center circle. 
		%essentially stems from the center circle. 
		On the right-hand side, the twist is made associated with the propagator of a composite operator in the opened vertex. 
		%be associated alternatively to opening either of the nodal vertices between the g-1PI parts. 
		If we {would} open all the vertices on the circled line and take all the contributions to EE, 
		it would give an overcounting of EE. 
		%We have to take just one of them, otherwise we confront the overcounting.   
	}
	\label{compredundancy}
\end{figure}
%%%%%%%%%%%%%%%%%%%%%%%%%%%%%%%%%%%%%%%%%%
%The choice of the organizing is equivalent to those of opening a vertex. For example, $\phi^6\to(\phi^2,\phi^4),\,(\phi^3,\phi^3)$ and so on. One obtains the vertex corrections by inserting the summation of twisted delta functions into the opened vertex, and taking the $M$-derivative. 

According to \eqref{nptfunc}, 
if we naively summed all the contributions to EE from 
opening all the $\lambda_n$-vertices, EE would be given by
\aln
{
	S_{\lambda_n}^{\text{(naive)}}=&-\int d^{d+1}x\,d^{d+1}y\,\partial_M\left(\sum_{m=1}^{M-1}\delta^{d+1}(\hat{g}^m x-y)\right)\Bigg|_{M\to1}\,\nonumber\\
	&\times\frac{\lambda_n}{n}\frac{1}{2}\Bigl\{C^{(n)}_{2,n-2}\,\langle\,[\phi^2](x)\,[\phi^{n-2}](y)\,\rangle+C^{(n)}_{3,n-3}\,\langle\,[\phi^3](x)\,[\phi^{n-3}](y)\,\rangle+\cdots\Bigr\}\nonumber\\
	=&-\frac{V_{d-1}}{12}\int\frac{d^{d-1}k_\parallel}{(2\pi)^{d-1}}\frac{\lambda_n}{n}\Bigl\{C^{(n)}_{2,n-2}\,G_{\phi^2\phi^{n-2}}(\bm{0};k_\parallel)+C^{(n)}_{3,n-3}\,G_{\phi^3\phi^{n-3}}(\bm{0};k_\parallel)+\cdots\Bigr\},
	\label{EEvertnaive}
}
where 
$C^{(n)}_{2,n-2}$, $C^{(n)}_{3,n-3}$, $\cdots$ are combinatorial factors 
to reconnect the endpoints. Endpoints can be decomposed into two sets as in Fig.\ref{compredundancy}
and then regarded as a composite operator. 
% for the way to form the  composite operators by organizing the endpoints of the $n$-point function. 
In the definition of $C^{(n)}_{2,n-2}$ etc., 
we distinguish the left and right sets of  endpoints, $x$ and $y$, for later convenience, 
and thus divide by 2 in the second line of \eqref{EEvertnaive}
to avoid an overcounting. 

Here note that
we should not include a decomposition of  $n$ endpoints $\phi^n$ into $(\phi,\phi^{n-1})$ %$\phi^n\to(\phi,\phi^{n-1})$
because it does not correspond to opening a vertex, rather it generates a non-1PI diagram in the 
ordinary sense. %, and 
Such contributions would  lead to an overcounting of the propagator contributions. 
Also, note  that we should not consider a reconnection of $n$ endpoints 
in which some of them do not participate in the propagation of the composite operator;
e.g. a diagram such that a pair of endpoints forms a closed loop and the other $n-2$ endpoints are
decomposed into two sets to form the propagator of the composite operator.  
This kind of diagram 
is  absorbed into the renormalizations of the coupling constant $\lambda_{n-2}$. 

Now let us go to step 2 to obtain the correct vertex contributions. 
$S_{\lambda_n}^{(\text{naive})}$ is not the correct one because of the redundancies we neglected.  
Let us consider the effects of redundancies  separately for each type of composite operator  in \eqref{EEvertnaive}. 

For simplicity, let us consider the $G_{\phi^k\phi^k}$-type contribution in \eqref{EEvertnaive}
which  emerges from a $2k$-point vertex by decomposing into 2 sets of $k$ and $k$. 
We simply  assume it is not mixed with other operators here. 
The simplest example  is $[\phi^2]$ in the $\phi^4$ theory, as described in Section \ref{s:vertexEE}. 
If the redundancies were neglected, 
the contribution to EE from this operator naively would take the form
\aln{
	\left.S_{\lambda_{2k}}^{\text{(naive)}}\right|_{\phi^k\phi^k}
	=-\frac{V_{d-1}}{12}\int\frac{d^{d-1}k_\parallel}{(2\pi)^{d-1}}\,\frac{\lambda_{2k}}{2k}\,C^{(2k)}_{k,k}\,G_{\phi^k\phi^k}(\bm{0};k_\parallel). 
	\label{EEvertnaivek}
}
Reflecting the twisting of the composite operator, 
the Green function is restricted to $\bm{k}=\bm{0}$ modes. 
The Green function can be expanded with respect to the g-1PI\index{{g-1PI}} self-energy {introduced in Section \ref{s:vertexEE}} as\footnote{
	For a consistent expansion, we have defined the combinatorics factors ($C$'s) by distinguishing the two endpoints.}
\aln{
	G_{\phi^k\phi^k}=\Sigma^{(g)}_{\phi^k\phi^k}&+\Sigma^{(g)}_{\phi^k\phi^k}\left(-\frac{\lambda_{2k}}{2k}C^{(2k)}_{k,k}\right)\Sigma^{(g)}_{\phi^k\phi^k}\nonumber\\
	&+\Sigma^{(g)}_{\phi^k\phi^k}\left(-\frac{\lambda_{2k}}{2k}C^{(2k)}_{k,k}\right)\Sigma^{(g)}_{\phi^k\phi^k}\left(-\frac{\lambda_{2k}}{2k}C^{(2k)}_{k,k}\right)\Sigma^{(g)}_{\phi^k\phi^k}+\cdots.
	\label{compexpand}
}
The correct formula must take the redundancies caused by $\mathbb{Z}_M$ gauge invariance
into account. Such redundancies occur in \eqref{compexpand} when there are 
more than one $\Sigma^{(g)}_{\phi^k\phi^k}$ as shown in  Fig.\ref{compredundancy}. 
The coefficients of these terms in \eqref{compexpand} overcount the effects of the twist. 
%The summation of them is therefore overcounting. In fact, we have considered opening all the vertices by considering the variation (\ref{nptfunc}). It means that \eqref{EEvertnaivek} has overcounted contributions in non g-1PI diagrams. In particular, the term consisting of $m$ g-1PI parts has acquired the extra factor of $m$. 

The resolution to avoid overcounting   is simple. 
For the term consisting of $m$ g-1PI parts in \eqref{EEvertnaive}, we should divide it by $m$ (the divide-by-multiplicity method). 
Consequently, 
by replacing $G_{\phi^k\phi^k}$ in the naive estimation \eqref{EEvertnaivek} with
\aln{
	\Sigma^{(g)}_{\phi^k\phi^k}&+\frac{1}{2}\,\Sigma^{(g)}_{\phi^k\phi^k}\left(-\frac{\lambda_{2k}}{2k}C^{(2k)}_{k,k}\right)\Sigma^{(g)}_{\phi^k\phi^k}\nonumber\\
	&+\frac{1}{3}\,\Sigma^{(g)}_{\phi^k\phi^k}\left(-\frac{\lambda_{2k}}{2k}C^{(2k)}_{k,k}\right)\Sigma^{(g)}_{\phi^k\phi^k}\left(-\frac{\lambda_{2k}}{2k}C^{(2k)}_{k,k}\right)\Sigma^{(g)}_{\phi^k\phi^k}+\cdots,
	\label{complogexpand}
}
we get the correct contributions to EE as 
\aln{
	\left.S_{\lambda_{2k}}\right|_{\phi^k\phi^k}
	=-\frac{V_{d-1}}{12}\int\frac{d^{d-1}k_\parallel}{(2\pi)^{d-1}}\,\ln (1-\left(-\frac{\lambda_{2k}}{2k}C^{(2k)}_{k,k}\right)\Sigma^{(g)}_{\phi^k\phi^k}). 
}
This leads {\eqref{e:EEcomp2}}. 
In the 2PI formalism, the result is interpreted that only the 1-loop diagram provides single 
twist contributions of propagators and all the other diagrams {cancel} each other. 
In the above discussions, we did not separate diagrams into 1-loop and others, but 
instead used the very basic relation of {\eqref{nptfunc}}. 
Then, using the property of the $\mathbb{Z}_M$ redundancy, the logarithmic factor for the 1-loop diagram
naturally appears. 

The above discussion can be straightforwardly generalized to more general composite operators with operator mixings. 
When we have a set of operators $\{\mathcal{O}_a\}$ 
by opening vertices, we consider   g-1PI self-energies $\Sigma^{(g)}_{\mathcal{O}_a\mathcal{O}_b}$ 
and a matrix generalization of the nodal structure of   $(\lambda_n/n)\times C^{(n)}_{ab}$. 
It is also straightforward  when the fundamental fields are mixed with other operators; 
it is sufficient to consider  $\hat{\Sigma}^{(g)}\hat{G}_0$ in the formulation. 
As a result, we arrive at the unified form of \eqref{EE-generalform2}.
\chapter{Appendices for Chapter III\&IV}
\ifpdf
    \graphicspath{{Chapter2/Figs/Raster/}{Chapter2/Figs/PDF/}{Chapter2/Figs/}}
\else
    \graphicspath{{Chapter2/Figs/Vector/}{Chapter2/Figs/}}
\fi

\section{Worldline formalism of QFT}\label{app:worldline}
The \textbf{worldline formalism}\index{worldline formalism}~\cite{Corradini:2015tik} relates a usual second quantized formalism of QFT to the first quantized formalism, i.e. a relativistic point particle. This is achieved by the Schwinger proper time parametrization of the propagator and an auxiliary field\index{auxiliary field} (=Lagrange multiplier\index{Lagrange multiplier}) imposing the primary constraint\index{primary constraint}. For simplicity, we focus on the free scalar QFT.

The transition amplitude (propagator) and the Euclidean action of the relativistic point particle are
\begin{align}
    \braket{x(\lambda)}{x(0)}&= \int \mathcal{D}x e^{-I_{part}[x]} \quad (\mathcal{D}x=dx)\nonumber\\
    I_{part}[x]&=m\int\dd{s}=m\int \dd{\lambda} \sqrt{-g_{\mu\nu} \dot{x}^\mu \dot{x}^\nu}\equiv \int \dd{\lambda} L_{part},
    \label{eq:point-part-action}
\end{align}
where $s$ is the proper distance and $\lambda$ is an arbitrary choice of parametrization of the particle trajectory $x^\mu=x^\mu(\lambda)$. The conjugate momentum is given by
\begin{equation}
    p_\mu=\frac{\delta L_{part}}{\delta\dot{x}^\mu}=\frac{m\dot{x}_\mu}{\sqrt{-\dot{x}^2}}.
\end{equation}
The primary constraint is a constraint purely from the conjugate momentum: $p^2+m^2=0$. It must be imposed with the Lagrange multiplier. Denoting it by $h$, the action with the explicit dependence on the conjugate momentum is given by
%\footnote{
%}
\begin{equation}
    I_{part}[x,p,h]=\int\dd{\lambda} \qty[-p\cdot \dot{x} + \frac{h}{2}(p^2+m^2)].
    \label{eq:point-part-action2}
\end{equation}
(The signs in the action are determined from replacing $i\lambda_L$ by $\lambda$ in the Lorentzian action
$
    \int\dd{\lambda_L} \qty[p_L\cdot \dot{x} - \frac{h}{2}(p_L^2+m^2)],
$
where $p_L$ is determined from the variation with respect to $\partial_{\lambda_L} x$. We take the negative sign of the Lagrange multiplier so that the canonical Hamiltonian has a positive sign.)
Indeed, we can obtain the original action \eqref{eq:point-part-action} after eliminating the auxiliary field $h$ and integrating the conjugate momentum $p$ out.

Our goal is to express the QFT propagator in terms of this new point particle action \eqref{eq:point-part-action2}.
In the same manner as we did for the free energy in Section \ref{sec:heat-kernel}, the QFT propagator can be written as follows:
\begin{align}
    \ev{\phi(x)\phi(x')} &= \bra{x}\int\dd{h} e^{-h G^{-1}}\ket{x'} \\
    &= \int\dd{h} e^{-hm^2} \mel{x}{e^{h\square}}{x'},
\end{align}
where we denoted the bare mass by $m$ here. We can further rewrite this using the conjugate momentum $p$:
\begin{align}
    \ev{\phi(x)\phi(x')} &= \int\dd{h} \int \frac{dp}{2\pi} e^{-ip(x-x')-h(p^2+m^2)} \\
    &= \int\dd{h} \int \frac{dp}{2\pi} e^{-\int_0^1\dd{\lambda} \qty[-p\cdot \dot{x} + h (p^2+m^2) ]}\Big|_{x(0)=x',\, x(1)=x} \\
    &= \int\limits_{\substack{x(0)=x',\\ x(1)=x}} \! \mathcal{D}x \int \mathcal{D}h \mathcal{D}p \, e^{-I_{part}[x,p,e]} \quad \qty(\mathcal{D}h=\dd{h},\, \mathcal{D}p=\frac{\dd{p}}{2\pi}).
\end{align}
This is what we want. After integrating $h$ and $p$, we can show (up to normalization)
\begin{equation}
    \ev{\phi(\mathfrak{p}_1)\phi(\mathfrak{p}_2)}=\int_{\mathfrak{p}_1}^{\mathfrak{p}_2} \mathcal{D}\mathfrak{p} e^{m\int \dd{s}[\mathfrak{p}]},
\end{equation}
where $\mathfrak{p}$ and $\mathfrak{p}'$ denote spacetime points.

\section{Geodesic approximation of Euclidean two-point function of CFT}\label{app:geodesic-approx}
The \textbf{geodesic approximation}\index{geodesic approximation}~\cite{Balasubramanian:1999zv} predicts the bulk two-point function is well approximated by the geodesics $d(A,B)$ connecting two points $A$ and $B$ as follows \eqref{eq:geodesic-approx}:
\begin{equation}
    \ev{\phi(X_A)\phi(Y_B)}\approx \frac{\Gamma(\Delta)}{\Gamma(\Delta+1-d/2)}e^{-\Delta d(A,B)/R},
\end{equation}
In this appendix, we demonstrate the above geodesic approximation correctly reproduces the CFT two-point function \eqref{eq:CFT-2pt} via the extrapolate dictionary \eqref{eq:dict}. For a detailed calculation on geodesics, see Appendix \ref{sec:Embedd}.

Since two points at the asymptotic boundary are largely separated due to hyperbolicity, \eqref{eq:geodesic-xi} is approximated by
\begin{equation}
    d(A,B)\approx R\log\frac{2}{\xi} \Leftrightarrow e^{d(A,B)/R}\approx \frac{2}{\xi}.
\end{equation}
Furthermore, it can be written in the Euclidean embedding coordinates, where $X_j^{(E)}=X_j$ for $j=1,\cdots, d-1$ and $X_{j}^{(E)}=iX_j$ for $j=d+1$ in \eqref{eq:geodesic-embedding}:
\begin{align}
    \frac{1}{\xi}=-\frac{X_A\cdot X_B}{R^2}
    %=\frac{X_0(A)X_0(B)-X_{d+1}^{(E)}(A)X_{d+1}^{(E)}(B)}{R^2}
    =-\frac{1}{R^2}\qty(-\frac{X_{+ A}X_{- B}+X_{- A}X_{+ B}}{2} + \sum_i X_{i\, A}^{(E)} X_{i\, B}^{(E)}),
\end{align}
where $X_{\pm}\equiv X_0 \pm X_{d}$ and $i=1,\cdots, d-1, d+1$. Using the extrapolate dictionary \eqref{eq:dict}, the (Euclidean) CFT two-point function is written as
\begin{align}
    \ev{O(P_A)O(P_B)} &\approx \frac{\Gamma(\Delta)}{\Gamma(\Delta+1-d/2)} \qty(\epsilon^2 \frac{2}{\xi})^{-\Delta} \\
    &= \frac{\Gamma(\Delta)}{\Gamma(\Delta+1-d/2)} \qty[ \frac{1}{R^2}\qty(P_{+ A}P_{- B}+P_{- A}P_{+ B} - 2\sum_i P_{i\, A}P_{i\, B})]^{-\Delta},
    \label{eq:CFT-geodesic}
\end{align}
where $P_{\pm}\equiv \lim_{z\rightarrow \epsilon} (\epsilon X_{\pm})$ and $P_i \equiv \lim_{z\rightarrow \epsilon} (\epsilon X_i^{(E)})$ are the embedding coordinates specifying the boundary points~\cite{Kaplan:2016,Penedones:2016voo,Tamaoka:2019pqo}.\footnote{Note that while the bulk embedding coordinates $X$'s are divergent at the asymptotic boundary, $P$'s remain finite as $\epsilon\rightarrow 0$.}

To see whether \eqref{eq:CFT-geodesic} really reproduces the CFT two-point function \eqref{eq:CFT-2pt}, we introduce a particular parametrization.\footnote{We can instead work in the embedding coordinates at the asymptotic boundary $P$ to derive the CFT two-point function. While symmetry is manifest in this approach, here we stick to a particular coordinate system to derive the familiar form \eqref{eq:CFT-2pt}.} Since the UV cutoff is simple in Poincar\'e coordinates, we work in these coordinates. After Wick rotations $X_{d+1}^{(E)}=i X_{d+1},\, t_E=it$, \eqref{eq:poincare-coords} becomes
\begin{equation}
    \begin{split}
        X_0 &= R \frac{\alpha^2+z^2+\bm{x}^2+t_E^{2}}{2\alpha z}\\
        X_{d+1}^{(E)} &= R \frac{t^{(E)}}{z}\\
        X_{\bar{i}} &= R \frac{x^{\bar{i}}}{z} \quad (\bar{i}=1,\cdots,d-1) \\
        X_d & = R \frac{-\alpha^2+z^2+\bm{x}^2+t_E^{2}}{2\alpha z}.
    \end{split}
\end{equation}
At the asymptotic boundary $z=\epsilon\ll R$, they become
\begin{equation}
    \begin{split}
        P_+ &= \epsilon \left.\qty(X_0 + X_{d+1}^{(E)})\right|_{z=\epsilon} = R \frac{x^2}{\alpha}\\
        P_- &= \epsilon \left.\qty(X_0 - X_{d+1}^{(E)})\right|_{z=\epsilon} = R \alpha\\
        P_i & = Rx^i \quad (i=1,\cdots, d-1,d+1)
    \end{split}
    ,
\end{equation}
where $x^i=(x^{\bar{i}}, t^{(E)})$ and $x^2=\sum_i (x^i)^2 = \bm{x}^2+t^{(E)\, 2}$. Substituting these into \eqref{eq:CFT-geodesic}, we obtain
\begin{equation}
    \ev{O(P_A)O(P_B)} \approx \frac{\Gamma(\Delta)}{\Gamma(\Delta+1-d/2)}
    \frac{1}{(x_A-x_B)^{2\Delta}},
\end{equation}
which is consistent with the CFT result \eqref{eq:CFT-2pt} up to the normalization. If one takes the $\Delta\rightarrow \infty$ limit, the prefactor approaches unity.

\section{Gravity action for AdS}\label{app:grav-action}
The Lorentzian gravity action\index{Lorentzian gravity action} of AdS is given by\footnote{In the Euclidean signature, where $t_E=it$ and the Euclidean action is $(-i)$ times the Lorentizian action: $I_E=-iI_L$, the overall sign flips since $iI_L=i\int \dd{t} L= -\int\dd{t_E} (-L)$.}
\begin{equation}
    I_{grav}=\frac{1}{16\pi G_N} \int \dd[d+1]{x} \sqrt{g} \qty(\mathcal{R}+\frac{d(d-1)}{R^2})+\frac{1}{8\pi G_N}\int_{\partial \mathrm{AdS}}\dd[d]{x} \sqrt{\gamma} K + I_{ct},
\end{equation}
where $\gamma=\det(\gamma_{ab})$ is the determinant of the induced metric\index{induced metric} at the asymptotic boundary and $K=\gamma^{ab}\nabla_a n_b$ is the \textbf{extrinsic curvature}\index{extrinsic curvature} ($n^a$ is the normal vector to the asymptotic boundary). The first term is the usual \textbf{Einstein-Hilbert action}\index{Einstein-Hilbert action} with a cosmological constant, the second term is the \textbf{Gibbons-Hawking (GH) term}\index{Gibbons-Hawking term}\index{GH term |see Gibbons-Hawking term }. It is responsible for canceling the terms involving the derivative of the metric perturbation when we consider the variation of $\mathcal{R}$~\cite{ammon_erdmenger_2015,Wald:1984rg}. The remaining terms in the variation are proportional to the metric perturbation, thus it vanishes if we impose the Dirichlet boundary condition\index{Dirichlet boundary condition} at the asymptotic boundary $\partial\mathrm{AdS}$. (Later, we choose a different boundary condition on the second boundary so that the variation from the GH term vanishes.)
%The Gibbons-Hawking term does not affect to the bulk equation of motion since it is a surface term. Nevertheless it is 
The \textbf{counterterm}\index{counterterm} action $I_{ct}$ is determined so that it cancels divergences to all orders with a finite number of counterterms at the asymptotic boundary~\cite{deHaro:2000vlm}. This procedure is known as \textbf{holographic renormalization}\index{holographic renormalization}. For reviews, see~\cite{Skenderis:2002wp}, Section 5.5 and 6.3.2 of~\cite{ammon_erdmenger_2015}, and references therein. For example, the IR divergence (infinite volume) of AdS (i.e. UV divergence in the dual CFT) can be removed by
\begin{equation}
    I_{ct}=\frac{1}{8\pi G_N} \frac{d-1}{R}\int\dd[d]{x}\sqrt{\gamma}.
\end{equation}
When $d=2$, this is sufficient to remove all the divergences~\cite{Balasubramanian:1999re}. However, in higher dimensions, this is not enough. When $d\le 6$, the counterterm action is given by
\begin{equation}
    I_{ct}=\frac{1}{8\pi G_N}\int\dd[d]{x} \sqrt{\gamma} \qty[\frac{d-1}{R}+\frac{R}{2(d-2)}\mathcal{R}-\frac{R^3}{2(d-2)^2(d-4)}\qty(\mathcal{R}^{ab}\mathcal{R}_{ab}-\frac{d}{4(d-1)}\mathcal{R}^2)+\cdots],
\end{equation}
where curvatures are made of $\gamma_{ab}$, to remove at least power-law divergences~\cite{Natsuume:2014sfa}.

\section{Derivation of conformal Ward identity}\label{app:conf-ward}
In the following, we derive the \textbf{conformal Ward identity}\index{conformal Ward identity}
\begin{equation}
    \frac{1}{2\pi i} \oint_{z_i} \dd{z} \epsilon^z (z) T(z) \Phi(z_i) = \delta_\epsilon \Phi(z_i).
    \label{eq:CWI}
\end{equation}
$\oint_z$ means the integration contour is a curve going around $z$ counterclockwise. (We will explain other notations later.)

Let us consider arbitrary field(s) $X$. Given the action $I[\phi,g]$ of the field $\phi$ and a background metric $g_{\mu\nu}$, the expectation value of $X$ is given by
\begin{equation}
    \ev{X}=\int\mathcal{D}\phi X e^{-I[\phi,g]}.
    \label{eq:evX}
\end{equation}
This is equivalent to
\begin{equation}
    \int\mathcal{D}\phi' X' e^{-I[\phi',g]}
    \label{eq:evX2}
\end{equation}
as we just reparametrize the field ($X'\equiv X(\phi')$).
If we choose $\phi'=\phi+\delta_\epsilon \phi$ (field after an infinitesimal conformal (Weyl) transformation\index{Weyl transformation} $z'=z+\epsilon^z(z)$), \eqref{eq:evX} = \eqref{eq:evX2} becomes
\begin{align}
    \ev{X} &= \int\mathcal{D}\phi' (X+\delta_\epsilon X) e^{-I[\phi+\delta_\epsilon \phi,g]} \\
    &= \int\mathcal{D}\phi (X+\delta_\epsilon X) e^{-I[\phi+\delta_\epsilon \phi,g]} \\
    &= \int\mathcal{D}\phi (X+\delta_\epsilon X) e^{-I[\phi,g]} e^{-\delta_\epsilon I}\\
    &= \int\mathcal{D}\phi (X+\delta_\epsilon X) e^{-I[\phi,g]} (1-\delta_\epsilon I)\\
    &= \ev{X} + \ev{\delta_\epsilon X} - \ev{\delta_\epsilon I}.
    \label{eq:variation-X}
\end{align}
In the second line, we assumed the path integral measure is diffeomorphism\index{diffeomorphism} invariant.\footnote{For advanced readers: This means the theory has no \textbf{gravitational anomaly}\index{gravitational anomaly}.} Note that since the metric is not dynamical in the calculation,
\begin{equation}
    \delta_\epsilon I = I[\phi,g] +\int \dd[d]{x} \frac{\delta I}{\delta \phi(x)} \delta_\epsilon \phi(x).
\end{equation}

Meanwhile, the variation of a CFT action should vanish when both the metric and field are dynamical due to the diffeomorphism\index{diffeomorphism} invariance:
\begin{equation}
    0=\delta I \equiv \delta_\epsilon I +\int \dd[2]{x} \frac{\delta I}{\delta g_{\mu\nu}(x)}\delta_\epsilon g_{\mu\nu}(x).
    \label{eq:diff-action}
\end{equation}
The infinitesimal diffeomorphism of the metric is written as\footnote{The minus sign is due to the pullback for the variation of metric.}
\begin{equation}
    \delta_\epsilon g_{\mu\nu}=-(\nabla_\mu \epsilon_\nu+\nabla_\nu \epsilon_\mu).
\end{equation}
Furthermore, the variation of the action with respect to metric is related to the \textbf{energy-momentum tensor}\index{energy-momentum tensor}, defined as
\begin{equation}
    T^{\mu\nu}\equiv \frac{2}{\sqrt{g}}\frac{\delta I}{\delta g_{\mu\nu}}.
\end{equation}
Then, taking the symmetric property of the energy-momentum tensor into account, \eqref{eq:diff-action} is reduced to
\begin{equation}
    0=\delta_\epsilon I - \int\dd[2]{x} \sqrt{g} (\nabla_\mu\epsilon_\nu) T^{\mu\nu}.
\end{equation}
Substituting $\delta_\epsilon I$ from the above equation into \eqref{eq:variation-X}, we obtain
\begin{equation}
    \ev{\delta_\epsilon X} = \int \dd[2]{x}\sqrt{g} \nabla_\mu\epsilon_\nu(x) \ev{T^{\mu\nu}(x) X}.
\end{equation}
If we choose $X=\prod_{i=1}^n \Phi_i(x_i)$, where $\Phi_i$ is a local (composite) field, we obtain
\begin{equation}
    \sum_i \ev{\Phi_1(x_1)\cdots \delta_\epsilon\Phi_i (x_i)\cdots\Phi_n(x_n)} = \int \dd[2]{x}\sqrt{g} \nabla_\mu\epsilon_\nu(x) \ev{T^{\mu\nu}(x) \prod_i \Phi_i(x_i)}.
    \label{eq:CWI-before}
\end{equation}
Let us consider a case where the nonzero component of the vector field $\epsilon$ is only $\epsilon=\epsilon^z \partial_z$ and otherwise zero. 
Then, it follows that for the flat Euclidean space $\dd{s}^2=\dd{(x+it_E)}\dd{(x-it_E)}=\dd{z}\dd{\bar{z}}$, \eqref{eq:CWI-before} is given by\footnote{From the metric, 
$g_{z\bar{z}}=g_{\bar{z}z}=1/2$ and
$\epsilon_{\bar{z}}=g_{\bar{z}z}\epsilon^z=\epsilon^z/2$. Furthermore, the conformal invariance implies the nonzero components of the energy-momentum tensor are $T^{\bar{z}\bar{z}}(z)$ and $T^{zz}(\bar{z})$.}
\begin{equation}
    \sum_i \ev{\Phi_1(z_1,\bar{z}_1)\cdots \delta_\epsilon\Phi_i (z_i,\bar{z}_i)\cdots\Phi_n(z_n,\bar{z}_n)} = \frac{1}{4i}\int\dd{\bar{z}}\wedge\dd{z} \partial_{\bar{z}}\epsilon^z(z,\bar{z}) \ev{T^{\bar{z}\bar{z}}(z) \prod_i \Phi_i(z_i,\bar{z}_i)},
\end{equation}
where we used $\dd[2]{x}\equiv \dd{x}\wedge \dd{t_E}=\frac{1}{2i} \dd{\bar{z}\wedge\dd{z}}$.\footnote{For a more pedagogical lecture, see Section 3.3 of~\cite{BB3163561X}.}

Let us consider a case $\epsilon^z$ is holomorphic everywhere except at insertion points $(z_1,\bar{z}_1),\cdots,(z_n,\bar{z}_n)$. By using Stokes' theorem\index{Stokes' theorem}, it reduces to a contour integral along a curve $C$ going around all the insertion points counterclockwise:\footnote{Since we apply Stokes' theorem to the region outside of insertion points, the contour runs counterclockwise for the outside region. This is the reason why we have $-C$ in the second line.}
\begin{align}
    \sum_i \ev{\Phi_1(z_1)\cdots \delta_\epsilon\Phi_i (z_i)\cdots\Phi_n(z_n)} & = \frac{1}{4i}\int \partial_{\bar{z}} \ev{\epsilon^z(z) T^{\bar{z}\bar{z}}(z) \prod_i \Phi_i(z_i)}\dd{\bar{z}}\wedge\dd{z}\\
    & = \frac{1}{4i}\oint_{-C} \dd{z} \ev{\epsilon^z(z) T^{\bar{z}\bar{z}}(z) \prod_i \Phi_i(z_i)} \qq{(holomorphic part)}.
\end{align}
Finally, defining the holomorphic energy-momentum tensor by $T(z)=-2\pi T_{zz}=-\pi T^{\bar{z}\bar{z}}(z)/2$, we obtain
\begin{equation}
    \sum_i \ev{\Phi_1(z_1)\cdots \delta_\epsilon\Phi_i (z_i)\cdots\Phi_n(z_n)} = \frac{1}{2\pi i}\oint_C \dd{z} \ev{\epsilon^z(z) T(z) \prod_i \Phi_i(z_i)}.
\end{equation}
The contour $C$ can be divided into small pieces winding around each insertion point. Thus, at the operator level, \eqref{eq:CWI} holds.

\section{Calculation of the energy-momentum tensor in the replicated space}\label{app:replica-calc-EM}
This section provides explicit calculations to derive \eqref{eq:replica-EM}.
Starting from $\Sigma_n$, a replicated space, we perform two conformal transformations:
\begin{equation}
    z\mapsto \zeta\equiv\frac{z-u}{z-v},
    \label{eq:conf-map1}
\end{equation}
mapping to a conical geometry like the right of Fig.\ref{fig:euclid-pathint}. Then,
\begin{equation}
    \zeta\mapsto w\equiv \zeta^{1/n},
    \label{eq:conf-map2}
\end{equation}
mapping to a complex plane $\mathbb{C}$. Since the Schwarzian term is generally cumbersome to calculate, we consider composing two relatively simple conformal transformations. In the following, we omit the superscript ' showing the field transformation.

Since \eqref{eq:conf-map2} is a bit simpler if we consider the $w$-derivatives of $\zeta=w^n$. Thus, we want the Schwarzian\index{Schwarzian |see Schwarzian derivative } to be $\{\zeta;w\}$, not $\{w;\zeta\}$. Following \eqref{eq:EM-trf-rule}, we can simplify the equation with $\ev{T(w)}_{\mathbb{C}}=0$:
\begin{align}
    \ev{T(\zeta)} &= -\frac{c}{12} \qty(\frac{dw}{d\zeta})^2 \{\zeta;w\} \\
    &=  -\frac{c}{12} \frac{1}{n^2 w^{2n-2}} \qty[\frac{(n-1)(n-2)}{w^2} - \frac{3}{2} \qty(\frac{n-1}{w})^2] \\
    &= \frac{c}{24}\qty(1-\frac{1}{n^2})\frac{1}{w^{2n}}= \frac{c}{24}\qty(1-\frac{1}{n^2})\frac{1}{\zeta^{2}}= \frac{c}{24}\qty(1-\frac{1}{n^2})\frac{(z-v)^2}{(z-u)^{2}}.
    \label{eq:EM-tensor-conf-map-2}
\end{align}
Note that
\begin{equation}
    \frac{(w^n)^{(m+k)}}{(w^n)^{(m)}}=\frac{(n-m)\cdots (n-(m+k-1))}{w^k}.
\end{equation}

For \eqref{eq:conf-map1}, the calculation is very simple since it is a global conformal transformation so that the Schwarzian vanishes. Aiming at equating with \eqref{eq:EM-tensor-conf-map-2},
\begin{align}
    \ev{T(\zeta)} &= \qty(\frac{dz}{d\zeta})^2 T(z) \\
    &=  \frac{(z-v)^4}{(u-v)^2} T(z).
    \label{eq:EM-tensor-conf-map-1}
\end{align}
By equating \eqref{eq:EM-tensor-conf-map-1} and \eqref{eq:EM-tensor-conf-map-2}, we obtain \eqref{eq:replica-EM}.

\section{Embedding the global AdS to the flat Space}\label{sec:Embedd}
The $(d+2)$-dimensional \textbf{anti-de Sitter (AdS) spacetime}\index{anti-de Sitter spacetime}\index{AdS} is obtained as embedding\index{embedding formalism}  to flat spacetime $\mathbf{R}^{2,d+1}$ with  
 \begin{equation}
     \Vec{X}\cdot \Vec{X}= - (X_0)^2 - (X_{d+2})^2 + (X_1)^2 + \cdots + (X_{d+1})^2 = - R^2 \label{eq:hyper}
 \end{equation}
The problem of finding geodesics $\Vec{X} (s)$ in  AdS is thus a variational problem with the constraint \eqref{eq:hyper}. 
In addition, since we focus on the spacelike geodesic (not null $ds^2=0$), we also need to impose the mass shell condition\index{mass shell condition}
\begin{equation}
    \Vec{P}^2+m^2=0,
    \label{eq:mass-shell}
\end{equation}
where $\Vec{P}$ is a conjugate momentum for $\Vec{X}$. Given the Lagrangian for the geodesic
\begin{equation}
    L_0=m\sqrt{\dot{\Vec{X}}\cdot \dot{\Vec{X}}},
\end{equation}
the conjugate momentum is given by
\begin{equation}
    \Vec{P}=\pdv{L_0}{\dot{\Vec{X}}},
\end{equation}
where $\dot{}=\frac{\dd}{\dd s}$.
Note that the canonical Hamiltonian\index{canonical Hamiltonian} ${\Vec{P}}\cdot\dot{\Vec{X}}-L_0$ is zero. However, since \eqref{eq:mass-shell} is a \textbf{primary constraint}\index{primary constraint}, i.e. it is derived only from the definition of the conjugate momentum, the actual time evolution is generated by the following Hamiltonian:
\begin{equation}
    H=\frac{h}{2}(\Vec{P}^2+m^2)+h\lambda(s)F(\vec{X}),   \quad
     F(\Vec{X})=  \Vec{X}\cdot \Vec{X} + R^2.
\end{equation}
$h\lambda (s)$ and $h/2$ are Lagrange multipliers\index{Lagrange multiplier} associated with each constraint \eqref{eq:hyper} and \eqref{eq:mass-shell}, respectively. The action is simply given by $S=\int\dd s (\vec{P}\cdot \dot{\vec{X}}-H)$. By integrating $\vec{P}$ in the path integral, the new action is
\begin{equation}
    S=\int \dd s \left[\frac{1}{2h}\dot{\Vec{X}}\cdot \dot{\Vec{X}} -\frac{h}{2}m^2 -h\lambda(s)F(\vec{X})\right].
\end{equation}
By rescaling the affine parameter $s\rightarrow s/h$, finally the action becomes
 \begin{align}
     S = \int ds L,\quad L = \frac{1}{2} \dot{\Vec{X}}\cdot \dot{\Vec{X}} - \lambda(s) F(\Vec{X}),
 \end{align}
where we dropped the constant term $-m^2/2$ in the Lagrangian. This is nothing but the action of a free particle constrained on the surface \eqref{eq:hyper}.

Equations of the motion for $\Vec{X},\lambda$ are thus given by  
 \begin{align}
     \Ddot{\Vec{X}}+ 2 \lambda \Vec{X} = 0, \label{eq:two-deriv}\\%\quad 
     \Vec{X}\cdot \Vec{X} + R^2 =0 \label{eq:constraint}
 \end{align}
 Taking derivatives two times  of the constraint gives 
 \begin{equation}
     \Ddot{\Vec{X}} \cdot \Vec{X} = - \dot{\Vec{X}}\cdot \dot{\Vec{X}} = -1.
     \label{eq:constraint2}
 \end{equation}
 The second equality holds because 
 the line element is given by
 \begin{equation}
     ds^2=-\eta_{AB} d{X}^A d{X}^B.
 \end{equation}
% we focus on spacelike geodesics
%\red{
%since the geodesic distance is given by
%\begin{equation}
%    \int \sqrt{-\eta_{AB} d{X}^A d{X}^B} = \int ds \sqrt{- \dv{\vec{X}}{s}\cdot\dv{\vec{X}}{s}} = \int ds,
%\end{equation}
%where the second equality holds as $s$ is the proper distance parametrizing the geodesic.
%}

By plugging \eqref{eq:two-deriv} into \eqref{eq:constraint2} and using \eqref{eq:constraint}, the Lagrange multiplier is fixed to be %now
%
%\red{
%Then,
%\begin{equation}
%    \Ddot{\vec{X}}=-\frac{\vec{X}}{\vec{X}\cdot\vec{X}}=\frac{\vec{X}}{R^2}.
%\end{equation}
%}
% 
% This fixes the Lagrange multiplier to be,
%
 \begin{equation}
     \lambda(s)= -\frac{1}{2 R^2},
 \end{equation}
and 
 \begin{align}
 \Ddot{\Vec{X}}-\frac{1}{R^2}\Vec{X}=0 \;  \Rightarrow \; 
 \Vec{X}(s) = \Vec{p} \cosh{\frac{s}{R}} + \Vec{q} \sinh{\frac{s}{R}}.
 \end{align}
 The constraint implies $\Vec{p}$ and $\Vec{q}$ satisfy
 \begin{equation}
     \Vec{p}\cdot \Vec{p} =  -\Vec{q}\cdot \Vec{q}=- R^2, \quad  \Vec{p}\cdot \Vec{q} = 0.
 \end{equation}
    We can freely set $s=0$ at the one end of the geodesic (call it A). Thus 
    \begin{equation}
         \Vec{p}= \Vec{X}_A
    \end{equation}
    We call another end as B, and by taking the inner product 
    of $\Vec{X}_B= \Vec{p} \cosh{\frac{s_B}{R}} + \Vec{q} \sinh{\frac{s_B}{R}}$ and $\Vec{X}_A$
    \begin{equation}
        \cosh{\frac{s_B}{R}} = -\frac{\Vec{X}_A\cdot\Vec{X}_B}{R^2}\equiv \xi^{-1}
        \label{eq:geodesic-embedding}
    \end{equation}
    and we obtain
    \begin{equation}
        \Vec{q} = \frac{1}{\sinh{\frac{s_B}{R}}}{\Vec{X}_B- \cosh{\frac{s_B}{R}}\Vec{X}_A}.
    \end{equation}
    
    Since $s=0$ at point A, the geodesic distance between A and B is simply given by $d(A,B)=s_B$. Thus, in terms of the embedding coordinates, it is written as
    \begin{equation}
        d(A,B)=s_B=R\cosh^{-1}{\xi^{-1}}=R \log\left(\frac{1+\sqrt{1-\xi^2}}{\xi}\right)
        \label{eq:geodesic-xi}
    \end{equation}
    by inverting \eqref{eq:geodesic-embedding}.
%The length of the geodesics between A and B $d(A,B)$ is just the difference of parameter 
%\begin{equation}
%     d(A,B)= \int_A^B \sqrt{d\Vec{X}\cdot d \Vec{X}}=s_B
%\end{equation}

Let us use the above formalism to compute the length of the geodesic connecting two points in global $AdS_{3}$, where the metric takes the form \eqref{Gmet}. The relation between  global coordinates\index{global coordinates} $(\tau, \theta, r)$  and embedding coordinates\index{embedding coordinates} $\vec{X}$ is 
\begin{align}
    X_{0}= \sqrt{R^2 +r^2}\cos{\tau}, \quad
    X_{1} = r \sin{\theta},\quad
    X_{2} = -r \cos{\theta},\quad
    X_{3}= \sqrt{R^2 +r^2}\sin{\tau}.
\end{align}

This relation allows us to write the relevant inner product $-R^2/\xi=\Vec{X}_A\cdot \Vec{X}_B$ in terms of the global coordinates
\begin{align}\label{AppendixA geodsics}
    \Vec{X}_A\cdot \Vec{X}_B& =  -{X_A}_0 {X_B}_0 -{X_A}_3 {X_B}_3+{X_A}_1 {X_B}_1+{X_A}_2 {X_B}_2 \\
    &= -\sqrt{\left(R^2+r_A^2\right)\left(R^2+{r_B^2}\right)} \cos(\tau_A-\tau_B) + r_A r_B\cos(\theta_A-\theta_B).
\end{align}
Then, we conclude the geodesic length\index{geodesic length} between two generic points is given by 
\begin{align}\label{AppendixA EE}
    d(A,B)=& R \cosh^{-1}{\left(\sqrt{\left(1+\frac{r_A^2}{R^2}\right)\left(1+\frac{r_B^2}{R^2}\right)} \cos(\tau_A-\tau_B) - \frac{r_A r_B}{R^2}\cos(\theta_A-\theta_B)\right)}.
\end{align}
In the case $r_A,r_B \gg R$, recalling $\cosh^{-1}{x}= \log{(x+\sqrt{x^2-1})}$, we obtain
\begin{equation}
    d(A,B)= R\log{\left(\frac{2r_A r_B}{R^2}\left(\cos(\tau_A-\tau_B)-\cos(\theta_A-\theta_B)\right)\right)}.
\end{equation}

%\section{Detail of the boundary curve}
\section{Details of the gravity-inducing boundary curve}\label{app:curve}
%\subsection{Gravity induced boundary}
 From (\ref{mappo}), we can associate $x$ with $\theta$ for each fixed time $t$ at the conformal boundary. By solving the fourth equation of (\ref{mappo}), which is just a quadratic equation for $x$, we can get 
 \begin{equation}\label{theta to x}
      x= \alpha \cot{\theta}\left(-1\pm \sqrt{1+ \tan^2{\theta}\left(1+\frac{t^2}{\alpha^2}\right)}\right).
 \end{equation}
The branch should be selected so that the map (\ref{theta to x}) is smooth for $\theta\in [-\pi,\pi)$. Thus for $-\frac{\pi}{2}< \theta < \frac{\pi}{2}$, we take $+$ sign; for $-\pi < \theta < -\frac{\pi}{2}$ and $\frac{\pi}{2}< \theta < \pi$, we take $-$ sign. In the other region, we take the sign so that the map has $2 \pi$ periodicity. Please see Fig.\ref{fig:theta to x}.
\begin{figure}[h]
    \centering
    \includegraphics[width = 100mm]{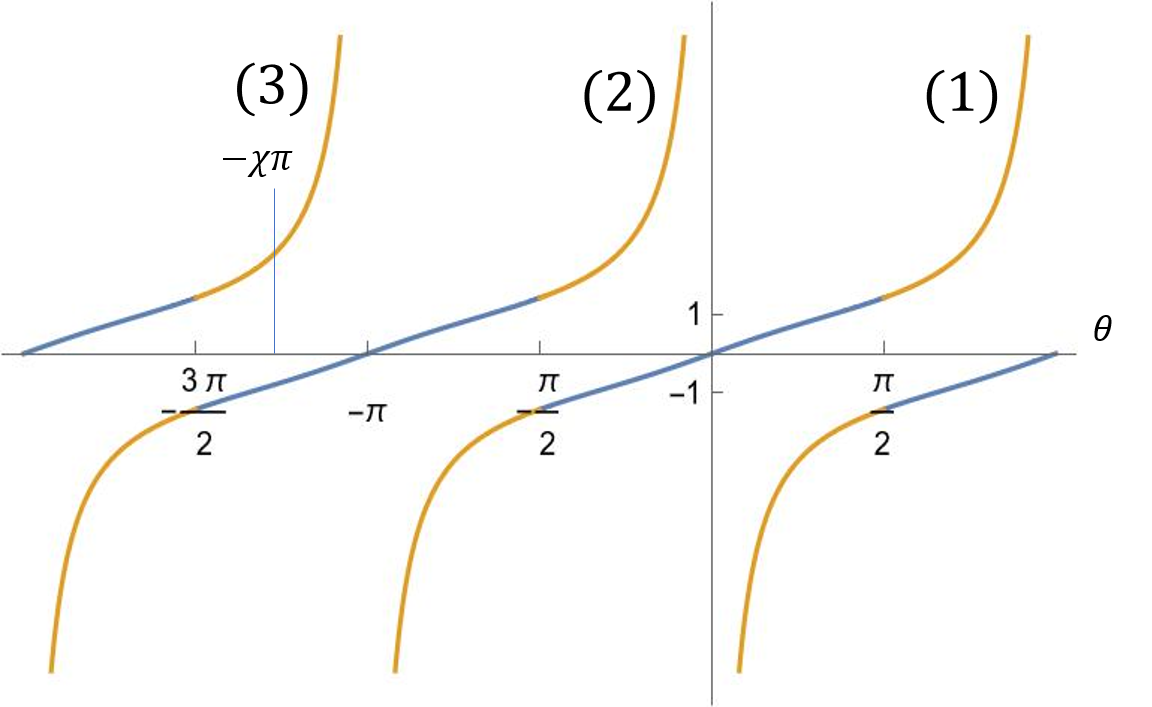}
    \caption{The blue curve corresponds to $+$ sign and the orange curve corresponds to $-$ of (\ref{theta to x}). We should take (1) and (3) branch for example.}
    \label{fig:theta to x}
\end{figure}
 Especially we have the gravity-inducing boundary $Z(t)$  at $\theta = -\frac{\pi}{\chi}$, where
 \begin{align}
     Z(t) = \left \{
     \begin{array}{ll}
        \alpha \cot{(\frac{\pi}{\chi})} \left(1+ \sqrt{1+\tan^2{(\frac{\pi}{\chi})}(1+\frac{t^2}{\alpha^2})}\right), & \frac{2}{3} < \chi < 1  \\
        \alpha \cot{(\frac{\pi}{\chi})} \left(1 - \sqrt{1+\tan^2{(\frac{\pi}{\chi})}(1+\frac{t^2}{\alpha^2})}\right).   &  \frac{1}{2} \leq \chi < \frac{2}{3}
     \end{array}
     \right.
 \end{align}
 For $\chi \leq \frac{1}{2}$, $Z(t)\leq 0$ and it looks there is no physical region for the holographic CFT. 
 \\The critical point $\chi = \frac{1}{2}$  can be seen when we consider whether the map between the tilde global coordinate$(\tilde{r},\ti{\theta},\ti{\tau})$ and the tilde Poincar\'e$(\ti{z},\ti{x},\ti{t})$ is well-defined.
%$\gamma=\tan\left(\pi\s{\frac{R^2}{R^2-M}}\right)$ is only positive for a limited value of $M$. $\gamma$ is not necessarily positive.
%Furthermore, the correct brane profile is obtained only when $0\le M<\frac{3}{4}R^2$ by carefully inspecting the coordinate transformation as we see in the following.
By a careful inspection of coordinate transformation \eqref{mappo}, choosing a particular domain for $\theta$ (or $\tilde{\theta}$) corresponds to considering a brane profile with a fixed winding number.
Let us first discuss the $-\pi\le \theta <\pi$ case in detail and comment briefly about other branches at last.% and the next subsection.

In the original global AdS spacetime with the massive particle, we chose the domain of $\theta$ to be $[-\pi,\pi)$. This results in the domain of $\tilde{\theta}$ \eqref{tildtheta}.
Since we discuss the profile of the EOW brane by bringing the solution from the Poincar\'e metric via \eqref{mappo} from the metric with $\tilde{\theta}$ \eqref{deficitg}, we are implicitly assuming the existence of the map between them. However, due to the deficit angle in $\tilde{\theta}$ and hence not covering the whole Poincar\'e spacetime, such a single-valued map does not exist for some values of $M>0$ given a fixed domain of $\theta$.

\begin{figure}
  \centering
  \includegraphics[width=13cm]{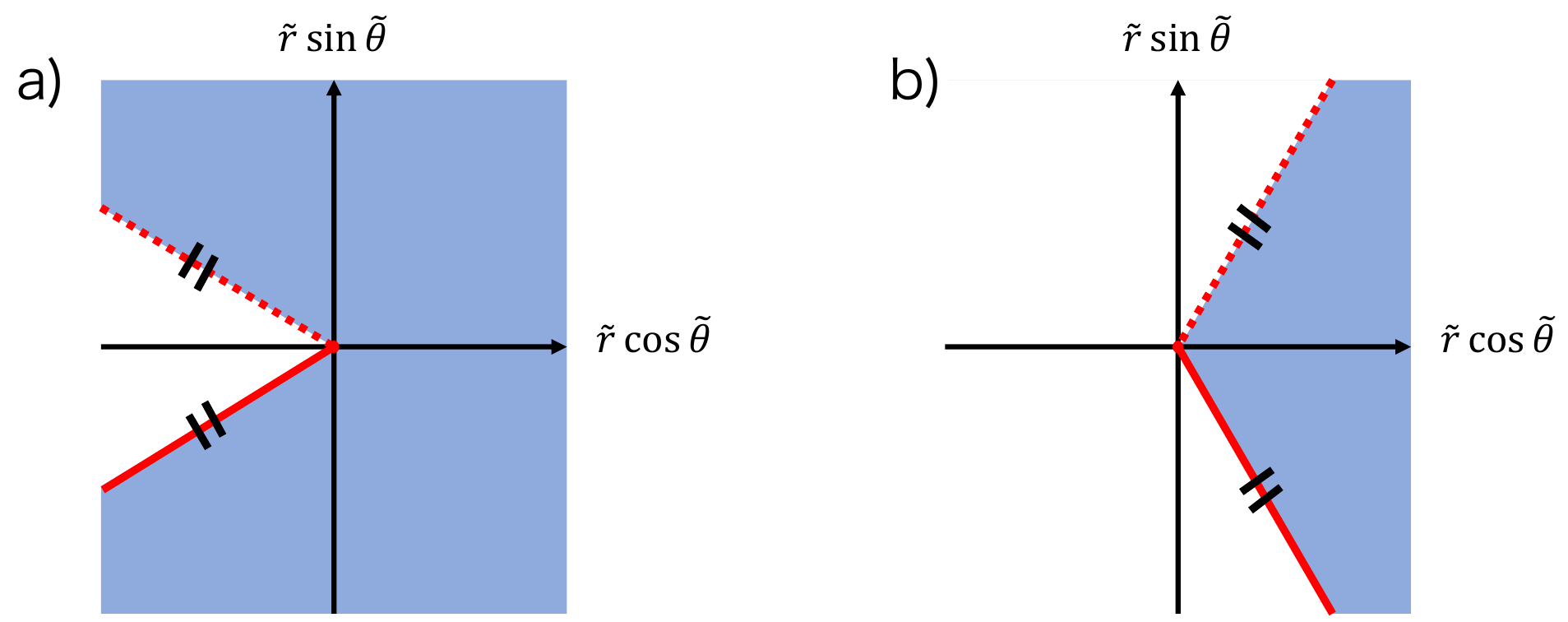}
  \caption{The allowed region for $-\pi \chi\le \tilde{\theta} < \pi \chi$ when a) $\frac{1}{2}\le\chi(<1)$ and b) $\chi\ge\frac{1}{2}$. The strike-though denotes an identification of two lines.}
\label{fig:deficit-geom}
\end{figure}

The deficit geometry in the rescaled coordinates $(\tilde{r},\tilde{\theta})$ is shown in Fig.\ref{fig:deficit-geom}.
From \eqref{mappo}, we expect the rescaled metric \eqref{deficitg} to be mapped to the Poincar\'e metric
\ba
ds^2=R^2\left(\frac{d\tilde{z}^2-d\tilde{t}^2+d\tilde{x}^2}{\tilde{z}^2}\right).
\ea
Let us focus on the relation
\ba
\tan\tilde{\theta}=\frac{\alpha^2 -\tilde{z}^2-\tilde{x}^2+\tilde{t}^2}{2\alpha\tilde{x}},
\label{tan}
\ea
By dividing into two cases according to the value of  ${\chi}<1$.

First, let us consider 
$\frac{1}{2}\le{\chi}$
as shown in a) in Fig.\ref{fig:deficit-geom}, where we need some case analyses.
For $-\frac{\pi}{2}<\ti{\theta}<\frac{\pi}{2}$, $\cos\tilde{\theta}\ge 0 \Leftrightarrow \tilde{z}^2+\tilde{x}^2-\tilde{t}^2\le \alpha^2$, there is no constraint on $\tan \tilde{\theta}$ and the conformal boundary region  $-\frac{\pi}{2}<\ti{\theta}<\frac{\pi}{2}$ is mapped to the conformal boundary region $x^2< t^2 +\alpha^2$. On the other hand, for $\cos\tilde{\theta} < 0$, $\tan \tilde{\theta}$ is constrained due to \eqref{tildtheta} as follows.
\ba
\begin{cases}
\tan\tilde{\theta} < \tan\left(\pi \chi\right) <0 & \text{for}\ \tilde{x}>0 \quad (\tilde{\theta}>\pi/2),\\
\tan\tilde{\theta} \ge \tan\left(-\pi \chi\right) >0 & \text{for}\ \tilde{x}<0 \quad (\tilde{\theta}<-\pi/2).
\end{cases}
\label{cond-for-tan}
\ea
In our setup, we are interested in the region $x>0$, which corresponds to $0<\ti{\theta}<\chi \pi$ and only consider that region for a while. Rewriting \eqref{cond-for-tan} in terms of Poincar\'e coordinates using \eqref{tan},
\ba
\tilde{x}\ge -\alpha \tan \left(\pi \chi \right)
+\sqrt{\alpha^2\tan^2 \left(\pi \chi\right)+\alpha^2+\tilde{t}^2}
\ea
and
\ba
\tilde{x}^2-\tilde{t}^2\ge \alpha^2
\ea
for $\tilde{x}\ge 0$ on the conformal boundary $\tilde{z}=0$. %Actually 
There is a region that satisfies both inequalities. The region $ \frac{\pi}{2}\geq \frac{\pi}\chi$ is mapped to the region in the tilde Poincar\'e coordinates satisfying the above inequalities.

Next, we consider $\chi\le 1/2$
as shown in b) in Fig.\ref{fig:deficit-geom}. We need
\ba
\tan\left(-\pi \chi\right) \le \tan\tilde{\theta}< \tan\left(\pi \chi\right).
\ea
However, we have no solution satisfying this in Poincar\'e coordinates.

From these analyses, we require
\ba
\frac{1}{2}<\chi = \sqrt{\frac{R^2-M}{R^2}} (\le 1) \quad \Longleftrightarrow \quad 0\le M <\frac{3}{4}R^2 (<R^2)
\label{restrictionM}
\ea
for the existence of the map from the rescaled coordinates $(\tilde{\tau}, \tilde{r},\tilde{\theta})$ to the Poincar\'e coordinates $(\tilde{t},\tilde{z},\tilde{x})$. Since we focus on the brane profile obtained in this manner, we restrict to \eqref{restrictionM} for $M<R^2$.
Within this domain, the boundary of the BCFT is indeed given by $\tilde{\theta}=0$ and $2\pi-\pi \chi$ for $M>0$.\footnote{For $M=0$, the boundary is given by $\tilde{\theta}=0,-\pi$.}

%\textbf{Puzzle:} 
The nontrivial endpoint of the EOW brane comes from $\tilde{\theta}=-\pi$. However, this is outside of the domain for $M>0$ as we chose $(-\pi<)-\pi \chi\le \tilde{\theta} < \pi \chi(<\pi)$. Within this domain, it seems that there is only one solution $\tilde{\theta}=0$. However, there should be two solutions from the viewpoint of the global patch. This paradox is resolved when you consider a proper origin of the angle such that the entire brane is within the fundamental domain.

The periodicity of $\tilde{\theta}$ given by $2\pi \chi$ is greater than $\pi/2$ regarding the previous analysis. This implies that by shifting the domain of $\tilde{\theta}$, we can always have two nodes in $\sin\tilde{\theta}$. It means the initial choice of $-\pi\le \theta<\pi$ may not be appropriate; however, for $1/2<\chi$ all configurations are connected continuously and we will not cause a problem by evading using periodicity as we have discussed.

\section{Branch choice in the identity conformal block}\label{app:branch}
In addition to the branch of $w^\prime$ discussed above, we need to choose the correct branch\index{branch} to analytically continue to the Lorentzian correlator \cite{Asplund:2014coa}. To see the monodromy\index{monodromy} around $z,\bar{z}=1$, we expand $z$ and $\bar{z}$ around $1$ provided $\ap$ is sufficiently small compared to other quantities. Since $e^{-i\varphi}$ factor only shifts the imaginary part, it is irrelevant to the discussion here. The cross ratios are expanded as follows:
\begin{align}
    z_{dis}&\propto\left(\frac{x-t-i\ap}{x-t+i\ap}\frac{x+t-i\ap}{x+t+i\ap}  \right)^\kappa &=1-\frac{4i\kappa x}{(x-t)(x+t)} \ap + O(\ap^2) \label{eq:1}\\
    \bar{z}_{dis}&\propto\left(\frac{x-t-i\ap}{x-t+i\ap}\frac{x+t-i\ap}{x+t+i\ap}  \right)^{-\kappa} 
    &=1+\frac{4i\kappa x}{(x-t)(x+t)} \ap + O(\ap^2)\label{eq:2}\\
    z_{con}&\propto\left(\frac{x_1-t+i\ap}{x_1-t-i\ap}\frac{x_2-t-i\ap}{x_2-t+i\ap}\right)^\kappa 
    &=1+\frac{2i\kappa (x_2-x_1)}{(x_1-t)(x_2-t)} \ap + O(\ap^2) \label{eq:3}\\
    \bar{z}_{con}&\propto\left(\frac{x_1+t+i\ap}{x_1+t-i\ap}\frac{x_2+t-i\ap}{x_2+t+i\ap}\right)^\kappa
    &=1+\frac{2i\kappa (x_2-x_1)}{(x_1+t)(x_2+t)} \ap + O(\ap^2). \label{eq:4}
\end{align}

When the sign of the imaginary part changes, $z$ or $\bar{z}$ moves to another sheet. From \eqref{eq:1}, $z_{dis}\rightarrow e^{-2\pi i} z_{dis}$ for $t\gg x$. From \eqref{eq:2}, $\bar{z}_{dis}\rightarrow e^{2\pi i} \bar{z}_{dis}$ for $t\gg x$. These branches for the disconnected entropy correspond to $m=0$ to $m=1$ in \eqref{eq:dis-first-br} and \eqref{eq:dis-next-br}. The transition between the two branches is determined so that the whole expression is continuous.

From \eqref{eq:3}, $z_{con}\rightarrow e^{2\pi i} z_{con}$ for $x_1\ll t\ll x_2$. Finally, from \eqref{eq:4}, there is no change in the branch as the sign of the imaginary part is always positive. For the connected entropy, by considering the Euclidean branch corresponding to the dominant OPEs\index{OPE |see operator product expansion} due to the approximation of the conformal block\index{conformal block}. Then, it should cancel the additional phase during $x_1\ll t\ll x_2$ and leads to the trivial branch \eqref{eq:conn-EE-CFT}.
\chapter{Appendix for Chapter V}
\ifpdf
    \graphicspath{{Chapter3/Figs/Raster/}{Chapter3/Figs/PDF/}{Chapter3/Figs/}}
\else
    \graphicspath{{Chapter3/Figs/Vector/}{Chapter3/Figs/}}
\fi

\section{Tensor network representation of quantum mechanics}\label{app:TN-rep}
It is useful to introduce a graphical representation of states and operators in quantum mechanics for various simplifications. In particular, such representations are known as \textbf{tensor networks}\index{tensor networks} in the context of quantum many-body systems. In the following, the basis of Hilbert space is fixed unless otherwise noted.

A ket (bra) is specified by a (co)vector. There is a one-to-one map with a fixed set of basis:
\begin{align}
    \ket{\psi}&=\psi^i \ket{\bm{e}_i}\mapsto \psi^i \\
    \bra{\psi}&=\bra{\bm{e}^i} \psi_i \mapsto \psi_i ,
\end{align}
where $\{\ket{\bm{e}_i}\}_i$ and $\{\bra{\bm{e}^i}\}_i$ are basis vectors for $\mathcal{H}$ and $\mathcal{H}^\ast$, respectively. The tensor network representation of a state (either bra or ket) is
\begin{equation}
\includegraphics[width=1cm, valign=c, clip]{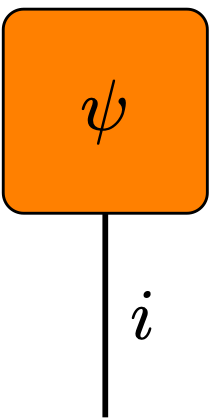}\quad \text{or}\quad
\includegraphics[width=2cm, valign=c, clip]{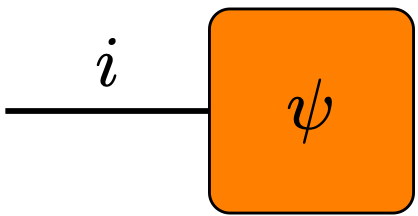}
\label{eq:ket-TN}
\end{equation}
as a (co)vector has one free index $i$. A line carrying an index is called a leg or \textbf{bond}\index{bond}. Although we will not care about the orientation of the leg, If we write the ket like \eqref{eq:ket-TN}, the bra is denoted by a tensor with a leg in the opposite orientation,
\begin{equation}
\includegraphics[width=1cm, valign=c, clip]{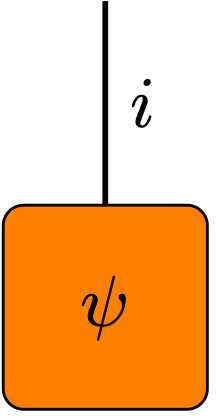}\quad \text{or}\quad
\includegraphics[width=2cm, valign=c, clip]{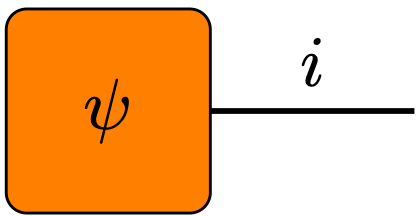}.
\end{equation}
We usually omit the index and simply write the ket vector as
\begin{equation}
    \ket{\psi}=\includegraphics[width=1cm,valign=c,clip]{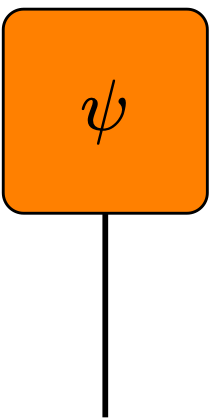}.
\end{equation}

An operator $\hat{O}:\mathcal{H}\rightarrow\mathcal{H}$ has a tensor network representation in the same manner. Since
\begin{equation}
    \hat{O}=\ket{\bm{e}_i}O^i_{\ j} \bra{\bm{e}_j},
\end{equation}
it is described by a tensor with two legs: a matrix $O^i_{\ j}$:
\begin{equation}
\includegraphics[width=1.2cm, valign=c, clip]{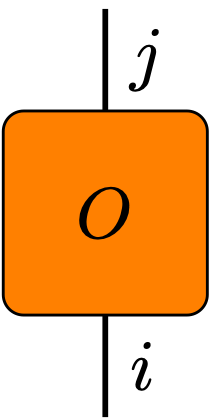}
\quad\text{or}\quad
\includegraphics[width=2cm, valign=c, clip]{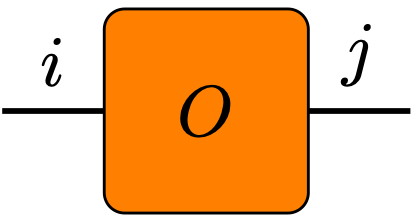}\quad .
\end{equation}

An EPR state \eqref{eq:epr-state} is denoted by a curve
\begin{equation}
    \frac{1}{\sqrt{d}}\includegraphics[width=1cm, valign=c, clip]{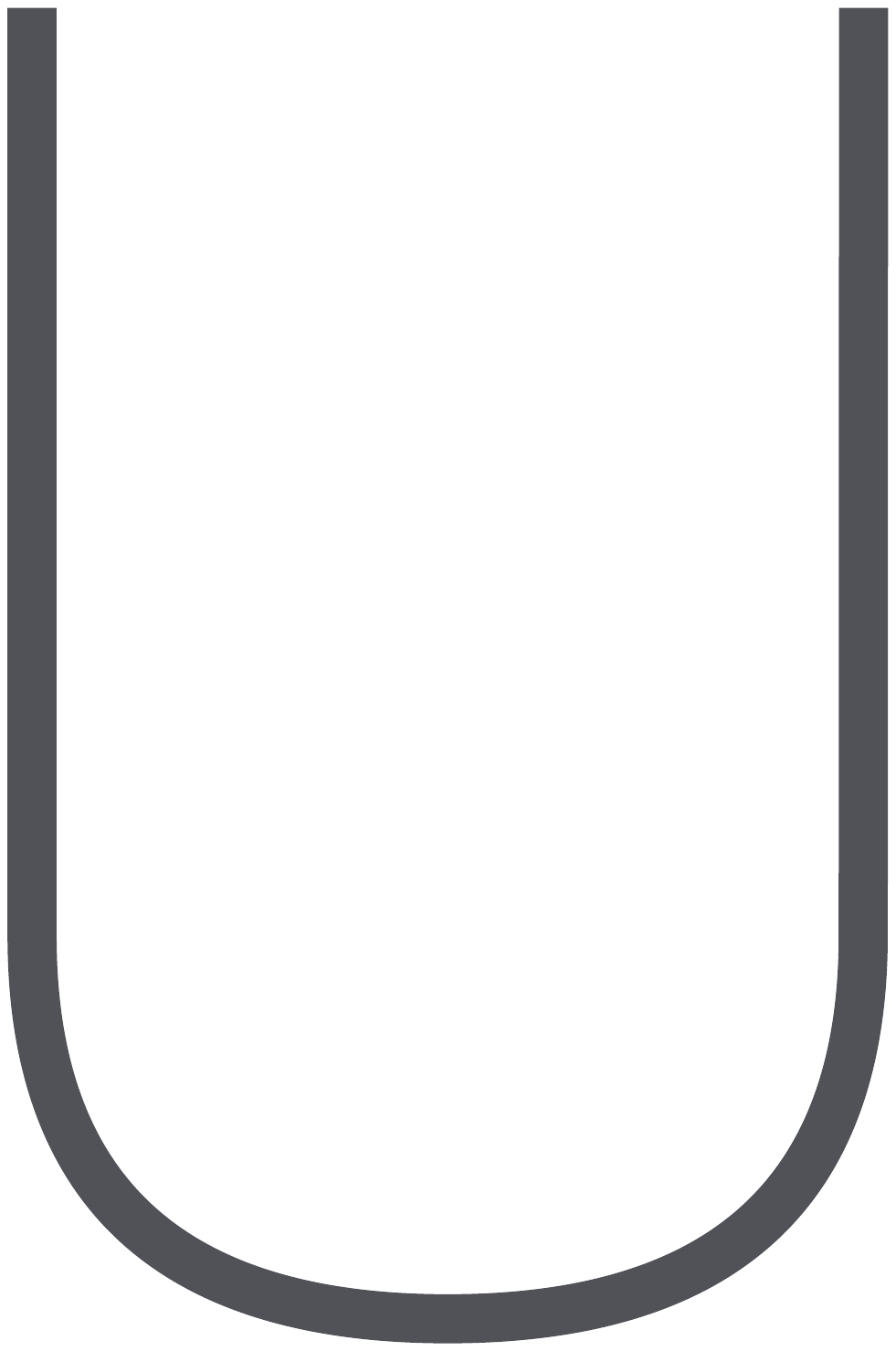}.
    \label{eq:epr-tensor}
\end{equation}

To contract indices, we can just connect the two corresponding legs. For example,
\begin{equation}
    \bar{\phi}^i=T^i_{\ jk} \phi^{(A)\, j}\phi^{(B)\, k}
    \quad
    \Leftrightarrow
    \quad
    \tikzset{every picture/.style={line width=0.75pt}}
    \begin{tikzpicture}[x=0.75pt,y=0.75pt,yscale=-1,xscale=1,baseline=-55]
    %\path (68,68); %set diagram left start at 68, and has height of 148
    \draw  [color={rgb, 255:red, 0; green, 0; blue, 0 }  ,draw opacity=0.33 ][fill={rgb, 255:red, 74; green, 144; blue, 226 }  ,fill opacity=0.67 ] (104,27.6) .. controls (104,23.95) and (106.95,21) .. (110.6,21) -- (130.4,21) .. controls (134.05,21) and (137,23.95) .. (137,27.6) -- (137,47.4) .. controls (137,51.05) and (134.05,54) .. (130.4,54) -- (110.6,54) .. controls (106.95,54) and (104,51.05) .. (104,47.4) -- cycle ;
    \draw  [color={rgb, 255:red, 0; green, 0; blue, 0 }  ,draw opacity=0.33 ][fill={rgb, 255:red, 248; green, 231; blue, 28 }  ,fill opacity=1 ] (140.99,109) -- (106.01,68.98) -- (176.01,69.02) -- cycle ;
    \draw  [color={rgb, 255:red, 0; green, 0; blue, 0 }  ,draw opacity=0.33 ][fill={rgb, 255:red, 74; green, 144; blue, 226 }  ,fill opacity=0.67 ] (146,27.6) .. controls (146,23.95) and (148.95,21) .. (152.6,21) -- (172.4,21) .. controls (176.05,21) and (179,23.95) .. (179,27.6) -- (179,47.4) .. controls (179,51.05) and (176.05,54) .. (172.4,54) -- (152.6,54) .. controls (148.95,54) and (146,51.05) .. (146,47.4) -- cycle ;
    \draw    (121,54) -- (121,69) ;
    \draw    (162,54) -- (162,69) ;
    \draw    (140.99,109) -- (140.99,124) ;
    \draw  [color={rgb, 255:red, 0; green, 0; blue, 0 }  ,draw opacity=0.33 ][fill={rgb, 255:red, 74; green, 144; blue, 226 }  ,fill opacity=0.67 ] (19,52.6) .. controls (19,48.95) and (21.95,46) .. (25.6,46) -- (45.4,46) .. controls (49.05,46) and (52,48.95) .. (52,52.6) -- (52,72.4) .. controls (52,76.05) and (49.05,79) .. (45.4,79) -- (25.6,79) .. controls (21.95,79) and (19,76.05) .. (19,72.4) -- cycle ;
    \draw    (34.99,79) -- (34.99,94) ;
    \draw (141,89) node    {$T$};
    \draw (123,35) node    {$\phi ^{( A)}$};
    \draw (165,35) node    {$\phi ^{( B)}$};
    \draw (171,62) node  [font=\footnotesize]  {$k$};
    \draw (131,62) node  [font=\footnotesize]  {$j$};
    \draw (142,133) node  [font=\small]  {$i$};
    \draw (35,61) node    {$\overline{\phi }$};
    \draw (36,103) node  [font=\small]  {$i$};
    \draw (78,68) node    {$=$};
    \draw (200,80) node  {.};
\end{tikzpicture}
\end{equation}
In the basis-independent notation, this corresponds to $\ket{\overline{\phi}}=\hat{T}\ket{\phi^{(A)}}\otimes\ket{\phi^{(B)}}$. $\hat{T}$ is a linear map from $\mathcal{H}\otimes\mathcal{H}$ to $\mathcal{H}$. The contracted indices are called \textbf{virtual}\index{virtual bonds} or \textbf{internal bonds}\index{internal bonds} while uncontracted indices carried by the original state in $\mathcal{H}$ are called \textbf{physical bonds}\index{physical bonds}.
The dimension of contracted indices is called the \textbf{bond dimension}\index{bond dimension} $\chi$.

%With these graphical notations at hand, some identities in linear algebra can be understood intuitively.

\nomenclature[g-p]{$\chi$}{bond dimension}   

%\section{Holographic tensor networks}

%\subsection{MERA}
%[Milsted-Vidal]
%\subsection{Perfect tensor networks}
%\subsubsection{Operator pushing}

%\subsection{Random tensor networks}

%\section{Operational interpretation of pseudo entropy}

% ********************************** Back Matter *******************************
% Backmatter should be commented out, if you are using appendices after References
%\backmatter

% ********************************** Bibliography ******************************
\begin{spacing}{0.9}

% To use the conventional natbib style referencing
% Bibliography style previews: http://nodonn.tipido.net/bibstyle.php
% Reference styles: http://sites.stat.psu.edu/~surajit/present/bib.htm

\bibliographystyle{JHEP}
\cleardoublepage
\bibliography{References/references2,References/EE,References/Island,References/master,References/HED} % Path to your References.bib file

% If you would like to use BibLaTeX for your references, pass `custombib' as
% an option in the document class. The location of 'reference.bib' should be
% specified in the preamble.tex file in the custombib section.
% Comment out the lines related to natbib above and uncomment the following line.

%\printbibliography[heading=bibintoc, title={References}]

\end{spacing}

% ********************************** Appendices ********************************

%\begin{appendices} % Using appendices environment for more functunality

%\include{Appendix1/appendix1}
%\include{Appendix2/appendix2}

%\end{appendices}

% *************************************** Index ********************************
%\addcontentsline{toc}{chapter}{\indexname}
\fontsize{12}{11}\selectfont\raggedright
\printthesisindex % If index is present

\end{CJK}
\end{document}